\documentclass[11pt]{amsart}

\usepackage{amsfonts}
\usepackage{mathrsfs}
\usepackage{mathtools}
\usepackage{amssymb}
\usepackage{latexsym,amsthm}
\usepackage[all]{xy}
\usepackage{graphicx}
\usepackage{wrapfig}
\usepackage{array,multirow}
\usepackage{bm}
\usepackage{enumitem,caption}

\usepackage{tensor}

\usepackage{longtable}

\usepackage{wasysym}

\usepackage{hyperref}
\hypersetup{%
    pdfborder = {0 0 0},
    colorlinks,
    citecolor=red!50!black,
    filecolor=Darkgreen,
    linkcolor=blue!40!black,
    urlcolor=cyan!50!black!90
}

\usepackage{tikz}
\usetikzlibrary{matrix,arrows,shapes,snakes}
\newcommand{\btp}{\begin{tikzpicture}[baseline=-.25em,scale=0.25,line width=0.7pt]}
\newcommand{\etp}{\end{tikzpicture}}

\allowdisplaybreaks[1]

\numberwithin{equation}{section}

\newtheorem{thrm}{Theorem}[section]
\newtheorem{prop}[thrm]{Proposition}
\newtheorem{crl}[thrm]{Corollary}
\newtheorem{lemma}[thrm]{Lemma}
\newtheorem{result}[thrm]{Result}	
\newtheorem{conj}[thrm]{Conjecture}

\theoremstyle{definition}
\newtheorem{defn}[thrm]{Definition}
\newtheorem{exam}[thrm]{Example}

\theoremstyle{remark}
\newtheorem{rmk}[thrm]{Remark}

\newcommand{\rmkend}{\ensuremath{\diameter}}
\newcommand{\examend}{\ensuremath{\diameter}}
\newcommand{\defnend}{\ensuremath{\diameter}}

\setlist[itemize,1]{leftmargin=.4in}
\setlist[enumerate,1]{leftmargin=.4in,label=(\roman*)}
\setlist[description,1]{leftmargin=.4in,font=\normalfont\itshape}

\newcommand{\nc}{\newcommand}
\newcommand{\rnc}{\renewcommand}

\nc{\al}{\alpha}
\nc{\eps}{\epsilon}
\nc{\veps}{\varepsilon}
\nc{\ga}{\gamma}
\nc{\Ga}{\Gamma}
\nc{\ka}{\kappa}
\nc{\la}{\lambda}
\nc{\La}{\Lambda}
\nc{\del}{\delta}
\nc{\om}{\omega}
\nc{\si}{\sigma}
\nc{\Ups}{\upsilon}
\nc{\vphi}{\varphi}

\nc{\id}{\mathrm{id}}
\nc{\Id}{\mathrm{Id}}
\nc{\gr}{\mathrm{gr}}
\rnc{\t}{\mathrm{t}}
\nc{\rk}{\mathrm{rank}}

\nc{\flL}{\phi_1}
\nc{\flR}{\phi_2}
\nc{\flLR}{\phi_{12}}

\nc{\bmodN}{\bmod \hspace{-3pt} '}

\nc{\Ug}{U\mathfrak{g}}
\nc{\Ub}{U\mathfrak{b}}

\nc{\ud}{\underline}
\nc{\tl}{\tilde}

\nc{\mbA}{\mathbf{A}}
\nc{\mbb}{\mathbf{b}}
\nc{\mbB}{\mathbf{B}}
\nc{\mbc}{\mathbf{c}}
\nc{\mbC}{\mathbf{C}}
\nc{\mbd}{\mathbf{d}}
\nc{\mbD}{\mathbf{D}}
\nc{\mbe}{\mathbf{e}}
\nc{\mbE}{\mathbf{E}}
\nc{\mbf}{\mathbf{f}}
\nc{\mbF}{\mathbf{F}}
\nc{\mbg}{\mathbf{g}}
\nc{\mbH}{\mathbf{H}}
\nc{\mbh}{\mathbf{h}}
\nc{\mbi}{\mathbf{i}}
\nc{\mbI}{\mathbf{I}}
\nc{\mbj}{\mathbf{j}}
\nc{\mbJ}{\mathbf{J}}
\nc{\mbk}{\mathbf{k}}
\nc{\mbK}{\mathbf{K}}
\nc{\mbL}{\mathbf{L}}
\nc{\mbM}{\mathbf{M}}
\nc{\mbQ}{\mathbf{Q}}
\nc{\mbq}{\mathbf{q}}
\nc{\mbr}{\mathbf{r}}
\nc{\mbT}{\mathbf{T}}
\nc{\mbu}{\mathbf{u}}
\nc{\mbU}{\mathbf{U}}
\nc{\mbv}{\mathbf{v}}
\nc{\mbV}{\mathbf{V}}
\nc{\mbw}{\mathbf{w}}
\nc{\mbW}{\mathbf{W}}
\nc{\mbX}{\mathbf{X}}
\nc{\mbY}{\mathbf{Y}}
\nc{\mbZ}{\mathbf{Z}}

\nc{\mbbA}{\mathbb{A}}
\nc{\mbbB}{\mathbb{B}}
\nc{\mbbD}{\mathbb{D}}
\nc{\mbbF}{\mathbb{F}}
\nc{\mbbV}{\mathbb{V}}
\nc{\mbbH}{\mathbb{H}}
\nc{\mbbK}{\mathbb{K}}
\nc{\mbbL}{\mathbb{L}}
\nc{\mbbP}{\mathbb{P}}
\nc{\mbbU}{\mathbb{U}}

\nc{\mcA}{\mathcal{A}}
\nc{\mcB}{\mathcal{B}}
\nc{\mcC}{\mathcal{C}}
\nc{\mcD}{\mathcal{D}}
\nc{\mcE}{\mathcal{E}}
\nc{\mcF}{\mathcal{F}}
\nc{\mcH}{\mathcal{H}}
\nc{\mcK}{\mathcal{K}}
\nc{\mcO}{\mathcal{O}}
\nc{\mcQ}{\mathcal{Q}}
\nc{\mcS}{\mathcal{S}}
\nc{\mcP}{\mathcal{P}}
\nc{\mcU}{\mathcal{U}}
\nc{\mcT}{\mathcal{T}}
\nc{\mcV}{\mathcal{V}}
\nc{\mcY}{\mathcal{Y}}
\nc{\mcZ}{\mathcal{Z}}

\nc{\mfa}{\mathfrak{a}}
\nc{\mfA}{\mathfrak{A}}
\nc{\mfb}{\mathfrak{b}}
\nc{\mfB}{\mathfrak{B}}
\nc{\mfC}{\mathfrak{C}}
\nc{\mfd}{\mathfrak{d}}
\nc{\mfD}{\mathfrak{D}}
\nc{\mfe}{\mathfrak{e}}
\nc{\mfE}{\mathfrak{E}}
\nc{\mff}{\mathfrak{f}}
\nc{\mfF}{\mathfrak{F}}
\nc{\mfg}{\mathfrak{g}}
\nc{\mfgl}{\mathfrak{g}\mathfrak{l}}
\nc{\mfh}{\mathfrak{h}}
\nc{\mfH}{\mathfrak{H}}
\nc{\mfJ}{\mathfrak{J}}
\nc{\mfk}{\mathfrak{k}}
\nc{\mfK}{\mathfrak{K}}
\nc{\mfl}{\mathfrak{l}}
\nc{\mfL}{\mathfrak{L}}
\nc{\mfM}{\mathfrak{M}}
\nc{\mfm}{\mathfrak{m}}
\nc{\mfn}{\mathfrak{n}}
\nc{\mfN}{\mathfrak{N}}
\nc{\mfo}{\mathfrak{o}}
\nc{\mfP}{\mathfrak{P}}
\nc{\mfQ}{\mathfrak{Q}}
\nc{\mfS}{\mathfrak{S}}
\nc{\mfsl}{\mathfrak{s}\mathfrak{l}}
\nc{\mfso}{\mathfrak{s}\mathfrak{o}}
\nc{\mfsp}{\mathfrak{s}\mathfrak{p}}
\nc{\mft}{\mathfrak{t}}
\nc{\mfU}{\mathfrak{U}}
\nc{\mfu}{\mathfrak{u}}
\nc{\mfUqsl}{\mathfrak{U}_q\mathfrak{sl}}
\nc{\mfUsl}{\mathfrak{Usl}}
\nc{\mfV}{\mathfrak{V}}
\nc{\mfX}{\mathfrak{X}}
\nc{\mfY}{\mathfrak{Y}}
\nc{\mfz}{\mathfrak{z}}

\nc{\mfgf}{\mathfrak{g}^{\rm fin}}

\nc{\mrmd}{\mathrm{d}}

\nc{\sal}{\check{\alpha}}
\nc{\cbeta}{\check{\beta}}
\nc{\cd}{\check{d}}
\nc{\cf}{\check{f}}
\nc{\cdelta}{\check{\delta}}
\nc{\ccr}{\check{r}}
\nc{\cs}{\check{s}}

\nc{\bv}{\breve{v}}

\nc{\tc}{\tilde{c}}
\nc{\tr}{\tilde{r}}
\nc{\ts}{\tilde{s}}
\nc{\tv}{\tilde{v}}

\nc{\Aut}{{\rm Aut}}
\nc{\Out}{{\rm Out}}
\nc{\sgn}{{\rm sgn}}
\rnc{\mod}{{\rm mod}}

\nc{\msA}{\mathsf{A}}
\nc{\msB}{\mathsf{B}}
\nc{\msC}{\mathsf{C}}
\nc{\msc}{\mathsf{c}}
\nc{\msD}{\mathsf{D}}
\nc{\msd}{\mathsf{d}}
\nc{\mse}{\mathsf{e}}
\nc{\msw}{\mathsf{w}}
\nc{\msq}{\mathsf{q}}
\nc{\msg}{\mathsf{g}}
\nc{\msE}{\mathsf{E}}
\nc{\msf}{\mathsf{f}}
\nc{\msF}{\mathsf{F}}
\nc{\msh}{\mathsf{h}}
\nc{\msH}{\mathsf{H}}
\nc{\msI}{\mathsf{I}}
\nc{\msJ}{\mathsf{J}}
\nc{\msK}{\mathsf{K}}
\nc{\msL}{\mathsf{L}}
\nc{\msP}{\mathsf{P}}
\nc{\msQ}{\mathsf{Q}}
\nc{\msR}{\mathsf{R}}
\nc{\mss}{\mathsf{s}}
\nc{\msS}{\mathsf{S}}
\nc{\msT}{\mathsf{T}}
\nc{\msU}{\mathsf{U}}
\nc{\msV}{\mathsf{V}}
\nc{\msX}{\mathsf{X}}
\nc{\msY}{\mathsf{Y}}
\nc{\msZ}{\mathsf{Z}}

\nc{\End}{\mathrm{End}}
\nc{\Ext}{\mathrm{Ext}}
\nc{\GL}{\mathrm{GL}}
\nc{\Hom}{\mathrm{Hom}}
\nc{\Ima}{\mathrm{Image}}
\nc{\Ind}{\mathrm{Ind}}
\nc{\Ker}{\mathrm{Ker}}
\nc{\RHom}{\mathrm{RHom}}
\nc{\EndCq}{\mathrm{End}\bb{C}(q)}
\nc{\GSat}{\mathrm{GSat}}
\nc{\Sat}{\mathrm{Sat}}
\nc{\WSat}{\mathrm{WSat}}
\nc{\Cyc}{\mathrm{Cyc}}
\nc{\Dih}{\mathrm{Dih}}
\nc{\Sym}{\mathrm{Sym}}

\nc{\gim}{{g.i.m.}}

\nc{\mtc}{\mathtt{c}}
\nc{\mtD}{\mathtt{D}}
\nc{\mte}{\mathtt{e}}
\nc{\mtE}{\mathtt{E}}
\nc{\mtf}{\mathtt{f}}
\nc{\mtF}{\mathtt{F}}
\nc{\mth}{\mathtt{h}}
\nc{\mtH}{\mathtt{H}}
\nc{\mtV}{\mathtt{V}}
\nc{\mtX}{\mathtt{X}}
\nc{\mty}{\mathtt{y}}

\nc{\ddeg}{\mathtt{deg}}
\nc{\dimm}{\mathtt{dim}}
\nc{\lmod}{\mathtt{lmod}}
\nc{\opp}{\mathtt{opp}}
\nc{\rmod}{\mathtt{rmod}}

\nc{\mmod}{\mathrm{mod}}
\nc{\nbh}{\mathrm{nbh}}

\nc{\mf}{\mathfrak}
\nc{\mc}{\mathcal}
\nc{\ms}{\mathsf}
\nc{\bb}{\mathbb}

\nc{\lrh}{\leftrightharpoons}
\nc{\iso}{\stackrel{\sim}{\longrightarrow}}
\nc{\liso}{\stackrel{\sim}{\longleftarrow}}
\nc{\wh}{\widehat}
\nc{\wt}{\widetilde}
\nc{\lra}{\longrightarrow}
\nc{\ra}{\rightarrow}
\nc{\into}{\hookrightarrow}
\nc{\onto}{\twoheadrightarrow}

\nc{\N}{\mathbb{N}}
\nc{\Z}{\mathbb{Z}}
\nc{\Q}{\mathbb{Q}}
\nc{\R}{\mathbb{R}}
\nc{\C}{\mathbb{C}}
\nc{\K}{\mathbb{K}}
\nc{\A}{\mathbb{A}}

\nc{\RT}{\mathscr{T}}

\nc{\SW}{\mathsf{SW}}

\nc{\ot}{\otimes}
\nc{\op}{\oplus}
\nc{\ol}{\overline}
\nc{\un}{\underline}

\nc{\lan}{\langle}
\nc{\ran}{\rangle}

\nc{\Xp}{x_{j_1,\ldots,j_p}^{-1}}
\nc{\Xeta}{x_{\eta_1,\ldots,\eta_e}^{-1}}

\nc{\Dj}{{\text{\it\tiny DJ}}}

\nc{\UQ}{U_q(\mfh')_\Theta}

\nc{\Usl}{\mfU_q\mfsl_2}
\nc{\Asl}{\mfU_q\wh{\mfsl}_2}

\nc{\Ugl}{\mfU_q\mfgl_2}
\nc{\Agl}{\mfU_q\wh{\mfgl}_2}

\nc{\UEgl}{\mfU_q^{ext}\mfgl_2}
\nc{\AEgl}{\mfU_q^{ext}\wh{\mfgl}_2}

\nc{\UDsl}{\mfU_q^{\text{\tiny DJ}}\mfsl_2}
\nc{\ADsl}{\mfU_q^{\text{\tiny DJ}}\wh{\mfsl}_2}

\nc{\bT}{\bar{T}}
\nc{\bt}{\bar{t}}

\nc{\cT}{\mathcal{T}}

\nc{\bcT}{\bar{\mathcal{T}}}
\nc{\bct}{\bar{\tau}}

\nc{\spl}[1]{\begin{align}\begin{split}#1\end{split}\end{align}}
\nc{\eqd}[1]{\begin{equation}\begin{aligned}#1\end{aligned}\end{equation}}
\nc{\eqa}[1]{\begin{align}#1\end{align}}
\nc{\eqg}[1]{\begin{gather}#1\end{gather}}
\nc{\eq}[1]{\begin{equation}#1\end{equation}}
\nc{\eqn}[1]{\begin{align*}#1\end{align*}}

\nc{\casesl}[3][1.2]{\ensuremath{\begingroup \rnc{\arraystretch}{#1} \, {\left\{ \begin{array}{@{}#2@{}} #3 \end{array} \right.} \, \endgroup}}
\nc{\casesm}[3][1.2]{\ensuremath{\begingroup \rnc{\arraystretch}{#1} {\begin{array}{@{}#2@{}} #3 \end{array}} \endgroup }}
\nc{\casesr}[3][1.2]{\ensuremath{\begingroup \rnc{\arraystretch}{#1} \, {\left. \begin{array}{@{}#2@{}} #3 \end{array} \right\}} \,  \endgroup}}

\nc{\tx}[1]{\qu\text{#1}\qu}

\nc{\eqrefs}[2]{\text{(\ref{#1}-\ref{#2})}}

\nc{\didi}[2]{\ensuremath{[#1\!:\!#2]}}

\nc{\red}{\color{red}}
\nc{\blu}{\color{blue}}
\nc{\br}{\color{brown}}
\nc{\grn}{\color{green!55!black}}
\nc{\gry}{\color{gray}}
\nc\el{\nonumber\\}
\nc\nn{\nonumber}

\nc\Tr{{\rm Tr}}

\nc{\sm}[1]{\text{\tiny{\rm #1}}}

\nc{\npb}{\nopagebreak}

\nc{\smallsum}{\text{\small$\sum\hspace{.5mm}$}}

\nc\tdeg{\stackrel{\sim}{\smash{\deg}\rule{0pt}{1.1ex}}}

\nc{\adr}{{\rm ad}_\mathrm{r}\,}
\nc{\adl}{{\rm ad}_\mathrm{l}\,}
\nc{\ad}{{\rm ad}}
\nc{\Ad}{{\rm Ad}}

\nc{\qu}{\quad}
\nc{\qq}{\qquad}

\rnc{\d}[1]{\frac{\rm d}{{\rm d}#1}}

\nc{\litem}[1]{\medskip\noindent\hspace{-13pt} $\circ$ #1:}

\renewcommand{\,}{\kern 0.1em} 

\nc{\shape}[1]{\begin{minipage}[c]{10mm} \vspace{1pt} \begin{tikzpicture}[scale=.18] #1 \end{tikzpicture} \vspace{1pt} \end{minipage}}

\def\evenAalt{\draw (0,0) rectangle (3,3) ;\foreach \x in {0,.5,1,1.5,2,2.5} {\fill (\x,3-\x) +(.25,0) rectangle ++(.5,-.25);\fill (\x,3-\x) +(0,-.25) rectangle ++(.25,-.5);}}
\def\evenArot{\draw (0,0) rectangle (3,3) ;\foreach \x in {0,.25,.5,.75,1,1.25} {\fill (1.5+\x,3-\x) +(0,0) rectangle ++(.25,-.25);\fill (\x,1.5-\x) +(0,0) rectangle ++(.25,-.25);}}
\def\evenantidiag{\draw (0,0) rectangle (3,3) ;\foreach \x in {0,.25,.5,.75,1,1.25,1.5,1.75,2,2.25,2.5,2.75} {\fill (\x,\x) +(0,0) rectangle ++(.25,.25);}}
\def\evendiag{\draw (0,0) rectangle (3,3) ;\foreach \x in {0,.25,.5,.75,1,1.25,1.5,1.75,2,2.25,2.5,2.75} {\fill (\x,3-\x) +(0,0) rectangle ++(.25,-.25);}}
\def\evendiagoutblocks{\draw (0,0) rectangle (3,3) ;\foreach \x in {.5,.75,1,1.25,1.5,1.75,2,2.25} {\fill (\x,3-\x) +(0,0) rectangle ++(.25,-.25);}; \foreach \x in {.0,.25,2.5,2.75} {\fill (\x,0) +(0,0) rectangle ++(.25,.25);\fill (\x,.25) +(0,0) rectangle ++(.25,.25);\fill (\x,2.5) +(0,0) rectangle ++(.25,.25);\fill (\x,2.75) +(0,0) rectangle ++(.25,.25);}}
\def\evendiagmidblocks{\draw (0,0) rectangle (3,3) ;\foreach \x in {0,.25,1,1.25,1.5,1.75,2.5,2.75} {\fill (\x,3-\x) +(0,0) rectangle ++(.25,-.25);}; \foreach \x in {.5,.75,2,2.25} {\fill (\x,.5) +(0,0) rectangle ++(.25,.25);\fill (\x,.75) +(0,0) rectangle ++(.25,.25);\fill (\x,2) +(0,0) rectangle ++(.25,.25);\fill (\x,2.25) +(0,0) rectangle ++(.25,.25);}}
\def\evendiagalt{\draw (0,0) rectangle (3,3) ;\foreach \x in {0,.25,.5,.75,1,1.25,1.5,1.75,2,2.25,2.5,2.75} {\fill (\x,3-\x) +(0,0) rectangle ++(.25,-.25);} \foreach \x in {0,.5,1,1.5,2,2.5}{\fill(\x,\x+.25) +(0,0) rectangle ++(.25,.25);\fill(\x+.25,\x) +(0,0) rectangle ++(.25,.25);}}
\def\evendiagaltin{\draw (0,0) rectangle (3,3) ;\foreach \x in {0,.25,.5,.75,1,1.25,1.5,1.75,2,2.25,2.5,2.75} {\fill (\x,3-\x) +(0,0) rectangle ++(.25,-.25);} \foreach \x in {0,.5,2,2.5}{\fill(\x,\x+.25) +(0,0) rectangle ++(.25,.25);\fill(\x+.25,\x) +(0,0) rectangle ++(.25,.25);}}

\def\evendiagaltinout{ \draw (0,0) rectangle (3,3) ;\foreach \x in {0,.25,.5,.75,1,1.25,1.5,1.75,2,2.25,2.5,2.75} {\fill (\x,3-\x) +(0,0) rectangle ++(.25,-.25);} \foreach \x in {.25,.75,1.75,2.25}{\fill(\x,\x+.25) +(0,0) rectangle ++(.25,.25);\fill(\x+.25,\x) +(0,0) rectangle ++(.25,.25);} }
\def\evenalt{\draw (0,0) rectangle (3,3); \foreach \x in {0,.5,1,1.5,2,2.5}{\fill(\x,\x+.25) +(0,0) rectangle ++(.25,.25);\fill(\x+.25,\x) +(0,0) rectangle ++(.25,.25);}}
\def\evencross{\draw (0,0) rectangle (3,3) ;\foreach \x in {0,.25,.5,.75,1,1.25,1.5,1.75,2,2.25,2.5,2.75} {\fill (\x,3-\x) +(0,0) rectangle ++(.25,-.25);\fill (\x,\x) +(0,0) rectangle ++(.25,.25);}}
\def\evencrossin{\draw (0,0) rectangle (3,3) ;\foreach \x in {0,.25,.5,.75,2,2.25,2.5,2.75} {\fill (\x,3-\x) +(0,0) rectangle ++(.25,-.25);\fill (\x,\x) +(0,0) rectangle ++(.25,.25);}; \foreach \x in {1,1.25,1.5,1.75} {\fill (\x,3-\x) +(0,0) rectangle ++(.25,-.25);}}
\def\evencrossinout{\draw (0,0) rectangle (3,3) ;\foreach \x in {.5,.75,2,2.25} {\fill (\x,3-\x) +(0,0) rectangle ++(.25,-.25);\fill (\x,\x) +(0,0) rectangle ++(.25,.25);}; \foreach \x in {0,.25,1,1.25,1.5,1.75,2.5,2.75} {\fill (\x,3-\x) +(0,0) rectangle ++(.25,-.25);} }
\def\evencrossinshift{\draw (0,0) rectangle (3,3) ;\foreach \x in {0,.25,.5,2,2.25,2.5} {\fill (\x,3-\x) +(0,0) rectangle ++(.25,-.25);\fill (\x,.25+\x) +(0,0) rectangle ++(.25,.25);};\foreach \x in {.75,1,1.25,1.5,1.75,2.75} {\fill (\x,3-\x) +(0,0) rectangle ++(.25,-.25);}}
\def\evendoublecross{\draw (0,0) rectangle (3,3) ;\foreach \x in {0,.25,.5} {\fill (\x,1.5+\x) +(0,0) rectangle ++(.25,.25); \fill(1.5-\x,3-\x) +(0,0) rectangle ++(-.25,-.25); \fill (1.5+\x,\x) + (0,0) rectangle ++ (.25,.25);\fill (3-\x,1.5-\x) + (0,0) rectangle ++ (-.25,-.25);} \foreach \x in {0,.25,.5,.75,1,1.25,1.5,1.75,2,2.25,2.5,2.75} {\fill (\x,3-\x) +(0,0) rectangle ++(.25,-.25);}}
\def\evendoublecrossin{\draw (0,0) rectangle (3,3) ;\foreach \x in {0,.25} {\fill (\x,1.5+\x) +(0,0) rectangle ++(.25,.25); \fill(1.5-\x,3-\x) +(0,0) rectangle ++(-.25,-.25); \fill (1.5+\x,\x) + (0,0) rectangle ++ (.25,.25);\fill (3-\x,1.5-\x) + (0,0) rectangle ++ (-.25,-.25);} \foreach \x in {0,.25,.5,.75,1,1.25,1.5,1.75,2,2.25,2.5,2.75} {\fill (\x,3-\x) +(0,0) rectangle ++(.25,-.25);}}
\def\evendoublecrossinbig{\draw (0,0) rectangle (3,3) ;\foreach \x in {0} {\fill (\x,1.5+\x) +(0,0) rectangle ++(.25,.25); \fill(1.5-\x,3-\x) +(0,0) rectangle ++(-.25,-.25); \fill (1.5+\x,\x) + (0,0) rectangle ++ (.25,.25);\fill (3-\x,1.5-\x) + (0,0) rectangle ++ (-.25,-.25);} \foreach \x in {0,.25,.5,.75,1,1.25,1.5,1.75,2,2.25,2.5,2.75} {\fill (\x,3-\x) +(0,0) rectangle ++(.25,-.25);}}

\def\odddiag{\draw (0,0) rectangle (3.25,3.25) ;\foreach \x in {0,.25,.5,.75,1,1.25,1.5,1.75,2,2.25,2.5,2.75,3} {\fill (\x,3.25-\x) +(0,0) rectangle ++(.25,-.25);}}
\def\odddiagoutblocks{\draw (0,0) rectangle (3.25,3.25) ;\foreach \x in {.5,.75,1,1.25,1.5,1.75,2,2.25,2.5} {\fill (\x,3.25-\x) +(0,0) rectangle ++(.25,-.25);}; \foreach \x in {.0,.25,2.75,3} {\fill (\x,0) +(0,0) rectangle ++(.25,.25);\fill (\x,.25) +(0,0) rectangle ++(.25,.25);\fill (\x,2.75) +(0,0) rectangle ++(.25,.25);\fill (\x,3) +(0,0) rectangle ++(.25,.25);}}
\def\odddiagmidblocks{\draw (0,0) rectangle (3.25,3.25) ;\foreach \x in {0,.25,.5,.75,1,1.25,1.5,1.75,2,2.25,2.5,2.75,3} {\fill (\x,3.25-\x) +(0,0) rectangle ++(.25,-.25);}; \foreach \x in {.75,1,2,2.25} {\fill (\x,.75) +(0,0) rectangle ++(.25,.25);\fill (\x,1) +(0,0) rectangle ++(.25,.25);\fill (\x,2) +(0,0) rectangle ++(.25,.25);\fill (\x,2.25) +(0,0) rectangle ++(.25,.25);}}
\def\odddiaginblock{\draw (0,0) rectangle (3.25,3.25) ;\foreach \x in {0,.25,.5,.75,1,2,2.25,2.5,2.75,3} {\fill (\x,3.25-\x) +(0,0) rectangle ++(.25,-.25);}; \foreach \x in {1.25,1.5,1.75} {\fill (\x,1.25) +(0,0) rectangle ++(.25,.25);\fill (\x,1.5) +(0,0) rectangle ++(.25,.25);\fill (\x,1.75) +(0,0) rectangle ++(.25,.25);}}
\def\oddcrossin{\draw (0,0) rectangle (3.25,3.25) ;\foreach \x in {0,.25,.5,.75,1,2,2.25,2.5,2.75,3} {\fill (\x,3.25-\x) +(0,0) rectangle ++(.25,-.25);\fill (\x,\x) +(0,0) rectangle ++(.25,.25);} ;\foreach \x in {1,1.25,1.5,1.75} {\fill (\x,3.25-\x) +(0,0) rectangle ++(.25,-.25);}}
\def\oddcrossinout{\draw (0,0) rectangle (3.25,3.25) ; \foreach \x in {0,.25,1,1.25,1.5,1.75,2,2.25,2.5,2.75,3} {\fill (\x,3.25-\x) +(0,0) rectangle ++(.25,-.25);}  \foreach \x in {.5,.75,1,2,2.25,2.5} {\fill (\x,3.25-\x) +(0,0) rectangle ++(.25,-.25);\fill (\x,\x) +(0,0) rectangle ++(.25,.25);}}

\def\odddiagaltin{\draw (0,0) rectangle (3.25,3.25);\foreach \x in {0,.25,.5,.75,1.75,2,2.25,2.5,2.75,3} {\fill (\x,3.25-\x) +(0,0) rectangle ++(.25,-.25);};\foreach \x in {0,.5,2.25,2.75}{\fill (\x,\x+.25) +(0,0) rectangle ++(.25,.25);\fill (\x+.25,\x) +(0,0) rectangle ++(.25,.25);};\foreach \x in {1,1.25,1.5} {\fill (\x,3.25-\x) +(0,0) rectangle ++(.25,-.25);}}
\def\odddiagaltinout{\draw (0,0) rectangle (3.25,3.25);\foreach \x in {0,.25,.5,.75,1.75,2,2.25,2.5,2.75,3} {\fill (\x,3.25-\x) +(0,0) rectangle ++(.25,-.25);};\foreach \x in {0.25,.75,2,2.5}{\fill (\x,\x+.25) +(0,0) rectangle ++(.25,.25);\fill (\x+.25,\x) +(0,0) rectangle ++(.25,.25);};\foreach \x in {1,1.25,1.5} {\fill (\x,3.25-\x) +(0,0) rectangle ++(.25,-.25);}}

\makeatletter
\newcommand*\rel@kern[1]{\kern#1\dimexpr\macc@kerna}
\newcommand*\widebar[1]{%
  \begingroup
  \def\mathaccent##1##2{%
    \rel@kern{0.8}%
    \overline{\rel@kern{-0.8}\macc@nucleus\rel@kern{0.2}}%
    \rel@kern{-0.2}%
  }%
  \macc@depth\@ne
  \let\math@bgroup\@empty \let\math@egroup\macc@set@skewchar
  \mathsurround\z@ \frozen@everymath{\mathgroup\macc@group\relax}%
  \macc@set@skewchar\relax
  \let\mathaccentV\macc@nested@a
  \macc@nested@a\relax111{#1}%
  \endgroup
}
\makeatother
\makeatletter
\def\nobreakhline{%
  \noalign{\ifnum0=`}\fi
    \penalty\@M
    \futurelet\@let@token\LT@@nobreakhline}
\def\LT@@nobreakhline{%
  \ifx\@let@token\hline
    \global\let\@gtempa\@gobble
    \gdef\LT@sep{\penalty\@M\vskip\doublerulesep}
  \else
    \global\let\@gtempa\@empty
    \gdef\LT@sep{\penalty\@M\vskip-\arrayrulewidth}
  \fi
  \ifnum0=`{\fi}%
  \multispan\LT@cols
     \unskip\leaders\hrule\@height\arrayrulewidth\hfill\cr
  \noalign{\LT@sep}%
  \multispan\LT@cols
     \unskip\leaders\hrule\@height\arrayrulewidth\hfill\cr
  \noalign{\penalty\@M}%
  \@gtempa}
\makeatother

\rnc\appendixname{}

\begin{document}

\begin{flushright}
DMUS-MP-16/05
\end{flushright}

\bigskip

\bigskip

\title[Reflection matrices, coideal subalgebras and generalized Satake diagrams]
{Reflection matrices, coideal subalgebras \\ and generalized Satake diagrams of affine type}

\author{Vidas Regelskis}
\address{Department of Mathematics, University of Surrey, Guildford, GU2 7YX, UK and \newline \mbox{\hspace{.31cm}} Department of Mathematics, University of York, York, YO10 5DD, UK}
\email{vidas.regelskis@york.ac.uk}

\author{Bart Vlaar}
\address{School of Mathematical Sciences, University of Nottingham, Nottingham, NG7 2RD, UK and \newline \mbox{\hspace{.31cm}} Department of Mathematics, University of York, York, YO10 5DD, UK}
\email{bart.vlaar@york.ac.uk}

\begin{abstract}
We present a generalization of the theory of quantum symmetric pairs as developed by Kolb and Letzter. We introduce a class of generalized Satake diagrams that give rise to (not necessarily involutive) automorphisms of the second kind of symmetrizable Kac-Moody algebras $\mfg$. These lead to right coideal subalgebras $B_{\bm c,\bm s}$ of quantized enveloping algebras $U_q(\mfg)$. 

In the case that $\mfg$ is a twisted or untwisted affine Lie algebra of classical type Jimbo found intertwiners (equivariant maps) of the vector representation of $U_q(\mfg)$ yielding trigonometric solutions to the parameter-dependent quantum Yang-Baxter equation. 
In the present paper we compute intertwiners of the vector representation restricted to the subalgebras $B_{\bm c,\bm s}$ when $\mfg$ is of type ${\rm A}^{(1)}_n$, ${\rm B}^{(1)}_n$, ${\rm C}^{(1)}_n$ and ${\rm D}^{(1)}_n$. These intertwiners are matrix solutions to the parameter-dependent quantum reflection equation known as trigonometric reflection matrices. 
They are symmetric up to conjugation by a diagonal matrix and in many cases satisfy a certain sparseness condition: there are at most two nonzero entries in each row and column.  
Conjecturally, this classifies all such solutions in vector spaces carrying this representation.

A group of Hopf algebra automorphisms of $U_q(\mfg)$ acts on these reflection matrices, allowing us to show that each reflection matrix found is equivalent to one with at most two additional free parameters. 
Additional characteristics of the reflection matrices such as eigendecompositions and affinization relations are also obtained. 
The eigendecompositions suggest that for all these matrices there should be a natural interpretation in terms of representations of Hecke-type algebras.
\end{abstract}


\maketitle


\setcounter{tocdepth}{1} 
\tableofcontents


\section{Introduction} \label{Sec:1}


\subsection{The Yang-Baxter equation}

Let $\mfg$ be a derived symmetrizable Kac-Moody algebra defined over an algebraically closed field of characteristic $0$, say $\C$. 
Let $q$ be an indeterminate.
The quantized enveloping algebra $U_q(\mfg)$ of $\mfg$ over $\C(q)$, discovered independently by Drinfeld \cite{Dr1,Dr2} and Jimbo \cite{Ji1} plays an important role in representation theory and the theory of quantum integrable models, see {\it e.g.}~\cite{CRAS,JiMi,KlSg,STS}. 

One of the key properties of $U_q(\mfg)$ is the existence of the universal R-matrix. 
More precisely, one extends $U_q(\mfg)$ by certain Cartan elements to obtain an extended algebra $U_q(\mfg^{\rm ext})$ which is, up to a completion of the tensor product, a quasitriangular Hopf algebra \cite{Dr2,KhTo}.
Its universal R-matrix yields various important relations and structures in the theory of quantum groups. 
Subsequently it can be used to generate different classes of quantum integrable models and explore their properties, see {\it e.g.}~\cite{BGKNR,FrHz,FrRt,GmTL}. 
However to construct a universal R-matrix for a particular choice of $\mfg$ is a rather nontrivial task; as an alternative we highlight the so-called \emph{method of the intertwining equation}, pioneered by Jimbo in \cite{Ji3}.

Let $\K$ be an algebraic closure of $\C(q)$ and let $\RT_u: U_q(\mfg) \to \K^N$ be any finite-dimensional representation of $U_q(\mfg)$ over $\K$ depending rationally on $u\in\K$, called the spectral parameter. 
Furthermore, let $\Delta :  U_q(\mfg) \to  U_q(\mfg) \ot  U_q(\mfg)$ be the coproduct of $ U_q(\mfg)$.
 
It is well-known \cite{Ji2} that, assuming certain irreducibility and regularity conditions, an intertwiner $\hat R(u/v) \in \End(\K^N)^{\ot 2}$ of the tensor product of two such representations,
\eq{
\hat R(\tfrac uv)\, (\RT_u \ot \RT_v)(\Delta(a)) = (\RT_v \ot \RT_u)(\Delta(a))\, \hat R(\tfrac uv) , \label{intro:2}
}
is unique, up to a scalar, and satisfies the quantum Yang-Baxter equation (YBE)
\[ 
(\hat R(u/v) \ot \Id)(\Id \ot \hat R(u))(\hat R(v) \ot \Id) = (\Id \ot \hat R(v))(\hat R(u) \ot \Id)(\Id \ot \hat R(u/v)), 
\]
an identity of $\End(\K^N)^{\ot 3}$-valued rational functions; here $\Id : \K^N \to \K^N$ is the identity operator.
Such solutions $\hat R(u)$ are known as ``trigonometric'' solutions since, roughly speaking, a rational dependence on $u$ can be restated as a trigonometric dependence on $x$ by setting $u={\rm e}^{2 \sqrt{-1}\, x}$.
In this way Jimbo found trigonometric solutions of the YBE when $\mfg$ is an (untwisted or twisted) affine Lie algebra of classical Lie type and $\RT_u$ is the vector (first fundamental) representation of $U_q(\mfg)$.


\subsection{The reflection equation}

The reflection equation (RE), also called the boundary quantum Yang-Baxter equation, is a quaternary analogue of the quantum Yang-Baxter equation and plays a fundamental role in the representation theory of coideal quantum groups \cite{CGM,GoMo,Mo,MoRa,MRS} and quantum integrable systems with open boundary conditions \cite{Sk,Ch1,Ch3}. 
In general it is not straightforward to classify solutions of the RE, called reflection matrices or K-matrices. 
First of all, in contrast to the YBE there are several reflection equations, such as twisted and untwisted left and right reflection equations. 
However, owing to the symmetry properties of the R-matrices we consider in this paper, left and right versions and in many cases also twisted and untwisted versions are related via certain transformation rules. 
Hence it is enough to find a solution for one of them. 
More significantly, the scale of the task of solving a particular reflection equation lies in the fact that given an R-matrix there exists a large number of inequivalent solutions. 
In the theory of quantum integrable systems each K-matrix corresponds to a different choice of boundary conditions that are compatible with the underlying integrability \cite{Mk}.

Many K-matrices for various R-matrices are known. 
The task is much simpler when R-matrices involved in the reflection equation are of rational type. 
In this case they are associated to Yangians and exhibit additional symmetries that do not extend to the trigonometric case in general; thus, classifying solutions of the rational reflection equation is largely well understood \cite{AACDFR1,DMS,GuRg1,MoRa}. 
When R-matrices are of trigonometric type, K-matrices that are solutions of the \emph{constant} reflection equation have been thoroughly studied in \cite{KSS,NoSu}. 
Solutions of the (nonconstant) reflection equation with a particular choice for the R-matrix have been studied by many authors, see {\it e.g.}~\cite{AbRi,BCDR,BCR,BeFo,Ch1,DeGe,DeMk,dVGR,FNR,Gb,JKKKM}. 
In \cite{MLS}, the \emph{densest} solutions to the untwisted reflection equation are found: those with most nonzero entries. 
However it appears that they are not the most \emph{general}: not every possible solution to the untwisted reflection equation can be obtained from these by specializations of their free parameters.
In other words, a unifying framework for classifying trigonometric K-matrices does not exist so far. 
Consequently, also the specialization to the rational case requires further investigation.

In this paper we address the classification problem of trigonometric K-matrices from the point of view of quantum symmetric pair (QSP) algebras. Let us explain what we mean by this. 
Let $\theta : \mfg \to \mfg$ be an involutive Lie algebra automorphism. The pair of algebras $(\mfg,\mfg^\theta)$ is called a symmetric pair. 
Let $\theta$ be of the second kind, that is the standard Borel subalgebra $\mfb^+$ of $\mfg$ is required to satisfy $\dim (\theta(\mfb^+)\cap\mfb^+)<\infty$. 
Then the works \cite{Le1,Le2,Ko1} by Letzter and Kolb provide quantized versions of such symmetric pairs, denoted $B_{\bm c,\bm s}$, which are right coideal subalgebras of $U_q(\mfg)$ depending on parameter tuples $\bm c$ and $\bm s$.

Note that the Kac-Moody algebra $\mfg$, as well as the quantum group $U_q(\mfg)$, can be defined combinatorially, by means of its Drinfeld-Jimbo realization, in terms of a Dynkin diagram. In the same way these fixed-point Lie subalgebras and coideal subalgebras, as well as the involution $\theta$ and its quantum analogon, can be defined in terms of a decoration of this diagram, called a Satake diagram.
Furthermore, only for suitable (combinatorially conditioned) values of the parameters $\bm c$ and $\bm s$, does the algebra $B_{\bm c,\bm s}$ specialize to the universal enveloping algebra of $\mfg^\theta$; in this case $B_{\bm c,\bm s}$ is called a QSP algebra.

By a suitable adjustment to the setting of coideal subalgebras, Jimbo's method of the intertwining equation can be used to find K-matrices. 
More precisely, we will be interested in pairs $(K(u),\eta)$, where $K(u) \in \End(\K^N)$, depending rationally on $u$, is invertible and $\eta \in \K^\times$, satisfying the untwisted boundary intertwining equation
\begin{flalign}
\hspace{40mm} K( u )\, \RT_{\eta u}(b) &= \RT_{\eta/u}(b)\, K( u) && \text{for all  } b \in B_{\bm c,\bm s}, \hspace{25mm} \label{intro:3} \\
\intertext{or the twisted boundary intertwining equation}
\hspace{40mm} K( u )\, \RT_{\eta u}(b) &= \RT^\t_{\eta/u}(S(b))\, K( u) && \text{for all  } b \in B_{\bm c,\bm s}, \hspace{25mm} \label{intro:4}
\end{flalign} 
where $\t$ is the usual transposition of matrices and $S$ is the antipode map. Then, following an argument made in \cite{DeGe,DeMk}, provided the representation $\RT_u\ot \RT_v$ of $U_q(\mfg)$ restricts to an irreducible representation of $B_{\bm c,\bm s}$ and an additional regularity condition is satisfied, $K(u)$ satisfies the appropriate (untwisted or twisted) reflection equation. 
This method of finding trigonometric K-matrices was used, for example, in \cite{DeGe,DeMk,Gb} for affine Toda field theories on the half-line, and in \cite{dLMR,dLR} for integrable open spin-chain models, where certain superalgebra analogues of $B_{\bm c, \bm s}$ were constructed.

Note that the fact that a quantum group has many inequivalent coideal subalgebras is responsible for the multitude of inequivalent K-matrices. 
For the coideal subalgebras considered in this paper this is reflected in the fact that to a Dynkin diagram are associated many different Satake diagrams.
It is also worth noting that the problem of classifying K-matrices is closely related to the problem of constructing universal K-matrices, {\it i.e.}
~analogues of universal R-matrices in the setting of coideal subalgebras. 
Certain universal K-matrices associated with coideal subalgebras $B_{\bm c,\bm s} \subseteq U_q(\mfg)$ were recently shown to exist by Balagovi\'c and Kolb in \cite{BgKo2,Ko2}; also see \cite{BaWa}. 
The constant \mbox{K-matrices} classified in \cite{NDS,NoSu} are, up to scalar factors, images of these universal K-matrices of $U_q(\mfg)$ under its vector representation, when $\mfg$ is of finite type. 
The constant K-matrices can also be obtained as $\lim_{u \to 0} K(u)$ whenever $K(u)$ satisfies \eqref{intro:3} or \eqref{intro:4} for all $b$ in a natural affinization of $B_{\bm c,\bm s}$.
The main drawback of this method is that there is not yet an explicit expression for the universal K-matrix, even if $\mfg$ is of finite type.


\subsection{Summary of present work} 

We study coideal subalgebras $B_{\bm c,\bm s} \subseteq U_q(\mfg)$ defined in terms of more general decorations of Dynkin diagrams called \emph{generalized Satake diagrams}.
These diagrams yield automorphisms $\theta$, which are not always involutive, and Lie subalgebras $\mfk_{\bm c} \subseteq \mfg$ depending on the tuple $\bm c$, which are not always fixed-point subalgebras. However, the subalgebras $\mfk_{\bm c}$ preserve a property of $\mfg^\theta$ (where it is trivially true): the only Cartan elements in $\mfk_{\bm c}$ are those fixed by $\theta$. 
We call this the \emph{intersection property}. 
The coideal subalgebras $B_{\bm c,\bm s}$ enjoy a quantum version of this property; this in turn is essential to showing that these algebras at $q=1$ specialize to the universal enveloping algebra of $\mfk_{\bm c}$. 
We call these $B_{\bm c,\bm s}$ the \emph{quantum pair (QP) algebras}.

The diagrammatic approach allows us to easily construct the algebra $B_{\bm c,\bm s}$ for any generalized Satake diagram. Given a generalized Satake diagram, the corresponding algebra $B_{\bm c,\bm s}$ is parametrized by tuples $\bm c$ and $\bm s$ with entries in $\K^\times$ and $\K$, respectively. 
If $B_{\bm c,\bm s}$ is to be a suitable quantum deformation of $U(\mfk')$, conditions must be imposed on $\bm c$ and $\bm s$. In this paper we will only consider the constraints as derived in \cite[Sec.~5.1]{Ko1}. 

We restrict our attention to the case where $\mfg$ is an untwisted affine Kac Moody algebra of classical Lie type, {\it i.e.}~of type ${\rm A}^{(1)}_n$, ${\rm B}^{(1)}_n$, ${\rm C}^{(1)}_n$ or ${\rm D}^{(1)}_n$, and $\RT_u$ is the vector representation of $U_q(\mfg)$.
The solutions $K(u)$ of the corresponding boundary intertwining equations \eqref{intro:3} or \eqref{intro:4} have elegant representation-theoretic properties.
In particular, with respect to the standard basis of $\K^N$, these matrices are rather sparse; the number of nonzero entries in each row and column of $K(u)$ is at most 4. 
In fact, this number is at most 2 if the algebra $B_{\bm c,\bm s}$ is \emph{quasistandard}.
This amounts to an additional combinatorial condition on the tuple $\bm s$, see Definition \ref{D:quasi}, which fails only for certain coideal subalgebras of $U_q(\wh{\mfso}_N)$ and $U_q(\wh{\mfsp}_N)$. 
It is worth noting that less sparse K-matrices can be obtained by imposing mixed constraints on the tuples $\bm c$ and $\bm s$, namely combinations of those given in \cite[Sec.~5.1]{Ko1} and the so-called $q$-Onsager type constraints studied by Baseilhac and Belliard in \cite[Sec.~2]{BsBe1}. 
We leave the analysis of such cases for a further study.

The main result of this paper is Theorem \ref{T:all-K}, which provides explicit formulas of K-matrices for QP algebras when $\mfg = \wh\mfsl_N$, $\wh\mfso_N$ or $\wh\mfsp_N$. 
We conjecture that this classifies all matrix solutions to the untwisted and twisted reflection equations that have at most two nonzero entries in each row or column and are symmetric (with respect to the main diagonal or counter-diagonal, respectively) up to conjugation by a diagonal matrix. 
We also study supplementary properties of the K-matrices such as unitarity and regularity relations, eigendecompositions and, when $B_{\bm c,\bm s}$ is an affine extension of a coideal subalgebra of the corresponding quantum group of finite type, affinization identities. 
This eigendecomposition allows us to obtain the minimal polynomial of such a K-matrix and leads to connections with the representation theory of Hecke-type algebras.


\subsection{Outline}

\enlargethispage{1em}

This paper is organized as follows. 
In section \ref{sec:alg} we recall the necessary preliminaries regarding Kac-Moody Lie algebras $\mfg$. 
In particular we introduce the notion of a \emph{generalized} Satake diagram and use it to define an automorphism and a subalgebra of $\mfg$. 
The motivation for this is that certain properties of Satake diagrams and the associated symmetric pairs hold for these more general diagrams.
Then in Section \ref{sec:QG} we consider the associated quantized enveloping algebras following \cite{KlSg,KaWa,ChPr1} for the quantum groups and following \cite{Ko1, BgKo2} for the coideal subalgebras.

In Section \ref{sec:natrep} we discuss the vector representation of $U_q(\mfg)$, the evaluation module obtained by affinization (in the homogeneous grading) of the vector (first fundamental) representation of the corresponding Hopf subalgebra of finite type.

In Section \ref{sec:Rmat} we discuss the quantum Yang-Baxter equation and review properties of quantum R-matrices associated with the vector representation. The main references for these sections are \cite{Ji2,FRT} and \cite[Sec.~8.4 and 8.7]{KlSg}. 

In Section \ref{sec:Kmat} we survey properties of twisted and untwisted reflection equations and discuss their solutions: twisted and untwisted reflection matrices. 
In particular, we explain how solutions to the boundary intertwining equations \eqref{intro:3} and \eqref{intro:4} for arbitrary coideal subalgebras yield solutions of the quantum reflection equation. 

In Section \ref{sec:rotdress} we describe the ``rotational'' symmetries which relate K-matrices associated with QP algebras that are equivalent under the Hopf algebra automorphisms corresponding to the symmetries of the Satake diagrams, {\it cf.}~\cite[Def.~9.1]{Ko1}. 
This means we only need to study one representative for each family of equivalent QP algebras. 
In addition, we discuss the ``dressing'' method that allows us to write a K-matrix as being diagonally equivalent to a {\it bare} K-matrix, which contains at most two free parameters; the other free parameters appear in the diagonal change-of-basis matrices. Also dressing is underpinned by an equivalence of QP algebras by particular Hopf algebra automorphisms.

In Section \ref{sec:Satdiagclassification} we extend the classification of affine Satake diagrams in \cite{BBBR} to generalized Satake diagrams (restricting to untwisted affine diagrams of classical type). We introduce diagrammatic terminology and a new notation for the generalized Satake diagrams that are directly in aid of the classification of reflection matrices.

In Section \ref{sec:Results} we present the main results of this paper. In Section \ref{sec:mainresult} we make general statements regarding the obtained K-matrices for QP algebras when $\mfg = \wh\mfsl_N$, $\wh\mfso_N$ or $\wh\mfsp_N$. 
Sections \ref{sec:K:tw}--\ref{sec:K:low-rank} provide details of the computations that were performed with the computer algebra system {\it Wolfram Mathematica}. 

Section \ref{sec:conclusions} contains a brief overview of some further directions of study that are implied by the results of the present paper. 
In Section \ref{sec:qintsys} we review how reflection equations naturally underpin integrability {in quantum integrable systems with reflecting boundary conditions expressed in terms of qKZ equations, transfer matrices as well as Hamiltonians. In Section \ref{sec:RTT} we give a synopsis of twisted quantum loop algebras in the RTT presentation of quantum groups. Such algebras have been constructed in \cite{CGM,MRS}, and their representations were studied in \cite{GoMo,Mo}. In Section~\ref{sec:Hecke} we demonstrate that some of the obtained K-matrices provide representations of baxterized operators in cyclotomic Hecke algebras studied in \cite{IsOg,KuMv}. This part is aimed to motivate a similar study for cyclotomic Birman-Wenzl-Murakami algebras. 

Appendix \ref{App:Satakediagrams} contains the classification of generalized Satake diagrams associated to affine Dynkin diagrams of classical Lie type. In Appendix \ref{App:LowRank} we list isomorphisms of affine generalized Satake diagrams of low rank. Appendix \ref{App:summary} contains a summary of key properties of the K-matrices studied in this paper.
In Appendix \ref{App:qOns} we present an example of a coideal subalgebra of a generalized $q$-Onsager type and its reflection matrix.
\smallskip


\subsection{Notation}

We will use the following conventions throughout the document:
\begin{itemize} [itemsep=0.25em]
\item the symbol \rmkend \ indicates the end of a definition, remark or example;
\item for any proposition $P$, we write
\[ \del_P := \begin{cases} 1 & \text{if } P \text{ is true}, \\ 0 & \text{if } P \text{ is false}; \end{cases} \]
\item we abbreviate $\del_{i=j}$ by the usual Kronecker delta symbol $\del_{ij}$.
\end{itemize}


\subsection{Acknowledgements}

This collaboration was established at the international conference ``New Trends in Quantum Integrability 2014'' held at the University of Surrey in August, 2014. B.V. thanks the organizers for the hospitality. 
The authors are grateful to Stefan Kolb for valuable comments, for pointing out the reference \cite{BBBR} and for sharing the results of \cite{Ko2} ahead of publication. 
The authors also thank Pascal Baseilhac for a helpful discussion involving the publication \cite{BsKz1}. 
They also thank Martina Balagovi\'{c}, Nicolas Guay, Maxim Nazarov, Nicolai Reshetikhin and Jasper Stokman for useful comments and suggestions. 

Part of this work was done during V.R.'s visit to the University of Alberta. 
V.R. thanks the University of Alberta for the hospitality and also gratefully acknowledges the Engineering and Physical Sciences Research Council (EPSRC) of the United Kingdom for the Postdoctoral Fellowship under the grant EP/K031805/1. 
B.V. is thankful for financial support from EPSRC under the grants EP/L000865/1 and EP/N023919/1 and from the Netherlands Organization for Scientific Research (NWO) under a Free Competition grant (``Double affine Hecke algebras, Integrable Models and Enumerative Combinatorics").

Preliminary results of this work have been announced at the mini-workshops ``Coideal Subalgebras of Quantum Groups'' in Oberwolfach in February 2015 and ``Vertex Algebras and Quantum Groups'' in Banff in February 2016. The authors are grateful to the organizers.


\section{Kac-Moody Lie algebras and generalized Satake diagrams} \label{sec:alg}


\subsection{Kac-Moody Lie algebras} \label{sec:KM}

Let $I$ be a finite set and let $A=(a_{ij})_{i,j \in I}$ be a symmetrizable generalized Cartan matrix,
{\it i.e.}~for all $i,j \in I$ we have $a_{ii}=2$, $a_{ij} \in \Z_{\le 0}$ if $i \ne j$ and 
\eq{ 
\label{DAisAD} d_i\,a_{ij} = d_j\,a_{ji}
} 
for some $d_i \in \Q_{>0}$.
For any $X \subseteq I$, consider the submatrix $A_X:=(a_{ij})_{i,j \in X}$.
We will assume that $A$ is indecomposable: for all $\emptyset \ne X \subset I$ we have $A \ne A_X \oplus A_{I \backslash X}$ (up to any permutation of $I$), in other words there exists $(i,j) \in X \times I \backslash X$ such that $a_{ij} \ne 0$.
For most of Sections \ref{sec:alg} and \ref{sec:QG}, we will consider general $A$, but for the remainder of this paper we will assume that $A$ is of \emph{finite type}, {\it i.e.}~$\det(A_X) \ne 0$ for all $X \subseteq I$, or of \emph{affine type}, {\it i.e.}~$\det(A)=0$ and $\det(A_X) \ne 0$ for all $X \subset I$.

Consider a complex vector space $\mfh$ with basis $\{h_i\}_{i\in I}$.
It can be extended to a Lie algebra in which $\mfh$ appears as a Cartan subalgebra: 
\begin{defn}
The \emph{derived Kac-Moody Lie algebra} associated to $A$ is the Lie algebra $\mfg 
= \langle e_i,f_i,h_i \rangle_{i \in I}$, subject to the following relations for $i,j \in I$:
\begin{gather}
\label{mfg:rels1} [h_i,h_j]=0, \qq [h_i,e_j]=a_{ij}e_j, \qq [h_i,f_j]=-a_{ij}f_j, \qq [e_i,f_j]=\del_{ij} h_i, \\
\label{mfg:Serre} \ad(e_i)^{1-a_{ij}}(e_j) = \ad(f_i)^{1-a_{ij}}(f_j) = 0 \qq \text{if } i \ne j,
\end{gather}
where 
\eq{ \label{mfg:adj} \ad(x)(y):=[x,y].} 
We denote the upper and lower nilpotent subalgebras by $\mfn^+ := \langle e_i \rangle_{i \in I}$ and $\mfn^- := \langle f_i \rangle_{i \in I}$, respectively, and the upper and lower Borel subalgebras by $\mfb^\pm := \langle \mfh, \mfn^\pm \rangle$.
\hfill \defnend
\end{defn}

We briefly review some additional properties of $\mfg$, see {\it e.g.}~\cite{Ca} for more detail. 
Choose simple roots $\{\al_j \}_{j \in I}$, {\it i.e.}~a linearly independent subset of $\mfh^*$ satisfying $\al_j(h_i) = a_{ij}$ for all $i,j \in I$. 
There exists a symmetric bilinear form $(\cdot,\cdot)$ on $\mfh$ such that $(h_i,h_j)= d_j^{-1}a_{ij}$ for $i,j \in I$.
The induced bilinear form on $\mfh^*$ is given by $(\al_i,\al_j)= d_i a_{ij}$ for $i,j \in I$ so that $d_i=(\al_i,\al_i)/2$. 

The Weyl group $W \subset GL(\mfh)$ is generated by the simple reflections $\{ r_i \}_{i \in I}$; they act on $\mfh$ and dually on $\mfh^*$ according to
\eq{ \label{simplereflections}
r_j(h_i) = h_i - a_{ij} h_j, \qq r_i(\al_j) = \al_j - a_{ij} \al_i \qq \text{for all } i,j \in I;
}
in particular, the bilinear form $(\cdot,\cdot)$ is $W$-invariant.
It follows that $W$ is a Coxeter group, {\it i.e.}~$r_i^2=1$ for all $i \in I$ and the \emph{braid relations associated to $A$} are satisfied:
\eq{ \label{braidrelations}
\underbrace{r_i r_j r_i \cdots}_{m_{ij} \text{ factors}} = 
\underbrace{r_j r_i r_j \cdots}_{m_{ij} \text{ factors}} 
\qq \text{for } i,j \in I \text{ such that } i \ne j,
}
where $m_{ij}=2,3,4,6$ and $\infty$ if $a_{ij}a_{ji}=0,1,2,3$ and $\ge4$, respectively. 

As an $\mfh$-module, $\mfg$ has the triangular decomposition
\eq{ \label{mfg:triangulardecomposition}
\mfg = \mfn^+ \oplus \mfh \oplus \mfn^-
}
and, more precisely, is a direct sum of root spaces:
\eq{ \label{mfg:rootspacedecomposition} 
\mfg = \bigoplus_{\mu \in \mfh^*} \mfg_\mu, \qq \text{where } \mfg_\mu = \{ x \in \mfg \, | \, \forall i \in I: \, [h_i,x]=\mu(h_i)x \} \qu \text {for } \mu \in \mfh^*,
}
with $\pi_\mu: \mfg \to \mfg_\mu$ the corresponding projection.
Note that $\mfg_0 = \mfh$. 
Let $Q = \sum_{i \in I} \Z \al_i \subset \mfh^*$ be the root lattice.
Then \eqref{mfg:rootspacedecomposition} induces a $Q$-grading: 
\eq{ \label{mfg:grading}
[\mfg_\mu,\mfg_\nu] \subseteq \mfg_{\mu+\nu} \qq \text{for all } \mu,\nu \in Q. 
}
Furthermore the root system $\Phi = \{ \mu \in \mfh^* \backslash \{ 0 \} \, | \, \mfg_\mu \ne \{ 0 \} \}$ is contained in $Q^+ \cup (-Q^+)$, where $Q^+:=\sum_{i \in I} \Z_{\ge 0} \al_i$.
We have $r_i(\Phi)=\Phi$ for all $i \in I$; call $\beta \in \Phi$ \emph{real} if $\beta = w(\al_i)$ for some $w \in W$, $i \in I$. 
If $\beta \in \Phi$ is real then $\dim(\mfg_{\beta})=1$, $(\beta,\beta)>0$ and $m \beta \in \Phi$ for $m \in \Z$ implies $m \in \{\pm 1 \}$.

\medskip

We will consider various Lie algebra automorphisms on $\mfg$, closely following \cite[Sec.~2.1-2.3]{Ko1}.
First of all, denote by $\om$ the Chevalley involution on $\mfg$:
\eq{ \label{def:Chev} \om(e_i) = -f_i, \qq \om(f_i) = -e_i, \qq \om(h_i)=-h_i \qq \text{for all } i \in I. }
Denote by $\Aut(A)$ the group of \emph{diagram automorphisms} of $A$:
\[ \Aut(A) = \{ \si: I \to I \text{ bijective} \, | \, \forall i,j \in I: a_{ij} = a_{\si(i)\si(j)} \}. \]
Note that $\Aut(A) < \Aut(\mfg)$ by setting, for arbitrary $\si \in \Aut(A)$, 
\[ \si(e_i) = e_{\si(i)}, \qq \si(f_i) = f_{\si(i)}, \qq \si(h_i) = h_{\si(i)} \qq \text{for all } i \in I. \] 
The induced map on $\mfh^*$, also denoted $\si$, satisfies $\si(\al_i) = \al_{\si(i)}$ for all $i \in I$; note that $\si(Q)=Q$.

Consider the Kac-Moody group $G$ associated to $A$. The adjoint action of $\mfg$ on itself corresponds to the group homomorphism $\Ad: G \to \Aut(\mfg)$ (see \cite[1.3]{KaWa} for more details):
\[ 
\Ad(\exp(x)) = \exp(\ad(x)) \qquad \text{for all } x \in \mfg. 
\]
This allows us to extend the action of the Weyl group $W$ on $\mfh$ to an action on $\mfg$ which preserves Lie brackets.
For $i \in I$ define $m_i \in G$ by
\eq{ \label{def:mX} m_i = \exp(e_i)\exp(-f_i)\exp(e_i) }
and consider $\Ad(m_i) \in \Aut(\mfg)$ (note that $\ad(e_i)$ and $\ad(-f_i)$ act locally nilpotently on $\mfg$ for all $i \in I$).
According to \cite[\S 3.8]{Ka}, $\Ad(m_i)(\mfg_{\mu}) = \mfg_{r_i(\mu)}$ and $\Ad(m_i)|_{\mfh} = r_i$.
Crucially, the $m_i$ satisfy the braid relations \eqref{braidrelations} associated to $A$.
Hence given $w \in W$ with reduced decomposition $w=r_{i_1}r_{i_2}\cdots r_{i_\ell}$ with $i_1,\ldots,i_\ell \in I$, $m_w := m_{i_1} m_{i_2} \cdots m_{i_\ell} \in G$ is well-defined.


\subsection{Generalized Satake diagrams} \label{sec:Satdiag}

Given further combinatorial data, namely a subset $X \subseteq I$ and $\tau \in \Aut(A)$ satisfying certain additional properties, there is a canonical way to construct an involutive Lie algebra automorphism of $\mfg$ and an associated fixed point Lie subalgebra of $\mfg$, see \cite[Sec. 2]{Ko1} and references therein.
By relaxing some conditions on $(X,\tau)$, essentially the same construction produces a Lie algebra automorphism $\theta$, not necessarily involutive, and a Lie subalgebra $\mfk$ whose elements are not all fixed by $\theta$; nevertheless these objects retain crucial properties.

It is customary to associate to $(I,A)$ a Dynkin diagram as follows:
\begin{itemize} [itemsep=0.25em]
\item Identify the set of vertices (nodes) with $I$.
\item Between nodes $i$ and $j$, draw an edge of multiplicity $\max(-a_{ij},-a_{ji})$ provided $a_{ij}a_{ji} \le 4$ and draw an edge with the label $(-a_{ij},-a_{ji})$ otherwise. 
\item Furthermore, if $d_i>d_j$ orient the edge towards $j$.
\end{itemize}
Then pairs $(X,\tau)$ with $X \subseteq I$ and $\tau \in \Aut(A)$ such that $\tau^2 = \id$ will be referred to as \emph{decorated diagrams} with the decorations applied to the Dynkin diagram associated to $(I,A)$ as follows:
\begin{itemize} [itemsep=0.25em]
\item Fill all nodes labelled by elements of $X$.
\item For all nontrivial orbits $\{ i , \tau(i) \}$ (i.e. $i \ne \tau(i)$) a bidirectional arrow is drawn between the nodes $i$ and $\tau(i)$.
\end{itemize}

\begin{exam}
Let $A$ be the Cartan matrix of type ${\rm B}^{(1)}_3$ with $I=\{ 0,1,2,3 \}$ (assuming the standard labelling, {\it i.e.}~$a_{02}=a_{20}=a_{12}=a_{21}=a_{23}=-1$, $a_{32}=-2$ and $a_{ij}=0$ for all other values $i,j \in I$ with $i \ne j$). 
Then the pair $(X,\tau) = (\{ 2 \},(01))$ corresponds to the decorated Dynkin diagram

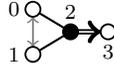
\begin{center}\begin{tikzpicture}[line width=0.7pt,scale=1]
\draw[thick] (-.5,.3) -- (0,0) -- (-.5,-.3);
\draw[double,->] (0,0) -- (.4,0);
\draw[<->,gray] (-.5,.2) -- (-.5,-.2);
\filldraw[fill=white] (-.5,.3) circle (.1) node[left=1pt]{\scriptsize $0$};
\filldraw[fill=white] (-.5,-.3) circle (.1) node[left=1pt]{\scriptsize $1$};
\filldraw[fill=black] (0,0) circle (.1) node[above=1pt]{\scriptsize $2$};
\filldraw[fill=white] (.5,0) circle (.1) node[below=1pt]{\scriptsize $3$};
\end{tikzpicture}\end{center} 

\noindent This diagram does not satisfy the condition \cite[Def. 2.3 (3)]{Ko1} so does not correspond to an involutive Lie algebra automorphism of $\wh\mfso_7$. \hfill \examend
\end{exam}

Let $X \subseteq I$. 
Recall the notation $A_X$ and in the same vein denote
\[
\mfg_X = \langle e_i,f_i,h_i \rangle_{ i \in X}, \qq \mfh_X = \langle h_i \rangle_{i \in X}, \qq W_X = \langle r_i \rangle_{i \in X}, \qq Q_X = \sum_{i\in X} \Z \al_i.
\]
Furthermore, we introduce some terminology that is natural from the interpretation of $(I,A)$ as a Dynkin diagram.
For $i,j \in X$, a \emph{path} from $i$ to $j$ in $X$ is a tuple $\bm k \in X^{D+1}$ for some $D \in \Z_{\ge 0}$, called the \emph{length} of $\bm k$, such that $k_1=i$, $k_{D+1}=j$ and for all $1 \le r \le D$ we have $a_{k_r k_{r+1}} \ne 0$.
We call $X$ \emph{connected} if it is nonempty and for all $i,j \in X$ there exists a path from $i$ to $j$ in $X$.

A \emph{component} of $X$ is a subset $Y \subseteq X$ such that $a_{ij} =0$ for all $i \in Y$, $j \in X\backslash Y$ (note that components may be disconnected or empty).
A \emph{decomposition} of $X$ is a collection $\{X_1,\ldots,X_k \}$ of disjoint components of $X$ for some $k \in \Z_{\geq 0}$ such that $X = \cup_{t=1}^k X_t$; note that components may be disconnected or empty, which can be useful in decompositions.

If $A_X$ is of finite type and $\{X_1,\ldots,X_k\}$ is a decomposition of $X$, then 
\[ 
w_X  = \prod_{t=1}^k w_{X_t} \in W_X, \qquad \rho^\vee_X = \sum_{t=1}^j \rho^\vee_{X_t} \in \mfh_X 
\]
where $w_X$ is the longest element in $W_X$ and $\rho^\vee_X$ is the half-sum of positive coroots (note that $w_{X_s}$ commutes with $w_{X_t}$ as elements of $W_X$, for all $1 \leq s,t \leq m$).
Also, for $j \in I \backslash X$, denote 
\eq{ \label{def:X(j)} X(j) = \bigcup_{t=1 \atop \exists i \in X_t \, a_{ij} \ne 0}^k X_t, } 
{\it i.e.}~$X(j)$ is the smallest component of $X$ containing all $i \in X$ such that $a_{ij} \ne 0$ or, equivalently, $X(j)$ is the largest component of $X$ such that $X(j) \cup \{j \}$ is connected.
Note that $X(j)=\emptyset$ if $a_{ij}=0$ for all $i \in X$.

\begin{lemma} \label{lem:Xjdecomposition}
Let $X \subseteq I$ with $A_X$ of finite type and let $j \in I \backslash X$.
For all $i \in X(j)$ there exist $v_i(\al_j) \in \Z_{>0}$ such that
\eq{ \label{Xjdecomposition} w_X(\al_j)=\al_j+\sum_{i \in X(j)} v_i(\al_j) \al_i.}
\end{lemma}

\begin{proof}
As a consequence of \eqref{simplereflections}, \eqref{Xjdecomposition} holds with $v_i(\al_j) \in \Z_{\ge 0}$. It remains to prove that for all $i \in X(j)$, $v_i(\al_j) > 0$. If $X(j)=\emptyset$ then $w_X(\al_j)=\al_j$ and we are done.
For all $i \in X(j)$ we have
\[ \al_j(w_X(h_i)) = w_X(\al_j)(h_i) = \al_j(h_i) + \sum_{i' \in X(j)} v_{i'}(\al_j) \al_{i'}(h_i). \]
Because $w_X$ sends the fundamental Weyl alcove in $\mfh_X$ to minus itself, {\it i.e.}~$-w_X$ permutes the sets $\{ \al_i \}_{i \in X}$ and $\{ h_i \}_{i \in X}$, we have $w_X(h_i)=-h_{\si_X(i)}$ for some bijection $\si_X:X \to X$.
It follows that
\eq{ \label{Xjeqn}  -a_{\si_X(i) \, j} - a_{ij} - \sum_{i' \in X(j)\backslash \{i\}} v_{i'}(\al_j) a_{i \, i'} = 2 v_i(\al_j),}
with each term nonnegative.
Define the distance $D_i$ of the node $i \in X(j)$ to the node $j$ to be the minimum of the lengths of all paths from $i$ to $j$ in $X(j) \cup \{j \}$.
By definition of $X(j)$ the $D_i$ are well-defined positive integers and there exists $D_{\rm max} \in \{ 1,2,\ldots,|I|-1\}$ such that $D_i \le D_{\rm max}$ for all $i \in X(j)$ and, conversely, for all $D \in \Z_{>0}$ there exists $i \in X(j)$ with $D_i=D$ precisely if $D \le D_{\rm max}$.

We now proceed by induction with respect to $D_i$.
The initial case is when $D_i=1$, {\it i.e.}~$a_{ij} \ne 0$.
Then the left-hand side of \eqref{Xjeqn} is positive as it contains the positive term $-a_{ij}$, yielding $v_i(\al_j)>0$ as required.
Suppose now that the statement is true for all $i \in X(j)$ with $D_i<D$ for some $D \in \Z_{>0}$. 
If there are no $i \in X(j)$ with $D_i=D$, then $D>D_{\rm max}$ and we are done.
Hence assume that $D_i=D$ for some $i \in X(j)$; then there exists $i' \in X(j)$ such that $a_{i'i}\ne 0$, $D_{i'}=D-1$ and, by the induction hypothesis, $v_{i'}(\al_j)>0$.
In this case the left-hand side of \eqref{Xjeqn} is positive as it contains the positive term $-v_{i'}(\al_j)a_{i \, i'}$, so that $v_i(\al_j)>0$ as required.
\end{proof}

We will only be interested in pairs $(X,\tau)$ where $X$ and $\tau$ satisfy certain compatibility criteria.
If $A_X$ is of finite type, $-w_X$ is an involution of the set $\{ \al_i \}_{i \in X}$. 
A useful compatibility criterion on the pair $(X,\tau)$ is that $-w_X$ and $\tau$ coincide on $Q_X$, i.e.
\eq{ \label{Satdiag1a} 
-w_X(\al_i) = \al_{\tau(i)} \qq \text{for all } i \in X.
}
Note that \eqref{Satdiag1a} implies that $\tau(X)=X$ and furthermore $\tau(Y)=Y$ if $Y$ is a component of $X$.

\begin{defn}
Let $X \subseteq I$ with $A_X$ of finite type and let $\tau\in \Aut(A)$ be an involution such that \eqref{Satdiag1a} is satisfied.
The \emph{Lie algebra automorphism associated to $(X,\tau)$} is 
\begin{flalign} \label{def:theta} 
&& \theta = \theta(X,\tau) := \Ad(m_{w_X}) \, \tau \, \omega \in \Aut(\mfg). && \defnend
\end{flalign}
\end{defn}

Hence $\theta$ acts on $\mfh$ as the involution $-w_X \tau$.
The dual involution on $\mfh^*$ also acts as $-w_X \tau$ and we also denote it by $\theta$.
We collect several basic properties of $\theta$ for later use.

\begin{prop} \label{prop:theta:properties}
Let $X \subseteq I$ with $A_X$ of finite type and let $\tau\in \Aut(A)$ be an involution such that \eqref{Satdiag1a} is satisfied.
Then 
\eqa{ 
\label{theta:gX} && \theta(x) &= x && \text{for all } x \in \mfg_X, \\
\label{theta:rootspace} && \theta(\mfg_\mu) &= \mfg_{\theta(\mu)} && \text{for all } \mu \in Q, \\
\label{theta:squared} && \theta^2(x) &= (-1)^{2\mu(\rho^\vee_X)}x \hspace{-10mm} && \text{for all } x \in \mfg_\mu, \\
\label{htheta:decomposition} && \mfh^\theta &= \bigoplus_{i \in X} \C h_i \; \oplus && \hspace{-28mm} \bigoplus_{j \in I^* \atop j \ne \tau(j)} \C (h_j - h_{\tau(j)}), 
}
where $I^*$ is any subset of $I \backslash X$ that intersects each $\tau$-orbit in a singleton. 
\end{prop}

\begin{proof}
Note that \eqref{theta:gX} is \cite[Lem. 4.9]{BBBR}. 
Secondly, \eqref{theta:rootspace} follows immediately from the corresponding properties of the component Lie algebra automorphisms, in particular $\Ad(m_{w_X})(\mfg_\mu) = \mfg_{w_X(\mu)}$ for all $\mu \in Q$.
Next, for \eqref{theta:squared} note that $\Ad(m_{w_X})$ commutes with $\tau$ and with $\omega$ (see \cite[Prop. 2.2 (3)]{Ko1}), so that $\theta^2 = \Ad(m_{w_X}^2)$.
Now \cite[Cor. 4.10.3]{BBBR} implies the statement.

We note that, by duality, \eqref{htheta:decomposition} is equivalent to 
\eq{ \label{equivalence1}
\theta(\mu) = \mu \qu \Longleftrightarrow \qu \exists \bm m \in \Z^I: \, \mu = \sum_{i \in X} m_i \al_i + \sum_{j \in I^* \atop j \ne \tau(j)} m_j (\al_j -  \al_{\tau(j)}) 
}
for all $\mu \in Q$.
Consider $A_i(X,\tau) := \al_i - \al_{\tau(i)} + w_X(\al_{\tau(i)}) - w_X(\al_i) \in Q_X$ for $i \in I$.
By applying $\theta$ we obtain 
\eq{ \label{A0}
A_i(X,\tau)=0
} 
(also see \cite[Lem.~3.1]{BgKo1}).
Hence from \eqref{Satdiag1a} we infer that, for all $\mu \in Q$,
\eq{ \label{equivalence2}
\theta(\mu)=\mu \qu \Longleftrightarrow \qu \sum_{j \in I^* \atop j=\tau(j)} m_j(\al_j + w_X(\al_j)) + \sum_{j \in I^* \atop j \ne \tau(j)} (m_j+m_{\tau(j)})(\al_j + w_X(\al_{\tau(j)})) = 0. 
}
Note that for $j \in I \backslash X$, $w_X(\al_j) - \al_j \in Q_X$ so that the right-hand side of \eqref{equivalence2} implies 
\[ \sum_{j \in I^* \atop j=\tau(j)} 2m_j\al_j + \sum_{j \in I^* \atop j \ne \tau(j)} (m_j+m_{\tau(j)})(\al_j + \al_{\tau(j)}) \in Q_X \]
and we obtain the right-hand side of \eqref{equivalence1}. 
The implication $\Leftarrow$ in \eqref{equivalence1} follows from a straightforward calculation using \eqref{Satdiag1a} and \eqref{A0}.  \hfill \qedhere

\end{proof}

\begin{defn}
Let $X \subseteq I$ with $A_X$ of finite type and let $\tau\in \Aut(A)$ be an involution such that \eqref{Satdiag1a} is satisfied.
Furthermore, let $\bm c \in (\C^\times)^{I \backslash X}$ and denote 
\eq{ \label{def:g_i} 
g_i := \begin{cases} f_i & \text{if } i \in X, \\  f_i + c_i \theta(f_i) & \text{otherwise} \end{cases}
}
and
\eq{ \label{def:nplusX}
\mfn^+_X := \mfn^+ \cap \mfg_X = \langle e_i \rangle_{i \in X}.
}
The \emph{Lie algebra associated to the triple $(X,\tau,\bm c)$} is
\begin{flalign}
\label{def:kc} 
&& \mfk_{\bm c} = \mfk_{\bm c}(X,\tau) = \langle \mfn^+_X, \mfh^\theta, \{ g_i \}_{i \in I} \rangle \subseteq \mfg. && \defnend
\end{flalign}
\end{defn}

Now we recall the notion of a Satake diagram (also called ``admissible pair'') for symmetrizable Kac-Moody algebras as per \cite[Defn.~2.3]{Ko1} (but compare {\it e.g.}~\cite[Def. 4.10 (b)]{BBBR}).

\begin{defn} \label{def:Satdiag}
Let $X \subseteq I$ with $A_X$ of finite type and let $\tau\in \Aut(A)$ be an involution such that \eqref{Satdiag1a} is satisfied.
The pair $(X,\tau)$ is called a \emph{Satake diagram (associated to $A$)} if
\eq{ \label{Satdiag2a} 
\tau(j)=j \; \implies \; \alpha_j(\rho^\vee_X) \in \Z \qq \text{for all } j \in I \backslash X
}
where $\rho^\vee_X$ is the half-sum of positive coroots associated to $\mfg_X$.
The class of all Satake diagrams associated to $A$ is denoted $\Sat(A)$. \hfill \defnend
\end{defn}

For a complete listing of Satake diagrams associated to affine Cartan matrices (both twisted and untwisted, both of classical and exceptional Lie type), see \cite[Sec.~6]{BBBR}.
If $(X,\tau)$ is a Satake diagram, then the automorphism $\theta$ can be made into an involution and the Lie subalgebra $\mfk_{\bm c}$ into a fixed point subalgebra of this involution, as will become clear in the proof of the following proposition. 
More relevant to us is the conclusion drawn with respect to the Cartan elements contained in the algebra $\mfk_{\bm c}$.

\begin{prop} \label{prop:kc:intersection1}
If $(X,\tau) \in \Sat(A)$ then there exists $\bm c \in (\C^\times)^{I \backslash X}$ such that the \emph{intersection property} holds:
\eq{ \label{intersection} 
\mfk_{\bm c} \cap \mfh = \mfh^\theta.
}
\end{prop}

\begin{proof}
Fix $(X,\tau) \in \Sat(A)$. 
Consider $\wt H:=\Hom(Q,\C^\times)$, the group of characters of $Q$ over $\C$. 
For $\chi \in \wt H$, define $\Ad(\chi) \in \Aut(\mfg)$ by
\[ 
\Ad(\chi)(a) = \chi(\mu)(a) \qq \text{for all } a \in \mfg_\mu, \, \mu \in Q. 
\]
As $X=\tau(X)$ we have $\al_j(2 \rho^\vee_X) = \al_{\tau(j)}(2 \rho^\vee_X) \in \Z$ for all $j \in I \backslash X$.
Choose $\chi \in \wt H$ such that
\begin{alignat*}{90}
\frac{\chi(\al_{\tau(j)})}{\chi(\al_j)} &= (-1)^{\al_j(2 \rho^\vee_X)} \qq && \text{if } j \notin X \text{ and } j \ne \tau(j), \\
\chi(\al_j) &= 1 && \text{otherwise,}
\end{alignat*}
see \cite[Eqn. (2.7)]{Ko1} and \cite[(5.1-5.2)]{BgKo2}.
Then $\theta':= \Ad(\chi) \theta$ is an involution as per \cite[Thm. 2.5]{Ko1}. 
Defining $\bm c \in (\C^\times)^{I \backslash X}$ by $c_j = \chi(\theta(\al_j))^{-1}$, we have $c_j \theta(f_j) = \theta'(f_j)$ for all $j \in I \backslash X$ and hence $\mfk_{\bm c}$ equals the fixed-point Lie subalgebra $\mfg^{\theta'}$.
Moreover, owing to \cite[Lem. 2.8]{Ko1}, $\mfg^{\theta'} \cap \mfh = \mfh^{\theta'}$.
Noting that $\Ad(\chi)$ fixes $\mfh$ pointwise, we derive \eqref{intersection}. 
\end{proof}

\begin{rmk}
A Lie algebra automorphism $\psi \in \Aut(\mfg)$ is said to be \emph{of the second kind} if ${\rm dim}(\psi(\mfb^+) \cap \mfb^+)<\infty$.
The automorphism $\theta$ defined by \eqref{def:theta} is of the second kind, since $A_X$ is of finite type.
Hence, if $\psi$ equals the involution $\theta'$ appearing in the proof of Proposition \ref{prop:kc:intersection1}, then $\psi$ is of the second kind and $(\mfg,\mfg^{\psi})$ is called the \emph{symmetric pair associated to $(X,\tau)$}.
All involutive automorphisms of $\mfg$ of the second kind are related to a Satake diagram in this way, see \cite[Thm.~{2.7}]{Ko1}. 
\hfill \rmkend
\end{rmk}

Let us consider some basic examples of pairs $(X,\tau)$ satisfying \eqref{Satdiag1a} and study the Lie algebra automorphism $\theta$ and the Lie subalgebra $\mfk_{\bm c}$.
 
\begin{exam} \label{ex:Satdiags}
Let $A$ be of finite type and rank 2. 
Write $I=\{1,2\}$ for the labelling set. 
Choose $\tau=\id$ and $X=\{2\}$.
Then $-w_X=-r_2$ fixes $\al_2$ so that \eqref{Satdiag1a} is satisfied.  
Also, $\theta = \Ad(r_2) \omega$ (in particular, $\mfh^\theta = \C h_2$) and $\mfk_{\bm c}$ is generated by $e_2$, $g_2 = f_2$ and $g_1 = f_1+c_1 \theta(f_1)$. 

\begin{enumerate}


\item
Suppose $A$ is of type ${\rm B}_2$, {\it i.e.}~$\mfg = \mfso_5$.
The Satake diagram is
$\begin{tikzpicture}[baseline=-0.25em,line width=0.7pt,scale=1]
\draw[double,->] (0,0) -- (0.4,0);
\filldraw[fill=white] (0,0) circle (.1) node[above=1pt]{\footnotesize 1};
\filldraw[fill=black] (.5,0) circle (.1) node[above=1pt]{\footnotesize 2};
\end{tikzpicture}$.
Note that $\al_1(\rho^\vee_X)= \frac{1}{2}\al_1(h_2) = \frac{1}{2}a_{21} = -1$.
Hence $(X,\tau) \in \Sat(A)$.
According to Proposition \ref{prop:kc:intersection1}, for $c_1= \chi(\theta(\al_1))^{-1}=1$ we have \eqref{intersection}. 
In fact, a direct calculation establishes that \eqref{intersection} holds for all $c_1 \in \C^\times$.
Namely, note that
\begin{align*} 
\theta(f_1) &= -\exp(\ad(e_2)) \exp(\ad(-f_2)) \exp(\ad(e_2))(e_1) \\
&= -\exp(\ad(e_2)) \exp(\ad(-f_2))(e_1+[e_2,e_1]+\tfrac{1}{2}[e_2,[e_2,e_1]]) \\
&= -\tfrac{1}{2}\exp(\ad(e_2))([e_2,[e_2,e_1]]) \\
&= -\tfrac{1}{2}[e_2,[e_2,e_1]]
\end{align*}
so that $g_1=f_1-\tfrac{c_1}{2}[e_2,[e_2,e_1]]$.
Defining
\[
g_{(2,1)} := [g_2,g_1] = [f_2,f_1] - c_1 [e_2,e_1], \qq g_{(2,2,1)} := [g_2,g_{(2,1)}] = [f_2,[f_2,f_1]]- 2c_1e_1
\]
we have $[g_{(2,2,1)},f_2]=0$ and 
\begin{align*}
[g_1,h_2] &= -2g_1, & [g_{(2,1)},h_2] &= 0, & [g_{(2,2,1)},h_2] &= 2g_{(2,2,1)}, \\
[g_1,e_2] &= 0, & [g_{(2,1)},e_2] &= -2g_1, & [g_{(2,2,1)},e_2] &= -2g_{(2,1)}, \\
[g_1,g_{(2,1)}] &= 2c_1 e_2, & [g_1,g_{(2,2,1)}] &= -2c_1 h_2, & [g_{(2,1)},g_{(2,2,1)}] &= -4c_1 g_2. 
\end{align*}
Hence 
\[ \mfk_{\bm c} = \C e_2 \oplus \C h_2 \oplus \C g_2 \oplus \C g_1 \oplus \C g_{(2,1)} \oplus \C g_{(2,2,1)} \cong \mfso_4 \]
so that $\mfk_{\bm c} \cap \mfh = \mfh^\theta$ for all $c_1 \in \C^\times$. 


\item
If $A$ is of type ${\rm C}_2$, {\it i.e.}~$\mfg = \mfsp_4$. The Satake diagram is
$\begin{tikzpicture}[baseline=-0.25em,line width=0.7pt,scale=1]
\draw[double,<-] (0.1,0) -- (0.5,0);
\filldraw[fill=white] (0,0) circle (.1) node[above=1pt]{\footnotesize1};
\filldraw[fill=black] (.5,0) circle (.1) node[above=1pt]{\footnotesize2};
\end{tikzpicture}$.
Now $\al_1(\rho^\vee_X)= \frac{1}{2}\al_1(h_2) = \frac{1}{2}a_{21} = -\frac{1}{2}$, so $(X,\tau) \notin \Sat(A)$.
Nevertheless, we have
\eqn{
\theta(f_1) &= -\exp(\ad(e_2))\exp(\ad(-f_2))(e_1+[e_2,e_1]) \\ &= -\exp(\ad(e_2))([e_2,e_1]) = [e_1,e_2]
}
so that $g_1=f_1+c_1[e_1,e_2]$.
Define $g_{(1,2)}:=[g_1,g_2]=[f_1,f_2]+ c_1 e_1$. Then $[g_{(1,2)},g_2]=0$ and
\[ 
\qq [g_1,h_2] = -g_1, \qq [g_{(1,2)},h_2] = g_{(1,2)}, \qq [g_1,e_2] = 0, \qq [g_{(1,2)},e_2] = -[g_1,h_2] = g_1. 
\]
Also define
\[ 
g_{(1,1,2)} := [g_1,g_{(1,2)}] = [f_1,[f_1,f_2]]  -2 c_1 (h_1+h_2) - c_1^2 [e_1,[e_1,e_2]].
\]
Straightforward calculations show that this element is central in $\mfk_{\bm c}$. Hence
\[ \mfk_{\bm c} = \C e_2 \oplus \C h_2 \oplus \C g_2 \oplus \C g_1 \oplus \C g_{(1,2)} \oplus \C g_{(1,1,2)} \]
so that, again, for all $c_1 \in \C^\times$ we have $\mfk_{\bm c} \cap \mfh = \mfh^\theta$. 
Hence, the subalgebra $\mfk_{\bm c}$ has a nontrivial Levi decomposition in terms of the simple Lie algebra $\langle e_2,g_2,h_2 \rangle \cong \mfsl_2$ and the solvable Lie algebra $\langle g_1, g_{(1,2)}, g_{(1,1,2)} \rangle$ (isomorphic to the three-dimensional Heisenberg algebra).


\item
Finally, if $A$ is of type ${\rm A}_2$, {\it i.e.}~$\mfg = \mfsl_3$. the Satake diagram is
$\begin{tikzpicture}[baseline=-0.25em,line width=0.7pt,scale=1]
\draw[thick] (0,0) -- (0.5,0);
\filldraw[fill=white] (0,0) circle (.1) node[above=1pt]{\footnotesize1};
\filldraw[fill=black] (.5,0) circle (.1) node[above=1pt]{\footnotesize2};
\end{tikzpicture}$.
As in the previous example, $(X,\tau) \notin \Sat(A)$.
Also, $\theta(f_1) = [e_1,e_2]$ so that $g_1=f_1+c_1[e_1,e_2]$.
Defining $g_{(1,2)}=[g_1,g_2]=[f_1,f_2]+c_1 e_1$ we obtain
\[ [g_1,g_{(1,2)}] = [f_1+c_1[e_1,e_2],[f_1,f_2]+c_1 e_1] = -c_1(2h_1 + h_2) \in \mfk_{\bm c} \cap \mfh, \]
but $\theta([g_1,g_{(1,2)}])=-[g_1,g_{(1,2)}]$, so that the intersection condition \eqref{intersection} fails for all values of $\bm c \in (\C^\times)^{I \backslash X}$. 
In fact, since $c_1\ne0$ a straightforward computation yields $\mf{k}_{\bm c} = \mfsl_3$. \hfill \examend
\end{enumerate}
\end{exam}

Clearly, Example \ref{ex:Satdiags} (iii) can be generalized to the following statement. 
Suppose $X \subseteq I$ with $A_X$ is of finite type and $\tau\in \Aut(A)$ is an involution such that \eqref{Satdiag1a} and \eqref{intersection} are satisfied. 
Then there are no $(i,j) \in X \times I \backslash X$ such that $\tau(j)=j$, $X(j)= \{ i \}$ and $a_{ij}=a_{ji}=-1$. 
In other words,
\eq{ \label{Satdiag2b}
\theta(\al_j) = -(\al_i+\al_j) \qu \implies \qu a_{ji} \ne -1 \qq \text{for all } (i,j) \in X \times I \backslash X.
}
Here we have used Lemma \ref{lem:Xjdecomposition} to rewrite the condition on $(i,j)$ in a more compact way.
This leads to the following definition.

\begin{defn} \label{def:genSatdiag}
Let $X \subseteq I$ with $A_X$ of finite type and let $\tau\in \Aut(A)$ be an involution such that \eqref{Satdiag1a} is satisfied.
The pair $(X,\tau)$ is called a \emph{generalized Satake diagram (associated to $A$)} if \eqref{Satdiag2b} holds.
We denote the class of all generalized Satake diagrams associated to $A$ by $\GSat(A)$. 
\hfill \defnend
\end{defn}

Diagrammatically, condition \eqref{Satdiag2b} means that among all the connected components of all decorated diagrams obtained from $(X,\tau)$ by repeatedly deleting unfilled $\tau$-orbits (along with all adjacent edges) there are no copies of 
$\begin{tikzpicture}[baseline=-0.25em,line width=0.7pt,scale=1]
\draw[thick] (0,0) -- (0.5,0);
\filldraw[fill=white] (0,0) circle (.1);
\filldraw[fill=black] (.5,0) circle (.1);
\end{tikzpicture}$.

\begin{exam}
The following decorated diagrams are all examples of generalized Satake diagrams:
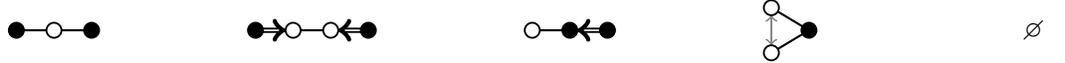
\begin{flalign*}
\qq && 
\begin{tikzpicture}[baseline=-0.25em,line width=0.7pt,scale=1]
\draw[thick] (-.5,0) -- (.5,0);
\filldraw[fill=black] (-.5,0) circle (.1);
\filldraw[fill=white] (0,0) circle (.1);
\filldraw[fill=black] (.5,0) circle (.1);
\end{tikzpicture} 
&& 
\begin{tikzpicture}[baseline=-0.25em,line width=0.7pt,scale=1]
\draw[double,->] (-.5,0) -- (-.1,0);
\draw[thick] (0,0) -- (0.5,0);
\draw[double,<-] (.6,0) -- (1,0);
\filldraw[fill=black] (-.5,0) circle (.1);
\filldraw[fill=white] (0,0) circle (.1);
\filldraw[fill=white] (.5,0) circle (.1);
\filldraw[fill=black] (1,0) circle (.1);
\end{tikzpicture} 
&& 
\begin{tikzpicture}[baseline=-0.25em,line width=0.7pt,scale=1]
\draw[thick] (0,0) -- (0.5,0);
\draw[double,<-] (.6,0) -- (1,0);
\filldraw[fill=white] (0,0) circle (.1);
\filldraw[fill=black] (.5,0) circle (.1);
\filldraw[fill=black] (1,0) circle (.1);
\end{tikzpicture}
&& 
\begin{tikzpicture}[baseline=-0.25em,line width=0.7pt,scale=1]
\draw[thick] (-.5,.3) -- (0,0) -- (-.5,-.3);
\draw[<->,gray] (-.5,.2) -- (-.5,-.2);
\filldraw[fill=white] (-.5,.3) circle (.1);
\filldraw[fill=white] (-.5,-.3) circle (.1);
\filldraw[fill=black] (0,0) circle (.1);
\end{tikzpicture} 
&&
\qq\examend
\end{flalign*}
\end{exam}

If \eqref{Satdiag2b} does not hold, {\it i.e.}~there exist $(i,j) \in X \times I \backslash X$ such that $\tau(j)=j$, $X(j)=\{i\}$, $a_{ij} = a_{ji}=-1$, then $\al_j(\rho^\vee_X) =  \frac{1}{2} \al_j(h_i) = \frac{a_{ij}}{2} = -\frac{1}{2}$ and we see that \eqref{Satdiag2a} fails too. 
We obtain the following statement.

\begin{lemma}
$\Sat(A) \subseteq \GSat(A)$.
\end{lemma}

\begin{defn} \label{def:weakSatdiag}
Let $X \subseteq I$ with $A_X$ of finite type and let $\tau\in \Aut(A)$ be an involution such that \eqref{Satdiag1a} is satisfied.
The pair $(X,\tau)$ is called a \emph{weak Satake diagram (associated to $A$)} if $(X,\tau) \in \WSat(A):=\GSat(A) \backslash \Sat(A)$. 
\hfill \defnend
\end{defn}

Example \ref{ex:Satdiags} (ii) shows that, at least for some $A$, $\WSat(A)$ is non-empty. More precisely, 
if $A$ is of type ${\rm A}_n$ or ${\rm A}_n^{(1)}$ then $\WSat(A)$ is empty; for all other finite or affine Cartan matrices, there exist weak Satake diagrams.
All generalized Satake diagrams for untwisted affine Cartan matrices of classical Lie type are listed in Appendix \ref{App:Satakediagrams}. 

\medskip 

The motivation for considering generalized Satake diagrams is that the corresponding Lie subalgebra $\mfk_{\bm c}$ satisfy the key intersection property \eqref{intersection} (and in the case of Satake diagrams, not just for one specific choice of $\bm c$ as in the proof of Proposition \ref{prop:kc:intersection1}).

\begin{rmk}
We will see that the quantum analogons $B_{\bm c, \bm s}$ of these Lie subalgebras satisfy a similar intersection condition (see Proposition \ref{prop:q-intersectionproperty}) and that the vector representation restricted to $B_{\bm c, \bm s}$ has nontrivial intertwiners, yielding solutions to the reflection equation.
It would be good to find a Lie-algebraic motivation for the notion of generalized Satake diagrams that does not rely on the algebra $\mfk_{\bm c}$, which is defined in a rather \emph{ad hoc} manner. \hfill \rmkend
\end{rmk}

In order to prove \eqref{intersection}, it is convenient to have a vector space basis of $\mfk_{\bm c}$. 
This can be established in a similar way to the arguments in \cite[Secs.~5.3 \& 6]{Ko1}.
First of all, define a partial order on $Q$ in the usual way: $\al \ge \beta$ if and only if $\al - \beta \in Q^+$; we also write $\al > \beta$ if and only if $\al \ge \beta$ and $\al \ne \beta$.
For $\bm i = (i_1,\ldots,i_\ell) \in I^\ell$ with $\ell \in \Z_{>0}$ define
\eq{ \label{def:efalphatuple}
\al_{\bm i} = \sum_{m=1}^\ell \al_{i_m}, \qq
f_{\bm i} = \ad(f_{i_1}) \cdots \ad(f_{i_{\ell-1}}) (f_{i_\ell}), \qq
g_{\bm i} = \ad(g_{i_1}) \cdots \ad(g_{i_{\ell-1}})(g_{i_\ell}).
}
Let $j \in I$.
From \eqref{mfg:rels1}, \eqref{theta:gX} and the explicit formula \eqref{def:g_i} we immediately obtain
\eq{
\label{kc:rels1} [e_i,g_j] = \del_{ij} h_i \in \mfh^\theta \qu \text{for all } i \in X, \qq [h,g_j] = - \al_j(h) g_j \qu \text{for all } h \in \mfh^\theta.
}
It follows that, as vector spaces,
\eq{ 
\label{kc:decomposition1} \mfk_{\bm c} = \mfn^+_X + \mfh^\theta + \langle g_j \rangle_{j \in I} = 
\mfn^+_X + \mfh^\theta + \sum_{\ell \in \Z_{\ge 0}} \sum_{\bm i \in I^\ell} \C g_{\bm i}.
}
Now we establish some relations for elements of $\mfk_{\bm c}$.
As a preliminary result, we have the following.

\begin{lemma} \label{lem:giij}
Let $(X,\tau) \in \GSat(A)$ and $i \in I \backslash X$ be such that $i = \tau(i)$. Also let $m \in \Z_{\ge 0}$.
\begin{enumerate} [itemsep=0.25em]
\item If $j \in X$ is such that $j \in X(i)$ then
\[ \ad(g_i)^m(g_j) - \ad(f_i)^m(f_j) - c_i^m \theta(\ad(f_i)^m(f_j)) 
\in \mfn_X^+ + \mfh_X + \sum_{\al_{\bm j} < m \al_i + \al_j} \C g_{\bm j}. \]
\item If $j \in I \backslash X$ is such that $j \notin \{ i, \tau(i) \}$ and $X(i) = \emptyset$ then
\[ \ad(g_i)^m(g_j) - \ad(f_i)^m(f_j) - c_i^m c_j \theta(\ad(f_i)^m(f_j)) 
\in \sum_{\al_{\bm j} < m \al_i + \al_j} \C g_{\bm j}. \]
\end{enumerate}
\end{lemma}

\begin{proof}
We prove both statements by induction with respect to $m$. 
\begin{enumerate} [itemsep=0.25em] 
\item For $m=1$, the statement follows immediately from $[f_i+c_i \theta(f_i),f_j] = f_{(i,j)} + c_i \theta(f_{(i,j)})$. 
Suppose $m \in \Z_{>1}$ and suppose the statement is satisfied when we replace $m$ by any smaller positive integer.
The induction hypothesis then yields
\eq{ \label{gij2:intermediate}
\begin{aligned}
& \ad(g_i)^m(g_j) - \ad(f_i)^m (f_j) - c_i^m \theta(\ad(f_i)^m (f_j) )  \\
& \qq = c_i \Big( [\theta(f_i),\ad(f_i)^{m-1}(f_j)] + c_i^{m-2} \big[ f_i, \theta(\ad(f_i)^{m-1}(f_j)) \big] \Big) \\
& \qq \qq + \ad(g_i)(\text{an element of } \mfn_X^+ + \mfh_X + \sum_{\al_{\bm j} < (m-1) \al_i + \al_j} \C g_{\bm j}.
\end{aligned}
}
We have $[\theta(f_i),\ad(f_i)^{m-1}(f_j)] \in \mfg_{w_X(\al_i)-(m-1)\al_i-\al_j}$; precisely if $w_X(\al_i) - (m-1)\al_i-\al_j \in \Phi \cup \{0\}$ this is nonzero.
This occurs only if $w_X(\al_i) \ge (m-1)\al_i-\al_j$ or $w_X(\al_i) < (m-1)\al_i-\al_j$. 
Because of $j \in X(i)$ and Lemma \ref{lem:Xjdecomposition} the latter inequality is not satisfied for any $m$. 
The former implies $m \le 2$, again by virtue of Lemma \ref{lem:Xjdecomposition}; by assumption this means $m=2$.
In this case $[\theta(f_i),\ad(f_i)(f_j)] \in \mfg_{w_X(\al_i)-\al_i-\al_j} \subseteq \mfn_X^+ + \mfh_X$. 
Owing to \eqref{theta:squared}, $\big[ f_i, \theta(\ad(f_i)^{m-1}(f_j)) \big]$ equals $\theta([\theta(f_i),\ad(f_i)(f_j)])$ up to a sign, so that it is an element of $\mfn_X^+ + \mfh_X$ as well.
Now \eqref{gij2:intermediate} implies that
\[ \ad(g_i)^m(g_j) - \ad(f_i)^m (f_j) - c_i^m \theta(\ad(f_i)^m (f_j) )  \in \mfn_X^+ + \mfh_X + \sum_{\al_{\bm j} < m \al_i + \al_j} \C g_{\bm j}.  \]
\item We will repeatedly use that $f_i = -\theta(e_i)$ and $\theta(f_i)=-e_i$; in particular $g_i = f_i - c_i f_{\tau(i)}$.
For $m=1$ the statement is simply a consequence of $[g_i,g_j] = [f_i,f_j] + c_i c_j \theta([f_i,f_j])$ since $[\theta(f_i),f_j]=-[e_i,f_j]=0$ and, owing to \eqref{theta:squared}, $[f_i,\theta(f_j)]$ equals $\theta([\theta(f_i),f_j])$ up to a sign.
Suppose $m \in \Z_{>1}$ and suppose the statement is true with $m$ replaced by any smaller positive integer.
By the induction hypothesis,
\eq{ \label{gij1:intermediate}
\begin{aligned}
& \ad(g_i)^m(g_j) - \ad(f_i)^m (f_j) - c_i^m c_j \theta(\ad(f_i)^m (f_j) ) \\
& \qq = - c_i \Big( [e_i,\ad(f_i)^{m-1}(f_j)] + c_i^{m-2} c_j \theta\big(\big[ e_i ,\ad(f_i)^{m-1}(f_j) \big]\big) \Big) \\
& \qq \qq + \ad(g_i)(\text{an element of }\mfn_X^+ + \mfh^\theta + \sum_{\al_{\bm j} < (m-1) \al_i + \al_j} \C g_{\bm j} ).
\end{aligned}
}
We have
\begin{align*}
[e_i,\ad(f_i)^{m-1}(f_j)] &= \sum_{r=1}^{m-1} \ad(f_i)^{r-1} \ad(h_i) \ad(f_i)^{m-1-r} (f_j) \\
&= -\sum_{r=1}^{m-1} (2(m-1-r) + a_{ij}) \ad(f_i)^{m-2} (f_j) \\
&= - (m-1) ( m-2+ a_{ij}  ) \ad(f_i)^{m-2} (f_j)
\end{align*}
so that 
\begin{align*}
&[e_i,\ad(f_i)^{m-1}(f_j)] + c_i^{m-2} c_j \theta\big(\big[ e_i ,\ad(f_i)^{m-1}(f_j) \big]\big) \\
&\qq = - (m-1) ( m-2+ a_{ij}  ) \big( \ad(f_i)^{m-2} (f_j) +  c_i^{m-2} c_j \theta(\ad(f_i)^{m-2} (f_j)) \big) \\
&\qq = - (m-1) ( m-2+ a_{ij}  ) \ad(g_i)^{m-2}(g_j) + (\text{an element of } \sum_{\al_{\bm j} < (m-2) \al_i + \al_j} \C g_{\bm j}) \\
&\qq \in \sum_{\al_{\bm j} < m \al_i + \al_j} \C g_{\bm j},
\end{align*}
because of the induction hypothesis once again and $\theta(\al_{\bm j})<0$ if $0<\al_{\bm j}<m \al_i + \al_j$.
Now \eqref{gij1:intermediate} yields the desired conclusion.\hfill \qedhere
\end{enumerate}
\end{proof}

For $i,j \in I$ denote $\la_{ij}:=(1-a_{ij})\al_i+\al_j \in Q^+$; note that \eqref{mfg:Serre} implies that $\la_{ij} \notin \Phi$. 

\begin{lemma} \label{lem:kc:Serre}
Let $(X,\tau) \in \GSat(A)$. 
Suppose that $\bm c \in (\C^\times)^{I \backslash X}$ satisfies 
\eq{ \label{Idiff}
a_{j,\tau(j)}=0 \qu \text{and} \qu X(j) = \emptyset \qu \implies \qu c_j = c_{\tau(j)} \qq \text{for all } j \in I \backslash X.
}
Then for all $i,j \in I$ such that $i \ne j$ the \emph{modified Serre relations} are satisfied
\[ 
\ad(g_i)^{1-a_{ij}}(g_j) \in \mfn_X^+ + \mfh^\theta + \sum_{\al_{\bm j} < \la_{ij}} \C g_{\bm j}. 
\]
\end{lemma}

\begin{proof}
We will repeatedly use the explicit formula \eqref{def:g_i}.
If $i \in X$, then using \eqref{theta:gX} we obtain
\[ 
\ad(g_i)^{1-a_{ij}}(g_j)  = \begin{cases} \ad(f_i)^{1-a_{ij}}(f_j) & \text{if } j \in X, \\ \ad(f_i)^{1-a_{ij}}(f_j) + c_j \theta(\ad(f_i)^{1-a_{ij}}(f_j)) & \text{otherwise} \end{cases} 
\]
which in both cases vanishes due to the Serre relations \eqref{mfg:Serre}.
Hence we may assume that $i \in I \backslash X$. 
By induction with respect to $m \in \Z_{\ge 0}$ we straightforwardly obtain
\eq{
\label{kc:Serre1} \ad(g_i)^m (g_j) \in \sum_{r=0}^m \Big( \mfg_{-\beta_{ij}^{(r)}} + \mfg_{-\theta(\beta_{ij}^{(r)})} \Big) 
}
where we have used \eqref{theta:rootspace} and denoted
\eq{
\label{def:beta} \beta_{ij}^{(r)} = (m-r) \al_i + \al_j - r w_X(\al_{\tau(i)}) \in Q. 
}
Note that the $r$-th term in \eqref{kc:Serre1} survives precisely if $\beta_{ij}^{(r)} \in \Phi \cup \{0 \}$.
Now fix $m = 1-a_{ij}$. 
In particular, the term with $r=0$ vanishes because of the Serre relations \eqref{mfg:Serre}.
Let $r>0$; because $\Phi \subset Q^+ \cup (-Q^+)$, the $r$-th term in \eqref{kc:Serre1} survives only if one of the following conditions is satisfied:
\eqa{
\label{betaminus}  (1-a_{ij}-r) \al_i + \al_j &< r w_X(\al_{\tau(i)})\\ 
\label{betaplus} (1-a_{ij}-r) \al_i + \al_j &\ge r w_X(\al_{\tau(i)}).
}
Let us first deal with the case \eqref{betaminus}.
Lemma \ref{lem:Xjdecomposition} implies $\{i,j\} \subseteq X(i) \cup \{ \tau(i) \}$; hence $\tau(i)=i$ and $j \in X(i)$.
Furthermore, as a consequence of Lemma \ref{lem:giij} (i) with $m=1-a_{ij}$ we have
\[ \ad(g_i)^{1-a_{ij}}(g_j) \in \mfn_X^+ + \mfh_X + \sum_{\al_{\bm j} < \la_{ij}} \C g_{\bm j} \]

Hence it suffices to consider \eqref{betaplus}. 
Then $X(i) \subseteq \{ j \}$ so there are two possibilities for $X(i)$, which we treat separately below.
\begin{description} [itemsep=.25em]
\item[$X(i) = \{ j \}$]
We must have $\tau(i)=i$ and $w_X(\al_{\tau(i)}) = \al_i - a_{j i} \al_j$, so that we obtain $1 \ge r |a_{j i}|$ and $1-a_{ij} \ge 2r$. 
Because we may assume $r \ne 0$, we obtain $r=1=-a_{ji}$ so that $\beta_{ij}^{(r)} = (-a_{ij}-1) \al_i$.
Since $\al_i$ is real, we must have $a_{ij} \in \{ -1, -2 \}$. 
Because $(X,\tau) \in \GSat(A)$ it must in fact be $a_{ij} = -2$. 
Hence the configuration as in Example \ref{ex:Satdiags} (i) applies; it follows that $\ad(g_i)^3(g_j)=0$.
\item[$X(i) = \emptyset$]
We have $\tau(i) \in \{i,j\}$ and $w_X(\al_{\tau(i)}) = \al_{\tau(i)}$. 
Assume that $\tau(i)=j$ so that $a_{ij}=a_{ji}$.
Then $r=1$, so that $\beta_{ij}^{(r)} = -a_{ij} \al_i$; since $\al_i$ is real it follows that $a_{ij} \in \{0,-1\}$. 
If $a_{ij}=0$, then by assumption we have $c_i = c_j$ so that
\[ \ad(g_i)^{1-a_{ij}}(g_j) = [f_i-c_i e_j,f_j-c_j e_i] = c_i (h_i-h_{\tau(i)}) \in \mfh^\theta. \]
If $a_{ij}=-1$ then a straightforward computation yields
\eqn{
\ad(g_i)^{1-a_{ij}}(g_j) &= [f_i-c_i e_j,[f_i,f_j]-c_i c_j [e_i,e_j]] \\ &= c_i ([h_j,f_i]+ c_j [h_i,e_j]) \\ &= c_i g_i \in \sum_{\al_{\bm j} < \la_{ij}} \C g_{\bm j}.
}
Finally, we must consider the case $\tau(i)=i$.
Lemma \ref{lem:giij} (i) with $m=1-a_{ij}$ implies that 
\[ 
\ad(g_i)^{1-a_{ij}}(g_j) \in \sum_{\al_{\bm j} < \la_{ij}} \C g_{\bm j}. \hfill \qedhere 
\]
\end{description}
\end{proof}

From \eqrefs{mfg:rels1}{mfg:Serre} it follows that $\mfn^-$ is spanned by $\bigcup_{\ell \in\Z_{>0}} \{ f_{\bm i} \}_{\bm i \in I^\ell}$.
Now choose $\mc{J} \subseteq \bigcup_{\ell \in \Z_{>0}} I^\ell$ such that $\{ f_{\bm j} \}_{\bm j \in \mc{J}}$ is a basis of $\mfn^-$; note that for each $\bm j \in \mc{J}$, $\al_{\bm j} \in \Phi \cap Q^+$.

\begin{thrm} \label{thm:kc:basis}
Let $(X,\tau) \in \GSat(A)$. 
Suppose that $\bm c \in (\C^\times)^{I \backslash X}$ satisfies \eqref{Idiff}.
Then, as vector spaces,
\[ 
\mfk_{\bm c}  = \mfn^+_X \oplus \mfh^\theta \oplus \bigoplus_{\bm j \in \mc{J}} \C g_{\bm j}. 
\]
\end{thrm}

\begin{proof}
First we prove that $\mfk_{\bm c} =  \mfn^+_X + \mfh^\theta + \sum_{\bm j \in \mc{J}} \C g_{\bm j}$.
Fix $\ell \in \Z_{>0}$ and fix $\bm k \in I^\ell$. We may assume that $\bm k \notin \mc{J}$.
Considering \eqref{kc:decomposition1} we see that it suffices to prove that 
\eq{ 
\label{eqn:basis} g_{\bm k} \in \mfn^+_X + \mfh^\theta + \sum_{\bm j \in \mc{J}} \C g_{\bm j}. 
}
Since $\{ f_{\bm j} \}_{\bm j \in \mc{J}}$ spans $\mfn^-$, we have $f_{\bm k} = \sum_{\bm j \in \mc{J}} a_{\bm j} f_{\bm j}$ for some $a_{\bm j} \in \C$. 
One obtains such a decomposition for $f_{\bm k}$ by repeatedly applying only Serre relations \eqref{mfg:Serre} for $f_i$ and $f_j$ for some $i,j \in I$ such that $\la_{ij} \le \al_{\bm k}$. 
This defines a finite sequence $((i_1,j_1),(i_2,j_2),\ldots,(i_r,j_r))$ of pairs of nodes for some $r \in \Z_{\ge 0}$.
Now repeatedly apply the modified Serre relations to $g_{\bm k}$ given by the same sequence of pairs of nodes, see Lemma \ref{lem:kc:Serre}.
At each step, using \eqref{kc:rels1} where necessary, we obtain (instead of zero) an element of 
\[ \mfn^+_X + \mfh^\theta + \sum_{\al_{\bm j} < \al_{\bm k}} \C g_{\bm j} = \mfn^+_X + \mfh^\theta + \sum_{\al_{\bm j} < \al_{\bm k} \atop \bm j \in \mc{J}} \C g_{\bm j} + \sum_{\al_{\bm j} < \al_{\bm k} \atop \bm j \notin \mc{J}} \C g_{\bm j} \] 
so that $g_{\bm k}$ lies in this subspace of $\mfk_{\bm c}$.
By induction on $\ell$ we obtain 
\[ g_{\bm k} \in \mfn^+_X + \mfh^\theta + \sum_{\al_{\bm j} < \al_{\bm k} \atop \bm j \in \mc{J}} \C g_{\bm j} \subseteq \mfn^+_X + \mfh^\theta + \sum_{\bm j \in \mc{J}} \C g_{\bm j} \]
as required.
It remains to show that the sum is direct.
Let $\bm j \in \mc{J}$. Then $f_{\bm j}$ is nonzero. 
Because of $\theta(\mfg_{\beta}) = \mfg_{-w_X \tau(\beta)}$ and the explicit formula \eqref{def:g_i}, we have
\eq{ \label{gi:weight}
 g_{\bm j} - f_{\bm j} \in  \mfn_X^+ + \mfh_X + \sum_{\al_{\bm k} < \al_{\bm j}} \C g_{\bm k}.
 }
Hence $f_{\bm j} = \pi_{-\al_{\bm j}}(g_{\bm j})$ for all $\bm j \in \mc{J}$. Thus the linear independence of $\{ f_{\bm k} \}_{\bm k \in \mc{J}}$ together with the triangular decomposition $\mfg = \mfn^+ \oplus \mfh \oplus \mfn^-$ implies that, as required,
\[ 
\mfn^+_X + \mfh^\theta + \sum_{\bm j \in \mc{J}} \C g_{\bm j} = \mfn^+_X \oplus \mfh^\theta \oplus \bigoplus_{\bm j \in \mc{J}} \C g_{\bm j}. \hfill \qedhere 
\]
\end{proof}

\begin{rmk}
It is possible to write explicit expressions for the right-hand side of the modified Serre relations appearing in Lemma \ref{lem:kc:Serre}, at least when $|a_{ij}| \le 4$, {\it i.e.}~at least when $A$ is of finite or affine type. 
Together with the relations \eqref{kc:rels1} and the usual relations for $\langle \mfn^+_X , \mfh^\theta \rangle$ inherited from $\mfg$, we claim that this yields an efficient presentation of $\mfk_{\bm c}$ in terms of generators and relations for $(X,\tau) \in \WSat(A)$.
See \cite[Sec.~7]{Ko2} for the analogous theory for the quantized versions of $U(\mfk_{\bm c})$ for $(X,\tau) \in \Sat(A)$: it can be immediately generalized to $(X,\tau) \in \WSat(A)$ because at no point is condition \eqref{Satdiag2a} necessary, just the weaker condition \eqref{Satdiag2b}. \hfill \rmkend
\end{rmk}

\begin{crl} \label{cor:kc:intersection2}
Let $(X,\tau) \in \GSat(A)$. 
Suppose that $\bm c \in (\C^\times)^{I \backslash X}$ satisfies \eqref{Idiff}.
Then the intersection condition \eqref{intersection} is satisfied.
\end{crl}

\begin{proof}
Clearly $\mfh^\theta \subseteq \mfk_{\bm c} \cap \mfh$; it remains to show the reverse inclusion.
Suppose $h \in \mfk_{\bm c} \cap \mfh$. 
From $\pi_{-\al_{\bm j}}(g_{\bm j}) = f_{\bm j} \in \mfn^-$ and the triangular decomposition $\mfg = \mfn^+ \oplus \mfh \oplus \mfn^-$ we immediately deduce that $h \in \mfn^+_X \oplus \mfh^\theta$. 
From \eqref{theta:gX} we obtain $h \in \mfh^\theta$.
\end{proof}

In the next section we will consider an analogue of $\mfk_{\bm c}$ in the quantum setting.
It is worth noting that the universal enveloping algebra corresponding to $\mfk_{\bm c}$ can be modified by scalar terms, allowing us to introduce another tuple $\bm s \in \C^{I \backslash X}$: $U(\mfk_{\bm c})_{\bm s}$ is generated as a unital associative algebra by the subalgebras $U(\mfn^+_X)$ and $U(\mfh^\theta)$ and the elements $g_i := f_i$ (for $i \in X$) and $g_i := f_i + c_i \theta(f_i) + s_i$ (for $j \in I\backslash X$), see \cite[Cor.~2.9]{Ko1}.

%

\section{Quantum groups and quantum pair algebras} \label{sec:QG}


\subsection{Drinfeld-Jimbo quantum groups} \label{sec:DJquantumgroups}

In the Drinfeld-Jimbo presentation \cite{Dr1,Dr2,Ji1} quantum groups are typically defined as unital associative algebras over $\C(q)$, the field of rational expressions in an indeterminate $q$, with specific generators and relations.
As we will be invoking Schur's lemma in Section \ref{sec:Kmat} when discussing intertwiners, instead we will define them over an algebraic closure $\K$ of $\C(q)$. 
Elsewhere one may replace $\K$ by a quadratic closure of $\C(q)$, in which case one must choose the $d_i$ appearing in \eqref{DAisAD} to be dyadic fractions.
Note that $q_i := q^{d_i} \in \K$ for $i \in I$.

\begin{defn} \label{D:DJQG}
We denote by $U_q(\mfg)$ the Hopf algebra over $\K$ with generators $x_i$, $y_i$ and invertible $k_i$ for $i \in I$ satisfying the defining relations
\begin{gather}
\label{Urelations1} k_i k_j = k_j k_i, \qq k_i x_j = q_i^{a_{ij}} x_j k_i, \qq k_i y_j = q_i^{-a_{ij}} y_j k_i, \qq  [x_i,y_j] = \del_{ij} \frac{k_i-k^{-1}_{i}}{q_i-q_i^{-1}}, \\
\label{qSerre} \sum_{r=0}^{1-a_{ij}} (-1)^r x_i^{(1-a_{ij}-r)} x_j x_i^{(r)} = \sum_{r=0}^{1-a_{ij}} (-1)^r y_i^{(1-a_{ij}-r)} y_j y_i^{(r)} = 0
\end{gather}
for all $i, j \in I$.
Here we have introduced the notations, for $i \in I$ and $r \in \Z_{\ge 0}$,
\[
[r]_q = \frac{q^r-q^{-r}}{q-q^{-1}}, \qq 
[r]_q! = [r]_q [r\!-\!1]_q \cdots [2]_q [1]_q, \qq 
x_i^{(r)} = \frac{x_i^r}{[r]_{q_i}!}, \qq 
y_i^{(r)} = \frac{y_i^r}{[r]_{q_i}!}.
\]
The following assignments for $i \in I$ determine the coproduct, counit and antipode:
\begin{flalign}
\nonumber \Delta(x_i) &= x_i \ot 1 + k_i \ot x_i, & \epsilon(x_i) &= 0, & S(x_i) &= - k_i^{-1} x_i, \\
\label{Hopfalg} \qq \qq \Delta(y_i) &= y_i \ot k^{-1}_i  + 1 \ot y_i, & \epsilon(y_i) &= 0, & S(y_i) &= - y_i k_i, \qq   \\
\nonumber \Delta(k_i) &= k_i \ot k_i, & \epsilon(k_i) &=1, & S(k_i) &= k_i^{-1}. && \defnend
\end{flalign}
\end{defn}  

For $ \mu = \sum_{i \in I} m_i \al_i \in Q$ with $\bm m \in \Z^I$ write $k_\mu = \prod_{i \in I} k_i^{m_i}$.
We have the Hopf subalgebra
\eq{
U_q( \mfh) := \langle k_i^{\pm 1} \rangle_{i \in I } = \bigoplus_{\mu \in Q} \K k_\mu.  
}
Define the quantum root space of $U_q( \mfg)$ corresponding to $\mu \in Q$ by
\[ 
U_q( \mfg)_\mu   = \{ a \in U_q( \mfg) \, | \, \forall i \in I \; k_i a k^{-1}_i = q^{(\al_i,\mu)} a \} \subset U_q( \mfg)
\]
and one obtains a $Q$-grading: $U_q(\mfg) = \bigoplus_{\mu \in Q} U_q(\mfg)_\mu$.\\

Following \cite[Sec.~3.2]{Ko1}, for $\chi \in \wt H_q := \Hom(Q,\K^\times)$, define $\Ad(\chi) \in \Aut_{\rm Hopf}(U_q( \mfg))$ by 
\[ \Ad(\chi)(a) = \chi(\mu)(a), \qq \text{for all } a \in U_q(\mfg_\mu), \mu \in Q. \]
In other words, $\Ad(\chi)(x_i) = \chi(\al_i) x_i$  and $\Ad(\chi)(y_i) = \chi(\al_i)^{-1} y_i$ for all $i \in I$ and $\Ad(\chi)$ acts as the identity on $U_q(\mfh)$.
We may also view $\Aut(A) < \Aut_{\rm Hopf}(U_q( \mfg))$ by setting
\eq{\label{sigmaUq} 
\si(x_i) = x_{\si(i)}, \qq \si(y_i) = y_{\si(i)}, \qq \si(k_i) = k_{\si(i)} \qq \text{for all } \sigma \in \Aut(A), \, i \in I.
}
In $\Aut_{\rm Hopf}(U_q( \mfg))$ we may form the semidirect product $\Ad(\wt H_q) \rtimes \Aut(A)$.
According to \cite[Thm.~2.1]{Tw}, in fact
\eq{ \label{Hopfalgautodecomp} \Aut_{\rm Hopf}(U_q( \mfg)) = \Ad(\wt H_q) \rtimes \Aut(A). }
Coideal subalgebras $\mc{B}, \mc{B}' \subseteq U_q(\mfg)$ are called {\it equivalent} if there exists $\phi \in \Aut_{\rm Hopf}(U_q( \mfg))$ such that $\phi(\mc{B}) = \mc{B}'$. 

\begin{rmk} \label{rmk:extendedaffqg}
The full quantum Kac-Moody algebra $U_q(\mfg^{\rm ext})$ is a larger Hopf algebra obtained as follows.
We extend $\mfh$ to a larger vector space $\mfh^{\rm ext}$ by adding basis elements $\{ \la_s \}_{s=1}^{n'}$ with $n'={\rm corank}(A)$ satisfying $\al_j(\la_s) \in \{0,1\}$ for all $j \in I$ and $1 \le s \le n'$. 
Then $U_q(\mfg^{\rm ext})$ is the extension of $U_q(\mfg)$ by the invertible generators $\{ \La_{s} \}_{s=1}^{n'}$ satisfying
\[ 
\La_{s} x_i = q^{ \al_i(\la_s)} x_i \La_{s}, \qq 
\La_{s} y_i = q^{- \al_i(\la_s)} y_i \La_{s}, \qq 
[\La_{s},k_i] = 0, \qq 
[\La_s,\La_{r}]=0 . 
\]
for $i \in I$ and $1 \le s,r \le {n'}$. The Hopf algebra structure on $U_q(\mfg^{\rm ext})$ is the one on $U_q(\mfg)$ extended~by
\[ 
\Delta(\La_{s}) = \La_{s} \ot \La_{s}, \qq \epsilon(\La_{s}) =1, \qq S(\La_{s}) = \La_{s}^{-1} 
\]
for $1\le s \le n'$. Note that a completion of $U_q(\mfg^{\rm ext})$ is used in the construction of the universal R-matrix when $A$ is of affine type, see {\it e.g.}~\cite{Dr1,FrRt,KhTo}.
If the minimal realization $( \mfh^{\rm ext}, \{ h_i \}_{i \in I},\{ \al_i \}_{i \in I})$ of $A$ is compatible with $\tau$ in the sense that there exists a permutation $\wt \tau$ of $\{ s \}_{1 \le s \le n'}$ such that $\al_{\tau(i)}(\la_{\wt \tau(s)}) = \al_i(\la_s)$ for all $i \in I$ and $1 \le s \le n'$, then the theory set out in the next section can be extended to $U_q(\mfg^{\rm ext})$; see \cite[Sec.~2.6]{Ko1} for more detail and note also \cite[Rmk.~8.2]{BgKo2}. \hfill \rmkend
\end{rmk}

We now introduce various algebra automorphisms of $U_q(\mfg)$, again following \cite[Sec.~3.4 \& 4.1]{Ko1}.
Define $\om_q: U_q( \mfg) \to U_q( \mfg)$ by
\eq{ \label{defn:tw}
\om_q(x_i) = - k^{-1}_i y_i, \qq
\om_q(y_i) = - x_i k_i, \qq
\om_q(k_i) = k_i^{-1} \qq \text{for } i \in I.
}
For $i \in I$ denote by $T_i$ the algebra automorphism of $U_q( \mfg)$ called $T''_{i,1}$~in~\cite[37.1]{Lu}:
\begin{equation} \label{defn:Ti}
\begin{aligned}
T_i(x_i) &= -y_i \, k_i , \qu & T_i(x_j) &= \sum_{r=0}^{-a_{ij}} (-q_i)^{-r}  x_i^{(-a_{ij}-r)} x_j \, x_i^{(r)}, 
\\  
T_i(y_i) &= -k_i^{-1} x_i, \qu & T_i(y_j) &= \sum_{r=0}^{-a_{ij}} (-q_i)^{r} y_i^{(r)} y_j\, y_i^{(-a_{ij}-r)} ,
\end{aligned}
\end{equation}
where $j \in I \backslash \{ i \}$, and $T_i(k_\mu) = k_{r_i(\mu)}$ for $\mu \in Q$; in particular $T_i(U_q(\mfg)_\mu) = U_q(\mfg)_{r_i(\mu)}$.
Owing to \cite[39.4.3]{Lu} the $T_i$ satisfy the braid relations \eqref{braidrelations} associated to $A$.
Hence we have a well-defined 
\[ 
T_w := T_{i_1} T_{i_2} \cdots T_{i_\ell} \in U_q( \mfg) 
\]
where $w=r_{i_1}r_{i_2}\cdots r_{i_\ell}$ with $i_1,\ldots,i_\ell \in I$ is a reduced decomposition for $w \in W$.


\subsection{Quantum pair algebras} \label{sec:qsymmpairs}

We consider an algebra automorphism which is a quantum analogon of $\theta(X,\tau)$ following \cite[Def.~4.3]{Ko1}.
\begin{defn}
Let $X \subseteq I$ with $A_X$ of finite type and let $\tau \in \Aut(A)$ be an involution such that \eqref{Satdiag1a} is satisfied. 
The \emph{algebra automorphism associated to $(X,\tau)$} is
\begin{flalign} \label{theta_q} 
&& \theta_q = \theta_q(X,\tau) := T_{w_X} \, \tau \, \om_q. && \defnend
\end{flalign}
\end{defn}

\begin{prop}
Let $X \subseteq I$ with $A_X$ of finite type and let $\tau \in \Aut(A)$ be an involution such that \eqref{Satdiag1a} is satisfied. 
Then
\begin{align}
\label{thetaq:UqgX} && \theta_q(a) &= a  && \text{for all } a \in U_q(\mfg_X), \\
\label{thetaq:rootspace} && \theta_q(U_q( \mfg)_\mu) &= U_q(\mfg)_{\theta(\mu)} && \text{for all } \mu \in Q, \\ 
\label{thetaq:Uqh} && \theta_q(k_\mu) &= k_{\theta(\mu)} && \text{for all } \mu \in Q, \\  
\label{Uqhthetaq:decomposition} &&  U_q(\mfh)^{\theta_q} &= \big\lan \{ k_i^{\pm 1} \}_{i \in X}, \{ k_j \, k_{\tau(j)}^{-1}, k_j^{-1} k_{\tau(j)} \}_{j \in I^*} \big\ran, \hspace{-41.5mm} && \qq && 
\end{align}
where $I^* \subseteq I \backslash X$ is as in Proposition \ref{prop:theta:properties}.
\end{prop}

\begin{proof}
Identity \eqref{thetaq:UqgX} is a consequence of \cite[Lem.~3.4]{Ko1}.
We obtain \eqref{thetaq:rootspace} as a consequence of the corresponding properties of the three constituent automorphisms $T_{w_X}$, $\tau$ and $\om_q$.
Next, \eqref{thetaq:Uqh} immediately follows from $\om_q(k_\mu) = k_{-\mu}$, $\tau(k_\mu) = k_{\tau(\mu)}$ and $T_i(k_\mu) = k_{r_i(\mu)}$ for $i \in I$.
Finally, \eqref{Uqhthetaq:decomposition} follows from \eqref{thetaq:Uqh} and the proof of \eqref{htheta:decomposition} in Proposition \ref{prop:theta:properties}.
\end{proof}

The following definition is due to \cite[Def.~5.6]{Ko1} (also see \cite[Sec.~7, Variation 2]{Le2}).

\begin{defn} \label{defn:rCSA} 
Let $X \subseteq I$ with $A_X$ of finite type and let $\tau \in \Aut(A)$ be an involution such that \eqref{Satdiag1a} is satisfied. 
Let ${\bm c}=(c_j)_{j\in I\backslash X}\in (\K^\times)^{I\backslash X}$ and ${\bm s}=(s_j)_{j\in I\backslash X}\in\K^{I\backslash X}$ and denote 
\eq{ \label{def:b_j}
b_i  = 
\begin{cases} 
y_i & \text{if } i \in X, \\
y_i + c_i \theta_q(y_i k_i) k^{-1}_i - s_i k^{-1}_i & \text{otherwise} 
\end{cases}
}
and
\[
U_q( \mfn_X^+)=\langle x_i \rangle_{i \in X}.
\]
The \emph{subalgebra associated to $(X,\tau,\bm c,\bm s)$} is
\begin{flalign} \label{def:Bcs} 
&& B_{\bm c,\bm s} = B_{\bm c,\bm s}(X,\tau) := \langle U_q( \mfn_X^+), U_q( \mfh)^{\theta_q}, \{ b_i \}_{i \in I} \rangle.  && \defnend
\end{flalign}
\end{defn}

\begin{prop}
Let $X \subseteq I$ with $A_X$ of finite type and let $\tau \in \Aut(A)$ be an involution such that \eqref{Satdiag1a} is satisfied. 
Let ${\bm c}=(c_j)_{j\in I\backslash X}\in (\K^\times)^{I\backslash X}$ and ${\bm s}=(s_j)_{j\in I\backslash X} \in \K^{I\backslash X}$.
Then $B_{\bm c,\bm s}$ is a right coideal of $U_q( \mfg)$: 
\[ \Delta(B_{\bm c,\bm s}) \subseteq B_{\bm c,\bm s} \otimes  U_q( \mfg). \]
\end{prop}

\begin{proof}
This is \cite[Prop.~5.2]{Ko1}.
We emphasize that neither \eqref{Satdiag2a} nor the weaker condition \eqref{Satdiag2b} is needed to establish this.
\end{proof} 

Note that a left coideal subalgebra is obtained by applying the antipode or its inverse to $B_{\bm c,\bm s}$. 
The arguments in \cite[Sec. 10]{Ko1}, which crucially do not rely on \eqref{Satdiag2a}, imply the following.

\begin{prop} \label{prop:q-intersectionproperty}
Let $(X,\tau) \in \GSat(A)$.
Suppose $\bm c \in (\K^\times)^{I\backslash X}$ and $\bm s \in \K^{I\backslash X}$ are such that
\eq{ \label{q-intersectionproperty}
B_{\bm c,\bm s} \cap U_q( \mfh) = U_q( \mfh)^{\theta_q}
}
holds.
Then $B_{\bm c,\bm s}(X,\tau)$ specializes to $U(\mfk_{\bm c}(X,\tau))_{\bm s}$ at $q=1$.
\end{prop} 

Equation \eqref{q-intersectionproperty} holds if the tuples $\bm c$~and~$\bm s$ satisfy certain constraints.

\begin{rmk}
The same constraints also arise when trying to find a nontrivial intertwiner for the restrictions of the vector representation $\RT_u$ of $U_q(\mfg)$ to $B_{\bm c,\bm s}$, see Theorem \ref{thrm:solexistence}.
This can be understood roughly as follows. 
The vector representation $\RT_u$ of $U_q(\mfg)$ is a particular representation on a finite-dimensional vector space, which has the property that elements of $U_q(\mfh)$ are mapped to diagonal matrices; on the other hand the images of the generators $b_j$ are non-diagonal matrices.
An intertwiner of $\RT_u|_{B_{\bm c,\bm s}}$ is a matrix $K(u)$ that commutes with the action of $B_{\bm c,\bm s}$ given by the representation $\RT_u$.
If there are too many elements of $U_q(\mfh)$ in $B_{\bm c,\bm s}$ compared to the number of $b_j$ (given by the cardinality of $I \backslash X$) then this puts too many restrictions on the intertwiner $K(u)$ and forces is to be as zero matrix.
It appears that the subalgebra $B_{\bm c,\bm s} \cap U_q(\mfh)$ is of the ``correct'' size precisely if it is as small as possible, {\it i.e.}~if it equals $U_q(\mfh)^{\theta_q}$. \hfill \rmkend
\end{rmk}

To describe these constraints on $\bm c$~and~$\bm s$ one needs to single out certain $\tau$-orbits of $I \backslash X$.
As in Proposition \ref{prop:theta:properties}, choose a subset $I^* \subseteq I \backslash X$ containing precisely one element of every $\tau$-orbit. 
Consider the following subsets of $I^*$:
\eqa{
I_{\rm diff} &= \{ j \in I^* \, | \, j \ne \tau(j) \text{ and }  (\al_j,\theta(\al_j)) \neq 0 \} \label{Idiffdefn}, \\
I_{\rm ns} &= \{ j \in I^* \,|\, j=\tau(j) \text{ and } a_{ij} = 0 \; \forall i \in X \} \label{Insdefn}, \\
I_{\rm nsf} & = 
\{ j \in I_{\rm ns} \, | \, a_{ij} \in 2\Z \;\forall i\in I_{\rm ns} \}. 
} 

\begin{defn}
Let $(X,\tau) \in \GSat(A)$.
We call a $\tau$-orbit $Y \subseteq I \backslash X$ \emph{special} if 
\begin{flalign*}
&& Y = \{j\} & \qu \text{such that } j \in I_{\rm nsf},\qq \text{or}\\
&& Y = \{j,\tau(j)\} & \qu \text{such that one of } j, \tau(j) \text{ is in } I_{\rm diff}. && \defnend
\end{flalign*}
\end{defn}

\begin{exam} \label{exam:distinguishedindexsets} 
In these examples we choose the standard labelling $I=\{0,1,\ldots,n\}$ for untwisted affine Dynkin diagrams; in particular $\{1,\ldots, n\}$ labels the associated finite Dynkin diagram.
\begin{enumerate}
\item Consider the Cartan matrix of type ${\rm A}_3^{(1)}$ with $I=\{0,1,2,3\}$. 
It can be straightforwardly verified that $(X,\tau) = (\{2\},(13)) \in \Sat(A) \subseteq \GSat(A)$. 
The diagram is

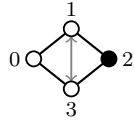
\begin{center}\begin{tikzpicture}[line width=0.7pt,scale=1]
\draw[thick] (1.5,.4) -- (1,0) -- (1.5,-.4) -- (2,0) -- (1.5,.4);
\draw[<->,gray] (1.5,.3) -- (1.5,-.3);
\filldraw[fill=white] (1,0) circle (.1) node[left=1pt]{\scriptsize $0$};
\filldraw[fill=white] (1.5,.4) circle (.1) node[above=1pt]{\scriptsize $1$};
\filldraw[fill=white] (1.5,-.4) circle (.1) node[below=1pt]{\scriptsize $3$};
\filldraw[fill=black] (2,0) circle (.1) node[right=1pt]{\scriptsize ${2}$};
\end{tikzpicture}\end{center}

\noindent We have $w_X = r_2$ and hence
\[ 
\theta(\al_0) = \al_0, \qu \theta(\al_1) = -\al_2 - \al_3, \qu \theta(\al_2) = \al_2, \qu\theta(\al_3) = -\al_1-\al_2. 
\]
Choose $I^*=\{0,1\}$.
We obtain $I_{\rm diff} = \{ 1 \}$ and $I_{\rm ns} = I_{\rm nsf} = \{ 0 \}$.
\item Fix $n \ge 4$ and consider the Cartan matrix of type ${\rm C}_n^{(1)}$ with $I=\{0,1,\ldots,n\}$.
Then $(X,\tau) = (\{0,1,\ldots,n-3\},\id) \in \GSat(A) \backslash \Sat(A)$. 
The diagram in this case is

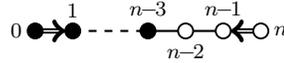
\begin{center}\begin{tikzpicture}[line width=0.7pt,scale=1]
\draw[double,<-] (-.1,0) -- (-.5,0);
\draw[thick,dashed] (0,0) -- (1,0);
\draw[thick] (1,0) -- (2,0);
\draw[double,<-] (2.1,0) --  (2.5,0);
\filldraw[fill=black] (-.5,0) circle (.1) node[left=1pt]{\scriptsize $0$};
\filldraw[fill=black] (0,0) circle (.1) node[above=1pt]{\scriptsize $1$};
\filldraw[fill=black] (1,0) circle (.1) node[above=1pt]{\scriptsize $n\!\!-\!\!3$};
\filldraw[fill=white] (1.5,0) circle (.1) node[below=1pt]{\scriptsize $n\!\!-\!\!2$};
\filldraw[fill=white] (2,0) circle (.1) node[above=1pt]{\scriptsize $n\!\!-\!\!1$};
\filldraw[fill=white] (2.5,0) circle (.1) node[right=1pt]{\scriptsize $n$};
\draw[] (0,.75) -- (0,.75);
\draw[] (0,-.75) -- (0,-.75);
\end{tikzpicture}\end{center}

\noindent Necessarily, $I^*=I\backslash X$. 
From $a_{n-1,n} = -2$ and $a_{n,n-1} = -1$ we derive
\begin{flalign*}
&&&& I_{\rm diff} = \emptyset, \qq  I_{\rm ns} = \{n-1,n\}, \qq I_{\rm nsf} = \{ n \}. \qq\qq\qu && \examend
\end{flalign*}
\end{enumerate}
\end{exam}

The special $\tau$-orbits in $I \backslash X$ are characterized diagrammatically as follows. 
Note that for $j \in I \backslash X$, by virtue of \cite[Lem.~5.3]{Ko1}, $\tau(X)=X$ and Lemma \ref{lem:Xjdecomposition}, the condition $j \ne \tau(j) \text{ and }  (\al_j,\theta(\al_j)) \neq 0$ appearing in the definition of $I_{\rm diff}$ is equivalent to the premise in \eqref{Idiff}.
Hence $I_{\rm diff}$ corresponds to those nontrivial $\tau$-orbits of unfilled nodes which neighbour each other or at least one filled node.
Next, $I_{\rm ns}$ corresponds precisely to all unfilled nodes which are fixed by $\tau$ and do not neighbour any filled nodes.
One obtains $I_{\rm nsf}$ by taking those nodes $j \in I_{\rm ns}$ all of whose neighbours $i \in I_{\rm ns}$ are connected to $j$ by edges with even multiplicity and not directed towards $j$ (in particular $d_i \le d_j$, {\it i.e.}~$\al_i$ must not be longer than~$\al_j$). 

The sets $I_{\rm diff}$ and $I_{\rm nsf}$ will be important for $\Ad(\wt{H}_q)$-equivalences of algebras $B_{\bm c,\bm s}$, see Corollary \ref{cor:nonremovable}.
For now, following \cite[(5.9) and (5.11)]{Ko1}, we define the sets of suitable parameters:
\eqa{
\mc{C} &= \mc{C}(X,\tau) = \{ {\bm c}\in(\K^\times)^{I\backslash X}\,|\, \forall j \in I^*:\, c_j \ne c_{\tau(j)} \Rightarrow j \in I_{\rm diff}
 \}, \label{C-family} \\
\mc{S} &= \mc{S}(X,\tau) = \{ {\bm s} \in \K^{I\backslash X}\,|\, \forall j \in I^*:\, s_j\ne 0 \Rightarrow j \in I_{\rm nsf} \} \label{S-family}.
}

\begin{defn}
Fix $(X,\tau) \in \GSat(A)$.
If $\bm c \in \mc{C}$ and $\bm s \in \mc{S}$ we call $B_{\bm c,\bm s}$ a \emph{quantum pair (QP) algebra}. 
If in addition $(X,\tau)\in\Sat(A)$, then $B_{\bm c,\bm s}$ is known as a {\it quantum symmetric pair (QSP) algebra}, see \cite[Def. 5.6]{Ko1}. 
\hfill \defnend
\end{defn}

\begin{prop} \label{prop:QPalgebras}
Let $(X,\tau) \in \GSat(A)$, $\bm c \in \mc{C}$ and $\bm s \in \mc{S}$. 
\begin{enumerate}
\item[(i)] The number of free parameters in $B_{\bm c,\bm s}$ is $|I^*|+|I_{\rm diff} \cup I_{\rm nsf}|$.
\item[(ii)] The intersection condition \eqref{q-intersectionproperty} holds.
\end{enumerate}
\end{prop}

\begin{proof}
To prove statement (i), note that for each $j \in I^*$ there is a free parameter $c_j \in \K^\times$. The defining conditions of the sets $\mc{C}$ and $\mc{S}$ in \eqrefs{C-family}{S-family} imply that there can only be one additional free parameter for each $j \in I_{\rm diff} \cup I_{\rm nsf}$.

Statement (ii) follows from the arguments in \cite[Secs. 5--7]{Ko1}. 
In these sections the only place in \cite{Ko1} where \eqref{Satdiag2a} is used is in \cite[Proof of Lem.~5.11, Step 1]{Ko1}. 
However, it is clear that the necessary condition in there is precisely \eqref{Satdiag2b}.
\end{proof}

\begin{rmk} \mbox{} 
\begin{enumerate}
\item 
We will see in Section \ref{sec:dressing} that $B_{\bm c,\bm s}$ is equivalent to a coideal subalgebra with $|I_{\rm diff} \cup I_{\rm nsf}|$ free parameters. 
This can be accomplished by applying a suitable Hopf algebra automorphism, which carries the remaining $|I^*|$ degrees of freedom.
\item For $(X,\tau)=(\emptyset,\id)$ Baseilhac and Belliard in \cite[Prop.~2.1]{BsBe1} considered the algebra $B_{\bm c,\bm s}$ with, according to some combinatorial prescription, either $s_j^2\,c_j^{-1}$ a particular rational expression in $q_j$, $s_j=0$ or $s_j \in\K$ free, for each $j \in I_{\rm ns} = I$.
In this case the algebra $B_{\bm c,\bm s}$ is called a \emph{generalized q-Onsager algebra}. 
However it appears that for arbitrary $(X,\tau) \in \GSat(A)$ the same prescription for $s_j$ with $j \in I_{\rm ns}$ yields a coideal subalgebra $B_{\bm c,\bm s}$ satisfying \eqref{q-intersectionproperty}; in this case it is convenient to redefine $\mc{S}$ and $I_{\rm nsf}$.
In order to obtain an algebra inequivalent to any QP algebra, one must have $I_{\rm nsf} \ne I_{\rm ns}$. 
A unified theoretical treatment of QP algebras and such generalized $q$-Onsager algebras will be presented elsewhere. \hfill \rmkend
\end{enumerate}
\end{rmk}


\subsection{Untwisted affine Dynkin diagrams of classical Lie type} \label{sec:affine}

We now restrict to the case where $A$ is of affine type. 
Then the kernel of $A$ is one-dimensional; in other words, there exists a unique tuple $(a_i)_{i \in I}$ of coprime positive integers such that
\eq{ \label{eq:gj} \sum_{j \in I} a_{ij} a_j = 0 }
(see {\it e.g.}~\cite[Ch.~17]{Ca}). 
Hence \eqref{DAisAD} implies that 
\eq{ k_c :=\prod_{j \in I} k_j^{a_j} }
is central in $U_q(\mfg)$.
The quotient $U_q(\mfg)/(k_c-1) U_q(\mfg)$ is called the \emph{quantum loop algebra}.
A theory of coideal subalgebras of $U_q(\mfg)/(k_c-1) U_q(\mfg)$ can be developed (see \cite[Sec. 11]{Ko1} and references therein) which is very similar to the one for $U_q(\mfg)$.
Furthermore, all finite-dimensional representations $\rho: U_q(\mfg) \to \End(V)$, including $\RT_u$, satisfy $\rho(k_c) = \Id_V$ (level-zero representations) so that $\rho$ factors through a representation of $U_q(\mfg)/(k_c-1) U_q(\mfg)$.

We furthermore assume that $A$ is an untwisted affine Cartan matrix.
We choose $I= \{0,\ldots, n\}$ such that $n= {\rm rank}(A)$ and $I \backslash \{ 0\}$ indexes the corresponding Cartan matrix $A^{\rm fin} := (a_{ij})_{1 \le i,j \le n}$ of finite type; in particular $a_0=1$.
Finally, we assume $A^{\rm fin}$ is one of the classical Lie types ${\rm A}_{n \ge 1}$, ${\rm B}_{n \ge 1}$, ${\rm C}_{n \ge 1}$, ${\rm D}_{n \ge 3}$.
Accordingly, the associated (semi)simple Lie algebra $\mfgf := \langle e_i,f_i \rangle_{1 \le i \le n} \subset \mfg$ equals $\mfsl_N$ with $N:=n+1$, $\mfso_N$ with $N:=2n+1$, $\mfsp_N$ with $N:=2n$ or $\mfso_N$ with $N:=2n$, respectively.
In Table~\ref{DynkinDiagrams} we list the untwisted affine Cartan matrices $A$ of classical Lie type and our choice of the corresponding $(d_i)_{i \in I}$, see \eqref{DAisAD}.

\begin{table}[t]
\caption{Dynkin diagrams of untwisted affine Lie algebras of classical types, the generalized Cartan matrix $A$ and the positive rationals $d_i$. 
We refer the reader to \cite[Appendix]{Ca} for more details on the Cartan matrices of finite and affine type. } \label{DynkinDiagrams}
\begin{tabular}{lllm{34mm}} 
\hline
\multirow{2}{*}{Name} & \multirow{2}{*}{Diagram} & Cartan matrix & \multirow{2}{*}{$d_i$ ($0 \le i \le n$)} \\
& & $A=(a_{ij})_{0 \le i,j \le n}$ & \\
\hline
${\rm A}^{(1)}_1$  
& 
\begin{tikzpicture}[baseline=-0.25em,line width=0.7pt,scale=1]
\draw[double] (0,0) -- (0.5,0);
\filldraw[fill=white] (0,0) circle (.1) node[left=1pt]{\scriptsize $0$};
\filldraw[fill=white] (.5,0) circle (.1) node[right=1pt]{\scriptsize $1$};
\end{tikzpicture} 
& 
\begin{minipage}{65mm} \vspace{3pt} \footnotesize $\begin{pmatrix}
2& \!-2 \\
-2& 2
\end{pmatrix}$ \end{minipage}
&
$d_0=d_1=1$
\\
${\rm A}^{(1)}_{n\ge2}$
&
\begin{tikzpicture} [baseline=-0.25em,scale=0.8,line width=0.7pt]
\draw[thick,domain=0:135] plot ({cos(\x)},{sin(\x)});
\draw[thick,dashed,domain=135:360] plot ({cos(\x)},{sin(\x)});
\filldraw[fill=white] ({-sqrt(2)/2},{sqrt(2)/2}) circle (.1) node[left]{\scriptsize $n$};
\filldraw[fill=white] (0,1) circle (.1) node[above]{\scriptsize 0};
\filldraw[fill=white] ({sqrt(2)/2},{sqrt(2)/2}) circle (.1) node[right]{\scriptsize 1};
\filldraw[fill=white] (1,0) circle (.1) node[right]{\scriptsize 2};
\end{tikzpicture} 
&
\begin{minipage}{65mm} \footnotesize $ \begin{pmatrix}
2 & \!-1 & & & \!-1 \\
-1 & 2 & \!-1 & & \\
 & \!-1 & 2 & \cdot & \\
 & & \cdot & \cdot & \!-1 \\
-1 & & & \!-1 & 2
\end{pmatrix} $ \vspace{2pt} \end{minipage}
&
$d_0= \ldots = d_n = 1$
\\
\hline
${\rm B}^{(1)}_{n\ge3}$
&
\begin{tikzpicture}[baseline=-0.25em,line width=0.8pt,scale=0.9]
\draw[thick] (-.6,.3) -- (0,0) -- (-.4,-.3);
\draw[thick] (0,0) -- (0.5,0);
\draw[thick,dashed] (.5,0) -- (1.5,0);
\draw[thick] (1.5,0) -- (2,0);
\draw[double,->] (2,0) -- (2.4,0);
\filldraw[fill=white] (-.6,.3) circle (.1) node[left=1pt]{\scriptsize $0$};
\filldraw[fill=white] (-.4,-.3) circle (.1) node[left=1pt]{\scriptsize $1$};
\filldraw[fill=white] (0,0) circle (.1) node[above]{\scriptsize $2$};
\filldraw[fill=white] (.5,0) circle (.1) node[above]{\scriptsize $3$};
\filldraw[fill=white] (1.5,0) circle (.1) node[above]{\scriptsize $n\!-\!2$ };
\filldraw[fill=white] (2,0) circle (.1) node[below]{\scriptsize $ n\!-\!1$};
\filldraw[fill=white] (2.5,0) circle (.1) node[right=1pt]{\scriptsize $n$};
\end{tikzpicture}
&
\begin{minipage}{65mm} \vspace{3pt} \footnotesize $\begin{pmatrix}
2 &  & \!-1 & & & & \\
 & 2 & \!-1 & & & & \\
-1 & \!-1 & 2 & \!-1 & & & \\
 & & \!-1 & 2 & \cdot & & \\
 & & & \cdot & \cdot & \!-1 & \\
 & & & & \!-1 & 2 & \!-1 \\
 & & & & & \!-2 & 2\\ 
\end{pmatrix} $ \vspace{2pt}
\end{minipage}
& 
$d_0 = \ldots = d_{n-1} = 1$ \newline
$d_n = \frac{1}{2}$
\\
\hline
${\rm C}^{(1)}_{n\ge2}$
&
\begin{tikzpicture}[baseline=-0.25em,line width=0.8pt,scale=.9]
\draw[double,<-] (-.1,0) -- (-.5,0);
\draw[thick] (0,0) -- (.5,0);
\draw[thick,dashed] (.5,0) -- (1.5,0);
\draw[thick] (1.5,0) -- (2,0);
\draw[double,<-] (2.1,0) --  (2.5,0);
\filldraw[fill=white] (-.5,0) circle (.1) node[left=1pt]{\scriptsize $0$};
\filldraw[fill=white] (0,0) circle (.1) node[above]{\scriptsize $1$};
\filldraw[fill=white] (.5,0) circle (.1) node[above]{\scriptsize $2$};
\filldraw[fill=white] (1.5,0) circle (.1) node[above]{\scriptsize $n\!-\!2$};
\filldraw[fill=white] (2,0) circle (.1) node[below=1pt]{\scriptsize $\! n\!-\!1$};
\filldraw[fill=white] (2.5,0) circle (.1) node[right=1pt]{\scriptsize $n$};
\end{tikzpicture}
&
\begin{minipage}{65mm} \vspace{3pt} \footnotesize $ \begin{pmatrix}
2 & \!-1 &  & & &  \\
-2 & 2 & \!-1 & & & \\
 & -1 & 2 & \cdot & & \\
  &  & \cdot & \cdot & \!-1 & \\
 & & & \!-1 & 2 & \!-2 \\
 & & & & \!-1 & 2\\ 
\end{pmatrix} $ \vspace{2pt} \end{minipage}
&
$d_0 = d_n = 2$ \newline
$d_1 = \ldots = d_{n-1} = 1$
\\
\hline
${\rm D}^{(1)}_{n\ge4}$
&
\begin{tikzpicture}[baseline=-0.35em,line width=0.8pt,scale=.9]
\draw[thick] (-.6,.3) -- (0,0) -- (-.4,-.3);
\draw[thick] (0,0) -- (.5,0);
\draw[thick,dashed] (.5,0) -- (1.5,0);
\draw[thick] (1.5,0) -- (2,0);
\draw[thick] (2.4,.3) -- (2,0) -- (2.6,-.3);
\filldraw[fill=white] (-.6,.3) circle (.1) node[left=1pt]{\scriptsize $0$};
\filldraw[fill=white] (-.4,-.3) circle (.1) node[left=1pt]{\scriptsize $1$};
\filldraw[fill=white] (0,0) circle (.1) node[above]{\scriptsize $2$};
\filldraw[fill=white] (0.5,0) circle (.1) node[above]{\scriptsize $3$};
\filldraw[fill=white] (1.5,0) circle (.1) node[above]{\scriptsize $n\!-\!3$};
\filldraw[fill=white] (2,0) circle (.1) node[below=1pt]{\scriptsize $\!\! n\!-\!2$};
\filldraw[fill=white] (2.4,.3) circle (.1) node[right=1pt]{\scriptsize $n\!-\!1$};
\filldraw[fill=white] (2.6,-.3) circle (.1) node[right=1pt]{\scriptsize $n$};
\end{tikzpicture} 
&
\begin{minipage}{65mm} \vspace{3pt} \footnotesize $ \begin{pmatrix}
2 &  & \!-1 & & & & & \\
 & 2 & \!-1 & & & & & \\
-1 & \!-1 & 2 & \!-1 & & & & \\
 & & \!-1 & 2 & \cdot & & & \\
 & & & \cdot & \cdot & \!-1 & & \\
 & & & & \!-1 & 2 & \!-1 & \!-1 \\
 & & & & & \!-1 & 2 &  \\
 & & & & & \!-1 & & 2\\ 
\end{pmatrix} $ \vspace{2pt}
\end{minipage}
&
$d_0= \ldots = d_n = 1$
\\
\hline
\end{tabular} 

\end{table}

Define the following elements of $\Aut(A)$:
\eq{ \label{AutAelts}
\begin{aligned}
\psi &:= \casesl{l}{\prod_{i=1}^{\lfloor \frac{n}{2} \rfloor} (i,n+1-i) \\[1pt] \id} && \casesm{l}{\\[-6pt] \text{for } {\rm A}_n^{(1)}, \\  \text{otherwise},} 
\\
\psi' &:= \textstyle\prod_{i=0}^{\lfloor \frac{n-1}{2} \rfloor} (i,n-i) && \text{for } {\rm A}_n^{(1)}, \\
\rho &:= (01\ldots n) && \text{for } {\rm A}_n^{(1)}, 
\\
\pi &:= \casesl{l}{\prod_{i=0}^{\frac{n-1}{2}}(i,i+\frac{n+1}{2}) \\[3pt] \prod_{i=0}^{\lfloor \frac{n-1}{2} \rfloor} (i,n-i)} && \casesm{l}{\\[-13pt] \text{for } {\rm A}_n^{(1)} \text{ if } n \text{ is odd,} \\[3pt] \text{for } {\rm C}_n^{(1)} \text{ and } {\rm D}_n^{(1)},} 
\\
\phi_1 &:= (01) && \text{for } {\rm B}_n^{(1)} \text{ and } {\rm D}_n^{(1)}, \\
\phi_2 &:= (n\!-\!1 \, n) && \text{for } {\rm D}_n^{(1)}, \\
\phi_{12} &:= (01)(n\!-\!1 \, n) && \text{for } {\rm D}_n^{(1)}.
\end{aligned}
}

In Table~\ref{AutA} we have summarized relevant properties of $\Aut(A)$.

\begin{table}[ht]
\caption{
Properties of $\Aut(A)$. 
For $A$ of type ${\rm A}^{(1)}_n$ or ${\rm D}^{(1)}_n$ we first list our chosen generators of $\Aut(A)$ and subsequently other relevant elements. 
For $A$ of type ${\rm B}^{(1)}_n$ or ${\rm C}^{(1)}_n$ the group $\Aut(A)$ is of order 2, so there is only one choice of generator (which is automatically an involution).
We denote by $\Cyc_N$, $\Dih_N$ and $\Sym_N$ the cyclic, dihedral and symmetric groups of order $N$, $2N$ and $N!$, respectively. 
}
\label{AutA}
{
\begin{tabular}{c|c|ccc}
Type & $\Aut(A)$ & \multicolumn{2}{c}{Important elements} & Relations \\
\hline 
& &
$\rho$ &
\begin{tikzpicture}[scale=.7,baseline=-3pt]
\draw[white] (0,1.5) -- (0,-1.5);
\draw[thick,domain=-45:90] plot ({cos(\x)},{sin(\x)});
\draw[thick,dashed,domain=90:315] plot ({cos(\x)},{sin(\x)});
\filldraw[fill=white] (0,1) circle (.1) node[above=-2pt]{\scriptsize $N\!\!-\!\!1$};
\filldraw[fill=white] ({sqrt(2)/2},{sqrt(2)/2}) circle (.1) node[right]{\scriptsize 0};
\filldraw[fill=white] (1,0) circle (.1) node[right]{\scriptsize 1};
\filldraw[fill=white] ({sqrt(2)/2},{-sqrt(2)/2}) circle (.1) node[right]{\scriptsize 2};
\draw[->,gray] ({-sqrt(2)/2+.1},{sqrt(2)/2}) .. controls ({-3*sqrt(2)/16},{3*(2+sqrt(2))/16}) .. ({-sqrt(2)/20},{1-sqrt(2)/20});
\draw[->,gray] ({sqrt(2)/20},{1-sqrt(2)/20}) .. controls ({3*sqrt(2)/16},{3*(2+sqrt(2))/16}) .. ({sqrt(2)/2-.1},{sqrt(2)/2}) ;
\draw[->,gray] ({sqrt(2)/2},{sqrt(2)/2-.1}) .. controls ({3*(2+sqrt(2))/16},{3*sqrt(2)/16}) .. ({1-sqrt(2)/20},{sqrt(2)/20});
\draw[->,gray] ({1-sqrt(2)/20},{-sqrt(2)/20}) .. controls ({3*(2+sqrt(2))/16},{-3*sqrt(2)/16}) .. ({sqrt(2)/2},{-sqrt(2)/2+.1});
\draw[->,gray] ({sqrt(2)/2-.1},{-sqrt(2)/2}) .. controls ({3*sqrt(2)/16},{-3*(2+sqrt(2))/16}) .. ({sqrt(2)/20},{-1+sqrt(2)/20}) ;
\end{tikzpicture}
& \multirow{2}{*}[-15pt]{
$ \begin{array}{c}
\rho^N = \psi^2 = \id \\
(\psi \cdot \rho)^2 = \id \\
\end{array} $
}
\\[-.25em]
\multirow{2}{*}[-20pt]{${\rm A}^{(1)}_{n \geq 1}$} & 
\multirow{2}{*}[-20pt]{$\begin{array}{ll} \Dih_N & \hspace{-2mm} \text{if } n \ge 2 \\ \Cyc_2 & \hspace{-2mm} \text{if } n=1 \end{array}$} 
& $\psi$ 
& \hspace{-8pt} \begin{tikzpicture}[scale=.7,baseline=-3pt]
\draw[thick] (0,.4) -- (-.5,0) -- (0,-.4);
\draw[thick,dashed] (0,.4) -- (1,.4);
\draw[thick] (1,.4) -- (1.5,0) -- (1,-.4);
\draw[thick,dashed] (0,-.4) -- (1,-.4);
\draw[<->,gray] (0,.3) -- (0,-.3);
\draw[<->,gray] (1,.3) -- (1,-.3);
\filldraw[fill=white] (-.5,0) circle (.1) node[left=-1pt]{\scriptsize 0};
\filldraw[fill=white] (0,.4) circle (.1) node[above=-1pt]{\scriptsize 1};
\filldraw[fill=white] (0,-.4) circle (.1) node[below=0pt]{\scriptsize $\hspace{-3pt} N\!\!-\!\!1$};
\filldraw[fill=white] (1,.4) circle (.1) node[above=-2pt]{\scriptsize $\tfrac{N}{2}\!\!-\!\!1$};
\filldraw[fill=white] (1,-.4) circle (.1) node[below=-1pt]{\scriptsize $\tfrac{N}{2}\!\!+\!\!1 \hspace{-3pt}$};
\filldraw[fill=white] (1.5,0) circle (.1) node[right=-1pt]{\scriptsize $\tfrac{N}{2}$};
\end{tikzpicture} 
\;
\begin{tikzpicture}[scale=.7, baseline=-3pt]
\draw[thick] (0,.4) -- (-.5,0) -- (0,-.4);
\draw[thick,dashed] (0,.4) -- (1,.4);
\draw[thick,domain=270:450] plot({1+.4*cos(\x)},{.4*sin(\x)});
\draw[thick,dashed] (0,-.4) -- (1,-.4);
\draw[<->,gray] (0,.3) -- (0,-.3);
\draw[<->,gray] (1,.3) -- (1,-.3);
\filldraw[fill=white] (-.5,0) circle (.1) node[left=-1pt]{\scriptsize 0};
\filldraw[fill=white] (0,.4) circle (.1) node[above=-1pt]{\scriptsize 1};
\filldraw[fill=white] (0,-.4) circle (.1) node[below=0pt]{\scriptsize $\hspace{-3pt} N\!\!-\!\!1$};
\filldraw[fill=white] (1,.4) circle (.1) node[above=-2pt]{\scriptsize $\tfrac{N\!-\!1}{2}$};
\filldraw[fill=white] (1,-.4) circle (.1) node[below=-1pt]{\scriptsize $\frac{N\!+\!1}{2} \hspace{-3pt}$};
\end{tikzpicture} 
& \\
\cline{3-5}
& &  & & \\[-1em]
& & $\psi'$ 
& \hspace{17pt} \begin{tikzpicture}[scale=.7, baseline=-4pt]
\draw[thick,domain=90:270] plot({.4*cos(\x)},{.4*sin(\x)});
\draw[thick,dashed] (0,.4) -- (1,.4);
\draw[thick,domain=270:450] plot({1+.4*cos(\x)},{.4*sin(\x)});
\draw[thick,dashed] (0,-.4) -- (1,-.4);
\draw[<->,gray] (0,.3) -- (0,-.3);
\draw[<->,gray] (1,.3) -- (1,-.3);
\filldraw[fill=white] (0,.4) circle (.1) node[above=-1pt]{\scriptsize $0$};
\filldraw[fill=white] (0,-.4) circle (.1) node[below=0pt]{\scriptsize $\hspace{-2pt} N\!\!-\!\!1$};
\filldraw[fill=white] (1,.4) circle (.1) node[above=-2pt]{\scriptsize $\tfrac{N}{2}\!\!-\!\!1 \hspace{-2pt}$};
\filldraw[fill=white] (1,-.4) circle (.1) node[below=-1pt]{\scriptsize $\tfrac{N}{2}$};
\end{tikzpicture} 
\qu ($N$ even) 
& $\psi'= \psi \cdot \rho$  \\[-.25em]
& & $\pi$ & 
\hspace{-6pt} \begin{tikzpicture} [baseline=-0.25em,scale=.7,line width=0.7pt]
\draw[thick,domain=0:90] plot ({cos(\x)},{sin(\x)});
\draw[thick,dashed,domain=90:180] plot ({cos(\x)},{sin(\x)});
\draw[thick,domain=180:270] plot ({cos(\x)},{sin(\x)});
\draw[thick,dashed,domain=270:360] plot ({cos(\x)},{sin(\x)});
\filldraw[fill=white] (0,1) circle (.1) node[above=-1pt] {\scriptsize $N\!\!-\!\!1$};
\filldraw[fill=white] ({sqrt(2)/2},{sqrt(2)/2}) circle (.1) node[right]{\scriptsize 0};
\filldraw[fill=white] (1,0) circle (.1) node[right]{\scriptsize 1};
\filldraw[fill=white] (0,-1) circle (.1) node[below=-1pt]{\scriptsize $\tfrac{N}{2}\!\!-\!\!1$};
\filldraw[fill=white] ({-sqrt(2)/2},{-sqrt(2)/2}) circle (.1) node[left]{\scriptsize $\tfrac{N}{2}$};
\filldraw[fill=white] (-1,0) circle (.1) node[left]{\scriptsize $\tfrac{N}{2}\!\!+\!\!1$};
\draw[<->,gray] (0,.9) -- (0,-.9);
\draw[<->,gray] ({.9*sqrt(2)/2},{.9*sqrt(2)/2}) -- ({-.9*sqrt(2)/2},{-.9*sqrt(2)/2});
\draw[<->,gray] (.9,0) -- (-.9,0);
\end{tikzpicture} 
\hspace{1pt} ($N$ even)  
& $\pi=\rho^{N/2}$ \\
\hline
${\rm B}^{(1)}_{n \geq 3}$ & 
$\Cyc_2$ & 
$\flL$ &
\begin{tikzpicture}[scale=.7, baseline=-3pt]
\draw[white] (0,.9) -- (0,-.8);
\draw[thick] (-.5,.3) -- (0,0) -- (-.5,-.3);
\draw[thick,dashed] (0,0) -- (1,0);
\draw[double,->] (1,0) --  (1.4,0);
\draw[<->,gray] (-.5,.2) -- (-.5,-.2);
\filldraw[fill=white] (-.5,.3) circle (.1) node[left=-1pt]{\scriptsize $0$};
\filldraw[fill=white] (-.5,-.3) circle (.1) node[left=-1pt]{\scriptsize $1$};
\filldraw[fill=white] (0,0) circle (.1) node[above=-1pt]{\scriptsize $2$};
\filldraw[fill=white] (1,0) circle (.1) node[below=-1pt]{\scriptsize $\hspace{-4pt} n\!\!-\!\!1$};
\filldraw[fill=white] (1.5,0) circle (.1) node[above=-1pt]{\scriptsize $n$};
\end{tikzpicture} 
&
$\phi_1^2 = \id$
\\
\hline
${\rm C}^{(1)}_{n \geq 2}$ 
& $\Cyc_2$
& $\pi$
& 
\begin{tikzpicture}[scale=.7, baseline=-4pt]
\draw[white] (0,.9)--(0,-.8);
\draw[thick] (0,.4) -- (-.5,0) -- (0,-.4);
\draw[thick,dashed] (0,.4) -- (1,.4);
\draw[thick,dashed] (0,-.4) -- (1,-.4);
\draw[double,<-] (1.1,-.4) -- (1.5,-.4);
\draw[double,<-] (1.1,.4) -- (1.5,.4);
\draw[<->,gray] (0,.3) -- (0,-.3);
\draw[<->,gray] (1,.3) -- (1,-.3);
\draw[<->,gray] (1.5,.3) -- (1.5,-.3);
\filldraw[fill=white] (-.5,0) circle (.1) node[left=-1pt]{\scriptsize $\tfrac{n}{2}$};
\filldraw[fill=white] (0,.4) circle (.1) node[above=-2pt]{\scriptsize $\tfrac{n}{2}\!\!-\!\!1$};
\filldraw[fill=white] (0,-.4) circle (.1) node[below=-1pt]{\scriptsize $\tfrac{n}{2}\!\!+\!\!1$};
\filldraw[fill=white] (1,.4) circle (.1) node[above=-1pt]{\scriptsize $1$};
\filldraw[fill=white] (1,-.4) circle (.1) node[below=-1pt]{\scriptsize $\hspace{-5pt} n\!\!-\!\!1 $};
\filldraw[fill=white] (1.5,.4) circle (.1) node[right=-1pt]{\scriptsize $0$};
\filldraw[fill=white] (1.5,-.4) circle (.1) node[right=-1pt]{\scriptsize $n$};
\end{tikzpicture} \qu
\begin{tikzpicture}[scale=75/100, baseline=-4pt]
\draw[thick,domain=90:270] plot({.4*cos(\x)},{.4*sin(\x)});
\draw[thick,dashed] (0,.4) -- (1,.4);
\draw[thick,dashed] (0,-.4) -- (1,-.4);
\draw[double,<-] (1.1,.4) -- (1.5,.4);
\draw[double,<-] (1.1,-.4) -- (1.5,-.4);
\draw[<->,gray] (0,.3) -- (0,-.3);
\draw[<->,gray] (1,.3) -- (1,-.3);
\draw[<->,gray] (1.5,.3) -- (1.5,-.3);
\filldraw[fill=white] (0,.4) circle (.1) node[above=-1pt]{\scriptsize $ \tfrac{n\!-\!1}{2}$};
\filldraw[fill=white] (0,-.4) circle (.1) node[below=-1pt]{\scriptsize $ \tfrac{n\!+\!1}{2}$};
\filldraw[fill=white] (1,.4) circle (.1) node[above=-1pt]{\scriptsize $1$};
\filldraw[fill=white] (1,-.4) circle (.1) node[below=-1pt]{\scriptsize $\hspace{-5pt} n\!\!-\!\!1 $};
\filldraw[fill=white] (1.5,.4) circle (.1) node[right=-1pt]{\scriptsize $0$};
\filldraw[fill=white] (1.5,-.4) circle (.1) node[right=-1pt]{\scriptsize $n$};
\end{tikzpicture} 
\hspace{3pt}
& $\pi^2 = \id$ \\[2pt]
\hline
\multirow{4}{*}[-3pt]{${\rm D}^{(1)}_{n > 4}$} & 
\multirow{4}{*}[-3pt]{$\Dih_4$ } 
& $\pi$
&
\begin{tikzpicture}[scale=.7,baseline=-4pt]
\draw[white] (0,1)--(0,-.9);
\draw[thick] (0,.4) -- (-.5,0)-- (0,-.4);
\draw[thick,dashed] (0,.4) -- (1,.4);
\draw[thick,dashed] (0,-.4) -- (1,-.4);
\draw[thick] (1.6,.7) -- (1,.4) -- (1.4,.2);
\draw[thick] (1.6,-.1) -- (1,-.4) -- (1.4,-.6);
\draw[<->,gray] (0,.3) -- (0,-.3);
\draw[<->,gray] (1,.3) -- (1,-.3);
\draw[<->,gray] (1.4,.1) -- (1.4,-.5);
\draw[<->,gray] (1.6,.6) -- (1.6,0);
\filldraw[fill=white] (-.5,0) circle (.1) node[left=-1pt]{\scriptsize $\tfrac{n}{2}$};
\filldraw[fill=white] (0,.4) circle (.1) node[above=-2pt]{\scriptsize $\tfrac{n}{2}\!\!-\!\!1$};
\filldraw[fill=white] (0,-.4) circle (.1) node[below=-1pt]{\scriptsize $\tfrac{n}{2}\!\!+\!\!1$};
\filldraw[fill=white] (1,.4) circle (.1) node[above=-1pt]{\scriptsize $2$};
\filldraw[fill=white] (1,-.4) circle (.1) node[below=-1pt]{\scriptsize $\hspace{-5pt} n\!\!-\!\!2$};
\filldraw[fill=white] (1.4,.2) circle (.1) node[below=-3pt]{\hspace{-11pt} \scriptsize 1};
\filldraw[fill=white] (1.4,-.6) circle (.1) node[right=-1pt]{\scriptsize $n\!\!-\!\!1$};
\filldraw[fill=white] (1.6,.7) circle (.1) node[right=-1pt]{\scriptsize $0$};
\filldraw[fill=white] (1.6,-.1) circle (.1) node[right=-1pt]{\scriptsize $n$};
\end{tikzpicture} \hspace{0pt}
\begin{tikzpicture}[scale=.7,baseline=-4pt]
\draw[thick,domain=90:270] plot({.4*cos(\x)},{.4*sin(\x)});
\draw[thick,dashed] (0,.4) -- (1,.4);
\draw[thick,dashed] (0,-.4) -- (1,-.4);
\draw[thick] (1.6,-.1) -- (1,-.4) -- (1.4,-.6);
\draw[thick] (1.6,.7) -- (1,.4) -- (1.4,.2);
\draw[<->,gray] (0,.3) -- (0,-.3);
\draw[<->,gray] (1,.3) -- (1,-.3);
\draw[<->,gray] (1.4,.1) -- (1.4,-.5);
\draw[<->,gray] (1.6,.6) -- (1.6,0);
\filldraw[fill=white] (0,.4) circle (.1) node[above=-1pt]{\scriptsize $\tfrac{n\!-\!1}{2}$};
\filldraw[fill=white] (0,-.4) circle (.1) node[below=-1pt]{\scriptsize $\tfrac{n\!+\!1}{2}$};
\filldraw[fill=white] (1,.4) circle (.1) node[above=-1pt]{\scriptsize $2$};
\filldraw[fill=white] (1,-.4) circle (.1) node[below=-1pt]{\scriptsize $\hspace{-5pt} n\!\!-\!\!2 $};
\filldraw[fill=white] (1.4,.2) circle (.1) node[below=-3pt]{\hspace{-11pt} \scriptsize 1};
\filldraw[fill=white] (1.4,-.6) circle (.1) node[right=-1pt]{\scriptsize $n\!\!-\!\!1$};
\filldraw[fill=white] (1.6,.7) circle (.1) node[right=-1pt]{\scriptsize $0$};
\filldraw[fill=white] (1.6,-.1) circle (.1) node[right=-1pt]{\scriptsize $n$};
\end{tikzpicture} 
&
\multirow{2}{*}[0pt]{
$ \begin{array}{c}
\pi^2 = \phi_1^2 = \id \\
(\pi \cdot \phi_1)^4 = \id \\
\end{array} $
}\\
&& $\flL$ &
\begin{tikzpicture}[scale=.7,baseline=-4pt]
\draw[white] (0,.6) -- (0,-.6);
\draw[thick] (-.5,.3) -- (0,0) -- (-.5,-.3);
\draw[thick,dashed] (0,0) -- (1,0);
\draw[thick] (1.4,.3) -- (1,0) -- (1.6,-.3);
\draw[<->,gray] (-.5,.2) -- (-.5,-.2);
\filldraw[fill=white] (-.5,.3) circle (.1) node[left=-1pt]{\scriptsize $0$};
\filldraw[fill=white] (-.5,-.3) circle (.1) node[left=-1pt]{\scriptsize $1$};
\filldraw[fill=white] (0,0) circle (.1) node[above=-1pt]{\scriptsize $2$};
\filldraw[fill=white] (1,0) circle (.1) node[below=-1pt]{\scriptsize $\hspace{-4pt} n\!\!-\!\!2$};
\filldraw[fill=white] (1.4,.3) circle (.1) node[right=-1pt]{\scriptsize $n\!\!-\!\!1$};
\filldraw[fill=white] (1.6,-.3) circle (.1) node[right=-1pt]{\scriptsize $n$};
\end{tikzpicture} 
& \\ 
\cline{3-5}
&&
$\flR$ &
\hspace{-5pt}
\begin{tikzpicture}[scale=.7,baseline=-4pt]
\draw[white] (0,.6) -- (0,-.6);
\draw[thick,dashed] (0,0) -- (1,0);
\draw[thick] (-.6,.3) -- (0,0) -- (-.4,-.3);
\draw[thick] (1.5,.3) -- (1,0) -- (1.5,-.3);
\draw[<->,gray] (1.5,.2) -- (1.5,-.2);
\filldraw[fill=white] (-.6,.3) circle (.1) node[left=-1pt]{\scriptsize $0$};
\filldraw[fill=white] (-.4,-.3) circle (.1) node[left=-1pt]{\scriptsize $1$};
\filldraw[fill=white] (0,0) circle (.1) node[above=-1pt]{\scriptsize $2$};
\filldraw[fill=white] (1,0) circle (.1) node[below=-1pt]{\scriptsize $\hspace{-5pt} n\!\!-\!\!2$};
\filldraw[fill=white] (1.5,.3) circle (.1) node[right=-1pt]{\scriptsize $n\!\!-\!\!1$};
\filldraw[fill=white] (1.5,-.3) circle (.1) node[right=-1pt]{\scriptsize $n$};
\end{tikzpicture} 
& $\flR = \pi \cdot \flL \cdot \pi$ \\
& & $\flLR$ &
\begin{tikzpicture}[scale=.7,baseline=-4pt]
\draw[thick] (-.5,.3) -- (0,0) -- (-.5,-.3);
\draw[thick,dashed] (0,0) -- (1,0);
\draw[thick] (1.5,.3) -- (1,0) -- (1.5,-.3);
\draw[<->,gray] (-.5,.2) -- (-.5,-.2);
\draw[<->,gray] (1.5,.2) -- (1.5,-.2);
\filldraw[fill=white] (-.5,.3) circle (.1) node[left=-1pt]{\scriptsize $0$};
\filldraw[fill=white] (-.5,-.3) circle (.1) node[left=-1pt]{\scriptsize $1$};
\filldraw[fill=white] (0,0) circle (.1) node[above=-1pt]{\scriptsize $2$};
\filldraw[fill=white] (1,0) circle (.1) node[below=-1pt]{\scriptsize $\hspace{-5pt} n\!\!-\!\!2$};
\filldraw[fill=white] (1.5,.3) circle (.1) node[right=-1pt]{\scriptsize $n\!\!-\!\!1$};
\filldraw[fill=white] (1.5,-.3) circle (.1) node[right=-1pt]{\scriptsize $n$};
\end{tikzpicture} 
& $\flLR = (\pi \cdot \flL)^2$ \\ 
\hline &&&& \\[-11pt]
\multirow{2}{*}[-6pt]{${\rm D}^{(1)}_4$} & 
\multirow{2}{*}[-6pt]{$\Sym_4$ } & 
\multicolumn{3}{c}{As for $n>4$, and additionally} \\[1pt]
& & (14) & 
\begin{tikzpicture}[scale=.7,baseline=-3pt]
\draw[thick] (-.6,.3) -- (0,0) -- (.4,.3);
\draw[thick] (-.4,-.3) -- (0,0) -- (.6,-.3);
\draw[<->,gray] (-.3,-.3) -- (.5,-.3);
\filldraw[fill=white] (-.6,.3) circle (.1) node[left]{\small 0};
\filldraw[fill=white] (.4,.3) circle (.1) node[right]{\small 3};
\filldraw[fill=white] (0,0) circle (.1) node[above]{\small 2};
\filldraw[fill=white] (-.4,-.3) circle (.1) node[left]{\small 1};
\filldraw[fill=white] (.6,-.3) circle (.1) node[right]{\small 4};
\end{tikzpicture} 
&
\hspace{-7pt} $\begin{array}{c}
(14)^2 = \id, ((14) \cdot \pi)^4 = \id \\
 ((14) \cdot \flL)^3=\id
\end{array}$\hspace{-7pt} 
\\
\hline
\end{tabular}
}
\end{table}



\section{The vector representation} \label{sec:natrep}

In this section, adhering closely to \cite[Sec.~8.4.1]{KlSg} we give the description of the vector (also called natural or first fundamental) representation\footnote{Our choice of the vector representation agrees with \cite[Sec.~8.4.1]{KlSg} up to the Chevalley automorphism; for $U_q(\wh{\mfsl}_N)$ it agrees with \cite[Lem.~3.4]{FrMn}.} $\RT$ of $U_q(\mfgf)$ on $\K^N$, when $\mfg^{\rm fin}$ is $\mfsl_N$, $\mfso_N$ and $\mfsp_N$. 
We then turn this representation into an evaluation representation $\RT_u$ of $U_q(\mfg)$ by means of affinization in the homogeneous grading defined by $x_i \mapsto \del_{i0}$, $y_i \mapsto -\del_{i0}$ (see \cite{KKMMNN}). 
However, we first need to introduce some additional notation that we will make use of throughout the remaining parts of the paper.


\subsection{Notation}

In will be convenient to label the standard basis of $\K^N$ by
\eq{   \renewcommand{\arraystretch}{1.2}
\langle N \rangle = 
\begin{cases}
\{1,2,\ldots,N\} & \text{for } \mfsl_N, \\
\{-n,\ldots,-1,1,\ldots,n\} & \text{for $\mfso_{2n}$ and $\mfsp_{2n}$}, \\
\{-n,\ldots,-1,0,1,\ldots,n\} & \text{for } \mfso_{2n+1} .
\end{cases}
}
Define $\vartheta, \vartheta_i, \nu_i \in \tfrac{1}{2} \Z$ for $i \in \langle N \rangle$ by
\eq{
\label{QQv} 
\vartheta = \casesl{l}{ (-1)^n , \\ 1, \\ -1 ,} \qq
\vartheta_i = \casesl{l}{ (-1)^i , \\ 1, \\ \sgn(i) ,} \qq
\nu_i = \casesl{l}{ i-\tfrac{N+1}{2} \\ i-\sgn(i)(n-\ka) \\ i } \qq
\casesm{l}{ \text{for } \mfsl_N , \\ \text{for } \mfso_N, \\ \text{for } \mfsp_N ,}
}
where
\eq{
\ka = \frac{N}2-\vartheta. \label{kappa}
}
We also note that for orthogonal and symplectic cases our tuple $(-\nu_{-n},\ldots,-\nu_{n})$ corresponds to the tuple $(\rho_{1},\ldots,\rho_N)$ in \cite[\S1]{FRT} and in \cite[Sec.~8.4.2]{KlSg}.

Let $E_{ij}$ denote the usual $N \times N$ elementary matrices, {\it i.e.}~$(E_{ij})_{kl}=\delta_{ik}\delta_{jl}$ with $i,j,k,l \in \langle N \rangle$. 
Given $m_i \in \Z$ for all $i \in \langle N \rangle$, we write 
\[ 
q^{\sum_{i \in \langle N \rangle} m_i E_{ii}} = \sum_{i \in \langle N \rangle} q^{m_i} E_{ii}. 
\]
We denote by $\Id=\sum_{i\in\lan N \ran}E_{ii}$ the identity matrix in $\End(\K^N)$. For orthogonal and symplectic cases, we additionally set 
\eq{
F_{ij} = E_{ij} - \vartheta_{i}\vartheta_{j} E_{-j,-i},\qu\text{so that}\qu F_{ij}+\vartheta_{i}\vartheta_{j} F_{-i,-j}=0
}
and, for ease of notation, we will write $\bar\jmath= n+1-j$ for any $j \in I$. 


\subsection{The representations $\RT$ and $\RT_u$}

The vector representation $\RT : U_q(\mfgf) \to \End(\K^N)$ is the $N$-dimensional irreducible highest-weight representation of $U_q(\mfgf)$ with the highest weight $(q,1,\ldots,1)\in\K^{n}$ and a highest weight vector $(1,0,\ldots,0)^\t\in\K^N$ defined as follows. For $\mfsl_N$ it is
\eqa{
& x_i \mapsto E_{i,i+1} , && y_i \mapsto E_{i+1,i} , &&  k_i \mapsto q^{E_{ii}-E_{i+1,i+1}} &&  \text{for }1\le i \le n. \label{rep:A}
\intertext{For $\mfso_N$ and $\mfsp_N$, upon defining $c=[2]_{q^{1/2}}^{1/2}$, it is}
& x_i \mapsto F_{-\bar\imath,-\bar\imath+1}, &&  y_i \mapsto F_{-\bar\imath+1,-\bar\imath}, && k_i \mapsto q^{F_{\bar\imath-1,\bar\imath-1}-F_{\bar\imath,\bar\imath}} && \text{for } 1\le i< n \label{rep:BCDi} 
}
and
\eqa{
& x_n \mapsto \casesl{l}{ c(E_{-1,0} {-} q^{-\frac12} E_{01}), \\ \frac12 F_{-1,1}, \\ F_{-2,1},} \hspace{-3mm} &&
y_n \mapsto \casesl{l}{ c(E_{0,-1} {-} q^{\frac12} E_{10}), \\ \frac12 F_{1,-1}, \\ F_{1,-2}, } \hspace{-3mm} &&  
k_n \mapsto \casesl{l}{ q^{-F_{11}}  \\ q^{-2 F_{11}}\\ q^{-F_{11}-F_{22}} } && 
\casesm{l}{ \text{for }\mfso_{2n+1}, \\ \text{for }\mfsp_{2n}, \ \\ \text{for }\mfso_{2n}.} \label{rep:BCDn}
\intertext{\qu Consider the natural inclusion of Hopf algebras $\iota : U_q(\mfgf) \to U_q({\mfg})$. 
Let $u\in \K^\times$. We extend the highest-weight representation $\RT : U_q(\mfgf) \to \End(\K^N)$ to a pseudo-highest-weight representation \mbox{$\RT_u : U_q({\mfg}) \to \End(\K^N)$}, depending polynomially on $u^{\pm1}$. 
We fix $\RT_u(\iota(a)) :=\RT(a)$ for all $a \in U_q(\mfgf)$ and for the remaining generators $x_0$, $y_0$ and $k_0$ we define}
& x_0 \mapsto \casesl{l}{ u\, E_{N1}, \\ u F_{n-1,-n}, \\ \tfrac12 u F_{n,-n}, } &&
y_0 \mapsto \casesl{l}{ u^{-1} E_{1N} ,\\ u^{-1} F_{-n,n-1}, \\ \tfrac12 u^{-1} F_{-n,n}, } &&  
k_0 \mapsto \casesl{l}{ q^{E_{NN}-E_{11}} \\ q^{F_{n-1,n-1}+F_{nn}} \\ q^{2 F_{nn}} } \hspace{-2mm} && 
\casesm{l}{ \text{for }\wh\mfsl_{N} \\ \text{for }\wh\mfso_{N}, \\ \text{for }\wh\mfsp_{N}.} \label{affrep}
}
It is a direct computation to verify that the assignments above provide a representation of $U_q({\mfg})$.


\subsection{Self-duality of representations $\RT$ and $\RT_u$} \label{sec:T-dual}

The representations $\RT$ and $\RT_u$ exhibit an additional property called (pseudo-)self-duality (see \cite[Sec.~8.4.2]{KlSg}) unless $A$ is of type ${\rm A}^{(1)}_n$ with $n>1$.  
Here we formulate a more general property which also includes type ${\rm A}^{(1)}_{n>1}$. 

\medskip

Let $\mc H$ be a Hopf algebra with antipode $S$ and let $\varrho$ be a representation of $\mc H$ on a finite-dimensional vector space $V$ over a field $F$.
The \emph{dual representation} $\varrho^*$ of $\mc H$ on $V^* := \Hom_F(V,F)$ is defined~by
\eq{ \label{dual:1}
(\varrho^*(a)\, \phi) (v) = \phi(\varrho(S^{-1}(a))\, v) \qu\text{for all } a\in \mc H,\; v\in V,\; \phi \in V^* .
}

If $\varrho^*$ is equivalent to $\varrho$,  i.e.~if there exists an invertible $\phi \in \Hom(V,V^*)$ such that for all $a \in \mc H$ we have $\varrho^*(a) \circ \phi = \phi \circ \varrho(a)$, then $\varrho$ is called a {\it self-dual} representation. 
The following definition is a generalization of this notion.

\begin{defn}
Let $\mc H$ be a Hopf algebra and $\si \in \Aut_{\rm Hopf}(\mc H)$ an involution.
Let $\varrho$ be a representation of $\mc H$ on a finite-dimensional vector space $V$ over a field $F$.
Then $\varrho$ is called \emph{$\si$-skewed self-dual} if $\varrho^*$ is equivalent to $\varrho \circ \si$. \hfill \defnend
\end{defn}

Recall the involution $\psi$ from \eqref{AutAelts}; note that $\psi$ restricts to an element of $\Aut(A^{\rm fin})$. 
We will show below that $\RT$ and $\RT_u$ are $\psi$-skewed self-dual representations. Let us identify the vector spaces $(\K^N)^*$ and $\K^N$ in the natural way. 
Denote the usual transposition in $\End(\K^N)$ by $\t$, {\it viz.}~$E_{ij}^\t = E_{ji}$ and write $\RT^\t(a) = (\RT(a))^\t$ for all $a \in U_q(\mfgf)$ and $\RT_u^\t(a) = (\RT_u(a))^\t$ for all $a \in U_q({\mfg})$.
Then the $\psi$-skewed self-duality of $\RT$ is equivalent to the statement that there exists $C\in\GL(\K^N)$ such that
\eq{ \label{dual:2} 
C \, \RT( \psi(a)) = \RT^{\,\t}(S(a))\, C \qu\text{for all }a\in U_q(\mfgf). 
}
For the representation $\RT_u$ we need to allow a shift in $u$. 
It follows that $\RT_{u}$ is $\psi$-skewed self-dual if and only if there exists $C\in\GL(\K^N)$ and $\wt q \in \K^\times$ such that
\eq{ \label{dual:3} 
C \, \RT_{u}(\psi(a)) =  \RT^{\,\t}_{\wt q^{\, 2} u}(S(a))\, C \qu\text{for all }a\in U_q({\mfg}).
}

We now give explicit expressions for $C$ and $\wt q$, using which we can derive \eqrefs{dual:2}{dual:3} (in other words, the $\psi$-skewed self-duality properties):
\eq{ \label{def:Cwtq}
C = \begin{cases} 
\displaystyle \sum_{i \in \langle N \rangle} \vartheta_i q^{\nu_i} E_{N+1-i,i} \\
\displaystyle \sum_{i \in \langle N \rangle} \vartheta_i q^{\nu_i} E_{-i,i}  
\end{cases}
\qu
\wt q = \begin{cases} 
(-1)^{N/2} q^{N/2} & \hspace{20pt} \text{for }\mfsl_{N}, \\[1.25em] 
q^{\ka} & \hspace{20pt} \text{otherwise}.
\end{cases}
}
Note that $C^2 = \vartheta\,\Id$.
Taking \eqref{def:Cwtq} into account, we have the following statements.

\begin{lemma} \label{L:T-dual}
The representation $\RT$ of $U_q(\mfgf)$ is $\psi$-skewed self-dual.
\end{lemma}

\begin{proof}
We only need to show this property for the generators of the algebra. For $\mfso_N$ and $\mfsp_N$ it was shown in Prop.~20 in \cite[Sec.~8.4.2]{KlSg}. For $\mfsl_N$ it follows by a direct computation using
\eq{ \label{CEC} 
C^{-1} E_{ij} C = (-q)^{i-j} E_{N-i+1,N-j+1}.
}
Recall that $\psi : x_i \mapsto x_{N-i}$, $y_i \mapsto y_{N-i}$ and $k_i \mapsto k_{N-i}$. Thus 
\eqn{
C^{-1}\RT^\t(S(k_i))\,C &= C^{-1}q^{E_{i+1,i+1}-E_{ii}}C = q^{E_{N-i,N-i}-E_{N-i+1,N-i+1}} = \RT(\psi(k_i)) , \\
C^{-1}\RT^\t(S(x_i))\,C &= -C^{-1} E_{i+1,i}\, q^{E_{i+1,i+1}-E_{ii}} C = E_{N-i,N-i+1} = \RT(\psi(x_i)) , \\
C^{-1}\RT^\t(S(y_i))\,C &= -C^{-1} q^{E_{ii}-E_{i+1,i+1}} E_{i,i+1} C = E_{N-i+1,N-i} = \RT(\psi(y_i)) . \qedhere
}
\end{proof}

\begin{lemma} \label{L:Tu-dual}
For any $u\in \K^\times$ the representation $\RT_u$ of $U_q({\mfg})$ is $\psi$-skewed self-dual.
\end{lemma}

\begin{proof}
It is enough to check that \eqref{dual:3} holds for $x_i$, $y_i$ and $k_i$ with $0\le i\le n$. By Lemma \ref{L:T-dual} we already know that this is true when $1\le i\le n$. Thus we only need to check that \eqref{dual:3} holds when $i=0$, which again follows by a direct computation. For $\wh\mfsl_N$, using \eqref{CEC}, we have 
\eqn{
C^{-1}\RT^\t_{\wt q^2u}(S(k_0))\,C &= C^{-1} q^{E_{11}-E_{NN}} C = q^{E_{NN}-E_{11}} = \RT_{\wt q^2u}(\psi(k_0)) , \\
C^{-1}\RT^\t_{\wt q^2u}(S(x_0))\,C &= -(-q)^{N} u\, C^{-1} E_{1N} q^{E_{11}-E_{NN}} C = u\,E_{N1} = \RT_{u}(\psi(x_0)) , \\
C^{-1}\RT^\t_{u}(S(y_0))\,C &= -(-q)^{-N} u^{-1} C^{-1} q^{E_{NN}-E_{11}} E_{N1} C = u^{-1}E_{1N} = \RT_{u}(\psi(y_0)) . 
\intertext{Next, analogous to \eqref{CEC}, for $\wh\mfso_N$ and $\wh\mfsp_N$ we derive $C^{-1} F_{ij} C = -q^{\nu_i-\nu_j} F_{ji}$. Thus, for $\wh\mfso_N$ we have $\wt q^2 = q^{N-2} = q^{2\nu_n}$. 
It follows that}
C^{-1}\RT^\t_{\wt q^2u}(S(k_0))\,C &= C^{-1} q^{-F_{n-1,n-1}-F_{nn}}C = q^{F_{n-1,n-1}+F_{nn}} = \RT_{u}(\psi(k_0)) , \\
C^{-1}\RT^\t_{\wt q^2u}(S(x_0))\,C &= -q^{2\nu_n} u\, C^{-1} F_{-n,n-1}\,q^{-F_{n-1,n-1}-F_{nn}} C = u\,F_{n-1,-n} = \RT_{u}(\psi(x_0)) , \\
C^{-1}\RT^\t_{\wt q^2u}(S(y_0))\,C &= -q^{-2\nu_n}u^{-1} C^{-1} q^{F_{n-1,n-1}+F_{nn}} F_{n-1,-n}\, C = u^{-1} F_{-n,n-1} = \RT_{u}(\psi(y_0)) . 
\intertext{For $\wh\mfsp_N$ we have $\wt q^2 = q^{N+2}= q^{2\nu_n+2}$ and}
C^{-1}\RT^\t_{\wt q^2u}(S(k_0))\,C &= C^{-1} q^{-2F_{nn}}C = q^{2F_{nn}} = \RT_{u}(\psi(k_0)) , \\
C^{-1}\RT^\t_{\wt q^2u}(S(x_0))\,C &= -\tfrac12 q^{2\nu_n+2} u\, C^{-1} F_{-n,n}\,q^{-2F_{nn}} C = \tfrac12 u\,F_{n,-n} = \RT_{u}(\psi(x_0)) , \\
C^{-1}\RT^\t_{\wt q^2u}(S(y_0))\,C &= -\tfrac12 q^{-2\nu_n-2}u^{-1} C^{-1} q^{2F_{nn}} F_{n,-n}\, C = \tfrac12 u^{-1} F_{-n,n} = \RT_{u}(\psi(y_0)) . \qedhere
} 
\end{proof}


\section{R-matrices and the quantum Yang-Baxter equation} \label{sec:Rmat}

This section summarizes and elaborates on various results involving trigonometric solutions of the quantum Yang-Baxter equation listed in \cite{Bzh,FRT,Ji2,Ji3,KlSg}.


\subsection{Constant R-matrices}

We start by defining certain linear operators acting on the tensor product $\K^N \ot \K^N$.
The identity matrix in $\End(\K^N)^{\ot 2}$ given by $\sum_{i,j \in \langle N \rangle} E_{ii} \ot E_{jj}$ will be denoted by $\Id$. (Note that $\Id$ also denotes the identity matrix in $\End(\K^N)$; it will always be clear from the context which $\Id$ is used.) We introduce the permutation operator $P$ and a projection-like operator $Q_q$, both in $\End(\K^N)^{\ot 2}$, by
\eq{ \label{PQ:defn}
P = \sum_{i,j \in \langle N \rangle} E_{ij} \ot E_{ji}, \qq Q_q = \begin{cases} 
\; 0 & \text{for } \mfsl_N, \\[.4em]
\displaystyle \sum_{i,j \in \langle N \rangle} \vartheta_i \vartheta_j q^{\nu_i-\nu_j} E_{ij} \ot E_{-i,-j} & \text{otherwise},
\end{cases}
}
which satisfy $P^2=\Id$ and $Q_q^2 = N Q_q$. Next, introduce the constant R-matrix $R_q\in \End(\K^N)^{\ot 2}$. According to \cite[\S1]{FRT}, for $\mfsl_N$ it is
\eqa{
R_q = \sum_{i,j \in \langle N \rangle} \left( q^{\delta_{ij}}E_{ii}\ot E_{jj} + \delta_{i<j} (q-q^{-1}) E_{ij}\ot E_{ji} \right) . \label{R:A}
}
For $\mfso_N$ and $\mfsp_N$ it is
\eqa{
R_q = \sum_{i,j \in \langle N \rangle} \left( q^{\delta_{ij}-\delta_{i,-j}}E_{ii}\ot E_{jj} + \delta_{i<j} (q-q^{-1})  (E_{ij}\ot E_{ji} - \vartheta_i \vartheta_j q^{\nu_i-\nu_j} E_{ij}\ot E_{-i,-j}) \right) . \label{R:BCD}}
The R-matrix $R_q$ is a solution of the constant quantum Yang-Baxter equation, {\it viz.}
\eq{ \label{R:YBE}
R_{12} R_{13} R_{23} = R_{23} R_{13} R_{12},
}
where $R_{ij}$ denotes $R_q$ acting nontrivially on the $i$-th and $j$-th factors of $(\K^N)^{\ot3}$ only.  It can be checked by a direct computation that matrices defined in \eqrefs{PQ:defn}{R:BCD} satisfy the identity
\eq{ \label{PQR} 
R_q - PR_q^{-1}P = (q-q^{-1}) (P-Q_q). 
}

In order to state some useful properties of the matrices $P$, $Q_q$ and $R_q$ that will be important later we need the following. We will use the notation $\t$ for transposition in $\End(\K^N)^{\ot 2}$, viz. $(E_{ij}\ot E_{kl})^\t = E_{ji}\ot E_{lk}$. 
Furthermore, $\t_1$ and $\t_2$ will denote partial transpositions on $\End(\K^N)^{\ot 2}$ with respect to the first and the second tensor factor, respectively: $(E_{ij}\ot E_{kl})^{\t_1} = E_{ji}\ot E_{kl}$ and $(E_{ij}\ot E_{kl})^{\t_2} = E_{ij}\ot E_{lk}$. 
\begin{description}

\item[PT-symmetry] We have
\eq{\label{R:PT-symm} P^\t=P, \qq Q_q^\t = PQ_qP, \qq R_q^\t = P R_q P.  }

\item[Bar-symmetry] We have
\eq{ \label{R:bar} 
Q_{q^{-1}} = P Q_q P, \qu R_{q^{-1}} = R_q^{-1}.  
}
\item[C-symmetry] In the orthogonal and symplectic cases we have the following identities involving $C$:
\eq{ \label{R:C-symm} 
\begin{aligned}
P^{\t_1} &= C_1 Q_{q^{-1}} C_1^{-1}, & Q_q^{\t_1} &= C_1 P C_1^{-1}, &  R_q^{\t_1} &= C_1 R_q^{-1} C_1^{-1}, \\
 P^{\t_2} &= C_2^\t Q_{q^{-1}} (C_2^\t)^{-1} ,  & Q_q^{\t_2} &= C_2^\t P (C_2^\t)^{-1},  &   R_q^{\t_2} &= C_2^{\t} R_q^{-1} (C_2^\t)^{-1},
\end{aligned}
}
where $C_1 = C \ot \Id$, $C_2 = \Id \ot C$. For $\mfsl_2$ we have
\eq{ \label{R:C-symm:sl2} 
\begin{aligned}
P^{\t_1} &= \frac{q \, C_1 R_q^{-1} C_1^{-1} - q^{-1} C_2 R_q C_2^{-1} }{q-q^{-1}}, & R_q^{\t_1} &= q C_1 R_q^{-1} C_1^{-1}, \\
P^{\t_2} &= \frac{q \, C^\t_2 R_q^{-1} (C^\t_2)^{-1} - q^{-1} C^\t_1 R_q (C^\t_1)^{-1} }{q-q^{-1}}, & R_q^{\t_2} &= q C_2^\t R_q^{-1} (C_2^\t)^{-1}.
\end{aligned}
}
In all cases, including $\mfsl_{N>2}$, we have 
\eq{  \label{R:RCC} 
C_1 C_2 P =  P^\t C_1 C_2, \quad C_1 C_2 Q_q =  Q_q^\t C_1 C_2, \quad C_1 C_2 R_q =  R_q^\t C_1 C_2.
}
\item[Eigenvalues and polynomial identities] 
Writing $\hat R_q = P R_q$ and $\hat Q_q = P Q_q$, from \eqref{PQR} we infer
\eqa{ 
\hat{Q}_q\hat{R}_q &= \hat{R}_q \hat{Q}_q=\vartheta q^{\vartheta -N} \hat{Q}_q, \label{hatQhatR} \\
\hat{Q}_q^2 &= \left(1+\vartheta\, [N-\vartheta]_q \right) \hat{Q}_q, \label{hatQhatQ}
}
and the identities
\eqg{
\hat{R}_q^2 = (q-q^{-1})\bigl(\hat{R}_q - \vartheta q^{\vartheta-N} \hat Q_q\bigr) + \Id, \label{R:min-id:1} \\
(\hat{R}_q - q \, \Id)(\hat{R}_q + q^{-1} \Id)(\hat{R}_q-\vartheta q^{\vartheta-N}\Id) = 0, \label{R:min-id:2}
}
which simplify to $(\hat{R}_q - q \, \Id)(\hat{R}_q + q^{-1} \Id)=0 $ for $\mfsl_N$.
More precisely, the eigenvalues $q$, $-q^{-1}$ and $\vartheta q^{\vartheta-N}$ of $\hat{R}_q$ have the following multiplicities:
\eq{ \label{evmultiplicities}
\begin{array}{llll}
\binom{N+1}{2}, & \binom{N}{2}, & 0 & \text{ for } \mfsl_N, \\[.5em]
\binom{N+1}{2} \!-\! 1, & \binom{N}{2}, & 1 & \text{ for } \mfso_N, \\[.5em]
\binom{N+1}{2}, & \binom{N}{2} \!-\! 1, & 1 & \text{ for } \mfsp_N,
\end{array}
}
which are the dimensions of the symmetric, antisymmetric and trivial representations~of~$\mfgf$.

\end{description}


\subsection{R-matrices with spectral parameter} \label{sec:Ru}

The R-matrix $R(u)\in\End(\K^N)^{\ot2}$, depending rationally on the {\it spectral parameter} $u$, is defined by 
\eq{
R_{q}(u) = R(u) = f_q(u)\,R_q + \frac{(q-q^{-1})\hspace{0.1mm}u}{q-q^{-1}u} \left( P - \frac{1-u}{q^{2\kappa}-u}\, Q_q \right)  \tx{where} f_q(u) =  \frac{1-u}{q-q^{-1}u} .
\label{Ru:defn}
}
The constant R-matrices are recovered by setting $u=0$: $R_q(0) = q^{-1} R_q$.
We will mostly write $R(u)$ and reserve the notation $R_{q}(u)$ to special cases such as $R_{q^{-1}}(u)$ when needed.
These R-matrices are solutions of the \emph{quantum Yang-Baxter equation with spectral parameters}, 
\eq{ \label{YBE}
R_{12}(\tfrac uv)\,R_{13}(\tfrac uw)\,R_{23}(\tfrac vw) = R_{23}(\tfrac vw) \,R_{13}(\tfrac uw)\,R_{12}(\tfrac uv),
}
where $R_{ij}(u)$ is $R(u)$ acting nontrivially on the $i$-th and $j$-th factors of $(\K^N)^{\ot3}$. This can be checked directly using \eqref{R:YBE}, the Yang-Baxter equation for $R_q$ \eqref{R:YBE} and identities such as \eqref{PQR}. 

The R-matrices with spectral parameter satisfy many additional properties that will be relevant later on. 
First we list properties in which the spectral parameter plays a role; these are not direct analogues of properties of $R_q$. 
\begin{description} [itemsep=0.5ex]

\item[Regularity] Immediately from \eqref{Ru:defn} we have
\eq{ \label{Ru:reg} R(1) = P. }

\item[Unitarity] From \eqrefs{YBE}{Ru:reg} it follows that $R(u)R_{21}(u^{-1})$ is a scalar multiple of $\Id$, where $R_{21}(u):=PR(u)P$.
Owing to the chosen normalization we in fact have, for generic values of $u$,
\eq{ 
\label{Ru:unit} R(u)^{-1} = R_{21}(u^{-1})
}
and, in particular, $R(u)$ is invertible. 

\item[Affinization identity] 
From \eqref{PQR} and \eqref{Ru:defn} one finds that $\hat R(u) := P R(u)$ satisfies
\spl{ \label{Ru:affinization} 
\hat R(u) &= \frac{ \hat R_q - u \hat R_q^{-1}}{q-q^{-1}u} + \frac{(q-q^{-1})(q^\ka - q^{-\ka}) u}{(q-q^{-1} u)(q^\ka - q^{-\ka} u)} \hat Q_q \\
&= \frac{(1-u)(q^\ka \hat R_q - q^{-\ka} u \hat R_q^{-1}) + (q-q^{-1})(q^\ka - q^{-\ka} )u \,\Id}{(q-q^{-1}u)(q^\ka-q^{-\ka}u)}.
}
Relations such as \eqref{Ru:affinization} are also known as Baxterization identities, see \cite[Sec.~8.7.1]{KlSg}.
\end{description}

\medskip

In addition to the properties listed above, the matrix $R(u)$ inherits properties obeyed by the matrices $P$, $Q_q$ and $R_q$. 
\begin{description} [itemsep=0.5ex]

\item[PT-symmetry] As a consequence of \eqref{R:PT-symm} we have
\eq{\label{Ru:PT-symm} R(u)^\t = R_{21}(u), \qq R(u)^{\t_i} = P R(u)^{\t_i} P. }

\item[Bar-symmetry] 
Define $\widebar R(u) = R_{q^{-1}}(u)$. Then from \eqref{R:bar} one obtains
\eq{ 
\label{Ru:bar} \widebar R(u) = R(u)^{-1}.
}

\item[C-symmetry] For $\mfso_N$, $\mfsp_N$ and $\mfsl_{2}$, owing to \eqrefs{R:C-symm}{R:RCC} and \eqref{Ru:unit}, we obtain
\eq{
\label{Ru:C-symm} R(u)^{\t_1} = \wt f_q(u) C_1 R(\wt q^{-2}u)^{-1} C_1^{-1}, 
\qq
R(u)^{\t_2} = \wt f_q(u) C_2^\t R(\wt q^{-2}u)^{-1} (C_2^\t)^{-1},
}
where $\wt f_q(u) := f_q(u)$ for $\mfsl_2$ and $\wt f_q(u) := \frac{f_q(u)}{f_q(q^{2\ka}u^{-1})}$ otherwise.
For all types of $A$ (including $\mfsl_{N>2}$) we have
\eq{ \label{Ru:RCC} 
R(u)^\t\, C_1 C_2 = C_1 C_2 R(u). 
}

\item[Eigenvalues and polynomial identities] As a consequence of (\ref{hatQhatR}-\ref{R:min-id:2}), the matrix $\hat R(u) $ 
satisfies 
\eq{\hat{Q}_q \hat R(u) = \hat R(u)\, \hat{Q}_q = \vartheta \,\frac{q^{-\vartheta} - q^\vartheta u}{q-q^{-1} u} \frac{q^{-\kappa} - q^{\kappa} u}{q^{\kappa}-q^{-\kappa} u}\, \hat{Q}_q , \label{hatQhatRu}}
and the identities
\eqg{
\hat R(u)^2 = \frac{(q-q^{-1})(1+u)}{q-q^{-1}u} \biggl( \hat R(u) - \vartheta \, \frac{q^{-\vartheta} - q^\vartheta u}{q-q^{-1} u} \left( \frac{1-u}{q^\kappa - q^{-\kappa} u} \right)^2 \,\hat{Q}_q \biggr)  + \frac{q^{-1}-q u}{q-q^{-1} u} \, \Id,  \label{Ru:min-id:1} \\[0.25em]
\biggl(\hat R(u) - \Id\biggr)  \biggl(\hat R(u) + \frac{q^{-1}-qu}{q-q^{-1}u} \,\Id\biggr) \biggl( \hat R(u) - \vartheta  \, \frac{q^{-\vartheta} - q^\vartheta u}{q-q^{-1} u} \frac{q^{-\kappa} - q^{\kappa} u}{q^{\kappa}-q^{-\kappa} u} \,\Id \biggr) = 0, \label{Ru:min-id:2}
}
which for $\mfsl_N$ simplifies to
\[ 
\Big( \hat R(u) - \Id \Big)  \Big( \hat R(u) + \frac{q^{-1}-qu}{q-q^{-1}u} \,\Id \Big) =0 .
\]
The multiplicities of the eigenvalues $1$, $-\frac{q^{-1}-qu}{q-q^{-1}u}$ and $\vartheta  \, \frac{q^{-\vartheta} - q^\vartheta u}{q-q^{-1} u} \frac{q^{-\kappa} - q^{\kappa} u}{q^{\kappa}-q^{-\kappa} u}$ of $\hat{R}(u)$ are as~in~\eqref{evmultiplicities}. 
\end{description}

\begin{rmk} 
Note that the R-matrices $R_q(u)$ for $\mfsl_2$ and $R_{q^{1/2}}(u)$ for $\mfsp_2$ are both equal to 
\eq{ \label{Ru:Baxter}
R(u) = P + f_q(u) \big( E_{11} \ot E_{22} + E_{22} \ot E_{11} - q^{-1} E_{12} \ot E_{21} - q E_{21} \ot E_{12} \big)
}
which, up to a similarity transformation, is Baxter's R-matrix for the six-vertex model \cite{Ba2}. \hfill \rmkend
\end{rmk}


\subsection{The R-matrix as intertwiner of vector representations} \label{sec:Rint}

The R-matrix $R(\tfrac uv)$ is the intertwiner for the representation $\RT_u \ot \RT_v$, that is 
\eq{ \label{R:intw} 
R(\tfrac uv)\,( \RT_u \ot \RT_v )( \Delta^{\rm op} (a))  = ( \RT_u \ot \RT_v )(\Delta (a))\, R(\tfrac uv) \qu\text{for all } a \in U_q(\mfg).
}
This equality can be used to define the R-matrix up to a scalar factor, provided \eqref{R:intw} has a nonzero solution. 
This follows from the well-known fact that for generic values of $u/v$ the tensor product $\RT_u \ot \RT_v$ is an irreducible representation of $U_q(\mfg)$ and an application of Schur's lemma. Indeed, Schur's lemma guarantees that $R(\frac uv)$ is invertible for generic values of $u/v$. Now let $R'(\tfrac uv)$ be any other nonzero solution to \eqref{R:intw}, then $( \RT_u \ot \RT_v )(\Delta (a))\, R(\tfrac uv)\, R'(\tfrac uv)^{-1} = R(\tfrac uv)\,( \RT_u \ot \RT_v )( \Delta^{\rm op} (a))\, R'(\tfrac uv)^{-1}  = R(\tfrac uv)\, R'(\tfrac uv)^{-1}\,( \RT_u \ot \RT_v )( \Delta(a))$. Since $\K$ is algebraically closed, it follows by Schur's lemma ({\it e.g.}~as formulated in \cite[Sec.~1.3]{EGHLSVY}) that $R(\tfrac uv)\, R'(\tfrac uv)^{-1}\in \K^\times \Id$. In \cite{Ji2} this approach was taken to construct R-matrices for the vector representation $\RT_u$ when $\mfg$ ranges over all (untwisted and twisted) affine Lie algebras of classical Lie type.

Furthermore, the properties listed in Section \ref{sec:Ru} can be obtained, even without relying on an explicit expression for $R(u)$. Most importantly, the Yang-Baxter equation \eqref{YBE} can be derived by observing that both sides of \eqref{YBE} intertwine the action of $U_q(\mfg)$ on $(\K^N)^{\ot 3}$ given by $(\RT_u \ot \RT_v \ot \RT_w)((\Delta \ot \id)(\Delta(a)))$. Since $\RT_u \ot \RT_v \ot \RT_w$ is also an irreducible representation of $U_q(\mfg)$ for generic values of $u/v$ and $v/w$, both sides must be identical up to a scalar.
To show that the scalar factor equals 1 one uses that $R(u)$ has a constant eigenvector of the form $e \ot e$ with $e \in \K^N$. Alternatively (and this argument can be applied in a more general context), one takes determinants to show that the scalar factor must be an $N^3$-th root of unity; since it also depends rationally on the spectral parameters it must be constant with respect to $u/v$ and $v/w$. By setting $u/v=v/w=1$ and using \eqref{Ru:reg} one sees that the factor~is~$1$.

\begin{rmk} \label{R:uni-R} \mbox{ }
\begin{enumerate} 
\item 
We note that the representation $\RT \ot \RT$ of the quantum group $U_q(\mfgf)$ is reducible. 
In particular, the equality $R_q ( \RT \ot \RT )( \Delta (a))  = ( \RT \ot \RT )(\Delta^\mathrm{op} (a)) R_q$ for all $a\in U_q(\mfgf)$ ({\it cf.}~\eqref{R:intw}) does not have a unique solution, thus cannot be viewed as the defining relation for the constant R-matrices \eqref{R:A} and \eqref{R:BCD}. 

\item 
Recall that a universal R-matrix $\mc{R}$ of a Hopf algebra $\mc{H}$ (see~{\it e.g.}~\cite{Dr1}) is an invertible element in $\mc{H} \ot \mc{H}$ satisfying the following properties:
\eqg{
\mc{R}\, \Delta^{\rm op}(a) = \Delta(a)\, \mc{R}   \qu\text{for all }\, a\in \mc{H},  \label{R-uni} \\
(\Delta \ot \id)(\mc{R} ) = \mc{R}_{23} \mc{R}_{13},\qq
(\id \ot \Delta)(\mc{R} ) = \mc{R}_{12} \mc{R}_{13},  \label{R-uni-2}
}
from which one obtains
\eqg{
\mc{R}_{12} \mc{R}_{13} \mc{R}_{23} = \mc{R}_{23} \mc{R}_{13} \mc{R}_{12}. \label{uni-YBE}
}
Given representations $\rho_i: \mc{H} \to \End(V_i)$ for some vector space $V_i$ with $i\in\{1,2,3\}$, denote $R^{(i,j)}=(\rho_i \ot \rho_j)(\mc{R})$ for $i,j\in\{1,2,3\}$.
Then applying $\rho_1 \ot \rho_2 \ot \rho_3$ to \eqref{uni-YBE} and viewing $R^{(i,j)}$ as an element of $\End(V_1\ot V_2\ot V_3)$ acting nontrivially on $V_i\ot V_j$ only, one obtains the Yang-Baxter equation 
$$
R^{(1,2)} R^{(1,3)} R^{(2,3)} = R^{(2,3)} R^{(1,3)} R^{(1,2)},
$$
whether $\rho_1 \ot\rho_2 \ot \rho_3$ is irreducible or not.
Such a universal R-matrix $\mc{R}$ exists if $\mc{H} = U_q(\mfgf)$, but not if $\mc{H} = U_q(\mfg)$.
Instead one must use a certain completion of the larger Hopf algebra $U_q(\mfg^{\rm ext})$ (see Remark \ref{rmk:extendedaffqg}).
\end{enumerate}
\end{rmk}


\section{K-matrices and the reflection equation} \label{sec:Kmat}


\subsection{Reflection equation}

Solutions of the constant reflection equation corresponding to Riemannian quantum symmetric spaces of classical type have been obtained in \mbox{\cite[Sec.~3]{NoSu}} and \cite[Sec.~2]{NDS}. 
In this section we will introduce versions of the (parameter-dependent) reflection equation (RE) and discuss these in the context of irreducible finite-dimensional representations of suitable coideal subalgebras of a Hopf algebra.

More precisely, let $R(u)$ be one of the R-matrices with spectral parameter discussed in Section~\ref{sec:Ru}. 
Our goal is to find invertible K-matrices $K(u) \in \End(\K^N)$, depending rationally on $u$, that are solutions of the \emph{untwisted reflection equation} \cite{Ch1}
\eqa{ 
R_{21}(\tfrac{u}{v})\, K_1(u)\, R(uv)\, K_2(v) &= K_2(v)\, R_{21}(uv)\, K_1(u)\, R(\tfrac{u}{v}) \label{RE} , 
}
and the \emph{twisted reflection equation} \cite{KuSk2}
\eqa{ 
R(\tfrac{u}{v})\, K_1(u)\, R(\tfrac{1}{uv})^{{\t}_1} K_2(v) &= K_2(v)\, R(\tfrac{1}{uv})^{{\t}_1} K_1(u)\, R(\tfrac{u}{v}) , \label{tRE} 
}
where $K_1(u) = K(u) \ot \Id$ and $K_2(u) = \Id \ot K(u)$. Using unitarity \eqref{Ru:unit} and PT-symmetry \eqref{Ru:PT-symm}, we note that the twisted reflection equation can be written alternatively as
\eqa{
R_{21}(\tfrac{u}{v})^{\t}\, K_1(u)\, (R(uv)^{-1})^{{\t}_2} \, K_2(v) &= K_2(v)\, (R_{21}(uv)^{-1})^{{\t}_1} K_1(u)\, R(\tfrac{u}{v}) . \label{tREalt} 
}
Equation \eqref{tREalt} is more natural in the sense that unitarity \eqref{Ru:unit} and PT-symmetry \eqref{Ru:PT-symm} are no longer needed when proving statements involving this reflection equation. 

We will be interested in solutions of the reflection equations above which are associated to right coideal subalgebras $\mc{B} \subset U_q({\mfg})$. 
This will be explained in detail in Section \ref{sec:K-intw} below. 
For now we will briefly comment on some general properties of solutions of \eqref{RE} and \eqref{tRE}.

\begin{lemma} \label{lem:tw-untw}
Assume that $\mfg^{\rm fin}=\mfsl_2$, $\mfso_N$ or $\mfsp_N$. Then $\wt K(u) := C K(\wt q u)$ is a solution of \eqref{tREalt} precisely if $K(u)$ is a solution of \eqref{RE}.
\end{lemma}

\begin{proof}
Using the second equation in \eqref{Ru:C-symm} we derive that
\[
(R(uv)^{-1})^{{\t}_2} C_2 = \wt f_q(uv) C_2 R(\wt q^2 u v)
\]
and hence, by conjugating with $P$,
\[
(R_{21}(uv)^{-1})^{{\t}_1} C_1 = \wt f_q(uv)  C_1 R_{21}(\wt q^2 u v).
\]
By virtue of \eqref{Ru:RCC} we see that the left-hand side of \eqref{tREalt} for $K(u) = \wt K(u)$ equals
\[
\wt f_q(uv) C_1 C_2 R_{21}(\tfrac{u}{v}) K_1(\wt q u) R(\wt q^2 u v) K_2(\wt q v),
\]
whereas the right-hand side becomes
\[
\wt f_q(uv) C_1 C_2 K_2(\wt q v) R_{21}(\wt q^2 u v) K_1(\wt q u) R(\tfrac{u}{v}).
\]
The equivalence of \eqref{RE} with $K(u)$ and \eqref{tREalt} with $\wt K(u) = C K(\wt q u)$ is now obvious.
\end{proof}

Because in most cases Lemma \ref{lem:tw-untw} implies that a classification of solutions of \eqref{RE} produces a classification of solutions of \eqref{tRE}, it is only necessary to consider \eqref{tRE} if $\mfgf = \mfsl_{N >2}$.

\begin{lemma} \label{lem:REbasic}
Suppose $K(u) \in \End(\K^N)$ is a solution of \eqref{RE} or \eqref{tREalt}.
Then
\begin{enumerate}
\item $\wt K(u) := K(-u)$ is a solution of the same equation;
\item $\wt K(u) := m(u)K(u)$ is a solution of the same equation for any $m(u) : \K \to \K$;
\item $K^Z(u) := Z(\tfrac{\eta}{u})^{-1}K(u)Z(\eta u)$ is a solution of \eqref{RE} or $K^Z(u) := Z(\tfrac{\eta}{u})^{\t}K(u)Z(\eta u)$ is a solution of \eqref{tREalt}, respectively, for any $\eta \in \K^\times$ and any $Z(u) \in \End(\K^N)$ depending rationally on $u$ and satisfying 
\eq{ 
\label{Ru:RZZ} [R(\tfrac{u}{v}),Z(u) \ot Z(v)]=0. 
} 
\end{enumerate}
\end{lemma}

\begin{proof}
Statements (i) and (ii) are obvious.
Statement (iii) is a special case of \cite[Prop.~2]{Sk} in the case of the untwisted reflection equation; for the twisted reflection equation it is entirely analogous, relying on the identities 
\eqn{
R(\tfrac{u}{v}) Z_1(\eta \, u) Z_2(\eta \, v) &= Z_1(\eta \, u) Z_2(\eta \, v)R(\tfrac{u}{v}),\\
R_{21}(\tfrac{u}{v})^{\t} Z_1(\tfrac{\eta}{u})^{{\t}_1} Z_2(\tfrac{\eta}{v})^{{\t}_2} &= Z_1(\tfrac{\eta}{u})^{{\t}_1} Z_2(\tfrac{\eta}{v})^{{\t}_2}R_{21}(\tfrac{u}{v})^{\t}, \\
Z_1(\eta \, u) (R(uv)^{-1})^{{\t}_2} Z_2(\tfrac{\eta}{v})^{{\t}_2} &= Z_2(\tfrac{\eta}{v})^{{\t}_2} (R(uv)^{-1})^{{\t}_2} Z_1(\eta \, u), \\
Z_2(\eta \, v) (R_{21}(uv)^{-1})^{{\t}_1} Z_1(\tfrac{\eta}{u})^{{\t}_1} &= Z_1(\tfrac{\eta}{u})^{{\t}_1} (R_{21}(uv)^{-1})^{{\t}_1} Z_2(\eta \, v) ,
}
all of which follow straightforwardly from \eqref{Ru:RZZ}. 
\end{proof}

Equation \eqref{Ru:RZZ} can be thought of as a version of the Yang-Baxter equation with one tensor factor replaced by the ground field $\K$. 
To our best knowledge, all invertible solutions of \eqref{Ru:RZZ}, with $R(u)$ defined by \eqref{Ru:defn}, form a group of matrices depending rationally on $u$ isomorphic to $\Ad(\wt H_q) \rtimes \Sigma_A$, {\it cf.}~\eqref{Hopfalgautodecomp}. The elements of $\Ad(\wt H_q)$ correspond to certain constant diagonal matrices, whereas elements of $\Sigma_A \leq \Aut(A)$ correspond to matrices possibly with a nontrivial dependence on $u$. We explore this in more detail in Section \ref{sec:rotdress}.

\begin{rmk} Equations \eqref{RE} and \eqref{tRE} are often referred to as right reflection equations, since they have an interpretation of a factorized scattering of particles off the right end of a semi-infinite line or a segment, see {\it e.g.}~\cite{Ch1,GhZa,Sk}. 
For example for \eqref{RE} this can be seen by writing it as
\eq{
\hat R(uv)\, K_2(v)\, \hat R(\tfrac uv)\, K_2(u) = K_2(v)\, \hat R(uv)\, K_2(u)\, \hat R(\tfrac uv). \label{PRE}
} 
where $\hat R(u) = P R(u)$; in this presentation K-matrices appear in the right tensorand only. Note that this is at odds with other conventions, see for instance \cite{BgKo2}, where \eqref{PRE} holds but with $\hat R (u) = R(u) P$ instead.
In Section \ref{sec:qintsys} we discuss how reflection equations appear in the theory of quantum integrable systems, including left reflection equations. \hfill \rmkend
\end{rmk}


\subsection{Boundary intertwining equation}  \label{sec:K-intw}

Fix a right coideal subalgebra $\mc{B} \subset U_q({\mfg}) $.
The \emph{untwisted boundary intertwining equation} for the pair $(K(u),\eta) \in \End(\K^N) \times \K^\times$ is the equation
\begin{alignat}{90}
K(u)\, \RT_{\eta \, u}(b) &= \RT_{\eta/u}(b) \, K(u) && \text{for all } b \in \mc{B}. \label{intw-untw} 
\intertext{The \emph{twisted boundary intertwining equation} for the pair $(K(u),\eta) \in \End(\K^N) \times \K^\times $ is the equation}
K(u)\, \RT_{\eta \, u}(b) &= \RT^{\,\t}_{\eta/u}(S(b)) \, K(u) \qq && \text{for all } b \in \mc{B}. \label{intw-tw}
\end{alignat}
We view these as equations for matrix-valued rational functions of $u$. We call $\eta$ the \emph{scaling parameter}. We say that $K(u)$ is an \emph{(un)twisted K-matrix} if it is a solution of the (un)twisted boundary intertwining equation, respectively.

The equations above are analogues of the intertwining equation \eqref{R:intw}. We will use these equations to find solutions of the reflection equations \eqref{RE} and \eqref{tREalt}, respectively, following the arguments presented in \cite[Sec.~2]{DeMk} and \cite[Sec.~3]{DeGe}. We will make this more precise shortly.
The main benefit of using boundary intertwining equations to find K-matrices is that they are linear in $K(u)$, whereas REs are quadratic in $K(u)$.

We make the following observations about solutions of \eqref{intw-untw} and \eqref{intw-tw}. Note that we will only be interested in nontrivial solutions $(K(u),\eta)$, {\it i.e.}~with $K(u)\ne 0$.

\begin{lemma} \label{lem:intwbasic} 
Let $\mc{B} \subset U_q({\mfg})$ be a right coideal subalgebra. Suppose $(K(u),\eta)$ is a solution of \eqref{intw-untw} or \eqref{intw-tw}. Then:
\begin{enumerate}
\item $(\wt K(u),-\eta)$ with $\wt K(u) := K(-u)$ is a solution of the same equation. 
\item $(\wt K(u),\eta)$ with $\wt K(u) := m(u) K(u)$ is a solution of the same equation for any $m(u) : \K \to \K$.
\end{enumerate}
\end{lemma}

\begin{proof}
We obtain (i) by applying $(u,\eta) \mapsto (-u,-\eta)$ in \eqref{intw-untw} and \eqref{intw-tw}. Statement (ii) is obvious.
\end{proof}

Lemma \ref{lem:intwbasic} corresponds directly to Lemma \ref{lem:REbasic} (i)-(ii). 
We will dicuss analogons of Lemma \ref{lem:REbasic} (iii) on the level of the intertwining equation in Section \ref{sec:rotdress}.

\begin{lemma} \label{L:sol-intw}
Let $\mc{B} \subset U_q({\mfg})$ be a right coideal subalgebra. Suppose $\RT_u|_{\mc{B}}$ is irreducible for generic values of $u$. If a nontrivial solution $(K(u),\eta)$ to \eqref{intw-untw} or \eqref{intw-tw} exists, then $K(u)$ is invertible for generic values of $u$ and is unique, up to a scalar factor. 
\end{lemma}

\begin{proof}
Suppose $u$ is generic. Let $(K(u),\eta)$ be a nontrivial solution of \eqref{intw-untw}. By Schur's lemma $K(u)$ is invertible. Now let $(K'(u),\eta)$ be any other nontrivial solution with the same~$\eta$. Then $K'(u)^{-1} K(u)\, \RT_{\eta \, u}(b) = K'(u)^{-1} \RT_{\eta/u}(b) \, K(u) = \RT_{\eta u}(b) \, K'(u)^{-1} K(u)$ for any $b\in\mc{B}$. It follows by Schur's lemma for algebraically closed fields that $K'(u)^{-1}  K(u)$ must be proportional to a scalar. Proving invertibility and uniqueness of a solution $K(u)$ of \eqref{intw-tw} is analogous.
\end{proof}

By multiplying $K(u)$ by a suitable polynomial in $u$ of minimal degree we may clear denominators in $K(u)$ and obtain a matrix with polynomial entries, which is uniquely defined up to an element of $\K^\times$.
We define $d_{\rm eff}(K)$, the \emph{effective degree} of $K(u)$, to be the degree of this polynomial matrix.

We now demonstrate that solutions of \eqref{intw-untw} and \eqref{intw-tw} indeed satisfy \eqref{RE} and \eqref{tRE}, under suitable assumptions.

\begin{prop} \label{P:B-intw-untw}
Let $\mc{B} \subset U_q({\mfg})$ be a right coideal subalgebra and let $(K(u),\eta)$ be a nontrivial solution of the boundary intertwining equation \eqref{intw-untw}. 
If $(\RT_u \ot \RT_v)|_{\mc{B}}$ is irreducible for generic values of $u$ and $v$, then $K(u)$ is a solution of the reflection equation \eqref{RE}.
\end{prop}

\begin{proof}
The main idea of this proof first appeared in \cite[Sec.~2.2]{DeMk} and \cite[Sec.~3]{DeGe}; it is an adaptation of an argument by Jimbo \cite[Proof of Prop.~3]{Ji2}. 
Let $u$ and $v$ be generic. Now note that \eqref{R:intw} is equivalent to
\eq{ \label{R:intw2} 
R_{21}(\tfrac uv)\,( \RT_{1/u} \ot \RT_{1/v} )( \Delta (a))  = ( \RT_{1/u} \ot \RT_{1/v})(\Delta^{\rm op} (a))\, R_{21}(\tfrac uv) \qu\text{for all } a \in U_q({\mfg}). 
}
Since $\mc{B}$ is a right coideal subalgebra of $U_q({\mfg})$, for all $b \in \mc{B}$ we have
\eqa{
K_2(v) (\RT_{\eta \, u} \ot \RT_{\eta \, v})(\Delta^{\rm op}(b))  &= (\RT_{\eta \, u} \ot \RT_{\eta/v})(\Delta^{\rm op}(b)) K_2(v), \label{K-intw-1} \\
K_1(u) (\RT_{\eta \, u} \ot \RT_{\eta \, v})(\Delta(b)) &= (\RT_{\eta/u} \ot \RT_{\eta \, v})(\Delta(b)) K_1(u). \label{K-intw-2}
}
Denote the left and right hand sides of \eqref{RE} by ${\rm RE}_1(u,v)$ and ${\rm RE}_2(u,v)$, respectively.
Applying ${\rm RE}_1(u,v)$ to $(\RT_{\eta \, u} \ot \RT_{\eta \, v}) (\Delta^{\rm op}(b))$ we obtain
\begin{eqnarray*}
&& \hspace{-3cm} R_{21}(\tfrac{u}{v})\,  K_1(u)\, R(uv)\, K_2(v)\, (\RT_{\eta \, u} \ot \RT_{\eta \, v}) (\Delta^{\rm op}(b)) = \\
& \overset{\eqref{K-intw-1}}{=}& \hspace{-.3cm} R_{21}(\tfrac{u}{v})\, K_1(u)\, R(uv)\, (\RT_{\eta \, u} \ot \RT_{\eta/v}) (\Delta^{\rm op}(b)) \, K_2(v) 
 \\
& \overset{\eqref{R:intw}}{=}& \hspace{-.3cm} R_{21}(\tfrac{u}{v})\, K_1(u)\, (\RT_{\eta \, u} \ot \RT_{\eta/v}) (\Delta(b))\, R(uv) \, K_2(v) 
\\
& \overset{\eqref{K-intw-2}}{=}& \hspace{-.3cm} R_{21}(\tfrac{u}{v})\, (\RT_{\eta/u} \ot \RT_{\eta/v}) (\Delta(b))\,K_1(u)\,  R(uv) \, K_2(v) 
\\
& \overset{\eqref{R:intw2}}{=}& \hspace{-.3cm} (\RT_{\eta/u} \ot \RT_{\eta/v}) (\Delta^{\rm op}(b))\,R_{21}(\tfrac{u}{v})\,  K_1(u)\,  R(uv) \, K_2(v) ,
\end{eqnarray*}
for all $b\in \mc{B}$. 
In a similar way we can show ${\rm RE}_2(u,v)$ also intertwines $(\RT_{\eta \, u} \ot \RT_{\eta \, v}) (\Delta^{\rm op}(b))$ with $ (\RT_{\eta/u} \ot \RT_{\eta/v}) (\Delta^{\rm op}(b))$ for all $b\in \mc{B}$. 

Since $(\RT_u \ot \RT_v)|_{\mc{B}}$ is irreducible, it follows that $\RT_u|_{\mc B}$ and $\RT_v|_{\mc B}$ are also irreducible and, by Lemma \ref{L:sol-intw}, both $K(u)$ and $K(v)$ are invertible. Hence ${\rm RE}_2(u,v)$ is also invertible and ${\rm RE}_2(u,v)^{-1}{\rm RE}_1(u,v)$ intertwines $(\RT_{\eta \, u} \ot \RT_{\eta \, v}) (\Delta^{\rm op}(\mc{B}))$ with itself. By Schur's lemma ${\rm RE}_1(u,v)=\zeta(u,v)\,{\rm RE}_2(u,v)$ for some $\zeta(u,v) \in \K$. 
It remains to show that $\zeta(u,v)=1$. 
Note that $\zeta(u,v)$ depends rationally on $u$ and $v$ since the operators appearing in \eqref{RE} depend rationally on $u$ and $v$.
Taking the determinant of both sides of \eqref{RE} yields $\zeta(u,v)^{N^2}=1$. 
Because the set of $N^2$-th roots of unity is discrete in $\K$, it follows that $\zeta(u,v) = \zeta$ must be constant.
Now by virtue of \eqref{Ru:reg}, taking $u=v$ in both sides of \eqref{RE} we obtain $\zeta =1 $ as required.
\end{proof}

To demonstrate the analogous statement in the twisted case, {\it i.e.}~that a solution of \eqref{intw-tw} is also a solution of \eqref{tRE}, we need the additional lemma stated below.

\begin{lemma} \label{L:R-intw}
The following identities hold for all $a\in U_q({\mfg})$ and generic values of $u$ and $v$:
\eqa{
(R(uv)^{-1})^{\t_2}  (\RT_{u} \ot \RT^{\,\t}_{1/v})(\id \ot S) (\Delta^{\rm op}(a)) &= (\RT_{u} \ot \RT^{\,\t}_{1/v})(\id \ot S) (\Delta(a))\, (R(uv)^{-1})^{\t_2}, \hspace{-3mm}  \label{R-intw-1} \\
\hspace{-1mm} (R_{21}(uv)^{-1})^{\t_1}  (\RT^{\,\t}_{1/u} \ot \RT_{v})(S \ot \id)(\Delta(a)) & =(\RT^{\,\t}_{1/u} \ot \RT_{v})(S \ot \id)(\Delta^{\rm op}(a)) (R_{21}(uv)^{-1})^{\t_1} , \hspace{-3mm}   \label{R-intw-2} \\
R_{21}(\tfrac{u}{v})^\t\, (\RT^{\t}_{1/u} \ot \RT^{\,\t}_{1/v})(S \ot S)(\Delta(a)) &= (\RT^{\,\t}_{1/u} \ot \RT^{\,\t}_{1/v})(S \ot S)(\Delta^{\rm op}(a))\, R_{21}(\tfrac{u}{v})^\t . \label{R-intw-3}
}
\end{lemma}

\begin{proof}
To prove \eqref{R-intw-1}, we will repeatedly use the identity 
\eq{ \label{partialtranspose} 
X^{\t_2} (Y \ot Z^\t) = ((Y \ot \Id_V) X (\Id_V \ot Z))^{\t_2} \in \End(V \otimes V), 
}
where $V$ is any vector space, $X \in \End(V \ot V)$ and $Y,Z \in \End(V)$.
First of all, note that
\[ 
a \mapsto (\RT_{u} \ot \RT^{\,\t}_{1/v})(\id \ot S) (\Delta^{\rm op}(a)) 
\]
defines an algebra homomorphism: $U_q({\mfg}) \to \End(\K^N)$.
To establish this it is sufficient to show that it preserves products. 
This follows from a straightforward argument involving \eqref{partialtranspose}, the fact that $S$ is an algebra antiautomorphism of $U_q({\mfg})$ and transposition is an algebra antiautomorphism of $\End(\K^N)$.
Hence it suffices to prove \eqref{R-intw-1} for the generators $x_i$, $y_i$ and $k_i$ ($i \in I$).
Note that \eqref{R:intw} is equivalent to
\eq{ \label{Rinv:intw} R(\tfrac uv)^{-1} (\RT_u \ot \RT_v)(\Delta(a)) = (\RT_u \ot \RT_v)(\Delta^{\rm op}(a)) R(\tfrac uv)^{-1}, \qq \text{for all } a \in U_q({\mfg}). }
To prove \eqref{R-intw-1} for $a = x_i$, note that
\[ 
(\id \ot S) (\Delta^{\rm op}(x_i)) = x_i \ot k_i^{-1} - 1 \ot k_i^{-1} x_i, \qq (\id \ot S) (\Delta(x_i)) = x_i \ot 1 - k_i \ot k_i^{-1} x_i 
\]
so that
\begin{align*}
&\hspace{-20pt} (R(uv)^{-1})^{\t_2} (\RT_{u} \ot \RT^{\,\t}_{1/v})(\id \ot S) (\Delta^{\rm op}(x_i)) = \\
&\overset{\hphantom{\eqref{partialtranspose}}}{=} (R(uv)^{-1})^{\t_2} (\RT_{u}(x_i) \ot \RT^{\,\t}_{1/v}(k_i^{-1})) - (R(uv)^{-1})^{\t_2} (\Id \ot \RT^{\,\t}_{1/v}(k_i^{-1} x_i)) \\
&\overset{\eqref{partialtranspose}}{=} \Big(  (\Id \ot \RT_{1/v}(k_i^{-1})) R(uv)^{-1} (\RT_{u}(x_i) \ot \Id) - (\Id \ot \RT_{1/v}(k_i^{-1} x_i)) R(uv)^{-1} \Big)^{\t_2} \\
&\overset{\eqref{Rinv:intw}}{=} \Big(  (\RT_{u}(x_i) \ot \Id) R(uv)^{-1} - (\Id \ot \RT_{1/v}(k_i^{-1})) R(uv)^{-1} (\RT_{u}(k_i) \ot \RT_{1/v}(x_i)) \Big)^{\t_2} \\
&\overset{\eqref{Rinv:intw}}{=} \Big((\RT_{u}(x_i) \ot \Id) R(uv)^{-1} - (\RT_{u}(k_i) \ot \Id) R(uv)^{-1} (\Id \ot \RT_{1/v}(k_i^{-1}x_i)) \Big)^{\t_2} \\
&\overset{\eqref{partialtranspose}}{=} (\RT_{u}(x_i) \ot \Id) (R(uv)^{-1})^{\t_2} - (\RT_{u}(k_i) \ot \RT_{1/v}^\t(k_i^{-1}x_i)) (R(uv)^{-1})^{\t_2} \\
&\overset{\hphantom{\eqref{partialtranspose}}}{=} (\RT_{u} \ot \RT^{\,\t}_{1/v})(\id \ot S) (\Delta(x_i))\, (R(uv)^{-1})^{\t_2}.
\end{align*}
The analogous calculation for $y_i$ is very similar and we leave it to the reader.
Finally, we have $(\id \ot S) (\Delta^{\rm op}(k_i)) = (\id \ot S) (\Delta(k_i)) = k_i \ot k_i^{-1}$ so that
\begin{align*}
(R(uv)^{-1})^{\t_2} (\RT_{u} \ot \RT^{\,\t}_{1/v})(\id \ot S) (\Delta^{\rm op}(k_i)) 
&\overset{\hphantom{\eqref{partialtranspose}}}{=} (R(uv)^{-1})^{\t_2} (\RT_{u}(k_i) \ot \RT^{\,\t}_{1/v}(k_i^{-1})) \\
&\overset{\eqref{partialtranspose}}{=} \Big( (\Id \ot \RT_{1/v}(k_i^{-1})) R(uv)^{-1} (\RT_{u}(k_i) \ot \Id) \Big)^{\t_2} \\
&\overset{\eqref{Rinv:intw}}{=} \Big( (\RT_u(k_i) \ot \Id) R(uv)^{-1} (\Id \ot \RT_{1/v}(k_i^{-1})) \Big)^{\t_2} \\
&\overset{\eqref{partialtranspose}}{=}  (\RT_{u}(k_i) \ot \RT^{\,\t}_{1/v}(k_i^{-1}))\, (R(uv)^{-1})^{\t_2} \\
&\overset{\hphantom{\eqref{partialtranspose}}}{=} (\RT_{u} \ot \RT^{\,\t}_{1/v})(\id \ot S) (\Delta(k_i))\, (R(uv)^{-1})^{\t_2}.
\end{align*}
To obtain \eqref{R-intw-2} we only need to conjugate \eqref{R-intw-1} with $P$ and swap $u$ and $v$. 
Finally, \eqref{R-intw-3} is equivalent to \eqref{R:intw}; this follows by transposing, conjugating by $P$ and using that the antipode is a coalgebra antiautomorphism.
\end{proof}

\begin{prop} \label{P:B-intw-tw}
Let $\mc{B} \subset U_q({\mfg})$ be a right coideal subalgebra and let $(K(u),\eta) \in \End(\K^N) \times \K^\times  $ be a nontrivial solution of the twisted boundary intertwining equation \eqref{intw-tw}. 
If $(\RT_u \ot \RT_v)|_{\mc{B}}$ is irreducible for generic values of $u$ and $v$, then $K(u)$ is a solution of the twisted reflection equation~\eqref{tRE}.
\end{prop}

\begin{proof}
Let $u$ and $v$ be generic. We will use the alternative form of the twisted RE \eqref{tREalt}. We need to use \eqref{R-intw-1} and \eqref{R-intw-3}, as well as the identities
\eqa{
K_2(v) (\RT_{\eta \, u} \ot \RT_{\eta \, v})(\Delta^{\rm op}(b))  &= (\RT_{\eta \, u} \ot \RT^{\,\t}_{\eta/v})(\id \ot S)(\Delta^{\rm op}(b)) K_2(v), \label{K-intw-3} \\
K_1(u) (\RT_{\eta \, u} \ot \RT_{\eta \, v})(\Delta(b)) &= (\RT^{\,\t}_{\eta/u} \ot \RT_{\eta \, v})(S \ot \id)(\Delta(b)) K_1(u), \label{K-intw-4}
}
for $b \in \mc{B}$. Applying the left-hand side of \eqref{tREalt} to $(\RT_{\eta \, u} \ot \RT_{\eta \, v}) (\Delta^{\rm op}(b))$ we obtain
\begin{eqnarray*}
&& \hspace{-3cm} R_{21}(\tfrac{u}{v})^\t K_1(u)\, (R(uv)^{-1})^{\t_2} K_2(v)\, (\RT_{\eta \, u} \ot \RT_{\eta \, v}) (\Delta^{\rm op}(b)) = 
\\
& \overset{\eqref{K-intw-3}}{=}& \hspace{-.3cm} R_{21}(\tfrac{u}{v})^\t K_1(u)\, (R(uv)^{-1})^{\t_2}  (\RT_{\eta \, u} \ot \RT^{\,\t}_{\eta/v})(\id \ot S) (\Delta^{\rm op}(b))\, K_2(v) 
\\
& \overset{\eqref{R-intw-1}}{=}& \hspace{-.3cm} R_{21}(\tfrac{u}{v})^\t K_1(u)\, (\RT_{\eta \, u} \ot \RT^{\,\t}_{\eta/v})(\id \ot S) (\Delta(b))\, (R(uv)^{-1})^{\t_2} K_2(v) 
\\
& \overset{\eqref{K-intw-4}}{=}& \hspace{-.3cm} R_{21}(\tfrac{u}{v})^\t \, (\RT^{\,\t}_{\eta/u} \ot \RT^{\,\t}_{\eta/v})(S \ot S) (\Delta(b))\, K_1(u)\, (R(uv)^{-1})^{\t_2} K_2(v) \\
& \overset{\eqref{R-intw-3}}{=}& \hspace{-.3cm} (\RT^{\,\t}_{\eta/u} \ot \RT^{\,\t}_{\eta/v})(S \ot S) (\Delta^{\rm op}(b))\, R_{21}(\tfrac{u}{v})^\t  K_1(u)\, (R(uv)^{-1})^{\t_2} K_2(v) ,
\end{eqnarray*}
for all $b\in \mc{B}$. It follows from a similar computation using \eqref{R:intw} and \eqref{R-intw-2} that the right-hand side of \eqref{tREalt} also intertwines $ (\RT_{\eta \, u} \ot \RT_{\eta \, v}) (\Delta^{\rm op}(b))$ with $(\RT^{\,\t}_{\eta/u} \ot \RT^{\,\t}_{\eta/v})(S \ot S) (\Delta^{\rm op}(b))$. 
The rest of the proof is analogous to that of Proposition \ref{P:B-intw-untw}. \qedhere
\end{proof}

Analogous to Lemma \ref{lem:tw-untw} we have the following.

\begin{prop} \label{prop:intw:tw-untw}
Assume that $\mfg^{\rm fin}=\mfsl_2$, $\mfso_N$ or $\mfsp_N$. Let $\mc{B} \subset U_q({\mfg})$ be a right coideal subalgebra. 
Then $(\wt K(u),\wt \eta) := (C K(\wt q u),\wt q \eta)$ is a solution of \eqref{intw-tw} precisely if $(K(u),\eta)$ is a solution of \eqref{intw-untw}.
\end{prop}

\begin{proof}
This follows immediately from \eqref{dual:3}.
\end{proof}

\begin{rmk} \label{R:uni-K}
There exist universal structures that are analogues of \eqref{R-uni} and \eqref{uni-YBE} for certain coideal subalgebras of quantum groups. 
Such universal structures have been proposed in \cite{DKM} and more recently in \cite{BaWa,BgKo2,Ko2}. 
More precisely, let $\mfg$ be an arbitrary symmetrizable Kac-Moody algebra. Following \cite[Defn.~4.12 and Cor.~7.7]{BgKo2} and analogous statements in \cite{Ko2}, we consider an associated quantized Kac-Moody algebra $\mc{U}= U_q(\mfg^{\rm ext})$ and a Satake diagram $(X,\tau)$ for the associated generalized Cartan matrix $A$. 
We assume the minimal realization $(\mfh^{\rm ext},\{h_i\}_{i \in I},\{\al_i\}_{i \in I})$ of $A$ is compatible with $\tau$.
In this case one may consider a right coideal subalgebra $\mc{B} = \mc{B}(X,\tau) \subseteq \mc{U}$ (essentially given by Definition \ref{defn:rCSA}) and a \emph{$\tau \tau_0$-universal K-matrix} exists, {\it i.e.}~an invertible element $\mc{K}$ in a suitable completion of $\mc{B} \ot \mc{U}$ satisfying
\begin{align}
\hspace{20mm} \mc{K}\,\Delta(b) &= (\id \ot \tau \tau_0)(\Delta(b))\,\mc{K} \qq \text{for all } b\in \mc{B} , \label{intw-uni}
\\
(\Delta\ot\id)(\mc{K}) &= \mc{R}^{\tau\tau_0}_{21}\,\mc{K}_{02}\,\mc{R}_{12}^{-1}, \label{DK1}
\\
(\id\ot\Delta)(\mc{K}) &= \mc{K}_{01}\mc{R}^{\tau\tau_0}_{12}\,\mc{K}_{02}\,\mc{R}_{21}. \label{DK2}
\end{align}
Here $\mc{R}$ is the universal R-matrix as in Remark \ref{R:uni-R}, we have written
$\mc{R}^{\tau\tau_0}_{12} = (\id \ot \tau\tau_0) (\mc{R}_{12})$ and $\mc{R}^{\tau\tau_0}_{21} = (\tau\tau_0\ot\id)(\mc{R}_{21})$,
and $\tau_0$ is a particular diagram automorphism satisfying certain compatibility conditions with $\tau$, see \cite[Sec.~7.1]{BgKo2}. 
It would be nice to generalize it to the larger class of coideal subalgebras defined in terms of generalized Satake diagrams studied in this paper.

Identity \eqref{DK2} implies that $\mc{K}$ satisfies the universal reflection equation
\eq{
\mc{K}_{01}\mc{R}^{\tau\tau_0}_{12}\,\mc{K}_{02}\,\mc{R}_{21} = \mc{R}_{12} \mc{K}_{02}\mc{R}^{\tau\tau_0}_{21}\,\mc{K}_{01} .\label{uni-RE}
}
If $\mc{U}$ is of finite type, for $\tau_0$ we may take the unique diagram automorphism such that the longest element $w_0 \in W$ satisfies $w_0(\al_i) = - \al_{\tau_0(i)}$ for all $i\in I$.
The existence of such a $\tau_0$ in the general Kac-Moody case is not proven.

Suppose that $\mc{U}$ is of affine type and its minimal realization is compatible with $\tau$. 
Let $\eps : \mc{B} \to \K$ be a one-dimensional representation. 
Then, upon applying $\eps \ot T_{u} \ot T_{v}$ to both sides of \eqref{uni-RE} and conjugating with $P$, one should obtain the untwisted reflection equation \eqref{RE} or the twisted reflection equation \eqref{tRE}, depending on the explicit form of $\tau\tau_0$. Similarly, applying $\eps \ot T_{u}$ to both sides of \eqref{intw-uni} one should obtain the untwisted intertwining equation \eqref{intw-untw} or the twisted intertwining equation \eqref{intw-tw}, or in other words, we expect that $K(u) = k_0(u) (\eps \ot T_{\eta u})(\mc{K})$ for some meromorphic function $k_0: \C \to \C$.  
If $\bm s \in \K^{I \backslash X}$ is such that $s_j \ne 0$ implies $j= \tau(j)$ (in particular, if $\bm s \in \mc{S}$) then we may choose $\tau_0 = \tau$, although this does not always produce the desired universal intertwining equation. The choice $\tau_0 = \psi$ seems to produce the correct intertwining equation, see Theorem \ref{thrm:solexistence}.
When $\mc{U}$ is of type ${\rm A}^{(1)}_{n>1}$ and $\tau \in \{ \id, \pi\}$ this imposes conditions $c_i = c_j$ for some $i,j \in I \backslash X$. 
It will become clear that in those cases it is in fact sufficient to consider QP algebras with $c_i=c_j$ for all $i,j \in I \backslash X$ (see Propositions \ref{prop:Htildeequiv} and \ref{prop:Gomegaequiv}), so this does not pose a restriction.~\rmkend
\end{rmk}

\begin{rmk}
In many cases it is relevant to consider R- and K-matrices acting on arbitrary vector spaces which carry irreducible representations of $U_q({\mfg})$ (or more generally of any symmetrizable quantized Kac-Moody algebra). 
Of a particular importance are R-matrices acting on tensor products of two vector spaces of different dimensions. For example, this enables one to generalize the construction of qKZ transport matrices, transfer matrices and Hamiltonians (see Section \ref{sec:qintsys}) as elements of $\End(V_1 \ot \cdots \ot V_L)$ with the vector spaces $\{V_i\}_{1\le i \le L}$ not all of the same dimension. 

In many instances higher-dimensional R- and K-matrices are related to their lower-dimensional counterparts by the so-called fusion rules or the bootstrap method, see {\it e.g.}~\cite{CDRS,FgKo,KRS,KuSk2,MeNe2,ZaZa}, see also \cite{Is,BbRg}. 
However it is unknown whether all solutions of higher-dimensional reflection equations can be obtained this way. 
An example is the K-matrix given in \cite[Eqns.~(17), (18), (20)]{BsKz1}. 
To obtain such K-matrices from a universal K-matrix formalism, it is crucial that, as in \cite{Ko2}, $\mc{K}$ lies in (a completion of) $\mc{B} \ot \mc{U}$, with $\mc{B}$ a right coideal subalgebra of a Hopf algebra $ \mc{U}$, and satisfies \eqref{DK1}. \hfill\rmkend
\end{rmk}


\subsection{Unitarity and regularity of K-matrices}

There exist notions of unitarity and regularity for K-matrices which are analogous to the corresponding properties for R-matrices \eqrefs{Ru:reg}{Ru:unit}. 
These properties are important in the theory of integrable systems with reflecting boundary conditions, see Section \ref{sec:qintsys}. 
First we deal with the untwisted case.
Let $(K(u),\eta)$ be a solution of \eqref{intw-untw}. We say that $K(u)$ is \emph{unitary} if, for generic values of $u$, 
\eq{ 
\label{Ku:unit} K(u)^{-1} = K(u^{-1}) .
}

\begin{lemma} \label{L:K-unit}
Let $\mc{B} \subset U_q({\mfg})$ be a right coideal subalgebra. Suppose $\RT_u|_{\mc{B}}$ is irreducible for generic values of $u$. If $(K(u),\eta)$ is nontrivial solution of \eqref{intw-untw}, then then $K(u^{-1})$ equals $K(u)^{-1}$ up to a scalar and there exists a rational function $m : \K \to \K$ such that $\wt K(u) := m(u) K(u)$ is unitary. 
\end{lemma}

\begin{proof}
Let $u$ be generic.
Substitute $u \to u^{-1}$ in \eqref{intw-untw} and left- and right-multiply with $K(u^{-1})^{-1}$. This gives
\[
K(u^{-1})^{-1}\, \RT_{\eta \, u}(b) = \RT_{\eta/u}(b) \, K(u^{-1})^{-1} \qq \text{for all } b \in \mc{B}.
\]
Hence, by irreducibility and Schur's lemma, $K(u)$ equals $K(u^{-1})^{-1}$ up to a scalar as required.

For the second part of the lemma, we may assume that $K(u)^{-1} = n(u) K(u^{-1})$ for some rational function $n(u)$ (since $K(u)$ is rational in $u$). 
By substituting $u \to u^{-1}$ it follows that $n(u^{-1})=n(u)$.
It is sufficient to show that there exists a rational function $m : \K \to \K$ such that $n(u) = m(u)m(u^{-1})$.
 Because $\K[u]$ is a unique factorization domain whose irreducibles are precisely the linear polynomials over $\K$ (here we need that $\K$ is algebraically closed), we have
\[
n(u) = c \prod_{a \in \K } (u-a)^{j_a} 
\]
with unique $c\in \K^\times$ and unique $j_a \in \Z$, finitely many of which are nonzero.
Again by unique factorization, $n(u)=n(u^{-1})$ implies that
\[
\prod_{a \in  \K^\times} (-a)^{j_a} = 1, \qq -2j_0 = \sum_{a \in \K^\times} j_a, \qq j_a = j_{a^{-1}} \qu\text{for } a \in \K^\times. 
\]
Hence $j_1$ and $j_{-1}$ are even.
Choose a subset $\mathbb{H} \subset \K^\times \backslash \{ \pm 1 \}$ such that for all $a \in \K^\times \backslash \{ \pm 1 \}$ either $a \in \mathbb{H}$ or $a^{-1} \in \mathbb{H}$.
Let $\wt c \in \K^\times$ be a square root of $c  (-1)^{-j_1/2}\prod_{a \in \mathbb{H}} (-a)^{-j_a}$ and define
\[
m(u) = \wt c (u-1)^{j_1/2} (u+1)^{j_{-1}/2} \prod_{a \in \mathbb{H}} (u-a)^{j_a}. 
\]
Then it can be straightforwardly checked that $m(u)\,m(u^{-1}) = n(u)$.
\end{proof}

Suppose $(K(u),\eta)$ is a nontrivial solution of \eqref{intw-untw}. If both $K(1)$ and $K(-1)$ are equal to $\pm \Id$ we call $K(u)$ \emph{doubly regular}; if precisely one of $K(1)$ and $K(-1)$ equals $\pm \Id$ we call $K$ \emph{singly regular}; finally, if neither $K(1)$ nor $K(-1)$ equals $\pm \Id$ we call $K$ \emph{non-regular}.

\begin{lemma} \label{L:regularity}
Let $\mc{B} \subset U_q({\mfg})$ be a right coideal subalgebra, $(K(u),\eta)$ a nontrivial solution of \eqref{intw-untw} and $\zeta \in \{\pm 1\}$. Suppose $\RT_{\zeta \eta}|_{\mc{B}}$ is irreducible. Then $K(\zeta)$ equals $\Id$ up to a scalar. Moreover, if $K(u)$ is unitary then $K(\zeta) = \pm \Id$.
\end{lemma}

\begin{proof}
Setting $u=\zeta$ in \eqref{intw-untw} we see that $K(\zeta)$ and $\Id$ intertwine the representation $\RT_{\zeta \eta}$ with itself. The first statement now follows using Schur's lemma. 
The second statement follows immediately by observing that $K(\zeta)=K(\zeta)^{-1}$ and $K(\zeta) = n_\zeta \Id$ for $n_\zeta \in \K^\times$ imply that $n_\zeta=n_\zeta^{-1}$.
\end{proof}

We now present an analogue of Lemma \ref{L:K-unit} for the twisted case.

\begin{lemma} \label{L:K-unit-tw}
Let $\mc{B} \subset U_q({\mfg})$ be a right coideal subalgebra such that $\psi(\mc{B})=\mc{B}$. Suppose $\RT_u|_{\mc{B}}$ is irreducible for generic values of $u$. If $(K(u),\eta)$ is  nontrivial solution of \eqref{intw-tw}, then there exists a rational function $n_{\rm tw} : \K \to \K$ such that
\eq{ \label{K-unit-tw} 
(C^{-1} K(\wt q^{-1} u))^{-1}  = n_{\rm tw}(u)\, C^{-1} K(\wt q^{-1} u^{-1}). 
}
\end{lemma}

\begin{proof}
Let $u$ be generic and let $b \in \mc{B}$.
From \eqref{intw-tw} and $\mc{B} = \psi ( \mc{B} )$ we derive
\eq{ \label{intw-tw-sigma} 
K(u)\, \RT_{\eta \, u}(\psi(b))  = \RT_{\eta/u}^\t (S^{-1}(\psi(b)))\, K(u)
}
so that, using Lemma \ref{L:Tu-dual},
\eqn{
C^{-1} K(\wt q^{-1} u^{-1}) C^{-1} \RT^{\,\t}_{\eta \wt q u^{-1}}(S^{-1}(b)) 
&\overset{\eqref{dual:3}}{=} C^{-1} K(\wt q^{-1} u^{-1})\, \RT_{\eta \wt q^{-1} u^{-1}}(\psi(b))\, C^{-1} \\
&\overset{\eqref{intw-tw-sigma}}{=} C^{-1} \RT_{\eta \wt q u}^\t(S^{-1}(\psi(b))) \,K(\wt q^{-1} u^{-1})\,  C^{-1} \\
&\overset{\eqref{dual:2}}{=}  \RT_{\eta \wt q^{-1} u}(b) \, C^{-1} K(\wt q^{-1} u^{-1}) \, C^{-1},
}
where we have used that $\psi \in \Aut_{\rm Hopf}(U_q({\mfg}))$ is involutive.
On the other hand, by left- and right-multiplying \eqref{intw-tw} by $K(u)^{-1}$ and replacing $u$ by $\wt q^{-1} u$ we obtain
$$ 
K(\wt q^{-1} u)^{-1} \, \RT^{\,\t}_{\eta \wt q u^{-1}}(S^{-1}(b)) = \RT_{\eta \wt q^{-1} u}(b) \, K(\wt q^{-1} u)^{-1}. 
$$
Because $C^{-1} K(\wt q^{-1} u^{-1}) C^{-1}$ and $K(\wt q^{-1} u)^{-1}$ intertwine the same representations of $\mc{B}$, by Schur's Lemma, they agree up to a scalar factor.
As $K(u)$ depends rationally on $u$, so does this factor.  
\end{proof}

For $\mfg^{\rm fin} = \mfsl_2$, $\mfso_N$ and $\mfsp_N$ the condition $\psi(\mc{B}) = \mc{B}$ in Lemma \ref{L:K-unit-tw} is trivially true; in this case Lemma~\ref{L:K-unit-tw} can also be obtained by combining Proposition \ref{prop:intw:tw-untw} and Lemma \ref{L:K-unit}.
For $\mfsl_{N>2}$ we will see in Lemma \ref{lem:twistedunitarityslN} that, for all QP algebras $\mc{B} =B_{\bm c,\bm s}(X,\tau)$ associated to the twisted reflection equation, the condition $\psi(\mc{B}) = \mc{B}$ is equivalent to a simple condition on $\bm c$.

\begin{samepage}
\begin{rmk} \mbox{}
\begin{enumerate}
\item 
It is possible to renormalize a twisted K-matrix so that $(C^{-1} K(\wt q^{-1} u))^{-1} = C^{-1} K(\wt q^{-1} u^{-1})$, {\it i.e.}~such that $n_{\rm tw}(u)=1$; the proof for Lemma \ref{L:K-unit} can be straightforwardly modified.
However it is more convenient for the presentation of our results to have a different proportionality factor $n_{\rm tw}(u)$ for different $\mc{B}$.
\item 
A notion of regularity for the twisted case, namely that $K(\pm \wt q^{-1})$ equals $C$ up to a scalar exists only if $\mfg^{\rm fin}=\mfsl_2$, $\mfso_N$ or $\mfsp_N$
In those cases it can be simply derived from the transformation appearing in Lemma \ref{lem:tw-untw}, provided that the corresponding untwisted K-matrix at $\pm 1$ equals $\Id$ up to a scalar. 
If $\mfg^{\rm fin}=\mfsl_{N>2}$, because $\psi$ does not fix $\mc{B} = B_{\bm c,\bm s}(X,\tau)$ pointwise, such a regularity property cannot be derived from the intertwining equation. 
Our case-by-case results (see Section \ref{sec:K:tw}) indeed confirm that $K(\pm \wt q^{-1})$ is not a scalar multiple of $C$. \hfill \rmkend
\end{enumerate}
\end{rmk}
\end{samepage}


\section{Rotation and dressing of quantum pair algebras and K-matrices} \label{sec:rotdress}

Let $\Sigma \leq \Aut_{\rm Hopf}(U_q(\mfg))$.
Coideal subalgebras $\mc{B}, \mc{B}' \subseteq U_q(\mfg)$ are called {\it $\Sigma$-equivalent} if there exists $\phi \in \Sigma$ such that $\phi(\mc{B}) = \mc{B}'$. 
Note that \eqref{Hopfalgautodecomp} implies that essentially there are two types of equivalences of coideal subalgebras of $U_q(\mfg)$: those induced by diagram automorphisms and those induced by characters of the root lattice; we will refer to them as \emph{rotational} and \emph{dressing} equivalences, respectively. 
In this section we will see that both types of equivalences are crucial for studying the intertwiners of $\RT_u|_{B_{\bm c,\bm s}}$.

\subsection{Rotation} \label{sec:rotation}
QP algebras which are $\Aut(A)$-equivalent are related by natural transformations on the underlying generalized Satake diagram and the tuples $\bm c$ and $\bm s$. 
Furthermore, for each generalized Cartan matrix $A$ of untwisted affine type, there exist a subgroup $\Sigma_A \le \Aut(A)$ such that the solutions of the boundary intertwining equation \eqref{intw-untw} or \eqref{intw-tw} of $\Sigma_A$-equivalent QP algebras are, ignoring a $u$-deformation, similar or congruent as matrices, respectively. Thus the classification of solutions of these intertwining equations is simplified: we only need to solve them for QP algebras $B_{\bm c,\bm s}(X,\tau)$ for representatives $(X,\tau)$ of the $\Sigma_A$-equivalence classes.


\subsubsection{$\Aut(A)$-equivalences of QP algebras} \label{sec:rotation:alg}

Here we review how $\Aut(A)$-equivalent coideal subalgebras $B_{\bm c,\bm s}$ are related to each other by transformations of the underlying generalized Satake diagram $(X,\tau)$ and the tuples $\bm c \in \mc{C}$, $\bm s \in \mc{S}$.
Let $\sigma \in \Aut(A)$ and $(X,\tau) \in \GSat(A)$. Then $(X^\si,\tau^\si) \in \GSat(A)$ where 
\eq{ 
X^\si := \si(X), \qq \tau^\si := \si \tau \si^{-1}. 
}
This follows immediately from
\eq{ \label{wXsigma} 
a_{\si(i) \si(j)}=a_{ij} \qu \text{for all } i,j \in I \text{ and} \qu  
\rho^\vee_{X^\si} = \si(\rho^\vee_X).
}
Furthermore, two generalized Satake diagrams $(X,\tau), (X',\tau')$ associated to $A$ are called $\Aut(A)$-equivalent if there exists $\sigma \in \Aut(A)$ such that $X' = X^\si$ and $\tau' = \tau^\si$.

For $\bm c \in (\K^\times)^{I \backslash X}$, $\bm s \in \K^{I \backslash X}$ and $\si \in \Aut(A)$, define $\si(\bm c) \in (\K^\times)^{I \backslash X}$, $\si(\bm s) \in \K^{I \backslash X}$ by $(\si(\bm c))_{\si(i)} = c_i$ and $(\si(\bm s))_{\si(i)} = s_i$ for $i \in I \backslash X$.
With respect to $\Aut(A)$-equivalence classes of coideal subalgebras of $U_q(\bm\hat \mfg)$ we have the following.

\begin{prop} \label{prop:rotateCSA}
Let $(X,\tau) \in \GSat(A)$ and $\si \in \Aut(A)$. 
Given $\bm c \in \mathcal{C}$ and $\bm s \in \mathcal{S}$, we have $\si(\bm c) \in \si(\mathcal C)$, $\si(\bm s) \in \si(\mathcal S)$ and
\eq{ \label{eq:rotateCSA}
\si( B_{\bm c,\bm s}(X,\tau)) = B_{\si(\bm c),\si(\bm s)}(X^\si,\tau^\si).
} 
\end{prop}

\begin{proof}
From $r_{\si(i)} = \si r_i \si^{-1}$ (as maps on $\mfh^*$) and $a_{\si(i)\si(j)}=a_{ij}$ one obtains that $T_{\si(i)} = \si T_i \si^{-1}$ so that, combined with $\tau(X)=X$, $T_{w_X} = \si T_{w_X} \si^{-1}$.
Now using \eqref{wXsigma} and the fact that $\tau^\si = \si \tau \si^{-1}$ it follows that 
\[ 
\theta_q^\si := \theta_q(X^\si,\tau^\si) = \si  \theta_q(X,\tau)  \si^{-1}.
\]
Since $\si$ permutes $Q$, it follows that $k_{\si(\mu)} = \si(k_\mu)$ for all $\mu \in Q$ so that $U_q(\mfh)^{\theta_q^\si} = \si (U_q(\mfh)^{\theta_q})$.
Also, from \eqref{sigmaUq} we immediately obtain $U_q(\mfg_{X^\si})= \si(U_q(\mfg_X))$. 

Fix tuples $\bm c \in (\K^\times)^{I\backslash X}$ and $\bm s \in \K^{I \backslash X}$. 
It follows that
\[ 
\si(b_j) = b_{\si(j)}|_{(X,\tau) \to (X^\si,\tau^\si)}, 
\]
which proves that \eqref{eq:rotateCSA} holds for $\bm c \in (\K^\times)^{I\backslash X}$ and $\bm s \in \K^{I\backslash X}$. 
To complete the proof, it suffices to prove that
\[ 
\mc{C}(X^\si,\tau^\si) = \si(\mc{C}(X,\tau)), \qq \mc{S}(X^\si,\tau^\si) = \si(\mc{S}(X,\tau)), 
\]
which follows directly from the definitions \eqrefs{Idiffdefn}{S-family}, $a_{\si(i)\si(j)}=a_{ij}$ and \eqref{wXsigma}.
\end{proof}


\subsubsection{Rotational symmetries of the representation $\RT_u$} \label{sec:natreprot}

Recall the definitions made in \eqref{AutAelts}.

\begin{defn}
Define $\Sigma_A \leq \Aut(A)$ as follows:
\eq{ \label{def:Sigma}
\Sigma_A = \left\{  \begin{aligned} 
\langle \rho \rangle & \cong \Cyc_N \; && \text{for }\wh\mfsl_N, \\
\langle \flL \rangle & \cong \Cyc_2 && \text{for }\wh\mfso_{2n+1}, \\
\langle \pi \rangle & \cong \Cyc_2 && \text{for }\wh\mfsp_{2n}, \\
\langle \flL, \pi \rangle  & \cong \Dih_4 && \text{for }\wh\mfso_{2n}.
\end{aligned} \right. \qq
}
For the specified generators of $\Sigma_A$ and generic $u$, define $Z^\si(u) \in \GL(\K^N)$ as follows:
\begin{flalign}
\nonumber \qq\qq\qq && Z^\rho(u) & = \sum_{i=1}^n E_{i,i+1} + u\, E_{N1} && \text{for }\wh\mfsl_{N}, \\
\label{def:Zsigma} && Z^{\flL}(u) &= \sum_{|i|<n} E_{ii} - u E_{n,-n} - u^{-1} E_{-n,n}  &&  \text{for }\wh\mfso_{N}, \\
\nonumber && Z^\pi(u) &= \sum_{i \in \langle N \rangle} (-1)^{i+1} (-\vartheta \, u)^{\delta_{i<0}} E_{i-\sgn(i) (n+1),i} && \text{for $\wh\mfsp_{2n}$ and $\wh\mfso_{2n}$}. && \qq\qq \defnend
\end{flalign}
\end{defn}

We wish to extend these assignments to a map $\xi_u : \Sigma_A \to \GL(\K^N): \si \mapsto Z^\si(u)$. 
If $\Sigma_A$ is cyclic, each of its elements can be written as $\si^i$ with $0 \le i < |\Sigma_A|$ with $\si$ the generator specified in \eqref{def:Sigma}.
For $\wh\mfso_{2n}$ we use the following expressions for the elements of $\Sigma_A$ in terms of the generators:
\[ 
\Sigma_A = \{ \id, \flL, \pi, \flL \pi, \pi \flL, \flL \pi \flL, \pi \flL \pi, \flL \pi \flL \pi \}. 
\]
Correspondingly we extend the assignments \eqref{def:Zsigma} multiplicatively, obtaining $Z^\si(u) \in \GL(\K^N) \subset \End(\K^N)$ for all $\si \in \Sigma_A$.
Note that the centre of $\GL(\K^N)$ is isomorphic to, and will be identified with, $\K^\times$.

\begin{prop} \label{prop:injective}
Let $u$ be generic. The composition of the above map $\xi_u$ with the canonical map $\GL(\K^N) \twoheadrightarrow \GL(\K^N) / \K^\times$ defines a group homomorphism $\Sigma_A \hookrightarrow \GL(\K^N)/\K^\times$.
\end{prop} 

\begin{proof}
In $\GL(\K^N) / \K^\times$, the correct relations are satisfied. This follows from 
\eq{\label{eq:Zpowers}
\begin{aligned}
Z^\rho(u)^N &= u \, \Id && \text{for }\wh\mfsl_N, \\
Z^{\flL}(u)^2 &= \Id && \text{for }\wh\mfso_{N}, \\
Z^\pi(u)^2 &= (-1)^n \vartheta \, u \, \Id  && \text{for $\wh\mfsp_{2n}$ and $\wh\mfso_{2n}$}, \\
\big( Z^{\flL}(u)Z^\pi(u) \big)^4 &= u^2 \, \Id && \text{for }\wh\mfso_{2n}.
\end{aligned}
}
Hence the pertinent map is a homomorphism.
The injectivity can be straightforwardly verified: in \eqref{eq:Zpowers} the exponents in the left hand sides are the smallest ones such that the respective right hand sides are multiples of $\Id$.
\end{proof}

The matrices $Z^\si(u)$ intertwine the representations $\RT_u \circ \si$ and $\RT_u$~of~$U_q(\mfg)$.

\begin{prop} \label{prop:Tsigma}
We have
\eq{\label{Tsigma} 
Z^\si(u) \, \RT_u(\si(a)) = \RT_u(a)\, Z^\si(u) \qq \text{for all } a \in U_q(\mfg).
}
\end{prop}

\begin{proof}
It is sufficient to check \eqref{Tsigma} for the specified generators $\si$ and for $a \in \{x_i^\pm, k_i^\pm \}_{i \in I}$.
This can be straightforwardly accomplished using \eqrefs{rep:A}{affrep} and \eqref{def:Zsigma}.
\end{proof}

Since the representation $\RT_u \ot \RT_v$ of $U_q(\mfg)$ is irreducible for generic values of $u/v$, Schur's lemma implies the following statement.

\begin{crl} \label{cor:RZZ1}
Let $\si \in \Sigma_A$.
The R-matrices defined by \eqref{Ru:defn} satisfy $[R(\tfrac{u}{v}),Z^\si(u) \ot Z^\si(v)]=0$.
\end{crl}

For $\wh\mfso_{2n}$ we will additionally use the notations
\eq{
\begin{aligned}
Z^{\flR} &:= Z^\pi(u)  Z^{\flL}(u) Z^\pi(u)^{-1} = \sum_{|i|>1} E_{ii} + E_{-1,1} + E_{1,-1}, \\
Z^{\flLR}(u) &:= Z^{\flL}(u) Z^{\flR} =\sum_{1<|i|<n} E_{ii} + E_{-1,1} + E_{1,-1} - u E_{n,-n} - u^{-1} E_{-n,n}.
\end{aligned}
}
Note that $Z^{\flR}$ does not depend on $u$. Indeed $\flR$ is a rotational symmetry of the representation $\RT$ of $U_q(\mfso_{2n})$: $Z^{\flR} \RT(\flR(a)) = \RT(a) Z^{\flR}$ for all $a \in U_q(\mfso_{2n})$. (We note that $Z^{\flR}$ corresponds to the matrix $T$ in \cite[Sec.~3]{Ji2}.)

\begin{rmk}
The matrices $Z^\si(u)$ can be expressed as images under $\RT_u$ of certain elements of $U_q(\mfg)$ fixed by $\si$:
\begin{align*} 
Z^\rho(u) &= \RT_u \sum_{i \in I} x_i && \text{for } \wh{\mfsl}_N , \\
Z^{\phi_1}(u) &= \RT_u \big( 1 + x_0 y_1 + y_0 x_1 - x_0 y_0 y_1 x_1 - y_0 x_0 x_1 y_1 \big)  && \text{for } \wh{\mfso}_N , \\
Z^\pi(u) &= \RT_u \bigg( x_{(n-1)--(2)} \, x_0 + x_0 \, x_{(2)--(n-1)} + x_{(1)--(n-2)} \, x_n + x_n \, x_{(n-2)--(1)}  \\
& \qq \qq + \sum_{i=2}^{n-1} \big( x_{(i-1)--(0)} \, x_{(2)--(\bar\imath -1)} + x_{(\bar\imath)--(n)} \, x_{(n-2)--(i)} \big) \bigg) && \text{for } \wh{\mfso}_{2n}, \\
Z^\pi(u) &= \RT_u  \sum_{i=1}^n \big( x_{(i-1)--(0)} \, x_{(1)--(\bar\imath -1)} + x_{(\bar \imath)--(n)} \, x_{(n-1)--(i)} \big)  && \text{for } \wh{\mfsp}_{2n},
\end{align*}
where $x_{(i)--(j)}=x_i x_{i+1} \cdots x_j$ if $i\le j$ or $x_{(i)--(j)}=x_i x_{i-1} \cdots x_j$ if $i\ge j$. 
We expect that these elements of $U_q(\mfg)$ are truncations of certain universal elements $\mc{Z}^\si$, satisfying $\mc{Z}^\si \si(a) = a \mc{Z}^\si$ for all $a \in U_q(\mfg)$. \hfill \rmkend
\end{rmk}

We emphasize that $\Sigma_A \ne \Aut(A)$ precisely for $\wh\mfsl_{N>2}$ and $\wh\mfso_8$.
In particular, for $\wh\mfsl_{N>2}$ we have $\psi \notin \Sigma_A$.
However, in order to study $\Aut(A)$-equivalence classes it is sufficient to study \mbox{$\Sigma_A$-equivalence} classes.

\begin{lemma} \label{lem:SigmaAslN}
Let $A$ be of type ${\rm A}^{(1)}_{n>1}$.
Then $\Aut(A)$-equivalent QP algebras are $\Sigma_A$-equivalent.
\end{lemma}

\begin{proof} 
Since $\Aut(A)$ is generated by $\rho$ and $\psi$, it is sufficient to show that 
$\psi(B_{\bm c,\bm s}) = \rho^m(B_{\bm c,\bm s})$ for some $m \in \{0,1,\ldots,n\}$.
First we prove that $(X^{\psi},\tau^{\psi}) = (X^{ \rho^m},\tau^{ \rho^m})$ for all $(X,\tau) \in \GSat(A)$.  
Because $\Aut(A) \cong \Dih_N$, we have $\tau = \rho^r \psi^{s}$ with $r \in \{0,1,\ldots,n\}$ and $s \in \{0,1\}$.
Suppose $s=1$. Then $\psi = \rho^{-r} \tau$ so that, the property $\tau(X)=X$ implies
\[
\tau^{\psi} = \psi \tau \psi^{-1} = \rho^{-r} \tau \rho^r = \tau^{\rho^{-r}}, \qq
X^{\psi} = \psi(X) =  \rho^{-r} (\tau(X)) = \rho^{-r}(X) = X^{\rho^{-r}}
\]
and we may take $m \in \{0,1,\ldots,n\}$ congruent to $-r$ modulo $n+1$.

If on the other hand $s = 0$ then $\tau = \rho^r$. 
Since $\tau$ is an involution we must have $r=0$ or $r=\frac{N}{2}$ with $N$ even. 
In both cases $\tau$ commutes with $\psi$. 
Let $X'$ be a component of $X$.
If $\tau = \id$, then note that $w_{X'} \ne \id$ if $X'$ is of type ${\rm A}_{t \ge 2}$; hence, $X' = \emptyset$, or $X' = \{i\}$ for some $i \in I$.
In that case condition \eqref{Satdiag2b} implies that $N$ is even and $X$ is of type ${\rm A}_1^{\times N/2}$.
If $\tau = \pi$, then because $\tau(X')=X'$, $X'$ must be a union of pairs $\{i,i+N/2\}$ (with addition modulo $N$ implied). Since $X \ne I$ it follows that $X' = \emptyset$.
The only possibilities are $(X,\tau) = (\emptyset,\id)$, $(X,\tau) = (\{0,2,\ldots,N-1\},\id)$, $(X,\tau) = (\{1,3,\ldots,N-2\},\id)$, $(X,\tau)=(\emptyset,\pi)$; in all cases $(X^{\psi},\tau^{\psi}) = (X,\tau)$ and we may take $m=0$. 

Proposition \ref{prop:rotateCSA} now implies that the equivalence of the Satake diagrams can be lifted to an equivalence of QP algebras. 
\end{proof}

As a byproduct of the proof of the previous lemma we are able to improve on Lemma \ref{L:K-unit-tw} for type ${\rm A}^{(1)}_{n>1}$.

\begin{lemma} \label{lem:twistedunitarityslN}
Let $(X,\tau) \in \Sat(A)$ with $A$ of type ${\rm A}^{(1)}_n$ such that $\tau \in \Sigma_A$. Let $\bm c \in (\K^\times)^{I \backslash X}$ such that $\bm c = \psi(\bm c)$. Let $(K(u),\eta)$ be a nontrivial solution of \eqref{intw-tw} for the coideal subalgebra $B_{\bm c,\bm s}(X,\tau)$. Then there exists $n_{\rm tw}(u) \in \K(u)$ such that \eqref{K-unit-tw} holds for generic values of $u$.
\end{lemma}

\begin{proof}
From the proof of Lemma \ref{lem:SigmaAslN} we know that $(X^{\psi},\tau^{\psi}) = (X,\tau)$; moreover we straightforwardly see that $I_{\rm diff} = \emptyset = I_{\rm nsf}$; in particular $\bm s = \bm 0:= (0,\ldots,0) \in \K^{I \backslash X}$. 
The proof of Proposition \ref{prop:rotateCSA} implies $\psi(B^{\bm c,\bm 0}(X,\tau)) = B^{\bm c,\bm 0}(X,\tau)$.
Hence, Lemma \ref{L:K-unit-tw} can be used.
\end{proof}

\begin{rmk} \label{R:non-Z}
For type ${\rm D}^{(1)}_4$ not all diagram automorphisms are elements of $\Sigma_A \cong \Dih_4$ (note that $\Sigma_A = \Aut(A) \cong \Dih_4$ for type ${\rm D}^{(1)}_{n>4}$). For example, $(14) \in \Aut(A)$ is not a symmetry of $\RT_u$. 
This is corroborated by our computations (see Sections \ref{sec:K:C1BD2}-\ref{sec:nqs} as opposed to Section \ref{sec:D4}). 
This is also true for the representation $\RT$ of $U_q(\mfso_8)$. 
Although $\Aut(A) \cong \Sym_3$ in type ${\rm D}_4$, the equation $Z\,\RT(\sigma(a)) = \RT(a)\,Z$ for all $a \in U_q(\mfso_8)$ has a nontrivial solution $Z\in\End(\K^N)$ for $\sigma \in \{ \id, (34) \}$ only, so that the $n=4$ case matches the situation for $n>4$, where $\Aut(A) = \langle \flR \rangle$. 
\hfill \rmkend
\end{rmk}

Elements of $\Sigma_A$ are called \emph{rotational symmetries of $\RT_u$} and the corresponding $Z^\si(u)$ are called \emph{rotation matrices}. For $\tau \in \Aut(A)$ we denote the corresponding $\Sigma_A$-conjugacy class by
\[ 
[\tau] :=  \{ \tau' \in \Aut(A) \, | \, \tau' = \sigma \tau \sigma^{-1} \text{ for some } \sigma \in \Sigma_A \}. 
\]
For example, for type ${\rm D}^{(1)}_4$ we have $[(14)] = \{ (03), (04), (13), (14) \}$.


\subsubsection{Rotation of K-matrices} \label{sec:rotating}

K-matrices associated to $\Sigma_A$-equivalent QP coideal subalgebras can be related as follows.

\begin{prop} \label{prop:rotateintw}
Let $\si \in \Sigma_A$ and recall the matrices $Z^\si(u)$ defined in Section \ref{sec:natreprot}.
Fix $(X,\tau) \in \GSat(A)$, $\bm c \in \mc{C}$ and $\bm s \in \mc{S}$. 
The pair $(K(u),\eta)$ is a solution to the boundary intertwining equation \eqref{intw-untw} or \eqref{intw-tw} for all $b \in B_{\bf c,\bf s}(X,\tau)$ precisely if $(K^\si(u),\eta)$, where
\eq{\label{eq:Ksigma} K^\si(u) := 
\begin{cases} 
Z^\si(\tfrac{\eta}{u})^{-1} K(u)\, Z^\si(\eta \, u) & \text{if $K(u)$ is untwisted}, \\ 
Z^\si(\tfrac{\eta}{u})^{\t} K(u)\, Z^\si(\eta \, u)  & \text{if $K(u)$ is twisted},
\end{cases}
} 
is a solution to the same equation for all $b \in \si( B_{\bm c,\bm s}(X,\tau))$.
\end{prop}

\begin{proof} 
Let $b \in B_{\bm c,\bm s}(X,\tau)$. Using the expression \eqref{eq:Ksigma} and the relation \eqref{Tsigma}, one straightforwardly derives that $K^\si(u)\, \RT_{\eta u}(\si(b)) = \RT_{\eta/u}(\si(b))\, K^\si(u)$ is equivalent to $K(u)\, \RT_{\eta u}(b) = \RT_{\eta/u}(b)\, K(u)$. In the twisted case one additionally uses that the antipode $S$ and $\si$ commute. \qedhere
\end{proof}

Proposition \ref{prop:rotateintw} implies that it is sufficient to solve the intertwining relation \eqref{intw-untw} for one representative $(X,\tau)$ of each $\Sigma_A$-orbit in $\GSat(A)$. 
K-matrices associated to other generalized Satake diagrams in the same $\Sigma_A$-orbit can then be obtained from the K-matrix associated to the chosen representative diagram by means of \eqref{eq:Ksigma}. 

\begin{exam} \label{exam:representative}
For ${\rm A}^{(1)}_3$ the following Satake diagrams $(X,\tau) $ form one $\Sigma_A$-orbit: 
\[ 
(\{2\},(1 3)), \qq (\{3\},(0 2) ), \qq (\{0\},(13)), \qq (\{1\},(02)). 
\]
The last three diagrams are related to the first by means of the automorphisms $\rho$, $\rho^2$ and $\rho^3$, respectively. 
Proposition \ref{prop:rotateintw} entails that it is sufficient to find K-matrix associated to the Satake diagram $ (\{2\},(13))$ and then apply \eqref{eq:Ksigma} with an appropriate power of~$\rho$.

To make things more explicit, set $(X',\tau') = (\{3\},(02))$. Suppose we wish to solve $K'(u)\,\RT_{\eta' u}(b) = \RT_{\eta' / u}(b)\,K'(u)$ for all $b \in B_{\bm c',\bm s'}(X',\tau')$ for some $\bm c' =  (c'_0, c'_1, c'_2) \in \mc{C}(X',\tau')$ and $\bm s' = (0,s'_1,0) \in \mc{S}(X',\tau')$.
Note that $X' = X^\si, \tau' = \tau^\si$ where $\si = (0123)$ and $(X,\tau) = (\{2\},(13))$ is the above representative.
In Section \ref{sec:A3} we solve $K(u)\,\RT_{\eta \, u}(b) = \RT_{\eta / u}(b)\,K(u)$ for all $b \in B_{\bm c,\bm s}(X,\tau)$ and all $\bm c = (c_0, c_1, c_2) \in \mc{C}(X,\tau)$, $\bm s = (s_0,0,0)\in \mc{S}(X,\tau)$.  
In particular, we may take $\bm c = \si^{-1}(\bm c') \in \si^{-1}(\mc{C}(X',\tau')) = \mc{C}(X,\tau)$ and $\bm s = \si^{-1}(\bm s') \in \si^{-1}(\mc{S}(X',\tau')) = \mc{S}(X,\tau)$.
As per Proposition \ref{prop:rotateCSA}, we have $c_i = c'_{\si(i)}$ and $s_i = s'_{\si(i)}$ so that
\eq{ \label{csrelation} 
c_0 = c'_1, \qq c_1 = c'_2, \qq c_3 = c'_0, \qq s_0 = s'_1. 
}
Consider a pair $(K(u),\eta)$ that solves \eqref{intw-untw} for these particular values $\bm c$ and $\bm s$; it is uniquely defined up to a scalar factor and a choice of sign. 
Then according to Proposition \ref{prop:rotateintw}, the pair $(K^\si(u),\eta)$ with $K^\si(u) = Z^\si(\tfrac{\eta}{u})^{-1} K(u) Z^\si(\eta \, u)$ given by \eqref{eq:Ksigma} solves $K^\si(u)\, \RT_{\eta \, u}(b) = \RT_{\eta / u}(b)\, K^\si(u)$ for all $b \in \si(B_{\bm c,\bm s}(X,\tau)) = B_{\si(\bm c),\si(\bm s)}(X^\si,\tau^\si) = B_{\bm c',\bm s'}(X',\tau')$ as required. 
To express $K'(u)=K^\si(u)$ and $\eta'=\eta$ in terms of the original $\bm c'$ and $\bm s'$, in the expressions for $K^\si(u)$ and $\eta$ one should replace the $c_i$ and $s_i$ by the appropriate $c'_j$ and $s'_j$ according to \eqref{csrelation}.
\hfill \examend \end{exam}


\subsection{Dressing} \label{sec:dressing}

Let us consider QP algebras which are $\Ad(\wt{H}_q)$-equivalent. 
We will see that the underlying diagrams $(X,\tau)$ must be the same and the respective tuples $\bm c$ and $\bm s$ are related by straightforward transformations. 
The corresponding solutions of the boundary intertwining equation \eqref{intw-untw} or \eqref{intw-tw} are similar or congruent as matrices, respectively. 
It allows us to reduce our workload further: we only have to solve the intertwining equations for QP algebras $B_{\bm c,\bm s}(X,\tau)$ for suitably chosen $\bm c \in \mc{C}$ and $\bm s \in \mc{S}$. 


\subsubsection{$\Ad(\wt{H}_q)$-equivalences of QP algebras} 

The following is a straightforward extension of \cite[Prop.~9.2 (1)]{Ko1} and \cite[Prop.~3.16 (2)]{BgKo1}. 

\begin{prop} \label{prop:Htildeequiv}
Fix $(X,\tau) \in \GSat(A)$, $\bm c \in \mc{C}$ and $\bm s \in \mc{S}$.
For all $\chi \in \wt H_q$ we have
\eq{ \label{wtHqequivalence}
\Ad( \chi ) (B_{\bm c,\bm s}) = B_{\bm c',\bm s'} 
}
where $\bm c'\in \mc{C}$ and $\bm s' \in \mc{S}$ are determined by
\eq{ \label{cstransformations}
c'_j = \chi(\al_j + w_X(\al_{\tau(j)})) c_j, \qq s'_j = \chi(\al_j) s_j, \qq \text{for } j \in I \backslash X.
}
\end{prop}

\begin{proof}
Note that \eqref{def:b_j} implies that, for $j \in I \backslash X$, we have
\[ b_j = y_j - c_j T_{w_X}(x_{\tau(j)})k^{-1}_j - s_j k^{-1}_j \]
so that, using \eqref{thetaq:rootspace}, we have
\eq{
\Ad ( \chi ) (b_j) = \chi(\al_j)^{-1} \Big( y_j - \chi(\al_j + w_X(\al_{\tau(j)})) c_j T_{w_X}(x_{\tau(j)}) k_j^{-1} -   \chi(\al_j)  s_j k_j^{-1} \Big).
}
The observation that $\Ad(\chi)$ acts on the natural generators of $U_q(\mfg_X)$ and $U_q(\mfh)^{\theta_q}$ by scalar multiplication completes the proof.
\end{proof}

Recall the set $I^* \subseteq I \backslash X$ parametrizing the $\tau$-orbits in $I \backslash X$. 
Let $\bm c \in \mc{C}$ and $\bm s \in \mc{S}$ be given and consider \eqref{cstransformations}. 
By choosing $\chi$ suitably, it is possible to fix those $c'_j$ with $j \in I^*$ to arbitrary values (e.g. to 1) as follows (here we need that $\K$ is quadratically closed).

Suppose $j$ is such that $j \ne \tau(j)$.
Then \eqref{A0} implies $c'_{\tau(j)}=c_{\tau(j)}c'_j/ c_j$. 
If $j \notin I_{\rm diff}$ we naturally have $c'_{\tau(j)}=c'_j$; however if $j \in I_{\rm diff}$ the value $c'_{\tau(j)}$ cannot be fixed to an arbitrary value, since $c'_j$ has been fixed already and $c_j$ and $c_{\tau(j)}$ are given. 
Alternatively, if $j = \tau(j)$ then either $j \notin I_{\rm nsf}$ or $j \in I_{\rm nsf}$. 
In the former case $s'_j=s_j=0$; in the latter case $a_{ij}=0$ for all $i \in X$ so, using Lemma \ref{lem:Xjdecomposition}, we obtain $\chi(w_X(\al_{\tau(j)}))=\chi(\al_j)$.
Therefore $s'_j = s_j  (c'_j/c_j)^{1/2}$; we see that $s'_j$ cannot be assigned an arbitrary value.

Combining this with Proposition \ref{prop:QPalgebras} (i) we obtain the following corollary.

\begin{crl} \label{cor:nonremovable}
Let $(X,\tau) \in \GSat(A)$, $\bm c \in \mc{C}$ and $\bm s \in \mc{S}$.
Then $B_{\bm c,\bm s}(X,\tau)$ is $\Ad(\wt H_q)$-equivalent to a QP algebra with $|I_{\rm diff} \cup I_{\rm nsf}|$ free parameters. The Hopf algebra automorphism accomplishing this equivalence will contain the remaining $|I^*|$ degrees of freedom.
\end{crl}

To have uniform expressions for the intertwiners $K(u)$ it is convenient to let the generators of $B_{\bm c,\bm s}$ associated to a special orbit depend on a free parameter $\la$ appearing in $K(u)$ in the following way.
For all $j \in I_{\rm diff}$ it is convenient to choose $c_j/c_{\tau(j)}$ equal to $\la^2$ times a signed integer power of $q_j$, unless $A$ is of type ${\rm C}^{(1)}_n$ or ${\rm D}^{(1)}_n$, $\tau \in [\pi]$ and $X$ is of maximal size, in which case we choose $c_j/c_{\tau(j)} = \la^4$.
Similarly for all $j \in I_{\rm nsf}$ we will choose $s_j/c_j^{1/2}$ equal to $\tfrac{\la - \la^{-1}}{q_j-q_j^{-1}}$ or $\tfrac{\la + \la^{-1}}{q_j-q_j^{-1}}$ times a square root of a signed integer power of $q_j$.


\subsubsection{Dressing symmetries of the representation $\RT_u$} \label{sec:natrepdress}

On the level of the representation $\RT_u$, Hopf algebra automorphisms induced by elements of $\wt H_q$ appear as conjugations by diagonal matrices; in other words there are diagonal matrices which intertwine the representations $\RT_u$ and $\RT_u \circ \Ad(\chi)$ for arbitrary $\chi \in \wt H_q$.
Let us make this statement precise.

\begin{defn} \label{defnG}
Denote
\eqa{ \label{Omega:defn}
\Omega_n &= \begin{cases} 
\{ \bm \omega \in (\K^\times)^N \, | \, \prod_{i=1}^N \om_i = 1 \} \hspace{4.3mm}\; & \text{for } \wh\mfsl_N, \\
(\K^\times)^n & \text{for } \wh\mfso_N \text{ and } \wh\mfsp_N.
\end{cases}
\intertext{For $\bm \om \in \Omega_n$, define the \emph{dressing matrix} $G(\bm \omega) \in {\rm SL}(\K^N)$ by}
\label{G:defn}
G(\bm \omega) &= 
\begin{cases} 
\sum\limits_{i=1}^N \omega_i E_{ii}
& \text{for }\wh\mfsl_{N}, \\[1em]
E_{00} + \sum\limits_{i=1}^n (\omega_i E_{-\bar \imath,-\bar \imath} + \omega_i^{-1} E_{\bar \imath \, \bar \imath}) & \text{for }\wh\mfso_{2n+1}, \\[1em]
\sum\limits_{i=1}^n (\omega_i E_{-\bar \imath,-\bar \imath} + \omega_i^{-1} E_{\bar \imath \, \bar \imath})  & \text{for $\wh\mfsp_{2n}$ and $\wh\mfso_{2n}$}.
\end{cases} 
\intertext{Finally, for $\bm \om \in \Omega_n$ and $\eta \in \K^\times$ define $\zeta(\bm \omega,\eta) \in \wt H_q$ as follows:}
\label{zetaomega:defn}
& \hspace{-6.4cm} \begin{aligned}
\zeta(\bm \omega,\eta)(\al_i) &= 
\om_i \om_{i+1}^{-1} \hspace{4.05cm} 
&& \text{ for } 1\le i \le n-1, 
\\
\zeta(\bm \omega,\eta)(\al_n) &= 
\casesl{l}{ 
\om_n \om_{n+1}^{-1} \\
\om_n \\
\om_n^2 \\
\om_{n-1} \om_{n}
} 
&& 
\casesm{l}{ 
\text{ for }\wh\mfsl_N, \\
\text{ for }\wh\mfso_{2n+1}, \\
\text{ for }\wh\mfsp_{2n}, \\
\text{ for }\wh\mfso_{2n},
}
\\
\hspace{45mm} 
\zeta(\bm \omega,\eta)(\al_0) &= 
\casesl{l}{ 
\eta \, \om_1^{-1} \om_N \\
\eta \, \om_1^{-1} \om_2^{-1} \\
\eta \, \om_1^{-2} 
} 
&& 
\casesm{l}{ 
\text{ for }\wh\mfsl_{N}, \\
\text{ for }\wh\mfso_{N}, \\
\text{ for }\wh\mfsp_{2n}.
} 
\end{aligned} 
}

\vspace{-1.5em}\hfill \defnend
\end{defn}

\smallskip

The matrices $G(\bm \om)$ intertwine the representations $\RT_{\eta u}$ and $\RT_u \circ \Ad(\zeta(\bm \om,\eta))$ ~of~$U_q(\mfg)$:

\begin{prop} \label{prop:Gomegaequiv}
Let $\bm \omega$ be as in Definition \ref{defnG} and let $\eta \in \K^\times$. For all $a \in U_q(\mfg)$ we have
\eq{ 
G(\bm \omega)\, \RT_{\eta u}(a) = \RT_u(\Ad(\zeta(\bm \omega,\eta) )(a))\, G(\bm \omega). \label{Tomega} 
}
\end{prop}

\begin{proof}
This follows from a straightforward check using the formulas \eqrefs{rep:A}{affrep} and \eqrefs{G:defn}{zetaomega:defn}.
\end{proof}

Analogously to Corollary \ref{cor:RZZ1} we have the following.

\begin{crl}
Let $\bm \omega$ be as in Definition \ref{defnG}.
The R-matrices defined by \eqref{Ru:defn} satisfy 
\eq{
[R(\tfrac{u}{v}),G(\bm \om) \ot G(\bm \om)]=0.
}
\end{crl}


\subsubsection{Dressing of K-matrices} \label{sec:Kmatdressing}

We have a statement analogous to Proposition \ref{prop:rotateintw} for $\Ad(\wt{H}_q)$-equivalent QP algebras.

\begin{prop} \label{prop:dressintw}
Let $\bm \omega$ be as in Definition \ref{defnG} and $\eta \in \K^\times$. 
Fix $(X,\tau) \in \GSat(A)$, $\bm c \in \mc{C}$ and $\bm s \in \mc{S}$. 
The pair $(K(u;\bm \omega),\eta)$, where
\eq{
\label{eq:Kdressing} 
K(u;\bm \om) := \begin{cases} 
G(\bm \om)^{-1} K(u)\, G(\bm \om)  & \text{if $K(u)$ is untwisted}, \\ 
G(\bm \om)\, K(u)\, G(\bm \om) & \text{if $K(u)$ is twisted},
\end{cases}
}
is a solution to the boundary intertwining equation \eqref{intw-untw} or \eqref{intw-tw}, respectively, for all $b \in B_{\bm c,\bm s}(X,\tau)$ precisely if $(K(u),1)$ is a solution to the same equation for all $b \in \Ad(\zeta(\bm \omega,\eta))(B_{\bf c,\bf s}(X,\tau))$.
\end{prop}

\begin{proof} 
This is completely analogous to the proof of Proposition \ref{prop:rotateintw}.
Let $b \in B_{\bm c,\bm s}(X,\tau)$. Using the expression \eqref{eq:Kdressing} and the relation \eqref{Tomega}, one straightforwardly derives that $K(u;\bm \om)\, \RT_{\eta u}(b) = \RT_{\eta/u}(b)\, K(u;\bm \omega)$ is equivalent to $K(u)\, \RT_{u}(\Ad(\zeta(\bm \omega,\eta))(b)) = \RT_{1/u}(\Ad(\zeta(\bm \omega,\eta) )(b))\, K(u)$. In the twisted case one additionally uses that that the antipode $S$ and $\Ad(\zeta(\bm \omega,\eta) )$ commute. 
\end{proof}

Proposition \ref{prop:dressintw} implies that we need to solve the intertwining equation \eqref{intw-untw} or \eqref{intw-tw} only for $\eta=1$ and for particular representatives of the $\Ad(\wt H_q)$-equivalence classes.
We will use this to simplify and structure the presentation of the K-matrices classified in this paper, by combining it with Corollary \ref{cor:nonremovable}; in that light it turns out to be natural to call the matrix $K(u;\bm \omega)$ a \emph{dressed K-matrix} and the matrix $K(u)$ a \emph{bare K-matrix}; furthermore the free parameters $\omega_i$ are called \emph{dressing parameters} and the remaining free parameters in $K(u)$ the \emph{non-removable free parameters}, which are associated to the sets $I_{\rm diff}$ and $I_{\rm nsf}$, {\it i.e.}~to the special orbits.

Let us explain this in more detail.
Consider $(X,\tau) \in \GSat(A)$ and let $\bm c \in \mc{C}$ and $\bm s \in \mc{S}$ be given.
According to Corollary \ref{cor:nonremovable}, $\Ad(\chi)(B_{\bm c,\bm s})=B_{\bm c',\bm s'}$ with $\chi$, depending on the free parameters $c_i$ with $i \in I^*$, chosen so that $B_{\bm c',\bm s'}$ only has $|I_{\rm diff} \cup I_{\rm nsf}|$ degrees of freedom (those $c_{\tau(i)}$ with $i \in I^* \cap I_{\rm diff}$ and those $s_i$ with $i \in I_{\rm nsf}$). 
Now choose 
\[ \chi = \zeta(\bm \omega,\eta). \] 
Then from Proposition \ref{prop:dressintw} we know that $(K(u;\bm \omega),\eta)$ satisfies the boundary intertwining equation associated to $B_{\bm c,\bm s}$ if and only if $(K(u),1)$ satisfies the boundary intertwining equation for $B_{\bm c',\bm s'}$.
According to Corollary \ref{cor:nonremovable}, the number of degrees of freedom in $\chi$ is $|I^*|$; since we equate $\chi$ to $\zeta(\bm \omega,\eta)$ there are also $|I^*|$ degrees of freedom among the dressing parameters and the scaling parameter. 
The scaling parameter $\eta$ will take one of these degrees of freedom and there are $|I^*|-1$ degrees of freedom to be distributed over $n$ dressing parameters $\om_i$. This means, unless $I^*=I$, there is some redundancy among the dressing parameters. 
It is always possible to set $n+1-|I^*|$ of them to 1; the remaining $|I^*|-1$ are called \emph{effective dressing parameters}.

In our case-by-case results we will write the ensuing relations between the algebra parameters $\bm c$ and $\bm s$ on the one hand and the effective dressing parameters, the scaling parameter and any non-removable free parameters on the other in the form $c_j=\ldots$ and $s_j=\ldots$ for $j \in I \backslash X$.

\begin{exam}
Let $A$ be of type ${\rm A}^{(1)}_3$. 
Then the choice $X = \{ 2\}$, $\tau = (13)$ defines a Satake diagram.
Choose $I^* = \{0,1\}$. We have $I_{\rm nsf} = \{0\}$ and $I_{\rm diff} = \{1\}$.
Note that $w_X = r_2$ so that $w_X(\al_1)-\al_1 = w_X(\al_3)-\al_3 = \al_2$.
Then $\bm c = (c_0,c_1,c_3)$ and $\bm s = (s_0,0,0)$ with four free parameters.
From Proposition \ref{prop:Htildeequiv} it follows that $B_{\bm c,\bm s} = \Ad(\chi) (B_{\bm c',\bm s'})$, where
\[ c'_0 = \chi(\al_0)^2 c_0 , \qq c'_1 = \chi(\al_1) \chi(\al_2) \chi(\al_3) c_1 , \qq c'_3 = \chi(\al_1) \chi(\al_2) \chi(\al_3) c_3 , \qq s'_0 = \chi(\al_0) s_0 \]
for some $\chi \in \wt H_q$.
On the other hand, Proposition \ref{prop:Gomegaequiv} suggests to consider $\Ad(\zeta(\bm \om,\eta))(B_{\bm c,\bm s})$.
The choice $\chi = \zeta(\bm \omega,\eta)$ implies
\[ c_0 = c'_0 \eta^{-2} \om_1^2 \om_4^{-2}, \qq c_1 = c'_1 \om_1^{-1} \om_4, \qq c_3 = c'_3 \om_1^{-1} \om_4, \qq s_0 = s'_0 \eta^{-1} \om_1 \om_4^{-1} \]
or, equivalently,
\begin{flalign*}
&& \eta^2 = \frac{c_0' c_1' c_3'}{c_0 c_1 c_3}, \qq \left( \frac{\om_1}{\om_4} \right)^2 = \frac{c_1 c_3}{c_1' c_3'}, \qquad \frac{c_1'}{c_3'} = \frac{c_1}{c_3}, \qquad (s_0')^2 = s_0^2 \frac{c_0'}{c_0} \left( \frac{c_1' c_3'}{c_1 c_3} \right)^2. && \examend
\end{flalign*}
\end{exam}


\section{Classification of generalized Satake diagrams of affine type} \label{sec:Satdiagclassification}

Satake diagrams associated to (untwisted and twisted) affine Dynkin diagrams were classified in \cite{BBBR}.
In this section we will in part review the classification in as much it pertains to Dynkin diagrams of untwisted affine Lie type A, B, C or D, extend it to include weak Satake diagrams, and introduce additional notation that will be useful in classifying K-matrices.
We restrict to untwisted affine Cartan matrices of classical Lie type, although many observations apply to twisted affine Cartan matrices or affine Cartan matrices of exceptional Lie type.


\subsection{Diagram involutions} \label{sec:diagraminvolutions} 

We introduce additional terminology involving diagram involutions.

\begin{defn}
Let $\tau \in \Aut(A)$ be an involution.
Those $Y \subseteq I$ consisting of the labelled nodes subject to $|i-i'|=1$ in the following subdiagrams and low-rank diagrams are called \emph{$\tau$-lateral}:
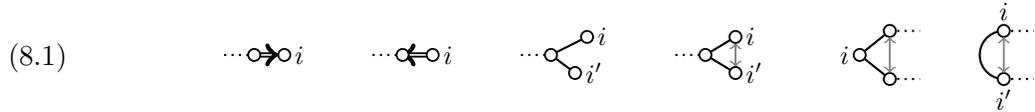
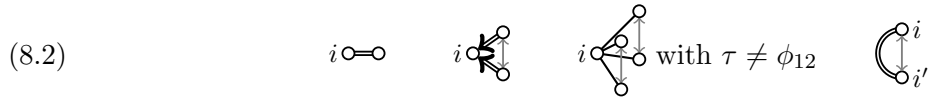
\begin{gather}
\label{lateralsets1}
\begin{tikzpicture}[baseline=-0.35em,scale=0.8,line width=0.7pt]
\draw[double,->] (0,0) -- (.4,0);
\draw[thick,dotted] (-.5,0) -- (0,0);
\filldraw[fill=white] (0,0) circle (.1);
\filldraw[fill=white] (.5,0) circle (.1) node[right]{\small$i$};
\end{tikzpicture}
\qq
\begin{tikzpicture}[baseline=-0.35em,scale=0.8,line width=0.7pt]
\draw[double,<-] (0,0) -- (.4,0);
\draw[thick,dotted] (-.5,0) -- (0,0);
\filldraw[fill=white] (0,0) circle (.1);
\filldraw[fill=white] (.5,0) circle (.1) node[right]{\small$i$};
\end{tikzpicture}
\qq
\begin{tikzpicture}[baseline=-0.35em,scale=0.8,line width=0.7pt]
\draw[thick] (.6,.3) -- (0,0) -- (.4,-.3);
\draw[thick,dotted] (-.5,0) -- (0,0);
\filldraw[fill=white] (.6,.3) circle (.1) node[right]{\small$i$};
\filldraw[fill=white] (.4,-.3) circle (.1) node[right]{\small$i'$};
\filldraw[fill=white] (0,0) circle (.1);
\end{tikzpicture}
\qq
\begin{tikzpicture}[baseline=-0.35em,scale=0.8,line width=0.7pt]
\draw[thick] (.5,.3) -- (0,0) -- (.5,-.3);
\draw[thick,dotted] (-.5,0) -- (0,0);
\draw[<->,gray] (.5,.2) -- (.5,-.2);
\filldraw[fill=white] (.5,.3) circle (.1) node[right]{\small$i$};
\filldraw[fill=white] (.5,-.3) circle (.1) node[right]{\small$i'$};
\filldraw[fill=white] (0,0) circle (.1);
\end{tikzpicture}
\qq 
\begin{tikzpicture}[baseline=-0.35em,scale=0.8,line width=0.7pt]
\draw[thick] (0,.4) -- (-.5,0) -- (0,-.4);
\draw[thick,dotted] (0,.4) -- (.5,.4);
\draw[thick,dotted] (0,-.4) -- (.5,-.4);
\filldraw[fill=white] (-.5,0) circle (.1) node[left]{\small $i$};
\filldraw[fill=white] (0,.4) circle (.1);
\filldraw[fill=white] (0,-.4) circle (.1);
\draw[<->,gray] (0,.3) -- (0,-.3);
\end{tikzpicture} 
\qq
\begin{tikzpicture}[baseline=-0.35em,scale=0.8,line width=0.7pt]
\draw[thick,domain=90:270] plot({.4*cos(\x)},{.4*sin(\x)});
\draw[thick,dotted] (0,.4)  -- (.5,.4) ;
\draw[thick,dotted] (0,-.4) -- (.5,-.4);
\filldraw[fill=white] (0,.4) circle (.1) node[above]{\small $i$};
\filldraw[fill=white] (0,-.4) circle (.1) node[below]{\small $i'$};
\draw[<->,gray] (0,.3) -- (0,-.3);
\end{tikzpicture} 
\\
\label{lateralsets2}
\begin{tikzpicture}[baseline=-0.35em,scale=0.8,line width=0.7pt]
\draw[double] (0,0) -- (.5,0);
\filldraw[fill=white] (0,0) circle (.1) node[left]{\small$i$};
\filldraw[fill=white] (.5,0) circle (.1);
\end{tikzpicture}
\qq
\begin{tikzpicture}[baseline=-0.35em,scale=0.8,line width=0.7pt]
\draw[double,<-] (1.09,.063) -- (1.5,.35);
\draw[double,<-] (1.09,-.063) -- (1.5,-.35);
\filldraw[fill=white] (1,0) circle (.1) node[left]{\small$i$};
\draw[<->,gray] (1.5,.25) -- (1.5,-.25);
\filldraw[fill=white] (1.5,.35) circle (.1);
\filldraw[fill=white] (1.5,-.35) circle (.1);
\end{tikzpicture} 
\qq
\begin{tikzpicture}[baseline=-0.35em,scale=0.8,line width=0.7pt]
\draw[thick] (.7,.7) -- (0,0) -- (.4,.2);
\draw[thick] (.7,-.1) -- (0,0) -- (.4,-.6);
\draw[<->,gray] (.4,.1) -- (.4,-.5);
\draw[<->,gray] (.7,.6) -- (.7,0);
\filldraw[fill=white] (0,0) circle (.1) node[left]{\small $i$};
\filldraw[fill=white] (.4,.2) circle (.1);
\filldraw[fill=white] (.4,-.6) circle (.1);
\filldraw[fill=white] (.7,.7) circle (.1);
\filldraw[fill=white] (.7,-.1) circle (.1);
\end{tikzpicture}
\text{ with } \tau \ne \phi_{12}
\qq
\begin{tikzpicture}[baseline=-0.35em,scale=0.8,line width=0.7pt]
\draw[double,domain=90:270] plot({.4*cos(\x)},{.4*sin(\x)});
\filldraw[fill=white] (0,.4) circle (.1) node[right]{\small $i$};
\filldraw[fill=white] (0,-.4) circle (.1) node[right]{\small $i'$};
\draw[<->,gray] (0,.3) -- (0,-.3);
\end{tikzpicture}
\end{gather}
Moreover, 
\begin{itemize}  [itemsep=.25ex]
\item the subsets $\{ i \}$ appearing in the first diagram of \eqref{lateralsets1} and the first diagram of \eqref{lateralsets2} are said to be of type ${\rm B}_1$;
\item the subsets $\{ i \}$ appearing in the second diagram of \eqref{lateralsets1} and the first diagram of \eqref{lateralsets2} are said to be of type ${\rm C}_1$;
\item the subsets $\{ i,i' \}$ appearing in the middle two diagrams of \eqref{lateralsets1} are said to be of type ${\rm D}_2$;
\item the subsets $\{i \}$ and $\{ i,i' \}$ appearing in the last two diagrams of \eqref{lateralsets1} or anywhere in \eqref{lateralsets2} are called \emph{hinges}.
\end{itemize}
The set of all $\tau$-lateral sets is denoted ${\mc L}(\tau)$. 
For $(X,\tau) \in \GSat(A)$, write 
\eq{  \label{Iprime}
I' = (I \backslash X) \backslash \bigcup_{Y \in \mc{L}(\tau)} Y.
}
If $\tau$ fixes $I'$ elementwise it is said to be of \emph{identity type}.
In particular, $\tau$ is called a \emph{flip} if it acts nontrivially on a $\tau$-lateral set and fixes all other nodes of $I$. 
\hfill \defnend
\end{defn}
Note that both nodes of the diagram of type ${\rm A}^{(1)}_1$ with $\tau = \psi = \id$ constitute hinges and subdiagrams of type ${\rm B}_1$ and ${\rm C}_1$; this is convenient as this diagram appears as the first case in many infinite series.
The conditions $|i-i'|=1$ and $\tau \ne \phi_{12}$ are imposed to ensure compatibility with $\Sigma_A$, the group of symmetries of $\RT$ and $\RT_u$, in the case that $\mfg^{\rm fin} = \mfso_8$ (see Remark \ref{R:non-Z}).
Hence we treat the diagram involutions $(03)$, $(04)$, $(13)$ and $(14)$ separately from $(01)$ and $(34)$ and the diagram involutions $(01)(34)$ separately from $(03)(14)$ and $(04)(13)$.

\smallskip

In view of Section \ref{sec:rotation}, for our purposes it is sufficient to consider $\Sigma_A$-equivalence classes of involutive diagram automorphisms. 
As representative diagram involutions we will use the following (see \eqref{AutAelts} for their definitions):
\begin{itemize}  [itemsep=.25ex]
\item The identity $\id$.
\item The flip $\flL$ for diagrams of type ${\rm B}^{(1)}_n$ and the flip $\flR$ for diagrams of type ${\rm D}^{(1)}_n$.
\item The composition of distinct commuting flips $\flLR$ for diagrams of type ${\rm D}^{(1)}_n$.
\item The involution $\psi$ for diagrams of type ${\rm A}^{(1)}_n$, which has two hinges.
\item For diagrams of type ${\rm A}^{(1)}_{N-1}$ with $N$ even: the involutions $\psi'$ (which has two hinges) and $\pi$ (there are no $\pi$-lateral sets if $N>2$). 
\item The involution $\pi$ for diagrams of type ${\rm C}^{(1)}_n$ and ${\rm D}^{(1)}_n$, which has one hinge.
\item For diagrams of type ${\rm D}^{(1)}_4$, the automorphism $(14) \notin [\flR]$.
Because $\Sigma_A < \Aut(A)$ and there is no a statement analogous to Lemma \ref{lem:SigmaAslN}, the classification of Satake diagrams in \cite{BBBR} in the case ${\rm D}_4^{(1)}$, yields all $\Sigma_A$-conjugacy classes of diagram involutions except $[(14)]$. 
\end{itemize}
Note that in type ${\rm A}^{(1)}_{n > 1}$ we have $\Aut(A) \backslash \Sigma_A = [\psi] \cup [\psi']$. 
On the other hand, $\psi = \id \in \Sigma_A$ and $\psi' = \pi \in \Sigma_A$ in type ${\rm A}^{(1)}_1$.


\subsection{Classification of generalized Satake diagrams of affine type}

As an aid to determining the set $\GSat(A)$ for untwisted affine Cartan matrices $A$ of classical Lie type, in Table~\ref{wXrhoveeX} we list explicit expressions for $w_X$ and $\rho^\vee_X$ for subsets $X \subseteq I$ such that $A_X$ is a of finite type. 

\begin{table}[t]
\caption{Expressions for $w_X$ and $\rho^\vee_X$ for $A_X$ of finite type. Note that $-w_X$ acts on $\{ \al_i \, | \, i \in X\}$ as the identity except in the cases ${\rm A}_{t \ge 2}$ and, for odd $t$, ${\rm D}_{t \ge 3}$, in which case it acts as the unique nontrivial diagram automorphism.}
\label{wXrhoveeX}
\setlength\extrarowheight{2pt}
\begin{tabular}{llm{59mm}m{38mm}}
\hline
Name & Diagram & $w_X$ & $\rho^\vee_X$  \\
\hline
A$_{t \geq 1}$ & 
\begin{minipage}[c]{25mm}
\begin{tikzpicture}[baseline=1em,scale=.9]
\draw[thick] (0,0) -- (.5,0);
\draw[thick,dashed] (.5,0) -- (1.5,0) ;
\draw[thick] (1.5,0) --  (2.0,0);
\filldraw[fill=black] (0,0) circle (.1) node[below=1pt] {\small $1$};
\filldraw[fill=black] (.5,0) circle (.1) node[below=1pt] {\small $2$};
\filldraw[fill=black] (1.5,0) circle (.1) node[below=1pt] {\small $t\!-\!1$};
\filldraw[fill=black] (2.0,0) circle (.1) node[right=1pt] {\small $t$};
\end{tikzpicture} 
\end{minipage}
& \hspace{-8pt} $\begin{array}{l} (r_1 r_2 \cdots r_{t-1} r_t r_{t-1} \cdots r_2 r_1) \cdot \\ \qu \cdot (r_2 \cdots r_{t-1} \cdots r_2) \cdots \\ \qq \cdot \begin{cases} \cdots r_{\frac{t+1}{2}} & \text{if } t \text{ is odd} \\ \cdots (r_{\frac{t}{2}} r_{\frac{t}{2}+1} r_{\frac{t}{2}}) & \text{if } t \text{ is even} \end{cases} \\[1.5em] \end{array}$
& $\displaystyle \sum_{i=1}^t \frac{i(t+1-i)}{2} h_i$ \\
\hline
B$_{t \geq 1}$ & 
\begin{minipage}[c]{25mm}
\begin{tikzpicture}[baseline=1em,scale=.9]
\draw[thick] (0,0) -- (.5,0);
\draw[thick,dashed] (.5,0) -- (1.5,0) ;
\draw[double,->] (1.5,0) --  (1.9,0);
\filldraw[fill=black] (0,0) circle (.1) node[below=1pt] {\small $1$};
\filldraw[fill=black] (.5,0) circle (.1) node[below=1pt] {\small $2$};
\filldraw[fill=black] (1.5,0) circle (.1) node[below=1pt] {\small $t\!-\!1$};
\filldraw[fill=black] (2.0,0) circle (.1) node[right=1pt] {\small $t$};
\end{tikzpicture} 
\end{minipage}
& \multirow{2}{*}{\hspace{-8pt} $\begin{array}{l} \\[-12pt] (r_1 r_2 \cdots r_{t-1} r_t r_{t-1} \cdots r_2 r_1) \cdot \\ \qu \cdot (r_2 \cdots r_{t-1} r_t r_{t-1} \cdots r_2) \cdots \\ \qq \cdots (r_{t-1} r_t r_{t-1}) r_t \end{array}$}
& \vspace{2pt} $\displaystyle \sum_{i=1}^t \frac{i(2t+1-i)}{2} h_i$ \vspace{3pt}  \\
\cline{1-2} \cline{4-4}
C$_{t \geq 1}$ & 
\begin{minipage}[c]{25mm}
\begin{tikzpicture}[baseline=1em,scale=.9]
\draw[thick] (0,0) -- (.5,0);
\draw[thick,dashed] (.5,0) -- (1.5,0) ;
\draw[double,<-] (1.6,0) --  (2.0,0);
\filldraw[fill=black] (0,0) circle (.1) node[below=1pt] {\small $1$};
\filldraw[fill=black] (.5,0) circle (.1) node[below=1pt] {\small $2$};
\filldraw[fill=black] (1.5,0) circle (.1) node[below=1pt] {\small $t\!-\!1$};
\filldraw[fill=black] (2.0,0) circle (.1) node[right=1pt] {\small $t$};
\end{tikzpicture}
\end{minipage}
& 
& \vspace{2pt} $\displaystyle \sum_{i=1}^t \frac{i(2t-i)}{2} h_i$ \vspace{3pt}  \\
\hline
D$_{t \geq 2}$ & 
\begin{minipage}[c]{25mm}
\begin{tikzpicture}[scale=.9]
\draw[thick] (0,0) -- (.5,0);
\draw[thick,dashed] (.5,0) -- (1.5,0) ;
\draw[thick] (1.9,.3) -- (1.5,0) -- (2.1,-.3);
\filldraw[fill=black] (0,0) circle (.1) node[below=1pt] {\small $1$};
\filldraw[fill=black] (.5,0) circle (.1) node[below=1pt] {\small $2$};
\filldraw[fill=black] (1.5,0) circle (.1) node[below=1pt] {\small $\hspace{-9pt} t\!-\!2$};
\filldraw[fill=black] (1.9,.3) circle (.1) node[above=1pt] {\small $t\!-\!1$};
\filldraw[fill=black] (2.1,-.3) circle (.1) node[below=1pt] {\small $t$};
\end{tikzpicture} 
\end{minipage}
& \hspace{-8pt} $\begin{array}{l} \\[-15pt] (r_1 r_2 \cdots r_{t-2} r_{t-1} r_t r_{t-2} \cdots r_2 r_1) \cdot \\ \qu \cdot (r_2 \cdots r_{t-2} r_{t-1} r_t r_{t-2} \cdots r_2) \cdots \\ \qq \cdots (r_{t-2} r_{t-1} r_t r_{t-2}) r_{t-1} r_t \end{array}$
& \hspace{-8pt} $\begin{array}{l} \\[-12pt]\displaystyle \sum_{i=1}^{t-2} \frac{i(2t-1-i)}{2} h_i + \\ \qu + \frac{t(t-1)}{4}(h_{t-1}+ h_t) \end{array}$ \\[24pt]
\hline
\end{tabular}
\end{table}

From now on assume that $(X,\tau) \in \GSat(A)$.
The following definitions will enable us to organize the pertinent diagrams into natural families.
For any $Y \subseteq I$ such that $\tau(Y)=Y$ we denote
\eq{
o_Y = \big| \{ Z \subseteq Y \, | \, Z \text{ is a } \tau-\text{orbit and } |Z|=2 \} \big|, \label{oY}
}
{\it i.e.}~the number of nontrivial $\tau$-orbits in $Y$.
Also, let $X_Y$ be the smallest component of $X$ containing $Y$; if $X$ does not contain $Y$ then set $X_Y=\emptyset$. 
It will also be convenient to write 
\eq{
p_Y = \big| \{ Z \subseteq X_Y \, | \, Z \text{ is a } \tau-\text{orbit} \} \big|.\label{pY}
}

We observe the following with respect to the lateral sets of $\tau$. 
If $\tau$ is of identity type then $|{\mc L}(\tau)|=2$, except in type ${\rm A}^{(1)}_{n>1}$ where $|{\mc L}(\tau)|=0$.
If $\tau=\pi$ then $|{\mc L}(\tau)|=1$, except again in type ${\rm A}^{(1)}_{n>1}$ where $|{\mc L}(\tau)|=0$.
Finally if $\tau \notin \Sigma_A$ then $|{\mc L}(\tau)|=2$ in type ${\rm A}^{(1)}_{n \ge 1}$ and $|{\mc L}(\tau)|=0$ in type~${\rm D}^{(1)}_4$. 

Let $Y \in {\mc L}(\tau)$.
Then $o_Y=1$ implies that $Y$ is a hinge or a subset of type ${\rm D}_2$.
From \eqref{Satdiag1a} one may draw two pertinent conclusions.
Namely, if $Y$ is a hinge then $X_Y$ is of type A$_t$ (or empty) and arranged symmetrically around the hinge.
Also, if $X_Y$ is of type ${\rm D}_{t \ge 2}$, then $p_Y$ is even.

Recall the subsets $X(j) \subseteq X$ defined by \eqref{def:X(j)} and $I' \subseteq I \backslash X$ defined by \eqref{Iprime}.

\begin{defn}
Let $(X,\tau) \in \GSat(A)$ with $\tau$ of identity type.  
We call $(X,\tau)$ \emph{plain} if $I'$ is connected and no connected component of $X$ is a proper subset of a $\tau$-lateral set.
If $X(j)$ has precisely two connected components for all $j \in I'$, we call $(X,\tau)$ \emph{alternating}. \hfill \defnend
\end{defn}

Note that if $(X,\tau)$ is alternating then $I_{\rm ns} = I_{\rm nsf}$.
There exist generalized Satake diagrams of identity type which are plain as well as alternating, namely those for which $I\backslash X$ consists of one $\tau$-orbit, {\it i.e.}~$|I^*|=1$ (in this case $A$ is of type ${\rm B}_n^{(1)}$, ${\rm C}_n^{(1)}$ or ${\rm D}_n^{(1)}$ and $(X,\tau) \in \Sat(A)$). 
For $A$ is of classical Lie type, using \eqref{Satdiag2b} one checks straightforwardly that all generalized Satake diagrams with $\tau$ of identity type are plain or alternating. 

\begin{rmk}
In the more restrictive case that $(X,\tau) \in \Sat(A)$, using the expressions for $\rho^\vee_X$ in Table \ref{wXrhoveeX}, from \eqref{Satdiag2a} one derives that the number of components of $X$ neighbouring $j$ which are of type ${\rm A}_1$ and not in $\mc{L}(\tau)$ has the same parity as the number of components of $X$ neighbouring $j$ of type ${\rm C}_{t \ge 1}$. \hfill \rmkend
\end{rmk}

If $(X,\tau) \in \GSat(A)$ is not of identity type, then any connected component of $X$ must be of type ${\rm A}_t$. 
More precisely, if $A$ is of type ${\rm A}^{(1)}_{N-1}$ with $N$ even and $\tau = \pi$, in the proof of Lemma \ref{lem:SigmaAslN} we have seen that $X$ must be empty.
If $\tau$ has at least one hinge, then $(X,\tau)$ with $A_X$ of finite type is a generalized Satake diagram precisely if each connected component of $X$ is of type ${\rm A}_t$ and symmetrically arranged around a hinge. 
Finally, if $A$ is of type ${\rm D}^{(1)}_4$ and $\tau \in [(14)]$ then there are three possibilities for $X$; for example, if $\tau = (14)$, then $X = \emptyset$, $X = \{ 0,3 \}$ or $X = \{ 1,2,4 \}$. 

Given an enumeration $Y_1,Y_2,\ldots,Y_k$ of lateral sets, we will abbreviate 
\[ o_i:=o_{Y_i}, \qq p_i:=p_{Y_i}, \qq X_i:=X_{Y_i}. \] 
From the above definitions and arguments we obtain that, if $A$ is an (untwisted or twisted) affine Cartan matrix of classical Lie type and $(X,\tau) \in \GSat(A)$, there are at most 2 $\tau$-lateral sets and $X$ decomposes as follows:
\eq{\label{X:decomp}
X = X_1 \cup X_{\rm alt} \cup X_2.
}
Here $X_{\rm alt}$ is the largest component of $X$ of the form ${\rm A}_1^{\times t}$ which does not contain any $\tau$-lateral sets (note that $X_{\rm alt} = \emptyset$ unless $(X,\tau)$ is alternating).

\smallskip

Classifying the generalized Satake diagrams now amounts to a straightforward enumeration of possibilities for $o_i$ and $p_i$.
For Satake diagrams, we recover the classification of \cite{BBBR}, apart from the subtlety of the case where $A$ is of type ${\rm D}^{(1)}_4$ and $\tau \in [(14)]$. 
Additionally, we obtain weak Satake diagrams when $A$ is of type ${\rm B}^{(1)}_n$, ${\rm C}^{(1)}_n$ or ${\rm D}^{(1)}_n$ and $\tau$ is of identity type; these are alternating if $A$ is of type ${\rm B}^{(1)}_n$ or ${\rm D}^{(1)}_n$ and plain if $A$ is of type ${\rm C}^{(1)}_n$.
The result is given in Table \ref{minitable}.

\begin{table}[t]
\setlength{\tabcolsep}{1.2pt}
\caption{Overview of the generalized Satake diagrams considered in this paper.}
\label{minitable}
\begin{center}

\\
\hline
\end{tabular}
\end{center}
\end{table}


\subsection{Notation and choice of representatives for generalized Satake diagrams of affine type} \label{sec:Satakenotation}
 
We will denote families of generalized Satake diagrams by the format T.$t$s where T $\in \{ {\rm A}, {\rm B}, {\rm C}, {\rm D} \}$ indicates the type of underlying Dynkin diagram; $t \in \{ 1,2,3,4 \}$ and s $\in \{ {\rm a},{\rm b},{\rm c} \}$ indicate a type and subtype, respectively.
The specifications $t=1$ and $t=2$ correspond to plain and alternating generalized Satake diagrams respectively. 
Furthermore, $t=3$ corresponds to $\tau \not \in \Sigma_A$ (in this case the Lie type is A or ${\rm D}_4$) and $t=4$ corresponds to $\tau = \pi$ (then the Lie type is A, C or D).

The subtype s is defined to be a, b or c according as $\sum_i o_i$ equals 0, 1 or 2, where the summation is over all $\tau$-lateral subsets of $I$.
Note that for generalized Satake diagrams of types A.1, A.2, A.4, C.1 and C.2 the distinction in subtypes is trivial and will be suppressed in the notation.
For example, the family A.3c consists of Satake diagrams whose underlying Dynkin diagram is of type ${\rm A}_{n \ge 1}$, where $\tau \in [\psi] \cup [\psi']$ and $\tau$ acts nontrivially on each hinge (at once we see that $n$ must be odd and in fact $\tau \in [\psi']$).

It is often convenient in the notation to unite families across different types, {\it e.g.}~BCD.1 is the union of the B.1, C.1 and D.1 families, A.124 is the union of the A.1, A.2 and A.4 families and D.12ac is the union of the D.1a, D.1c, D.2a and D.2c families.

\smallskip

By swapping the labels 1 and 2 if necessary, we may assume for generalized Satake diagrams of types A.3ac, C.12 and D.12ac ({\it i.e.}~those with two $\tau$-lateral sets such that $\tau$ commutes with $\pi$) that
\eq{ 
o_1 \le o_2, \qq \text{ and } \qq o_1=o_2 \; \implies \; p_1 \le p_2.
}

The representatives of the $\Sigma_A$-orbits of generalized Satake diagrams that we will use in Section \ref{sec:Results} are chosen in such a way that the involution $\tau$ is as specified in Section \ref{sec:diagraminvolutions}. 
The above conventions imply that, assuming the standard labelling, if there are two $\tau$-lateral sets, the affine node 0 is automatically in $Y_1$. 
Moreover, if a $\Sigma_A$-orbit intersects $\GSat_0(A)$ then we can and will choose the representative to be an element of $\GSat_0(A)$.

\smallskip

We will employ a second notation to describe the generalized Satake diagrams of affine type more precisely.
It has the format
\[ 
({\rm T}^{(1)}_n)^{\tau}_{\text{specification of }X}, 
\]
where T $\in \{ {\rm A}, {\rm B}, {\rm C}, {\rm D} \}$, so that T$^{(1)}_n$ indicates the type of the underlying affine Dynkin diagram. 
In view of \eqref{X:decomp} we specify $X$ is as follows:

\begin{itemize} [itemsep=.5ex]
\item For generalized Satake diagrams in the families A.3ac, C.1 and D.1ac (for which we have fixed $p_1 \le p_2$) the specification is written in the format $p_1,p_2$.
For generalized Satake diagrams in the families A.3b, B.1 and D.1b (for which we have imposed an ordering by specifying the types of the $\tau$-lateral subsets of $I$) the specification format is $p_1;p_2$.
\item For Satake diagrams in the family CD.4, $X$ has only one component and the specification format is simply $p_1$.
\item If $(X,\tau)$ is alternating, this will be indicated by ``alt''; for $(X,\tau)$ in the families BCD.2 there may be up to two components $X_1$, $X_2$ involving the nodes 0 and n, respectively, and we combine this with the aforementioned convention, yielding a specification format $p_1,\text{alt},p_2$ (for types C.2 and D.2ac) and $p_1,\text{alt},p_2$ (for types B.2 and D.2b).
\item In the families A.1 and A.4 we will specify $X$ by writing $\emptyset$.
\item For Satake diagrams of type ${\rm D}^{(1)}_4$ with $\tau \in [(14)]$ we simply specify $X$ as a subset of $I = \{0,1,2,3,4\}$.
\end{itemize} 

Appendix \ref{App:Satakediagrams} lists all subfamilies of $\Sigma_A$-equivalence classes of generalized Satake diagrams for untwisted affine Dynkin diagrams $(I,A)$ of classical Lie type, along with relevant properties. 
In Appendix \ref{App:LowRank} we list bijections between low-rank ($n \le 4$) $\Sigma_A$-equivalence classes.


\subsection{Restrictable diagrams}

Owing to the standard choice of labelling, deleting the node 0 from an untwisted affine Dynkin diagram, one obtains the corresponding Dynkin diagram of finite type. It leads to the following definition for (generalized) Satake diagrams.

\begin{defn} \label{def:resSatdiag}
Let $A$ be an untwisted affine Cartan matrix.
Define
\[ \GSat_0(A) = \left\{ (X,\tau) \in \GSat(A) \, | \, \tau(0)=0\not \in X \right\}.\]
The elements of $\GSat_0(A)$ are called \emph{restrictable generalized Satake diagrams}.
We also define
\[ \Sat_0(A) = \GSat_0(A) \cap \Sat(A) = \left\{ (X,\tau) \in \Sat(A) \, | \, \tau(0)=0\not \in X \right\}, \]
the element of which are called \emph{restrictable Satake diagrams}. \hfill \defnend
\end{defn}

There is a natural one-to-one correspondence between $\GSat(A^{\rm fin})$ and $\GSat_0(A)$ for the Cartan matrices under consideration.

\begin{prop}
Let $A$ be an untwisted affine Cartan matrix of classical Lie type.
For $\tau \in \Aut(A^{\rm fin})$, define $\wh \tau \in \Aut(A)$ by $\wh \tau(i) = \tau(i)$ if $i \in I \backslash \{ 0 \}$ and $\wh \tau(0)=0$.
The assignment $(X,\tau) \mapsto (X,\wh \tau)$ defines an invertible map ${\rm Aff}_A : \GSat(A^{\rm fin}) \to \GSat_0(A)$ such that ${\rm Aff}_A(\Sat(A^{\rm fin})) = \Sat_0(A)$. 
\end{prop}

\begin{proof}
Let $(X,\tau) \in \GSat(A^{\rm fin})$.
First we show that $(X,\wh \tau) \in \GSat(A)$. 
Note that evidently $X$ is of finite type and $\wh \tau$ is an involution. 
Condition \eqref{Satdiag1a} holds since $\wh \tau|_X = \tau|_X$. 
For condition \eqref{Satdiag2b}, consider $X(0)$, the union of connected components of $X$ neighbouring node $0 \in I \backslash X$,
We have the following case-by-case analysis for the four infinite families of $A$:
\begin{description}
\item[${\rm A}^{(1)}_{n>1}$] Here $X(0)$ cannot be a single node. In type A.1, $X=\emptyset$ and in type A.2, $X(0)=\{1,n\}$. In type A.3, either $X(0)=\emptyset$ or $X(0)\in \{1,n\}$.
\item[${\rm A}^{(1)}_1$ and ${\rm C}^{(1)}_n$] If $X(0)$ consists of a single node, it must be 1. 
We have $a_{10}=-2$.
\item[${\rm B}^{(1)}_n$ and ${\rm D}^{(1)}_n$] Again, $X(0)$ cannot be a single node. Namely, only $i=2$ is a possibility which would leave $1$ in $I \backslash X$, contradicting \eqref{Satdiag2b} for $(X,\tau) \in \GSat(A^{\rm fin})$.
\end{description}

Hence, if $X(0)=\{i\}$ for some $i \in \{1,\ldots,n\}$, $a_{0i}a_{i0} \ne 1$; therefore $(X,\wh \tau) \in \GSat(A)$.
It follows immediately from Definition \ref{def:resSatdiag} that in fact $(X,\tau) \in \GSat_0(A)$.
The invertibility of ${\rm Aff}_A$ follows at once upon noting that $(X,\tau|_{I \backslash \{ 0 \}}) \in \GSat(A^{\rm fin})$ if $(X,\tau) \in \GSat_0(A)$ and that restriction to $I \backslash \{ 0 \}$ is a left- and right-inverse of the map $\tau \mapsto \wh \tau$.

Finally, to show that ${\rm Aff}_A(\Sat(A^{\rm fin})) = \Sat_0(A)$, assume that $(X,\tau) \in \Sat(A^{\rm fin})$. 
Then $\al_j(\rho^\vee_X) \in \Z$ for all $j \in I\backslash X$ such that $\tau(j)=j$ by assumption, so that it remains to verify that $\al_0(\rho^\vee_X) \in \Z$.
Recall the tuple $(a_i)_{i \in I}$ satisfying \eqref{eq:gj}; as $A$ is of untwisted affine type we have $a_0=1$.
Then the basic imaginary root $\del:=\sum_{i\in I} a_i \al_i \in \Phi$ satisfies $\del(h_i)=0$ for all $i \in I$ and $\phi:=\del-\al_0=\sum_{i \in I \backslash \{ 0 \}} a_i \al_i$ is the highest root of $\Phi_{I \backslash \{ 0 \}}$.
Immediately we have $\del(\rho^\vee_X)=0$.
Also, $\phi(\rho^\vee_X) \in \Z$.
Indeed, having chosen a subset $(I \backslash \{ 0 \})^* \subseteq (I \backslash \{ 0 \}) \backslash X$ intersecting the $\tau$-orbits of $(I \backslash \{ 0 \})\backslash X$ in singletons, we obtain
\[
\phi(\rho^\vee_X) = \sum_{i \in X} a_i \al_i(\rho^\vee_X) + \sum_{i \in (I \backslash \{ 0 \})^* \atop i=\tau(i)} a_i \al_i(\rho^\vee_X) + \sum_{i \in (I \backslash \{ 0 \})^* \atop i \ne \tau(i)} (a_i \al_i + a_{\tau(i)} \al_{\tau(i)})(\rho^\vee_X).
\]
Each summation is in $\Z$: the first because $\rho^\vee_X$ is the sum of fundamental coweights associated to $\Phi_X$, the second because $(X,\tau) \in \Sat(A^{\rm fin})$ and the third because $a_i = a_{\tau(i)}$ (since $\phi = \tau(\phi)$) and $\al_{\tau(i)}(\rho^\vee_X) = \al_i(\rho^\vee_X)$ (since $X=\tau(X)$).
\end{proof}

Since ${\rm Aff}_A$ is bijective, generalized Satake diagrams of finite type appear in a natural way as subdiagrams of restrictable generalized Satake diagrams of affine type.
Hence, for Satake diagrams, we recover their classification (listed in \cite[\S 4 and \S 5]{Ara} and \cite[Sec.~7]{Le3}), see Table \ref{tab:finitesatakediagrams} in Appendix \ref{App:Satakediagrams}.
Note that the families B.12b, D.12c, A.3c and ACD.4 do not contain any restrictable Satake diagrams. 
generalized Satake diagrams from the remaining families are restrictable up to rotation by an element of $\Sigma_A$ precisely if $X \backslash X_{\rm alt}$ is connected or empty.


\subsection{The special orbits revisited}

Let $(X,\tau) \in \Sat(A)$. 
In \cite[Sec.~7, Variation 1]{Le2} it was noted that $|I_{\rm diff}| \le 1$ if $A$ is of finite type. 
Moreover, in \cite[Rmk.~9.3]{Ko1} it is argued that $|I_{\rm diff}| \le 2$ if $A$ is of affine type. 
It is natural to generalize this in two directions: also involve the set $I_{\rm nsf}$ and allow $(X,\tau) \in \GSat(A)$. It turns out that the same upper bounds hold.

\begin{lemma} \label{lem:specialorbits}
Let $A$ be a finite or untwisted affine Cartan matrix of classical Lie type and fix $(X,\tau) \in \GSat(A)$.
Then $|I_{\rm diff} \cup I_{\rm nsf}|$ is bounded above by 1 or 2, respectively. 
\end{lemma}
\begin{proof}
This follows from a straightforward case-by-case analysis using the tables in Appendix \ref{App:Satakediagrams}.
\end{proof}

It would be nice to prove this upper bound without resorting to the case-by-case analysis. Even if we would be able to prove this for generalized Satake diagrams of finite type, there would not be an obvious way to derive a proof of this for affine type, since Lemma \ref{lem:specialorbits} applies also to non-restrictable affine generalized Satake diagrams. 

\begin{rmk} \label{rem:specialorbits} \mbox{}

\begin{itemize}[itemsep=0.5em]
\item Lemma \ref{lem:specialorbits} applies equally for generalized Satake diagrams associated to twisted affine Cartan matrices, as well as finite and affine Cartan matrices of exceptional Lie type. 
Furthermore, the upper bound can be lowered by 1 in some cases. 
\item
From the tables in Appendix \ref{App:Satakediagrams} it also follows that instances of generalized Satake diagrams with $|I_{\rm diff} \cup I_{\rm nsf}|=2$ exist for all untwisted affine Cartan matrices of classical Lie type.
For example, let $A$ be of type ${\rm A}^{(1)}_{n \ge 2}$, {\it i.e.}~$I=\{0,1,\ldots,n\}$, $a_{ii}=2$ and, for $i \ne j$, $a_{ij}=a_{ji}=-1$ precisely if $|i-j|=1$ or $|i-j|=n$.
Note that $\Aut(A)$ is a dihedral group.
Then $|I_{\rm diff} \cup I_{\rm nsf}|=2$ for all $(X,\tau) \in \GSat(A)$ such that $|I^*|>1$ and $\tau \notin \Sigma_A$.
Moreover, each of the three possibilities for $(|I_{\rm diff}|,|I_{\rm nsf}|)$ is attained by some $(X,\tau)$ provided $|I|$ is even.  \hfill \rmkend
\end{itemize}
\end{rmk}

Combining Corollary \ref{cor:nonremovable} and Lemma \ref{lem:specialorbits} we obtain the following.

\begin{crl}
Fix $(X,\tau) \in \GSat(A)$ and let $\bm c \in \mc{C}$ and $\bm s \in \mc{S}$.
Then $B_{\bm c,\bm s}$ is $\Ad(\wt{H}_q)$-equivalent to a QP algebra with at most two free parameters.
\end{crl}


\section{Classification of K-matrices} \label{sec:Results}

In this section we present the main results of this paper, a comprehensive list trigonometric reflection matrices for QP algebras $B_{\bm c,\bm s}=B_{\bm c,\bm s}(X,\tau)$ in the vector representation of $U_q(\mfg)$ for all $(X,\tau) \in \GSat(A)$, when $A$ is of a classical untwisted affine type. 


\subsection{Methodology and the main results} \label{sec:mainresult}

The reflection matrices are obtained by solving \eqref{intw-untw} or \eqref{intw-tw} for all generators of the algebras $B_{\bm c,\bm s}$ for each representative $(X,\tau)$ of the $\Sigma_A$-orbits of the generalized Satake diagrams as defined in Section \ref{sec:Satdiagclassification}. 
We remind the reader that whenever generalized Satake diagrams are in the same $\Sigma_A$-orbit, the associated K-matrices are related according to Proposition \ref{prop:rotateintw}.
Now fix $(X,\tau) \in \GSat(A)$ to be such a representative.
We choose
\eq{ 
I^* = \{ 
j \in I \backslash X 
\, | \, 
j \ge \tau(j) \text{ if } \{j,\tau(j)\}=\{0,1\} \, \text{ and } \, j \le \tau(j) \text{ otherwise} 
\}
} 
The condition $\tau(0)=1$ only applies to some representative generalized Satake diagrams of types BD.2c and BD.1b; in those cases we will reiterate our choice for $I^*$. Note that this fixes the sets $I_{\rm diff}$ and $I_{\rm nse}$.
As explained in Section \ref{sec:Kmatdressing}, we express $\bm c$ and $\bm s$ in terms of  dressing parameters $\bm \omega$, non-removable free parameters and the scaling parameter $\eta$ in such a way that solving \eqref{intw-untw} or \eqref{intw-tw} for all $b_j$ and the generators of the subalgebras $U_q(\mfg_X)$ and $U_q(\mfh)^{\theta_q}$ for the given value of $\eta$ yields a dressed K-matrix $K(u;\bm \omega)$ given by \eqref{eq:Kdressing}.
Recall from Section \ref{sec:dressing} that the bare K-matrix $K(u)$ contains $|I_{\rm diff} \cup I_{\rm nsf}|$ degrees of freedom and the dressing matrix $G(\bm \omega)$ contains $|I^*| - 1$ degrees of freedom; the remaining degree of freedom is accounted for by the scaling parameter $\eta$.

Recall that if a nontrivial solution $(K(u),\eta)$ of \eqref{intw-untw} or \eqref{intw-tw} exists, $(K(-u),-\eta)$ is also a nontrivial solution (it may occur that $K(u)$ and $K(-u)$ coincide); as a matter of convenience we always present only one of them. 
Our case-by-case computations suggest that if a nontrivial solution of \eqref{intw-untw} or \eqref{intw-tw} exists, then $K(u)$ is unique up to this choice of sign and up to scalar multiples.
The following statement refers to the existence of solutions.

\begin{thrm} \label{thrm:solexistence}
Fix $(X,\tau) \in \GSat(A)$ and consider the corresponding coideal subalgebra $B_{\bm c,\bm s}$ with $\bm c \in (\K^\times)^{I \backslash X}$ and $\bm s \in \K^{I \backslash X}$.
The untwisted boundary intertwining equation \eqref{intw-untw} has nontrivial solutions if and only if $\tau \psi \in \Sigma_A$.
More precisely we have the following:
\begin{description} [itemsep=.25ex]
\item[Type ${\rm A}^{(1)}_{n > 1}$]
Either \eqref{intw-untw} or \eqref{intw-tw} has a nontrivial solution, and \eqref{intw-untw} has nontrivial solutions precisely if $\tau \notin \Sigma_A \cong \Cyc_N$. 

\item[Types ${\rm A}^{(1)}_1$, ${\rm B}^{(1)}_{n > 1}$ and ${\rm C}^{(1)}_{n > 1}$] We have $\psi = \id$ and $\Sigma_A = \Aut(A)$ and indeed we have found nontrivial solutions to \eqref{intw-untw} for all $(X,\tau) \in \GSat(A)$ which are in 1-to-1 correspondence with nontrivial solutions to \eqref{intw-tw} according to Proposition \ref{prop:intw:tw-untw}.

\item[Type ${\rm D}^{(1)}_{n>3}$] We have $\psi = \id$ and, provided $(X,\tau) \in \GSat(A)$ such that $\tau \in \Sigma_A$, again \eqref{intw-untw} turns out to have nontrivial solutions which are in 1-to-1 correspondence with nontrivial solutions to \eqref{intw-tw}. 
Only if $(X,\tau)$ is of type D.3, {\it i.e.}~$n=4$ and $\tau \in [ (14)]$, the condition $\tau \in \Sigma_A$ fails and, for all $\bm c \in (\K^\times)^{I \backslash X}$ and $\bm s \in \K^{I \backslash X}$, the only solution to \eqref{intw-untw} or \eqref{intw-tw} is $K(u)=0$. 
\end{description}
In accordance with the above, assume that $(X,\tau)$ is not of type D.3 and select the appropriate intertwining equation \eqref{intw-untw} or \eqref{intw-tw} if $A$ is of type ${\rm A}^{(1)}_{n>1}$.
The existence of solutions to \eqref{intw-untw} or \eqref{intw-tw} depends on $\bm c$ and $\bm s$ as follows.
The boundary intertwining equation has a nontrivial solution only if $\bm c \in \mc{C}$.
Assuming that $\bm c \in \mc{C}$, the only values of $\bm s$ for which there are nontrivial solutions to \eqref{intw-untw} or \eqref{intw-tw} are as follows:
\begin{itemize} [itemsep=.25ex]
\item If $|I^*|=1$ (in particular $I_{\rm ns} = \emptyset$), the boundary intertwining equation has a nontrivial solution for any $\bm s \in \K^{I \backslash X}$; these solutions are diagonal. 
Note that the additional term in $b_j$ is mapped by $\RT_u$ to a diagonal matrix and, as we will see, the solution $K(u)$ is itself diagonal and does not depend on $\bm s$. 

\item If $|I^*| >1$ and $I_{\rm ns}=I_{\rm nsf}$, the boundary intertwining equation has a nontrivial solution precisely if $\bm s \in \mc{S}$ ({\it i.e.}~$s_j = 0$ if $j \notin I_{\rm nsf}$). 

\item There exists a nontrivial solution if $\bm s \in (\K^\times)^{I \backslash X}$ satisfies \emph{$q$-Onsager type constraints}, see \cite[Prop.~2.1]{BsBe1}, which include particular $\K$-linear relations between $s_j^2$ and $c_j$ for certain values of $j \in I_{\rm ns}$.
These solutions are genuinely different from the aforementioned ones only if $I_{\rm nsf} \ne I_{\rm ns}$, which requires that $(X,\tau)$ is plain.
\end{itemize}
\end{thrm}

\begin{rmk} \mbox{}
\begin{enumerate}  [itemsep=.25ex]
\item The above theorem is based on computational work for low values of $n$; thus it is only a conjecture for general $n$. 
It is possible on a case-by-case basis to give lengthy but straightforward proofs which we omit here.
\item A detailed study of the reflection matrices associated to generalized $q$-Onsager algebras will be presented elsewhere, although we briefly review the K-matrix associated with the family A.1 in Appendix~\ref{App:qOns}. \hfill \rmkend
\end{enumerate}
\end{rmk}

Our main result is the following list of trigonometric reflection matrices for QP algebras of classical Lie type.
In order to present this succinctly, recall the sign $\vartheta$ as defined in \eqref{QQv} and, if $A$ is of type ${\rm A}^{(1)}_n$, denote, for $a \in \Z_{\ge 0}$,
\eq{ \label{modN}
a \bmodN N = \begin{cases} 
N & \text{if } a \text{ is an integer multiple of } N, \\ 
a \bmod N & \text{otherwise}. 
\end{cases} 
}
In particular, for $(X,\tau)$ of type A.3, we denote $t = \tau(0) \bmodN N$.

\begin{thrm} \label{T:all-K}
The bare twisted K-matrices are:
\begin{flalign*}
& {\rm A.1:} & K(u) &= \Id && \text{if } N>2, && \\
& {\rm A.2:} & K(u) &= \textstyle\sum_{1\le i\le {N}/{2}} \big( q^{1/2} E_{2i-1,2i}  - q^{-1/2} E_{2i,2i-1}\big) && \text{if } N>2 \text{ even}, && \\
& {\rm A.4:} & K(u) &= \textstyle\sum_{1 \le i \le {N}/{2}} \big( u E_{i+{N}/{2},i} + E_{i,i+{N}/{2}} \big) &&  \text{if } N>2 \text{ even}. &&
\end{flalign*}
In the other cases, in terms of
\[ \ell =
\begin{cases} 
p_1 & \text{for A.3}, \\ p_1 + o_1 & \text{for BCD.12}, \\ (n-o)/2-p & \text{for CD.4},
\end{cases} 
\qq 
r= 
\begin{cases} 
(N-o_1-o_2)/2-p_2 & \text{for A.3}, \\ n-p_2-o_2 & \text{for BCD.12}
\end{cases} 
\]
the bare untwisted K-matrices are given by the formula
\eq{ \label{K(u):X}
K(u) = \Id + \frac{u-u^{-1}}{k_1(u)} \left( M_1(u) + \frac{M_2(u)}{k_2(u)}  \right)
}
where
\eqa{ 
\label{k_i(u)}
k_1(u) &= \la\mu - u, \qq 
k_2(u) = 
\begin{cases} 
\la^{-1}- (\mu u)^{-1} & \text{for C.1 and BD.2}, \\
\la^{-1}+ (\mu u)^{-1} & \text{for other types}
\end{cases} 
\intertext{and}
M_1(u) &= 
\begin{cases} 
\displaystyle \sum_{1 \le i \le \ell} \la \, \mu \, u \, E_{ii} + \sum_{t+1-\ell \le i \le N}  E_{ii} \hspace{4.78cm}\; & \text{for A.3}, \\[1.5em]
\displaystyle \sum_{\bar \ell \le i \le n} \left(\la \, \mu \, u \, E_{-i,-i} + E_{ii} \right) & \text{for BCD.12}, \\[1.5em]
\displaystyle \sum_{1 \le  i \le n}  E_{ii} & \text{for CD.4}, 
\end{cases} 
\intertext{and}
M_2(u) &= 
\begin{cases} 
\displaystyle \sum_{ \ell < i \le  r} \left( \la E_{ii} + \la^{-1} E_{t+1-i,t+1-i} + E_{i,t+1-i} + E_{t+1-i,i} \right) & \text{for A.3}, \\[1.5em]
\displaystyle \sum_{\bar r \le i < \bar \ell} \left( \la E_{-i,-i} + \la^{-1} E_{ii} + E_{-i,i} + E_{i,-i} \right) & \text{for BD.1}, \\[1.5em]
\displaystyle -\,\del_{\ell,1}\, u^{-1}\,(\mu-\mu^{-1})\, E_{nn} - \del_{N,2n}\,\del_{r,n-1}\, (\la-\la^{-1}) \, E_{11} \\[.25em]
\displaystyle \qq + \sum_{\overline{r} \le i < \overline{\ell} } \big( {-}\la E_{-i,-i} + \la^{-1} E_{ii} + \eps_i\, (E_{-i-\eps_i,i} - E_{i,-i-\eps_i}) \big) & \text{for BD.2}, \\[1.5em]
\displaystyle \sum_{\bar r\le i\le \bar\ell} \big( {-}\la E_{-i,-i} + \la^{-1} E_{ii} + E_{-i,i} - E_{i,-i} \big) & \text{for C.1}, \\[1.5em]
\displaystyle \sum_{\bar r \le i < \bar \ell} \big(\la E_{-i,-i} + \la^{-1} E_{ii} + \eps_i\,(E_{-i-\eps_i,i} + E_{i,-i-\eps_i})\big) & \text{for C.2}, \\[1.5em]
\displaystyle \sum_{\bar \ell \le i \le n} \!\big( \la E_{-i,-i} + \la^{-1} E_{-\bar \imath,-\bar \imath} - E_{-i,-\bar\imath} - E_{-\bar\imath,-i}  \\[-.5em]
\displaystyle \hspace{3.95cm} -\, u^{-1}(\mu\,E_{ii} + \mu^{-1} E_{\bar \imath \bar \imath} - E_{i\bar\imath} - E_{\bar\imath i}) \big)  & \text{for CD.4},
\end{cases} 
\intertext{
with $\eps_i=(-1)^{\bar \imath-\ell}$ for BCD.2. The parameters $\la,\mu$ have the following values:
}
& \qu
\begin{cases} 
\la \in \K^\times \text{ is free}\qu\; & \text{for A.3, CD.4 and for C.1 if $r=n$,} \\ 
\la = q^{N/2-r} & \text{for BD.1 and C.2,} \\
\la = q^{N/2-r-\vartheta} & \text{for B.2, for C.1 if $r<n$ and for D.2 if $r < n-1$,}
\end{cases} 
\\
& \qu
\begin{cases} 
\mu \in \K^\times \text{ is free}\qu\; & \text{for A.3, for C.1 if $\ell=0$ and for BD.2 if $\ell \le 1$,} \\ 
\mu = q^{-\ell} & \text{for BD.1 and C.2,} \\
\mu = q^{-\ell + \vartheta} & \text{for C.1 if $\ell>0$, for BD.2 if $\ell>1$,} \\
\mu = q^{-n+2\ell}\la & \text{for CD.4}.
\end{cases} 
}


\noindent For D.4 if $\ell=1$ the above formulas should be replaced by
\[
K(u) = \Id + \frac{u-u^{-1}}{ k_1(u)} \left( M_1 + \frac{M^+_{2}(u)}{ k^+_{2}(u)} - \frac{M^-_2}{ k^-_2(u)} \right),
\]
where $M_1 = \sum_{1 \le i \le n} E_{ii}$ (as for $\ell \ne 1$) and
\eqn{ 
k^+_2(u) &= k_2(u) = \la^{-1}+ (\mu u)^{-1}, \qu & M^+_2(u) &= u^{-1} (\mu^{-1} E_{11} + \mu \, E_{nn}+E_{1n}+E_{n1}) , \\
k^-_2(u) &= (\al \, \la)^{-1} + \al (\mu \, u)^{-1}, &
M^-_2 &= (\al \, \la)^{-1} E_{-1,-1} + \al \, \la \, E_{-n,-n}+E_{-1,-n}+E_{-n,-1} ,
}
%
%
with $\mu = q^{-n+2} \al \, \la$ and free parameters $\al, \la \in \K^\times$. \medskip


\noindent 
For certain QP algebras of types BCD.1 the above formulas should be replaced as follows (the expressions for $k_1(u)$ and $k_2(u)$ and the values for $\la$ and $\mu$ remain the same). 
\eqn{
\intertext{For BD.1 with $(\ell,r)=(0,2)$:}
K(u) &= \Id + \frac{(u-u^{-1}) \la^2 u^2}{k_1(\frac{\nu_0 u}{\nu_1})\, k_1(\frac{\nu_1 u}{\nu_0})\, k_1(\nu_0 \nu_1 u)\, k_1(\frac{u}{\nu_0 \nu_1})} \Big( k_1(u)k_2(u) M_2 + \al(\tfrac{u}{\la})  \, \wt M_3^- + \al(\tfrac{\la}{u}) \,  \wt M_3^+ \Big) ,
\intertext{where $\al(u) = (\nu_1+\nu_1^{-1}) \, u - (\nu_0+\nu_0^{-1})$ and }
M_2 &= \sum_{n-1 \le i\le n} \big( \la E_{-i,-i} + \la^{-1} E_{i,i} + E_{-i,i} + E_{i,-i} \big) ,\\
\wt M_3^\pm &=  E_{\mp(n-1),\pm n} - E_{\pm n,\mp (n-1)} + \la^{\mp 1} (E_{\pm (n-1),\pm  n}-E_{\pm n,\pm (n-1)} \mp  (\nu_1+\nu_1^{-1}) E_{\pm  n,\pm n}). \\[.5em]
\intertext{For BCD.1 with $1\le \ell\le n-2$, $r=\ell+2$:}
K(u) &= \Id + \frac{u-u^{-1}}{k_1(u)} \bigg(\wt M_1(u) + \frac{k_2(u) \wt M_2(u) + \nu_+ \wt M_3(u) }{ k_2(-\nu^{-2} u) k_2(-\nu^2 u)} \bigg) , \\[-.5em]
\intertext{where}
\wt M_1(u) &= \sum_{\bar \ell -1 \le i \le n} (\la \mu u E_{-i,-i} + E_{ii}), \\[.25em] 
\wt M_2(u) &= \vartheta \la E_{2-\bar\ell,2-\bar\ell} + \la^{-1} E_{\bar\ell-2,\bar\ell-2} + \vartheta  E_{2-\bar\ell,\bar\ell-2} + E_{\bar\ell-2,2-\bar\ell} \\
& \qu - \mu u  E_{1-\bar \ell,1-\bar \ell} - \vartheta (\mu u)^{-1}  E_{\bar \ell-1,\bar \ell-1}  +  E_{1-\bar \ell,\bar \ell-1} + \vartheta E_{\bar \ell-1,1-\bar \ell}, \\[.25em]
\wt M_3(u) &= \vartheta \big( E_{1-\bar\ell,2-\bar\ell}-E_{2-\bar\ell,1-\bar\ell} \big) + \la^{-1}(E_{1-\bar\ell,\bar\ell-2}-\vartheta  E_{\bar\ell-2,1-\bar\ell}) \\
& \qu + (\mu u)^{-1} \big(\vartheta  E_{2-\bar\ell,\bar\ell-1}-E_{\bar\ell-1,2-\bar\ell} + \vartheta  \la^{-1}(E_{\bar\ell-2,\bar\ell-1}-E_{\bar\ell-1,\bar\ell-2}) \big). \\[.25em]
\intertext{For B.1 with $(\ell,r)=(n-1,n):$}
K(u) &= \Id + \frac{u-u^{-1}}{k_1(u)} \Bigg(\wt M_1(u) + \frac{k_1(u^{-1})^2 \wt M_2(u) + q^{1/4} [2]^{1/2}_{q^{1/2}} (\nu-\nu^{-1})  \wt M_3(u)}{k_1(\nu^2 u^{-1})\,k_1(\nu^{-2} u^{-1})\,k_2(u) }  \Bigg) , \\[-.5em]
\intertext{where}
\wt M_1(u) &= \sum_{1 \le i \le n} (q^{1/2} \mu u E_{-i,-i} + E_{ii}), \\ 
\wt M_2(u) &= - \mu u E_{-1,-1} - (\mu u)^{-1} E_{11} + E_{-1,1} + E_{1,-1}, \\
\wt M_3(u) &=  k_1 \big( u^{-1}) (u^{-1} (E_{01} + E_{10}) - \mu (E_{-1,0}+E_{0,-1}) \big) -  q^{1/4}  [2]_{q^{1/2}}^{1/2} \mu (\nu-\nu^{-1}) u^{-1}  E_{00}.
}
\end{thrm}

\smallskip

\begin{rmk} \mbox{}
\begin{enumerate} [itemsep=0.25ex]
\item The new parameters $\ell$, $r$ and $t$ are such that the tuples $(n,\ell)$, $(n,\ell,r)$ and $(n,\ell,r,t)$ uniquely parametrize generalized Satake diagrams of the families CD.4, BCD.12 and A.3, respectively (across all subfamilies a,b,c, when applicable).

\item In each case $M_1(u)$ is always diagonal and vanishes if $(X,\tau)$ is restrictable.
Matrix $M_2(u)$ contains the off-diagonal entries of $K(u)$, with $M_2(u) =0$ if $|I^*| = 1$.

\item Theorem \ref{T:all-K} summarizes the main results obtained in Sections \ref{sec:K:tw}--\ref{sec:K:0}. We used computational software to solve intertwining equations for $n$ in the range $1 \le n \le 15$, depending on the type of $B_{\bm c, \bm s}$, and extrapolated the results for general $n$. Thus it is a theorem for small $n$, but only a conjecture for large $n$. 
Proving that the above expressions solve the boundary interwining equation for arbitrary $n$ in most cases involve lengthy but direct computations; in some cases these calculations are somewhat simpler, see Remark \ref{A4:solution} for the A.4 case. 
The \emph{uniqueness} (up to a scalar multiple and up to a sign) of these solutions for general $n$ would follow if we could establish indecomposability of $\RT_u|_{B_{\bm c,\bm s}}$.

\item We have solved the intertwining equations \eqrefs{intw-untw}{intw-tw} for representative generalized Satake diagrams $(X,\tau)$ in each $\Sigma_A$-orbit, corresponding to the special values of the parameters $\ell$, $r$ and $t$ as mentioned in Theorem \ref{T:all-K}.
Expressions for the K-matrices corresponding to other generalized Satake diagrams are determined by Proposition \ref{prop:rotateintw}. 
However, in many cases the $(\ell,r,t)$-dependent formulas in Theorem \ref{T:all-K} are valid for more generalized Satake diagrams; see Remarks \ref{R:A3:rotate}, \ref{R:C1:rotate}, \ref{R:D2:rotate}, \ref{R:C2:rotate} and \ref{R:D1:rotate} for more precise statements in types A.3 and CD.12. \hfill \rmkend

\end{enumerate}
\end{rmk}

In each case we compare our results with known solutions of the reflection equation. 
In \cite{MLS} Malara and Lima-Santos give a detailed description of the ``densest'' solutions ({\it i.e.}~with most entries nonzero) of the untwisted reflection equation for a larger class of R-matrices, which are obtained by directly solving the reflection equation.
In particular, for type A they find K-matrices corresponding to Satake diagrams of type A.3 with $X$ empty.
However they admit they cannot completely describe the relations between these matrices, which is accounted for by the action of $\Sigma_A \cong \Cyc_N$ in our setup.
Their solutions for types B, C and D are related to generalized q-Onsager algebras, see \cite{BsBe1}. 
Some untwisted K-matrices found here are special parameter limits of these more general K-matrices and we highlight this in the results where appropriate. 
A more comprehensive comparison will be made in a future publication dealing with K-matrices associated to generalized $q$-Onsager algebras.


\subsubsection{Nonzero entries of K-matrices and quasistandard QP algebras} \label{sec:K-zeros}

The K-matrices classified in this paper are relatively sparse, {\it i.e.}~the number of nonzero entries with respect to the standard basis of $\K^N$ is small compared to $N^2$. 
In fact, it is clear from Theorem \ref{T:all-K} that, except in some cases of type BCD.1, the number of nonzero entries in each rows and column of each K-matrix is at most 2. 
This can be characterized in terms of the generalized Satake diagram underlying the associated QP algebra as follows.

\begin{defn} \label{D:quasi}
Let $A$ be an untwisted affine Cartan matrix of classical Lie type and let $(X,\tau) \in \GSat(A)$.
We call $(X,\tau)$ \emph{quasistandard} if for each $Y \in \mc{L}(\tau)$, $I_{\rm nsf} \cap Y = \{ j \}$ with $\al_j$ long (note that for simply-laced diagrams, all roots are considered to be long). 
Let $\bm c \in \mc{C}$ and $\bm s \in \mc{S}$.
The QP algebra $B_{\bm c,\bm s}(X,\tau)$ is called quasistandard if $s_j \ne 0$ implies that $\al_j$ is long and $\{j \} = Y \cap I_{\rm nsf}$ for some $Y \in \mc{L}(\tau)$.
\end{defn}

This terminology is natural: evidently all standard QP algebras (those with $s_j=0$ for all $j \in I \backslash X$, see \cite[Defn.~5.6]{Ko1}) are quasistandard. 
Most nonstandard QP algebras considered in this paper are also quasistandard; in fact, the only non-quasistandard QP algebras for untwisted affine Cartan matrices of classical Lie type occur in type BCD.1. 

Note that a generalized Satake diagram $(X,\tau)$ is quasistandard precisely if each element of $I_{\rm nsf}$ is one of the nodes labeled $j$ in the following subdiagrams:
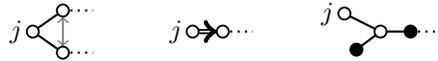
\[
\begin{tikzpicture}[baseline=-0.35em,scale=0.8,line width=0.7pt]
\draw[thick] (1.5,.35) -- (1,0) -- (1.5,-.35);
\draw[thick,dotted] (1.5,.35) -- (2,.35);
\draw[thick,dotted] (1.5,-.35) -- (2,-.35);
\filldraw[fill=white] (1,0) circle (.1) node[left]{\small$j$};
\draw[<->,gray] (1.5,.25) -- (1.5,-.25);
\filldraw[fill=white] (1.5,.35) circle (.1) ;
\filldraw[fill=white] (1.5,-.35) circle (.1) ;
\end{tikzpicture} 
\qq\;
\begin{tikzpicture}[baseline=-0.35em,scale=0.8,line width=0.7pt]
\draw[double,->] (0,0) -- (.4,0);
\draw[thick,dotted] (.5,0) -- (1,0);
\filldraw[fill=white] (0,0) circle (.1) node[left]{\small$j$};
\filldraw[fill=white] (.5,0) circle (.1);
\end{tikzpicture}
\qq
\begin{tikzpicture}[baseline=-0.35em,scale=0.8,line width=0.7pt]
\draw[thick] (-.6,.3) -- (0,0) -- (-.4,-.3);
\draw[thick] (0,0) -- (.5,0);
\draw[thick,dotted] (.5,0) -- (1,0);
\filldraw[fill=white] (-.6,.3) circle (.1) node[left]{\small$j$};
\filldraw[fill=black] (-.4,-.3) circle (.1);
\filldraw[fill=white] (0,0) circle (.1);
\filldraw[fill=black] (.5,0) circle (.1);
\end{tikzpicture}
\]
or one of the nodes labelled $j$ or $j'$ in one of the following low-rank diagrams:
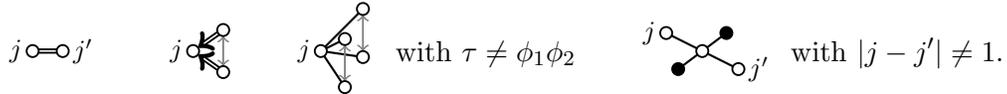
\[
\begin{tikzpicture}[baseline=-0.35em,scale=0.8,line width=0.7pt]
\draw[double] (0,0) -- (.5,0);
\filldraw[fill=white] (0,0) circle (.1) node[left]{\small $j$};
\filldraw[fill=white] (.5,0) circle (.1) node[right]{\small $j'$};
\end{tikzpicture}
\qq
\begin{tikzpicture}[baseline=-0.35em,scale=0.8,line width=0.7pt]
\draw[double,<-] (.09,.063) -- (.5,.35);
\draw[double,<-] (.09,-.063) -- (.5,-.35);
\filldraw[fill=white] (0,0) circle (.1) node[left]{\small $j$};
\draw[<->,gray] (.5,.25) -- (.5,-.25);
\filldraw[fill=white] (.5,.35) circle (.1);
\filldraw[fill=white] (.5,-.35) circle (.1);
\end{tikzpicture} 
\qq
\begin{tikzpicture}[baseline=-0.35em,scale=0.8,line width=0.7pt]
\draw[thick] (.7,.7) -- (0,0) -- (.4,.2);
\draw[thick] (.7,-.1) -- (0,0) -- (.4,-.6);
\draw[<->,gray] (.4,.1) -- (.4,-.5);
\draw[<->,gray] (.7,.6) -- (.7,0);
\filldraw[fill=white] (0,0) circle (.1) node[left]{\small $j$};
\filldraw[fill=white] (.4,.2) circle (.1);
\filldraw[fill=white] (.4,-.6) circle (.1);
\filldraw[fill=white] (.7,.7) circle (.1);
\filldraw[fill=white] (.7,-.1) circle (.1);
\end{tikzpicture}
\;\;\text{ with } \tau \ne \phi_1\phi_2
\qq
\begin{tikzpicture}[baseline=-0.35em,scale=0.8,line width=0.7pt]
\draw[thick] (-.6,.3) -- (0,0) -- (-.4,-.3);
\draw[thick] (.4,.3) -- (0,0) -- (.6,-.3);
\filldraw[fill=white] (-.6,.3) circle (.1) node[left]{\small $j$};
\filldraw[fill=black] (-.4,-.3) circle (.1);
\filldraw[fill=white] (0,0) circle (.1);
\filldraw[fill=white] (.6,-.3) circle (.1) node[right]{\small $j'$};
\filldraw[fill=black] (.4,.3) circle (.1);
\end{tikzpicture}
\,\text{ with } |j-j'| \ne 1.
\]
Hence, quasistandard generalized Satake diagrams with nonempty $I_{\rm nsf}$ occur in types A.3, C.1, BD.2 and CD.4, as can be easily seen from Table \ref{minitable}. (See also Table \ref{tab:satakediagrams} in Appendix \ref{App:Satakediagrams}.)

Non-quasistandard generalized Satake diagrams are of type BCD.1 and contain one of the following subdiagrams (where $j$ and $j'$ indicate the elements of $I_{\rm nsf}$): 
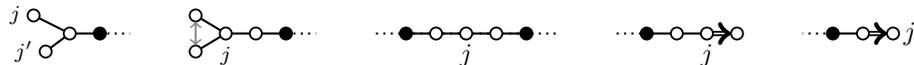
\[
\begin{tikzpicture}[baseline=-0.35em,line width=0.7pt,scale=.8]
\draw[thick] (-.6,.3) -- (0,0) -- (-.4,-.3);
\draw[thick] (0,0) -- (.5,0);
\draw[thick,dotted] (0.5,0) -- (1,0);
\filldraw[fill=white] (-.6,.3) circle (.1) node[left=1pt]{\scriptsize $j$};
\filldraw[fill=white] (-.4,-.3) circle (.1) node[left=1pt]{\scriptsize $j'$};
\filldraw[fill=white] (0,0) circle (.1);
\filldraw[fill=black] (.5,0) circle (.1);
\end{tikzpicture} 
\qq
\begin{tikzpicture}[baseline=-0.35em,line width=0.7pt,scale=0.8]
\draw[thick] (-.5,.3) -- (0,0) -- (-.5,-.3);
\draw[thick] (0,0) -- (1,0);
\draw[thick,dotted] (1,0) -- (1.5,0);
\filldraw[fill=white] (-.5,.3) circle (.1);
\filldraw[fill=white] (-.5,-.3) circle (.1);
\filldraw[fill=white] (0,0) circle (.1) node[below=1pt]{\scriptsize $j$};
\filldraw[fill=white] (.5,0) circle (.1);
\filldraw[fill=black] (1,0) circle (.1);
\draw[<->,gray] (-.5,.2) -- (-.5,-.2);
\end{tikzpicture}
\qq
\begin{tikzpicture}[baseline=-0.35em,scale=0.8,line width=0.7pt]
\draw[thick,dotted] (-1.5,0) -- (1.5,0);
\draw[thick] (-1,0) -- (1,0);
\filldraw[fill=black] (-1,0) circle (.1) ;
\filldraw[fill=white] (-.5,0) circle (.1) ;
\filldraw[fill=white] (0,0) circle (.1) node[below]{\small$j$};
\filldraw[fill=white] (.5,0) circle (.1);
\filldraw[fill=black] (1,0) circle (.1) ;
\end{tikzpicture}
\qq
\begin{tikzpicture}[baseline=-0.35em,scale=0.8,line width=0.7pt]
\draw[thick,dotted] (-1.5,0) -- (-1,0);
\draw[thick] (-1,0) -- (0,0);
\draw[double,->] (0,0) -- (.4,0);
\filldraw[fill=black] (-1,0) circle (.1) ;
\filldraw[fill=white] (-.5,0) circle (.1) ;
\filldraw[fill=white] (0,0) circle (.1) node[below]{\small$j$};
\filldraw[fill=white] (.5,0) circle (.1);
\end{tikzpicture}
\qq
\begin{tikzpicture}[baseline=-0.35em,scale=0.8,line width=0.7pt]
\draw[thick] (-.5,0) -- (0,0);
\draw[thick,dotted] (-.5,0) -- (-1,0);
\draw[double,->] (0,0) -- (.4,0);
\filldraw[fill=white] (.5,0) circle (.1) node[right]{\small$j$};
\filldraw[fill=white] (0,0) circle (.1);
\filldraw[fill=black] (-.5,0) circle (.1);
\end{tikzpicture} 
\]
or special low-rank versions of these. The corresponding QP algebras are studied in Section \ref{sec:nqs}.

\smallskip 

The K-matrices associated to quasistandard QP algebras are of a particular simple form. 
Recall that a \emph{generalized permutation matrix} is an element of $\End(\K^N)$ whose nonzero entries pattern in the same way as a permutation matrix, in other words a $\K$-linear combination of elementary matrices $E_{i,\gamma(i)}$ with $i \in \langle N \rangle$ for some permutation $\ga$ of $\langle N \rangle$. 
For example, for generic values of $u$, any matrix $Z^\si(u)$ with $\si \in \Sigma_A$ is a generalized permutation matrix.
The subgroup of generalized permutation matrices in $\GL(\K^N)$ normalizes the maximal abelian subgroup of $\GL(\K^N)$ consisting of diagonal matrices, whose Weyl group is $S_N$.
If the permutation $\gamma$ is an involution, we call a generalized permutation matrix a \emph{generalized involution matrix} (\gim).
Note that any \gim~squares to a diagonal matrix.

\begin{defn}
We call an element of $\End(\K^N)$ a \emph{generalized cross matrix} (g.c.m.) if it is a $\K$-linear combination of a diagonal matrix and a \gim \hfill \defnend
\end{defn}

{\def\arraystretch{1.3}
\begin{table}
\caption{Involution $\ga$.} 
\label{tab:gamma} 
\[
\begin{array}{l|l}
 \text{Type}  & \ga(i)  \\
\hline
\text{A.1} & i \\
\text{A.2} & \begin{cases} (i - (-1)^i) \bmodN N & \text{if }  X = \{ 1,3,\ldots,n\} \\ (i + (-1)^i) \bmodN N  & \text{if } X = \{ 0 ,2, \ldots, n-1\} \end{cases} \\
\text{A.4} & (i+\tfrac{N}{2}) \bmodN N \\
\hline
\text{A.3} & (\tau(0)+1-i) \bmodN N \\
\text{BCD.1} & -i  \\
\text{BCD.2} & \begin{cases} i  \qu \text{if } |i|=1 \text{ and } n+o_1+p_1 \text{ is odd, or } |i|=n \text{ and } o_1+p_1 \text{ is odd} \\ -i+\sgn(i) (-1)^{n+i+o_1+p_1} \qu \text{otherwise} \end{cases} \\
\text{CD.4} & \sgn(i) \widebar{|i|}
\end{array}
\]
\end{table} 
}

Note that, up to a reordering of the basis, a generalized cross matrix is a direct sum of 2$\times$2- and 1$\times$1-blocks.
It follows from Theorem \ref{T:all-K} that, for generic values of $u$, a K-matrix associated to a quasistandard QP algebra is a generalized cross matrix. 
The explicit expressions of the involution $\ga$ are listed in Table~\ref{tab:gamma}; see also Table \ref{table:summary} in Appendix \ref{App:summary}.
In certain cases, the K-matrix reduces to a generalized involution matrix, which are in addition independent of $u$ if and only if $(X,\tau) \in \GSat_0(A)$. 
This occurs precisely for the following QP algebras:
\begin{itemize}  [itemsep=.25ex]
\item QP algebras with one of the following underlying Satake diagrams ($n \ge 2$):
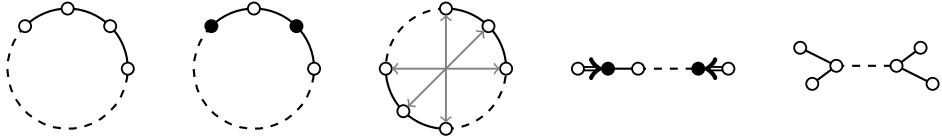
\[ 
\begin{tikzpicture} [baseline=-0.25em,scale=0.8,line width=0.7pt]
\draw[thick,domain=0:135] plot ({cos(\x)},{sin(\x)});
\draw[thick,dashed,domain=135:360] plot ({cos(\x)},{sin(\x)});
\filldraw[fill=white] ({-sqrt(2)/2},{sqrt(2)/2}) circle (.1);
\filldraw[fill=white] (0,1) circle (.1);
\filldraw[fill=white] ({sqrt(2)/2},{sqrt(2)/2}) circle (.1);
\filldraw[fill=white] (1,0) circle (.1);
\end{tikzpicture} 
\qq
\begin{tikzpicture} [baseline=-0.25em,scale=0.8,line width=0.7pt]
\draw[thick,domain=0:135] plot ({cos(\x)},{sin(\x)});
\draw[thick,dashed,domain=135:360] plot ({cos(\x)},{sin(\x)});
\filldraw[fill=black] ({-sqrt(2)/2},{sqrt(2)/2}) circle (.1);
\filldraw[fill=white] (0,1) circle (.1);
\filldraw[fill=black] ({sqrt(2)/2},{sqrt(2)/2}) circle (.1);
\filldraw[fill=white] (1,0) circle (.1);
\end{tikzpicture} 
\qq
\begin{tikzpicture} [baseline=-0.25em,scale=0.8,line width=0.7pt]
\draw[thick,domain=0:90] plot ({cos(\x)},{sin(\x)});
\draw[thick,dashed,domain=90:180] plot ({cos(\x)},{sin(\x)});
\draw[thick,domain=180:270] plot ({cos(\x)},{sin(\x)});
\draw[thick,dashed,domain=270:360] plot ({cos(\x)},{sin(\x)});
\filldraw[fill=white] (0,1) circle (.1);
\filldraw[fill=white] ({sqrt(2)/2},{sqrt(2)/2}) circle (.1);
\filldraw[fill=white] (1,0) circle (.1);
\filldraw[fill=white] (0,-1) circle (.1);
\filldraw[fill=white] ({-sqrt(2)/2},{-sqrt(2)/2}) circle (.1);
\filldraw[fill=white] (-1,0) circle (.1);
\draw[<->,gray] (0,.9) -- (0,-.9);
\draw[<->,gray] ({.9*sqrt(2)/2},{.9*sqrt(2)/2}) -- ({-.9*sqrt(2)/2},{-.9*sqrt(2)/2});
\draw[<->,gray] (.9,0) -- (-.9,0);
\end{tikzpicture}  
\qq
\begin{tikzpicture}[baseline=-0.25em,line width=0.7pt,scale=0.8]
\draw[double,<-] (-.1,0) -- (-.5,0);
\draw[thick] (0,0) -- (.5,0);
\draw[thick,dashed] (.5,0) -- (1.5,0);
\draw[double,<-] (1.6,0) --  (2,0);
\filldraw[fill=white] (-.5,0) circle (.1);
\filldraw[fill=black] (0,0) circle (.1);
\filldraw[fill=white] (.5,0) circle (.1);
\filldraw[fill=black] (1.5,0) circle (.1);
\filldraw[fill=white] (2,0) circle (.1);
\end{tikzpicture} 
\qq
\begin{tikzpicture}[baseline=-0.35em,line width=0.7pt,scale=.8]
\draw[thick] (-.6,.3) -- (0,0) -- (-.4,-.3);
\draw[thick,dashed] (0,0) -- (1,0);
\draw[thick] (1.4,.3) -- (1,0) -- (1.6,-.3);
\filldraw[fill=white] (-.6,.3) circle (.1);
\filldraw[fill=white] (-.4,-.3) circle (.1);
\filldraw[fill=white] (0,0) circle (.1);
\filldraw[fill=white] (1,0) circle (.1);
\filldraw[fill=white] (1.4,.3) circle (.1);
\filldraw[fill=white] (1.6,-.3) circle (.1);
\end{tikzpicture} 
\]
\item QP algebras of type A.3, C.1 or D.2 with $s_j = 0$ and $c_j = c_{\tau(j)}$ for all $j \in I^*$.
\item QP algebras of type CD.4 with $X = \emptyset$, $n$ even and $s_{n/2}=0$.
\item QP algebras such that $|I^*|=1$, in which case $K(u)$ is diagonal.
\end{itemize}

Another property of all untwisted K-matrices obtained in this paper is that $K_{ij}(u) \ne 0$ if and only if $K_{ji}(u) \ne 0$ for all $i,j \in \langle N \rangle$, i.e. that $K(u)$ is symmetric up to conjugation by a diagonal matrix.
In fact, it follows from the definitions of the representation $\RT_u$ and of the coideal subalgebras under consideration in this paper that solutions of \eqref{intw-untw} have this property. 
For the same reason, solutions $K(u)$ of \eqref{intw-tw} satisfy the property that $K_{ij}(u) \ne 0$ if and only if $K_{j'i'}(u) \ne 0$ for all $i,j \in \langle N \rangle$, where $i'=N+1-i$ if $\mfg = \wh\mfsl_N$ and $i'=-i$ otherwise, i.e. $C^{-1} K(\wt q^{-1} u)$ is symmetric up to conjugation by a diagonal matrix.
In other words, for the untwisted K-matrices classified in this paper the off-diagonal nonzero entries come in pairs; the number of such pairs is $|I^*| -1$ for quasistandard QP algebras of type A.3 or BCD.1, $2(|I^*|-1)$ for quasistandard algebras of type BCD.24 and $3(|I^*| -1)$ for non-quasistandard QP algebras of type BCD.1.
In particular, if $|I^*| = 1$ then an untwisted K-matrix is diagonal. 
Indeed, by construction any such K-matrix commutes with all diagonal matrices in $\End(\K^N)$, which follows from the fact that the QP algebra contains $n=|I|-1$ linearly independent elements of $U_q(\mfh)$ and from the explicit formulas for $\RT_u(k_i)$ for $0 \le i \le n$ (in particular one needs that $\RT_u(k_c) = \Id$). 

Attempting to solve the reflection equation \eqref{RE} or \eqref{tRE} starting from various ans\"{a}tze with prescribed matrix entries set to zero in $\End(\K^N)(u)$ for some low values of $N$, leads us to believe that the listing in Theorem \ref{T:all-K} is complete in the following sense.

\begin{conj}
Let $R(u)$ be given by \eqref{Ru:defn}.
All solutions $K(u)$ of the reflection equation \eqref{RE} or \eqref{tRE} which satisfy the conditions
\begin{enumerate}
\item $K(u)$ or $C^{-1} K(\wt q^{-1} u)$, respectively, is symmetric up to conjugation by a diagonal matrix;
\item no row or column in $K(u)$ has more than two nonzero entries;
\end{enumerate}
are generalized cross matrices.
Up to application of Lemma \ref{lem:tw-untw}, rotation, dressing, scalar multiple and choice of sign of the spectral parameter, they are listed in Theorem \ref{T:all-K} for quasistandard $(X,\tau)\in\GSat(A)$. 
\end{conj}


\subsubsection{Additional properties of K-matrices} \label{sec:Kmatrixsupproperties}

For each of the K-matrices obtained, where applicable we discuss the following properties in the case-by-case results:
\begin{description}[itemsep=.25ex]

\item[Unitarity]
In accordance with Lemma \ref{L:K-unit}, all untwisted K-matrices classified in Theorem \ref{T:all-K} have been normalized so as to satisfy the unitarity property \eqref{Ku:unit}. We will not highlight this property in the case-by-case results. 
For the twisted bare K-matrices of types A.124 we have chosen $\bm c \in (\K^\times)^{I \backslash X}$ to be of the form $(a,\ldots,a)$ for some $a \in \K^\times$. 
Lemma \ref{lem:twistedunitarityslN} implies that \eqref{K-unit-tw} holds for some there exists $n_{\rm tw}(u) \in \K$.
It is not convenient to renormalize these K-matrices to make $n_{\rm tw}(u)=1$; instead we will state \eqref{K-unit-tw} in each case.

\item[Regularity]
Unless otherwise stated, all untwisted K-matrices obtained are doubly regular; again, we will not emphasize double regularity in the case-by-case results. 
All doubly and singly regular K-matrices can be and have been normalized such that $K(\zeta) = +\Id$ for $\zeta=-1$, $\zeta=1$ or both (this can be accomplished by multiplying the K-matrix by $-1$ and/or by $u$ as required, neither of which affects the unitarity property). 

\item[Eigendecomposition]
All untwisted K-matrices obtained are diagonalizable for generic values of the non-removable free parameters.
More precisely, we have
\eq{ \label{eigendecomposition} 
K(u) = V D(u) V^{-1}, 
}
where $V$ is independent of $u$ and $D(u)$ is diagonal.
If $B_{\bm c,\bm s}$ is quasistandard, $V$ is a generalized cross matrix with the same involution $\gamma$ as $K(u)$. 
In each case we will give the matrices $V$ and $D(u)$. It will be convenient to use the notation
\eq{
h_i(u) = \frac{k_i(u^{-1})}{k_i(u)}. \label{p(u)}
}
We will order the eigenvalues and eigenvectors in such a way that $V$ is a generalized cross matrix with the same involution $\gamma$ as $K(u)$ (see Table \ref{tab:gamma}).
For twisted K-matrices one has to replace $K(u)$ by $C^{-1} K(\wt q^{-1} u)$ in \eqref{eigendecomposition}.

The number of distinct eigenvalues is low compared to $N$, namely at most 5 (4 in the quasistandard case). 
From \eqref{eigendecomposition} it follows that, in analogy with \eqref{Ru:min-id:2}, this number equals the degree of the minimal polynomial of $K(u)$; furthermore $[K(u),K(v)]=0$ for all $u,v$. 
These properties suggest that the K-matrices obtained here in terms of representations of coideal subalgebra of affine quantum groups can also be related to representations of suitable cyclotomic Hecke or Birman-Murakami-Wenzl algebras, {\it i.e.}~certain quotients of affine versions of these algebras. 
As an example, in Section \ref{sec:Hecke} we work this out in detail for K-matrices of type A.3. 
In turn, this may help in establishing a version of quantum Schur-Weyl duality for coideal subalgebras.

\item[Affinization]

If $(X,\tau) \in \GSat_0(A)$, the limit $K_0 := \lim_{u \to 0} K(u)$ exists and is invertible.
This can be seen by taking the limits $u\to 0$, $v \to 0$ in \eqref{RE} and \eqref{tREalt} appropriately and noting that $K_0$ is a solution to the constant untwisted or twisted reflection equations, viz.
\eqa{
\label{CRE} (R_q)_{21} (K_0)_1  R_q (K_0)_2  &= (K_0)_2 (R_q)_{21} (K_0)_1 R_q, \\
\label{CtRE} (R_q)^\t_{21} (K_0)_1 (R_q^{-1})^{\t_2} (K_0)_2 &= (K_0)_2 (R_q^{-1})_{21}^{\t_1} (K_0)_1 R_q,
}
respectively. 
The matrix $K_0$ satisfies the constant boundary intertwining equation $K_0\,\RT(b)=\RT(b)\,K_0$ in the untwisted case, or $K_0\,\RT(b) = \RT^\t(S(b))\,K_0$ in the twisted case, for all $b\in B_{\bm c,\bm s}\cap U_q(\mfg)$. 
Conversely, if $(X,\tau)$ is not restrictable, we find that $K_0$ is not invertible, so it cannot be interpreted as an intertwiner of any coideal subalgebra of $U_q(\mfg)$ in the usual sense. 

In \cite{NoSu} a classification of solutions of the constant twisted RE was given; the case of the quantum complex Grassmannian (corresponding to our type A.3, for which the reflection equation cannot be brought to the twisted form) was discussed in \cite{NDS}. 
In order to compare our results with those obtained by Noumi, Dijkhuizen and Sugitani, for each restrictable case we will indicate to which of their constant K-matrices our $K(u)$ corresponds. It is useful to remark that, owing to \eqref{R:PT-symm}, \eqref{CtRE} is equivalent to 
\[
\wt R_q (K_0)_1 \wt R_q^{\t_1} (K_0)_2 = (K_0)_2 \wt R_q^{\t_1} (K_0)_1 \wt R_q 
\]
with $\wt R_q = (R_q^{-1})_{21}$, which is the twisted RE used in \cite{NoSu}. 

If $X = I \backslash \{0\}$ we always have $K(u) = K_0 = \Id$. 
More generally, K-matrices associated to restrictable generalized Satake diagrams turn out to have at most three distinct eigenvalues, {\it i.e.}~there exists a linear combination of $K(u)$, $K(u)^{-1}$ and $\Id$ which is zero. 
In the quasistandard restrictable cases there are in fact at most two distinct eigenvalues (there is only one family of non-quasistandard restrictable Satake diagrams: BD.1 with $I \backslash X = \{0,1,2\}$). 

Analogously to \eqref{Ru:affinization} we find \emph{(boundary) affinization identities}: if $(X,\tau) \in \GSat_0(A)$ and $K(u)$ is the corresponding K-matrix then there exist $f_\pm(u) , f_0(u), g_\pm(u) \in \K$ such that  
\eqn{
\qq K(u) &= f_+(u) K_0 + f_0(u) \Id + f_- (u) K_0^{-1}, \\
\qq K(u) &= g_+(u) K_0 + g_- (u) K_0^{-1} \qq \qq \text{if } (X,\tau) \text{ is quasistandard}.
}

\item[Rotations]

When $(X^\si,\tau^\si)=(X,\tau)$ for some $\si \in \Sigma_A$ then according to Proposition \ref{prop:rotateintw}, the rotated K-matrix $K^\si(u)$ given by \eqref{eq:Ksigma} associated to the same generalized Satake diagram, can be presented as a modification of $K(u)$ in terms of dressing, scalar multiplication and a different choice of $\bm c \in \mc{C},\bm s \in \mc{S}$. When applicable, we list such identities for rotated K-matrices $K^\si(u)$. 

\item[Bar-symmetry] 

K-matrices of types A.3 and BCD.12 exhibit a boundary analogon of the ``bar-symmetry'' \eqref{Ru:bar}, namely $K(u)^{-1}$ can be expressed in terms of $J K(u) J$ with an involution applied to its parameters, where
\[ 
J =\sum_{i \in \langle N \rangle} E_{i',i}
\]
where $i'=N+1-i$ if $\mfg = \wh\mfsl_N$ and $i'=-i$ otherwise (see Section \ref{sec:K-zeros}).
The appearance of $J$ is quite natural since \eqref{Ru:bar} is equivalent to $\hat R(u)^{-1} = J_1 J_2 \hat R(u)|_{q \to q^{-1}} J_1 J_2$ owing to $R_{21}(u) = J_1 J_2 R(u) J_1 J_2$.

\item[Half-period]

In many cases there is a simple relation between $K(-u)$ and $K(u)$.
We call this \emph{half-period symmetry} since in the trigonometric realization of these K-matrices, obtained by substituting $u$ by $\exp(\sqrt{-1}\, t)$ for a new spectral parameter $t$, the transformation $u \to -u$ corresponds to $t \to t+\pi$.
If $0 \notin X$ the only generators of $B_{c,s}$ whose image under $\RT_{\eta u}$ depends on $\eta$ are $b_0$ and $b_{\tau(0)}$; if $0 \in X$ there is a unique $\tau$-orbit $Y \subseteq I \backslash X$ neighbouring the connected component of $X$ containing 0 and only the $b_j$ with $j \in Y$ depend on $\eta$.
As follows from Section \ref{sec:dressing}, if $\tau(0)=0 \notin X$, $c_0$ is proportional to $\eta^{-2}$ and $s_0$ is proportional to $\eta^{-1}$. If $\tau(0) \ne 0 \notin X$ then $c_0$ and $c_{\tau(0)}$ are both proportional to $\eta^{-1}$. Finally, if $0 \in X$ then for all $j \in Y$ we have $c_j$ proportional to $\eta^{-1}$.
Hence, if $\tau(0)=0 \notin X$ with $s_0=0$ then $K(u)=K(-u)$.
Furthermore, if $0 \in I_{\rm nsf}$ with $s_0 \ne 0$ or if either $\{0,\tau(0) \}$ or $Y$ intersects $I_{\rm diff}$ then $K(-u)$ can be obtained from $K(u)$ by applying to a nonremovable free parameter a certain transformation, which we will specify. 

\item[Reductions and diagonal cases] We highlight if any of the following behaviour occurs for special values of the ``classification'' parameters $\ell$, $r$ and $t$ or for special values of the free parameters:
\begin{itemize}[itemsep=0.25ex]

\item The effective degree $d_{\rm eff}(K)$ defined in Section \ref{sec:K-intw} is not equal to the number of distinct eigenvalues of $K(u)$ ({\it i.e.}~the degree of the minimal polynomial of $K(u)$).

\item The K-matrix is not doubly regular ({\it i.e.}~$K(1)$ or $K(-1)$ is not equal to $\pm \Id$); this is always accompanied by a reduction in $d_{\rm eff}(K)$. In particular, w.r.t.~the formulas listed in Theorem~\ref{T:all-K}, the functions $k_i(u)$ may develop a zero at $u= \pm 1$, yielding a reduction in the regularity property (in these cases it follows that the representation $\RT_u$ restricted to $B_{\bm c,\bm s}$ is reducible when evaluated at $u=\pm \eta$).

\item The K-matrix is a generalized involution matrix.

\item The K-matrix is independent of $q$. This is possible in some cases by expressing $\bm c$ and $\bm s$ suitably in terms of $q$, the dressing parameters $\omega_j$, the scaling parameter $\eta$ and additional free parameters.
\end{itemize}
Whenever we restrict to a certain values of $\ell$, $r$ and $t$, or impose a relation between $\ell$ and $r$ we will denote the corresponding K-matrix by notations such as $K_{\ell=\ell_0}(u)$, $K_{\ell=\ell_0\!,\,r=r_0}(u)$, $K_{\ell=r}(u)$, etc.
\end{description}

We present results in the case-by-case manner by grouping them according to the type of reflection equation and the number of special $\tau$-orbits, {\it i.e.}~by $|I_{\rm diff}\cup I_{\rm nsf}|$:

\begin{itemize} [itemsep=.25ex]

\item Families A.124: $|I_{\rm diff}\cup I_{\rm nsf}|=0$. The cases with $n>1$ are studied in Section \ref{sec:K:tw}. Some of the $n=1$ cases have nonempty $I_{\rm diff}\cup I_{\rm nsf}$ and are studied in Section \ref{sec:K:low-rank} (they are related to the $n=1$ case of the A.3 family).

\item Family A.3: $|I_{\rm diff}\cup I_{\rm nsf}|=2$ for most of the cases, with $|I_{\rm diff}\cup I_{\rm nsf}|<2$ for certain specializations; all cases are studied together in Section \ref{sec:A3}.

\item Families C.1 and BD.2: $I_{\rm diff}\cup I_{\rm nsf}=\emptyset$ for most of the cases, with $|I_{\rm diff}\cup I_{\rm nsf}|=1,2$ for certain specializations. 
These families mostly consist of weak Satake diagrams; the Satake diagrams are those such that $|I_{\rm diff}\cup I_{\rm nsf}|=2$ or $|I^*|=1$. 
All cases are studied together in Section \ref{sec:K:C1BD2}. 

\item Families CD.4: most of these Satake diagrams have $|I_{\rm diff}\cup I_{\rm nsf}|=1$ and are treated in Section~\ref{sec:CD4}. There is a special case with $|I_{\rm diff}|=2$, namely Satake diagrams of type D.4 with $|I \backslash X|=4$, which are dealt with in Section~\ref{sec:D4:spec}.

\item Families C.2 and BD.1: Most of the cases have $I_{\rm diff}\cup I_{\rm nsf}=\emptyset$ and are studied in Sections \ref{sec:C2}-\ref{sec:BD1}. 
The cases with $|I^*|=1$ are also of type C.1 or BD.2, respectively; in particular in the family BD.1 there are special cases with $|I_{\rm diff}|=1$ which are dealt with in Section \ref{sec:K:C1BD2}.

\item Finally, we study the non-quasistandard cases, which occur in the families BCD.1.
These are dealt with in Section~\ref{sec:nqs}. Of these, the BD.1 diagrams are Satake diagrams and the C.1 diagrams are weak Satake diagrams.
\end{itemize}

Low-rank QP algebras with $n=1$ and the exceptional D.3 cases (where $n=4$ and $\tau \in [(14)]$) are studied in Section~\ref{sec:K:low-rank}.
In Appendix \ref{App:summary} we give a summary of the main properties of the \mbox{K-matrices} classified in this paper.


\subsection{Twisted K-matrices of types A.1, A.2 and A.4} \label{sec:K:tw}

In this section we compute reflection matrices for the QP algebras associated with Satake diagrams $(X,\tau)\in\Sat(A)$ when $A$ is of type~A$^{(1)}_n$ and either $\tau=\id$ or $\tau=\pi$ that correspond to types (families) A.1, A.2 and A.4 listed in Table \ref{tab:A124:sat}. 
These QP algebras provide us with solutions of the twisted boundary intertwining equation \eqref{intw-tw} that are also solutions of the twisted reflection equation \eqref{tRE}. 
In this section we also give additional details of intermediate computations.
In later sections we will be more concise.

\begin{table}[h]
\caption{Satake diagrams and special $\tau$-orbits for families A.1, A.2 and A.4. } \label{tab:A124:sat}
\[
\begin{array}{ccccccc}
\text{Type} & \text{Name} & \text{Diagram} & \text{Restrictions} & I_{\rm diff} & I_{\rm ns} & I_{\rm nsf} \\ 
\hline 
\hline 
\rm A.1 
& 
\bigl( {\rm A}^{(1)}_1 \bigr)^\text{id}_0 
&
\begin{tikzpicture}[baseline=-0.4em,line width=0.7pt,scale=.8]
\draw[double] (0,0) -- (0.5,0);
\draw[] (0,.4);
\filldraw[fill=white] (0,0) circle (.1) node[left=1pt]{\scriptsize $0$};
\filldraw[fill=white] (.5,0) circle (.1) node[right=1pt]{\scriptsize $1$};
\end{tikzpicture}  
& n = 1 & \emptyset & \{0,1\} & \{0,1\} 
\\
\rm A.1 
& 
\bigl( {\rm A}^{(1)}_n \bigr)^\text{id}_0 
&
\hspace{5pt}
\begin{tikzpicture} [baseline=-0.25em,scale=0.8,line width=0.7pt]
\draw[thick,domain=0:135] plot ({cos(\x)},{sin(\x)});
\draw[thick,dashed,domain=135:360] plot ({cos(\x)},{sin(\x)});
\filldraw[fill=white] ({-sqrt(2)/2},{sqrt(2)/2}) circle (.1) node[left]{\scriptsize $n$};
\filldraw[fill=white] (0,1) circle (.1) node[above]{\scriptsize 0};
\filldraw[fill=white] ({sqrt(2)/2},{sqrt(2)/2}) circle (.1) node[right]{\scriptsize 1};
\filldraw[fill=white] (1,0) circle (.1) node[right]{\scriptsize 2};
\draw[](0,-1.2);
\end{tikzpicture} 
& n \ge 2 & \emptyset & I & \emptyset
\\
\hline 
\rm A.2 
& 
\bigl( {\rm A}^{(1)}_n \bigr)^\text{id}_\text{alt} 
& 
\begin{tikzpicture} [baseline=-0.25em,scale=0.8,line width=0.7pt]
\draw[thick,domain=0:135] plot ({cos(\x)},{sin(\x)});
\draw[thick,dashed,domain=135:360] plot ({cos(\x)},{sin(\x)});
\filldraw[fill=black] ({-sqrt(2)/2},{sqrt(2)/2}) circle (.1) node[left]{\scriptsize $n$};
\filldraw[fill=white] (0,1) circle (.1) node[above]{\scriptsize 0};
\filldraw[fill=black] ({sqrt(2)/2},{sqrt(2)/2}) circle (.1) node[right]{\scriptsize 1};
\filldraw[fill=white] (1,0) circle (.1) node[right]{\scriptsize 2};
\draw[](0,-1.2);
\end{tikzpicture} 
\;\text{ or }\;
\begin{tikzpicture} [baseline=-0.25em,scale=0.8,line width=0.7pt]
\draw[thick,domain=0:135] plot ({cos(\x)},{sin(\x)});
\draw[thick,dashed,domain=135:360] plot ({cos(\x)},{sin(\x)});
\filldraw[fill=white] ({-sqrt(2)/2},{sqrt(2)/2}) circle (.1) node[left]{\scriptsize $n$};
\filldraw[fill=black] (0,1) circle (.1) node[above]{\scriptsize 0};
\filldraw[fill=white] ({sqrt(2)/2},{sqrt(2)/2}) circle (.1) node[right]{\scriptsize 1};
\filldraw[fill=black] (1,0) circle (.1) node[right]{\scriptsize 2};
\draw[](0,-1.2);
\end{tikzpicture} 
&
n \text{ odd} & \emptyset & \emptyset & \emptyset
\\
\hline 
\rm A.4 
& 
\bigl( {\rm A}^{(1)}_1 \bigr)^{\pi} 
&
\begin{tikzpicture} [baseline=-0.25em,scale=0.8,line width=0.7pt]
\draw[thick,domain=0:360] plot ({.4*cos(\x)},{.4*sin(\x)});
\draw[<->,gray] (0,.3) -- (0,-.3);
\filldraw[fill=white] (0,.4) circle (.1) node[above]{\scriptsize 0};
\filldraw[fill=white] (0,-.4) circle (.1) node[below]{\scriptsize 1};
\end{tikzpicture}
&
n=1 & \{0\} & \emptyset  & \emptyset \\
\rm A.4 
& 
\bigl( {\rm A}^{(1)}_n \bigr)^{\pi} 
&
\begin{tikzpicture} [baseline=-0.25em,scale=0.8,line width=0.7pt]
\draw[thick,domain=45:135] plot ({cos(\x)},{sin(\x)});
\draw[thick,dashed,domain=135:225] plot ({cos(\x)},{sin(\x)});
\draw[thick,domain=225:315] plot ({cos(\x)},{sin(\x)});
\draw[thick,dashed,domain=315:405] plot ({cos(\x)},{sin(\x)});
\draw[<->,gray] (0,.9) -- (0,-.9);
\draw[<->,gray] ({.9*sqrt(2)/2},{.9*sqrt(2)/2}) -- ({-.9*sqrt(2)/2},{-.9*sqrt(2)/2});
\draw[<->,gray] ({-.9*sqrt(2)/2},{.9*sqrt(2)/2}) -- ({.9*sqrt(2)/2},{-.9*sqrt(2)/2});
\filldraw[fill=white] ({-sqrt(2)/2},{sqrt(2)/2}) circle (.1) node[left]{\scriptsize $n$};
\filldraw[fill=white] (0,1) circle (.1) node[above] {\scriptsize0};
\filldraw[fill=white] ({sqrt(2)/2},{sqrt(2)/2}) circle (.1) node[right]{\scriptsize 1};
\filldraw[fill=white] ({sqrt(2)/2},{-sqrt(2)/2}) circle (.1) node[right]{\scriptsize $(n\!-\!1)/2$};
\filldraw[fill=white] (0,-1) circle (.1) node[below]{\scriptsize $(n\!+\!1)/2$};
\filldraw[fill=white] ({-sqrt(2)/2},{-sqrt(2)/2}) circle (.1) node[left]{\scriptsize $(n\!+\!3)/2$};
\end{tikzpicture} 
&
n\ge 3, \text{ odd} & \emptyset  & \emptyset & \emptyset \\
\hline 
\end{array}
\]
\end{table}

The $n=1$ case is special for QP algebras of types A.1 and A.4 (this will be explained in detail below). We will postpone the study of this case to Section \ref{sec:K:low-rank}. 

We note that the K-matrices of types A.1, A.2 and A.4 are rather simple. 
Nevertheless we give a complete list of their properties (or rather, of the associated matrix $C^{-1} K(\wt q^{-1} u)$) to allow the reader to compare with K-matrices of types BCD.1, BCD.2 and CD.4, respectively.


\subsubsection{Family A.1} \label{sec:A1}

This family consists of Satake diagrams $(X,\tau)=(\emptyset,\id)$ para\-metrized by $n\ge1$ only. All diagrams in this family are quasistandard and restrictable, and have $I^*=I$. The special $\tau$-orbits are $I_{\rm diff} = \emptyset$, $I_{\rm nsf} = I$ if $n=1$ and $I_{\rm nsf} =  \emptyset$ otherwise; see Table \ref{tab:A124:sat}.
The case $n=1$ is special and will be studied in Section \ref{sec:K:low-rank}. Here we focus on the case $n\ge2$.

Since $I_{\rm diff}=I_{\rm nsf}=\emptyset$, we have $\bm s=\bm 0$ and the QP algebra $B_{\bm c,\bm 0}$ is generated by the elements $b_j$ defined by
\eq{
b_j = y_j - c_j\, x_{j} k_j^{-1}  \tx{for}  0\le j < N \label{A1:b_j}
}
with $\bm c = (c_0,\ldots,c_n )\in \mc C = (\K^\times)^I$.
Let $\eta\in\K^\times$ and $K(u)\in\End(\K^N)(u)$ be arbitrary. Combining \eqref{A1:b_j} with \eqref{rep:A} and \eqref{affrep} we have that
\eqn{
\RT_{\eta u}(b_0) &= (\eta \, u)^{-1} E_{1N} - q\, c_0 \eta \, u E_{N1} , & \RT^\t_{\eta/u}(S(b_0)) &= c_0 \, \eta \, u^{-1} E_{1N} - q \, \eta^{-1} u E_{N1} , \\
\RT_{\eta u}(b_j) &= E_{j+1,j} - q \,c_j E_{j,j+1} , & \RT^\t_{\eta/u}(S(b_j)) &= c_j E_{j+1,j} - q E_{j,j+1} ,
}
for $1\le j < N$. Solving the twisted boundary intertwining equation \eqref{intw-tw} for $b_j$ with $1\le j <N$ gives a unique, up to a scalar multiple, $q$-independent solution, which we denote by $K_{\bm c}(u)$, 
\eq{
K_{\bm c}(u) = \sum_{1\le i \le N} \Big( \prod_{1\le j < i} c_j \Big) E_{ii}, \label{A1:K-full}
}
satisfying the twisted reflection equation \eqref{tRE}. Solving \eqref{intw-tw} with \eqref{A1:K-full} for $b_0$ fixes $\eta^2=\prod_{j} c_j^{-1}$.
Next we want to rewrite $K_{\bm c}(u)$ in the canonical form. We introduce the effective dressing parameters $\om_1,\ldots,\om_N$, all in $\K^\times$, and set
\eq{ \label{A1:c_j}
c_0 = \eta^{-2} \om_1^2\,\om_N^{-2} \tx{and} c_j = \om_{j}^{-2}\om_{j+1}^2 \tx{for} 1 \le j \le n.
}
Then, upon multiplying by $\om_1^2$, \eqref{A1:K-full} becomes $\sum_{i=1}^N \om_i^2 E_{ii}$, which coincides with the type AI solution of the constant twisted reflection equation reported in \mbox{\cite[Sec.~3]{NoSu}}. Hence, upon undressing according to \eqref{eq:Kdressing}, we reach the following result.

\begin{result}
The bare K-matrix of type A.1 is $K(u)=\Id$.
\end{result}

The properties this K-matrix are rather trivial, nevertheless we state them for completeness. The affinization identity is simply $K(u)=K_0$. The rotational symmetry is $K^{\rho}(u) = K(u)$, {\it cf.}~\eqref{eq:Ksigma}. The twisted unitarity relation \eqref{K-unit-tw} is
\[ 
(C^{-1} K(\wt q^{-1} u))^{-1} =  (-1)^n C^{-1} K(\wt q^{-1} u^{-1}). 
\]
The eigendecomposition of $C^{-1} K(\wt q^{-1} u)$ is
\[ 
V = \Id + (-1)^{\frac{N+1}{2}}\! \sum_{i \in \langle N \rangle} \!\sgn(\tfrac{N+1}{2}-i) \vartheta_i q^{\nu_i} E_{N+1-i,i}, \qu D(u) = (-1)^{\frac{N+1}{2}} \!\!\sum_{1 \le i \le N} \!\Big( \del_{i \le \frac{N+1}{2}} - \del_{i>\frac{N+1}{2}} \Big) E_{ii} . 
\]


\subsubsection{Family A.2} \label{sec:A2}

This family consists of Satake diagrams $(X,\tau) = (\{ 1,3,\ldots,n\},\id)$ and $(X,\tau) = (\{ 0,2,\ldots,n\!-\!1\},\id)$ with $n\ge1$ odd; see Table \ref{tab:A124:sat}. All diagrams in this family are quasistandard. Proposition \ref{prop:rotateintw} with $\si = \rho$ allows us to restrict to the study of the restrictable ($0\notin X$) case only. 
Note that $X$ is of type ${\rm A}_1^{\times N/2}$. We have $I^* = \{0,2,\ldots,n-1\}$ and there are no special $\tau$-orbits, {\it i.e.}~\mbox{$I_{\rm diff} = I_{\rm ns} = I_{\rm nsf} = \emptyset$}.
In contrast to the A.1 family the $n=1$ case is not special; nevertheless we will give additional comments on it in Section~\ref{sec:K:low-rank}. 

The QP algebra $B_{\bm c,\bm 0}$ is generated by elements $x_i$, $y_i$, $k_i$ with $i\in X$ and elements $b_j$ defined~by
\eq{  
b_j = y_j - c_j\, T_{j-1} T_{j+1}(x_{j})\, k_j^{-1}  \tx{with} j \in \{0,2,\ldots,n\!-\!1\} \label{A2:b_j}
}
and $c_j \in \K^\times$ all independent; here $x_{-1} := x_n$. 
Next we proceed in a similar way as we did for the A.1 family above. Combining \eqref{A2:b_j} with \eqref{rep:A} and \eqref{affrep} we have that
\eqn{
\RT_{\eta u}(b_0) &= (\eta u)^{-1} E_{1N} - q^{-1} c_0 \eta u E_{N-1,2} , & \RT^\t_{\eta/u}(S(b_0)) &= - q \eta^{-1} u E_{N1} - q^{-4} c_0 \eta u^{-1} E_{2,N-1} , \\
\RT_{\eta u}(b_j) &= E_{j+1,j} + q^{-1} c_j E_{j-1,j+2} , & \RT^\t_{\eta/u}(S(b_j)) &= -q E_{j,j+1} - q^{-4} c_j E_{j+2,j-1} ,
}
for $j \in \{2,4,\ldots,n-1\}$. 
Solving the twisted boundary intertwining equation \eqref{intw-tw} for all generators of the QP algebra yields a unique, up to a scalar multiple, solution
\eq{
{ K_{\bm c}(u) =} \sum_{1\le i\le N/2} \Big(\prod_{1\le j < i} q^{-3} c_{2j}\Big) \Big( E_{2i,2i-1} - q\,E_{2i-1,2i}  \Big). \label{A2:K-full}
}
It satisfies the twisted reflection equation \eqref{tRE}. 
Substituting \eqref{A2:K-full} for $K(u)$ in \eqref{intw-tw} with $b=b_0$ fixes $\eta^2=q^{3N/2}\prod_{j} c_{j}^{-1}$. 
To obtain the canonical form of $K_{\bm c}(u)$ we introduce the effective dressing parameters $\om_2,\om_4,\ldots,\om_N$ and set 
\eq{ \label{A2:c_j}
c_0 = q^3 \eta^{-2} \om_2 \,\om^{-1}_{N} \tx{and} c_j = q^3 \om_{j}^{-1} \om_{j+2} \tx{for} j \in \{2,4,\ldots,N-2\}. 
}
Then, upon multiplying by $\om_2$, \eqref{A2:K-full} becomes $\sum_{1\le i\le N/2} \om_{2i} ( E_{2i,2i-1} - q E_{2i-1,2i} )$, which coincides with the type AII solution reported in \cite[Sec.~3]{NoSu}. Upon undressing and multiplying with $-q^{-1/2}$, we obtain the following result.

\begin{result}
The bare K-matrix of type A.2 is
\eq{
\label{A2:K} K(u) = \sum_{1\le i\le N/2} \Big( q^{1/2} E_{2i-1,2i} - q^{-1/2} E_{2i,2i-1} \Big).
}
\end{result}

As in the A.1 case, the affinization identity is simply $K(u)=K_0$. The rotational symmetry is $K^{\rho^2}(u) = K(u)$. The twisted unitarity relation \eqref{K-unit-tw} is
\[ 
(C^{-1} K(\wt q^{-1} u))^{-1} =  C^{-1} K(\wt q^{-1} u^{-1}). 
\]
The inverse of the \mbox{K-matrix} is given by the simple formula $K(u)^{-1} = -K(u)$. 
The eigendecomposition of $C^{-1} K(\wt q^{-1} u)$ is
\begin{gather*}
V = \Id + \sum_{i =1}^{N/2} \Big( q^{\nu_i-(-1)^i/2} E_{N+1-i+(-1)^i,i} - q^{(-1)^i/2-\nu_i} E_{i,N+1-i+(-1)^i}  \Big), \\
D(u) = \sum_{i \in \langle N \rangle} \Big( \del_{i \le 2 \lceil N/4 \rceil} - \del_{i>2 \lceil N/4 \rceil} \Big) E_{ii} .
\end{gather*}


\subsubsection{Family A.4} \label{sec:A4}

This family consists of Satake diagrams $(X,\tau) = (\emptyset,\pi )$ with $n\ge1$ odd. All diagrams in this family are quasistandard and non-restrictable. We choose $I^* = \{ 0,1,\ldots,N/2\!-\!1\}$ so that $I_{\rm ns}  =  I_{\rm nsf} = \emptyset$ and $I_{\rm diff} = \{ 0 \}$ if $n=1$ and $I_{\rm diff}=\emptyset$ otherwise; see Table~\ref{tab:A124:sat}.
As for the A.1 family, the $n=1$ case is special; we will study this case in Section~\ref{sec:K:low-rank}. Here we study the $n\ge3$ case only.

The QP algebra $B_{\bm c, \bm 0}$ is generated by the elements $k_jk_{j+N/2}^{-1}$ with $0\le j < N/2$ and elements $b_j$ with $0 \le j < N$ defined by
\eq{ \label{A4:b_j}
b_j = y_j - c_j\, x_{j+N/2}\, k_j^{-1},
}
}
where $c_j=c_{j+N/2}$ for $0 \le j < N/2$; all indices must be read modulo $N$.  Combining \eqref{A4:b_j} with \eqref{rep:A} and \eqref{affrep} we have that
\eqn{
\RT_{\eta u}(b_0) &= (\eta \, u)^{-1} E_{1N} - c_0 E_{N/2,N/2+1} , & \RT^\t_{\eta/u}(S(b_0)) &= -q \, \eta^{-1} u E_{N1} + q^{-1} c_0 E_{N/2+1,N/2} , \\
\RT_{\eta u}(b_{N/2}) &= E_{N/2+1,N/2} - c_0 \eta \, u E_{N1} , & \RT^\t_{\eta/u}(S(b_0)) &= -q E_{N/2,N/2+1} + q^{-1} c_0 \eta \, u^{-1} E_{1N} , \\
\RT_{\eta u}(b_j) &= E_{j+1,j} - c_j E_{j+N/2,j+N/2+1} , & \RT^\t_{\eta/u}(S(b_j)) &= - q \, E_{j,j+1} + q^{-1} c_j E_{j+N/2+1,j+N/2} , \\
\RT_{\eta u}(b_{j+N/2}) &= E_{j+N/2+1,j+N/2} - c_{j} E_{j,j+1} , & \RT^\t_{\eta/u}(S(b_{j+N/2})) &= - q \, E_{j+N/2,j+N/2+1} + q^{-1} c_{j} E_{j+1,j} ,
}
for $0< j < N/2$. Solving the twisted boundary intertwining equation \eqref{intw-untw} for all generators of the QP algebra fixes $\eta = \prod_{0 \le j < N/2} q^{-1}c_j$ and gives a unique, up to a scalar multiple, solution
\eq{
 K_{\bm c}(u) =  \sum_{1 \le i \le N/2} \Bigl( \prod_{i \le j < N/2} q\,c_{j}^{-1} \Bigr) ( u \,E_{i+N/2,i} + E_{i,i+N/2} ) \label{A4:K-full}
}
satisfying the twisted reflection equation \eqref{tRE}. Next, in terms of the effective dressing parameters $\om_1,\ldots,\om_{N/2}$ and the scaling parameter $\eta$ we set 
\eq{ \label{A4:c_j}
c_0 = q \, \eta^{-1} \om_1/\om_{N/2} \tx{and} c_j = q \, \om_{j+1}/\om_{j} \tx{for all} 1 \le j < N/2. 
}
Then, upon multiplying by $\om_{N/2}$, \eqref{A4:K-full} becomes $\sum_{1\le i\le N/2} \om_i ( u E_{i+N/2,i} + E_{i,i+N/2} )$ and we arrive at the following result.

\begin{result}
The bare K-matrix of type A.4 is
\eq{ 
\label{A4:K} K(u) = Z^{\pi}(u) = \sum_{1 \le i \le N/2} ( u E_{i+N/2,i} + E_{i,i+N/2} ).
}
\end{result}

This K-matrix has the following properties. The constant K-matrix $K_0=\lim_{u\to0}K(u)$ is not invertible; this is typical for K-matrices associated with non-restrictable Satake diagrams. The rotational symmetry is $K^{\rho}(u) = K(u)$.
The twisted unitarity relation \eqref{K-unit-tw} is 
\[
(C^{-1} K(\wt q^{-1} u))^{-1} = - \wt q  \, C^{-1} K(\wt q^{-1} u^{-1}).
\]
The inverse if given by $K(u)^{-1} = u^{-1} K(u)$. The effective degree is $d_{\rm eff}=1$. The matrix $C^{-1} K(\wt q^{-1} u)$ has four distinct eigenvalues. Its eigendecomposition is 
\begin{gather*}
\begin{aligned}
V&= \Id + (-1)^{N/4} q^{N/4} \sum_{i=1}^{\lfloor N/4 \rfloor} \Big( (-1)^{\nu_i} q^{\nu_i} \big( E_{N/2 +1-i,i} + E_{N+1-i,i+N/2} \big) \\
& \hspace{45mm} - (-1)^{-\nu_i - N/2} q^{-\nu_i - N/2} \big( E_{i,N/2+1-i} + E_{i+N/2,N+1-i} \big) \Big),
\end{aligned} \\
D(u) = (-1)^{(N+2)/4} q^{-N/4} \sum_{i=1}^{N/2} (-1)^{\del_{i \le (N+2)/4}} \big( u E_{ii} + E_{N/2+i,N/2+i} \big).
\end{gather*}

\begin{rmk} \label{A4:solution}
The boundary intertwining equation can be straightforwardly verified for the A.4 family (we keep the assumption $N>2$). 
Owing to Proposition \ref{prop:dressintw}, it suffices to verify it for the ``bare'' case when $c_{j+N/2} = c_j=q$ for all $0 \le j <N/2$. 
From $K(u)=Z^\pi(u)$ and Proposition \ref{Tsigma} it follows that $K(u)\,\RT_u(b)\,K(u)^{-1} = \RT_u(\pi(b))$ for all $b \in B_{\bm c,\bm 0}$ and the boundary intertwining equation simplifies~to 
\eq{ \label{bintw:A4} 
\RT_u(\pi(b)) = \RT^\t_{1/u}(S(b)) \qu \text{for all}\qu b \in \{ b_j, k_j k_{j+N/2}^{-1} \}_{0 \le j < N}. 
}
When $b=k_j k_{j+N/2}^{-1}$ this is clear: we have $\pi(b) = b^{-1} = S(b)$ and $\RT_u|_{U_q(\mfh')} = \RT_1|_{U_q(\mfh')} = \RT^\t_{1/u} |_{U_q(\mfh')}$. 
It remains to verify \eqref{bintw:A4} when $b=b_j$, see \eqref{A4:b_j}.
We have
\begin{align*}
\RT_u(\pi(b_j)) &= \RT_u \big(y_{j+N/2}  - q x_j k_{j+N/2}^{-1} \big) = \RT_u \big(y_{j+N/2}) - q\, \RT_u \big(x_j\big) \RT_1(k_{j+N/2}^{-1}) ,
\\
\RT^\t_{1/u}(S(b_j)) &= \RT^\t_{1/u} \big( q\, k_j k_{j+N/2}^{-1} x_{j + N/2} - y_j k_j \big) \\
&= q\, \RT^\t_{1/u}(x_{j + N/2})\, \RT_1(k_{j+N/2}^{-1})\, \RT_1(k_j) - \RT_1(k_j)\, \RT^\t_{1/u}(y_j). 
\end{align*}
Hence, it is sufficient to establish the equations
\[
\RT_u \big(y_{j+N/2})  =  q \RT^\t_{1/u}(x_{j + N/2})\, \RT_1(k_{j+N/2}^{-1})\, \RT_1(k_j), \qq 
q \RT_u \big(x_j\big)\, \RT_1(k_{j+N/2}^{-1}) =  \RT_1(k_j)\, \RT^\t_{1/u}(y_j). 
\]
By transposing and replacing $j$ by $j+N/2$ in the first equation and replacing $u$ by $1/u$ in the second, we see that this system is equivalent to 
\[
\RT_1(k_j) \, \RT^\t_u \big(y_j)  =  q\, \RT_1(k_{j+N/2}^{\pm 1}) \, \RT_{1/u}(x_j) 
\]
(for both sign choices in the exponent). But this is a direct consequence of \eqref{rep:A} and \eqref{affrep}. \hfill \rmkend
\end{rmk}


\subsection{Untwisted K-matrices of type A.3} \label{sec:A3}

The A.3 family consists of Satake diagrams $(X,\tau) \in \Sat(A)$ with $\tau \in \Aut(A) \backslash \Sigma_A = [\psi] \cup [\psi']$ and $A$ of type ${\rm A}^{(1)}_n$. All diagrams in this family are quasistandard and are parametrized by the tuple $(N,p_1,p_2,o_1,o_2)$ and location of the affine node. We assume the underlying Dynkin diagram is labelled clockwise.
We recall that $o_i \in \{0,1\}$ determines the type of the unique $\tau$-lateral set $Y_i$ and $p_i$ equals the number of $\tau$-orbits in $X_i$. According to the value of $o_1+o_2$ we distinguish three subfamilies: A.3a (with $o_1+o_2=0$), A.3b (with $o_1+o_2=1$) and A.3c (with $o_1+o_2=2$).

Set $\ell=p_1 + \lfloor (t+N)/2\rfloor$ and $r=\lfloor t/2\rfloor - p_2$, where $t = \tau(0) \bmodN N$ (recall the notation \eqref{modN}). It follows that $o_1= N-t$ and $o_2 = t \bmod 2$, and the quadruple $(N,\ell,r,t)$ uniquely determines $(X,\tau)$. By rotating $(X,\tau)$ with an appropriate $\si \in \Sigma_A$ we may assume that either $\tau = \psi'$ with $N$ even (subfamily A.3c) or $\tau = \psi$ (subfamilies A.3ab). This determines the representative diagrams (note that $p_1 \le p_2$ for even $N$) and implies that $0 \in Y_1$, $t \in \{N-1,N\}$ and $\ell$, $r$ are bounded by $0 \leq \ell \leq r \leq \lfloor t/2\rfloor$ and, for even $N$, $\ell+r \leq \lfloor t/2\rfloor$. We thus have that $|I^*|-1 = r-\ell$, $X_1 = \{0,\ldots,\ell-1\} \cup \{ t-\ell+1,\ldots,n\}$ and $X_2 = \{r+1,\ldots,t-r-1\}$. A Satake diagram is restrictable precisely if $(\ell,t)=(0,N)$. 
The representative diagrams and special $\tau$-orbits, subject to the choice $I^* = \{ \ell, \ldots,r\}$, are listed in Table \ref{tab:A3:sat}. In all cases $I_{\rm nsf} = I_{\rm ns}$.

{\arraycolsep=3pt \def\arraystretch{1.41}
\begin{table}[h]
\caption{Family A3: representative Satake diagrams.} \label{tab:A3:sat}
\[
\begin{array}{ccccccc}
\text{Type} & \text{Name} & \text{Diagram} & (o_1,o_2) & \rm Restrictions & I_{\rm diff} & I_{\rm nsf}
\\ \hline \hline
\text{A.3a} & \bigl( {\rm A}^{(1)}_{n} \bigr)^{\psi}_{p_1,p_2} & 
\begin{tikzpicture}[baseline=-0.25em,line width=0.7pt,scale=0.8]
\draw[thick] (0,.4) -- (-.5,0) -- (0,-.4);
\draw[thick,dashed] (0,.4) -- (1,.4);
\draw[thick] (1,.4) -- (1.5,.4);
\draw[thick,dashed] (1.5,.4) -- (2.5,.4);
\draw[thick] (2.5,.4) -- (3,.4);
\draw[thick,dashed] (3,.4) -- (4,.4);
\draw[thick] (4,.4) -- (4.5,0) -- (4,-.4);
\draw[thick,dashed] (0,-.4) -- (1,-.4);
\draw[thick] (1,-.4) -- (1.5,-.4);
\draw[thick,dashed] (1.5,-.4) -- (2.5,-.4);
\draw[thick] (2.5,-.4) -- (3,-.4);
\draw[thick,dashed] (3,-.4) -- (4,-.4);
\draw[<->,gray] (0,.3) -- (0,-.3);
\draw[<->,gray] (1,.3) -- (1,-.3);
\draw[<->,gray] (1.5,.3) -- (1.5,-.3);
\draw[<->,gray] (2.5,.3) -- (2.5,-.3);
\draw[<->,gray] (3,.3) -- (3,-.3);
\draw[<->,gray] (4,.3) -- (4,-.3);
\draw[snake=brace] (-.6,.6) -- (1.1,.6) node[midway,above]{\scriptsize $p_1$};
\draw[snake=brace] (2.9,.6) -- (4.6,.6) node[midway,above]{\scriptsize $p_2$};
\filldraw[fill=black] (-.5,0) circle (.1) node[left=-1pt]{\scriptsize 0};
\filldraw[fill=black] (0,.4) circle (.1);
\filldraw[fill=black] (0,-.4) circle (.1);
\filldraw[fill=black] (1,.4) circle (.1);
\filldraw[fill=black] (1,-.4) circle (.1);
\filldraw[fill=white] (1.5,.4) circle (.1) node[above=1pt]{\scriptsize $\ell$};
\filldraw[fill=white] (1.5,-.4) circle (.1) node[below=1pt]{\scriptsize $\! N\!-\!\ell$};
\filldraw[fill=white] (2.5,.4) circle (.1) node[above=1pt]{\scriptsize $r$};
\filldraw[fill=white] (2.5,-.4) circle (.1) node[below=1pt]{\scriptsize $N\!-\!r\!$};
\filldraw[fill=black] (3,.4) circle (.1);
\filldraw[fill=black] (3,-.4) circle (.1);
\filldraw[fill=black] (4,.4) circle (.1);
\filldraw[fill=black] (4,-.4) circle (.1);
\filldraw[fill=black] (4.5,0) circle (.1);
\end{tikzpicture} 
& (0,0)  & \begin{array}{c} t=N \text{ even} \\ 0 < \ell < r \le \tfrac{N}{2}\!-\!\ell \end{array} & \{\ell,r\}  & \emptyset 
\\[-.15em]
\text{A.3a} & \bigl( {\rm A}^{(1)}_{n} \bigr)^{\psi}_{0,p_2} 
&
\begin{tikzpicture}[baseline=-0.25em,line width=0.7pt,scale=0.8]
\draw[thick] (1.5,.4) -- (1,0) -- (1.5,-.4);
\draw[thick,dashed] (1.5,.4) -- (2.5,.4);
\draw[thick] (2.5,.4) -- (3,.4);
\draw[thick,dashed] (3,.4) -- (4,.4);
\draw[thick] (4,.4) -- (4.5,0) -- (4,-.4);
\draw[thick,dashed] (1.5,-.4) -- (2.5,-.4);
\draw[thick] (2.5,-.4) -- (3,-.4);
\draw[thick,dashed] (3,-.4) -- (4,-.4);
\draw[<->,gray] (1.5,.3) -- (1.5,-.3);
\draw[<->,gray] (2.5,.3) -- (2.5,-.3);
\draw[<->,gray] (3,.3) -- (3,-.3);
\draw[<->,gray] (4,.3) -- (4,-.3);
\draw[snake=brace] (2.9,.6) -- (4.6,.6) node[midway,above]{\scriptsize $p_2$};
\filldraw[fill=white] (1,0) circle (.1) node[left=-1pt]{\scriptsize $0$};
\filldraw[fill=white] (1.5,.4) circle (.1);
\filldraw[fill=white] (1.5,-.4) circle (.1);
\filldraw[fill=white] (2.5,.4) circle (.1) node[above=1pt]{\scriptsize $r$};
\filldraw[fill=white] (2.5,-.4) circle (.1) node[below=1pt]{\scriptsize $N\!-\!r$};
\filldraw[fill=black] (3,.4) circle (.1);
\filldraw[fill=black] (3,-.4) circle (.1);
\filldraw[fill=black] (4,.4) circle (.1);
\filldraw[fill=black] (4,-.4) circle (.1);
\filldraw[fill=black] (4.5,0) circle (.1);
\end{tikzpicture} 
& (0,0)  & \begin{array}{c} t=N \text{ even} \\ 0 = \ell  < r < \tfrac{N}{2} \end{array} & \{r\} & \{0\}
\\
\text{A.3a} & \bigl( {\rm A}^{(1)}_{n} \bigr)^{\psi}_{0,0} & 
\begin{tikzpicture}[baseline=-0.25em,line width=0.7pt,scale=0.8]
\draw[thick] (1.5,.4) -- (1,0) -- (1.5,-.4);
\draw[thick,dashed] (1.5,.4) -- (2.5,.4);
\draw[thick,dashed] (1.5,-.4) -- (2.5,-.4);
\draw[thick] (2.5,.4) -- (3,0) -- (2.5,-.4);
\draw[<->,gray] (1.5,.3) -- (1.5,-.3);
\draw[<->,gray] (2.5,.3) -- (2.5,-.3);
\filldraw[fill=white] (1,0) circle (.1) node[left=-1pt]{\scriptsize $0$};
\filldraw[fill=white] (1.5,.4) circle (.1);
\filldraw[fill=white] (1.5,-.4) circle (.1);
\filldraw[fill=white] (2.5,.4) circle (.1);
\filldraw[fill=white] (2.5,-.4) circle (.1);
\filldraw[fill=white] (3,0) circle (.1) node[right=-1pt]{\scriptsize $\frac{N}{2}$};
\end{tikzpicture} 
& (0,0) & \begin{array}{c} t=N \text{ even} \\ \ell = 0, \; r =\tfrac{N}{2} \end{array} & \emptyset & \{0,\tfrac{N}{2}\}
\\[1em]
\text{A.3a} & \bigl( {\rm A}^{(1)}_{n} \bigr)^{\psi}_{\frac{N}{2}-p_2,p_2} & 
\begin{tikzpicture}[baseline=-0.25em,line width=0.7pt,scale=0.8]
\draw[thick] (0,.4) -- (-.5,0) -- (0,-.4);
\draw[thick,dashed] (0,.4) -- (1,.4);
\draw[thick] (1,.4) -- (2,.4);
\draw[thick,dashed] (2,.4) -- (3,.4);
\draw[thick] (3,.4) -- (3.5,0) -- (3,-.4);
\draw[thick,dashed] (0,-.4) -- (1,-.4);
\draw[thick] (1,-.4) -- (2,-.4);
\draw[thick,dashed] (2,-.4) -- (3,-.4);
\draw[<->,gray] (0,.3) -- (0,-.3);
\draw[<->,gray] (1,.3) -- (1,-.3);
\draw[<->,gray] (1.5,.3) -- (1.5,-.3);
\draw[<->,gray] (2,.3) -- (2,-.3);
\draw[<->,gray] (3,.3) -- (3,-.3);
\draw[snake=brace] (-.6,.6) -- (1.1,.6) node[midway,above]{\scriptsize $\tfrac{N}{2}\!-\!p_2$};
\draw[snake=brace] (1.9,.6) -- (3.6,.6) node[midway,above]{\scriptsize $p_2$};
\filldraw[fill=black] (-.5,0) circle (.1) node[left=-1pt]{\scriptsize 0};
\filldraw[fill=black] (0,.4) circle (.1);
\filldraw[fill=black] (0,-.4) circle (.1);
\filldraw[fill=black] (1,.4) circle (.1);
\filldraw[fill=black] (1,-.4) circle (.1);
\filldraw[fill=white] (1.5,.4) circle (.1) node[above=1pt]{\scriptsize $\ell$};
\filldraw[fill=white] (1.5,-.4) circle (.1) node[below=1pt]{\scriptsize $\! N\!-\!\ell$};
\filldraw[fill=black] (2,.4) circle (.1);
\filldraw[fill=black] (2,-.4) circle (.1);
\filldraw[fill=black] (3,.4) circle (.1);
\filldraw[fill=black] (3,-.4) circle (.1);
\filldraw[fill=black] (3.5,0) circle (.1);
\end{tikzpicture} 
& (0,0)  & \begin{array}{c} t=N \text{ even} \\ 0 < \ell = r \le \tfrac{N}{4} \end{array} & \{\ell \}  & \emptyset 
\\[-.15em]
\text{A.3a} & \bigl( {\rm A}^{(1)}_{n} \bigr)^{\psi}_{0,\frac{N}{2}} & 
\begin{tikzpicture}[baseline=-0.25em,line width=0.7pt,scale=0.8]
\draw[thick] (1.5,.4) -- (1,0) -- (1.5,-.4);
\draw[thick,dashed] (1.5,.4) -- (2.5,.4);
\draw[thick,dashed] (1.5,-.4) -- (2.5,-.4);
\draw[thick] (2.5,.4) -- (3,0) -- (2.5,-.4);
\draw[<->,gray] (1.5,.3) -- (1.5,-.3);
\draw[<->,gray] (2.5,.3) -- (2.5,-.3);
\filldraw[fill=white] (1,0) circle (.1) node[left=-1pt]{\scriptsize $0$};
\filldraw[fill=black] (1.5,.4) circle (.1);
\filldraw[fill=black] (1.5,-.4) circle (.1);
\filldraw[fill=black] (2.5,.4) circle (.1);
\filldraw[fill=black] (2.5,-.4) circle (.1);
\filldraw[fill=black] (3,0) circle (.1) node[right=-1pt]{\scriptsize $\frac{N}{2}$};
\end{tikzpicture} 
& (0,0) & \begin{array}{c} t=N \text{ even} \\ \ell = r=  0 \end{array} & \emptyset & \emptyset
\\[.75em] \hline
\text{A.3b} & \bigl( {\rm A}^{(1)}_{n} \bigr)^{\psi}_{p_1;p_2} & 
\begin{tikzpicture}[baseline=-0.25em,line width=0.7pt,scale=0.8]
\draw[thick] (0,.4) -- (-.5,0) -- (0,-.4);
\draw[thick,dashed] (0,.4) -- (1,.4);
\draw[thick] (1,.4) -- (1.5,.4);
\draw[thick,dashed] (1.5,.4) -- (2.5,.4);
\draw[thick] (2.5,.4) -- (3,.4);
\draw[thick,dashed] (3,.4) -- (4,.4);
\draw[thick,domain=270:450] plot({4+.4*cos(\x)},{.4*sin(\x)});
\draw[thick,dashed] (0,-.4) -- (1,-.4);
\draw[thick] (1,-.4) -- (1.5,-.4);
\draw[thick,dashed] (1.5,-.4) -- (2.5,-.4);
\draw[thick] (2.5,-.4) -- (3,-.4);
\draw[thick,dashed] (3,-.4) -- (4,-.4);
\draw[<->,gray] (0,.3) -- (0,-.3);
\draw[<->,gray] (1,.3) -- (1,-.3);
\draw[<->,gray] (1.5,.3) -- (1.5,-.3);
\draw[<->,gray] (2.5,.3) -- (2.5,-.3);
\draw[<->,gray] (3,.3) -- (3,-.3);
\draw[<->,gray] (4,.3) -- (4,-.3);
\draw[snake=brace] (-.6,.6) -- (1.1,.6) node[midway,above]{\scriptsize $p_1$};
\draw[snake=brace] (2.9,.6) -- (4.1,.6) node[midway,above]{\scriptsize $p_2$};
\filldraw[fill=black] (-.5,0) circle (.1) node[left=-1pt]{\scriptsize 0};
\filldraw[fill=black] (0,.4) circle (.1);
\filldraw[fill=black] (0,-.4) circle (.1);
\filldraw[fill=black] (1,.4) circle (.1);
\filldraw[fill=black] (1,-.4) circle (.1);
\filldraw[fill=white] (1.5,.4) circle (.1) node[above=1pt]{\scriptsize $\ell$};
\filldraw[fill=white] (1.5,-.4) circle (.1) node[below=1pt]{\scriptsize $\!N \! - \! \ell$};
\filldraw[fill=white] (2.5,.4) circle (.1) node[above=1pt]{\scriptsize $r$};
\filldraw[fill=white] (2.5,-.4) circle (.1) node[below=1pt]{\scriptsize $N\!-\!r\!$};
\filldraw[fill=black] (3,.4) circle (.1);
\filldraw[fill=black] (3,-.4) circle (.1);
\filldraw[fill=black] (4,.4) circle (.1);
\filldraw[fill=black] (4,-.4) circle (.1);
\end{tikzpicture} 
& (0,1) & \begin{array}{c} t=N \text{ odd} \\ 0< \ell < r \le \tfrac{N+1}{2} \end{array} & \{\ell,r\} & \emptyset
\\[-.15em]
\text{A.3b} & \bigl( {\rm A}^{(1)}_{n} \bigr)^{\psi}_{0;p_2} &  
\begin{tikzpicture}[baseline=-0.25em,line width=0.7pt,scale=0.8]
\draw[thick] (1.5,.4) -- (1,0) -- (1.5,-.4);
\draw[thick,dashed] (1.5,.4) -- (2.5,.4);
\draw[thick] (2.5,.4) -- (3,.4);
\draw[thick,dashed] (3,.4) -- (4,.4);
\draw[thick,domain=270:450] plot({4+.4*cos(\x)},{.4*sin(\x)});
\draw[thick,dashed] (1.5,-.4) -- (2.5,-.4);
\draw[thick] (2.5,-.4) -- (3,-.4);
\draw[thick,dashed] (3,-.4) -- (4,-.4);
\draw[<->,gray] (1.5,.3) -- (1.5,-.3);
\draw[<->,gray] (2.5,.3) -- (2.5,-.3);
\draw[<->,gray] (3,.3) -- (3,-.3);
\draw[<->,gray] (4,.3) -- (4,-.3);
\draw[snake=brace] (2.9,.6) -- (4.1,.6) node[midway,above]{\scriptsize $p_2$};
\filldraw[fill=white] (1,0) circle (.1) node[left=1pt]{\scriptsize $0$};
\filldraw[fill=white] (1.5,.4) circle (.1);
\filldraw[fill=white] (1.5,-.4) circle (.1);
\filldraw[fill=white] (2.5,.4) circle (.1) node[above=1pt]{\scriptsize $r$};
\filldraw[fill=white] (2.5,-.4) circle (.1) node[below=-.5pt]{\scriptsize $N\!-\!r$};
\filldraw[fill=black] (3,.4) circle (.1);
\filldraw[fill=black] (3,-.4) circle (.1);
\filldraw[fill=black] (4,.4) circle (.1);
\filldraw[fill=black] (4,-.4) circle (.1);
\end{tikzpicture} & (0,1) & \begin{array}{c} t=N \text{ odd} \\ 0 = \ell < r \le \tfrac{N+1}{2} \end{array} & \{r\} & \{0\} \\
\text{A.3b} & \bigl( {\rm A}^{(1)}_{n} \bigr)^{\psi}_{\frac{N\!-\!1}{2}-p_2;p_2} & 
\begin{tikzpicture}[baseline=-0.25em,line width=0.7pt,scale=0.8]
\draw[thick] (0,.4) -- (-.5,0) -- (0,-.4);
\draw[thick,dashed] (0,.4) -- (1,.4);
\draw[thick] (1,.4) -- (2,.4);
\draw[thick,dashed] (2,.4) -- (3,.4);
\draw[thick,domain=270:450] plot({3+.4*cos(\x)},{.4*sin(\x)});
\draw[thick,dashed] (0,-.4) -- (1,-.4);
\draw[thick] (1,-.4) -- (2,-.4);
\draw[thick,dashed] (2,-.4) -- (3,-.4);
\draw[<->,gray] (0,.3) -- (0,-.3);
\draw[<->,gray] (1,.3) -- (1,-.3);
\draw[<->,gray] (1.5,.3) -- (1.5,-.3);
\draw[<->,gray] (2,.3) -- (2,-.3);
\draw[<->,gray] (3,.3) -- (3,-.3);
\draw[snake=brace] (-.6,.6) -- (1.1,.6) node[midway,above]{\scriptsize $\frac{N\!-\!1}{2}\!-\!p_2$};
\draw[snake=brace] (1.9,.6) -- (3.1,.6) node[midway,above]{\scriptsize $p_2$};
\filldraw[fill=black] (-.5,0) circle (.1) node[left=-1pt]{\scriptsize 0};
\filldraw[fill=black] (0,.4) circle (.1);
\filldraw[fill=black] (0,-.4) circle (.1);
\filldraw[fill=black] (1,.4) circle (.1);
\filldraw[fill=black] (1,-.4) circle (.1);
\filldraw[fill=white] (1.5,.4) circle (.1) node[above=1pt]{\scriptsize $\ell$};
\filldraw[fill=white] (1.5,-.4) circle (.1) node[below=1pt]{\scriptsize $\!N \! - \! \ell$};
\filldraw[fill=black] (2,.4) circle (.1);
\filldraw[fill=black] (2,-.4) circle (.1);
\filldraw[fill=black] (3,.4) circle (.1);
\filldraw[fill=black] (3,-.4) circle (.1);
\end{tikzpicture} 
& (0,1) & \begin{array}{c} t=N \text{ odd} \\ 0< \ell = r \le \tfrac{N+1}{2} \end{array} & \{\ell \} & \emptyset
\\[-.15em]
\text{A.3b} & \bigl( {\rm A}^{(1)}_{n} \bigr)^{\psi}_{0;\frac{N\!-\!1}{2}} &  
\begin{tikzpicture}[baseline=-0.25em,line width=0.7pt,scale=0.8]
\draw[thick] (1.5,.4) -- (1,0) -- (1.5,-.4);
\draw[thick,dashed] (1.5,.4) -- (2.5,.4);
\draw[thick,domain=270:450] plot({2.5+.4*cos(\x)},{.4*sin(\x)});
\draw[thick,dashed] (1.5,-.4) -- (2.5,-.4);
\draw[<->,gray] (1.5,.3) -- (1.5,-.3);
\draw[<->,gray] (2.5,.3) -- (2.5,-.3);
\filldraw[fill=white] (1,0) circle (.1) node[left=1pt]{\scriptsize $0$};
\filldraw[fill=black] (1.5,.4) circle (.1);
\filldraw[fill=black] (1.5,-.4) circle (.1);
\filldraw[fill=black] (2.5,.4) circle (.1); 
\filldraw[fill=black] (2.5,-.4) circle (.1); 
\end{tikzpicture} & (0,1) & \begin{array}{c} t=N \text{ odd} \\ \ell = r = 0 \end{array} & \emptyset & \emptyset
\\[.75em] 
\hline
\text{A.3c} & \bigl( {\rm A}^{(1)}_{n} \bigr)^{\psi'}_{p_1,p_2} & 
\hspace{6pt} 
\begin{tikzpicture}[baseline=-0.25em,line width=0.7pt,scale=0.8]
\draw[thick,domain=90:270] plot({.4*cos(\x)},{.4*sin(\x)});
\draw[thick,dashed] (0,.4) -- (1,.4);
\draw[thick] (1,.4) -- (1.5,.4);
\draw[thick,dashed] (1.5,.4) -- (2.5,.4);
\draw[thick] (2.5,.4) -- (3,.4);
\draw[thick,dashed] (3,.4) -- (4,.4);
\draw[thick,domain=270:450] plot({4+.4*cos(\x)},{.4*sin(\x)});
\draw[thick,dashed] (0,-.4) -- (1,-.4);
\draw[thick] (1,-.4) -- (1.5,-.4);
\draw[thick,dashed] (1.5,-.4) -- (2.5,-.4);
\draw[thick] (2.5,-.4) -- (3,-.4);
\draw[thick,dashed] (3,-.4) -- (4,-.4);
\draw[<->,gray] (0,.3) -- (0,-.3);
\draw[<->,gray] (1,.3) -- (1,-.3);
\draw[<->,gray] (1.5,.3) -- (1.5,-.3);
\draw[<->,gray] (2.5,.3) -- (2.5,-.3);
\draw[<->,gray] (3,.3) -- (3,-.3);
\draw[<->,gray] (4,.3) -- (4,-.3);
\draw[snake=brace] (-.1,.6) -- (1.1,.6) node[midway,above]{\scriptsize $p_1$};
\draw[snake=brace] (2.9,.6) -- (4.1,.6) node[midway,above]{\scriptsize $p_2$};
\filldraw[fill=black] (0,.4) circle (.1) node[left=2pt]{\scriptsize $0$};
\filldraw[fill=black] (0,-.4) circle (.1) node[left=2pt]{\scriptsize $n$};
\filldraw[fill=black] (1,.4) circle (.1);
\filldraw[fill=black] (1,-.4) circle (.1);
\filldraw[fill=white] (1.5,.4) circle (.1) node[above=1pt]{\scriptsize $\ell$};
\filldraw[fill=white] (1.5,-.4) circle (.1) node[below=1pt]{\scriptsize $\! n\!-\!\ell$};
\filldraw[fill=white] (2.5,.4) circle (.1) node[above=1pt]{\scriptsize $r$};
\filldraw[fill=white] (2.5,-.4) circle (.1) node[below=1pt]{\scriptsize $n\!-\!r\!$};
\filldraw[fill=black] (3,.4) circle (.1);
\filldraw[fill=black] (3,-.4) circle (.1);
\filldraw[fill=black] (4,.4) circle (.1);
\filldraw[fill=black] (4,-.4) circle (.1);
\end{tikzpicture} 
& (1,1) & \begin{array}{c} t\!+\!1=N \text{ even} \\ 0 \le \ell < r < \tfrac{N}{2}\!-\!\ell \end{array} & \{\ell,r\} & \emptyset 
\\ 
\text{A.3c} & \bigl( {\rm A}^{(1)}_{n} \bigr)^{\psi'}_{\frac{N\!-\!2}{2}-p_2,p_2} & 
\hspace{6pt} 
\begin{tikzpicture}[baseline=-0.25em,line width=0.7pt,scale=0.8]
\draw[thick,domain=90:270] plot({.4*cos(\x)},{.4*sin(\x)});
\draw[thick,dashed] (0,.4) -- (1,.4);
\draw[thick] (1,.4) -- (2,.4);
\draw[thick,dashed] (2,.4) -- (3,.4);
\draw[thick,domain=270:450] plot({3+.4*cos(\x)},{.4*sin(\x)});
\draw[thick,dashed] (0,-.4) -- (1,-.4);
\draw[thick] (1,-.4) -- (2,-.4);
\draw[thick,dashed] (2,-.4) -- (3,-.4);
\draw[<->,gray] (0,.3) -- (0,-.3);
\draw[<->,gray] (1,.3) -- (1,-.3);
\draw[<->,gray] (1.5,.3) -- (1.5,-.3);
\draw[<->,gray] (2,.3) -- (2,-.3);
\draw[<->,gray] (3,.3) -- (3,-.3);
\draw[snake=brace] (-.1,.6) -- (1.1,.6) node[midway,above]{\scriptsize $\frac{N\!-\!2}{2}-\!p_2$};
\draw[snake=brace] (1.9,.6) -- (3.1,.6) node[midway,above]{\scriptsize $p_2$};
\filldraw[fill=black] (0,.4) circle (.1) node[left=2pt]{\scriptsize $0$};
\filldraw[fill=black] (0,-.4) circle (.1) node[left=2pt]{\scriptsize $n$};
\filldraw[fill=black] (1,.4) circle (.1);
\filldraw[fill=black] (1,-.4) circle (.1);
\filldraw[fill=white] (1.5,.4) circle (.1) node[above=1pt]{\scriptsize $\ell$};
\filldraw[fill=white] (1.5,-.4) circle (.1) node[below=1pt]{\scriptsize $\! n\!-\!\ell$};
\filldraw[fill=black] (2,.4) circle (.1);
\filldraw[fill=black] (2,-.4) circle (.1);
\filldraw[fill=black] (3,.4) circle (.1);
\filldraw[fill=black] (3,-.4) circle (.1);
\end{tikzpicture} 
& (1,1) & \begin{array}{c} t\!+\!1=N \text{ even} \\ 0 \le \ell = r < \tfrac{N}{4} \end{array} & \{\ell\} & \emptyset 
\\
\hline
\end{array}
\]
\end{table}
}

The QP algebra is generated by $x_i$, $y_i$, $k_i$ with $i\in X$ and $k_j k_{\tau(j)}^{-1}$ with $j \in (I \backslash X) \backslash I_{\rm nsf}$ and elements $b_j$ whose reduced expressions are  
\begin{alignat*}{99}
b_\ell &= \casesl{l}{
y_\ell - c_\ell \, x_\ell k_\ell^{-1} - s_\ell \,k_\ell^{-1} \\
y_\ell - c_\ell \, T_{w_{X_1}} (x_{t-\ell})\, k_\ell^{-1}  \\
y_\ell - c_\ell \, T_{w_{X_2}} (x_\ell) \, k_\ell^{-1} \\
y_\ell - c_\ell \, T_{w_{X_1}} T_{w_{X_2}} (x_{t-\ell})\, k_\ell^{-1}  \\
} & \qu & \casesm{l}{
\text{if } \ell<r,\; (\ell,t)=(0,N), \\
\text{if } \ell<r,\; (\ell,t) \ne (0,N), \\
\text{if } \ell=r,\; (\ell,t) = (0,N), \\
\text{if } \ell=r,\; (\ell,t) \ne (0,N), \\
}
\\
b_{t-\ell} &= \casesl{l}{
y_{t-\ell} - c_{t-\ell} \, T_{w_{X_1}} (x_\ell)\, k_{t-\ell}^{-1} \\
y_{t-\ell} - c_{t-\ell} \, T_{w_{X_1}} T_{w_{X_2}} (x_\ell)\, k_{t-\ell}^{-1} \\
} && \casesm{l}{
\text{if } \ell<r ,\; (\ell,t) \ne (0,N), \\
\text{if } \ell=r ,\; (\ell,t) \ne (0,N), \\
}
\\
b_r &= \casesl{l}{
y_r - c_r \, x_r k_r^{-1} - s_r \,k_r^{-1} \\
y_r - c_r \, T_{w_{X_2}} (x_{t-r})\, k_r^{-1} \\
} && \casesm{l}{
\text{if } \ell<r=N/2, \\
\text{if } \ell<r<N/2, \\
}\\
b_{t-r} &= y_{t-r} - c_{t-r}\,T_{w_{X_2}} (x_r)\,k_{t-r}^{-1} && \text{if } \ell<r<N/2, \\
b_{j} &= y_j - c_j\, x_{t-j}\, k_{j}^{-1} && \text{if } \ell<j<r \text{ or } t-r<j<t-\ell, 
\intertext{where $c_j = c_{t-j}$ for all $\ell<j<r$. Next, we introduce the effective dressing parameters $\om_{\ell+1},\ldots,\om_r$, the scaling parameter $\eta$ and additional free parameters $\la$ and $\mu$, all in $\K^\times$, and we set}
c_{\ell} &= 
\casesl{l}{ 
-q^{-1} \eta^{-2} \om_1^{2} \\
 (-q)^{N-t+2\ell} \, \eta^{-1}\mu\, \om_{\ell+1} \\
(-q)^N \eta^{-2}\\
 -(-q)^{N-t+2\ell} \eta^{-1} \la\, \mu  
} 
&& 
\casesm{l}{ 
\text{if } \ell<r ,\; (\ell,t)=(0,N), \\
\text{if } \ell<r ,\; (\ell,t) \ne (0,N), \\
\text{if } \ell=r ,\; (\ell,t)=(0,N), \\
\text{if } \ell=r ,\; (\ell,t) \ne (0,N),
} 
\\
c_{t-\ell} &= 
\casesl{l}{ 
- (\eta\mu)^{-1} \om_{\ell+1} \\ 
 -(-q)^{t-2\ell} (\eta\la\mu)^{-1} 
}
 &&
\casesm{l}{ 
\text{if } \ell<r,\; (\ell,t)\ne(0,N), \\
\text{if } \ell=r,\; (\ell,t)\ne(0,N), 
}
\\
c_{r} &= \casesl{l}{ 
\la\, \om_r^{-1} \\
-q^{-1} \om_{r}^{-2} 
} 
\hspace{-2pt} && 
\casesm{l}{ 
\text{if } \ell<r<N/2, \\ 
\text{if } \ell<r=N/2,  
}  \\
c_{t-r} &= -(-q)^{t-2r} \la^{-1} \om_r^{-1} \!\!\! && \text{if } \ell < r < N/2, \\
c_j &= -\om_j^{-1}\om_{j+1} && \text{if } \ell < j < r, \\
s_0 &= \frac{\mu - \mu^{-1}}{q - q^{-1}} \,\eta^{-1} \om_1 && \text{if } \ell<r ,\; (\ell,t)=(0,N), \\
s_{N/2} &= \frac{\la - \la^{-1}}{q - q^{-1}} \, \om_{N/2}^{-1} && \text{if } \ell<r ,\; (\ell,r)=(0,N/2). 
\end{alignat*}
Then solving the boundary intertwining equation \eqref{intw-untw} for all generators of the QP algebra subject to the assignments above we obtain the following solution of the reflection equation \eqref{RE}, which is  independent of $q$ for all $\ell$, $r$ and $t$.

\begin{result} \label{Res:A3}
The bare K-matrix of type A.3 is of the form \eqref{K(u):X} with $\la,\mu\in\K^\times$ and
\eq{ \label{A3:K}
\begin{aligned}
 M_1(u) &= \sum_{1 \le i \le \ell} \la \, \mu \, u \, E_{ii} + \sum_{t-\ell < i \le N}  E_{ii}, \\
 M_2(u) &= \sum_{ \ell < i \le  r} \big( \la E_{ii} + \la^{-1} E_{t+1-i,t+1-i} + E_{i,t+1-i} + E_{t+1-i,i} \big).
\end{aligned}
}
\end{result}

The K-matrices corresponding to the Satake diagrams with $X = \emptyset$ and at least one element of $I$ fixed by $\tau$ ({\it i.e.}~$o_1=0$) were found in \cite{AbRi}. Up to rotation and dressing, they correspond to the $K(u)$ defined above when $(\ell,r,t)=(0,\lfloor \tfrac{N}{2} \rfloor,N)$; in this case it is a linear combination of a diagonal matrix and the antidiagonal matrix $J$. 
In \cite{MLS}, the so-called type~I solutions, for which only one pair of non-diagonal entries is nonzero, correspond to the rotated and dressed versions of $K(u)$ when $r-\ell=1$. 
The type II solutions in {\it ibid.\@} are the rotated and dressed analogues of $K(u)$ with $r-\ell>1$. 
Around the same time, in \cite[Sec.~5]{KuMv} the dressed versions of this class of K-matrices were derived by means of a baxterization procedure (also see Remark~\ref{R:Kulish} in Section~\ref{sec:Hecke}). Finally we note that for restrictable Satake diagrams, $K(u)$ coincides with the one in \cite[eq.\@ (31)]{CGM} by setting $\la=\xi$ and $\mu=1$. 

\smallskip

The K-matrix defined above satisfies the following properties. 

\smallskip

\begin{description}[itemsep=1ex]

\item[Eigendecomposition] \hfill  $V= \Id + \la \sum_{\ell < i \le r} \big( E_{i,t+1-i}-E_{t+1-i,i} \big),$ \hfill   \hphantom{\it Eigendecomposition}
\eqn{ 
D(u) &= \displaystyle u^2h_1(u) \!\sum_{1 \le i \le \ell}\!  E_{ii} + \!\sum_{\ell <i \le t-r}\! E_{ii} + h_1(u)\,h_2(u) \!\sum_{t-r < i \le  t-\ell}\! E_{ii} + h_1(u) \!\sum_{t-\ell<i \le N}\! E_{ii}.
}

\item[Affinization] For $(\ell,t)=(0,N)$ the Satake diagram is restrictable and the affinization identity is
\[
\qq
K(u) = \frac{\la u^{-1} K_0 - \la^{-1} u K_0^{-1} + \mu_- \Id }{\la u^{-1} - \la^{-1} u + \mu_-} = \frac{ \la (u^{-1} + \mu_-\la_-^{-1} ) K_0 - \la^{-1} (u + \mu_-\la_-^{-1})  K_0^{-1}}{\la (u^{-1} + \mu_-\la_-^{-1}) - \la^{-1} (u + \mu_-\la_-^{-1})}.
\]
Here we have introduced the notation $\nu_- := \nu - \nu^{-1}$ for $\nu \in \K^\times$. 
Also, $K_0 = \lim_{u \to 0} K(u)$ is the constant K-matrix, independent of $\mu$, having eigenvalues $1$ and $-\la^{-2}$ with respective multiplicities $\bar r$ and $r$. It coincides with $J^\si$ in \cite[eq.~(2.14)]{NDS} upon identifying $\la=q^{-\si}$ and conjugating by $\sum_{i=1}^N E_{i,N+1-i}$,
and with $J^\xi$ in \cite[eq.~(30)]{CGM} upon identifying $\la = \xi$. 

\item[Bar-symmetry] \hfill $K(u)^{-1} = Z^\rho(u)^{o_1} J K(u)|_{\la \to \la^{-1},\mu \to \mu^{-1}} J Z^\rho(u^{-1})^{-o_1}.$ \hfill \hphantom{\it Bar-symmetry}

\item[Half-period] \hfill $K(-u) = K(u)|_{\mu \to - \mu}.$ \hfill \hphantom{\it Half-period}

\item[Rotations] For $N$ odd there is no nontrivial rotational symmetry. For $N$ even the Satake diagrams with $p_1=p_2$ are invariant under rotation by $\pi = \rho^\ell$.
The corresponding K-matrices satisfy
\[
\qq K^{\pi}(u) = u^2 h_1(u)\, K(u)|_{\la \leftrightarrow -\mu^{-1}} 
.
\]
\item[Reductions] For $\ell<r$ and $\mu = \pm \la^{\pm 1}$ with $\la$ generic it has $d_{\rm eff}=2$ and is singly regular, $K(\mp 1)|_{\mu = \pm \la} \ne \Id$ and $K(\pm 1)|_{\mu = \pm \la^{-1}} \ne \Id$. For $\ell<r$ and $\la=-\mu=\pm1$ it has $d_{\rm eff}=1$ and is a non-regular \gim:
\[
\qq K(u)|_{\la=- \mu = \pm 1} = \sum_{1 \le i \le \ell} (u E_{ii} + u^{-1} E_{t+1-i,t+1-i}) \mp \sum_{\ell < i \le r} (E_{i,t+1-i} + E_{t+1-i,i}) + \sum_{r < i \le t-r \atop {\text{or } i>t}} E_{ii}.
\]

\item[Diagonal cases]
For $\ell=r$ it is parametrized by $\xi := \la \mu$ only,
\eq{ \label{A3:K:diag}
\qq K_{\ell=r}(u) = u^2\, \frac{\xi - u^{-1}}{\xi - u}  \sum_{1 \le i \le \ell } E_{ii} + \sum_{\ell < i \le t-\ell} E_{ii} + \frac{\xi-u^{-1}}{\xi-u} \sum_{t- \ell <i \le N}  E_{ii}.
}
Additionally setting $\xi=\pm 1$ yields
\[
\qq K_{\ell=r}(u)|_{\xi = \pm 1} = \mp u \sum_{1 \le i \le \ell } E_{ii} + \sum_{\ell < i \le t-\ell} E_{ii} \mp u^{-1} \sum_{t- \ell <i \le N}  E_{ii}.
\]
For $\ell=r=0$ and $t=N$ it specializes to $K_{\ell=r=0,\,t=N}(u) = \Id$; it is associated to the unique restrictable Satake diagram with $|I^*|=1$, namely $(X,\tau) = (I \backslash \{0\},\psi)$. 

\noindent For $\ell<r$ we have the following diagonal limits:
\begin{align*}
\qq \lim_{\la \to 0} K(u) &= \sum_{1 \le i \le N} u^{-2\delta_{i>t-r}} E_{ii}, \qq \lim_{\la \to \infty} K(u) = \sum_{1 \le i \le N} u^{2\delta_{i\le r}} E_{ii},  \\
\qq \lim_{\mu \to 0} K(u) &=  \sum_{1 \le i \le N} u^{-2\delta_{i>t-\ell}} E_{ii}, \qq \lim_{\mu \to \infty} K(u) = \Id. 
\end{align*}
Moreover, by first imposing a relation between $\la$ and $\mu$ and then taking such a limit we may recover (a rotated version of) a K-matrices associated to a particular Satake diagram with $|I^*| = 1$. 
More precisely, given $0 \le r \le \ell \le (N-o_1-o_2)/2$, choose $o_1' \in \{0,1\}$ of the same parity as $r-\ell$ and write $\ell'= (r-\ell+o_1')/2$. Consider the Satake diagram of type A.3 with the same $N$, but with $\tau(0) = o_1'$ and $I^* = \{ \ell' \}$; this is a representative Satake diagram and denote the corresponding diagonal K-matrix by $K'(u)$. It is given by \eqref{A3:K:diag} with $\ell$ replaced by $\ell'$ and $t$ replaced by $N-o_1'$, with a free parameter $\xi \in \K^\times$.
We have: 
\eqn{
\lim_{\la \to 0} K(u)|_{\mu \to -\xi^{-1} \la} 
&= \sum_{1\le i \le t-r} E_{ii} + \frac{\xi-u^{-1}}{\xi-u} \sum_{t-r<i\le t-\ell} E_{ii} + u^{-2} \sum_{t-\ell < i \le N} E_{ii} \\
&= Z^\rho(u^{-1})^{o_1 + \ell + \ell'} K'(u)\, Z^\rho(u)^{-o_1-\ell-\ell'}.
}
\end{description}

\begin{rmk} \label{R:A3:rotate} 
Consider a non-representative Satake diagram $(X,\tau)$ of type A.3 parametrized by $(N,\ell',r',t')$ with $0\le \ell' \le r' \le t'/2$ and $t'<N-1$, {\it i.e.}~one of the following:
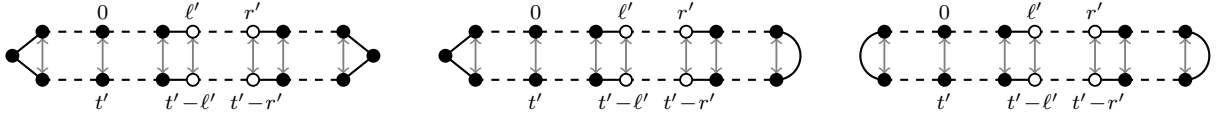
\[ 
\begin{tikzpicture}[baseline=-0.25em,line width=0.7pt,scale=0.8]
\draw[thick] (0,.4) -- (-.5,0) -- (0,-.4);
\draw[thick,dashed] (0,.4) -- (2,.4);
\draw[thick] (2,.4) -- (2.5,.4);
\draw[thick,dashed] (2.5,.4) -- (3.5,.4);
\draw[thick] (3.5,.4) -- (4,.4);
\draw[thick,dashed] (4,.4) -- (5,.4);
\draw[thick] (5,.4) -- (5.5,0) -- (5,-.4);
\draw[thick,dashed] (0,-.4) -- (2,-.4);
\draw[thick] (2,-.4) -- (2.5,-.4);
\draw[thick,dashed] (2.5,-.4) -- (3.5,-.4);
\draw[thick] (3.5,-.4) -- (4,-.4);
\draw[thick,dashed] (4,-.4) -- (5,-.4);
\draw[<->,gray] (0,.3) -- (0,-.3);
\draw[<->,gray] (1,.3) -- (1,-.3);
\draw[<->,gray] (2,.3) -- (2,-.3);
\draw[<->,gray] (2.5,.3) -- (2.5,-.3);
\draw[<->,gray] (3.5,.3) -- (3.5,-.3);
\draw[<->,gray] (4,.3) -- (4,-.3);
\draw[<->,gray] (5,.3) -- (5,-.3);
\filldraw[fill=black] (-.5,0) circle (.1);
\filldraw[fill=black] (0,.4) circle (.1);
\filldraw[fill=black] (0,-.4) circle (.1);
\filldraw[fill=black] (1,.4) circle (.1) node[above=1pt]{\scriptsize $0$};
\filldraw[fill=black] (1,-.4) circle (.1) node[below=1pt]{\scriptsize $t'$};
\filldraw[fill=black] (2,.4) circle (.1);
\filldraw[fill=black] (2,-.4) circle (.1);
\filldraw[fill=white] (2.5,.4) circle (.1) node[above=1pt]{\scriptsize $\ell'$};
\filldraw[fill=white] (2.5,-.4) circle (.1) node[below=1pt]{\scriptsize $\! t' \!-\!\ell'$};
\filldraw[fill=white] (3.5,.4) circle (.1) node[above=1pt]{\scriptsize $r'$};
\filldraw[fill=white] (3.5,-.4) circle (.1) node[below=1pt]{\scriptsize $t'\!-\!r'\!\!$};
\filldraw[fill=black] (4,.4) circle (.1);
\filldraw[fill=black] (4,-.4) circle (.1);
\filldraw[fill=black] (5,.4) circle (.1);
\filldraw[fill=black] (5,-.4) circle (.1);
\filldraw[fill=black] (5.5,0) circle (.1);
\end{tikzpicture} 
\qq
\begin{tikzpicture}[baseline=-0.25em,line width=0.7pt,scale=0.8]
\draw[thick] (0,.4) -- (-.5,0) -- (0,-.4);
\draw[thick,dashed] (0,.4) -- (2,.4);
\draw[thick] (2,.4) -- (2.5,.4);
\draw[thick,dashed] (2.5,.4) -- (3.5,.4);
\draw[thick] (3.5,.4) -- (4,.4);
\draw[thick,dashed] (4,.4) -- (5,.4);
\draw[thick,domain=270:450] plot({5+.4*cos(\x)},{.4*sin(\x)});
\draw[thick,dashed] (0,-.4) -- (2,-.4);
\draw[thick] (2,-.4) -- (2.5,-.4);
\draw[thick,dashed] (2.5,-.4) -- (3.5,-.4);
\draw[thick] (3.5,-.4) -- (4,-.4);
\draw[thick,dashed] (4,-.4) -- (5,-.4);
\draw[<->,gray] (0,.3) -- (0,-.3);
\draw[<->,gray] (1,.3) -- (1,-.3);
\draw[<->,gray] (2,.3) -- (2,-.3);
\draw[<->,gray] (2.5,.3) -- (2.5,-.3);
\draw[<->,gray] (3.5,.3) -- (3.5,-.3);
\draw[<->,gray] (4,.3) -- (4,-.3);
\draw[<->,gray] (5,.3) -- (5,-.3);
\filldraw[fill=black] (-.5,0) circle (.1);
\filldraw[fill=black] (0,.4) circle (.1);
\filldraw[fill=black] (0,-.4) circle (.1);
\filldraw[fill=black] (1,.4) circle (.1) node[above=1pt]{\scriptsize $0$};
\filldraw[fill=black] (1,-.4) circle (.1) node[below=1pt]{\scriptsize $t'$};
\filldraw[fill=black] (2,.4) circle (.1);
\filldraw[fill=black] (2,-.4) circle (.1);
\filldraw[fill=white] (2.5,.4) circle (.1) node[above=1pt]{\scriptsize $\ell'$};
\filldraw[fill=white] (2.5,-.4) circle (.1) node[below=1pt]{\scriptsize $\! t' \!-\!\ell'$};
\filldraw[fill=white] (3.5,.4) circle (.1) node[above=1pt]{\scriptsize $r'$};
\filldraw[fill=white] (3.5,-.4) circle (.1) node[below=1pt]{\scriptsize $t'\!-\!r'\!\!$};
\filldraw[fill=black] (4,.4) circle (.1);
\filldraw[fill=black] (4,-.4) circle (.1);
\filldraw[fill=black] (5,.4) circle (.1);
\filldraw[fill=black] (5,-.4) circle (.1);
\end{tikzpicture} 
\qq
\begin{tikzpicture}[baseline=-0.25em,line width=0.7pt,scale=0.8]
\draw[thick,domain=90:270] plot({.4*cos(\x)},{.4*sin(\x)});
\draw[thick,dashed] (0,.4) -- (2,.4);
\draw[thick] (2,.4) -- (2.5,.4);
\draw[thick,dashed] (2.5,.4) -- (3.5,.4);
\draw[thick] (3.5,.4) -- (4,.4);
\draw[thick,dashed] (4,.4) -- (5,.4);
\draw[thick,domain=270:450] plot({5+.4*cos(\x)},{.4*sin(\x)});
\draw[thick,dashed] (0,-.4) -- (2,-.4);
\draw[thick] (2,-.4) -- (2.5,-.4);
\draw[thick,dashed] (2.5,-.4) -- (3.5,-.4);
\draw[thick] (3.5,-.4) -- (4,-.4);
\draw[thick,dashed] (4,-.4) -- (5,-.4);
\draw[<->,gray] (0,.3) -- (0,-.3);
\draw[<->,gray] (1,.3) -- (1,-.3);
\draw[<->,gray] (2,.3) -- (2,-.3);
\draw[<->,gray] (2.5,.3) -- (2.5,-.3);
\draw[<->,gray] (3.5,.3) -- (3.5,-.3);
\draw[<->,gray] (4,.3) -- (4,-.3);
\draw[<->,gray] (5,.3) -- (5,-.3);
\filldraw[fill=black] (0,.4) circle (.1);
\filldraw[fill=black] (0,-.4) circle (.1);
\filldraw[fill=black] (1,.4) circle (.1) node[above=1pt]{\scriptsize $0$};
\filldraw[fill=black] (1,-.4) circle (.1) node[below=1pt]{\scriptsize $t'$};
\filldraw[fill=black] (2,.4) circle (.1);
\filldraw[fill=black] (2,-.4) circle (.1);
\filldraw[fill=white] (2.5,.4) circle (.1) node[above=1pt]{\scriptsize $\ell'$};
\filldraw[fill=white] (2.5,-.4) circle (.1) node[below=1pt]{\scriptsize $\! t' \!-\!\ell'$};
\filldraw[fill=white] (3.5,.4) circle (.1) node[above=1pt]{\scriptsize $r'$};
\filldraw[fill=white] (3.5,-.4) circle (.1) node[below=1pt]{\scriptsize $t'\!-\!r'\!\!$};
\filldraw[fill=black] (4,.4) circle (.1);
\filldraw[fill=black] (4,-.4) circle (.1);
\filldraw[fill=black] (5,.4) circle (.1);
\filldraw[fill=black] (5,-.4) circle (.1);
\end{tikzpicture} 
\]
These diagrams can be obtained from those in Table \ref{tab:A3:sat} by the action of $\rho=(01\ldots n)\in \Sigma_A$. Let $m= \lfloor (N\!-\!t')/2 \rfloor$. Then $(N,\ell'\!+\!m,r'\!+\!m,t'\!+\!2m)$ parametrizes a representative Satake diagram.  
It turns out that the formula appearing in Result \ref{Res:A3} defining bare K-matrix also applies to the non-representative Satake diagram given by $(N,\ell',r',t')$, so that up to a factor the K-matrix $K_{\ell=\ell',r=r',t=t'}(u)$ defined in Result \ref{Res:A3} equals the rotated K-matrix $K_{\ell=\ell'\!+\!m,r=r'\!+\!m,t=t'\!+\!2m}^{\rho^m}(u)$. Hence, if $0\le \ell \le r \le t/2 \le N/2$, the K-matrix given by Result \ref{Res:A3} is a solution of the reflection equation. In fact, we have:
\begin{flalign}
&& \qu K^{\rho^{m}}_{\ell=\ell'+m,\,r=r'+m,\,t=t'+2m}(u) = K_{\ell=\ell',\,r=r',\,t=t'}(u). && \rmkend
\end{flalign}
\end{rmk}


\subsection{Untwisted K-matrices of types C.1 and BD.2} \label{sec:K:C1BD2}

The C.1 family consists of generalized Satake diagrams $(X,\id)\in \GSat(A)$ with $A$ of type C$^{(1)}_n$. 
Moreover $(X,\tau)\in\Sat(A)$ if and only if $|I^*|\in\{1,n+1\}$. 
This family has both quasistandard and non-quasistandard generalized Satake diagrams. 
The quasistandard diagrams have $I_{\rm diff}=\emptyset$ and $I_{\rm nsf}\subseteq\{\,0,n\}$. 
The non-quasistandard diagrams have $I_{\rm nsf}=\{i\}$ with $2\le i\le n-2$ and lead to K-matrices that are not generalized cross matrices. 
These cases will be studied in Section \ref{sec:nqs}

The BD.2 family consists of generalized Satake diagrams $(X,\tau)\in \GSat(A)$ with $\tau\in\lan\phi_1,\phi_2\ran$ and $A$ of type B$^{(1)}_n$ or D$^{(1)}_n$. 
In this case $(X,\tau)\in\Sat(A)$ if and only if $X=X_{\rm alt}$ for type D.2 or $|I^*|=1$. 
All diagrams in the BD.2 family are quasistandard.


\subsubsection{Family C.1 (quasistandard case)} \label{sec:C1} 

Generalized Satake diagrams in this family are parametrized by the tuple $(N,p_1,p_2)$ with $N \ge 4$. 
By rotating with $\pi \in \Sigma_A$, if necessary, we may assume that the affine node satisfies $0 \in X \Rightarrow 0 \in X_1$.
With this choice, the diagrams are restrictable precisely if $p_1 = 0$.

Set $\ell=p_1$ and $r=n-p_2$, so that $\ell$ and $r$ are bounded by $0 \le \ell \le r \le n-\ell$. It follows that $(X,\id)\in\Sat(A)$ if $\ell=r$ or if $\ell=0$ and $r=n$. The cases with $1 \le \ell \le n-3$ and $r=\ell+2$ are non-quasistandard and will be studied in Section \ref{sec:nqs}.
We have $I^* = I \backslash X = \{\ell,\ell+1,\ldots,r\}$ so that $|I^*| -1 = r - \ell$ and $X_1 =  \{0,1,\ldots,\ell-1\}$ and $X_2 =  \{ r+1,r+2,\ldots,n\}$, which, unless $\ell=1$ or $r=n-1$, are of respective type $C_{\ell}$ and $C_{n-r}$; otherwise they are of type~$A_1$. The representative diagrams and special $\tau$-orbits are listed in Table \ref{tab:BD1:gen}. In all cases $I_{\rm diff} = \emptyset$ and $I_{\rm ns} =\{ \ell+\delta_{\ell \ne 0},\ell+\delta_{\ell \ne 0}+1,\ldots, r-\delta_{r \ne n} \}$. 

{
\arraycolsep=2pt\def\arraystretch{1.2}
\begin{table}[h]
\caption{Family C.1: representative quasistandard generalised Satake diagrams.} \label{tab:BD1:gen}
\[
\begin{array}{ccccc}
\rm Type & \rm Name & \rm Diagram & \rm Restrictions & I_{\rm nsf} \\ 
\hline\hline 
\text{C.1} & \bigl( {\rm C}_n^{(1)} \bigr)^\text{id}_{p_1,p_2}  &
\begin{tikzpicture}[baseline=-0.35em,line width=0.7pt,scale=0.8]
\draw[double,<-] (-.1,0) -- (-.5,0);
\draw[thick,dashed] (0,0) -- (1,0);
\draw[thick] (1,0) -- (1.5,0);
\draw[thick,dashed] (1.5,0) -- (2.5,0);
\draw[thick] (2.5,0) -- (3,0);
\draw[thick,dashed] (3,0) -- (4,0);
\draw[double,<-] (4.1,0) --  (4.5,0);
\filldraw[fill=black] (-.5,0) circle (.1) node[left=1pt]{\scriptsize $0$};
\filldraw[fill=black] (0,0) circle (.1) node[above=1pt]{\scriptsize $1$};
\filldraw[fill=black] (1,0) circle (.1); 
\filldraw[fill=white] (1.5,0) circle (.1) node[above=1pt]{\scriptsize $\ell$};
\filldraw[fill=white] (2.5,0) circle (.1) node[above=1pt]{\scriptsize $r$};
\filldraw[fill=black] (3,0) circle (.1);
\filldraw[fill=black] (4,0) circle (.1) node[above]{\scriptsize $n\!\!-\!\!1$};
\filldraw[fill=black] (4.5,0) circle (.1) node[right=1pt]{\scriptsize $n$};
\draw[snake=brace] (1.1,-.2) -- (-0.6,-.2) node[midway,below]{\scriptsize $p_1$};
\draw[snake=brace] (4.6,-.2) -- (2.9,-.2) node[midway,below]{\scriptsize $p_2$};
\draw[] (0,.7) -- (0,.7);
\draw[] (0,-.75) -- (0,-.75);
\end{tikzpicture} & 
\begin{array}{c} r \ne \ell+2 \\ 0 < \ell \le r \le n-\ell \end{array} &
\emptyset \\
\text{C.1} & \bigl( {\rm C}_n^{(1)} \bigr)^\text{id}_{0,p_2}  &
\begin{tikzpicture}[baseline=-0.35em,line width=0.7pt,scale=0.8]
\draw[double,<-] (-.1,0) -- (-.5,0);
\draw[thick,dashed] (0,0) -- (1,0);
\draw[thick] (1,0) -- (1.5,0);
\draw[thick,dashed] (1.5,0) -- (2.5,0);
\draw[double,<-] (2.6,0) --  (3,0);
\filldraw[fill=white] (-.5,0) circle (.1) node[left=1pt]{\scriptsize $0=\ell$};
\filldraw[fill=white] (0,0) circle (.1) node[above=1pt]{\scriptsize $1$};
\filldraw[fill=white] (1,0) circle (.1) node[above=1pt]{\scriptsize $r$};
\filldraw[fill=black] (1.5,0) circle (.1);
\filldraw[fill=black] (2.5,0) circle (.1) node[above]{\scriptsize $n\!\!-\!\!1$};
\filldraw[fill=black] (3,0) circle (.1) node[right=1pt]{\scriptsize $n$};
\draw[snake=brace] (3.1,-.2) -- (1.4,-.2) node[midway,below]{\scriptsize $p_2$};
\draw[] (0,.75) -- (0,.75);
\draw[] (0,-.75) -- (0,-.75);
\end{tikzpicture} & 
0 = \ell < r < n &
\{ 0 \} \\
\text{C.1} & \bigl( {\rm C}_n^{(1)} \bigr)^\text{id}_{0,n}  &
\begin{tikzpicture}[baseline=-0.35em,line width=0.7pt,scale=0.8]
\draw[double,<-] (-.1,0) -- (-.5,0);
\draw[thick,dashed] (0,0) -- (1,0);
\draw[double,<-] (1.1,0) --  (1.5,0);
\filldraw[fill=white] (-.5,0) circle (.1) node[left=1pt]{\scriptsize $0 = \ell = r$};
\filldraw[fill=black] (0,0) circle (.1) node[above=1pt]{\scriptsize $1$};
\filldraw[fill=black] (1,0) circle (.1) node[above]{\scriptsize $n\!\!-\!\!1$};
\filldraw[fill=black] (1.5,0) circle (.1) node[right=1pt]{\scriptsize $n$};
\draw[] (0,.75) -- (0,.75);
\draw[] (0,-.75) -- (0,-.75);
\end{tikzpicture} & 0 = \ell = r < n &
\emptyset \\
\text{C.1} & \bigl( {\rm C}_n^{(1)} \bigr)^\text{id}_{0,0}  &
\begin{tikzpicture}[baseline=-0.25em,line width=0.7pt,scale=0.8]
\draw[double,<-] (-.1,0) -- (-.5,0);
\draw[thick,dashed] (0,0) -- (1,0);
\draw[double,<-] (1.1,0) --  (1.5,0);
\filldraw[fill=white] (-.5,0) circle (.1) node[left=1pt]{\scriptsize $0 = \ell$};
\filldraw[fill=white] (0,0) circle (.1) node[above=1pt]{\scriptsize $1$};
\filldraw[fill=white] (1,0) circle (.1) node[above]{\scriptsize $n\!\!-\!\!1$};
\filldraw[fill=white] (1.5,0) circle (.1) node[right=1pt]{\scriptsize $r=n$};
\draw[] (0,.75) -- (0,.75);
\draw[] (0,-.7) -- (0,-.7);
\end{tikzpicture} & 0 = \ell < r = n &
\{ 0 , n\} \\\hline
\end{array}
\]
\end{table}
}

{\allowdisplaybreaks

The QP algebra is generated by $x_i$, $y_i$, $k_i$ with $i\in X$ and 
\begin{alignat*}{99}
b_\ell &= \casesl{l}{
y_0 - c_0 \, T_{w_{X_2}}(x_0) \, k_0^{-1} \\
y_0 - c_0 x_0 k_0^{-1} - s_0 k_0^{-1} \\
y_\ell - c_\ell \, T_{w_{X_1}}(x_{\ell})\, k_\ell^{-1} \\
y_\ell - c_\ell \, T_{w_{X_1}}T_{w_{X_2}}(x_\ell)\, k_\ell^{-1} \\
} & \qu & \casesm{l}{
\text{if } 0 = \ell = r, \\ 
\text{if } 0 = \ell < r, \\
\text{if } 0 < \ell < r, \\ 
\text{if } 0 < \ell = r, 
} \\
b_j &= y_j - c_j \, x_j\, k_j^{-1} && \text{if } \ell<j<r, \\
b_r &=  \casesl{l}{
y_r - c_r \, T_{w_{X_2}}(x_r)\, k_r^{-1} \\ 
y_n - c_n x_n k_n^{-1} - s_n k_n^{-1}
} && \casesm{l}{
\text{if } \ell<r<n, \\ 
\text{if } \ell<r=n,
} 
\intertext{with all $c_j$ independent. Next, in terms of the effective dressing parameters $\om_{\ell+1},\ldots, \om_r$, the scaling parameter $\eta$ and additional free parameters $\la$ and $\mu$ we set}
c_\ell &= \casesl{l}{
q^{2n} \eta^{-2} \\
q^{-2} \eta^{-2}\,\om_1^{4} \\
(-q)^{\ell} \eta^{-1}\,\om_{\ell+1}^{2} \\
-(-q)^{n+1} \eta^{-1} 
} && \casesm{l}{
\text{if } 0 = \ell = r, \\ 
\text{if } 0 = \ell < r , \\
\text{if } 0 < \ell < r , \\ 
\text{if } 0 < \ell = r , 
} \\
c_j &= q^{-1} \om_j^{-2} \om_{j+1}^2 && \text{if } \ell < j < r ,
\\
c_r &= 
\casesl{l}{(-q)^{n-r} \om_r^{-2} \\ q^{-2} \om_n^{-4} } 
&& 
\casesm{l}{\text{if } \ell < r < n, \\  \text{if } \ell < r = n ,} \\
s_j &= 
\casesl{l}{
\frac{\mu+\mu^{-1}}{q^{2} - q^{-2}} \, \eta^{-1}\om_1^2  \\[.25em]
\frac{\la+\la^{-1}}{q^{2} - q^{-2}}\, \om_n^{-2} 
} && \casesm{l}{
\text{if } j=0 = \ell < r, \\[.25em] 
\text{if } \ell<r=j=n. 
}
\end{alignat*}
Solving the boundary intertwining equation \eqref{intw-untw} for all generators of the QP algebra we obtain the following result. 

}

\begin{result} \label{Res:C1}
The bare K-matrix of type C.1 is of the form \eqref{K(u):X} with
\eq{
M_1(u) = \sum_{\bar \ell \le i \le n} ( \la \mu u E_{-i,-i} + E_{ii} ), \qu M_2(u) = \sum_{\bar r\le i < \bar\ell} ( -\la \, E_{-i,-i} + \la^{-1} E_{ii} + E_{-i,i} - E_{i,-i} )  \label{C1:K}
} 
where $\la=q^{\bar r}$, $\mu=q^{-\ell-1}$ except $\mu\in\K^\times$ if $\ell=0$ and $\la\in\K^\times$ if $r=n$.
\end{result}

This K-matrix has the following properties.

\smallskip

\begin{description}[itemsep=1ex]

\item[Eigendecomposition] \hfill $V=\Id+\la \displaystyle \sum_{\bar r\le i< \bar \ell} \big( E_{-i,i}+E_{i,-i} \big),$ \hfill   \hphantom{\it Eigendecomposition}
\[
D(u)=\sum_{1\le i< \bar r} \!\big( E_{-i,-i} + E_{ii} \big) + \sum_{\bar r\le i< \bar \ell} \!\Big( E_{-i,-i}+h_1(u)\,h_2(u)E_{ii} \Big) + h_1(u) \!\sum_{\bar \ell \le i \le n} \! \Big( u^2\, E_{-i,-i} + E_{ii} \Big) . 
\]

\item[Affinization] For $\ell=0$ the generalized Satake diagram is restrictable and the affinization identity~is
\[
K(u) = \frac{\la u^{-1} K_0 + \la^{-1} u K_0^{-1} - \mu_+ \Id}{ \la u^{-1} + \la^{-1} u - \mu_+ } = \frac{\la (u^{-1}-\mu_{+}\la_{+}^{-1}) K_0 + \la^{-1} (u-\mu_{+}\la_{+}^{-1}) K_0^{-1}}{ \la (u^{-1}-\mu_{+}\la_{+}^{-1}) + \la^{-1} (u-\mu_{+}\la_{+}^{-1}) }.
\]
Here $\nu_+:=\nu+\nu^{-1}$ for all $\nu \in \K^\times$.
The constant K-matrix $K_0$ has eigenvalues $1$ and $\la^{-2}$ with respective multiplicities $N-r$ and $r$. 
If additionally $r=n$, then $(X,\tau)\in \Sat(A)$, and, upon setting $\la^2=-1$ and multiplying by $C$ and dressing, $K_0$ corresponds to the CI solution of the constant twisted RE reported in \cite[Sec.~3]{NoSu}.

\item[Bar-symmetry] \hfill $K(u)^{-1}  = J K(u)|_{\la \to \la^{-1},\, \mu \to \mu^{-1}} J$. \hfill \hphantom{\it Bar-symmetry} 

\item[Half-period] \hfill $K(-u) = K(u)|_{\mu \to -\mu}$ if $\ell=0$. \hfill \hphantom{\it Half-period} 

\item[Rotations]

\; \hfill $\begin{aligned}[t]
\; K^\pi(u) &= -u\, K(u) && \text{ if } \ell+r=n\text{ and } \ell >0,
\\
K^\pi(u) &= h_2(u)\, K(u)|_{\la \leftrightarrow \mu} && \text{ if }  \ell+r=n\text{ and } \ell = 0,
\end{aligned}$ \hfill \hphantom{\it Rotations} \\[.25em]
where $\la \leftrightarrow \mu$ corresponds to $(c_0,s_0) \leftrightarrow (c_n,s_n)$.

\item[Reductions] 

For $\ell+r=n$ and $\ell<r$ it has $d_{\rm eff}=3$ and for $\ell=r=n/2$ it has $d_{\rm eff}=2$. In both cases it is singly regular, $K(1)\ne\Id$.
For $\ell=0$, $r>\ell$ and $\mu\in\{\pm\la,\pm\la^{-1}\}$ it has $d_{\rm eff}=1$ and is singly regular, \mbox{$K(\pm1)|_{\mu\in\{\pm\la,\pm\la^{-1}\}}\ne\Id$.} 
For $\ell=0$, $r=n$ and $\mu^2=\la^2=-1$ (so that $s_0=s_n=0$) it is a non-regular \gim: $K(u)|_{\la^2 = \mu^2 = -1} = \pm G(\bm\om)^{-1}\, J\, G(\bm\om)$ with $\bm\om=((-1)^{1/4},\ldots,(-1)^{1/4}) \in \K^n$.

\item[Diagonal cases] For $\ell=r$: 
\[
K(u) = \Id + \frac{u-u^{-1}}{k_1(u)}\,M_1(u).
\]
For $\ell=0$, $r>0$:
\[
\lim\limits_{\mu \to \infty} K(u) = \lim\limits_{\mu \to 0} K(u) = \Id .
\]
For $\ell=0$, $r=n$:
\begin{gather*} 
\lim_{\la \to \infty} K(u)  = u^2 \lim_{\la \to 0} K(u) = \sum_{1 \leq i \leq n} \big( u^2 E_{-i,-i} + E_{ii}  \big), \\
\lim_{\la \to \infty} K(u)|_{\mu = \la^{\pm1}} = {-} u \lim_{\la \to 0} K(u)|_{\mu = \la^{\pm1}} = \sum_{1 \le i \le n} \big( {-}u E_{-i,-i} + E_{ii} \big), \\
\lim_{\la \to \infty} K(u)|_{\mu = -\la^{\pm1}} = u \lim_{\la \to 0} K(u)|_{\mu = -\la^{\pm1}} = \sum_{1 \le i \le n} \big( u E_{-i,-i} + E_{ii} \big). 
\end{gather*}
\end{description}

\begin{rmk} \label{R:C1:rotate} 
Non-representative generalized Satake diagrams $(X,\tau)$ of type C.1 are parametrized by $(n,\ell',r')$ such that $0 \le \ell' \le r' \le n$ and, crucially, $\ell'+r'>n$. 
Then $(X,\tau)$ is the $\pi$-rotation of the representative diagram parametrized by $(n,n-r',n-\ell')$.
Hence, owing to Proposition \ref{prop:rotateintw}, $K^\pi_{\ell = n-r',r = n-\ell'}(u)$, with $K(u)$ the K-matrix defined in Result \ref{Res:C1}, is a bare K-matrix associated to $(X,\tau)$. 
In fact, the formula appearing in Result \ref{Res:C1} defining bare K-matrices turns out to apply to any non-representative diagram of type C.1.
If $\ell'+r'>n$, up to a factor the K-matrix $K(u)$ defined in Result \ref{Res:C1} equals the rotated K-matrix $K^\pi_{\ell = n-r',r = n-\ell'}(u)$:
\begin{flalign} 
&& K^\pi_{\ell = n-r'\!,\,r = n-\ell'}(u)= \Big( \tfrac{k_1(u)}{k_1(u^{-1})} K(u) \Big) \Big|_{\ell = \ell'\!,\, r= r'} . \label{Zpi:bijection} && \rmkend
\end{flalign}
\end{rmk}


\subsubsection{Family BD.2} \label{sec:BD2}

Generalized Satake diagrams in this family are parametrized by the quintuple $(N,p_1,p_2,o_1,o_2)$ with $N \ge 7$. According to the value of $o_1+o_2$ we distinguish three subfamilies: BD.2a (with $o_1+o_2=0$), BD.2b (with $o_1+o_2=1$) and D.2c (with $o_1+o_2=1$). Recall that $p_i$, the number of $\tau$-orbits in $X_i=X_{Y_i}$, is even unless $N$ is odd and $i=2$. Rotating with a suitable $\si \in \Sigma_A$, if necessary, we may assume that the affine node satisfies $0 \in X \Rightarrow 0 \in X_1$. With this choice, the diagrams are restrictable if $o_1 = p_1 = 0$.

Set $\ell=p_1+o_1$ and $r=n-p_2-o_2$, so that $\ell$ and $r$ are bounded by $0\le \ell \le r \le n$ and $r-\ell$ is even. 
The diagram $(X,\tau)$ is restrictable if $o_1=p_1=0$. 
Moreover, $(X,\tau)\in\Sat(A)$ if $\ell=r$ or if $(\ell,r)\in\{(0,n),(0,n-1),(1,n-1)\}$. 
We have $|I^*|-1 = (r-\ell)/2$ and $X = X_1 \cup X_{\rm alt} \cup X_2 $.
Here $X_1=\{0,1,\ldots,\ell-1\}$ is of type ${\rm D}_\ell$ if $\ell\ge2$, and empty otherwise.
Also, $X_2=\{r+1,r+2,\ldots,n\}$ is of type ${\rm D}_{n-r}$ if $r\le n-2$ and $N$ even, of type ${\rm B}_{n-r}$ if $r\le n-1$ and $N$ odd, and empty otherwise. 
Finally, $X_{\rm alt} = \{\ell+1,\ell+3,\ldots, r-1 \}$ is of type ${\rm A}_1^{\times \frac{r-\ell}2}$ if $\ell<r$ and empty otherwise. 
The representative diagrams and special $\tau$-orbits are listed in Table~\ref{tab:BD2:all}. In all cases $I_{\rm nsf} = I_{\rm ns}$.

{\arraycolsep=3pt \def\arraystretch{1.21}
\begin{table}[h]
\caption{Family BD.2: representative generalized Satake diagrams.} \label{tab:BD2:all}
\[
\begin{array}{cllccccc}
\text{Type}\!\!  & \qu\text{Name}\!\! & \hspace{14mm} \text{Diagram} & \hspace{-15pt} (o_1,o_2) \hspace{-8pt} & \rm Restrs. & I_{\rm diff} & \hspace{-1pt}  I_{\rm nsf} \hspace{-3pt}  \\ 
\hline\hline 
\text{B.2a} & \bigl( {\rm B}_n^{(1)} \bigr)^\text{id}_{p_1;{\rm alt};p_2} \!\! &
\begin{tikzpicture}[baseline=-0.35em,line width=0.7pt,scale=0.8]
\draw[thick] (-.6,.3) -- (0,0) -- (-.4,-.3);
\draw[thick,dashed] (0,0) -- (1,0);
\draw[thick] (1,0) -- (2.5,0);
\draw[thick,dashed] (2.5,0) -- (3.5,0);
\draw[thick] (3.5,0) -- (4.5,0);
\draw[thick,dashed] (4.5,0) -- (5.5,0);
\draw[double,->] (5.5,0) --  (5.9,0);
\filldraw[fill=black] (-.6,.3) circle (.1) node[left=1pt]{\scriptsize $0$};
\filldraw[fill=black] (-.4,-.3) circle (.1) node[left=1pt]{\scriptsize $1$};
\filldraw[fill=black] (0,0) circle (.1) node[above=1pt]{\scriptsize $2$};
\filldraw[fill=black] (1,0) circle (.1);
\filldraw[fill=white] (1.5,0) circle (.1) node[above=1pt]{\scriptsize $\ell$};
\filldraw[fill=black] (2,0) circle (.1);
\filldraw[fill=white] (2.5,0) circle (.1);
\filldraw[fill=black] (3.5,0) circle (.1);
\filldraw[fill=white] (4,0) circle (.1) node[above=1pt]{\scriptsize $r$};
\filldraw[fill=black] (4.5,0) circle (.1);
\filldraw[fill=black] (5.5,0) circle (.1) node[above]{\scriptsize $n\!\!-\!\!1$};
\filldraw[fill=black] (6,0) circle (.1) node[right=1pt]{\scriptsize $n$};
\draw[snake=brace] (-.7,.55) -- (1.1,.55) node[midway,above]{\scriptsize $p_1$};
\draw[snake=brace] (6.1,-.25) -- (4.4,-.25) node[midway,below]{\scriptsize $p_2$};
\end{tikzpicture} \!\! & (0,0) & \begin{array}{c} \ell \text{ even} \\ 2 \le \ell \le r \le n \end{array}  & \emptyset & \emptyset  \\
\text{B.2b} & \bigl( {\rm B}_n^{(1)} \bigr)^{\flL}_{p_1;{\rm alt};p_2} \!\! &
\hspace{2.3pt}\begin{tikzpicture}[baseline=-0.35em,line width=0.7pt,scale=0.8]
\draw[thick] (-.5,.3) -- (0,0) -- (-.5,-.3);
\draw[thick,dashed] (0,0) -- (1,0);
\draw[thick] (1,0) -- (2.5,0);
\draw[thick,dashed] (2.5,0) -- (3.5,0);
\draw[thick] (3.5,0) -- (4.5,0);
\draw[thick,dashed] (4.5,0) -- (5.5,0);
\draw[double,->] (5.5,0) --  (5.9,0);
\filldraw[fill=black] (-.5,.3) circle (.1) node[left=1pt]{\scriptsize $0$};
\filldraw[fill=black] (-.5,-.3) circle (.1) node[left=1pt]{\scriptsize $1$};
\filldraw[fill=black] (0,0) circle (.1) node[above=1pt]{\scriptsize $2$};
\filldraw[fill=black] (1,0) circle (.1);
\filldraw[fill=white] (1.5,0) circle (.1) node[above=1pt]{\scriptsize $\ell$};
\filldraw[fill=black] (2,0) circle (.1);
\filldraw[fill=white] (2.5,0) circle (.1);
\filldraw[fill=black] (3.5,0) circle (.1);
\filldraw[fill=white] (4,0) circle (.1) node[above=1pt]{\scriptsize $r$};
\filldraw[fill=black] (4.5,0) circle (.1);
\filldraw[fill=black] (5.5,0) circle (.1) node[above]{\scriptsize $n\!\!-\!\!1$};
\filldraw[fill=black] (6,0) circle (.1) node[right=1pt]{\scriptsize $n$};
\draw[snake=brace] (-.6,.55) --  (1.1,.55) node[midway,above]{\scriptsize $p_1$};
\draw[snake=brace] (6.1,-.25) -- (4.4,-.25) node[midway,below]{\scriptsize $p_2$};
\draw[<->,gray] (-.5,.2) -- (-.5,-.2);
\end{tikzpicture} \!\! & (1,0) & \begin{array}{c} \ell \text{ odd}  \\ 3 \le \ell \le r \le n  \end{array}  & \emptyset & \emptyset \\
\hline
\text{B.2a} & \bigl( {\rm B}_n^{(1)} \bigr)^\text{id}_{0;{\rm alt};p_2} \!\! &
\begin{tikzpicture}[baseline=-0.35em,line width=0.7pt,scale=0.8]
\draw[thick] (-.6,.3) -- (0,0) -- (-.4,-.3);
\draw[thick] (0,0) -- (1,0);
\draw[thick,dashed] (1,0) -- (2,0);
\draw[thick] (2,0) -- (3,0);
\draw[thick,dashed] (3,0) -- (4,0);
\draw[double,->] (4,0) --  (4.4,0);
\filldraw[fill=white] (-.6,.3) circle (.1) node[left=1pt]{\scriptsize $0$};
\filldraw[fill=black] (-.4,-.3) circle (.1) node[left=1pt]{\scriptsize $1$};
\filldraw[fill=white] (0,0) circle (.1) node[above=1pt]{\scriptsize $2$};
\filldraw[fill=black] (.5,0) circle (.1);
\filldraw[fill=white] (1,0) circle (.1);
\filldraw[fill=black] (2,0) circle (.1);
\filldraw[fill=white] (2.5,0) circle (.1) node[above=1pt]{\scriptsize $r$};
\filldraw[fill=black] (3,0) circle (.1);
\filldraw[fill=black] (4,0) circle (.1) node[above]{\scriptsize $n\!\!-\!\!1$};
\filldraw[fill=black] (4.5,0) circle (.1) node[right=1pt]{\scriptsize $n$};
\draw[snake=brace] (4.6,-.25) -- (2.9,-.25) node[midway,below]{\scriptsize $p_2$};
\end{tikzpicture} & (0,0) & \begin{array}{c} \ell =0 \\ 0\le r \le n  \end{array}  & \emptyset & \{0\} 
\\[1.25em]
\text{B.2b} & \bigl( {\rm B}_n^{(1)} \bigr)^{\flL}_{0;{\rm alt};p_2} \!\! &
\hspace{13.6pt}\begin{tikzpicture}[baseline=-0.35em,line width=0.7pt,scale=0.8]
\draw[thick] (-.5,.3) -- (0,0) -- (-.5,-.3);
\draw[thick] (0,0) -- (.5,0);
\draw[thick,dashed] (.5,0) -- (1.5,0);
\draw[thick] (1.5,0) -- (2.5,0);
\draw[thick,dashed] (2.5,0) -- (3.5,0);
\draw[double,->] (3.5,0) --  (3.9,0);
\filldraw[fill=white] (-.5,.3) circle (.1) node[left=1pt]{\scriptsize $0$};
\filldraw[fill=white] (-.5,-.3) circle (.1) node[left=1pt]{\scriptsize $1$};
\filldraw[fill=black] (0,0) circle (.1) node[above=1pt]{\scriptsize $2$};
\filldraw[fill=white] (.5,0) circle (.1);
\filldraw[fill=black] (1.5,0) circle (.1);
\filldraw[fill=white] (2,0) circle (.1) node[above=1pt]{\scriptsize $r$};
\filldraw[fill=black] (2.5,0) circle (.1);
\filldraw[fill=black] (3.5,0) circle (.1) node[above]{\scriptsize $n\!\!-\!\!1$};
\filldraw[fill=black] (4,0) circle (.1) node[right=1pt]{\scriptsize $n$};
\draw[snake=brace] (4.1,-.25) -- (2.4,-.25) node[midway,below]{\scriptsize $p_2$};
\draw[<->,gray] (-.5,.2) -- (-.5,-.2);
\end{tikzpicture}  & (1,0) & \begin{array}{c} \ell=1 \\ 1 \le r\le n   \end{array}   & \{0\} & \emptyset
\\ 
\hline\hline
\text{D.2a} & \bigl( {\rm D}^{(1)}_n \bigr)^\id_{p_1,{\rm alt},p_2} \!\! &
\begin{tikzpicture}[baseline=-0.35em,line width=0.7pt,scale=0.8]
\draw[thick] (-.6,.3) -- (0,0) -- (-.4,-.3);
\draw[thick,dashed] (0,0) -- (1,0);
\draw[thick] (1,0) -- (2.5,0);
\draw[thick,dashed] (2.5,0) -- (3.5,0);
\draw[thick] (3.5,0) -- (4.5,0);
\draw[thick,dashed] (4.5,0) -- (5.5,0);
\draw[thick] (5.9,.3) -- (5.5,0) -- (6.1,-.3);
\filldraw[fill=black] (-.6,.3) circle (.1) node[left=1pt]{\scriptsize $0$};
\filldraw[fill=black] (-.4,-.3) circle (.1) node[left=1pt]{\scriptsize $1$};
\filldraw[fill=black] (0,0) circle (.1) node[above=1pt]{\scriptsize $2$};
\filldraw[fill=black] (1,0) circle (.1);
\filldraw[fill=white] (1.5,0) circle (.1) node[above=1pt]{\scriptsize $\ell$};
\filldraw[fill=black] (2,0) circle (.1);
\filldraw[fill=white] (2.5,0) circle (.1);
\filldraw[fill=black] (3.5,0) circle (.1);
\filldraw[fill=white] (4,0) circle (.1) node[above=1pt]{\scriptsize $r$};
\filldraw[fill=black] (4.5,0) circle (.1);
\filldraw[fill=black] (5.5,0) circle (.1) node[above]{\scriptsize $\hspace{-8pt} n\!\!-\!\!2$};
\filldraw[fill=black] (5.9,.3) circle (.1) node[right=1pt]{\scriptsize $n\!\!-\!\!1$};
\filldraw[fill=black] (6.1,-.3) circle (.1) node[right=1pt]{\scriptsize $n$};
\draw[snake=brace] (-.7,.5) -- (1.1,.5) node[midway,above]{\scriptsize $p_1$};
\draw[snake=brace] (6.2,-.5) -- (4.4,-.5) node[midway,below]{\scriptsize $p_2$};
\end{tikzpicture} \!\! & (0,0) & \begin{array}{c} \ell, \, n\!-\!r \text{ even} \\ 2 \le \ell \le r \le n\!-\!2 \\ \ell+r < n \end{array}  & \emptyset & \emptyset \\
\text{D.2b} & \bigl( {\rm D}^{(1)}_n \bigr)^{\flR}_{p_1;{\rm alt};p_2} \!\! & 
\begin{tikzpicture}[baseline=-0.35em,line width=0.7pt,scale=0.8]
\draw[thick] (-.6,.3) -- (0,0) -- (-.4,-.3);
\draw[thick,dashed] (0,0) -- (1,0);
\draw[thick] (1,0) -- (2.5,0);
\draw[thick,dashed] (2.5,0) -- (3.5,0);
\draw[thick] (3.5,0) -- (4.5,0);
\draw[thick,dashed] (4.5,0) -- (5.5,0);
\draw[thick] (6,.3) -- (5.5,0) -- (6,-.3);
\filldraw[fill=black] (-.6,.3) circle (.1) node[left=1pt]{\scriptsize $0$};
\filldraw[fill=black] (-.4,-.3) circle (.1) node[left=1pt]{\scriptsize $1$};
\filldraw[fill=black] (0,0) circle (.1) node[above=1pt]{\scriptsize $2$};
\filldraw[fill=black] (1,0) circle (.1);
\filldraw[fill=white] (1.5,0) circle (.1) node[above=1pt]{\scriptsize $\ell$};
\filldraw[fill=black] (2,0) circle (.1);
\filldraw[fill=white] (2.5,0) circle (.1);
\filldraw[fill=black] (3.5,0) circle (.1);
\filldraw[fill=white] (4,0) circle (.1) node[above=1pt]{\scriptsize $r$};
\filldraw[fill=black] (4.5,0) circle (.1);
\filldraw[fill=black] (5.5,0) circle (.1) node[above]{\scriptsize $\hspace{-8pt} n\!\!-\!\!2$};
\filldraw[fill=black] (6,.3) circle (.1) node[right=1pt]{\scriptsize $n\!\!-\!\!1$};
\filldraw[fill=black] (6,-.3) circle (.1) node[right=1pt]{\scriptsize $n$};
\draw[snake=brace] (-.7,.5) -- (1.1,.5) node[midway,above]{\scriptsize $p_1$};
\draw[snake=brace] (6.1,-.5) -- (4.4,-.5) node[midway,below]{\scriptsize $p_2$};
\draw[<->,gray] (6,.2) -- (6,-.2);
\end{tikzpicture} \!\! & (0,1) & \begin{array}{c} \ell \text{ even},\, n\!-\!r \text{ odd} \\  2 \le \ell \le r \le n\!-\!3  \end{array}  & \emptyset & \emptyset
\\
\text{D.2c} & \bigl( {\rm D}^{(1)}_n \bigr)^{\flLR}_{p_1,{\rm alt},p_2} \!\! & \hspace{2.3pt}
\begin{tikzpicture}[baseline=-0.35em,line width=0.7pt,scale=0.8]
\draw[<->,gray] (-.5,.2) -- (-.5,-.2);
\draw[thick] (-.5,.3) -- (0,0) -- (-.5,-.3);
\draw[thick,dashed] (0,0) -- (1,0);
\draw[thick] (1,0) -- (2.5,0);
\draw[thick,dashed] (2.5,0) -- (3.5,0);
\draw[thick] (3.5,0) -- (4.5,0);
\draw[thick,dashed] (4.5,0) -- (5.5,0);
\draw[thick] (6,.3) -- (5.5,0) -- (6,-.3);
\filldraw[fill=black] (-.5,.3) circle (.1) node[left=1pt]{\scriptsize $0$};
\filldraw[fill=black] (-.5,-.3) circle (.1) node[left=1pt]{\scriptsize $1$};
\filldraw[fill=black] (0,0) circle (.1) node[above=1pt]{\scriptsize $2$};
\filldraw[fill=black] (1,0) circle (.1);
\filldraw[fill=white] (1.5,0) circle (.1) node[above=1pt]{\scriptsize $\ell$};
\filldraw[fill=black] (2,0) circle (.1);
\filldraw[fill=white] (2.5,0) circle (.1);
\filldraw[fill=black] (3.5,0) circle (.1);
\filldraw[fill=white] (4,0) circle (.1) node[above=1pt]{\scriptsize $r$};
\filldraw[fill=black] (4.5,0) circle (.1);
\filldraw[fill=black] (5.5,0) circle (.1) node[above]{\scriptsize $\hspace{-8pt} n\!\!-\!\!2$};
\filldraw[fill=black] (6,.3) circle (.1) node[right=1pt]{\scriptsize $n\!\!-\!\!1$};
\filldraw[fill=black] (6,-.3) circle (.1) node[right=1pt]{\scriptsize $n$};
\draw[snake=brace] (-.6,.5) -- (1.1,.5) node[midway,above]{\scriptsize $p_1$};
\draw[snake=brace] (6.2,-.5) -- (4.4,-.5) node[midway,below]{\scriptsize $p_2$};
\draw[<->,gray] (6,.2) -- (6,-.2);
\end{tikzpicture} \!\! & (1,1) & \begin{array}{c} \ell,\, n\!-\!r \text{ odd} \\ 3 \le \ell \le r \le n\!-\!3 \\ \ell+r < n  \end{array}  & \emptyset & \emptyset 
\\ \hline
\text{D.2a} & \bigl( {\rm D}^{(1)}_n \bigr)^\id_{0,{\rm alt},p_2} \!\! &
\begin{tikzpicture}[baseline=-0.35em,line width=0.7pt,scale=0.8]
\draw[thick] (-.6,.3) -- (0,0) -- (-.4,-.3);
\draw[thick] (0,0) -- (1,0);
\draw[thick,dashed] (1,0) -- (2,0);
\draw[thick] (2,0) -- (3,0);
\draw[thick,dashed] (3,0) -- (4,0);
\draw[thick] (4.4,.3) -- (4,0) -- (4.6,-.3) ;
\filldraw[fill=white] (-.6,.3) circle (.1) node[left=1pt]{\scriptsize $0$};
\filldraw[fill=black] (-.4,-.3) circle (.1) node[left=1pt]{\scriptsize $1$};
\filldraw[fill=white] (0,0) circle (.1) node[above=1pt]{\scriptsize $2$};
\filldraw[fill=black] (.5,0) circle (.1);
\filldraw[fill=white] (1,0) circle (.1);
\filldraw[fill=black] (2,0) circle (.1);
\filldraw[fill=white] (2.5,0) circle (.1) node[above=1pt]{\scriptsize $r$};
\filldraw[fill=black] (3,0) circle (.1);
\filldraw[fill=black] (4,0) circle (.1) node[above]{\scriptsize $\hspace{-8pt} n\!\!-\!\!2$};
\filldraw[fill=black] (4.4,.3) circle (.1) node[right=1pt]{\scriptsize $n\!\!-\!\!1$};
\filldraw[fill=black] (4.6,-.3) circle (.1) node[right=1pt]{\scriptsize $n$};
\draw[snake=brace] (4.7,-.5) -- (2.9,-.5) node[midway,below]{\scriptsize $p_2$};
\end{tikzpicture} & (0,0) & \begin{array}{c} \ell=0,\,n\!-\!r \text{ even} \\ 0 \le r \le n\!-\!2  \end{array} & \emptyset & \{ 0 \}  \\
\text{D.2b} & \bigl( {\rm D}^{(1)}_n \bigr)^{\flR}_{0;{\rm alt};p_2} \!\! & 
\begin{tikzpicture}[baseline=-0.35em,line width=0.7pt,scale=0.8]
\draw[thick] (-.6,.3) -- (0,0) -- (-.4,-.3);
\draw[thick] (0,0) -- (1,0);
\draw[thick,dashed] (1,0) -- (2,0);
\draw[thick] (2,0) -- (3,0);
\draw[thick,dashed] (3,0) -- (4,0);
\draw[thick] (4.5,.3) -- (4,0) -- (4.5,-.3) ;
\filldraw[fill=white] (-.6,.3) circle (.1) node[left=1pt]{\scriptsize $0$};
\filldraw[fill=black] (-.4,-.3) circle (.1) node[left=1pt]{\scriptsize $1$};
\filldraw[fill=white] (0,0) circle (.1) node[above=1pt]{\scriptsize $2$};
\filldraw[fill=black] (.5,0) circle (.1);
\filldraw[fill=white] (1,0) circle (.1);
\filldraw[fill=black] (2,0) circle (.1);
\filldraw[fill=white] (2.5,0) circle (.1) node[above=1pt]{\scriptsize $r$};
\filldraw[fill=black] (3,0) circle (.1);
\filldraw[fill=black] (4,0) circle (.1) node[above]{\scriptsize $\hspace{-8pt} n\!\!-\!\!2$};
\filldraw[fill=black] (4.5,.3) circle (.1) node[right=1pt]{\scriptsize $n\!\!-\!\!1$};
\filldraw[fill=black] (4.5,-.3) circle (.1) node[right=1pt]{\scriptsize $n$};
\draw[snake=brace] (4.6,-.5) -- (2.9,-.5) node[midway,below]{\scriptsize $p_2$};
\draw[<->,gray] (4.5,.2) -- (4.5,-.2);
\end{tikzpicture}  & (0,1) & \begin{array}{c} \ell=0,\,n\!-\!r \text{ odd} \\  0 \le r \le n\!-\!3  \end{array} & \emptyset & \{ 0 \} 
\\
\text{D.2c} & \bigl( {\rm D}^{(1)}_n \bigr)^{\flLR}_{0,{\rm alt},p_2} \!\! & \hspace{13.7pt}
\begin{tikzpicture}[baseline=-0.35em,line width=0.7pt,scale=0.8]
\draw[thick] (-.5,.3) -- (0,0) -- (-.5,-.3);
\draw[thick] (0,0) -- (.5,0);
\draw[thick,dashed] (.5,0) -- (1.5,0);
\draw[thick] (1.5,0) -- (2.5,0);
\draw[thick,dashed] (2.5,0) -- (3.5,0);
\draw[thick] (4,.3) -- (3.5,0) -- (4,-.3);
\filldraw[fill=white] (-.5,.3) circle (.1) node[left=1pt]{\scriptsize $0$};
\filldraw[fill=white] (-.5,-.3) circle (.1) node[left=1pt]{\scriptsize $1$};
\filldraw[fill=black] (0,0) circle (.1) node[above=1pt]{\scriptsize $2$};
\filldraw[fill=white] (.5,0) circle (.1);
\filldraw[fill=black] (1.5,0) circle (.1);
\filldraw[fill=white] (2,0) circle (.1) node[above=1pt]{\scriptsize $r$};
\filldraw[fill=black] (2.5,0) circle (.1);
\filldraw[fill=black] (3.5,0) circle (.1) node[above]{\scriptsize $\hspace{-8pt} n\!\!-\!\!2$};
\filldraw[fill=black] (4,.3) circle (.1) node[right=1pt]{\scriptsize $n\!\!-\!\!1$};
\filldraw[fill=black] (4,-.3) circle (.1) node[right=1pt]{\scriptsize $n$};
\draw[snake=brace] (4.1,-.5) -- (2.4,-.5) node[midway,below]{\scriptsize $p_2$};
\draw[<->,gray] (-.5,.2) -- (-.5,-.2);
\draw[<->,gray] (4,.2) -- (4,-.2);
\end{tikzpicture}  & (1,1) & \begin{array}{c} \ell=1,\,n\!-\!r \text{ odd} \\ 1 \le r \le n\!-\!3  \end{array}  & \{ 0 \} & \emptyset 
\\
\hline
\text{D.2a} & \bigl( {\rm D}^{(1)}_{n} \bigr)^\id_{0,{\rm alt},0} \!\! & 
\hspace{0pt}
\begin{tikzpicture}[baseline=-0.35em,line width=0.7pt,scale=.8]
\draw[thick] (-.6,.3) -- (0,0) -- (-.4,-.3);
\draw[thick] (0,0) -- (1,0);
\draw[thick,dashed] (1.0,0) -- (2.0,0);
\draw[thick] (2.0,0) -- (2.5,0);
\draw[thick] (2.9,.3) -- (2.5,0) -- (3.1,-.3);
\filldraw[fill=white] (-.6,.3) circle (.1) node[left=1pt]{\scriptsize $0$};
\filldraw[fill=black] (-.4,-.3) circle (.1) node[left=1pt]{\scriptsize $1$};
\filldraw[fill=white] (0,0) circle (.1) node[above=1pt]{\scriptsize $2$};
\filldraw[fill=black] (.5,0) circle (.1);
\filldraw[fill=white] (1,0) circle (.1);
\filldraw[fill=black] (2.0,0) circle (.1);
\filldraw[fill=white] (2.5,0) circle (.1) node[below]{\scriptsize $\hspace{-4pt}n\!\!-\!\!2$};
\filldraw[fill=black] (2.9,.3) circle (.1) node[right=1pt]{\scriptsize $n\!\!-\!\!1 \hspace{-5pt}$};
\filldraw[fill=white] (3.1,-.3) circle (.1) node[right=1pt]{\scriptsize $n$};
\end{tikzpicture} 
& (0,0) & \begin{array}{c} \ell=0,\,r=n \\ n \text{ even}\end{array} & \emptyset & \!\!\{ 0, n \} \\[1.2em] 
%
\text{D.2b} & \bigl( {\rm D}^{(1)}_{n} \bigr)^{\flR}_{0;{\rm alt};0} \!\! &
\hspace{0pt}
\begin{tikzpicture}[baseline=-0.35em,line width=0.7pt,scale=.8]
\draw[thick] (-.6,.3) -- (0,0) -- (-.4,-.3);
\draw[thick] (0,0) -- (1,0);
\draw[thick,dashed] (1.0,0) -- (2.0,0);
\draw[thick] (2.0,0) -- (3,0);
\draw[thick] (3.5,.3) -- (3,0) -- (3.5,-.3);
\filldraw[fill=white] (-.6,.3) circle (.1) node[left=1pt]{\scriptsize $0$};
\filldraw[fill=black] (-.4,-.3) circle (.1) node[left=1pt]{\scriptsize $1$};
\filldraw[fill=white] (0,0) circle (.1) node[above=1pt]{\scriptsize $2$};
\filldraw[fill=black] (.5,0) circle (.1);
\filldraw[fill=white] (1,0) circle (.1);
\filldraw[fill=black] (2,0) circle (.1);
\filldraw[fill=white] (2.5,0) circle (.1);
\filldraw[fill=black] (3,0) circle (.1) node[below]{\scriptsize $\hspace{-5pt}n\!\!-\!\!2$};
\filldraw[fill=white] (3.5,.3) circle (.1) node[right=1pt]{\scriptsize $n\!\!-\!\!1 \hspace{-5pt}$};
\filldraw[fill=white] (3.5,-.3) circle (.1) node[right=1pt]{\scriptsize $n$};
\draw[<->,gray] (3.5,.2) -- (3.5,-.2);
\end{tikzpicture} 
& (0,1) & \begin{array}{c} \ell=0,\,r=n\!-\!1\\ n \text{ odd}\end{array}  & \{ n \} & \{ 0 \} \\[1.2em]
%
\text{D.2c} & \bigl( {\rm D}^{(1)}_{n} \bigr)^{\flLR}_{0,{\rm alt},0} \!\! & 
\hspace{13.8pt}
\begin{tikzpicture}[baseline=-0.35em,line width=0.7pt,scale=.8]
\draw[thick] (0,.3) -- (.5,0) -- (0,-.3);
\draw[thick] (.5,0) -- (1,0);
\draw[thick,dashed] (1,0) -- (2,0);
\draw[thick] (2,0) -- (3,0);
\draw[thick] (3.5,.3) -- (3,0) -- (3.5,-.3);
\filldraw[fill=white] (0,.3) circle (.1) node[left=1pt]{\scriptsize $0$};
\filldraw[fill=white] (0,-.3) circle (.1) node[left=1pt]{\scriptsize $1$};
\draw[<->,gray] (0,.2) -- (0,-.2);
\filldraw[fill=black] (.5,0) circle (.1) node[above=1pt]{\scriptsize $2$};
\filldraw[fill=white] (1,0) circle (.1);
\filldraw[fill=black] (2,0) circle (.1);
\filldraw[fill=white] (2.5,0) circle (.1);
\filldraw[fill=black] (3,0) circle (.1) node[below]{\scriptsize $\hspace{-5pt}n\!\!-\!\!2$};
\filldraw[fill=white] (3.5,.3) circle (.1) node[right=1pt]{\scriptsize $n\!\!-\!\!1 \hspace{-5pt}$};
\filldraw[fill=white] (3.5,-.3) circle (.1) node[right=1pt]{\scriptsize $n$};
\draw[<->,gray] (3.5,.2) -- (3.5,-.2);
\end{tikzpicture} 
& (1,1) & \begin{array}{c} \ell=1,\,r=n\!-\!1\\ n \text{ even}\end{array} & \!\!\{0,n\} & \emptyset \\\hline 
\end{array}
\]
\end{table}
}

Recall that $I' = (I\backslash X) \backslash (Y_1 \cup Y_2)$. The QP algebra is generated by $x_i$, $y_i$, $k_i$ with $i \in X$, $(k_0 k_1^{-1})^{\pm 1}$ if $\ell=1$, $(k_n k_{n-1}^{-1})^{\pm 1}$ if $r=n-1$ and $N$ even, and elements $b_j$ whose reduced expressions are
\begin{alignat*}{99}
b_0 & = 
\casesl{l}{
y_0 - c_0\, T_{w_{X_2}}(x_0)\, k_0^{-1} \\
y_0 - c_0\, x_0\, k_0^{-1} - s_0 k_0^{-1} \\
y_0 - c_0\, T_{w_{X_2}}(x_1)\, k_0^{-1} \\
y_0 - c_0\, T_2(x_1)\, k_0^{-1} \\
}
& &
\casesm{l}{
\text{if } 0=\ell=r,\\
\text{if } 0=\ell<r,\\
\text{if } 1=\ell=r, \\
\text{if } 1=\ell<r, \\
} 
\\
b_1 & = 
\casesl{l}{
y_1 - c_1\, T_{w_{X_2}}(x_0)\, k_1^{-1} \\
y_1 - c_1\, T_2(x_0)\, k_1^{-1} \\
}
& &
\casesm{l}{
\text{if } 1=\ell=r, \\
\text{if } 1=\ell<r, \\
} 
\\
b_\ell & = 
\casesl{l}{
y_\ell - c_\ell\, T_{\ell+1}T_{w_{X_1}}(x_\ell)\, k_\ell^{-1} \\
y_\ell - c_\ell\, T_{w_{X_2}}T_{w_{X_1}}(x_\ell)\, k_\ell^{-1} \\
}
& &
\casesm{l}{
\text{if } 1<\ell<r, \\
\text{if } 1<\ell=r, \\
} 
\\
b_j & = y_j - c_j\, T_{j-1}T_{j+1}(x_j)\, k_j^{-1} & & \text{if } j \in I', \!\!
\\
b_r & = y_r - c_r\, T_{r-1}T_{w_{X_2}}(x_r)\, k_r^{-1} & & \text{if } \ell < r < n-(-1)^N,
\\
b_{n-1} & = y_{n-1} - c_{n-1}\, T_{n-2}(x_n)\, k_{n-1}^{-1}  & \qu & \text{if } \ell < r = n-(-1)^N,
\\
b_n & = 
\casesl{l}{
y_n - c_n\, T_{n-2}(x_{n-1})\, k_{n}^{-1} \\
y_n - c_n\, x_n\, k_n^{-1} - s_n k_n^{-1} \\
}
& &
\casesm{l}{
\text{if } \ell < r = n-(-1)^N \\
\text{if } \ell < r = N/2,
} 
\intertext{where all $c_j$ are independent. Next, in terms of the effective dressing parameters $\om_{\ell+2},\om_{\ell+4},\ldots,\om_{r}$, the scaling parameter $\eta$ and additional free parameters $\la$ and $\mu$ we set }
c_0 & = 
\casesl{l}{
q^{2(\lceil N/2\rceil-1)} \eta^{-2} \\
q^{-1} \eta^{-2}\, \om_2^2 \\
-(-q)^{\lceil N/2\rceil-1} (\mu \,\eta)^{-1} \\
-q\, \mu^{-1}\, \eta^{-1}\,\om_3 \\
} \qq
& &
\casesm{l}{
\text{if } 0=\ell=r,\\
\text{if } 0=\ell<r, \\
\text{if } 1=\ell=r,  \\
\text{if } 1=\ell<r, \\
}
\\
c_1 & = 
\casesl{l}{
-(-q)^{\lceil N/2\rceil-1} \mu\,\eta^{-1} \\
q\,\mu\,\eta^{-1}\,\om_3 \\
}
& &
\casesm{l}{
\text{if } 1=\ell=r, \\
\text{if } 1=\ell<r, \\
} 
\\
c_\ell & = 
\casesl{l}{
-(-q)^\ell \,\eta^{-1}\, \om_{\ell+2} \\
-(-q)^{\lceil N/2\rceil-1} \,\eta^{-1} \\
}
& &
\casesm{l}{
\text{if } 1<\ell<r, \\
\text{if } 1<\ell=r, \\
} 
\\
c_j & = q\, \om_j^{-1} \om_{j+2} & & \text{if } j \in I', \!\!
\\
c_r & = -(-q)^{\lceil N/2\rceil-r}\, \om_r^{-1} & & \text{if } \ell<r<n-(-1)^N,
\\
c_{n-1} & = q\,\la\,\,\om_{n-1}^{-1} & & \text{if } \ell<r = n-(-1)^N,
\\
c_n & = 
\casesl{l}{
-q\,\la^{-1}\,\om_{n-1}^{-1} \\
q^{-1}\, \om_n^{-2} \\
}
& &
\casesm{l}{
\text{if } \ell<r=n-(-1)^N, \\
\text{if } \ell<r=N/2,
} 
\\
s_0 &= \frac{\mu + \mu^{-1}}{q-q^{-1}} \eta^{-1} \om_2 && \text{if } 0=\ell<r, \\
s_n &= \frac{\la + \la^{-1}}{q-q^{-1}} \om_n^{-1} && \text{if } \ell<r=N/2.
\end{alignat*}
Then, solving the boundary intertwining relation \eqref{intw-untw} for all generators of the QP algebra we obtain the following result.

\begin{result}
The bare K-matrix of type BD.2 is of the form \eqref{K(u):X} with 
\spl{
M_1(u) &= \sum_{\overline{\ell} \le i \le n} \big( \la \mu u E_{-i,-i} + E_{ii} \big), \\
M_2(u) &= -\del_{\ell,1}\, u^{-1}\,(\mu-\mu^{-1})\, E_{nn} - \del_{N,2n}\del_{r,n-1}\, (\la - \la^{-1}) \, E_{11} \\
& \qu + \sum_{\overline{r} \le i < \overline{\ell} } \big( {-}\la E_{-i,-i} + \la^{-1} E_{ii} + \eps_i\, (E_{-i-\eps_i,i} - E_{i,-i-\eps_i}) \big). \label{D2:K}
}
where $\eps_i = (-1)^{\bar\imath-\ell}$, $\mu = q^{-\ell+1}$, $\la = q^{N/2 - r - 1}$ except $\mu\in\K^\times$ if $\ell\in\{0,1\}$ and $\la\in\K^\times$ if $r\in\{n-1,n\}$ for type D.2 only. 
\end{result}

This K-matrix has the following properties.

\smallskip

\begin{description} [itemsep=1ex]

\item[Eigendecomposition] \hfill  $V= \Id - \la \sum_{\bar r \le i < \bar \ell}  \eps_i \big( E_{-i-\eps_i,i}-E_{i,-i-\eps_i} \big),$ \hfill   \hphantom{\it Eigendecomposition}
\eqn{ 
D(u) &= \sum_{1 \le i < \bar r} \big( E_{-i,-i} + E_{ii} \big) + \sum_{\bar r \le i < \bar \ell} \Big( E_{-i,-i}+h_1(u)\,h_2(u)E_{ii} \Big) 
\\
& - \tfrac{u-u^{-1}}{k_1(u)\,k_2(u)} \big( \del_{\ell1} u^{-1}\mu_- E_{nn} + \del_{N,2n}\del_{r,n-1} \la_- \big) + h_1(u) \sum_{\bar \ell \le i \le n} \big( u^2\,E_{-i,-i} + E_{ii} \big).
}

\item[Affinization] For $\ell=0$ the generalized Satake diagram is restrictable and the affinization identity is as in the C.1 case. If additionally $r=n$ if $n$ is even or $r=n-1$ if $n$ is odd, then $(X,\tau)\in \Sat(A)$, and, upon setting $\la=\sqrt{-1}$ and multiplying by $C$ and dressing, $K_0$ corresponds to the CI solution of the constant twisted RE reported in \mbox{\cite[Sec.~3]{NoSu}.}

\item[Bar-symmetry] \qu $K(u)^{-1} = Z^{\flL}(u)^{o_1} (Z^{\flR})^{o_2} J K(u)|_{\la \to -\la^{-1},\mu \to -\mu^{-1}} J Z^{\flL}(u^{-1})^{-o_1} (Z^{\flR})^{-o_2}.$

\item[Half-period]  \hfill $K(-u) = K(u)|_{\mu \to -\mu}$ if $\ell\in\{0,1\}$. \hfill \hphantom{\it Half-period}

\item[Rotations]\; \hfill $\begin{aligned}[t]
\; K^{\flL}(u) &= K(u)|_{\mu \to \mu^{-1}} && \text{ if } \ell = 1.
\\
K^{\flR}(u) &= h_1(u) h_2(u)\, K(u)|_{\la \to \la^{-1}}  && \text{ if } r = n-1 \text{ for D.2},
\\
K^{\pi}(u) &= h_2(u)  \, K(u)|_{\la \leftrightarrow \mu} && \text{ if } \ell+ r = n  \text{ and $\ell\le1$ for D.2} ,
\\[.25em]
K^{\pi}(u) &=-u \, K(u) && \text{ if } \ell + r = n  \text{ and $\ell\ge2$  for D.2} , 
\end{aligned}$ \hfill \hphantom{\it Rotations} \\
where $\mu \to \mu^{-1}$ corresponds to $c_0 \leftrightarrow c_1$, $\la \to \la^{-1}$ to $c_{n-1} \leftrightarrow c_n$ and $\la\leftrightarrow \mu$ to $(c_0,c_1,s_0) \leftrightarrow (c_n,c_{n-1},s_n)$.

\item[Reductions] For $\ell=1$ and $\mu\in\{\pm\la,\pm\la^{-1}\}$ it has $d_{\rm eff}=3$. For $\ell=0$ and $\mu\in\{\pm\la,\pm\la^{-1}\}$ it has $d_{\rm eff}=1$. In the latter two cases it is singly regular, \mbox{$K(\pm1)|_{\mu\in\{\pm\la,\pm\la^{-1}\}}\ne\Id$.} 

\noindent For $N$ even, $\ell+r=n$ and $1<\ell<r$ it has $d_{\rm eff}=3$. For $\ell=r=n/2$  it has $d_{\rm eff}=2$. In both cases it is singly regular, $K(1)\ne \Id$. 

\noindent For $N$ and $n$ even, $l=0$, $r=n$ and $\la^2 =\mu^2= -1$ (so that $s_0=s_n=0$) it is a non-regular \gim: $K(u)|_{\la^2 = \mu^2 = -1} = -\sqrt{-1}\,\sum_{\overline{r} \le i < \overline{\ell} } \,\eps_i\,(E_{-i-\eps_i,i} - E_{i,-i-\eps_i})$.

\item[Diagonal cases] For $\ell=r$:
\begin{gather*}
K(u) = \Id + \frac{u-u^{-1}}{k_1(u)} \sum_{\bar \ell \le i \le n} ( \la \mu u E_{-i,-i} + E_{ii}).
\\
\intertext{For $\ell\le1$, $r\ge n\!-\!1$ for D.2:}
\lim_{\la\to0} K(u) = \sum_{1\le i \le n} ( E_{-i,-i} + u^{-2} E_{ii} ), \qq
\lim_{\la \to 0} K(u)|_{\mu = \pm\la} = \sum_{1 \le i \le n} (E_{-i,-i} \mp u^{-1}  E_{ii}) ,
\\
\lim_{\la \to 0} K(u)|_{\mu = \pm\la^{-1}} = \del_{\ell1} (\mp u E_{-n,-n} + u^{2} E_{nn}) +  \sum_{1 \le i < \bar\ell} (E_{-i,-i} \mp u^{-1} E_{ii}), 
\\
\lim_{\la\to \infty} K(u) = \del_{r,n-1}(E_{-1,-1}+u^2 E_{11}) + \sum_{\bar r\le i \le n} ( u^2 E_{-i,-i} + E_{ii} ),
\\
\lim_{\la \to \infty} K(u)|_{\mu = \pm\la} = \del_{\ell1} (u^2 E_{-n,-n}  \mp u^{-1} E_{nn}) + \del_{r,n-1} (E_{-1,-1} \mp u E_{11}) + \sum_{\bar r \le i < \bar\ell} (\mp u E_{-i,-i} + E_{ii}), 
\\
\lim_{\la \to \infty} K(u)|_{\mu = \pm\la^{-1}} = \del_{r,n-1} (E_{-1,-1} \mp u E_{11}) + \sum_{\bar r \le i \le n} (\mp u E_{-i,-i} + E_{ii}). 
\\
\intertext{For $\ell=0$, $r\ge2$:}
\lim_{\mu \to 0} K(u) = \lim_{\mu \to \infty} K(u) = \Id . 
\\
\intertext{For $\ell=1$, $r\ge2$:}
\lim_{\mu \to 0} K(u) = \Id, \qq 
\lim_{\mu \to \infty} K(u) = \Id + (u^2-1) E_{-n,-n} + (u^{-2}-1) E_{nn} .
\end{gather*}
\end{description}

\begin{rmk} \label{R:D2:rotate}
K-matrices of type D.2ac also satisfy \eqref{Zpi:bijection}, see Remark \ref{R:C1:rotate}. \hfill \rmkend
\end{rmk}


\subsection{Untwisted K-matrices of type CD.4} \label{sec:K:I}

The CD.4 family consists of $(X,\pi)\in \Sat(A)$ with $A$ of type C$^{(1)}_n$ or D$^{(1)}_n$. 
All diagrams in this family are quasistandard, non-restrictable and $|I_{\rm diff} \cup I_{\rm nsf}|$ equals 1 (generic case) or 2 (special case).


\subsubsection{Family CD.4 (generic case)} \label{sec:CD4}

Satake diagrams in this family are parametrized by the tuple $(n,o_1,p_1)$. According to the value of $o_1$ we distinguish two subfamilies: CD.4a (with $o_1=0$) and CD.4b (with $o_1=1$).

Set $\ell=(n-o_1)/2-p_1 = |I^*| - 1$ so that $0\le \ell \le (n-o_1)/2 $. We additionally require $\ell\ne1$ for type D.4; the excluded case is the special case and is studied in Section \ref{sec:D4:spec} below. In terms of this parametrization we have $X = \{ \ell+1, \ldots, n-\ell-1 \}$, which, unless $\ell = n/2$, is of type ${\rm A}_{n-2\ell-1}$. 
We choose $I^*=\{0,1,\ldots,\ell\}$. The representative diagrams and special $\tau$-orbits are listed in Table~\ref{tab:CD4:gen}.

\begin{table}[h]
\caption{Family CD.4: representative Satake diagrams with $|I_{\rm diff} \cup I_{\rm nsf}|=1$.} \label{tab:CD4:gen}
\[
\arraycolsep=3pt\def\arraystretch{1.2}
\begin{array}{cccccc}
\text{Type} & \text{Name} & \text{Diagram} & \text{Restrictions} & I_{\rm diff} & I_{\rm ns} = I_{\rm nsf} \\ 
\hline 
\hline
\text{C.4a} & \bigl( {\rm C}_n^{(1)} \bigr)^\pi_{p_1} & 
\begin{tikzpicture}[baseline=-0.35em,line width=0.7pt,scale=0.8]
\draw[thick] (0,.4) -- (-.5,0) -- (0,-.4);
\draw[thick,dashed] (0,.4) -- (1,.4);
\draw[thick,dashed] (0,-.4) -- (1,-.4);
\draw[thick] (1,.4) -- (1.5,.4);
\draw[thick] (1,-.4) -- (1.5,-.4);
\draw[thick,dashed] (1.5,.4) -- (2.5,.4);
\draw[thick,dashed] (1.5,-.4) -- (2.5,-.4);
\draw[double,<-] (2.6,.4) -- (3,.4);
\draw[double,<-] (2.6,-.4) -- (3,-.4);
\draw[<->,gray] (3,.3) -- (3,-.3);
\draw[<->,gray] (2.5,.3) -- (2.5,-.3);
\draw[<->,gray] (1.5,.3) -- (1.5,-.3);
\draw[<->,gray] (1,.3) -- (1,-.3);
\draw[<->,gray] (0,.3) -- (0,-.3);
\filldraw[fill=black] (-.5,0) circle (.1);
\filldraw[fill=black] (0,.4) circle (.1) ;
\filldraw[fill=black] (0,-.4) circle (.1) ;
\filldraw[fill=black] (1,.4) circle (.1);
\filldraw[fill=black] (1,-.4) circle (.1);
\filldraw[fill=white] (1.5,.4) circle (.1) node[above=1pt]{\scriptsize $\ell$};
\filldraw[fill=white] (1.5,-.4) circle (.1) node[below=1pt] {\scriptsize $n\!\!-\!\!\ell$};
\filldraw[fill=white] (2.5,.4) circle (.1) node[above=1pt]{\scriptsize $1$};
\filldraw[fill=white] (2.5,-.4) circle (.1) node[below=1pt]{\scriptsize $n\!\!-\!\!1$};
\filldraw[fill=white] (3,.4) circle (.1) node[right=1pt]{\scriptsize $0$};
\filldraw[fill=white] (3,-.4) circle (.1) node[right=1pt]{\scriptsize $n$};
\draw[snake=brace] (-.6,.6) -- (1.1,.6) node[midway,above]{\scriptsize $p_1$};
\end{tikzpicture} 
& \begin{array}{c} n \text{ even} \\ 0 \le \ell < n/2 \end{array} & \{\ell \} & \emptyset
\\
\text{C.4a} & \bigl( {\rm C}_n^{(1)} \bigr)^\pi_0 & 
\begin{tikzpicture}[baseline=-0.35em,line width=0.7pt,scale=0.8]
\draw[thick] (0,.4) -- (-.5,0) -- (0,-.4);
\draw[thick,dashed] (0,.4) -- (1,.4);
\draw[thick,dashed] (0,-.4) -- (1,-.4);
\draw[double,<-] (1.1,.4) -- (1.5,.4);
\draw[double,<-] (1.1,-.4) -- (1.5,-.4);
\draw[<->,gray] (0,.3) -- (0,-.3);
\draw[<->,gray] (1,.3) -- (1,-.3);
\draw[<->,gray] (1.5,.3) -- (1.5,-.3);
\filldraw[fill=white] (-.5,0) circle (.1) node[left=1pt]{\scriptsize $n/2$};
\filldraw[fill=white] (0,.4) circle (.1);
\filldraw[fill=white] (0,-.4) circle (.1);
\filldraw[fill=white] (1,.4) circle (.1) node[above=1pt]{\scriptsize $1$};
\filldraw[fill=white] (1,-.4) circle (.1) node[below=1pt]{\scriptsize $n\!\!-\!\!1$};
\filldraw[fill=white] (1.5,.4) circle (.1) node[right=1pt]{\scriptsize $0$};
\filldraw[fill=white] (1.5,-.4) circle (.1) node[right=1pt]{\scriptsize $n$};
\end{tikzpicture} 
& \begin{array}{c} n \text{ even} \\ \ell = n/2 \end{array}  & \emptyset & \{n/2\}
\\
\hline
\text{C.4b} & \bigl( {\rm C}_n^{(1)} \bigr)^\pi_{p_1} 
& 
\begin{tikzpicture}[baseline=-0.35em,line width=0.7pt,scale=0.8]
\draw[thick,domain=90:270] plot({.4*cos(\x)},{.4*sin(\x)});
\draw[thick,dashed] (0,.4) -- (1,.4);
\draw[thick,dashed] (0,-.4) -- (1,-.4);
\draw[thick] (1,.4) -- (1.5,.4);
\draw[thick] (1,-.4) -- (1.5,-.4);
\draw[thick,dashed] (1.5,.4) -- (2.5,.4);
\draw[thick,dashed] (1.5,-.4) -- (2.5,-.4);
\draw[double,<-] (2.6,.4) -- (3,.4);
\draw[double,<-] (2.6,-.4) -- (3,-.4);
\draw[<->,gray] (0,.3) -- (0,-.3);
\draw[<->,gray] (1,.3) -- (1,-.3);
\draw[<->,gray] (1.5,.3) -- (1.5,-.3);
\draw[<->,gray] (2.5,.3) -- (2.5,-.3);
\draw[<->,gray] (3,.3) -- (3,-.3);
\filldraw[fill=black] (0,.4) circle (.1) ;
\filldraw[fill=black] (0,-.4) circle (.1) ;
\filldraw[fill=black] (1,.4) circle (.1);
\filldraw[fill=black] (1,-.4) circle (.1);
\filldraw[fill=white] (1.5,.4) circle (.1) node[above=1pt]{\scriptsize $\ell$};
\filldraw[fill=white] (1.5,-.4) circle (.1) node[below=1pt] {\scriptsize $n\!\!-\!\!\ell$};
\filldraw[fill=white] (2.5,.4) circle (.1) node[above=1pt]{\scriptsize $1$};
\filldraw[fill=white] (2.5,-.4) circle (.1) node[below=1pt]{\scriptsize $n\!\!-\!\!1$};
\filldraw[fill=white] (3,.4) circle (.1) node[right=1pt]{\scriptsize $0$};
\filldraw[fill=white] (3,-.4) circle (.1) node[right=1pt]{\scriptsize $n$};
\draw[snake=brace] (-.1,.6) -- (1.1,.6) node[midway,above]{\scriptsize $p_1$};
\end{tikzpicture} 
& \begin{array}{c} n \text{ odd} \\ 0 \le \ell < n/2 \end{array} & \{\ell\} & \emptyset
\\
\hline
\hline
\text{D.4a} & \bigl( {\rm D}^{(1)}_n \bigr)^\pi_{p_1} 
& 
\begin{tikzpicture}[baseline=-0.35em,line width=0.7pt,scale=0.8] 
\draw[] (0,1.5);
\draw[thick] (0,.4) -- (-.5,0) -- (0,-.4);
\draw[thick,dashed] (0,.4) -- (1,.4);
\draw[thick,dashed] (0,-.4) -- (1,-.4);
\draw[thick] (1,.4) -- (1.5,.4);
\draw[thick] (1,-.4) -- (1.5,-.4);
\draw[thick,dashed] (1.5,.4) -- (2.5,.4);
\draw[thick,dashed] (1.5,-.4) -- (2.5,-.4);
\draw[thick] (3.1,.7) -- (2.5,.4) -- (2.9,.2);
\draw[thick] (3.1,-.1) -- (2.5,-.4) -- (2.9,-.6);
\draw[<->,gray] (0,.3) -- (0,-.3);
\draw[<->,gray] (1,.3) -- (1,-.3);
\draw[<->,gray] (1.5,.3) -- (1.5,-.3);
\draw[<->,gray] (2.5,.3) -- (2.5,-.3);
\draw[<->,gray] (2.9,.1) -- (2.9,-.5);
\draw[<->,gray] (3.1,.6) -- (3.1,0);
\filldraw[fill=black] (-.5,0) circle (.1);
\filldraw[fill=black] (0,.4) circle (.1);
\filldraw[fill=black] (0,-.4) circle (.1);
\filldraw[fill=black] (1,.4) circle (.1);
\filldraw[fill=black] (1,-.4) circle (.1);
\filldraw[fill=white] (1.5,.4) circle (.1) node[above=1pt]{\scriptsize $\ell$};
\filldraw[fill=white] (1.5,-.4) circle (.1) node[below=.5pt]{\scriptsize $n\!\!-\!\!\ell$};
\filldraw[fill=white] (2.5,.4) circle (.1) node[above=.5pt]{\scriptsize $2$};
\filldraw[fill=white] (2.5,-.4) circle (.1) node[below]{\scriptsize $n\!\!-\!\!2\hspace{5pt} $};
\filldraw[fill=white] (2.9,.2) circle (.1) node[below=-3pt]{\hspace{-11pt} \scriptsize 1};
\filldraw[fill=white] (2.9,-.6) circle (.1) node[right=1pt]{\scriptsize $n\!\!-\!\!1$};
\filldraw[fill=white] (3.1,.7) circle (.1) node[right=1pt]{\scriptsize $0$};
\filldraw[fill=white] (3.1,-.1) circle (.1) node[right=1pt]{\scriptsize $n$};
\draw[snake=brace] (-.6,.6) -- (1.1,.6) node[midway,above]{\scriptsize $p_1$};
\end{tikzpicture} 
& \begin{array}{c} n \text{ even} \\ 0 \le \ell < n/2 \\ \ell \ne 1 \end{array} & \{\ell \} & \emptyset
\\[-1em]
\text{D.4a} & \bigl( {\rm D}^{(1)}_n \bigr)^\pi_0 
& 
\begin{tikzpicture}[baseline=-0.35em,line width=0.7pt,scale=0.8] 
\draw[] (0,1.5);
\draw[thick] (0,.4) -- (-.5,0) -- (0,-.4);
\draw[thick,dashed] (0,.4) -- (1,.4);
\draw[thick,dashed] (0,-.4) -- (1,-.4);
\draw[thick] (1.6,.7) -- (1,.4) -- (1.4,.2);
\draw[thick] (1.6,-.1) -- (1,-.4) -- (1.4,-.6);
\draw[<->,gray] (0,.3) -- (0,-.3);
\draw[<->,gray] (1,.3) -- (1,-.3);
\draw[<->,gray] (1.4,.1) -- (1.4,-.5);
\draw[<->,gray] (1.6,.6) -- (1.6,0);
\filldraw[fill=white] (-.5,0) circle (.1) node[left=1pt]{\scriptsize $n/2$};
\filldraw[fill=white] (0,.4) circle (.1);
\filldraw[fill=white] (0,-.4) circle (.1);
\filldraw[fill=white] (1,.4) circle (.1) node[above=.5pt]{\scriptsize $2$};
\filldraw[fill=white] (1,-.4) circle (.1) node[below]{\scriptsize $n\!\!-\!\!2\hspace{5pt} $};
\filldraw[fill=white] (1.4,.2) circle (.1) node[below=-3pt]{\hspace{-11pt} \scriptsize 1};
\filldraw[fill=white] (1.4,-.6) circle (.1) node[right=1pt]{\scriptsize $n\!\!-\!\!1$};
\filldraw[fill=white] (1.6,.7) circle (.1) node[right=1pt]{\scriptsize $0$};
\filldraw[fill=white] (1.6,-.1) circle (.1) node[right=1pt]{\scriptsize $n$};
\end{tikzpicture} 
& \begin{array}{c} n \text{ even} \\ \ell = n/2 \end{array} & \emptyset & \{ n/2 \}
\\
\hline
\text{D.4b} & \bigl( {\rm D}^{(1)}_n \bigr)^\pi_{p_1} 
&
\hspace{14pt}
\begin{tikzpicture}[baseline=-0.35em,line width=0.7pt,scale=0.8]
\draw[thick,domain=90:270] plot({.4*cos(\x)},{.4*sin(\x)});
\draw[thick,dashed] (0,.4) -- (1,.4);
\draw[thick,dashed] (0,-.4) -- (1,-.4);
\draw[thick] (1,.4) -- (1.5,.4);
\draw[thick] (1,-.4) -- (1.5,-.4);
\draw[thick,dashed] (1.5,.4) -- (2.5,.4);
\draw[thick,dashed] (1.5,-.4) -- (2.5,-.4);
\draw[thick] (3.1,.7) -- (2.5,.4) -- (2.9,.2);
\draw[thick] (3.1,-.1) -- (2.5,-.4) -- (2.9,-.6);
\draw[<->,gray] (0,.3) -- (0,-.3);
\draw[<->,gray] (1,.3) -- (1,-.3);
\draw[<->,gray] (1.5,.3) -- (1.5,-.3);
\draw[<->,gray] (2.5,.3) -- (2.5,-.3);
\draw[<->,gray] (2.9,.1) -- (2.9,-.5);
\draw[<->,gray] (3.1,.6) -- (3.1,0);
\filldraw[fill=black] (0,.4) circle (.1);
\filldraw[fill=black] (0,-.4) circle (.1);
\filldraw[fill=black] (1,.4) circle (.1);
\filldraw[fill=black] (1,-.4) circle (.1);
\filldraw[fill=white] (1.5,.4) circle (.1) node[above=1pt]{\scriptsize $\ell$};
\filldraw[fill=white] (1.5,-.4) circle (.1) node[below=.5pt]{\scriptsize $n\!\!-\!\!\ell$};
\filldraw[fill=white] (2.5,.4) circle (.1) node[above=.5pt]{\scriptsize $2$};
\filldraw[fill=white] (2.5,-.4) circle (.1) node[below]{\scriptsize $n\!\!-\!\!2\hspace{5pt} $};
\filldraw[fill=white] (2.9,.2) circle (.1) node[below=-3pt]{\hspace{-11pt} \scriptsize 1};
\filldraw[fill=white] (2.9,-.6) circle (.1) node[right=1pt]{\scriptsize $n\!\!-\!\!1$};
\filldraw[fill=white] (3.1,.7) circle (.1) node[right=1pt]{\scriptsize $0$};
\filldraw[fill=white] (3.1,-.1) circle (.1) node[right=1pt]{\scriptsize $n$};
\draw[snake=brace] (-.1,.6) -- (1.1,.6) node[midway,above]{\scriptsize $p_1$};
\end{tikzpicture}
& \begin{array}{c} n \text{ odd} \\ 0 \le \ell < n/2 \\ \ell \ne 1 \end{array} & \{\ell\} & \emptyset 
\\
\hline
\end{array}
\]
\end{table}


{\allowdisplaybreaks

The QP algebra is generated by $x_i$, $y_i$, $k_i$ with $i\in X$ and $(k_j k_{n-j}^{-1})^{\pm1}$ with $j\in I^*$ and elements $b_j$ whose reduced expressions are 
\begin{alignat*}{99}
b_{\ell} &= 
\casesl{l}{ 
y_\ell - c_\ell \, T_{w_X} (x_{n-\ell})\, k_\ell^{-1} \\[.2em]
y_\ell - c_\ell \, x_\ell k_\ell^{-1} - s_\ell \,k_\ell^{-1}
}
& \qu & \casesm{l}{
\text{if } \ell<n/2 , \\[.2em]
\text{if } \ell=n/2 ,
}
\\
b_{n-\ell} &= y_{n-\ell} - c_{n-\ell}\, T_{w_X} (x_{\ell})\, k_{n-\ell}^{-1} && \text{if } \ell<n/2, 
\\
b_{j} &= y_{j} - c_{j}\, x_{n-j}\, k_{j}^{-1} && \text{if } j < \ell \text{ or } j > n\!-\!\ell,
\intertext{where $c_j=c_{n-j}$ if $j\notin I_{\rm diff}$. Next, in terms of the effective dressing parameters $\om_1,\ldots,\om_\ell$, the scaling parameter $\eta$, an additional free parameter $\la$ and $\mu=q^{-n+2\ell} \la$, we set}
c_\ell &= \casesl{l}{
q^{-\vartheta}\, \eta^{-1} \la^{2} \\  
\la\,\om_\ell^{-1}\\
- q^{-1} \om_\ell^{-2} \\
}
&&
\casesm{l}{
\text{if } \ell=0 , \\
\text{if } \frac{1+\vartheta}2 < \ell <n/2 , \\
\text{if } \ell=n/2 ,
} \\
c_{n-\ell} &= \casesl{l}{  
q^{-\vartheta}\,\eta^{-1} \mu^{-2} \\
(-1)^n\mu^{-1} \om_\ell^{-1}
}
&&
\casesm{l}{
\text{if } \ell=0 , \\
\text{if } \frac{1+\vartheta}2 < \ell <n/2, 
}
\\
c_j &= \casesl{l}{  
\eta^{-1} \om_1^{\frac{3-\vartheta}2} \om_2^{\frac{1+\vartheta}2} \\
- \om_{j}^{-1} \om_{j+1}
}
&&
\casesm{l}{
\text{if } 0 = j < \ell, \\
\text{if } 0 < j < \ell, 
} \\
s_\ell &= \frac{\mu - \mu^{-1}}{q - q^{-1}} \, \om_\ell^{-1} && \text{if } \ell=n/2 ,
\end{alignat*}
where $\vartheta=-1$ for type C and $\vartheta=1$ for type D, see \eqref{QQv}.
Then, solving the boundary intertwining equation \eqref{intw-untw} for all generators of the QP algebra we obtain the following result.
}

\begin{result} \label{Res:CD4}
The bare K-matrix of type CD.4, when $|I_{\rm diff}\cup I_{\rm nsf}|=1$, is of the form \eqref{K(u):X} with $M_1(u) = \sum_{1 \le  i \le n}  E_{ii}$ and
\eq{ \label{D2:CD4}
M_2(u) = \sum_{\bar \ell \le i \le n} \!\Big( \la E_{-i,-i} + \la^{-1} E_{-\bar \imath,-\bar \imath} + E_{-i,-\bar\imath} + E_{-\bar\imath,-i} - u^{-1}(\mu\,E_{ii} + \mu^{-1} E_{\bar \imath \bar \imath} + E_{i\bar\imath} + E_{\bar\imath i}) \Big), 
}
where $\mu = q^{-n+2\ell}\la$ and $\la\in \K^\times$. 
\end{result}

This K-matrix has the following properties.

\smallskip

\begin{description}[itemsep=1ex]

\item[Eigendecomposition] \hfill  $V= \Id + \sum_{1 \le i \le \ell } \big( \la ( E_{-\bar \imath,-i}-E_{-i,-\bar \imath} ) + \mu ( E_{\bar \imath,i}-E_{i,\bar \imath} ) \big),$ \hfill   \hphantom{\it Eigendecomposition}
\eqn{ 
D(u) &= \displaystyle \sum_{1 \le i \le \ell} \Big( h_1(u) h_2(u)  E_{-i,-i}+ u^{-2}h_2(u) E_{ii} \Big) + \sum_{\ell <i \le n}\Big( E_{-i,-i}+ h_1(u) E_{ii} \Big).
}

\item[Rotations] \hfill $ 
\begin{aligned}[t]
\qq K^\pi(u) &= u^2 \,h_1(u)\, K(u)|_{\la \to q^{n-2\ell} \la^{-1}} , \\
K^{\flLR}(u) &= K(u) \qq \text{for D.4 if } \ell>0.
\end{aligned}
$ \hfill \hphantom{\it Rotations}

\item[Reductions] 

For $\ell = n/2$ it is independent of $q$, has $d_{\rm eff}=2$ and is singly regular, $K(-1) \ne \Id$. Setting $\la = \pm 1$ it becomes a nonregular \gim~with $d_{\rm eff}=1$. 

\noindent For $\ell<n/2$ and $\la = \pm q^{n/2-\ell}$ or $\la = \pm \sqrt{-1}\, q^{n/2-\ell}$ it has $d_{\rm eff}=2$ and is singly regular, $K(1) \ne \Id$ or $K(-1) \ne \Id$, respectively.

\item[Diagonal cases] For $\ell>0$ it has the following diagonal limits, both with $d_{\rm eff}=3$:
\eqn{ 
\qq \lim_{\la \to 0} K(u) &= u^{-2} h_2(u) \sum_{1 \le i \le \ell} (E_{-i,-i}+E_{ii}) + \sum_{\ell < i \le n} (E_{-i,-i} + u^{-2} E_{ii}), \\ 
\qq \lim_{\la \to \infty} K(u) &= \sum_{1 \le i < \bar \ell} (E_{-i,-i}+E_{ii}) + h_2(u) \sum_{\bar \ell \le i \le n} ( E_{-i,-i} + u^{-2} E_{ii} ).
}
For $\ell=0$ it is diagonal and has $d_{\rm eff}=2$. By writing $\xi = \la \, \mu = q^{-n} \la^2$, it is independent of $q$,
\[
\qq K (u) = \sum_{i=1}^n \left( E_{-i,-i} + \frac{\xi - u^{-1}}{\xi - u} E_{ii} \right),
\]
and has the half-period symmetry $K(-u)=K(u)|_{\xi \to -\xi}$.
\end{description}


\subsubsection{Family D.4 (special case)} \label{sec:D4:spec}

The diagrams with $|I_{\rm diff}\cup I_{\rm nsf}|=2$ are parametrized by the tuple $(n,o_1,p_1)$ with $p_1=(n-o_1)/2-1$, so that $\ell = 1$ and $X=\{2,\ldots,n-2\}$. We choose $I^*=\{0,1\}$. The representative Satake diagrams and the special $\tau$-orbits are listed in Table \ref{tab:D4:spec}.  In both cases $I_{\rm ns} = I_{\rm nsf} = \emptyset$.

The QP algebra is generated by $x_i$, $y_i$, $k_i$ with $i\in X$ and $(k_j k_{n-j}^{-1})^{\pm 1}$ with $j\in I^*$ and elements $b_j$ whose reduced expressions are
\[
b_j = y_j - c_j T_{w_X}(x_{n-j})\, k_j^{-1} , \qu b_{n-j} = y_{n-j} - c_{n-j}\, T_{w_X} (x_{j})\, k_{n-j}^{-1} \qu\text{with } j\in I^*,
\]
and all $c_j$, $c_{n-j}$ independent. 
Next, in terms of the effective dressing parameter $\om_1$, the scaling parameter $\eta$ and the additional free parameters $\al$ and $\la$ and with  $\mu=q^{-n+2} \al \, \la$ we set
\[
c_0 = \eta^{-1} \la \, \om_1, \qu\; c_1 = \al \, \la \, \om_1^{-1}, \qu\; c_{n-1} = (-1)^n(\mu \, \om_1)^{-1}, \qu\; c_n = (-1)^n(\eta \, \mu)^{-1} \al \, \om_1.
\]
Repeating the same steps as before we obtain the following result.

\begin{table}[h]
\caption{Family D.4: representative Satake diagrams with $|I_{\rm diff} \cup I_{\rm nsf}|=2$.} \label{tab:D4:spec}
\[
\arraycolsep=3pt
\begin{array}{ccccc}
\text{Type} & \text{Name} & \text{Diagram} & \text{Restrictions} & I_{\rm diff} \\ 
\hline \hline 
\text{D.4a} & \bigl( {\rm D}^{(1)}_n \bigr)^\pi_{n/2-1} & 
\begin{tikzpicture}[baseline=-0.35em,line width=0.7pt,scale=0.8]
\draw[thick] (0,.4) -- (-.5,0) -- (0,-.4);
\draw[thick,dashed] (0,.4) -- (1,.4);
\draw[thick,dashed] (0,-.4) -- (1,-.4);
\draw[thick] (1.6,.7) -- (1,.4) -- (1.4,.2);
\draw[thick] (1.6,-.1) -- (1,-.4) -- (1.4,-.6);
\draw[<->,gray] (0,.3) -- (0,-.3);
\draw[<->,gray] (1,.3) -- (1,-.3);
\draw[<->,gray] (1.4,.1) -- (1.4,-.5);
\draw[<->,gray] (1.6,.6) -- (1.6,0);
\filldraw[fill=black] (-.5,0) circle (.1);
\filldraw[fill=black] (0,.4) circle (.1);
\filldraw[fill=black] (0,-.4) circle (.1);
\filldraw[fill=black] (1,.4) circle (.1);
\filldraw[fill=black] (1,-.4) circle (.1);
\filldraw[fill=white] (1.4,.2) circle (.1) node[below=-3pt]{\hspace{-11pt} \scriptsize 1};
\filldraw[fill=white] (1.4,-.6) circle (.1) node[right=1pt]{\scriptsize $n\!\!-\!\!1$};
\filldraw[fill=white] (1.6,.7) circle (.1) node[right=1pt]{\scriptsize $0$};
\filldraw[fill=white] (1.6,-.1) circle (.1) node[right=1pt]{\scriptsize $n$};
\end{tikzpicture}  
& \begin{array}{c} n \text{ even} \\ \ell = 1 \end{array} & \{0,1\}
\\
\text{D.4b} & \bigl( {\rm D}^{(1)}_n \bigr)^\pi_{(n-3)/2} & 
\hspace{4pt} \begin{tikzpicture}[baseline=-0.35em,line width=0.7pt,scale=0.8]
\draw[thick,domain=90:270] plot({.4*cos(\x)},{.4*sin(\x)});
\draw[thick,dashed] (0,.4) -- (1,.4);
\draw[thick,dashed] (0,-.4) -- (1,-.4);
\draw[thick] (1.6,.7) -- (1,.4) -- (1.4,.2);
\draw[thick] (1.6,-.1) -- (1,-.4) -- (1.4,-.6);
\filldraw[fill=black] (0,.4) circle (.1);
\filldraw[fill=black] (0,-.4) circle (.1);
\draw[<->,gray] (0,.3) -- (0,-.3);
\filldraw[fill=black] (1,.4) circle (.1);
\filldraw[fill=black] (1,-.4) circle (.1);
\draw[<->,gray] (1,.3) -- (1,-.3);
\filldraw[fill=white] (1.4,.2) circle (.1) node[below=-3pt]{\hspace{-11pt} \scriptsize 1};
\filldraw[fill=white] (1.4,-.6) circle (.1) node[right=1pt]{\scriptsize $n\!\!-\!\!1$};
\draw[<->,gray] (1.4,.1) -- (1.4,-.5);
\filldraw[fill=white] (1.6,.7) circle (.1) node[right=1pt]{\scriptsize $0$};
\filldraw[fill=white] (1.6,-.1) circle (.1) node[right=1pt]{\scriptsize $n$};
\draw[<->,gray] (1.6,.6) -- (1.6,0);
\end{tikzpicture}
& \begin{array}{c} n \text{ odd} \\ \ell=1 \end{array} & \{0,1\} 
\\ \hline 
\end{array}
\]
\end{table}

\begin{result}
The bare K-matrix of type D.4, when $|I_{\rm diff}\cup I_{\rm nsf}|=2$, is
\eq{
K(u) = \Id + \frac{u-u^{-1}}{ k_1(u)} \bigg( M_1 - \frac{M^+_{2}(u)}{k^+_{2}(u)} + \frac{M^-_2}{ k^-_2(u)} \bigg),  
\label{D4:K:spec}
}
where $k_1(u)$ and $k^+_2(u) = k_2(u)$ are as in \eqref{k_i(u)}, $M_1 = \sum_{1 \le i \le n} E_{ii}$ is as in Result \ref{Res:CD4}, and the remaining terms are $\al$-deformations of the matrix $M_2(u)$ given by \eqref{D2:CD4}:
\eqn{ 
M^+_2(u) &= u^{-1} (\mu^{-1} E_{11} + \mu \, E_{nn}+E_{1n}+E_{n1}), \qu k^-_2(u) = (\al \, \la)^{-1} + \al (\mu \, u)^{-1} ,\\
M^-_2 &= (\al \, \la)^{-1} E_{-1,-1} + \al \, \la \, E_{-n,-n}+E_{-1,-n}+E_{-n,-1}.
}
Here $\mu = q^{-n+2} \al \, \la$ with $\al, \la \in \K^\times$.
\end{result}

Setting $\al=1$ the K-matrix defined above specializes to the one defined in Result \ref{Res:CD4} when $\ell=1$. 
The K-matrix defined in \eqref{D4:K:spec} has the following properties.

\smallskip

\enlargethispage{1em}

\begin{description}[itemsep=1ex]

\item[Eigendecomposition] \hfill  $V= \Id + \al \la ( E_{-n,-i}-E_{-i,-n} ) + \mu ( E_{n,1}-E_{1,n} ),$ \hfill   \hphantom{\it Eigendecomposition}
\eqn{ 
D(u) &= \displaystyle h_1(u)\, h^-_2(u) E_{-i,-i}+ u^{-2} h^+_2(u) E_{ii}+ \sum_{1 <i \le n} \Big( E_{-i,-i}+ h_1(u) E_{ii} \Big).
}

\item[Half-period] \hfill $K(-u) = K(u) \big|_{\al \to -\al,\; \la \to -\la}$. \hfill \hphantom{\it Half-period} \\[-.5em]

\item[Rotations] \hfill $ 
\begin{aligned}[t]
K^{\pi}(u) &= u^2 h_1(u) K(u) \big|_{\al \to \al^{-1} ,\; \la \to - (-q)^{n-2} \la^{-1}} , \hspace{-5mm}  \\
K^{\flLR}(u) &= K(u) \big|_{\al \to \al^{-1} ,\; \la \to \al \, \la}.
\end{aligned}
$ \hfill \hphantom{\it Rotations}

\item[Reductions] For $\la \in \{ \mu, -\mu^{-1}, \al^{-2} \mu \}$ or $\la \in \{ -\mu, \mu^{-1}, -\al^{-2} \mu \}$ it has $d_{\rm eff}=3$ and is singly regular, $K(-1) \ne \Id$ or $K(1) \ne \Id$, respectively. 

\item[Diagonal cases] \hfill $
\begin{aligned}[t]
\lim_{\al \to 0} K(u) &= \sum_{1 <i \le n} (E_{-i,-i}+u^{-2} E_{ii}) + u^{-2} E_{-1,-1} + E_{11}, \\
\lim_{\al \to \infty} K(u) &= \sum_{1 \le i < n} (E_{-i,-i}+E_{ii}) + u^2 E_{-n,-n} + u^{-2} E_{nn}, 
\end{aligned}$ \hfill \hphantom{\it Diagonal cases} \\
\hfill $\begin{aligned}[t]
\lim_{\la \to 0} K(u) &= \sum_{1 <i \le n} (E_{-i,-i} + u^{-2} E_{ii}) + \frac{q^{n-2} \al + u^{-1}}{q^{n-2} \al + u} E_{-1,-1} + \frac{q^{n-2} \al^{-1} + u^{-1}}{q^{n-2} \al^{-1} + u} E_{11} , \\
\lim_{\la \to \infty} K(u) &= \sum_{1 \le i < n} (E_{-i,-i} + E_{ii}) + \frac{1+q^{n-2} \al \, u}{1+q^{n-2} \al \, u^{-1}} E_{-n,-n} + \frac{q^{n-2} \al^{-1} + u^{-1}}{q^{n-2} \al^{-1} +u} E_{nn}, 
\end{aligned}$ \hfill \\
\[ \lim_{\al \to 0,\la \to \infty \atop \al \, \la \text{ fixed}} K(u) =\Id, \qq \lim_{\al \to \infty,\la \to 0 \atop \al \, \la \text{ fixed}} K(u) = \sum_{1 \le i \le n} (E_{-i,-i} + u^{-2} E_{ii}).\]

\end{description}


\subsection{Untwisted K-matrices of types C.2 and BD.1} \label{sec:K:0}

The C.2 family consists of Satake diagrams $(X,\id)\in \Sat(A)$ with $A$ of type C$^{(1)}_n$. All diagrams in this family are quasistandard and have $|I_{\rm diff}\cup I_{\rm nsf}|=0$. 
Note that diagrams with $|I\backslash X|=1$ are also members of the C.1 family and were studied in Section~\ref{sec:C1}. 

The BD.1 family consist of Satake diagrams $(X,\tau)\in \Sat(A)$ with $\tau\in\{\id,\phi_1,\phi_1,\phi_{12}\}$ and $A$ of type B$^{(1)}_n$ or D$^{(1)}_n$. 
This family has both quasistandard and non-quasistandard Satake diagrams. 
The quasistandard diagrams have $|I_{\rm nsf}|=0$ and $|I_{\rm diff}|$ equal to 0 (general case) or 1 (special case). 
The special case only occurs if $|I^*|=1$ and the corresponding diagrams are also members of the BD.2 family and were studied in Section~\ref{sec:BD2}.
The non-quasistandard diagrams have $|I_{\rm diff}|=0$ and $|I_{\rm nsf}|\in\{1,2\}$ and will be studied in Section \ref{sec:nqs}.


\subsubsection{Family C.2} \label{sec:C2}

Satake diagrams in this family are parametrized by the triple $(n,p_1,p_2)$ with $n\ge2$. 
By rotating with $\pi\in\Sigma_A$, if necessary, we may assume that the affine node satisfies $0 \in X \Rightarrow 0 \in X_1$. With this choice, the diagrams are restrictable precisely if $p_1=0$. 

Set $\ell=p_1$ and $r=n-p_2$, so that $\ell$ and $r$ are bounded by $0\le \ell \le r \le n-\ell$ and $r-\ell$ is even.  We also have $|I^*| -1 = (r-\ell)/2$ and $X_1 = \{0,\ldots,\ell-1\}$, $X_2 = \{ r+1,\ldots, n\}$ and $X_{\rm alt} = \{\ell+1,\ell+3,\ldots,r-3,r-1\}$. Thus $X_1$ is of type ${\rm C}_\ell$ (empty if $\ell=0$), $X_2$ is of type ${\rm C}_{n-r}$ (empty if $r=n$) and $X_{\rm alt}$ is of type ${\rm A}_1\!{}^{\times \frac{r-\ell}{2}}$ (empty if $\ell=r$). 
Note that $I^*=I\backslash X$ and there are no special $\tau$-orbits. The representative diagrams are listed in Table \ref{tab:C2}.

The QP algebra is generated by the elements $x_i$, $y_i$, $k_i$ with $i\in X$ and elements $b_j$ whose reduced expressions~are
\begin{alignat*}{99}
b_\ell &= \casesl{l}{ 
y_\ell - c_r\, T_{w_{X_1}} T_{r+1} (x_\ell )\, k_\ell^{-1} \\[.2em]
y_\ell - c_r\, T_{w_{X_1}} T_{w_{X_2}} (x_\ell )\, k_\ell^{-1} \\
}
& \qu & \casesm{l}{ \text{if } \ell<r, \\[.2em] \text{if } \ell=r,}
\\
b_r &=  y_r - c_r\, T_{r-1} T_{w_{X_2}} (x_r )\, k_r^{-1} && \text{if } \ell<r, \\
b_j &=  y_{j} - c_{j}\, T_{j-1}T_{j+1} (x_{j} )\, k_j^{-1} && \text{if } \ell < j < r \text{ and } j-\ell \text{ even} , 
\intertext{with $c_j$ all independent. Next, in terms of the effective dressing parameters $\om_{\ell+2}, \om_{\ell+4}, \ldots, \om_{r}$ and the scaling parameter $\eta$ we set}
c_\ell &= 
\casesl{l}{  
q^{2(n+1)} \eta^{-2} \\
q\, \eta^{-2} \om_{\ell+2}^{2}  \\ 
-(-q)^{n+1} \eta^{-1} \\ 
(-q)^{\ell+1} \eta^{-1} \om_{\ell+2}
} \hspace{-1mm} &&
\casesm{l}{ 
\text{if } 0=\ell=r, \\
\text{if } 0=\ell<r, \\ 
\text{if } 0<\ell=r<n, \\ 
\text{if } 0<\ell<r<n, 
} 
\\
c_r &= 
\casesl{l}{ 
(-q)^{\bar{r}} \om_r^{-1} \\ 
q\, \om_{n}^{-2}
} \hspace{-1mm} &&
\casesm{l}{
\text{if } \ell < r < n , \\ 
\text{if } \ell < r = n, 
} \\
c_j &= q\,\om_{j}^{-1}\om_{j+2} && \text{if } \ell < j < r \text{ and } j-\ell \text{ even}. \hspace{-.2cm}
\end{alignat*}
Then, solving the boundary intertwining equation \eqref{intw-untw} for all generators of the QP algebra we obtain the following result.

\begin{table}[h]
\caption{Family C.2: representative Satake diagrams.} \label{tab:C2}
\[
\arraycolsep=3pt
\begin{array}{ccccc}
\text{Type} & \text{Name} & \text{Diagram} & \text{Restrictions} \\ 
\hline \hline 
\text{C.2} & \bigl( {\rm C}^{(1)}_n \bigr)^\id_{p_1,{\rm alt},p_2} & 
\begin{tikzpicture}[baseline=-0.25em,line width=0.7pt,scale=0.8]
\draw[double,<-] (-.1,0) -- (-.5,0);
\draw[thick,dashed] (0,0) -- (1,0);
\draw[thick] (1,0) -- (2.5,0);
\draw[thick,dashed] (2.5,0) -- (3.5,0);
\draw[thick] (3.5,0) -- (4.5,0);
\draw[thick,dashed] (4.5,0) -- (5.5,0);
\draw[double,<-] (5.6,0) --  (6,0);
\filldraw[fill=black] (-.5,0) circle (.1) node[left=1pt]{\scriptsize $0$};
\filldraw[fill=black] (0,0) circle (.1) node[above=1pt]{\scriptsize $1$};
\filldraw[fill=black] (1,0) circle (.1); 
\filldraw[fill=white] (1.5,0) circle (.1) node[above=1pt]{\scriptsize $\ell$};
\filldraw[fill=black] (2,0) circle (.1);
\filldraw[fill=white] (2.5,0) circle (.1);
\filldraw[fill=black] (3.5,0) circle (.1);
\filldraw[fill=white] (4,0) circle (.1) node[above=1pt]{\scriptsize $r$};
\filldraw[fill=black] (4.5,0) circle (.1);
\filldraw[fill=black] (5.5,0) circle (.1) node[above]{\scriptsize $n\!-\!1$};
\filldraw[fill=black] (6,0) circle (.1) node[right=1pt]{\scriptsize $n$};
\draw[snake=brace] (1.1,-.2) -- (-0.6,-.2) node[midway,below]{\scriptsize $p_1$};
\draw[snake=brace] (6.1,-.2) -- (4.4,-.2) node[midway,below]{\scriptsize $p_2$};
\end{tikzpicture}  
& \begin{array}{c} r-\ell \text{ even} \\ 0<\ell\le r\le n-\ell \end{array} 
\\
\text{C.2} & \bigl( {\rm C}^{(1)}_n \bigr)^\id_{0,{\rm alt},p_2} & 
\begin{tikzpicture}[baseline=-0.25em,line width=0.7pt,scale=0.8]
\draw[double,->] (1.6,0) -- (2,0);
\draw[thick] (2,0) -- (2.5,0);
\draw[thick,dashed] (2.5,0) -- (3.5,0);
\draw[thick] (3.5,0) -- (4.5,0);
\draw[thick,dashed] (4.5,0) -- (5.5,0);
\draw[double,<-] (5.6,0) --  (6,0);
\filldraw[fill=white] (1.5,0) circle (.1) node[left=1pt]{\scriptsize $0=\ell$};
\filldraw[fill=black] (2,0) circle (.1);
\filldraw[fill=white] (2.5,0) circle (.1);
\filldraw[fill=black] (3.5,0) circle (.1);
\filldraw[fill=white] (4,0) circle (.1) node[above=1pt]{\scriptsize $r$};
\filldraw[fill=black] (4.5,0) circle (.1);
\filldraw[fill=black] (5.5,0) circle (.1) node[above]{\scriptsize $n\!-\!1$};
\filldraw[fill=black] (6,0) circle (.1) node[right=1pt]{\scriptsize $n$};
\draw[snake=brace] (6.1,-.2) -- (4.4,-.2) node[midway,below]{\scriptsize $p_2$};
\end{tikzpicture} 
& \begin{array}{c} r \text{ even} \\ 0\le r<n \end{array} 
\\ 
\text{C.2} & \bigl( {\rm C}^{(1)}_n \bigr)^\id_{0,{\rm alt},0} & 
\begin{tikzpicture}[baseline=-0.25em,line width=0.7pt,scale=0.8]
\draw[double,->] (1.6,0) -- (2,0);
\draw[thick] (2,0) -- (2.5,0);
\draw[thick,dashed] (2.5,0) -- (3.5,0);
\draw[double,<-] (3.6,0) --  (4,0);
\filldraw[fill=white] (1.5,0) circle (.1) node[left=1pt]{\scriptsize $0=\ell$};
\filldraw[fill=black] (2,0) circle (.1);
\filldraw[fill=white] (2.5,0) circle (.1);
\filldraw[fill=black] (3.5,0) circle (.1);
\filldraw[fill=white] (4,0) circle (.1) node[right=1pt]{\scriptsize $r=n$};
\draw[] (2,.5) -- (2,.5);
\draw[] (2,-.5) -- (2,-.5);
\end{tikzpicture} 
& n \text{ even}  
\\ \hline 
\end{array}
\]
\end{table}

\begin{result} \label{Res:C2}
The bare K-matrix of type C.2 is of the form \eqref{K(u):X} with
\eq{
\begin{aligned}
M_1(u) &= \sum_{\bar \ell \le i \le n} (\la \, \mu \, u \, E_{-i,-i} + E_{i,i}) , \\
M_2(u) &= \sum_{\bar r \le i < \bar \ell} \big(\la \, E_{-i,-i} + \la^{-1} E_{i,i} + \eps_i\,(E_{i,-i-\eps_i} + E_{-i-\eps_i,i})\big) , 
\end{aligned}\!\! \label{C2:K:parts}
}
where $\la = q^{N/2-r}$, $\mu = q^{-\ell}$ and $\eps_i=(-1)^{i+\bar \ell}$.
\end{result}

This K-matrix has the following properties.

\smallskip

\begin{description}[itemsep=1ex]

\item[Eigendecomposition] \hfill  $V= \Id + \la \sum_{\bar r \le i < \bar \ell} \eps_i ( E_{-i-\eps_i,i} - E_{i,-i-\eps_i} ),$ \hfill   \hphantom{\it Eigendecomposition}
\eqn{ 
D(u) &= \displaystyle \sum_{|i| < \bar r} \Big( E_{-i,-i}+E_{ii} \Big) + \sum_{\bar r \le i < \bar \ell}\! \Big( E_{-i,-i}+ h_1(u)\, h_2(u) E_{ii} \Big) + h_1(u) \!\sum_{\bar \ell \le  i \le n}\! \Big( u^2 E_{-i,-i} + E_{ii} \Big).
}

\item[Affinization] For $\ell=0$ the Satake diagram is restrictable and the affinization identity is
\[
K(u) =\frac{ \la u^{-1} K_0 - \la^{-1} u K_0^{-1} }{ \la u^{-1} - \la^{-1} u} ,
\]
where, upon multiplying by $C$ and dressing, $K_0$ corresponds to the CII solution of the constant twisted RE reported in \cite[Sec.~3]{NoSu}. In this case we also have a half-period symmetry: $K(-u)=K(u)$.

\item[Rotations] \; For $n$ even and $\ell+r=n$ we have $K^{\pi}(u)  = -u \, K(u)$.

\item[Bar-symmetry] $K(u)^{-1}  = G(\bm\om) J K(u)|_{\la \to \la^{-1},\, \mu \to \mu^{-1}} J G(\bm\om)$ where $\bm \om = (-1,\ldots,-1) \in \K^n$.

\item[Reductions] 

For $n$ even and $\ell+r=n$, when $0<\ell<n/2$, it has $d_{\rm eff}=3$ and is singly regular, $K(1)\ne \Id$; when $\ell=0$ it is a non-regular \gim:
\eq{ \label{C2:K:cst}
K(u) = - \sum_{1 \le i \le n} \eps_i (E_{i,-i-\eps_i} + E_{-i-\eps_i,i}).
}
\item[Diagonal cases] For $\ell=r$ it coincides with the one of type C.1 with $\ell=r$.

\end{description}

\begin{rmk} \label{R:C2:rotate}
K-matrices of type C.2 also satisfy \eqref{Zpi:bijection}, see Remark \ref{R:C1:rotate}. \hfill\rmkend
\end{rmk}


\subsubsection{Family BD.1 (generic case)} \label{sec:BD1}

Satake diagrams in this family are parametrized by the tuple $(N,p_1,p_2,o_1,o_2)$ with $N\ge7$. According to the value of $o_1+o_2$ we may distinguish two subfamilies if $N$ is odd: B.1a ($o_1+o_2=0$) and B.1b ($o_1=1,o_2=0$), and three subfamilies if $N$ is even: D.1a ($o_1+o_2=0$), D.1b ($o_1+o_2=1$) and D.1c ($o_1+o_2=2$). 
Recall that $p_i$, the number of $\tau$-orbits in $X_i=X_{Y_i}$, is even unless $N$ is odd and $i=2$. Rotating with a suitable $\si \in \Sigma_A$, if necessary, we may assume that the affine node satisfies $0 \in X \Rightarrow 0 \in X_1$. With this choice, the diagrams are restrictable if $o_1 = p_1 = 0$. 

Set $\ell=p_1 + o_1$ and $r=n-p_2-o_2$, so that $\ell$ and $r$ are bounded by $0 \le \ell \le r \le n$ and, if $N$ is even and $o_1=o_2$, in addition by $\ell+r \le n$. We restrict to the generic case, $|I_{\rm nsf}\cup I_{\rm diff}|=0$, hence $(\ell,r) \notin \{ (1,1), (\lfloor N/2\rfloor\!-\!1,\lceil N/2\rceil \!-\!1) \}$ and $r\ne\ell+2$.
In terms of this parametrization we have $|I^*| -1 = r - \ell$ and $X_1 = \{0,\ldots,\ell-1\}$ if $\ell>1$ and $X_2 = \{ r+1,\ldots,n\}$ if $N$ is odd or $r<n-1$, so that $X_1$ is of type ${\rm D}_\ell$ and $X_2$ is of type ${\rm B}_{n-r}$ if $N$ is odd and of type ${\rm D}_{n-r}$ if $N$ is even; otherwise respectively $X_1=\emptyset$ and $X_2 =\emptyset$. We choose $I^* = \{\ell,\ell+1,\ldots,r\}$. 
The representative diagrams are listed in Table \ref{tab:BD1:gen}. They have $I_{\rm ns} =\{ \ell+\delta_{\ell \ne 0},\ell+\delta_{\ell \ne 0}+1,\ldots, r-\delta_{r \ne n} \}$.

{
\arraycolsep=2pt\def\arraystretch{1.2}
\begin{table}[h]
\caption{Family BD.1: representative Satake diagrams with $|I_{\rm nsf} \cup I_{\rm diff}| = 0$. } \label{tab:BD1:gen}
\[
\begin{array}{cccccc}
\rm Type & \rm Name & \rm Diagram & \!\! (o_1,o_2)\!\! & \rm Restrictions \\ 
\hline\hline 
\text{B.1a} & \bigl( {\rm B}_n^{(1)} \bigr)^\text{id}_{p_1;p_2}  &
\begin{tikzpicture}[baseline=-0.35em,line width=0.7pt,scale=0.8]
\draw[thick] (-.6,.3) -- (0,0) -- (-.4,-.3);
\draw[thick,dashed] (0,0) -- (1,0);
\draw[thick] (1,0) -- (1.5,0);
\draw[thick,dashed] (1.5,0) -- (2.5,0);
\draw[thick] (2.5,0) -- (3,0);
\draw[thick,dashed] (3,0) -- (4,0);
\draw[double,->] (4,0) --  (4.4,0);
\filldraw[fill=black] (-.6,.3) circle (.1) node[left=1pt]{\scriptsize $0$};
\filldraw[fill=black] (-.4,-.3) circle (.1) node[left=1pt]{\scriptsize $1$};
\filldraw[fill=black] (0,0) circle (.1) node[above=1pt]{\scriptsize $2$};
\filldraw[fill=black] (1,0) circle (.1);
\filldraw[fill=white] (1.5,0) circle (.1) node[above=1pt]{\scriptsize $\ell$};
\filldraw[fill=white] (2.5,0) circle (.1) node[above=1pt]{\scriptsize $r$};
\filldraw[fill=black] (3,0) circle (.1);
\filldraw[fill=black] (4,0) circle (.1) node[above]{\scriptsize $n\!\!-\!\!1$};
\filldraw[fill=black] (4.5,0) circle (.1) node[right=1pt]{\scriptsize $n$};
\draw[snake=brace] (-.7,.55) -- (1.1,.55) node[midway,above]{\scriptsize $p_1$};
\draw[snake=brace] (4.6,-.25) -- (2.9,-.25) node[midway,below]{\scriptsize $p_2$};
\end{tikzpicture} & (0,0) & \begin{array}{c} \ell \text{ even}, \, r \ne \ell+2 \\ 0 \le \ell \le r \le n \\ (\ell,r)\ne(n\!-\!1,n)  \end{array}  \\[1.5em]
\text{B.1b} & \bigl( {\rm B}_n^{(1)} \bigr)^{\flL}_{p_1;p_2}  &
\begin{tikzpicture}[baseline=-0.35em,line width=0.7pt,scale=0.8]
\draw[thick] (-.5,.3) -- (0,0) -- (-.5,-.3);
\draw[thick,dashed] (0,0) -- (1,0);
\draw[thick] (1,0) -- (1.5,0);
\draw[thick,dashed] (1.5,0) -- (2.5,0);
\draw[thick] (2.5,0) -- (3,0);
\draw[thick,dashed] (3,0) -- (4,0);
\draw[double,->] (4,0) --  (4.4,0);
\filldraw[fill=black] (-.5,.3) circle (.1) node[left=1pt]{\scriptsize $0$};
\filldraw[fill=black] (-.5,-.3) circle (.1) node[left=1pt]{\scriptsize $1$};
\filldraw[fill=black] (0,0) circle (.1) node[above=1pt]{\scriptsize $2$};
\filldraw[fill=black] (1,0) circle (.1);
\filldraw[fill=white] (1.5,0) circle (.1) node[above=1pt]{\scriptsize $\ell$};
\filldraw[fill=white] (2.5,0) circle (.1) node[above=1pt]{\scriptsize $r$};
\filldraw[fill=black] (3,0) circle (.1);
\filldraw[fill=black] (4,0) circle (.1) node[above]{\scriptsize $n\!\!-\!\!1$};
\filldraw[fill=black] (4.5,0) circle (.1) node[right=1pt]{\scriptsize $n$};
\draw[snake=brace] (-.6,.55) --  (1.1,.55) node[midway,above]{\scriptsize $p_1$};
\draw[snake=brace] (4.6,-.25) -- (2.9,-.25) node[midway,below]{\scriptsize $p_2$};
\draw[<->,gray] (-.5,.2) -- (-.5,-.2);
\end{tikzpicture}  & (1,0) & \begin{array}{c} \ell \text{ odd}, \, r\ne \ell+2  \\ 0 \le \ell \le r \le n \\ (\ell,r) \ne (1,1) \\ (\ell,r) \ne (n\!-\!1,n) \end{array}  \\[1.5em] \hline
\text{D.1a} & \bigl( {\rm D}^{(1)}_n \bigr)^\id_{p_1,p_2}  &
\begin{tikzpicture}[baseline=-0.35em,line width=0.7pt,scale=0.8]
\draw[thick] (-.6,.3) -- (0,0) -- (-.4,-.3);
\draw[thick,dashed] (0,0) -- (1,0);
\draw[thick] (1,0) -- (1.5,0);
\draw[thick,dashed] (1.5,0) -- (2.5,0);
\draw[thick] (2.5,0) -- (3,0);
\draw[thick,dashed] (3,0) -- (4,0);
\draw[thick] (4.4,.3) -- (4,0) -- (4.6,-.3);
\filldraw[fill=black] (-.6,.3) circle (.1) node[left=1pt]{\scriptsize $0$};
\filldraw[fill=black] (-.4,-.3) circle (.1) node[left=1pt]{\scriptsize $1$};
\filldraw[fill=black] (0,0) circle (.1) node[above=1pt]{\scriptsize $2$};
\filldraw[fill=black] (1,0) circle (.1);
\filldraw[fill=white] (1.5,0) circle (.1) node[above=1pt]{\scriptsize $\ell$};
\filldraw[fill=white] (2.5,0) circle (.1) node[above=1pt]{\scriptsize $r$};
\filldraw[fill=black] (3,0) circle (.1);
\filldraw[fill=black] (4,0) circle (.1) node[above]{\scriptsize $\hspace{-8pt} n\!\!-\!\!2$};
\filldraw[fill=black] (4.4,.3) circle (.1) node[right=1pt]{\scriptsize $n\!\!-\!\!1$};
\filldraw[fill=black] (4.6,-.3) circle (.1) node[right=1pt]{\scriptsize $n$};
\draw[snake=brace] (-.7,.5) -- (1.1,.5) node[midway,above]{\scriptsize $p_1$};
\draw[snake=brace] (4.7,-.5) -- (2.9,-.5) node[midway,below]{\scriptsize $p_2$};
\end{tikzpicture} & (0,0) & \begin{array}{c} \ell, n\!-\!r \text{ even},  r\ne \ell+2 \\ 0 \le \ell \le r \le n-\ell \end{array}  \\
\text{D.1b} & \bigl( {\rm D}^{(1)}_n \bigr)^{\flR}_{p_1;p_2}  & 
\begin{tikzpicture}[baseline=-0.35em,line width=0.7pt,scale=0.8]
\draw[thick] (-.6,.3) -- (0,0) -- (-.4,-.3);
\draw[thick,dashed] (0,0) -- (1,0);
\draw[thick] (1,0) -- (1.5,0);
\draw[thick,dashed] (1.5,0) -- (2.5,0);
\draw[thick] (2.5,0) -- (3,0);
\draw[thick,dashed] (3,0) -- (4,0);
\draw[thick] (4.5,.3) -- (4,0) -- (4.5,-.3);
\filldraw[fill=black] (-.6,.3) circle (.1) node[left=1pt]{\scriptsize $0$};
\filldraw[fill=black] (-.4,-.3) circle (.1) node[left=1pt]{\scriptsize $1$};
\filldraw[fill=black] (0,0) circle (.1) node[above=1pt]{\scriptsize $2$};
\filldraw[fill=black] (1,0) circle (.1);
\filldraw[fill=white] (1.5,0) circle (.1) node[above=1pt]{\scriptsize $\ell$};
\filldraw[fill=white] (2.5,0) circle (.1) node[above=1pt]{\scriptsize $r$};
\filldraw[fill=black] (3,0) circle (.1);
\filldraw[fill=black] (4,0) circle (.1) node[above]{\scriptsize $\hspace{-6pt}n\!\!-\!\!2$};
\filldraw[fill=black] (4.5,.3) circle (.1) node[right=1pt]{\scriptsize $n\!\!-\!\!1$};
\filldraw[fill=black] (4.5,-.3) circle (.1) node[right=1pt]{\scriptsize $n$};
\draw[<->,gray] (4.5,.2) -- (4.5,-.2);
\draw[snake=brace] (-.7,.5) -- (1.1,.5) node[midway,above]{\scriptsize $p_1$};
\draw[snake=brace] (4.6,-.5) -- (2.9,-.5) node[midway,below]{\scriptsize $p_2$};
\end{tikzpicture}  & (0,1) & \begin{array}{c} \ell \text{ even}, n\!-\!r \text{ odd} \\  r \ne \ell+2 \\ 0 \le \ell \le r \le n  \\ (\ell,r)\ne(n\!-\!1,n\!-\!1) \end{array} 
\\[2em]
\text{D.1c} & \bigl( {\rm D}^{(1)}_n \bigr)^{\flLR}_{p_1,p_2}  & \hspace{5.9pt}
\begin{tikzpicture}[baseline=-0.35em,line width=0.7pt,scale=0.8]
\draw[thick] (-.5,.3) -- (0,0) -- (-.5,-.3);
\draw[thick,dashed] (0,0) -- (1,0);
\draw[thick] (1,0) -- (1.5,0);
\draw[thick,dashed] (1.5,0) -- (2.5,0);
\draw[thick] (2.5,0) -- (3,0);
\draw[thick,dashed] (3,0) -- (4,0);
\draw[thick] (4.5,.3) -- (4,0) -- (4.5,-.3);
\filldraw[fill=black] (-.5,.3) circle (.1) node[left=1pt]{\scriptsize $0$};
\filldraw[fill=black] (-.5,-.3) circle (.1) node[left=1pt]{\scriptsize $1$};
\draw[<->,gray] (-.5,.2) -- (-.5,-.2);
\filldraw[fill=black] (0,0) circle (.1) node[above=1pt]{\scriptsize $2$};
\filldraw[fill=black] (1,0) circle (.1);
\filldraw[fill=white] (1.5,0) circle (.1) node[above=1pt]{\scriptsize $\ell$};
\filldraw[fill=white] (2.5,0) circle (.1) node[above=1pt]{\scriptsize $r$};
\filldraw[fill=black] (3,0) circle (.1);
\filldraw[fill=black] (4,0) circle (.1) node[above]{\scriptsize $\hspace{-4pt} n\!\!-\!\!2$};
\filldraw[fill=black] (4.5,.3) circle (.1) node[right=1pt]{\scriptsize $n\!\!-\!\!1$};
\filldraw[fill=black] (4.5,-.3) circle (.1) node[right=1pt]{\scriptsize $n$};
\draw[<->,gray] (4.5,.2) -- (4.5,-.2);
\draw[snake=brace] (-.6,.5) -- (1.1,.5) node[midway,above]{\scriptsize $p_1$};
\draw[snake=brace] (4.6,-.5) -- (2.9,-.5) node[midway,below]{\scriptsize $p_2$};
\end{tikzpicture}  & (1,1) & \begin{array}{c} \ell, n\!-\!r \text{ odd}, \, r \ne \ell+2 \\ 0 \le \ell \le r \le n-\ell \\ (\ell,r)\ne (1,1) \\ (\ell,r) \ne (n\!-\!1,n\!-\!1)  \end{array}  
\\ \hline
\end{array}
\]
\end{table}
}

{\allowdisplaybreaks
The QP algebra is generated by $x_i$, $y_i$, $k_i$ with $i\in X$ and $(k_0k_1^{-1})^{\pm 1}$ if $\ell=1$ and $(k_{n-1}k_n^{-1})^{\pm 1}$ if $r=n-1$ and $N$ is even, and elements $b_j$ whose reduced expressions are
\begin{alignat*}{99}
b_0 &= \casesl{l}{
y_0 - c_0 \, T_{w_{X_2}}(x_{0})\, k_0^{-1}, \\
y_0 - c_0 \, x_{o_1} k_0^{-1}, \\
} & \qu & \casesm{l}{
\text{if } 0 = \ell \le r \le 1, \\ 
\text{if } 0 \le \ell \le 1 < r , 
} \\
b_1 &= y_1 - c_1 \, x_{1-o_1} k_1^{-1}  && \text{if } 0 \le \ell \le 1 < r, 
\\
b_\ell &= \casesl{l}{
y_\ell - c_\ell \, T_{w_{X_1}}(x_{\ell})\, k_\ell^{-1} \\
y_\ell - c_\ell \, T_{w_{X_1}}T_{w_{X_2}}(x_\ell)\, k_\ell^{-1} \\
} && \casesm{l}{
\text{if } 2\le \ell < r, \\ 
\text{if } 2\le \ell = r, 
} \\
b_j &= y_j - c_j \, x_j\, k_j^{-1} && \text{if } \ell<j<r<n-(-1)^N,  \\
b_r &= y_r - c_r \, T_{w_{X_2}}(x_r)\, k_r^{-1} && \text{if } \ell<r<n-(-1)^N, \\
b_{n-1} &= y_{n-1} - c_{n-1}\, x_{n-1+o_2}\, k_{n-1}^{-1} && \text{if } \ell<n-(-1)^N\le r, \\ 
b_{n} &= y_n - c_n\, x_{n-o_2}\, k_n^{-1} && \text{if } \ell<n-(-1)^N\le r,
\intertext{where $c_0=c_1$ if $\ell=1$ and $c_{n-1}=c_n$ if $r=n-(-1)^N$. Next, in terms of the effective dressing parameters $\omega_{\ell+1},\ldots,\omega_r$ and the scaling parameter $\eta$ we assign}
c_\ell &= \casesl{l}{
q^{2(\lceil N/2 \rceil -1)}\eta^{-2}\\
-(-q)^{\lceil N/2 \rceil -2} \eta^{-2} \om_{1}^2 \\ 
q^{-1} \eta^{-2} \om_{1}^2\om_{2}^2 \\ 
-(-q)^{\ell-1} \eta^{-1} \om_{\ell+1}^2\!\! \\
-(-q)^{\lceil N/2 \rceil - 1}\eta^{-1}
} & \qu & \casesm{l}{
\text{if } 0=\ell=r,\\ 
\text{if } 0=\ell<r=1, \\
\text{if } 0=\ell<r-1, \\ 
\text{if } 1\le\ell<r,  \\
\text{if } 2\le\ell=r, 
}   \\ 
c_j &= q^{-1} \om_j^{-2} \om_{j+1}^2 && \text{if } \ell<j<r, 
\\
c_r &= \casesl{l}{-(-q)^{\lceil N/2 \rceil-1-r} \om_r^{-2} \\ q^{-1} \om_{n-1}^{-2} \om_n^{-2} } && \casesm{l}{\text{if } \ell<r<\lceil N/2 \rceil,\!\! \\ \text{if } \ell<r=N/2. } 
\end{alignat*}
Then, solving the boundary intertwining equation \eqref{intw-untw} for all generators of the QP algebra we obtain the following result.
}

\begin{result} \label{res:BD1} 
The bare K-matrix of type BD.1 when $|I_{\rm nsf} \cup I_{\rm diff}| = 0$ is given by \eqref{K(u):X} with $M_1(u)$ as in \eqref{C2:K:parts} and $\la=q^{N/2-r}$, $\mu = q^{-\ell}$ as before, and 
\eq{ \label{BD1:D_2}
M_2(u) = \sum_{\bar r\le i < \bar \ell} ( \la \, E_{-i,-i} + \la^{-1} E_{ii} + E_{-i,i} + E_{i,-i}) .
}
\end{result}

For $N$ odd and $(\ell,r)=(0,n)$ this K-matrix, upon dressing, coincides with the one in \cite[(163-164)]{MLS}; for $N$ even and $(\ell,r)=(0,1)$ it corresponds to \cite[(199)]{MLS}. 
Moreover, it has the same properties as the one of type C.2 except for the following.

\smallskip

\begin{description}[itemsep=1ex]
\item[Eigendecomposition]  $V= \Id + \la \sum_{\bar r \le i < \bar \ell} \big( E_{-i,i} - E_{i,-i} \big)$ and $D(u)$ as for type C.2.

\item[Affinization] The constant K-matrix $K_0$ corresponds to the BDI solution of the constant twisted RE reported in \cite[Sec.~3]{NoSu}.

\item[Rotations] \hfill $K^{\flL}(u) = K(u)$ if $r>0$ and $K^{\flR}(u) = K(u)$ if $N$ is even. \hfill \hphantom{\it Rotations}

\item[Bar-symmetry] \hfill $K(u)^{-1}  = J K(u)|_{\la \to \la^{-1},\, \mu \to \mu^{-1}} J$. \hfill \hphantom{\it Bar-symmetry} 

\item[Reductions] For $N$ even and $(\ell,r)=(0,n)$ it is a non-regular \gim~$K(u)=-J$.

\item[Diagonal cases] For $\ell=r\ne1$ it coincides with the one of type BD.2 with $\ell=r$.

\end{description}

\medskip

\begin{rmk} \label{R:D1:rotate} 
K-matrices of type D.1ac also satisfy \eqref{Zpi:bijection}, see Remark \ref{R:C1:rotate}. \hfill \rmkend
\end{rmk}


\subsection{Untwisted K-matrices for non-quasistandard QP algebras of types BCD.1} \label{sec:nqs}

In this section we study non-quasistandard weak Satake diagrams for the family BCD.1. 
The representative diagrams and special $\tau$-orbits are listed in Table \ref{tab:BD1:nstd}. Note that $(X,\tau)\in \Sat(A)$ for the BD.1 family, but $(X,\tau)\in \WSat(A)$ for the C.1 family. It will be shown below that all these diagrams yield K-matrices that are not generalized cross matrices, which sets them apart from the other cases studied in this paper.

{
\arraycolsep=2pt\def\arraystretch{1.2}
\begin{table}[h]
\caption{Family BCD.1: representative non-quasistandard weak Satake diagrams. 
In all cases $I_{\rm diff} = \emptyset$ and $I_{\rm nsf} = I_{\rm ns}$, except when $(\ell,r)=(n-2,n)$ and $N$ is odd, in which case $I_{\rm ns} = \{ n-1,n\}$. }
\label{tab:BD1:nstd}
\[
\begin{array}{cccccc}
\text{Type} & \rm Name & \rm Diagram & \!\!(o_1,o_2)\!\! & \rm Restrictions & I_{\rm nsf}
\\ 

\hline\hline
\text{B.1a} & \bigl( {\rm B}_n^{(1)} \bigr)^\id_{0;n-2} &
\begin{tikzpicture}[baseline=-0.35em,line width=0.7pt,scale=0.8]
\draw[thick] (-.6,.3) -- (0,0) -- (-.4,-.3);
\draw[thick] (0,0) -- (.5,0);
\draw[thick,dashed] (.5,0) -- (1.5,0);
\draw[double,->] (1.5,0) --  (1.9,0);
\filldraw[fill=white] (-.6,.3) circle (.1) node[left=1pt]{\scriptsize $0$};
\filldraw[fill=white] (-.4,-.3) circle (.1) node[left=1pt]{\scriptsize $1$};
\filldraw[fill=white] (0,0) circle (.1) node[above=1pt]{\scriptsize $2$};
\filldraw[fill=black] (.5,0) circle (.1) node[above]{\scriptsize $3$};
\filldraw[fill=black] (1.5,0) circle (.1) node[above]{\scriptsize $n\!\!-\!\!1$};
\filldraw[fill=black] (2,0) circle (.1) node[right=1pt]{\scriptsize $n$};
\end{tikzpicture} & (0,0) & (\ell,r)=(0,2) & \{0,1\} \\ 
\hline

\text{D.1a} & \! \bigl( {\rm D}^{(1)}_n \bigr)^\id_{0,n-2} &
\begin{tikzpicture}[baseline=-0.35em,line width=0.7pt,scale=0.8]
\draw[thick] (-.6,.3) -- (0,0) -- (-.4,-.3);
\draw[thick] (0,0) -- (.5,0);
\draw[thick,dashed] (0.5,0) -- (1.5,0);
\draw[thick] (1.9,.3) -- (1.5,0) -- (2.1,-.3);
\filldraw[fill=white] (-.6,.3) circle (.1) node[left=1pt]{\scriptsize $0$};
\filldraw[fill=white] (-.4,-.3) circle (.1) node[left=1pt]{\scriptsize $1$};
\filldraw[fill=white] (0,0) circle (.1) node[above=1pt]{\scriptsize $2$};
\filldraw[fill=black] (.5,0) circle (.1) node[above=1pt]{\scriptsize $3$};
\filldraw[fill=black] (1.5,0) circle (.1) node[above]{\scriptsize $\hspace{-8pt} n\!\!-\!\!2$};
\filldraw[fill=black] (1.9,.3) circle (.1) node[right=1pt]{\scriptsize $n\!\!-\!\!1$};
\filldraw[fill=black] (2.1,-.3) circle (.1) node[right=1pt]{\scriptsize $n$};
\end{tikzpicture}  & (0,0) & \!\begin{array}{c} n \text{ even} \\ (\ell,r)=(0,2) \end{array} \! & \multirow{3}{*}{\{0,1\}} \\
\text{D.1b} & \bigl( {\rm D}^{(1)}_n \bigr)^{\flR}_{0;n-3} & 
\begin{tikzpicture}[baseline=-0.35em,line width=0.7pt,scale=0.8]
\draw[thick] (-.6,.3) -- (0,0) -- (-.4,-.3);
\draw[thick] (0,0) -- (.5,0);
\draw[thick,dashed] (0.5,0) -- (1.5,0);
\draw[thick] (2,.3) -- (1.5,0) -- (2,-.3);
\filldraw[fill=white] (-.6,.3) circle (.1) node[left=1pt]{\scriptsize $0$};
\filldraw[fill=white] (-.4,-.3) circle (.1) node[left=1pt]{\scriptsize $1$};
\filldraw[fill=white] (0,0) circle (.1) node[above=1pt]{\scriptsize $2$};
\filldraw[fill=black] (.5,0) circle (.1) node[above=1pt]{\scriptsize $3$};
\filldraw[fill=black] (1.5,0) circle (.1) node[above]{\scriptsize $\hspace{-4pt} n\!\!-\!\!2$};
\filldraw[fill=black] (2,.3) circle (.1) node[right=1pt]{\scriptsize $n\!\!-\!\!1$};
\filldraw[fill=black] (2,-.3) circle (.1) node[right=1pt]{\scriptsize $n$};
\draw[<->,gray] (2,.2) -- (2,-.2);
\end{tikzpicture}  & (0,1) & \!\begin{array}{c} n \text{ odd} \\ (\ell,r)=(0,2) \end{array} \! &  \\
\hline\hline 

\text{C.1} & \! \bigl( {\rm C}^{(1)}_n \bigr)^\id_{\ell;n-\ell-2} &
\begin{tikzpicture}[baseline=-0.25em,line width=0.7pt,scale=0.8]
\draw[double,<-] (-.1,0) -- (-.5,0);
\draw[thick,dashed] (0,0) -- (1,0);
\draw[thick] (1,0) -- (3,0);
\draw[thick,dashed] (3,0) -- (4,0);
\draw[double,<-] (4.1,0) --  (4.5,0);
\filldraw[fill=black] (-.5,0) circle (.1) node[left=1pt]{\scriptsize $0$};
\filldraw[fill=black] (0,0) circle (.1) node[above=1pt]{\scriptsize $1$};
\filldraw[fill=black] (1,0) circle (.1); 
\filldraw[fill=white] (1.5,0) circle (.1) node[above=1pt]{\scriptsize $\ell$};
\filldraw[fill=white] (2,0) circle (.1) node[below=1pt]{\scriptsize $\ell\!+\!1$};
\filldraw[fill=white] (2.5,0) circle (.1) node[above=1pt]{\scriptsize $\ell\!+\!2$};
\filldraw[fill=black] (3,0) circle (.1);
\filldraw[fill=black] (4,0) circle (.1) node[above]{\scriptsize $n\!-\!1$};
\filldraw[fill=black] (4.5,0) circle (.1) node[right=1pt]{\scriptsize $n$};
\draw[] (0,.75) -- (0,.75);
\end{tikzpicture}  & (0,0) & 1\le \ell \le n\!-\!3 & \{\ell\!+\!1\} \\
\hline\hline
\text{B.1a} & \bigl( {\rm B}_n^{(1)} \bigr)^\text{id}_{\ell;n-\ell-2}  &
\begin{tikzpicture}[baseline=-0.35em,line width=0.7pt,scale=0.8]
\draw[thick] (-.6,.3) -- (0,0) -- (-.4,-.3);
\draw[thick,dashed] (0,0) -- (1,0);
\draw[thick] (1,0) -- (3,0);
\draw[thick,dashed] (3,0) -- (4,0);
\draw[double,->] (4,0) --  (4.4,0);
\filldraw[fill=black] (-.6,.3) circle (.1) node[left=1pt]{\scriptsize $0$};
\filldraw[fill=black] (-.4,-.3) circle (.1) node[left=1pt]{\scriptsize $1$};
\filldraw[fill=black] (0,0) circle (.1) node[above=1pt]{\scriptsize $2$};
\filldraw[fill=black] (1,0) circle (.1);
\filldraw[fill=white] (1.5,0) circle (.1) node[above=1pt]{\scriptsize $\ell$};
\filldraw[fill=white] (2,0) circle (.1) node[below=1pt]{\scriptsize $\ell\!\!+\!\!1$};
\filldraw[fill=white] (2.5,0) circle (.1) node[above]{\scriptsize $\ell\!\!+\!\!2$};
\filldraw[fill=black] (3,0) circle (.1);
\filldraw[fill=black] (4,0) circle (.1) node[above]{\scriptsize $n\!\!-\!\!1$};
\filldraw[fill=black] (4.5,0) circle (.1) node[right=1pt]{\scriptsize $n$};
\end{tikzpicture} & (0,0) & 
\!\begin{array}{c} 1 \le \ell < n\!-\!1 \\ r=\ell\!+\!2, \, \ell \text{ even} \end{array} \!  & \,\multirow{3}{*}{$\{\ell\!+\!1\}$} 
\\
\text{B.1b} & \bigl( {\rm B}_n^{(1)} \bigr)^{\flL}_{\ell-1;n-\ell-2}\!  & 
\begin{tikzpicture}[baseline=-0.35em,line width=0.7pt,scale=0.8]
\draw[thick] (-.5,.3) -- (0,0) -- (-.5,-.3);
\draw[<->,gray] (-.5,.2) -- (-.5,-.2);
\draw[thick,dashed] (0,0) -- (1,0);
\draw[thick] (1,0) -- (3,0);
\draw[thick,dashed] (3,0) -- (4,0);
\draw[double,->] (4,0) --  (4.4,0);
\filldraw[fill=black] (-.5,.3) circle (.1) node[left=1pt]{\scriptsize $0$};
\filldraw[fill=black] (-.5,-.3) circle (.1) node[left=1pt]{\scriptsize $1$};
\filldraw[fill=black] (0,0) circle (.1) node[above=1pt]{\scriptsize $2$};
\filldraw[fill=black] (1,0) circle (.1);
\filldraw[fill=white] (1.5,0) circle (.1) node[above=1pt]{\scriptsize $\ell$};
\filldraw[fill=white] (2,0) circle (.1) node[below=1pt]{\scriptsize $\ell\!\!+\!\!1$};
\filldraw[fill=white] (2.5,0) circle (.1) node[above]{\scriptsize $\ell\!\!+\!\!2$};
\filldraw[fill=black] (3,0) circle (.1);
\filldraw[fill=black] (4,0) circle (.1) node[above]{\scriptsize $n\!\!-\!\!1$};
\filldraw[fill=black] (4.5,0) circle (.1) node[right=1pt]{\scriptsize $n$};
\end{tikzpicture} & (1,0) & 
\!\begin{array}{c} 1 \le \ell < n\!-\!1 \\  r=\ell\!+\!2, \, \ell \text{ odd} \end{array} \!  \\ \hline
\text{D.1a} & \bigl( {\rm D}^{(1)}_n \bigr)^\id_{\ell,n-\ell-2}  &
\!\begin{tikzpicture}[baseline=-0.35em,line width=0.7pt,scale=0.8]
\draw[thick] (-.6,.3) -- (0,0) -- (-.4,-.3);
\draw[thick,dashed] (0,0) -- (1,0);
\draw[thick] (1,0) -- (3,0);
\draw[thick,dashed] (3,0) -- (4,0);
\draw[thick] (4.4,.3) -- (4,0) -- (4.6,-.3);
\filldraw[fill=black] (-.6,.3) circle (.1) node[left=1pt]{\scriptsize $0$};
\filldraw[fill=black] (-.4,-.3) circle (.1) node[left=1pt]{\scriptsize $1$};
\filldraw[fill=black] (0,0) circle (.1) node[above=1pt]{\scriptsize $2$};
\filldraw[fill=black] (1,0) circle (.1);
\filldraw[fill=white] (1.5,0) circle (.1) node[above=1pt]{\scriptsize $\ell$};
\filldraw[fill=white] (2,0) circle (.1) node[below=1pt]{\scriptsize $\ell\!\!+\!\!1$};
\filldraw[fill=white] (2.5,0) circle (.1)  node[above=1pt]{\scriptsize $\ell\!\!+\!\!2$};
\filldraw[fill=black] (3,0) circle (.1);
\filldraw[fill=black] (4,0) circle (.1) node[above]{\scriptsize $\hspace{-8pt} n\!\!-\!\!2$};
\filldraw[fill=black] (4.4,.3) circle (.1) node[right=1pt]{\scriptsize $n\!\!-\!\!1$};
\filldraw[fill=black] (4.6,-.3) circle (.1) node[right=1pt]{\scriptsize $n$};
\end{tikzpicture}\! & (0,0) & 
\!\begin{array}{c} 1 \le \ell \le \frac{n}{2},\, r=\ell\!+\!2 \\ \ell \text{ even}, n \text{ even} \end{array} \! \\
\text{\!D.1b} & \bigl( {\rm D}^{(1)}_n \bigr)^{\flR}_{\ell,n-\ell-3}  &
\begin{tikzpicture}[baseline=-0.35em,line width=0.7pt,scale=0.8]
\draw[thick] (-.6,.3) -- (0,0) -- (-.4,-.3);
\draw[thick,dashed] (0,0) -- (1,0);
\draw[thick] (1,0) -- (3,0);
\draw[thick,dashed] (3,0) -- (4,0);
\draw[thick] (4.5,.3) -- (4,0) -- (4.5,-.3);
\filldraw[fill=black] (-.6,.3) circle (.1) node[left=1pt]{\scriptsize $0$};
\filldraw[fill=black] (-.4,-.3) circle (.1) node[left=1pt]{\scriptsize $1$};
\filldraw[fill=black] (0,0) circle (.1) node[above=1pt]{\scriptsize $2$};
\filldraw[fill=black] (1,0) circle (.1);
\filldraw[fill=white] (1.5,0) circle (.1)  node[above=1pt]{\scriptsize $\ell$};
\filldraw[fill=white] (2,0) circle (.1) node[below=1pt]{\scriptsize $\ell\!\!+\!\!1$};
\filldraw[fill=white] (2.5,0) circle (.1) node[above]{\scriptsize $\ell\!\!+\!\!2$};
\filldraw[fill=black] (3,0) circle (.1);
\filldraw[fill=black] (4,0) circle (.1) node[above]{\scriptsize $\hspace{-6pt}n\!\!-\!\!2$};
\filldraw[fill=black] (4.5,.3) circle (.1) node[right=1pt]{\scriptsize $n\!\!-\!\!1$};
\filldraw[fill=black] (4.5,-.3) circle (.1) node[right=1pt]{\scriptsize $n$};
\draw[<->,gray] (4.5,.2) -- (4.5,-.2);
\end{tikzpicture}  & (0,1) & 
\!\begin{array}{c} 1 \le \ell \le n\!-\!3,\, r=\ell\!+\!2 \\ \ell \text{ even}, n \text{ odd} \end{array} \! & \,\{\ell\!+\!1\} \\
\text{D.1c} & \bigl( {\rm D}^{(1)}_n \bigr)^{\flLR}_{\ell-1,n-\ell-3}\!  & \hspace{5.9pt}
\begin{tikzpicture}[baseline=-0.35em,line width=0.7pt,scale=0.8]
\draw[thick] (-.5,.3) -- (0,0) -- (-.5,-.3);
\draw[thick,dashed] (0,0) -- (1,0);
\draw[thick] (1,0) -- (3,0);
\draw[thick,dashed] (3,0) -- (4,0);
\draw[thick] (4.5,.3) -- (4,0) -- (4.5,-.3);
\filldraw[fill=black] (-.5,.3) circle (.1) node[left=1pt]{\scriptsize $0$};
\filldraw[fill=black] (-.5,-.3) circle (.1) node[left=1pt]{\scriptsize $1$};
\draw[<->,gray] (-.5,.2) -- (-.5,-.2);
\filldraw[fill=black] (0,0) circle (.1) node[above=1pt]{\scriptsize $2$};
\filldraw[fill=black] (1,0) circle (.1);
\filldraw[fill=white] (1.5,0) circle (.1) node[above=1pt]{\scriptsize $\ell$};
\filldraw[fill=white] (2,0) circle (.1) node[below=1pt]{\scriptsize $\ell\!\!+\!\!1$};
\filldraw[fill=white] (2.5,0) circle (.1) node[above]{\scriptsize $\ell\!\!+\!\!2$};
\filldraw[fill=black] (3,0) circle (.1);
\filldraw[fill=black] (4,0) circle (.1) node[above]{\scriptsize $\hspace{-4pt} n\!\!-\!\!2$};
\filldraw[fill=black] (4.5,.3) circle (.1) node[right=1pt]{\scriptsize $n\!\!-\!\!1$};
\filldraw[fill=black] (4.5,-.3) circle (.1) node[right=1pt]{\scriptsize $n$};
\draw[<->,gray] (4.5,.2) -- (4.5,-.2);
\end{tikzpicture}  & (1,1) & 
\!\begin{array}{c} 1 \le \ell \le \frac{n}{2},\, r=\ell\!+\!2 \\ \ell \text{ odd}, n \text{ even} \end{array} \! & \\
\hline\hline 
\text{B.1a} & \bigl( {\rm B}_n^{(1)} \bigr)^\text{id}_{n-1;0} &
\begin{tikzpicture}[baseline=-0.35em,line width=0.7pt,scale=0.8]
\draw[thick] (-.6,.3) -- (0,0) -- (-.4,-.3);
\draw[thick,dashed] (0,0) -- (1,0);
\draw[thick] (1,0) -- (1.5,0);
\draw[double,->] (1.5,0) --  (1.9,0);
\filldraw[fill=black] (-.6,.3) circle (.1) node[left=1pt]{\scriptsize $0$};
\filldraw[fill=black] (-.4,-.3) circle (.1) node[left=1pt]{\scriptsize $1$};
\filldraw[fill=black] (0,0) circle (.1) node[above=1pt]{\scriptsize $2$};
\filldraw[fill=black] (1,0) circle (.1);
\filldraw[fill=white] (1.5,0) circle (.1);
\filldraw[fill=white] (2,0) circle (.1) node[right=1pt]{\scriptsize $n$};
\end{tikzpicture} & (0,0) & 
\!\begin{array}{c} (\ell,r)=(n\!-\!1,n) \\  n \text{ odd}\end{array} \! & \{n\}
\\
\text{B.1b} & \bigl( {\rm B}_n^{(1)} \bigr)^{\flL}_{n-2;0} & 
\begin{tikzpicture}[baseline=-0.35em,line width=0.7pt,scale=0.8]
\draw[thick] (-.5,.3) -- (0,0) -- (-.5,-.3);
\draw[<->,gray] (-.5,.2) -- (-.5,-.2);
\draw[thick,dashed] (0,0) -- (1,0);
\draw[thick] (1,0) -- (1.5,0);
\draw[double,->] (1.5,0) --  (1.9,0);
\filldraw[fill=black] (-.5,.3) circle (.1) node[left=1pt]{\scriptsize $0$};
\filldraw[fill=black] (-.5,-.3) circle (.1) node[left=1pt]{\scriptsize $1$};
\filldraw[fill=black] (0,0) circle (.1) node[above=1pt]{\scriptsize $2$};
\filldraw[fill=black] (1,0) circle (.1);
\filldraw[fill=white] (1.5,0) circle (.1);
\filldraw[fill=white] (2,0) circle (.1) node[right=1pt]{\scriptsize $n$};
\end{tikzpicture} & (1,0) & \!\begin{array}{c} (\ell,r)=(n\!-\!1,n) \\  n \text{ even}\end{array} \!  & \{n\} \\\hline
\end{array}
\]
\end{table}
}


\subsubsection{Family BD.1 with $|I_{\rm nsf}|=2$}

Satake diagrams in this family are parametrized by $\ell=0$ and $r=0$ and have $X = \{3,4,\ldots,n\}$. The QP algebra is generated by $x_i$, $y_i$, $k_i$ with $i\in X$ and elements $b_j$ whose reduced expressions are
\[ 
b_j = y_j - c_j\, x_j k_j^{-1} - s_j k_j^{-1} \qu\text{for } j \le 1 \qu\text{and}\qu b_2 = y_2 - c_2 \,T_{w_X}(x_2)\,k_2^{-1} ,  
\]
with $c_0,c_1,c_2$ and $s_0,s_1$ all independent. Upon setting 
\begin{gather*}
c_0 = q^{-1} \eta^{-2} \om_1^2 \om_2^2 , \qq
c_1 = q^{-1} \om_1^{-2} \om_2^2 , \qq
c_2 = -(-q)^{n-2} \om_2^{-2} , 
\\
s_0 = \frac{\nu_0 + \nu_0^{-1}}{q-q^{-1}} \eta^{-1} \om_1 \om_2 , \qq
s_1 = \frac{\nu_1 + \nu_1^{-1}}{q-q^{-1}} \om_1^{-1} \om_2, 
\end{gather*}
where $\nu_0,\nu_1\in\K^\times$ are additional free parameters, and solving the boundary intertwining equation \eqref{intw-untw} for all generators of the QP algebra we obtain the following result.

\begin{result}
When $(\ell,r)=(0,2)$ the bare K-matrix of type BD.1 is 
\eq{
K(u) = \Id + \frac{(u-u^{-1}) \la^2 u^2}{k_1(\frac{\nu_0 u}{\nu_1})\, k_1(\frac{\nu_1 u}{\nu_0})\, k_1(\nu_0 \nu_1 u)\, k_1(\frac{u}{\nu_0 \nu_1})} \Big( k_1(u)\,k_2(u) M_2 + \al(\tfrac{u}{\la})  \, \wt M_3^- + \al(\tfrac{\la}{u}) \,  \wt M_3^+ \Big) 
}
where $\la  = q^{N/2-2}$, $\mu =1$ are as in Result \ref{res:BD1} and $\al(u) = (\nu_1+\nu_1^{-1})\, u - (\nu_0+\nu_0^{-1})$.
Furthermore, $M_2$ equals the expression $M_2(u)$ defined in \eqref{BD1:D_2}, i.e.
\[
M_2 = \sum_{n-1 \le i\le n} \big( \la E_{-i,-i} + \la^{-1} E_{i,i} + E_{-i,i} + E_{i,-i} \big) 
\]
and we have introduced
\[ 
\wt M_3^\pm =  E_{\mp(n-1),\pm n} - E_{\pm n,\mp (n-1)} + \la^{\mp 1} (E_{\pm (n-1),\pm  n}-E_{\pm n,\pm (n-1)} \mp  (\nu_1+\nu_1^{-1}) E_{\pm  n,\pm n}).
\]
\end{result}

Note that the rows and columns labeled $\pm n$ and $\pm (n-1)$ have four nonzero entries. This K-matrix has three distinct eigenvalues, however $d_{\rm eff}=4$ for generic values of $\nu_0$, $\nu_1$.
For $\nu_0^2=\nu_1^2=-1$ (so that $s_0=s_1=0$), it specializes to the one in Result \ref{res:BD1} with $(\ell,r)=(0,2)$. 
This K-matrix exhibits the following properties.

\medskip

\begin{description}[itemsep=1ex]

\item[Eigendecomposition] 
\eqn{ 
V&= \Id + (\nu_1+\nu_1^{-1}) E_{1-n,-n} -\nu_1 \big( E_{n-1,n} + E_{n,n-1} \big) \\
& \qq\;+ \la \big( \nu_1^{-1} E_{-n,n-1} - \nu_1^2 E_{-n,n} + E_{1-n,n-1} - \nu_1 E_{1-n,n} - E_{n-1,1-n} - E_{n,-n} \big), \\[.25em]
D(u) &= \sum_{i \le n-2} E_{ii} + u^2 \bigg( \tfrac{k_1(\nu_0^{-1} \nu_1 u^{-1})}{k_1(\nu_0^{-1} \nu_1 u)} \tfrac{k_1(\nu_0 \nu_1 u^{-1})}{k_1(\nu_0 \nu_1 u)} E_{n-1,n-1} + \tfrac{k_1(\nu_0 \nu_1^{-1} u^{-1})}{k_1(\nu_0 \nu_1^{-1} u)} \tfrac{k_1(\nu_0^{-1} \nu_1^{-1} u^{-1})}{k_1(\nu_0^{-1} \nu_1^{-1} u)} E_{n,n} \bigg).
}
Note that $V$ does not depend on $\nu_0$.
Its two rightmost columns are eigenvectors of $K(u)$ associated with the two non-unit eigenvalues. 
When specializing $\nu_0$ and $\nu_1$ to $\pm \sqrt{-1}$ these two columns do not tend to the rightmost columns of $V$ in the generic case of type BD.1 (see the eigendecomposition following Result \ref{res:BD1}). 
However both eigenvalues tend to $p_1(u)p_2(u)$ and the corresponding two-dimensional eigenspace can be given the expected basis.

\item[Affinization] 
The constant K-matrix $K_0$ is independent of $\nu_0$ and has eigenvalues $1$, $\la^{-2} \nu_1^2$ and $\la^{-2} \nu_1^{-2}$ with multiplicities $N-2$, 1 and 1, respectively.
The affinization identity is
\eqn{
K(u) &= \frac{(u-u^{-1})\big(\la  \, \al(\frac{\la}{u}) K_0 - \la^{-1}  \al(\frac{u}{\la}) K_0^{-1}\big) + \beta(u) \Id}{(u-u^{-1})\big(\la \, \al(\frac{\la}{u}) - \la^{-1}  \al(\frac{u}{\la}) \big)  + \beta(u)},
}
where $\beta(u) = - (\nu_1+\nu_1^{-1}) \, \al(\tfrac{\la}{u}) \, \al(\tfrac{u}{\la})- (\la-\la^{-1})\big(\al(\tfrac{\la}{u}) \, u -\al(\tfrac{u}{\la}) \, u^{-1}\big)$.

\item[Half-period] \hfill $K(-u) = K(u)|_{\nu_0 \to -\nu_0}$. \hfill \hphantom{\it Half-period}

\item[Rotations] \hfill $K^{\flL}(u) =  K(u)|_{\nu_0 \leftrightarrow \nu_1}$ and, for even $N$, $K^{\flR}(u) = K(u)$. \hfill \hphantom{\it Rotations}

\item[Reductions] For $(\nu_0/\nu_1)^2 = \la^{\pm 2}$ or $(\nu_0 \nu_1)^2 = \la^{\pm 2}$ it has $d_{\rm eff} = 2$ and is singly regular, $K(1)\ne\Id$ or $K(-1)\ne\Id$, depending on the choice of the square root in the previous identities.  

\item[Diagonal cases] \hfill $\displaystyle \lim_{\nu_0 \to 0} K(u) = \lim_{\nu_0 \to \infty} K(u) = \Id$, \hfill \hphantom{\it Diagonal cases}
\begin{gather*} 
\lim_{\nu_1 \to 0} K(u) = \lim_{\nu_1 \to \infty} K(u) = u^2 E_{-n,-n} + u^{-2} E_{nn}+ \sum_{i < n} (E_{-i,-i}+E_{ii}) , \\
\lim_{\nu_0 \to 0} K(u)|_{\nu_1 = \nu_0} = \lim_{\nu_0 \to \infty} K(u)|_{\nu_1 = \nu_0} = h_1(u) ( u^2 E_{-n,-n} + E_{n,n} ) + \sum_{i < n} (E_{-i,-i}+E_{ii}), \\
\lim_{\nu_0 \to 0} K(u)|_{\nu_1 = -\nu_0^{-1}} = \lim_{\nu_0 \to \infty} K(u)|_{\nu_1 = -\nu_0^{-1}} = h_2(u) ( E_{-n,-n} +u^{-2} E_{n,n}) + \sum_{i < n} (E_{-i,-i}+E_{ii}) .
\end{gather*} 

\end{description}


\subsubsection{Family C.1 with $I_{\rm nsf}=\{j\}$ with $2 \le j \le n-2$}

Weak Satake diagrams in this family are parametrized by $\ell=j-1$ and $r=j+1$ with $j \in \{2,\ldots,n-2\}$. Hence $X=X_1\cup X_2$ with $X_1=\{0,\ldots,\ell-1\}$ and $X_2=\{r+1,\ldots,n\}$. The QP algebra is generated by $x_i$, $y_i$, $k_i$ with $i\in X$ and elements $b_j$ whose reduced expressions are
\begin{gather*} 
b_{\ell+1} =  y_{\ell+1} - c_{\ell+1} x_{\ell+1} k_{\ell+1}^{-1} - s_{\ell+1} k_{\ell+1}^{-1}, \\
b_\ell = y_\ell - c_\ell T_{w_{X_1}}(x_\ell) k_\ell^{-1}, \qq b_{\ell+2} = y_{\ell+2} - c_{\ell+2} T_{w_{X_2}}(x_{\ell+2}) k_{\ell+2}^{-1},
\end{gather*} 
with $c_{\ell}$, $c_{\ell+1}$, $c_{\ell+2}$, $s_{\ell+1}$ independent. Upon setting the parameters $c_j$ in the same way as in the general case given in Section \ref{sec:C1} and $s_{\ell+1}$ as  
\begin{alignat*}{99}
c_\ell &= \casesl{l}{ q^{-2} \eta^{-2} \om_{1}^4 \\ (-q)^{\ell} \eta^{-1} \om_{\ell+1}^2 } &\qu & \casesm{l}{  \text{if }\ell=0, \\ \text{if }\ell>0, } 
\\
c_{\ell+1} &= q^{-1} \om^{-2}_{\ell+1} \om^2_{\ell+2} , 
\\
c_{\ell+2} &= \casesl{l}{ (-q)^{\bar \ell-1} \om^{-2}_{\ell+2} \\ q^{-2} \om_{n}^{-4} } && \casesm{l}{ \text{if } \ell+2<n, \\ \text{if } \ell+2=n, }
\\
s_{\ell+1} &= \frac{\nu + \nu^{-1}}{q-q^{-1}}\, \om^{-1}_{\ell+1} \om_{\ell+2} ,
\end{alignat*}
and repeating the same steps as before we obtain the following result.

\begin{result} \label{Res:C1-nstd}
When $1\le \ell\le n-3$, $r=\ell+2$ the bare K-matrix of type C.1  is 
\eq{
K(u) = \Id + \frac{u-u^{-1}}{k_1(u)} \bigg(\wt M_1(u) + \frac{k_2(u) \wt M_2(u) + ( \nu +\nu^{-1} ) \wt M_3(u) }{ k_2(-\nu^{-2} u) k_2(-\nu^2 u)} \bigg) \label{K:C1-nstd}
}
with $\la = q^{\bar r}$, $\mu = q^{-\ell-1}$ as in Result \ref{Res:C1} and 
\begin{align*} 
\wt M_1(u) &= \sum_{\bar \ell -1 \le i \le n} \big(\la \mu u E_{-i,-i} + E_{ii}\big), 
\\ 
\wt M_2(u) &= {-}\la E_{2-\bar\ell,2-\bar\ell} + \la^{-1} E_{\bar\ell-2,\bar\ell-2} + E_{2-\bar\ell,\bar\ell-2} - E_{\bar\ell-2,2-\bar\ell} \\
& \qq - \mu u  E_{1-\bar \ell,1-\bar \ell} + (\mu u)^{-1}  E_{\bar \ell-1,\bar \ell-1}  +  E_{1-\bar \ell,\bar \ell-1} - E_{\bar \ell-1,1-\bar \ell}\,, 
\\
\wt M_3(u) &= -E_{1-\bar\ell,2-\bar\ell}+E_{2-\bar\ell,1-\bar\ell} + \la^{-1}(E_{1-\bar\ell,\bar\ell-2}+E_{\bar\ell-2,1-\bar\ell}) \\
& \qq - (\mu u)^{-1} \big( E_{2-\bar\ell,\bar\ell-1}+E_{\bar\ell-1,2-\bar\ell} + \la^{-1}(E_{\bar\ell-2,\bar\ell-1}-E_{\bar\ell-1,\bar\ell-2}) \big).
\end{align*}
\end{result}

This K-matrix has four non-zero entries in the rows and columns labelled $ \pm (\bar\ell-1)$ and $\pm (\bar\ell-2)$.
For $\nu^2 =-1$ it specializes to the one in Result \ref{Res:C1} with $1 \le \ell < n - 2$ and $r=\ell+2$. This K-matrix has the following properties.

\smallskip

\begin{description}[itemsep=1ex]

\item[Eigendecomposition] 
\begin{align*} V&= \Id  + (\nu+\nu^{-1}) E_{2-\bar \ell , 1-\bar \ell} -\nu \big( E_{\bar \ell-1,\bar \ell-2} + E_{\bar \ell-2,\bar \ell-1} \big) \\
& \qq\; + \la\, \big( {-}\nu^2 E_{1-\bar \ell,\bar \ell-1} + E_{2-\bar \ell,\bar \ell-2} + E_{\bar \ell-2,2-\bar \ell} + E_{\bar \ell-1,1-\bar \ell} + \nu^{-1} E_{1-\bar \ell , \bar \ell-2 } - \nu E_{2-\bar \ell , \bar \ell - 1}\big) ,
\\[0.25em]
D(u) &= \!\sum_{1-\bar \ell \le i < \bar \ell-2}\! E_{ii} + p_1(u)\bigg( \tfrac{k_2(-\nu^2 u^{-1})}{k_2(-\nu^2 u)} E_{\bar \ell-2,\bar \ell-2} +  \tfrac{k_2(-\nu^{-2} u^{-1})}{k_2(-\nu^{-2} u)} E_{\bar \ell-1,\bar \ell-1} + \!\sum_{\bar \ell \le i \le n}\! (u^2 E_{-i,-i} + E_{ii})\! \bigg).
\end{align*}

\item[Rotations] \hfill $K^{\pi}(u) = -u \, K(u)$ if $\ell=\frac n2-1$. \hfill \hphantom{\it Rotations}
\item[Reductions] For $\nu^2 = (\la/\mu)^{\pm 1}$ it has $d_{\rm eff}=3$ and is singly regular, $K(-1)\ne \Id$.

\item[Diagonal cases] \qu $\displaystyle \lim_{\nu \to 0} K(u) = \lim_{\nu \to \infty} K(u) = p_1(u) \sum_{i \ge \bar\ell-1}  ( u^2 E_{-i,-i}+E_{ii}) + \sum_{|i|<\bar\ell-1} E_{ii}$. 
\end{description}


\subsubsection{Family BD.1 with $I_{\rm nsf}=\{j\}$ with $\al_{j}$ long}

Satake diagrams in this family are parametrized by $\ell$ and $r$ satisfying $1\le \ell< \lceil N/2 \rceil-2$ and $r=\ell+2$, so that $j=\ell+1$. Hence $X=X_1 \cup X_2$ with $X_1=\{0,\ldots,\ell-1\}$ if $\ell>1$ or $X_1=\emptyset$ if $\ell=1$, and $X_2=\{r+1,\ldots,n\}$.
When $\ell=1$, the QP algebra is generated by $x_i$, $y_i$, $k_i$ with $i\in X$, $(k_0k_1^{-1})^{\pm 1}$ and the elements $b_j$ whose reduced expressions are
\begin{gather*} 
b_2 = y_2 - c_2 x_2 k_2^{-1} - s_2 k_2^{-1}, \\
b_{j} = y_{j} - c_{1} \,x_{1-j}\, k_{j}^{-1} \text{ for } j\in\{0,1\}, \qq b_3 = y_3 - c_3 T_{w_{X_2}}(x_3) k_3^{-1}. 
\end{gather*}
When $1<\ell<\lceil N/2 \rceil -1$, the QP algebra is generated by $x_i$, $y_i$, $k_i$ with $i \in X$ and
\begin{gather*} 
b_{\ell+1} =  y_{\ell+1} - c_{\ell+1} x_{\ell+1} k_{\ell+1}^{-1} - s_{\ell+1} k_{\ell+1}^{-1}, \\
b_\ell = y_\ell - c_\ell T_{w_{X_1}}(x_\ell) k_\ell^{-1}, \qq b_{\ell+2} = y_{\ell+2} - c_{\ell+2} T_{w_{X_2}}(x_{\ell+2}) k_{\ell+2}^{-1}.
\end{gather*} 
Upon setting the parameters $c_j$ in the same way as in the generic case, {\it cf.}~Section \ref{sec:BD1},  
\[
c_\ell = (-q)^{\ell-1} \eta^{-1} \om_{\ell+1}^2 , \qu 
c_{\ell+1} = q^{-1} \om^{-2}_{\ell+1} \om^2_{\ell+2} , \qu
c_{\ell+2} = (-q)^{\bar \ell-3} \om^{-2}_{\ell+2} , 
\]
and $s_{\ell+1} = \frac{\nu + \nu^{-1}}{q-q^{-1}} \om_{\ell+2} \om^{-1}_{\ell+1}$, and repeating the same steps as before we obtain the following result.

\begin{result} 
When $1\le \ell\le n-2$, $r=\ell+2$ the bare K-matrix of type BD.1 is given by the formula \eqref{K:C1-nstd} only with $\la = q^{N/2 - \ell -2}$, $\mu = q^{-\ell}$ as in Result \ref{res:BD1} and 
\begin{align*} 
\wt M_1(u) &= \sum_{\bar \ell -1 \le i \le n} (\la \mu u E_{-i,-i} + E_{ii}), \\ 
\wt M_2(u) &= \la E_{2-\bar\ell,2-\bar\ell} + \la^{-1} E_{\bar\ell-2,\bar\ell-2} + E_{2-\bar\ell,\bar\ell-2} + E_{\bar\ell-2,2-\bar\ell} + \\
& \qq - \mu u  E_{1-\bar \ell,1-\bar \ell} - (\mu u)^{-1}  E_{\bar \ell-1,\bar \ell-1}  +  E_{1-\bar \ell,\bar \ell-1} + E_{\bar \ell-1,1-\bar \ell}, \\
\wt M_3(u) &= E_{1-\bar\ell,2-\bar\ell}-E_{2-\bar\ell,1-\bar\ell} + \la^{-1}(E_{1-\bar\ell,\bar\ell-2}-E_{\bar\ell-2,1-\bar\ell}) + \\
& \qq + (\mu u)^{-1} \big( E_{2-\bar\ell,\bar\ell-1}-E_{\bar\ell-1,2-\bar\ell} +  \la^{-1}(E_{\bar\ell-2,\bar\ell-1}-E_{\bar\ell-1,\bar\ell-2}) \big).
\end{align*}
\end{result}

This K-matrix is very similar to that in Result \ref{Res:C1-nstd}. For $\nu^2=-1$ it specializes to the one in Result \ref{res:BD1} with $1 \le \ell < \lceil N/2 \rceil - 2$ and $r=\ell+2$. We list below only those properties that differ from the ones above.

\smallskip

\begin{description}[itemsep=1ex]

\item[Eigendecomposition] $D(u)$ is the same as for the non-quasistandard C.1 case, but $V$ is given by
\begin{align*} V&= \Id + \la \big( {-}\nu^2 E_{1-\bar \ell,\bar \ell-1} + E_{2-\bar \ell,\bar \ell-2} - E_{\bar \ell-2,2-\bar \ell} - E_{\bar \ell-1,1-\bar \ell} + \nu^{-1} E_{1-\bar \ell , \bar \ell-2 } - \nu E_{2-\bar \ell , \bar \ell - 1}\big) \\
& \qq\; + (\nu+\nu^{-1}) E_{2-\bar \ell , 1-\bar \ell} -\nu \big( E_{\bar \ell-1,\bar \ell-2} + E_{\bar \ell-2,\bar \ell-1} \big).
\end{align*}

\item[Rotations] There are additional rotations $K^{\flL}(u) = K(u)$ and $K^{\flR}(u) = K(u) $ if $N$ and even.

\item[Reductions] The singly regular K-matrix has $K(-1)\ne \Id$ instead. 

\end{description}


\subsubsection{Family B.1 with $I_{\rm nsf}=\{j\}$ with $\al_j$ short}

Satake diagrams in this family are parametrized by $(\ell,r)=(n\!-\!1,n)$ with $N$ odd. We have $j=\ell+1$ and $X=\{0,1,\ldots,n-2\}$. The QP algebra is generated by $x_i$, $y_i$, $k_i$ with $i\in X$ and
\[
b_{n-1} = y_{n-1} - c_{n-1} T_{w_{X_1}}(x_{n-1})\, k_{n-1}^{-1}, \qq b_n = y_n - c_n x_n k_n^{-1} - s_n k_n^{-1} . 
\]
Upon setting $c_{n-1}$, $c_n$, $s_n$ as
\[
c_{n-1} = - (-q)^{n-2} \eta^{-1} \om_{n}^2 , \qq c_n = - \om_n^{-2} , \qq s_n =  q^{1/4} \frac{\nu - \nu^{-1}}{q^{1/2}-q^{-1/2}} \,\om_n^{-1}, 
\]
where $\nu \in \K^\times$ (for $c_{n-1}$ and $c_n$ this is the same as in Section \ref{sec:BD1}), and proceeding in the same way as before we obtain the following result.

\begin{result} 
When $(\ell,r)=(n-1,n)$ the bare K-matrix of type B.1 is 
\spl{
K(u) &= \Id + \frac{u-u^{-1}}{k_1(u)} \Bigg(\wt M_1(u) + \frac{k_1(u^{-1})^2 \wt M_2(u) + q^{1/4} (\nu-\nu^{-1})[2]^{1/2}_{q^{1/2}} \wt M_3(u)}{k_1(\nu^2 u^{-1})\,k_1(\nu^{-2} u^{-1})\,k_2(u) }  \Bigg) 
}
with $\la  = q^{1/2}$, $\mu  = q^{1-n}$ as in Result~\ref{res:BD1} and
\eqn{ 
\wt M_1(u) &= \sum_{1 \le i \le n} (\la \mu u E_{-i,-i} + E_{ii}), \\ 
\wt M_2(u) &= - \mu u E_{-1,-1} - (\mu u)^{-1} E_{11} + E_{-1,1} + E_{1,-1}, \\
\wt M_3(u) &= k_1 ( u^{-1}) \big(u^{-1} (E_{01} + E_{10}) - \mu (E_{-1,0}+E_{0,-1}) \big) -  q^{1/4} (\nu-\nu^{-1}) [2]_{q^{1/2}}^{1/2} \mu u^{-1}  E_{00}.
}
\end{result}

This K-matrix has three nonzero entries in the rows and columns labelled $-1$, 0 and 1. (Note the appearance of the factor $[2]_{q^{1/2}}^{1/2}$, {\it cf.}~\eqref{rep:BCDn}.)
For $\nu^2=1$ (so that $s_n=0$) it specializes to the one in Result \ref{res:BD1} with $N$ odd and $(\ell,r)=(n-1,n)$. 
This K-matrix has the following properties.

\bigskip

\begin{description}[itemsep=1ex]

\item[Eigendecomposition] 
\begin{align*} V &= \Id + \la \bigg( E_{-1,1} - \nu^2 E_{1,-1} \\ 
& \hspace{1.9cm}+ q^{-1/4} [2]_{q^{1/2}}^{-1/2} \Big( \mu \big( (\nu-\nu^{-1}) E_{01} + [2]_{q^{1/2}} \nu E_{0,-1} \big)  + \mu^{-1} \big( \nu^{-1} E_{10} - \la^{-1} \nu E_{-1,0} \big) \Big) \bigg), 
\\
D(u) &=p_1(u) \bigg( \tfrac{k_1(\nu^{-2} u)}{k_1(\nu^{-2} u^{-1})} E_{-1,-1} + \tfrac{k_1(\nu^2 u)}{k_1(\nu^2 u^{-1})} E_{00} + p_2(u) E_{11} + \sum_{\bar \ell \le i \le n} \big( u^2 E_{-i,-i} + E_{ii} \big) \bigg).
\end{align*}

\item[Rotations] \hfill $K^{\flL}(u) = K(u)$. \hfill \hphantom{\it Rotations}

\item[Diagonal cases] \qu $\displaystyle\lim_{\nu \to 0} K(u) = \lim_{\nu \to \infty} K(u) = p_1(u) \Big( p_2(u) E_{00} + \sum_{i=1}^n  ( u^2 E_{-i,-i}+E_{ii}) \Big)$. \hfill \hphantom{\it Diagonal cases}
\end{description}


\subsection{Special low-rank cases} \label{sec:K:low-rank}


Here we highlight some peculiarities that occur when the underlying Lie algebra is of low rank. In Appendix \ref{App:LowRank} we have listed all isomorphisms of ($\Sigma_A$-classes of) affine Satake diagrams of low rank ($n \le 4$). 
When two affine Kac-Moody algebras are isomorphic, necessarily the rank $n$ is the same for both. However the value for $N$ (the dimension of the vector representation) can be different. In particular, $N=2$ for ${\rm A}^{(1)}_1$ and ${\rm C}^{(1)}_1$ but $N=3$ for ${\rm B}^{(1)}_1$. Similarly, $N=4$ for ${\rm C}^{(1)}_2$ but $N=5$ for ${\rm B}^{(1)}_2$. Finally, $N=4$ for ${\rm A}^{(1)}_3$ but $N=6$ for ${\rm D}^{(1)}_3$. 
For simplicity, here we will focus on the cases when vector spaces have the same dimensions.


\subsubsection{Untwisted cases of $\rm{A}^{(1)}_1$ and $\rm{C}^{(1)}_1$} \label{sec:n=1:untw}

The solutions of the untwisted RE \eqref{RE} involving Baxter's R-matrix \eqref{Ru:Baxter} are given by the low-rank cases of the formulas for K-matrices of type A.3 in Result \ref{Res:A3}, type C.1 in Result \ref{Res:C1}, type C.4 in Result \ref{Res:CD4} and type C.2 in Result \ref{Res:C2}, upon a suitable choices of the classification parameters and relabelling of the basis vectors. 
Note that these K-matrices are independent of $q$, thus the substitution $q \to q^{1/2}$ required to identify R-matrices of types $\rm{A}^{(1)}_1$ and $\rm{C}^{(1)}_1$ does not affect the K-matrices. 
Here we review these K-matrices starting with the most simple case and working our way up to the K-matrix of the most general form. \\

\noindent $\bigl( \text{A}^{(1)}_1 \bigr)^{\psi}_{0,1}  \cong \bigl( \text{C}^{(1)}_1 \bigr)^{\text{id}}_{\rm 0,alt,1}$
\begin{minipage}[c]{12mm} 
\begin{tikzpicture}[baseline=-0.25em,scale=0.2,line width=0.7pt]
	\draw[double] (0,0) -- (2,0) ;
	\filldraw[fill=white] (0,0) circle (.5) node[left=1pt]{\small $0$}; 
	\filldraw[fill=black] (2,0) circle (.5) node[right]{\small $1$}; 
\end{tikzpicture} 
\end{minipage} \\

\noindent The Satake diagram is $(X,\tau) = (\{1\}, \id)$. The bare K-matrix in this case is simply 
\eq{ 
\label{n=1:id} K(u) = \Id. 
}
By rotating with $\rho = (01)$ we obtain $(X^\rho,\tau^\rho)=(\{0\},\id)$ and the K-matrix becomes 
$$ 
\qq K^{(01)}(u) = Z^{(01)}(\tfrac{\eta}{u})^{-1} K(u)\, Z^{(01)}(\eta u)  = u^2 E_{11} + E_{22} .
$$

\noindent $\bigl( \text{A}^{(1)}_1 \bigr)^{\psi'}_{0,0} \cong \bigl( \text{C}^{(1)}_1 \bigr)^{\pi}_{0}$
\begin{minipage}[c]{12mm} 
\begin{tikzpicture}[baseline=-0.25em,scale=0.2,line width=0.7pt]
	\draw[double] (0,0) -- (2,0) ;
	\filldraw[fill=white] (0,0) circle (.5) node[left=1pt]{\small $0$}; 
	\filldraw[fill=white] (2,0) circle (.5) node[right]{\small $1$}; 
	\draw [<->,gray] (2,0.6) arc[x radius=1cm, y radius =1cm, start angle=0, end angle=180];
\end{tikzpicture}
\end{minipage} \\

\noindent The Satake diagram is now $(X,\tau) = (\emptyset, (01))$. The bare K-matrix has a free parameter $\xi\in \K^\times$:
\eq{  \label{n=1:gen-diag}
\qq K(u) =E_{11} + \frac{\xi - u^{-1}}{\xi - u} E_{22} .
}
This K-matrix was one of several found by Cherednik in \cite{Ch1} and featured in Sklyanin's adaptation of the algebraic Bethe ansatz to quantum integrable systems with boundaries \cite{Sk}.

\medskip

\noindent $\bigl( \text{A}^{(1)}_1 \bigr)^{\psi}_{0,0} \cong \bigl( \text{C}^{(1)}_1 \bigr)^{\text{id}}_{0,0}$
\begin{minipage}[c]{12mm} 
\begin{tikzpicture}[baseline=-0.25em,scale=0.2,line width=0.7pt]
	\draw[double] (0,0) -- (2,0) ;
	\filldraw[fill=white] (0,0) circle (.5) node[left=1pt]{\small $0$}; 
	\filldraw[fill=white] (2,0) circle (.5) node[right]{\small $1$}; 
\end{tikzpicture} \end{minipage} \\

\noindent The Satake diagram is $(X,\tau) = (\emptyset, \id)$. The bare K-matrix has two free parameters:
\eq{ \label{n=1:gen-nondiag}
K(u) = \frac{(\mu-\mu^{-1})\Id + (\la - \la^{-1})(u E_{11} + u^{-1} E_{22}) + (u-u^{-1})(E_{12} + E_{21})}{(\la \mu -u)(\la^{-1} + (\mu u)^{-1})}.
}
This K-matrix was first obtained in \cite{dVGR}. 

\begin{rmk}
Specialize $q$ to be a nonzero complex number which is not a root of unity and consider the reflection equation \eqref{RE} with the ${\rm A^{(1)}_1}$-type R-matrix.
Then it can be checked directly that $K(u)$ given by \eqref{n=1:gen-diag} is the most general meromorphic diagonal solution.
Here we have ignored the limit cases $\xi \to 0$, $\xi \to \infty$, which up to a scalar yield the two solutions $\Id$ and $u^2 E_{11}+E_{22}$ above.
Also, $K(u)$ given by \eqref{n=1:gen-nondiag} is, upon dressing, the most general nondiagonal meromorphic solution.
\end{rmk}


\subsubsection{Twisted cases of $\rm{A}^{(1)}_1$ and $\rm{C}^{(1)}_1$} \label{sec:n=1:tw}

Solutions of the twisted RE \eqref{tRE} can be obtained from the ones above using Proposition \ref{prop:intw:tw-untw} as follows. \\

\noindent {\it From} $\bigl( \text{A}^{(1)}_1 \bigr)^{\psi}_{0,1}$ {\it to}  $\bigl( \text{A}^{(1)}_1 \bigr)^{\id}_{\text{alt}}${\it :} \\

\noindent The K-matrix $K(u)=\Id$ ({\it cf.} \eqref{n=1:id}) corresponds to $\wt K(u) = C$, which is indeed what formula \eqref{A2:K} yields in the case $n=1$. \\

\noindent {\it From} $\bigl( \text{A}^{(1)}_1 \bigr)^{\psi'}_{0,0}$ {\it to} $\bigl( \text{A}^{(1)}_1 \bigr)^\rho${\it :} \\

\noindent The diagonal K-matrix \eqref{n=1:gen-diag} corresponds the twisted K-matrix
\eq{ 
\label{A4:K:n=1} \wt K(u) = C K (q u ) \propto (u^{-1} + q \xi) E_{12} - (q u + \xi) E_{21}, 
}
which has one more free parameter compared to the K-matrix in \eqref{A4:K}. 
This owes to the fact that the set $I_\text{diff}$ is empty in the case when $n=1$, so that an additional parameter plays a role in the intertwining equation.
By setting $\xi = -1$ in \eqref{A4:K:n=1} one recovers (a scalar multiple of) the matrix $K(u) = u\,E_{21} + E_{12}$, the low-rank limit of the formula \eqref{A4:K}.\\

\noindent  {\it From} $\big( \text{A}^{(1)}_1 \big)^{\psi}_{0,0}$ {\it to} $\big( \text{A}^{(1)}_1 \big)^{\id}_0${\it :}  \\

\noindent The nondiagonal bare K-matrix discussed above leads to the following solution of the twisted RE: 
\begin{equation} 
\label{A1:K:n=1}  
\begin{aligned}
\wt K(u) &= C K (q u ) \\
&  \propto  (q u - (qu)^{-1})(q^{1/2} E_{11}- q^{-1/2}E_{22}) \\
& \qu - (q^{1/2}(\mu^{-1} - \mu) + q^{-1/2} (\la^{-1} - \la) u^{-1}) E_{12} \\
& \qu + (q^{-1/2} (\mu^{-1} - \mu) + q^{1/2} (\la^{-1} - \la)u) E_{21} .  
\end{aligned}
\end{equation}
Upon dressing, this K-matrix has three free parameters. However, the $n=1$ analogue of the K-matrix of type A.1 found in Section \ref{sec:A1}, namely $\om_1^2 E_{11} + \om_2^2 E_{22}$, only has one (upon multiplying the K-matrix by a scalar). It can be obtained from \eqref{A1:K:n=1} by setting $\la^2 = \mu^2 = 1$ and left- and right-multiplying by the appropriate $G(\omega_1,\omega_2)$, {\it cf.} \eqref{eq:Kdressing}.

In Appendix \ref{App:qOns} we will discuss the q-Onsager K-matrix for $\text{A}^{(1)}_{n>1}$; the one above can be viewed as a special case of the q-Onsager case provided we take into account the following subtlety. Uniquely for $n=1$ among the A.1 family, $I_{\rm nsf} = I$ and hence the generators $b_0$, $b_1$ have four independent parameters $c_0$, $c_1$, $s_0$, $s_1$ that enter the intertwining relation (in the q-Onsager case for $n>1$ a relation is imposed between $c_j$ and $s_j$ for $j \in I$). This exceptional status for $n=1$ where $a_{01} = a_{10} = -2$ is accounted for in the description by Baseilhac and Belliard, see \cite[Prop.~2.1]{BsBe1}.

In terms of dressing parameters $\omega_1,\omega_2$, additional free parameters $\la,\mu$ and the scaling parameter $\eta$ set 
\[
c_0 = - q\,\eta^{-2} \om_1^{-2}\om_2^2, \qu
c_1 = - q^{-1} \om_1^2 \om_2^{-2}, \qu
s_0 = q\,\frac{\mu-\mu^{-1}}{q-q^{-1}} \, \eta^{-1}\om_1^{-1}\om_2, \qu
s_1 = \frac{\la-\la^{-1}}{q-q^{-1}}\, \om_1\,\om_2^{-1}.  
\]
Solving the boundary intertwining equation we find, up to dressing by $G(\om_1,\om_2)$ and multiplication by a scalar, the K-matrix given by \eqref{A1:K:n=1}.


\subsubsection{Exceptional cases of $\rm{D}^{(1)}_4$} \label{sec:D4}

Let $(I,A)$ be of type $\rm{D}^{(1)}_4$. There are three $\Sigma_A$-equivalence classes of Satake diagrams of type D.3. We consider their representatives $(X,\tau)$ with $\tau = (14)$ and $X=\emptyset$, $X=\{1,2,4\}$ or $X=\{0,3\}$, which are $\Aut(A)$-equivalent to the Satake diagrams $(\emptyset,(34))$, $(\{2,3,4\},(34))$ and $(\{0,1\},(34))$ respectively (see Table \ref{Tbl:isos:4b} in Appendix \ref{App:LowRank}), which were studied in Sections \ref{sec:K:C1BD2}-\ref{sec:K:0}. However as they are not $\Sigma_A$-equivalent, these exceptional Satake diagrams need to be considered separately.

In particular, we need to solve the boundary intertwining equation \eqref{intw-untw} for the exceptional QSP algebras $B_{\bm c,\bm s}$ whose details are listed in Table \ref{Tbl:so8}. By doing so in each case we find that the only solution of \eqref{intw-untw} is the trivial solution, $K(u)=0$. 
For example, in the case with $X=\{0,3\}$ this can be established as follows. First, note that the representation $\RT_u$ of generators of $U_q(\hat \mfh)$ is given~by
\eqn{ 
& \RT_u(k_1 k_4^{-1}) = \sum_{1\le i\le3} (q^{-1} E_{-i,-i} +q\, E_{ii} ) + q\, E_{-4,-4} + q^{-1} E_{44}, \\
& \RT_u(k_0) = q^{-1} (E_{-4,-4} + E_{-3,-3}) + q\, (E_{33} + E_{44}) + \sum_{1\le i\le2} (E_{-i,-i} + E_{i,i}), \\
& \RT_u(k_3) = q^{-1} ( E_{-1,-1} + E_{22}) + q\, (E_{-2,-2}+E_{11}) + \sum_{3\le i\le4} (E_{-i,-i}+E{i,i}).
}
Then observe that for all $i,j \in \langle 4 \rangle$ with $i \ne j$ the $(i,i)$- and $(j,j)$-entries are distinct for at least one of three matrices above. 
Hence the relation $K(u)\,\RT_{\eta u}(b)=\RT_{\eta/u}(b)\,K(u)$ for $b \in \{ k_1 k_4^{-1} , k_0, k_3 \}$ implies that $K(u)$ is a diagonal matrix, say $K(u) = \sum_{i \in \langle 4 \rangle} k^{(i)}(u) E_{ii}$ with some $k^{(i)}(u)\in\K(u)$.
Next, the representation $\RT_u$ of generators $x_j$ for $j\in\{0,3\}$ and $b_j$ for $j\in\{1,2\}$ is given by
\eqn{ 
& \RT_u(x_3) = E_{-2,-1} - E_{12}, && \RT_u(b_1) = E_{-3,-4} - E_{43} + c_1 (E_{-1,2} - E_{-2,1}), \\
& \RT_u(x_0) = u\,(E_{3,-4} - E_{4,-3}), && \RT_u(b_2) = E_{-2,-3} - E_{32} + c_2\, q^{-1} u\,(E_{1,4} - E_{-1,4}) .
}
By a direct computation we see that $K(u)\,\RT_{\eta u}(x_2)=\RT_{\eta/u}(x_2)\,K(u)$ implies that $k^{(\pm2)}(u) = k_{(\pm1)}(u)$. 
Likewise, equality $K(u)\,\RT_{\eta u}(b_1)=\RT_{\eta/u}(b_1)\,K(u)$ further implies that $k^{(\pm 4)}(u) = k^{(\pm 3)}(u)$ and $k^{(1)}(u) = k^{(-1)}(u)$, and $K(u)\,\RT_{\eta u}(x_0)  = \RT_{\eta/u}(x_0)\,K(u)$ implies that $k^{(-3)}(u) = u^2 k^{(3)}(u)$. 
It remains to compute $K(u)\,\RT_{\eta u}(b_2) = \RT_{\eta/u}(b_2)\,K(u)$, giving two more relations, $k^{(-1)}(u) = k^{(3)}(u)$ and $k^{(-1)}(u) = u^2 k^{(3)}(u)$, which are only true if $k^{(-1)}(u) = k^{(3)}(u) = 0$. Hence $K(u)=0$.

{\def\arraystretch{1.2}
\begin{table}
\caption{Exceptional QSP algebras for ${\rm D}^{(1)}_4$} \label{Tbl:so8}
\begin{tabular}{lll}
\hline
Name & Diagram & Generators \\ \hline
$({\rm D}^{(1)}_4)^{(14)}_{\emptyset}$ &
\begin{minipage}[c]{17mm}
\begin{tikzpicture}[scale=80/100]
\draw[thick] (-.6,.3) -- (0,0) -- (-.4,-.3);
\draw[thick] (.4,.3) -- (0,0) -- (.6,-.3);
\draw[<->,gray] (-.3,-.3) -- (.5,-.3);
\filldraw[fill=white] (-.6,.3) circle (.1) node[left]{\scriptsize 0};
\filldraw[fill=white] (-.4,-.3) circle (.1) node[left]{\scriptsize 1};
\filldraw[fill=white] (0,0) circle (.1) node[above]{\scriptsize 2};
\filldraw[fill=white] (.6,-.3) circle (.1) node[right]{\scriptsize 4};
\filldraw[fill=white] (.4,.3) circle (.1) node[right]{\scriptsize 3};
\end{tikzpicture}
\end{minipage} 

& 

\begin{minipage}[c]{10cm}
\vspace{1mm}
$(k_1 k_4^{-1})^{\pm 1}$,\\
$b_j = y_j - c_j \, x_j \,k_j^{-1} \text{ for } j\in\{0,2,3\}$, \\
$b_j = y_j - c_j \, x_{5-j} k_j^{-1} \text{ for } j\in\{1,4\}$, \\
\hphantom{\qquad} with $c_0,c_1,c_2,c_3,c_4 \in \K^\times$ such that $c_1 = c_4$. \\[-3mm]
\end{minipage} 

\\\hline

$({\rm D}^{(1)}_4)^{(14)}_{\{1,2,4\}}$ & 
\begin{minipage}[c]{17mm}
\begin{tikzpicture}[scale=80/100]
\draw[thick] (-.6,.3) -- (0,0) -- (-.4,-.3);
\draw[thick] (.4,.3) -- (0,0) -- (.6,-.3);
\draw[<->,gray] (-.3,-.3) -- (.5,-.3);
\filldraw[fill=white] (-.6,.3) circle (.1) node[left]{\scriptsize 0};
\filldraw[fill=black] (-.4,-.3) circle (.1) node[left]{\scriptsize 1};
\filldraw[fill=black] (0,0) circle (.1) node[above]{\scriptsize 2};
\filldraw[fill=black] (.6,-.3) circle (.1) node[right]{\scriptsize 4};
\filldraw[fill=white] (.4,.3) circle (.1) node[right]{\scriptsize 3};
\end{tikzpicture}
\end{minipage} 

& 

\begin{minipage}[c]{10.5cm}
\vspace{1mm}
$x_i,\;y_i,\;k_i \text{ for } i \in \{1,2,4\}$,\\
$b_j = y_j - c_j\, T_2 T_1 T_4 T_2 (x_j)\,k_j^{-1} \text{ for } j\in\{0,3\}$, with $c_0,c_3 \in \K^\times$. \\[-2.5mm]
\end{minipage}

\\\hline

$({\rm D}^{(1)}_4)^{(14)}_{\{0,3\}}$ &
\begin{minipage}[c]{17mm}
\begin{tikzpicture}[scale=80/100]
\draw[thick] (-.6,.3) -- (0,0) -- (-.4,-.3);
\draw[thick] (.4,.3) -- (0,0) -- (.6,-.3);
\draw[<->,gray] (-.3,-.3) -- (.5,-.3);
\filldraw[fill=black] (-.6,.3) circle (.1) node[left]{\scriptsize 0};
\filldraw[fill=white] (-.4,-.3) circle (.1) node[left]{\scriptsize 1};
\filldraw[fill=white] (0,0) circle (.1) node[above]{\scriptsize 2};
\filldraw[fill=white] (.6,-.3) circle (.1) node[right]{\scriptsize 4};
\filldraw[fill=black] (.4,.3) circle (.1) node[right]{\scriptsize 3};
\end{tikzpicture}
\end{minipage} 

& 

\begin{minipage}[c]{12cm}
\vspace{1mm}
$(k_1 k_4^{-1})^{\pm 1} \text{ and }x_i,\;y_1,\;k_i \text{ for } i \in \{ 0,3\}$,\\
$b_2 = y_2 - c_2 \, T_0 T_3 (x_2)\, k_2^{-1}$, \\
$b_j = y_j - c_j \,x_{5-j}\, k_j^{-1} \text{ for } j\in\{1,4\}$, with $c_1,c_2,c_4 \in \K^\times$ such that $c_1 = c_4$. \\[-2.5mm]
\end{minipage}
\\\hline 
\end{tabular}
\end{table}
}

\begin{rmk} \label{rmk:nosolutions} 
It would be interesting to understand if the absence of nontrivial solutions to \eqref{intw-untw} for the exceptional coideal subalgebras $B_{\bm c,\bm s}$ described in Table \ref{Tbl:so8} is an artefact of the vector representation $\RT_u$ or it is also true for other finite-dimensional representations of $U_q(\wh{\mfso}_8)$ that remain irreducible when restricted to $B_{\bm c,\bm s}$. \hfill \rmkend
\end{rmk}


\section{Outlook and applications} \label{sec:conclusions}


In this section we apply the results obtained to related topics in mathematical physics and representation theory. 


\subsection{Quantum integrable systems with boundaries} \label{sec:qintsys}
We briefly review some existing results that place the right REs \eqref{RE} and \eqref{tRE} and counterparts of them for the left boundary in the context of quantum integrable systems using the example of a quantum Heisenberg spin chain of length $k \in \Z_{>0}$ (see {\it e.g.}~\cite{Ba2}). 
Such a system describes particles labelled by $1,2,\ldots,k$, each with $N \in \Z_{>0}$ internal degrees of freedom called ``spin'' ({\it e.g.}~magnets), which are arranged on a line.
Since it is a quantum-mechanical system, the space of states is a vector space, namely $(\C^N)^{\otimes k}$.

The particles are assumed to obey nearest-neighbour interactions only. 
Such interactions can be encoded in a solution of the quantum Yang-Baxter equation \eqref{YBE}, {\it i.e.}~an R-matrix. 
If particles 1 and $k$ are not considered to be neighbours the chain is called \emph{open}. 
In particular, the spin chain may be thought of as being between two boundaries in which case the particle-boundary interaction is related to a choice of solution of an appropriate reflection equation, {\it i.e.}~a K-matrix.
The integrability of such a system can be characterized using different frameworks. 
We will consider here the quantum Knizhnik-Zamolodchikov (qKZ) equations, the transfer matrix and the Hamiltonian.
The integrability in each formalism always follows from just the Yang-Baxter and appropriate reflection equations; the relations between the different formalisms require unitarity and regularity conditions on the R- and K-matrices.
Therefore it applies to all K-matrices classified in this paper except possibly where the regularity property $K(\pm 1) = \Id$ is needed. 
One may distinguish between boundary conditions given by untwisted and twisted K-matrices;
we will first treat the untwisted case in detail and briefly review the twisted case in Section \ref{sec:integrability:tw}.

To simplify the presentation the indeterminate $q$ is specialized to a nonzero complex number which is not a root of unity. Furthermore we assume that we are given matrices $R(u)$, $K^+(u)$, $K^-(u)$, meromorphically depending on $u$ such that $R(u) \in \End(\C^N \ot \C^N)$ satisfies the Yang-Baxter equation \eqref{YBE}, $K^+(u) \in \End(\C^N)$ satisfies the right reflection equation \eqref{RE} and $K^-(u) \in \End(\C^N)$ satisfies the \emph{left reflection equation}:
\eq{ 
R(\tfrac{u}{v})\, K^-_1(u)\, R_{21}(uv)\, K^-_2(v) = K^-_2(v)\, R(uv)\, K^-_1(u)\, R_{21}(\tfrac{u}{v}).  \label{LRE}
}
Using the fact that image of a right coideal subalgebra under the antipode (or its inverse) is a left coideal subalgebra, the entire construction of solutions of \eqref{RE} followed in this paper in terms of right coideal subalgebras can be modified to left coideal subalgebras in order to produce solutions of \eqref{LRE}.
However, it can be checked that if $K^+(u)$ is a solution of \eqref{RE}, then $K^-(u)=C^{-1} K^+(u)^\t C$ is a solution of \eqref{LRE} (and all solutions of \eqref{LRE} arise in this way). 
Thus, assuming $R(u)$ is one of the R-matrices discussed in Section \ref{sec:Ru}, Section \ref{sec:Results} provides a pool of solutions of both \eqref{RE} and \eqref{LRE}.
Although the two given K-matrices $K^+(u)$, $K^-(u)$ are necessarily associated to coideal subalgebras of the same affine quantum group, these two coideal subalgebras do not need to be related to each other in any further way (in particular the underling generalized Satake diagrams do not need to be the same).


\subsubsection{qKZ equations}

The qKZ equations are noteworthy difference equations appearing in a wide range of contexts in mathematical physics, representation theory and beyond; in their original form they are defined in terms of R-matrices only \cite{DFZJ1,EFK,FrRt,JiMi,Sm,TaVa}. 
Cherednik \cite{Ch2,Ch3} studied generalizations of these equations defined in terms of an R-matrix datum associated to an arbitrary affine root system; taking this to be of type A one recovers the aforementioned equations. 
If instead we choose the type to be B, C or D, we obtain the boundary qKZ equations.

More precisely, for $p \in \C^k$, the qKZ equations are the (multiplicative) difference equations for meromorphic functions $f:\C^k \to (\C^N)^{\otimes k}$ defined by
\eq{
f(z_1, \ldots, p z_r, \ldots, z_k) = \mc{A}_r(\bm z;p) f(\bm z) \qq \text{for } 1\leq r \leq k. \label{bqKZ}
}
Here the \emph{transport matrices} $\mc{A}_r(\bm z;p)$ are particular elements of $\End((\C^N)^{\otimes k})$, depending meromorphically on $\bm z$ and $p^{1/2}$ and satisfying the condition
\eq{ \label{qKZconsistency} 
\mc{A}_r(z_1, \ldots, p z_s, \ldots, z_k;p)\mc{A}_s(\bm z;p) = \mc{A}_s(z_1, \ldots, p z_r, \ldots, z_k;p)\mc{A}_r(\bm z;p) \qq \text{for } 1 \le r,s \le k;
}
consequently the system \eqref{bqKZ} is consistent (and thus defines a ``discrete connection'').
For the qKZ transport matrices at $p=1$, also called \emph{scattering matrices}, \eqref{qKZconsistency} turns into an ordinary integrability (i.e. commutativity) statement. The qKZ equations can then be seen as $p$-deformations of ordinary eigenvector equations of the scattering matrices (with eigenvalue 1).

Specifically, for the boundary qKZ equations, we have
\eq{ \label{bqKZmatrix}
\begin{aligned}
\mc{A}_r(\bm z;p) &=  R_{r,r-1}(\tfrac{pz_r}{z_{r-1}}) \cdots R_{r1}(\tfrac{pz_r}{z_1}) K^-_r(p^{1/2}z_r) R_{1r}(z_1z_r) \cdots R_{r-1,r}(z_{r-1}z_r) \\
& \qq  \times  R_{r+1,r}(z_rz_{r+1}) \cdots  R_{kr}(z_rz_k)  K^+_r(z_r) R_{kr}(\tfrac{z_k}{z_r})^{-1} \cdots R_{r+1,r}(\tfrac{z_{r+1}}{z_r})^{-1}.
\end{aligned}
}
In this case the consistency condition \eqref{qKZconsistency} is a consequence of \eqref{YBE}, \eqref{RE} and \eqref{LRE}.

Solutions of the boundary qKZ equations in special cases exist in various forms. They are best understood for representations of $U_q(\wh{\mfsl}_2)$.
Correlation functions of $U_q(\wh{\mfsl}_2)$-vertex operators with respect to so-called boundary states were obtained in \cite{JKKKM,JKKMW}.
In other types of solutions the importance of the representation theory of affine Hecke algebras of type B, C or D comes to the fore.
For example, Laurent polynomial solutions in terms of nonsymmetric Macdonald-Koornwinder polynomials were found in \cite{StVl}; when $p$ is a specific rational power of $q$, polynomials solutions involve connections with combinatorics, see \cite{DFZJ2} and references therein.
Finally, if the underlying QSP algebra is of type $\big(\rm{A}^{(1)}_{1}\big)^{\psi'}_{0,0} \cong \big(\rm{C}^{(1)}_1 \big)^{\pi}_0$, solutions in the form of Jackson integrals (bilateral series) and integrals have been obtained \cite{RSV1,RSV2,RSV3}; these solutions are in terms of boundary Bethe vectors as introduced by Sklyanin \cite{Sk} and thus K-matrices and R-matrices themselves. 

The formalism of commuting transfer matrices originates in the work of Baxter \cite{Ba2} on two-dimensional vertex models of statistical mechanics. Of particular relevance is his result that the transfer matrix of a two-dimensional vertex model is simply related to the Hamiltonian of a quantum spin chain \cite{Ba1}.
In the 1970s and 1980s the concept of commuting transfer matrices was developed further by the St. Petersburg school \cite{Fa,KuSk1}. This directly led to the RTT formulation of quantum groups \cite{FRT}.
Commuting transfer matrices for reflecting systems can be constructed from solutions of the Yang-Baxter equation and appropriate reflection equations, following the ideas and conventions used in \cite{Sk,MeNe1}.

To formalize this, it is convenient to consider the larger vector space $\C^N \ot (\C^N)^{\ot k}$ where the solitary tensor factor, called \emph{auxiliary space}, is labelled 0; for any $M \in \End(\C^N \ot (\C^N)^{\ot k})$ denote by $\Tr_0 M \in \End( (\C^N)^{\ot k})$ the relative trace with respect to the auxiliary space.
Define $K^{\rm d}(u) \in \End(\C^N)$ meromorphically depending on $u$ by means of
\eq{\label{Krelation}
 K^-_1(u) = \Tr_0 K^{\rm d}_0(u) P_{01}  R_{01}(u^2) \qq \text{or, equivalently,} \qu K^{\rm d}_1(u) = \Tr_0 K^-_0(u) P_{01} \wt R_{01}(u^2) .
}
As a consequence of \eqref{YBE}, \eqref{Krelation} provides a one-to-one correspondence between meromorphic solutions of the left RE \eqref{LRE} and meromorphic solutions of the \emph{dual reflection equation} (also see \cite{Sk,MeNe1,Vl} and \cite[Sec.~1.2]{WYCS}):
\eq{ 
R_{12}(u/v)^{-1} K^{\rm d}_1(u) \wt{R}_{21}(uv) K^{\rm d}_2(v) = K^{\rm d}_2(v) \wt{R}_{12}(uv) K^{\rm d}_1(u) R_{21}(u/v)^{-1}, \label{dualRE} }
where $\wt{R}(u) := ((R(u)^{t_1})^{-1})^{t_1}$.
We have assumed here that $R(u)^{t_1}$ is generically invertible, which is the case for the R-matrices considered in this paper.

The (inhomogeneous) boundary transfer matrix is the operator $\mc{T}(u;z_1,\ldots,z_k) \in \End((\C^N)^{\otimes k})$, depending meromorphically on $u$ and $\bm z$ and defined by
\eq{\label{btransfermatrix}
\begin{aligned}
\mc{T}(u;\bm z) &:= \Tr_0 K^{\rm d}_0(u) R_{10}(u z_1) \cdots R_{k0}(u z_k) K^+_0(u) R_{0k}(\tfrac{u}{z_k}) \cdots R_{01}(\tfrac{u}{z_1}).
\end{aligned}
}
We have
\eq{ \label{tmatrixcommutativity} [\mc{T}(u;\bm z),\mc{T}(v;\bm z)]=0}
as a consequence of \eqref{YBE}, \eqref{LRE} and \eqref{dualRE}.

Owing to \eqref{Ru:reg} certain transfer matrices \emph{interpolants} are equal to the scattering matrices:
\eq{ \label{ATrelation}
\mc{T}(z_r;\bm z) =\mc{A}_r(\bm z;1), \qq \mc{T}(z_r^{-1};\bm z) =  \mc{A}_r(\bm z;1)^{-1} \qq \text{for } 1 \le r \le k.
}
For the second relation of \eqref{ATrelation} we also need the unitarity of the left and right K-matrices, viz. $K^\pm(u)^{-1}=K^\pm(u^{-1})$. 
Then, under suitable conditions on the R-matrix datum, see e.g. \cite[Sec. 4.2]{Vl}, the commutativity of the $\mc{T}(u;\bm z)$ can be derived from the commutativity of the $\mc{A}_r(\bm z;1)$.

Sklyanin related the boundary transfer matrices in the homogeneous case $(z_1=\ldots=z_k=1)$ to quantum Hamiltonians for spin chains in the case that $R(u)$ and $K^+(u)$ have a suitable regularity condition at $u=1$, provided the R-matrix satisfies additional symmetries \cite[Prop.~4]{Sk}.
Similar expressions for these Hamiltonians can be found in \cite[Eqn.~(25-26)]{MeNe1} and, in the $U_q(\wh{\mfsl}_2)$-case with various choices for the K-matrices, \cite{BsBe2,BsKz2,JKKKM,JKKMW,StVl}.
We will now generalize this to a larger class of R-matrices and associated K-matrices.
For any meromorphic matrix-valued function $X$ we denote $X'(u) = \d{u} X(u)$ and define the Hamiltonian for the open chain of length $k$ as 
\eq{ \label{bHamiltonian} \mc{H} :=  (K^-_1)'(1)  + 2 \sum_{s=1}^{k-1} \hat R_{s, s+1}'(1) +  (K^+_k)'(1). }
Following \cite[Rmk. 3]{Sk} we note that $K^+$ and $K^-$ appear in a symmetric fashion in \eqref{bHamiltonian}.
\begin{prop}\label{prop:ATHrelation}
Suppose we have a meromorphic $\End(\C^N \ot \C^N)$-valued function $R$ satisfying \eqrefs{Ru:reg}{Ru:unit} and meromorphic $\End(\C^N)$-valued functions $K^\pm$ satisfying $K^\pm(1)=\Id$.
Defining $\mc{A}_r$, $\mc{T}$ and $\mc{H}$ according to \eqref{bqKZmatrix}, \eqref{btransfermatrix} and \eqref{bHamiltonian}, respectively.
We have
\eq{ \label{ATHrelation} \mc{H} = \mc{T}'( 1;\bm 1) = \frac{\partial \mc{A}_r(\bm z;1)}{\partial z_r} |_{\bm z = \bm 1} \qq \text{for } 1 \le r \le k. }
\end{prop}

\begin{proof}
From the first relation in \eqref{Krelation} we obtain
\eq{ \label{dualKat1} \Id = \Tr K^{\rm d}(1), \qq (K^-_1)'(1) = \Tr_0 (K^{\rm d}_0)'(1) + 2 \Tr_0  K^{\rm d}_0(1)  \hat R_{01}'(1)}
where $\hat R(u)  := P R(u)$. 
We may express the homogenized transfer matrix as
\[
\mc{T}(u;\bm 1) = \Tr_0 K^{\rm d}_0(u) \hat R_{01}(u) \hat R_{12}(u) \cdots \hat R_{k-1,k}(u) K^+_k(u) \hat R_{k-1,k}(u) \cdots \hat R_{12}(u) \hat R_{01}(u).\]
Thus, by a straightforward calculation,
\[
\mc{T}'(1;\bm 1) = \Tr_0 (K^{\rm d}_0)'(1) + 2 \sum_{s=0}^{k-1} \Tr_0 K^{\rm d}_0(1) \hat R_{s , s+1}'(1) + \Tr_0 K^{\rm d}_0(1) (K^-_k)'(1).
\]
Using \eqref{dualKat1} we obtain $ \mc{T}'( 1;\bm 1) = \mc{H}$.
To obtain the identity $\frac{\partial}{\partial z_r} \mc{A}_r(\bm z;1)|_{\bm z = \bm 1} = \mc{H}$, note that it is straightforward to establish that $\frac{\partial}{\partial z_r} \mc{A}_r(\bm z;1)|_{\bm z = \bm 1} = \mc{T}'( 1;\bm 1) $ by differentiating the first identity in \eqref{ATrelation} with respect to $z_r$ and using that $\mc{T}(1;\bm z) = \Id^{\otimes k}$ for generic values of $\bm z$.
\end{proof}

\pagebreak 

\begin{rmk} \mbox{}
\begin{enumerate}
\item
Equation \eqref{tmatrixcommutativity} and Proposition \ref{prop:ATHrelation} yield $[\mc{H},\mc{T}(u;\bm 1)]=0$, relating diagonalization problems of the Hamiltonian to simultaneous diagonalization problems of the transfer matrices.
Higher-order Hamiltonians can be defined through linear combinations of higher-order derivatives of the homogenized transfer matrix $\mc{T}(u;\bm 1)$ at $u=1$.
\item
Proposition \ref{prop:ATHrelation} is applicable to all singly or doubly regular K-matrices discussed in this paper. In the case when $K(-1) = \Id \ne K(1)$ one should redefine $K^\pm(u)$ as $K^\pm(-u)$. \hfill \rmkend
\end{enumerate}
\end{rmk}


It turns out that there exist natural rotation and dressing transformation formulas for the qKZ transport matrices, transfer matrices and Hamiltonians defined above, provided the transformations at the level of the K-matrices are chosen in a compatible way at the two boundaries.
More precisely, fix $\eta \in \K^\times$.
According to Lemma \ref{lem:REbasic} (iii), the assignment
\eq{ \label{Krighttransform} K^{+,Z}(u):=Z(\eta/u)^{-1} K^+(u) Z(\eta u)}
maps solutions of the right RE \eqref{RE} to itself. 
Here $Z(u)$ is a solution of \eqref{Ru:RZZ}, e.g. a rotation matrix $Z^\si(u)$ with $\si \in \Sigma_A$ or a dressing matrix $G(\bm \om)$ with $\bm \om$ as in Definition \ref{defnG}.
Fix another parameter $\zeta \in \K^\times$.
Similarly, the assignments
\eq{ \label{Klefttransform}
K^{-,Z}(u):=Z(\zeta u)^{-1} K^-(u) Z(\zeta/u) \qu \text{and} \qu K^{\rm d,Z}(u):= Z(\zeta u)^{-1} K^{\rm d}(u) Z(\zeta /u)
} 
permute the solutions of the left RE \eqref{LRE} and those of the dual RE \eqref{dualRE}, respectively. 

It follows from repeated application of \eqref{Ru:RZZ} that
\eqn{
\mc{A}^Z_r(\bm z;p) &:= \mc{A}_r(\bm z;p)|_{K^\pm(u) \mapsto K^{\pm,Z}(u),\zeta \to p^{1/2} \eta} \hspace{-10pt} &&=  \Big( \prod_{s \ne r} Z_s(\eta z_s) \Big)^{-1} Z_r(\eta p z_r)^{-1} \mc{A}_r(\bm z;p)  \prod_s Z_s(\eta z_s) , \\
\mc{T}^Z(u;\bm z) &:= \mc{T}(u;\bm z)|_{K^\pm(u) \mapsto K^{\pm,Z}(u), \zeta \to \eta} \hspace{-10pt} &&= \Big( \prod_{s} Z_s(\eta z_s) \Big)^{-1} \mc{T}(u;\bm z)  \prod_s Z_s(\eta z_s) , \\
\mc{H}^Z &:= \mc{H}|_{K^\pm(u) \mapsto K^{\pm,Z}(u), \zeta \to \eta} &&= \Big( \prod_{s} Z_s(\eta) \Big)^{-1} \mc{H} \prod_{s} Z_s(\eta). 
}
Note that these $Z$-transformed objects satisfy the same consistency and compatibility conditions as before, e.g. 
\[ 
\mc{A}^Z_r(z_1, \ldots, p z_s, \ldots, z_k;p) \mc{A}^Z_s(\bm z;p) = \mc{A}^Z_s(z_1, \ldots, p z_r, \ldots, z_k;p) \mc{A}^Z_r(\bm z;p) \qq \text{for } 1 \le r,s \le k .
\]

We end this discussion with a comment about the diagonalizability of $\mc{A}_r(\bm z;p)$ and $\mc{T}(u;\bm z)$.
Note that $\RT(a) := \RT_u(a)$ are independent of $u$ for all $a \in U_q(\mfh)$ and the R-matrices considered in this paper satisfy 
\eq{ \label{icerule} 
[\hat R(u/v),(\RT \ot \RT)(\Delta(a))]=0 \qq \text{for all } a \in U_q(\mfh).
}
This is essentially a generalization of the \emph{ice rule} (see \cite{Ba2}) enjoyed by the matrix $R(u)$ given by \eqref{Ru:Baxter} associated to $U_q(\wh{\mfsl}_2)$.
Suppose $K^+(u)$ is not diagonal ({\it e.g.}~$|I^*|\ne1$). Then $K^+(u)$ and $\RT(a)$ do not commute for all $a \in U_q(\mfh)$. 
Hence operators such as $\mc{A}_r(\bm z;p), \mc{T}(u;\bm z) \in \End((\C^N)^{\otimes k})$ do not commute with $\RT^{\ot k}(\Delta^{(k)}(a))$ for all $a \in U_q(\mfh)$, where we have recursively defined
\[ \Delta^{(k)} = (\Delta \otimes \id^{\otimes(k-2)}) \Delta^{(k-1)}: U_q(\mfg) \to U_q(\mfg)^{\ot k} \text{ for } k \in \Z_{>1}, \qq \Delta^{(1)} = \id. \]
Furthermore, the diagonalization \eqref{eigendecomposition} of $K^+(u)$ does not help, because as far as the authors are aware the equation
\[ R(v) W_1(u) V_2 = \wt W_1(u) V_2 R(v) \]
is not solvable for invertible $W(u),\wt W(u) \in \End(\K^N)$ unless $V = W(u) = \wt W(u)$.
It may be possible to diagonalize $\mc{A}_r(\bm z;p)$ and $\mc{T}(u;\bm z)$ by using R-matrices with a \emph{dynamical parameter}, see \cite{FK} for the case of $U_q(\wh\mfsl_2)$. 
One may expect that K-matrices which are generalized cross matrices ({\it i.e.}~those corresponding to quasistandard QP algebras) are most likely to be amenable to (generalizations of) this technique.


\subsubsection{Integrability in the twisted case} \label{sec:integrability:tw}

Here we briefly outline the analogues of the above constructions and statements in the twisted case. 
For more detail on quantum integrable systems with twisted boundary conditions see e.g.~\cite{BCDR,BCR,Gb}. 
Notably, in \cite{AACDFR2,Do1} transfer matrices and Hamiltonians are constructed for such systems; here we use a different approach although the underlying reflection equations are essentially the same.
Assume we have a meromorphic solution $\wt K^+(u)$ of the right twisted RE \eqref{tREalt} and a meromorphic solution $\wt K^-(u)$ of the \emph{left twisted reflection equation}
\eq{
R(\tfrac{u}{v})\, \wt K^-_1(u)\, (R_{21}(uv)^{-1})^{\t_1} \wt K^-_2(v) = \wt K^-_2(v)\, (R(uv)^{-1})^{\t_2} \wt K^-_1(u)\, R_{21}(\tfrac{u}{v})^\t. \label{LtREalt} 
} 
We remark that \eqref{tREalt} describes particle-to-antiparticle scattering and \eqref{LtREalt} describes antiparticle-to-particle scattering, so that composition of these two K-matrices corresponds to particle-to-particle scattering. 
Note that \eqref{LtREalt} is equivalent to \eqref{tREalt} by PT-symmetry \eqref{Ru:PT-symm}. 
Similar to the untwisted case, solutions of \eqref{LtREalt} can be bijectively related to solutions of \eqref{tREalt}: if $\wt K^+(u)$ is a solution of \eqref{tREalt}, then $\wt K^-(u) = C^{-1} \wt K^+(u)^\t (C^{-1})^\t$ is a solution of \eqref{LtREalt}. 

The twisted boundary qKZ transport operator is given by
\[
\begin{aligned}
\wt{\mc{A}}_r(\bm z;p) &= R_{r,r-1}(\tfrac{pz_r}{z_{r-1}}) \cdots R_{r1}(\tfrac{pz_r}{z_1}) \wt K^-_r(p^{1/2}z_r) \left(R_{1r}(z_1z_r)^{-1}\right)^{\t_r} \cdots \left(R_{r-1,r}(z_{r-1}z_r)^{-1}\right)^{\t_r} \\
& \qq \times \left( R_{r+1,r}(z_rz_{r+1})^{-1} \right)^{\t_r} \cdots \left( R_{kr}(z_rz_k)^{-1} \right)^{\t_r} \wt K^+_r(z_r) R_{kr}(\tfrac{z_k}{z_r})^{-1} \cdots R_{r+1,r}(\tfrac{z_{r+1}}{z_r})^{-1}. 
\end{aligned}
\]
and satisfies \eqref{qKZconsistency} as a consequence of \eqref{YBE}, \eqref{tREalt} and \eqref{LtREalt}.

The analogon of the relation \eqref{Krelation} in the twisted case is
\[
\wt K^-_1(u) = \Tr_0 \wt K^{\rm d}_0(u) P_{01} (R_{01}(u^2)^{-1})^{\t_1 } , \qquad 
\wt K^{\rm d}_1(u) = \Tr_0 \wt K^-_0(u) P_{01} R_{01}(u^2)^{\t_1} .
\]
As follows from a proof along the lines of \cite[App.~B]{Vl} and relying on \eqref{YBE}, it provides a bijection between solutions of the left twisted RE \eqref{LtREalt} and solutions of the dual twisted RE:
\eq{ \label{dualtRE}
R(\tfrac{u}{v})^{-1} \wt K^{\rm d}_1(u) R_{21}(uv)^{\t_1} \wt K^{\rm d}_2(v) = \wt K^{\rm d}_2(v) R(u v)^{\t_2} \wt K^{\rm d}_1(u) (R_{21}(\tfrac{u}{v})^{-1})^\t.
}
The twisted boundary transfer matrix is given by
\eq{ \label{btransfermatrixtwisted}
\wt{\mc{T}}(u;\bm z) := \Tr_0 \wt K^{\rm d}_0(u) (R_{10}(uz_1)^{-1})^{\t_0} \cdots (R_{k0}(uz_k)^{-1})^{\t_0} \wt K^+_0(u) R_{0k}(\tfrac{u}{z_k}) \cdots R_{01}(\tfrac{u}{z_1}). 
}
and satisfies \eqref{tmatrixcommutativity}.
We have the following relation between the twisted boundary transfer matrix and the twisted boundary qKZ transport matrix:
\[ \wt{\mc{T}}(z_r;\bm z) =  (\text{scalar}) \cdot \wt{\mc{A}}_r(\bm z;1) \qq \text{for } 1 \le r \le k. \]

If $\mfg^{\rm fin} \ne \mfsl_{N>2}$ then owing to Lemma \ref{lem:tw-untw} and its analogon for the left K-matrices, {\it viz.}~$\wt K^-(u) = K^-(\wt q u) C^{-1}$, we also have 
\[ \wt{\mc{T}}(\wt q^{-2} z_r^{-1};\bm z) =  (\text{scalar}) \cdot \wt{\mc{A}}_r(\bm z;1)^{-1} \qq \text{for } 1 \le r \le k. \]
For the same reason, for the Hamiltonian
\[ 
\wt{\mc{H}} := (\wt K^-_1)'(\wt q^{-1}) + 2 \sum_{s=1}^{k-1} \hat R'_{s,s+1}(1) + (\wt K^+_k)'(\wt q^{-1}) 
\]
we have the following version of Proposition \ref{prop:ATHrelation}:
\[ 
\wt{\mc{H}} := \wt{\mc T}'(\wt q^{-1};\wt q^{-1} \bm 1) = \frac{\partial \wt{\mc{A}}_r(\bm z;1)}{\partial z_r}|_{\bm z = \wt q^{-1} \bm 1} \qq \text{for } 1 \le r \le k. 
\]
Whenever the twisted K-matrix cannot be specialized to $C$ (for example owing to being a generalized involution matrix, as in types A.124) it is not possible to relate the Hamiltonian in this way to the transfer matrix and qKZ transport matrices.

Finally, the analogues of \eqrefs{Krighttransform}{Klefttransform} in the twisted case are 
\eq{ \label{Ktwtransform}
\begin{gathered}
\wt K^{+,Z}(u):= Z(\eta/u)^\t \wt K^+(u) Z(\eta u), \\
\wt K^{-,Z}(u):= Z(\zeta u)^{-1} \wt K^-(u) (Z(\zeta/u)^{-1})^\t, \qquad \wt K^{\rm d,Z}(u):= Z(\zeta/u)^{-1} \wt K^+(u) (Z(\zeta/u)^{-1})^\t.
\end{gathered}
}
The transformation rules for the qKZ transport matrices, transfer matrices and Hamiltonian are
\eqn{
\wt{\mc{A}}^Z_r(\bm z;p) &:= \wt{\mc{A}}_r(\bm z;p)|_{\wt K^\pm(u) \mapsto \wt K^{\pm,Z}(u),\zeta \to p^{1/2} \eta} \hspace{-10pt} &&=  \Big( \prod_{s \ne r} Z_s(\eta z_s) \Big)^{-1} Z_r(\eta p z_r)^{-1} \wt{\mc{A}}_r(\bm z;p) \prod_s Z_s(\eta z_s) , \\
\wt{\mc{T}}^Z(u;\bm z) &:= \wt{\mc{T}}(u;\bm z)|_{\wt K^\pm(u) \mapsto \wt K^{\pm,Z}(u), \zeta \to \eta} \hspace{-10pt} &&= \Big( \prod_{s} Z_s(\eta z_s) \Big)^{-1} \wt{\mc{T}}(u;\bm z)  \prod_s Z_s(\eta z_s) , \\
\wt{\mc{H}}^Z &:= \wt{\mc{H}}|_{\wt K^\pm(u) \mapsto \wt K^{\pm,Z}(u), \zeta \to \eta} &&= \Big( \prod_{s} Z_s(\eta) \Big)^{-1} \wt{\mc{H}} \Big( \prod_{s} Z_s(\eta)  \Big). 
}


\subsection{R-matrix presentation of quantum loop algebras and coideal subalgebras} \label{sec:RTT}

It is well-known (see {\it e.g.}~\cite[Sec.~2.3]{DeMk}) that K-matrices can be used to construct coideal subalgebras of quantum loop algebras, the so-called twisted quantum loop algebras, using the RTT presentation of quantum groups \cite{FRT}. Motivated by the pioneering works of Cherednik \cite{Ch1} and Sklyanin \cite{Sk}, Olshanski \cite{Ol} constructed twisted orthogonal and symplectic Yangians, denoted by $Y^+(N)$ and $Y^-(N)$, respectively, that are coideal subalgebras of the Yangian $Y(N)$ \cite{Dr1,Dr2}. These twisted Yangians are certain quantizations of the twisted current Lie algebra $\mfgl_N[x]^\theta$, where the involution $\theta$ is such that $\theta(x)=-x$ and $\mfgl_N^\theta=\mfso_N$ or $\mfgl_N^\theta=\mfsp_N$, corresponding to the symmetric pairs of type AI or AII, respectively. 

In \cite{MRS}, Molev, Ragoucy and Sorba constructed their q-analogues, coideal subalgebras of quantum loop algebras, called twisted orthogonal and twisted symplectic q-Yangians denoted by $Y^{\rm tw}_q(\mfso_N)$ and $Y^{\rm tw}_q(\mfsp_N)$, respectively. It was shown by Kolb \cite[Thm.~11.7]{Ko1} that certain specializations of the q-Yangians $Y^{\rm tw}_q(\mfso_N)$ and $Y^{\rm tw}_q(\mfsp_N)$ are isomorphic to QSP algebras of types A.1 and A.2, respectively. 

In \cite{CGM}, Chen, Guay and Ma constructed a family of twisted quantum loop algebras, denoted by $U_q^r(\wh\mfgl_N)$ and parametrized by $0\le r\le \lfloor N/2 \rfloor$, that correspond to the pairs of type AIII and are q-analogues of reflection algebras $\mc{B}(N,r)$ introduced by Molev and Ragoucy in \cite{MoRa}. Kolb noted in \cite[Rmk.~11.8]{Ko1}, that it is to be expected that these algebras are isomorphic to a certain family of QSP algebras. 

In this section we recall the RTT presentation of quantum loop algebras and briefly survey the description of coideal subalgebras by summarizing and extending the constructions presented in \cite{CGM,DeMk,MRS}. We then make a connection with the results obtained in the previous sections. In particular, we demonstrate how one obtains boundary intertwining equations studied in Section~\ref{sec:K-intw} starting with the RTT presentation of coideal subalgebras and answer the question raised in \mbox{\cite[Rmk.~11.8]{Ko1}}.


\subsubsection{R-matrix presentation of quantum groups}

Let $\mc{R}$ be the universal R-matrix of $U_q(\mfg^{\rm ext})$ and let $\mfb^{\rm ext, \pm}$ denote the standard Borel subalgebras of $\mfg^{\rm ext}$. Introduce $L$-operators by
\eq{
L^+(u) = (\RT_{u} \ot \id)(\mc{R}), \qq
L^-(u) = (\RT_{u} \ot \id)(\mc{R}_{21}^{-1}). \label{L-ops}
}
Recall that the universal R-matrix $\mc{R}$ is an element in a completion of $U_q(\mfb^{\rm ext,+}) \ot U_q(\mfb^{\rm ext,-})$ \cite{KhTo}. Thus the matrix elements $L^\pm_{ij}(u)$ of $L^\pm(u)$, where $i,j\in\lan N\ran$, are formal series 
\[
L^\pm_{ij}(u) = \sum_{r\ge 0} L^\pm_{ij}[r]\, u^{\pm r} \in U_q(\mfb^{\rm ext,\pm})[[u^{\pm1}]]
\] 
with coefficients satisfying $L^{\pm}_{ii}[0] \, L^\mp_{ii}[0] = 1$ for all $i$ and $L^-_{ij}[0] = L^+_{ji}[0] = 0$ for all $i<j$, which follows from the fact that $\RT_{u}(y_i)$ and $\RT_{u}(x_i)$ for ${1\le i \le n}$ are lower and upper triangular matrices, respectively; see Section \ref{sec:natrep} (the elements $L^\pm_{ij}[0]$ also satisfy additional symmetry relations that depend on the choice of $\mfg$, see \cite[\S 2]{FRT}). 
Introduce the elements $L^\pm_1(u)$ and $L^\pm_2(u)$ by
\[
L^\pm_1(u) = \sum_{ i,j\in\lan N\ran} E_{ij} \ot I \ot L^\pm_{ij}(u) ,  \qu
L^\pm_2(u) = \sum_{i,j\in\lan N\ran} I \ot E_{ij} \ot L^\pm_{ij}(u) .  
\]

\begin{prop}
The $L$-operators satisfy the following commutation relations
\eqa{
R(\tfrac uv)\, L^\pm_1(u)\, L^\pm_2(v) &= L^\pm_2(v)\, L^\pm_1(u)\, R(\tfrac uv) , \label{RTT12} \\
R(\tfrac uv)\, L^+_1(u)\, L^-_2(v) &= L^-_2(v)\, L^+_1(u)\, R(\tfrac uv) . \label{RTT3} 
}
\end{prop}

\begin{proof}
This is a standard computation. We follow the arguments presented in the proof of Prop.~27 in \cite[Sec.~8.5.1]{KlSg}. Let $\si : a \ot b \mapsto b \ot a$ denote the flip operator. 
Applying \mbox{$\RT_{u} \ot \RT_{v} \ot \id$} to both sides of $\mc{R}_{12} \mc{R}_{13} \mc{R}_{23} = \mc{R}_{23} \mc{R}_{13} \mc{R}_{12}$  we obtain the relation with the `$+$' sign in \eqref{RTT12}. Next, applying \mbox{$\id \ot \RT_{u} \ot \RT_{v}$} to both sides of $\mc{R}_{23} \mc{R}_{12}^{-1} \mc{R}_{13}^{-1} = \mc{R}_{13}^{-1} \mc{R}_{12}^{-1} \mc{R}_{23}$ yields the relation $R_{23}(\tfrac uv)\,\si(L^-_2(u))\,\si(L^-_3(v)) = \si(L^-_3(v))\,\si(L^-_2(u))\,R_{23}(\tfrac uv)$, which is equivalent to the one with the `$-$' sign in \eqref{RTT12}. Lastly, applying $\RT_{u} \ot \id \ot \RT_{v}$ to both sides of $\mc{R}_{13} \mc{R}_{12} \mc{R}_{23}^{-1} = \mc{R}_{23}^{-1} \mc{R}_{12} \mc{R}_{13}$ we obtain the relation $R_{13}(\tfrac uv)\,L^+_1(u)\,\si(L^-_3(v))= \si(L^-_3(v))\,L^+_1(u)\,R_{13}(\tfrac uv)$, which is equivalent to \eqref{RTT3}.
\end{proof}

Relations \eqrefs{RTT12}{RTT3} are conveniently called the RTT (or LLR) relations. The algebra generated by the coefficients $L^\pm_{ij}[r]$ with $r\ge0$ is called the RTT (or R-matrix) presentation of $U_q(\mfg)$; we will denote this algebra by $U^R_q(\mfg)$. Note that $U^R_q(\mfg)$ is a quantum loop algebra. This is because the universal $R$-matrix factorizes as $\mc{R}= \mathscr{R}_0\mathscr{R}$ with quasi $R$-matrix $\mathscr{R}$ independent of the central element $k_c$ and $(\RT_{u} \ot \id)(\mathscr{R}_0)=(\id \ot \RT_{u})(\mathscr{R}_0)=I$, see Section \ref{sec:affine} and \cite[Sec.~4]{FrRt}. The subalgebra generated by the coefficients $L^\pm_{ij}[0]$ is isomorphic to $U_q(\mfg^{\rm fin})$. (Isomorphisms between different presentations of $U_q(\wh\mfgl_N)$ and $U_q(\wh\mfsl_N)$ are studied in \cite{FrMn}. The RTT presentation of $U_q(\wh\mfso_N)$ and $U_q(\wh\mfsp_N)$ is studied in \cite{GRW}.) 

\smallskip

The coalgebra structure on $U^R_q(\mfg)$ is as follows. Recall that $(\Delta\ot\id)(\mc{R}) = \mc{R}_{23}\mc{R}_{13}$ \eqref{R-uni-2}. Let $\bar L^\pm_{ij}(u)$ denote the matrix elements of the inverses $L^\pm(u)^{-1}$. Then the coproduct on the elements $L^\pm_{ij}(u)$ and $\bar L^\pm_{ij}(u)$ is given by
\eq{
\Delta(L^\pm_{ij}(u)) = \sum_{a\in \lan N \ran} L^\pm_{ia}(u) \ot L^\pm_{aj}(u) , \qq
\Delta(\bar L^\pm_{ij}(u)) = \sum_{a\in \lan N \ran} \bar L^\pm_{aj}(u) \ot \bar L^\pm_{ia}(u). \label{L:cop}
}
%


\subsubsection{R-matrix presentation of coideal subalgebras} 

We will focus on coideal subalgebras of $U_q(\mfg)$ that are constructed via certain embeddings of the so-called $S$-operators of untwisted and twisted type. In particular, we introduce untwisted and twisted \mbox{$S$-operators} by
\eq{
S(u) = L^-(\tfrac1u)^{-1} K(u) \,L^+(u)  , \qq 
\wt S(u) = L^+(\tfrac1u)^\t\, \wt K(u) \,L^-(u) , \label{S-mat}
}
where $K(u)$ and $\wt K(u)$ are solutions of the untwisted and twisted reflection equations, respectively with matrix elements being rational functions in $u$ over $\K$, so that expanding in series at $u=0$ we have $K(u) \in \End(\K^N)[[u]]$ and at $u=\infty$ we have $\wt K(u) \in \End(\K^N)[[u^{-1}]]$. With these assumptions $S$-operators are well-defined elements in $U^R_q(\mfg)[[u]]$ and $U^R_q(\mfg)[[u^{-1}]]$, respectively.

Let $S_1(u)$, $S_2(u)$ and $\wt S_1(u)$, $\wt S_2(u)$ be defined in an analogous way as $L^\pm_1(u)$ and $L^\pm_2(u)$. We can write commutation relations for matrix elements of $S(u)$ and $\wt S(u)$ as follows.

\begin{prop}
The untwisted $S$-operator satisfies the untwisted algebraic reflection equation
\eqa{ 
R_{21}(\tfrac{u}{v})\,S_1(u)\,R(uv)\,S_2(v) &= S_2(v)\,R_{21}(uv)\,S_1(u)\,R(\tfrac{u}{v}), \label{S:RE}
}
in $\End((\K^N)^{\ot2}) \ot U^R_q(\mfg)[[u,v]]$. The twisted $S$-operator satisfies the twisted algebraic reflection equation
\eqa{ 
R(\tfrac{u}{v})\,\wt S_1(u)\,R(\tfrac{1}{uv})^{\t_1}\,\wt S_2(v) &= \wt S_2(v)\,R(\tfrac{1}{uv})^{\t_1}\,\wt S_1(u)\,R(\tfrac{u}{v})  \label{S:tRE} 
}
in $\End((\K^N)^{\ot2}) \ot U^R_q(\mfg)[[u^{-1},v^{-1}]]$. 

\end{prop}

\begin{proof}
It is a standard computation to show that $S(u)$ satisfies \eqref{S:RE}. This was first observed in \cite[Prop.~2]{Sk}; for a detailed proof see {\it e.g.}~\cite[Thm.~6.3]{CGM} (note that $R(u,v)$ in {\it ibid.}~corresponds to our $R_{21}(\tfrac{u}{v})$ for $\mfg = \wh{\mfsl}_N$). The proof of \eqref{S:tRE} is analogous; this was shown in \cite[Sec.~3]{MRS}
\end{proof}

Let $\mcB$ denote the algebra generated by the coefficients $S_{ij}[r]$ with $r\ge0$ of the matrix entries $S_{ij}(u)$ of $S(u)$, and let the algebra $\wt B$ be defined analogously in terms of $\wt S(u)$. 
The coalgebra structure on $\mcB$ and $\wt \mcB$ is described by the proposition below.

\begin{prop}
The algebras $\mcB$ and $\wt \mcB$ are right coideal subalgebras of $U^R_q(\mfg)$. 
\end{prop}

\begin{proof}
This is easy to show using \eqref{S-mat} and \eqref{L:cop}. A direct computation gives
\eqn{
\Delta(S_{ij}(u)) &= \sum_{b,c \in \lan N \ran} \Delta\big( \bar L^{-}_{ib}(\tfrac1u) K_{bc}(u) L^{ + }_{cj}(u)\big) = \sum_{a,b,c,d\in \lan N \ran} \bar L^{-}_{ab}(\tfrac1u)\,K_{bc}(u)\, L^{+}_{cd}(u) \ot \bar L^{-}_{ia}(\tfrac1u) \, L^{+}_{dj}(u) \\ 
&= \sum_{a,d\in \lan N \ran} S_{ad}(u) \ot \bar L^{-}_{ia}(\tfrac1u)\, L^{+}_{dj}(u) \in (\mcB \ot U^R_q(\mfg))[[u]] 
}
in the untwisted case, and
\eqn{
\Delta(\wt S_{ij}(u)) &= \sum_{b,c\in \lan N \ran} \Delta\big( L^{+}_{bi}(\tfrac1u)\, \wt K_{bc}(u)\, L^{-}_{cj}(u) \big) = \sum_{a,b,c,d\in \lan N \ran} L^{+}_{ba}(\tfrac1u)\,\wt K_{bc}(u)\, L^{-}_{cd}(u) \ot L^{+}_{ai}(\tfrac1u)\, L^{-}_{dj}(u) \\ 
& = \sum_{a,d\in \lan N \ran} \wt S_{ad}(u) \ot L^{+}_{ai}(\tfrac1u)\, L^{-}_{dj}(u) \in (\wt \mcB \ot U^R_q(\mfg) )[[u^{-1}]]
}
in the twisted case.
\end{proof}

We are now ready to obtain the RTT presentation of the boundary intertwining relations.

\begin{prop}
The following intertwining relations hold
\eqa{
K_2(v)\, (\id \ot \RT_{v}) ( S_1(u)) &= (\id \ot \RT_{1/v})(S_1(u))\, K_2(v), \label{SK} \\
\wt K_2(v)\, (\id \ot \RT_{v}) (\wt S_1(u)) &= (\id \ot \RT^\t_{1/v})(\id \ot S) (\wt S_1(u))\, \wt K_2(v)  . \label{tSK} 
}
\end{prop}

\begin{proof}
We demonstrate \eqref{SK} first. Observe that
\eqn{
(\id \ot \RT_{v})(L^{+}(u)) &= (\RT_{u} \ot \RT_{v})(\mc{R}) = r_0(\tfrac uv)\, R(\tfrac uv) , \\
(\id \ot \RT_{v})(L^{-}(\tfrac1u)^{-1}) &= (\RT_{1/u} \ot \RT_{v})(\mc{R}_{21}) = r_0(uv)\,R_{21}(uv) .
}
Combining this with \eqref{S-mat} we obtain
\eqn{
(\id \ot \RT_{v})(S_1(u)) &= r_0(\tfrac uv)\,r_0({uv})\, R_{21}(uv) \, K_1(u)\, R(\tfrac{u}{v}) , \\
(\id \ot \RT_{1/v})(S_1(u)) &= r_0(\tfrac uv)\,r_0({uv})\,R_{21}(\tfrac uv) \, K_1(u)\, R(uv) ,
}
Finally, by combining the expressions above with \eqref{RE} we deduce that \eqref{SK} is indeed true. 
Let us now demonstrate \eqref{tSK}. For the left hand side we have
\begin{gather*}
(\id \ot \RT_{v})(L^{+}(\tfrac1u)^\t) = (\RT_{1/u}^\t \ot \RT_{v})(\mc{R}) = r_0(\tfrac{1}{uv})\, R(\tfrac{1}{uv})^{\t_1} , \\
(\id \ot \RT_{v})(L^{-}(u)) = (\RT_{u} \ot \RT_{v})(\mc{R}^{-1}_{21}) = r_0(\tfrac vu)^{-1} R_{21}(\tfrac vu)^{-1} = r_0(\tfrac vu)^{-1} R(\tfrac uv) ,
\end{gather*}
giving
$$
(\id \ot \RT_{v})(\wt S_1(u)) = r_0(\tfrac{1}{uv})\,r_0(\tfrac vu)^{-1} R(\tfrac{1}{uv})^{\t_1}\,\wt K_1(u)\,R(\tfrac uv) .
$$
For the right hand side of \eqref{tSK} we compute
\begin{gather*}
(\id \ot \RT^\t_{1/v})(\id \ot S)(L^{+}(\tfrac1u)^\t) = (\RT^\t_{1/u} \ot \RT^\t_{1/v})(\id \ot S)(\mc{R}) = r_0(\tfrac vu)^{-1}(R(\tfrac vu)^{-1})^{\t} = r_0(\tfrac vu)^{-1} R(\tfrac uv) , \\
(\id \ot \RT^\t_{1/v})(\id \ot S)(L^{-}(u)) = P(\RT^\t_{1/v} \ot \RT_{u})(S\ot \id)(\mc{R}^{-1})P = r_0(\tfrac{1}{uv}) PR(\tfrac 1{uv})^{\t_1}P = r_0(\tfrac{1}{uv}) R(\tfrac 1{uv})^{\t_1} .
\end{gather*}
Hence
$$
(\id \ot \RT^\t_{1/v})(\id \ot S) (\wt S_1(u)) = r_0(\tfrac{1}{uv})\,r_0(\tfrac vu)^{-1} R(\tfrac uv)\, \wt K_1(u)\,R(\tfrac{1}{uv})^{\t_1}.
$$
Combining the equations above with \eqref{tRE} proves \eqref{tSK} as required. \qedhere
\end{proof}

The boundary intertwining equations \eqrefs{SK}{tSK} can be written in a more convenient form as
$$
\begin{aligned}
K(v)\, \RT_{v}(b) &= \RT_{1/v}(b)\, K(v)  && \text{for all } b\in \mcB, \\
\wt K(v)\, \RT_{v}(b) &= \RT^\t_{1/v}(S(b))\, \wt K(v) && \text{for all } b\in \wt \mcB .
\end{aligned}
$$
By introducing an additional parameter $\eta$ into the definitions \eqref{L-ops} we recover the boundary intertwining equations studied in Section \ref{sec:K-intw}.

\smallskip

We have shown above that K-matrices for the vector representation of QP algebras can be used to construct coideal subalgebras of quantum loop algebras in the RTT presentation of quantum groups. The natural question to ask is whether coideal subalgebras in the RTT presentation and in the quantum symmetric pair description are isomorphic as algebras. As we have mentioned in the beginning of this section, a positive answer was given by Kolb in \cite[Thm.~11.7]{Ko1} for QSP algebras of types AI and AII for a certain choice of tuples $\bm c$ (that in our notation correspond to types A.1 and A.2). It is natural to expect that the same is true for all the remaining cases.

\begin{rmk}
In \cite[Rmk.~11.8]{Ko1}, Kolb noted that a family $U_q^r(\wh\mfgl_N)$ of twisted quantum loop algebras of type AIII/IV constructed by Chen, Guay and Ma in \cite{CGM} is expected to be isomorphic (by taking the intersection with $U_q(\wh\mfsl_N)$) to a family of QSP algebras of type AIII/AIV, that in our notation corresponds to type A.3. Coideal subalgebras $U_q^r(\wh\mfgl_N)$ are constructed via the untwisted $S$-operator, which in our notation corresponds to $S(u)$, and using the K-matrix $G^\xi(u)$, which coincides with the restrictable K-matrix $K_{s=N}(u)$ of type A.3 (after setting $\la=\xi$ and $\mu=1$). 
Thus the family $U^r_q(\wh\mfgl_N)$ can be viewed as a subfamily of a larger family $U^{\ell,r}_q(\wh\mfgl_N)$ defined using the `general' K-matrix $K(u)$ of type A.3. Note that this K-matrix provides a two--parameter one--dimensional representation of $U^{\ell,r}_q(\wh\mfgl_N)$ parametrized by $\la$ and $\mu$. This is an analogy of the one--parameter one--dimensional representation $V(\ga)$ of the reflection algebras $\mcB(N,r)$ constructed in \cite[Sec.~4.3]{MoRa}. This answers the question of Kolb about the role of the second parameter. \hfill \rmkend
\end{rmk}


\subsection{Representations of the cyclotomic Hecke algebra} \label{sec:Hecke}

Affine Hecke algebras play an important role in the representation theory of $U_q(\wh\mfsl_N)$ (see {\it e.g.}~\cite[Sec.~12.3]{ChPr1} and \cite{ChPr2}) and its coideal subalgebras  (see {\it e.g.}~\cite{CGM,JoMa}). Moreover, its cyclotomic quotients have wide applications in Kazhdan-Lusztig theory \cite{BaWa} and the theory of quantum integrable models with open boundary conditions (see {\it e.g.}~\cite{BbRg,Do2,IsOg,Is,LvMt}). In this section we briefly recall representation of the Hecke algebra on $(\K^N)^{\ot k}$ when $k\in\Z_{\ge2}$. We then demonstrate that baxterized solutions of the reflection equation within a cyclotomic Hecke algebra found in \cite{KuMv} and \cite[Sec.~1]{IsOg} can be represented by K-matrices of type A.3 obtained in Section \ref{sec:A3}. This motivates a study of representations of cyclotomic Birman--Murakami--Wenzl (BMW) algebras on $(\K^N)^{\ot k}$. In particular, we expect that K-matrices of types B, C, D given by Theorem \ref{T:all-K} provide representations of certain baxterized solutions of the reflection equation obtained in \cite[Sec.~2]{IsOg}. Such connections between coideal subalgebras and representations of cyclotomic Hecke or BMW algebras point towards an analogue of the Schur--Weyl duality for coideal subagebras, see \cite[Sec.~7.2]{CGM} and \cite[Sec.~7 \& 8]{JoMa} for such a duality for certain coideal subalgebras of type A.3; the role of K-matrices in this case is analogous to that of R-matrices in the Schur--Weyl duality for the quantum group $U_q(\mfg)$.
%


\subsubsection{Finite and affine Hecke algebras} We start by recalling the relevant definitions.

\begin{defn}
The Hecke algebra $H_k(q)$ is the unital associative $\K$-algebra with generators $\si_1,\ldots,\si_{k-1}$ and defining relations 
\begin{flalign}  \label{Hecke}
&&&& \si_i \si_{i+1} \si_i = \si_{i+1} \si_i \si_{i+1}, \qu 
\si_i^2 = (q-q^{-1}) \si_i + 1 , \qu
 \si_i \si_j &= \si_j \si_i \tx{if} |i-j|>1. && \defnend
\end{flalign}
\end{defn}

Let $\hat R_q = P R_q$ be as in Section \ref{sec:Rmat}, and let $\hat R_{i,i+1}$ denote $\hat R_q$ acting nontrivially on spaces $i$ and $i+1$ in the tensor product $(\K^N)^{\ot k}$ only. Then there exists a unique representation $\pi_k$ of the Hecke algebra $H_k(q)$ on $(\K^N)^{\ot k}$ such that $\pi_k(\si_i)=\hat R_{i,i+1}$. Introduce {\it baxterized} elements
\eq{
\si_i(u) = \frac{\si_i - u\, \si_i^{-1}}{q-q^{-1}u} \label{si(u)}
}
satisfying unitarity ($\si_i(u)\si_i(u^{-1})=1$) and the baxterized braid relation
\eq{
\si_i(u)\, \si_{i+1}(uv)\, \si_i(v) = \si_{i+1}(v)\, \si_i(uv)\, \si_{i+1}(u). \label{Hecke-YBE}
}
Recall that the R-matrix with a spectral parameter for the vector representation of $U_q(\wh\mfsl_N)$ is given by \eqref{Ru:affinization}
\eq{  
R(u) = \frac{ R_q - u P R_q^{-1} P}{q-q^{-1}u}. \label{Rq(u):A} 
}
Denote by $\hat R_{i,i+1}(u) = P_{i,i+1} R_{i,i+1}(u)$ the R-matrix above acting nontrivially on spaces $i$ and $i+1$ in $(\K^N)^{\ot k}$ only. Comparing \eqref{si(u)} and \eqref{Rq(u):A} we see that $\pi_k(\si_i(u))=\hat R_{i,i+1}(u)$. Moreover, \eqref{Hecke-YBE} is mapped to the braided quantum Yang-Baxter equation ({\it cf.}~\eqref{YBE})
$$
\hat R_{i,i+1}(u)\,\hat R_{i+1,i+2}(u v)\,\hat R_{i,i+1}(v) = \hat R_{i+1,i+2}(v)\,\hat R_{i,i+1}(u v)\,\hat R_{i+1,i+2}(u).
$$

\begin{defn}
The affine Hecke algebra $\hat H_k(q)$ of type $A_{k-1}$ is an extension of $H_k(q)$ by one more element $\tau_1$ satisfying
\begin{flalign}
&&&& \si_1 \tau_1 \si_1 \tau_1 = \tau_1 \si_1 \tau_1 \si_1  \qq {and} \qq \tau_j \si_i = \si_i \tau_j \tx{for} j\ne i,i+1. \label{Hecke-aff} && \defnend
\end{flalign}
\end{defn}

It was shown in \cite[Prop.~1]{IsOg} that the baxterized element
\eq{
\tau_1(u) = \frac{\tau_1 - \xi u}{\tau_1 - \xi u^{-1}} \label{tau(u)}
}
is a unitary ($\tau_1(u)\tau_1(u^{-1})=1$) solution of the baxterized reflection equation
\eq{
\si_1(u/v)\,\tau_1(u)\,\si_1(uv)\, \tau_1(v) = \tau_1(v)\,\si_1(u v)\,\tau_1(u)\,\si_1(u/v). \label{Hecke-RE}
}
This is a local solution in the sense that $[\tau_1(u),\si_j]=0$ for all $1<j<k$, and a regular solution: $\tau_i(\pm1)=1$.


\subsubsection{Cyclotomic Hecke algebra} 

Fix a positive integer $m\in \Z_{\ge 1}$ and let $\al_j\in\K$ for $0\le j\le m$. Consider the two-sided ideal $\hat I(\al_1,\ldots,\al_m)$ of $\hat H_k(q)$ generated by the identity 
\eq{
\tau_1^{m+1} + \sum_{0\le j \le m} \al_j \tau_1^j = 0. \label{char-id}
}
The cyclotomic Hecke algebra is $H_{m,1,k}(\al_0,\ldots,\al_m;q) := \hat H_k(q) / \hat I(\al_0,\ldots,\al_m)$ (for more details see~\cite{Ari}). The identity \eqref{char-id} is called the {\it characteristic identity}.

\smallskip


Let us focus on the case $m=1$. The characteristic identity is now $\tau_1^2 + \al_1 \tau_1 + \al_0 = 0$, and the cyclotomic Hecke algebra in this case amounts to the Hecke algebra of type B. We have
\[
\frac{1}{\tau_1 - \xi u^{-1}} = - \frac{(\al_1 u  + \xi)(\tau_1 + \xi u^{-1} + \al_1) }{ (\xi u^{-1}+ \al_1)(\al_0 u + \xi^2 u^{-1} + \al_1 \xi) } .
\]
Hence $\tau_1(u)$ specializes to the simplest polynomial solution \cite[Eq.~(1.17)]{IsOg}:
\eq{
\tau_1(u) = \frac{u-u^{-1}}{\xi u^{-2} + \al_1 u^{-1}+ \al_0 \xi^{-1}}\left( \tau_1 + \frac{\al_1 u + \xi + \al_0 \xi^{-1}}{u-u^{-1}}\right)  . \label{tau(u)-1}
}
It is a direct computation to check that \eqref{tau(u)-1} is indeed a solution of \eqref{Hecke-RE}. 
Our goal is to show that $\pi_k$ can be naturally extended to $H_{1,1,k}(\al_0,\al_1;q)$ and that this extension maps $\tau_1(u)$ to a scalar multiple of the inverse of a particular restrictable K-matrix of type A.3 (see Section \ref{sec:A3}).
More precisely, define $\bar K(u) := K_{t=N,\,\ell=0}(u)^{-1}$ with $0< r \le N/2$.
Note that $K_{t=N,\,\ell=0}(u)$ does not depend on $q$; hence \eqref{Ru:bar} implies that $\bar K(u)$ is a solution of the left reflection equation \eqref{LRE}. Define $\bar K = \la^{-1}\lim_{u\to 0} \bar K(u)$. In particular,
\eq{
\bar K = \la^{-1}\Id - \sum_{1\le i \le r} (\la E_{ii} + \la^{-1} E_{N-i+1,N-i+1} + E_{i,N-i+1} + E_{N-i+1,i}) .  \label{KA30-left}
}
The matrix $\bar K$ satisfies $(\bar K + \la\,\Id)(\bar K - \la^{-1}\Id) = 0$ (its minimal polynomial) and the constant left reflection equation
\eq{ 
\hat R_q \, (\bar K \ot \Id) \, \hat R_q \, (\bar K \ot \Id) = (\bar K \ot \Id)\, \hat R_q \, (\bar K \ot \Id)\,\hat R_q . \label{cLRE}
}  
We also note that $\bar K(u)$ can be written as 
\eq{
\bar K(u) = \frac{u-u^{-1}}{(\la^{-1} + \mu^{-1} u)(\la \mu - u^{-1})} \left(\bar K + \frac{(\la - \la^{-1})\,u + \mu - \mu^{-1}}{u-u^{-1}}\, \Id \right) . \label{KA3-left}
}

Due to the automorphism $\tau_1 \mapsto c\,\tau_1$ of $\hat H_k(q)$, where $c \in \C^\times$ is arbitrary, we can fix one of the free parameters without loss of generality. Thus we choose $\al_1 = \la - \la^{-1}$, $\al_0 = -1$ and $\xi = \mu $, so that $\la$ and $\mu$ are independent free parameters. Then $(\tau_1 + \la)(\tau_1 - \la^{-1}) = 0$ and \eqref{tau(u)-1} becomes
\eq{
\tau_1(u) = \frac{u^2(u-u^{-1})}{(\la^{-1} + \mu^{-1} u)(\la \mu - u)} \left( \tau_1 + \frac{(\la - \la^{-1})\,u + \mu - \mu^{-1}}{u-u^{-1}} \right) . \label{tau(u)-2}
}
Denote by $\bar K_i$ and $\bar K_i(u)$ the matrices $\bar K$ in \eqref{KA30-left} and $\bar K(u)$  in \eqref{KA3-left}, respectively, acting on the space $i$ in $(\K^N)^{\ot k}$ only. 

\begin{prop} \label{P:Hecke-2}
The assignments 
\eq{
\tilde \pi_k(\si_i) = \hat R_{i,i+1} \tx{for} 1 \le i < k \qu\text{and}\qu \tilde \pi_k(\tau_1) = \bar K_1  \label{tau-rep}
}
with $\bar K_1$ as in \eqref{KA30-left} define a unique representation $\tilde \pi_k$ of $H_{1,1,k}(\al_0,\al_1;q)$ on $(\K^N)^{\ot k}$.
\end{prop}

\begin{proof}
Since $\tilde \pi_k$ restricts on $H_k(q)$ to the representation $\pi_k$ it is enough to show that $\bar K_1$ satisfies the same relations as $\tau_1$. But this is straightforward: we have $(\bar K_1 + \la\,\Id)(\bar K_1 - \la^{-1}\Id) = 0$, $[\bar K_1,\hat R_{j,j+1}]=0$ for all $1<j<k$ and \eqref{cLRE}.  
\end{proof}

By comparing \eqref{tau(u)-2} with \eqref{KA3-left} we obtain the following.

\begin{crl}
The baxterized solution $\tau_1(u)$ given by \eqref{tau(u)-1} satisfies:
\eq{ 
\tilde{\pi}_k(\tau_1(u)) = \frac{1-\la \mu u}{1 - \la \mu u^{-1}} \, \bar K_1(u) .
} 
\end{crl}
%


Next, let us study the case when $m=2$. The characteristic identity is now of degree three, $\tau_1^3 + \al_2 \tau_1^2 + \al_1 \tau_1 + \al_0 = 0$, giving
$$
\frac{1}{\tau_1 - \xi u^{-1 }} = -\frac{\sum_{0\le j \le 2} \al_{j+1} (\frac{\xi}{u})^j}{\sum_{0\le j \le 3} \al_j (\frac{\xi}{u})^j}  \Bigg( 1+ \frac{\frac{\xi}{u} + \al_2 }{\sum_{0\le j \le 2} \al_{j+1} (\frac{\xi}{u})^j } \left( y + \frac{y^2}{\frac{\xi}{u} + \al_2}\right)\! \Bigg),
$$
where $\al_3=1$, and $\tau_1(u)$ specializes to \cite[Eq.~(1.18)]{IsOg}:
\eq{
\tau_1(u) = \frac{\xi(u-u^{-1})}{ \sum_{0\le j \le 3} \al_j (\frac{\xi}{u})^j } \Bigg( \tau_1^2 + \big(\tfrac{\xi}{u} + \al_2 \big)\tau_1 + \frac{\frac{\xi^2}{u} + \al_2 \xi + \al_1 u + \frac{\al_0}{\xi}}{u-u^{-1}} \Bigg). \label{tau(u)-3}
}
Showing that \eqref{tau(u)-3} is indeed a solution of \eqref{Hecke-RE} requires to use the identity
$$
(q-q^{-1})\, \tau_1 \tau_1 \sigma_1 \tau_1-\sigma_1 \tau_1 \tau_1 \sigma_1 \tau_1 + \tau_1 \sigma_1 \tau_1 \tau_1 \si^{-1}_1 = 0 .
$$
Note that the inverse of $\tau_1$ is not required in checking \eqref{Hecke-RE}. Hence we need not assume that $\tau_1$ has an inverse. This will be important later on.

Now we want to show that $\pi_k$ can be naturally extended to $H_{2,1,k}(\al_0,\al_1,\al_2;q)$ and that $\tau_1(u)$ given by \eqref{tau(u)-3} is mapped to a scalar multiple of the inverse of a particular non-restrictable K-matrix of type A.3, namely $\bar K(u) = K_{\ell=0}(u)^{-1}$ with $0<t<N$ and $0<r\le t/2$. Define $\bar K = \la^{-1}\lim_{u\to0}\bar K(u)$. More precisely,
\eq{
\bar K = \sum_{1\le i \le t-r} \la^{-1}E_{ii} - \sum_{1\le i \le r} (\la E_{ii} + E_{i,t-i+1} + E_{t-i+1,i}) . \label{KA30-left2}
}
The matrix above is not invertible and satisfies $(\bar K + \la \Id)(\bar K - \la^{-1} \Id )\bar K = 0$ (its minimal polynomial) and, as before, \eqref{cLRE}. With these conventions $\bar K(u)$ can be written as
\eq{
\bar K(u) = \frac{u-u^{-1}}{(1+ \frac{\mu}{\la u})(\la \mu - \frac1u)} \left( \bar K^{2} + \big(\tfrac{\mu}{u} + \la - \la^{-1}\big) \bar K + \frac{(1+ \frac{\mu}{\la u})(\la \mu - u)}{u-u^{-1}}\,\Id \right) . \label{KA3-left2}
}

Set $\al_2 = \la- \la^{-1}$, $\al_1 = -1$, $\al_0=0$ and $\xi = \mu$. Then \eqref{tau(u)-3} becomes
\eq{
\tau_1(u) = \frac{u^2(u-u^{-1})}{(1+ \frac{\mu}{\la u})(\la \mu - u)} \left( \tau_1^{2} + \big(\tfrac{\mu}{u} + \la - \la^{-1})\, \tau_1 - \frac{(1+ \frac{\mu}{\la u})(\la \mu - u)}{u-u^{-1}}\,\Id \right) . \label{tau(u)-4}
}
The characteristic identity of $\tau_1$ can now be written as $(\tau_1 + \la)(\tau - \la^{-1})\tau_1 = 0$. Next, as in the $m=1$ case, we denote by $K_i$ and $K_i(u)$ the matrices $\bar K$ in \eqref{KA30-left2} and $\bar K$ in \eqref{KA3-left2}, respectively, acting on the space $i$ in $(\K^N)^{\ot k}$ only.

\begin{prop} \label{P:Hecke-3}
The assignments 
\eq{
\tilde \pi_k(\si_i) = \hat R_{i,i+1} \tx{for} 1 \le i < k \qu\text{and}\qu \tilde \pi_k(\tau_1) = \bar K_1  \label{tau-rep-2}
}
with $\bar K_1$ as in \eqref{KA30-left2} define a unique representation $\tilde \pi_k$ of $H_{2,1,k}(\al_0,\al_1,\al_2;q)$ on $(\K^N)^{\ot k}$.
\end{prop}

\begin{proof}
The proof is analogous to that of Proposition \ref{P:Hecke-2}. 
\end{proof}

Finally, by comparing \eqref{tau(u)-3} with \eqref{KA3-left2} we obtain the following.

\begin{crl}
The baxterized solution $\tau_1(u)$ given by \eqref{tau(u)-3} satisfies:
\eq{ 
\tilde{\pi}_k(\tau_1(u)) = \frac{1 - \la \mu u}{1 -\la \mu u^{-1}} \, \bar K_1(u) . \label{tau-rep-3}
} 
\end{crl}

\begin{rmk} \label{R:Kulish}
The representation \eqref{tau-rep-3} was first obtained in \cite[Sec.~5]{KuMv}. 
In particular, the K-matrix $\bar K$ in \eqref{KA30-left2} and $\bar K(u)$ in \eqref{KA3-left2} correspond to the one in equation $(36)$ in {\it ibid.} and the matrix $K(x)$ immediately below it, respectively.
Note that these K-matrices have three free parameters. 
However one of them can be removed by the dressing method, which leaves two non-removable parameters, as expected.\hfill \rmkend
\end{rmk}



\appendix


\section{Generalized Satake diagrams of classical Lie type} \label{App:Satakediagrams}


In Table \ref{tab:satakediagrams}, following \cite{BBBR}, we present a complete classification of generalized Satake diagrams whose underlying Dynkin diagram is untwisted affine and of classical Lie type in terms of $n$, $o_i$ (the number of nontrivial $\tau$-orbits in $Y_i$) and $p_i$ (the number of $\tau$-orbits in $X_i  = X_{Y_i}$). 
We additionally indicate the restrictions placed on the parameters and $|\Sigma_A(X,\tau)|$, the cardinality of the $\Sigma_A$-orbit under consideration. 
Recall that $N=n+1$ if the diagram is of type ${\rm A}^{(1)}_n$.


\setlength{\tabcolsep}{3pt}
\begin{longtable}{m{7mm} m{20mm} m{47mm} m{11mm} | m{39mm} m{25mm}}
\caption{Classification of generalized Satake diagrams of affine type} \label{tab:satakediagrams} \\
\hline
& & & & & \\[-11pt]
\multicolumn{2}{l}{Name} & Diagram & $o_1,o_2$ & Restrictions & $|\Sigma_A(X,\tau)|$ \\
\hline
\endfirsthead
\hline
& & & & & \\[-11pt]
\multicolumn{2}{l}{Name} & Diagram & $o_1,o_2$ & Restrictions & $|\Sigma_A(X,\tau)|$ \\
\hline
\endhead
A.1 
& $\bigl( {\rm A}^{(1)}_n \bigr)^{\rm id}_\emptyset$ 
& 
\begin{tikzpicture}[baseline=-0.35em,line width=0.7pt,scale=65/100]
\clip (-1.1,-1.2) rectangle (1.1,1.2);
\draw[thick,domain=0:135] plot ({cos(\x)},{sin(\x)});
\draw[thick,dashed,domain=135:360] plot ({cos(\x)},{sin(\x)});
\filldraw[fill=white] ({-sqrt(2)/2},{sqrt(2)/2}) circle (.1);
\filldraw[fill=white] (0,1) circle (.1);
\filldraw[fill=white] ({sqrt(2)/2},{sqrt(2)/2}) circle (.1);
\filldraw[fill=white] (1,0) circle (.1);
\end{tikzpicture}
& -
& -
& 1 
\\
\hline
A.2
& $\bigl( {\rm A}^{(1)}_n \bigr)^{\rm id}_{\rm alt}$ 
& 
\begin{tikzpicture}[baseline=-0.35em,line width=0.7pt,scale=65/100]
\clip (-1.1,-1.2) rectangle (1.1,1.2);
\draw[thick,domain=0:135] plot ({cos(\x)},{sin(\x)});
\draw[thick,dashed,domain=135:360] plot ({cos(\x)},{sin(\x)});
\filldraw[fill=black] ({-sqrt(2)/2},{sqrt(2)/2}) circle (.1);
\filldraw[fill=white] (0,1) circle (.1);
\filldraw[fill=black] ({sqrt(2)/2},{sqrt(2)/2}) circle (.1);
\filldraw[fill=white] (1,0) circle (.1);
\end{tikzpicture} 
& -
& $N$ even
& 2
\\
\hline
\end{longtable}

In \cite[\S 4 and \S 5]{Ara} (also see \cite[Sec.~3]{NoSu} and \cite[Sec.~7]{Le3}) all Satake diagrams of finite type were classified and a notation involving roman numerals was introduced. 
We recall this classification in Table \ref{tab:finitesatakediagrams}, indicating the corresponding symmetric pair, and extend it to all generalized Satake diagrams by including weak Satake diagrams of types CI and BDIII.
We provide the name of the generalized Satake diagram, both in terms of the roman-numeral notation and the notation of the present paper, suitably adapted to the finite case. 
Finally we give the corresponding restrictable generalized Satake diagram of affine type in the classification of Table \ref{tab:satakediagrams}.
The types AIV and BDII are distinguished from the larger families AIII and BDI, respectively, by the condition that $|I^*| = 1$; also note that in \cite[Sec.~7]{Le3} the Satake diagrams given by $(X,\tau) = (I, w_I)$ are not listed, since in this case $\mfg_X=\mfg$, thus $B_{\bm c,\bm s}=U_q(\mfg)$.

\setlength{\tabcolsep}{4pt}
\begin{longtable}{m{37mm} m{10mm} m{14mm} m{32mm} m{23mm} m{7mm} m{17mm}}
\caption{Classification of generalized Satake diagrams of finite type} \label{tab:finitesatakediagrams} \\
\hline
& & & \\[-11pt]
Symm. pair & \multicolumn{2}{c}{Name} & Diagram & Restrictions & \multicolumn{2}{c}{Affinization} \\
\hline
\endfirsthead
\hline
& & & \\[-11pt]
Symm. pair & \multicolumn{2}{c}{Name} & Diagram & Restrictions & \multicolumn{2}{c}{Affinization} \\
\hline
\endhead
$(\mfsl_N,\mfso_N)$ & AI & $\bigl( {\rm A}_n \bigr)^{\rm id}_\emptyset$ & 
\begin{tikzpicture}[baseline=-.5em,line width=0.7pt,scale=75/100]
\clip (-.7,-.4) rectangle (1.7,.4);
\draw[thick] (-.5,0) -- (0,0);
\draw[thick,dashed] (0,0) -- (1,0);
\draw[thick] (1,0) --  (1.5,0);
\filldraw[fill=white] (-.5,0) circle (.1);
\filldraw[fill=white] (0,0) circle (.1);
\filldraw[fill=white] (1,0) circle (.1);
\filldraw[fill=white] (1.5,0) circle (.1);
\end{tikzpicture} 
&
& A.1
& $\bigl( {\rm A}^{(1)}_n \bigr)^{\rm id}_\emptyset$
\\
\hline
$(\mfsl_N,\mfsp_N)$ & AII & $\bigl( {\rm A}_n \bigr)^{\rm id}_{\rm alt}$ & 
\begin{tikzpicture}[baseline=-.5em,line width=0.7pt,scale=75/100]
\clip (-.7,-.4) rectangle (2.2,.4);
\draw[thick] (-.5,0) -- (0,0);
\draw[thick,dashed] (0,0) -- (1,0);
\draw[thick] (1,0) --  (2,0);
\filldraw[fill=black] (-.5,0) circle (.1);
\filldraw[fill=white] (0,0) circle (.1);
\filldraw[fill=black] (1,0) circle (.1);
\filldraw[fill=white] (1.5,0) circle (.1);
\filldraw[fill=black] (2,0) circle (.1);
\end{tikzpicture} 
& $N$ even
& A.2
& $\bigl( {\rm A}^{(1)}_n \bigr)^\id_{\rm alt}$
\\
\hline
\multirow{2}{*}{\casesm[.9]{l}{\\[-5pt] (\mfsl_N, \mfsl_N \cap \\ \qu (\mfgl_{\lfloor \frac{N}{2} \rfloor - p} \oplus \mfgl_{\lceil \frac{N}{2} \rceil + p} ))}}
& \multirow{2}{*}{\casesm[1]{l}{\\[-5pt] $AIII$, \\ $AIV$}}
& $\bigl( {\rm A}_n \bigr)^{\psi}_{p}$ &
\begin{tikzpicture}[baseline=-2em,line width=0.7pt,scale=75/100]
\clip (1.3,-.6) rectangle (4.6,1.1);
\draw[thick,dashed] (1.5,.4) -- (2.5,.4);
\draw[thick] (2.5,.4) -- (3,.4);
\draw[thick,dashed] (3,.4) -- (4,.4);
\draw[thick] (4,.4) -- (4.5,0) -- (4,-.4);
\draw[thick,dashed] (1.5,-.4) -- (2.5,-.4);
\draw[thick] (2.5,-.4) -- (3,-.4);
\draw[thick,dashed] (3,-.4) -- (4,-.4);
\draw[snake=brace] (2.9,.55) -- (4.6,.55) node[midway,above]{\scriptsize $p$};
\filldraw[fill=white] (1.5,.4) circle (.1);
\filldraw[fill=white] (1.5,-.4) circle (.1);
\draw[<->,gray] (1.5,.3) -- (1.5,-.3);
\filldraw[fill=white] (2.5,.4) circle (.1);
\filldraw[fill=white] (2.5,-.4) circle (.1);
\draw[<->,gray] (2.5,.3) -- (2.5,-.3);
\filldraw[fill=black] (3,.4) circle (.1);
\filldraw[fill=black] (3,-.4) circle (.1);
\draw[<->,gray] (3,.3) -- (3,-.3);
\filldraw[fill=black] (4,.4) circle (.1);
\filldraw[fill=black] (4,-.4) circle (.1);
\draw[<->,gray] (4,.3) -- (4,-.3);
\filldraw[fill=black] (4.5,0) circle (.1);
\end{tikzpicture} 
& \multirow{2}{*}[-4pt]{$0 \leq p \le \tfrac{\lfloor N-1 \rfloor }{2}$}
& A.3a
& $\bigl( {\rm A}^{(1)}_n \bigr)^{\psi}_{0,p}$  
\\*[-1em]
& &  $\bigl( {\rm A}_n \bigr)^{\psi}_{p}$  &
\begin{tikzpicture}[baseline=-2em,line width=0.7pt,scale=75/100]
\clip (1.3,-.6) rectangle (4.6,1.1);
\draw[thick,dashed] (1.5,.4) -- (2.5,.4);
\draw[thick] (2.5,.4) -- (3,.4);
\draw[thick,dashed] (3,.4) -- (4,.4);
\draw[thick,domain=270:450] plot({4+.4*cos(\x)},{.4*sin(\x)});
\draw[thick,dashed] (1.5,-.4) -- (2.5,-.4);
\draw[thick] (2.5,-.4) -- (3,-.4);
\draw[thick,dashed] (3,-.4) -- (4,-.4);
\draw[snake=brace] (2.9,.55) -- (4.1,.55) node[midway,above]{\scriptsize $p$};
\filldraw[fill=white] (1.5,.4) circle (.1);
\filldraw[fill=white] (1.5,-.4) circle (.1);
\draw[<->,gray] (1.5,.3) -- (1.5,-.3);
\filldraw[fill=white] (2.5,.4) circle (.1);
\filldraw[fill=white] (2.5,-.4) circle (.1);
\draw[<->,gray] (2.5,.3) -- (2.5,-.3);
\filldraw[fill=black] (3,.4) circle (.1);
\filldraw[fill=black] (3,-.4) circle (.1);
\draw[<->,gray] (3,.3) -- (3,-.3);
\filldraw[fill=black] (4,.4) circle (.1);
\filldraw[fill=black] (4,-.4) circle (.1);
\draw[<->,gray] (4,.3) -- (4,-.3);
\end{tikzpicture}
& 
& A.3b
& $\bigl( {\rm A}^{(1)}_n \bigr)^{\psi}_{0;p}$ 
\\[-1em]
\hline 
\hline 
\casesm[.9]{l}{(\mfso_{2n+1}, \\ \qu \mfso_{n-p} \oplus \mfso_{n+1+p})}
& BI, BII 
& $\bigl( {\rm B}_n \bigr)^{\rm id}_p$
&
\begin{tikzpicture}[baseline=-1em,line width=0.7pt,scale=75/100]
\clip (1.3,-.2) rectangle (4.8,.7);
\draw[thick,dashed] (1.5,0) -- (2.5,0);
\draw[thick] (2.5,0) -- (3,0);
\draw[thick,dashed] (3,0) -- (4,0);
\draw[double,->] (4,0) --  (4.4,0);
\filldraw[fill=white] (1.5,0) circle (.1);
\filldraw[fill=white] (2.5,0) circle (.1);
\filldraw[fill=black] (3,0) circle (.1);
\filldraw[fill=black] (4,0) circle (.1);
\filldraw[fill=black] (4.5,0) circle (.1);
\draw[snake=brace] (2.9,.15) -- (4.6,.15) node[midway,above]{\scriptsize $p$};
\end{tikzpicture} 
& $0 \le p \le n$
& B.1a 
& $\bigl( {\rm B}^{(1)}_n \bigr)^{\id}_{0;p}$
\\
\hline
-
& BIII 
& $\bigl( {\rm B}_n \bigr)^{\rm id}_{{\rm alt},p} $ 
& 
\begin{tikzpicture}[baseline=-1em,line width=0.7pt,scale=75/100]
\clip (.8,-.2) rectangle (5.3,.7);
\draw[thick] (1,0) -- (2,0);
\draw[thick,dashed] (2,0) -- (3,0);
\draw[thick] (3,0) -- (3.5,0);
\draw[thick,dashed] (3.5,0) -- (4.5,0);
\draw[double,->] (4.6,0) --  (5,0);
\filldraw[fill=black] (1,0) circle (.1);
\filldraw[fill=white] (1.5,0) circle (.1);
\filldraw[fill=black] (2,0) circle (.1);
\filldraw[fill=white] (3,0) circle (.1);
\filldraw[fill=black] (3.5,0) circle (.1);
\filldraw[fill=black] (4.5,0) circle (.1);
\filldraw[fill=black] (5,0) circle (.1);
\draw[snake=brace] (3.4,.15) -- (5.1,.15) node[midway,above]{\scriptsize $p$};
\end{tikzpicture} 
& \casesm[.9]{l}{0 \leq p \le n\\ n-p \text{ even}} 
& B.2a
& $\bigl( {\rm B}^{(1)}_n \bigr)^{\rm id}_{0;{\rm alt};p} $  \\
\hline
\hline 
\casesm[.9]{l}{(\mfsp_{2n},\mfgl_n) \text{ if } p=0 \\ (\mfsp_{2n},\mfsp_{2n}) \text{ if } p=n}
& CI
& $\bigl( {\rm C}_{n} \bigr)^{\rm id}_{p} $ 
&
\begin{tikzpicture}[baseline=-1em,line width=0.7pt,scale=75/100]
\clip (1.3,-.2) rectangle (4.8,.7);
\draw[thick,dashed] (1.5,0) -- (2.5,0);
\draw[thick] (2.5,0) -- (3,0);
\draw[thick,dashed] (3,0) -- (4,0);
\draw[double,<-] (4,0) --  (4.4,0);
\filldraw[fill=white] (1.5,0) circle (.1);
\filldraw[fill=white] (2.5,0) circle (.1);
\filldraw[fill=black] (3,0) circle (.1);
\filldraw[fill=black] (4,0) circle (.1);
\filldraw[fill=black] (4.5,0) circle (.1);
\draw[snake=brace] (2.9,.15) -- (4.6,.15) node[midway,above]{\scriptsize $p$};
\end{tikzpicture} 
& $0 \le p \le n$
& C.1
& $\bigl( {\rm C}^{(1)}_{n} \bigr)^{\rm id}_{0,p}$ \\
\hline
$(\mfsp_{2n},\mfsp_{n - p} \oplus \mfsp_{n + p})$
& CII 
& $\bigl( {\rm C}_n \bigr)^{\rm id}_{{\rm alt},p} $ 
& 
\begin{tikzpicture}[baseline=-1em,line width=0.7pt,scale=75/100]
\clip (.8,-.2) rectangle (5.3,.7);
\draw[thick] (1,0) -- (2,0);
\draw[thick,dashed] (2,0) -- (3,0);
\draw[thick] (3,0) -- (3.5,0);
\draw[thick,dashed] (3.5,0) -- (4.5,0);
\draw[double,<-] (4.6,0) --  (5,0);
\filldraw[fill=black] (1,0) circle (.1);
\filldraw[fill=white] (1.5,0) circle (.1);
\filldraw[fill=black] (2,0) circle (.1);
\filldraw[fill=white] (3,0) circle (.1);
\filldraw[fill=black] (3.5,0) circle (.1);
\filldraw[fill=black] (4.5,0) circle (.1);
\filldraw[fill=black] (5,0) circle (.1);
\draw[snake=brace] (3.4,.15) -- (5.1,.15) node[midway,above]{\scriptsize $p$};
\end{tikzpicture} 
& \casesm[.9]{l}{0 \leq p \le n\\ n-p \text{ even}} 
& C.2
& $\bigl( {\rm C}^{(1)}_n \bigr)^{\rm id}_{0,{\rm alt},p} $ 
\\
\hline
\hline
\multirow{2}{*}[-7pt]{$(\mfso_{2n},\mfso_{n - p} \oplus \mfso_{n + p})$}
& \multirow{2}{*}{\casesm[1]{l}{\\[-5pt] $DI$, \\ $DII$}}
& $\bigl({\rm D}_{n}\bigr)^{\id}_p$
&
\begin{tikzpicture}[baseline=-1.75em,line width=0.7pt,scale=75/100]
\clip (.8,-.4) rectangle (4.7,1.1);
\draw[thick,dashed] (1.5,0) -- (2.5,0);
\draw[thick] (2.5,0) -- (3,0);
\draw[thick,dashed] (3,0) -- (4,0);
\draw[thick] (4.4,.3) -- (4,0) -- (4.6,-.3);
\filldraw[fill=white] (1.5,0) circle (.1);
\filldraw[fill=white] (2.5,0) circle (.1);
\filldraw[fill=black] (3,0) circle (.1);
\filldraw[fill=black] (4,0) circle (.1);
\filldraw[fill=black] (4.4,.3) circle (.1);
\filldraw[fill=black] (4.6,-.3) circle (.1);
\draw[snake=brace] (2.9,.45) -- (4.7,.45) node[midway,above]{\scriptsize $p$};
\end{tikzpicture} 
& \casesm[.9]{l}{1 \leq p \le n\\ p \text{ even}} 
& D.1a
& $\bigl( {\rm D}^{(1)}_n \bigr)^{\id}_{0,p}$
\\*[-.75em]
& & $\bigl({\rm D}_{n}\bigr)^{\flR}_p$
&
\begin{tikzpicture}[baseline=-1.75em,line width=0.7pt,scale=75/100]
\clip (.8,-.4) rectangle (4.7,1.1);
\draw[thick,dashed] (2.5,0) -- (1.5,0);
\draw[thick] (3,0) -- (2.5,0);
\draw[thick,dashed] (4,0) -- (3,0);
\draw[thick] (4.5,.3) -- (4,0) -- (4.5,-.3);
\draw[<->,gray] (4.5,.2) -- (4.5,-.2);
\filldraw[fill=white] (1.5,0) circle (.1);
\filldraw[fill=white] (2.5,0) circle (.1);
\filldraw[fill=black] (3,0) circle (.1);
\filldraw[fill=black] (4,0) circle (.1);
\filldraw[fill=black] (4.5,.3) circle (.1);
\filldraw[fill=black] (4.5,-.3) circle (.1);
\draw[snake=brace] (2.9,.45) -- (4.6,.45) node[midway,above]{\scriptsize $p$};
\end{tikzpicture}
& \casesm[.9]{l}{1 \leq p < n \\ p \text{ even}} 
& D.1b
& $\bigl( {\rm D}^{(1)}_n \bigr)^{\flR}_{0;p}$
\\[-.75em]
\hline
\multirow{2}{*}[-8pt]{\casesm[.9]{l}{(\mfso_{2n},\mfgl_n) \text{ if } p=0 \\ (\mfso_{2n},\mfso_{2n}) \text{ if } p=n}}
& \multirow{2}{*}[-8pt]{DIII}
& $\bigl({\rm D}_n \bigr)^{\id}_{{\rm alt},p}$
&
\begin{tikzpicture}[baseline=-1.75em,line width=0.7pt,scale=75/100]
\clip (.8,-.4) rectangle (5.2,1.1);
\draw[thick] (1,0) -- (2,0);
\draw[thick,dashed] (2,0) -- (3,0);
\draw[thick] (3,0) -- (3.5,0);
\draw[thick,dashed] (3.5,0) -- (4.5,0);
\draw[thick] (4.9,.3) -- (4.5,0) -- (5.1,-.3);
\filldraw[fill=black] (1,0) circle (.1);
\filldraw[fill=white] (1.5,0) circle (.1);
\filldraw[fill=black] (2,0) circle (.1);
\filldraw[fill=white] (3,0) circle (.1);
\filldraw[fill=black] (3.5,0) circle (.1);
\filldraw[fill=black] (4.5,0) circle (.1);
\filldraw[fill=black] (4.9,.3) circle (.1);
\filldraw[fill=black] (5.1,-.3) circle (.1);
\draw[snake=brace] (3.4,.45) -- (5.2,.45) node[midway,above]{\scriptsize $p$};
\end{tikzpicture}
& \casesm[.9]{l}{0 \leq p \le n\\ p, n \text{ even}}
& D.2a
& $\bigl( {\rm D}^{(1)}_n \bigr)^{\id}_{0,{\rm alt},p}$
\\*[-.75em]
& 
& $\bigl({\rm D}_n \bigr)^{\flR}_{{\rm alt},p}$
&
\begin{tikzpicture}[baseline=-1.75em,line width=0.7pt,scale=75/100]
\clip (.8,-.4) rectangle (5.2,1.1);
\draw[thick] (1,0) -- (2,0);
\draw[thick,dashed] (2,0) -- (3,0);
\draw[thick] (3,0) -- (3.5,0);
\draw[thick,dashed] (3.5,0) -- (4.5,0);
\draw[thick] (5,.3) -- (4.5,0) -- (5,-.3);
\draw[<->,gray] (5,.2) -- (5,-.2);
\filldraw[fill=black] (1,0) circle (.1);
\filldraw[fill=white] (1.5,0) circle (.1);
\filldraw[fill=black] (2,0) circle (.1);
\filldraw[fill=white] (3,0) circle (.1);
\filldraw[fill=black] (3.5,0) circle (.1);
\filldraw[fill=black] (4.5,0) circle (.1);
\filldraw[fill=black] (5,-.3) circle (.1);
\filldraw[fill=black] (5,.3) circle (.1);
\draw[snake=brace] (3.4,.45) -- (5.1,.45) node[midway,above]{\scriptsize $p$};
\end{tikzpicture}
& \casesm[.9]{l}{0 \leq p \le n\\ p \text{ even},\, n \text{ odd}}
& D.2b
& $\bigl( {\rm D}^{(1)}_n \bigr)^{\flR}_{0,{\rm alt},p}$
\\*[-.75em]
\hline
\end{longtable}


\section{Generalized Satake diagrams of low rank} \label{App:LowRank}


In this appendix we list isomorphisms between generalized Satake diagrams of low rank when the underling Dynkin diagram is of affine type.
For the restrictable cases these are in one-to-one correspondence with isomorphisms of irreducible Hermitian symmetric spaces listed in \cite[Sec.~X.6]{He}. 
Mainly, we consider degenerations of affine Dynkin diagrams that follow from the
well-known isomorphisms of simple Lie algebras of low rank: 
\eqn{
& {\rm A}_1 \cong {\rm B}_1 \cong {\rm C}_1 , \qq\;\; 
{\rm B}_2 \cong {\rm C}_2  , \qq\;
{\rm A_4} \cong {\rm D_3} , \\
& \mfsl_2 \cong \mfso_3 \cong \mfsp_2 , \qq
\mfso_5 \cong  \mfsp_4, \qq
\mfsl_4 \cong \mfso_6. 
}
For the affine Dynkin diagrams the corresponding isomorphisms are
\[
{\rm A}^{(1)}_1 \cong {\rm B}^{(1)}_1 \cong {\rm C}^{(1)}_1, \qq {\rm B}^{(1)}_2 \cong {\rm C}^{(1)}_2, \qq {\rm A}^{(1)}_3 \cong {\rm D}^{(1)}_3 
\]
allowing us to let $n$ for types B$^{(1)}_n$, C$^{(1)}_n$ and D$^{(1)}_n$ start at the values $n=1$, $n=1$ and $n=3$, respectively, with the corresponding diagrams listed in Table \ref{Tbl:degenerate} being the natural degenerations of the families of diagrams listed in Table~\ref{DynkinDiagrams} in Section \ref{sec:natrep}.

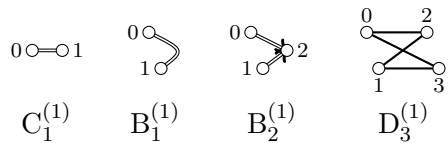
\begin{table}[h]
\caption{ Degenerate Dynkin diagrams of affine type.} \label{Tbl:degenerate} 
\[
\begin{array}{cccc}
\begin{minipage}[c]{12mm}
\begin{tikzpicture}[scale=80/100]
\draw[double] (0,0) -- (.5,0);
\filldraw[fill=white] (0,0) circle (.1) node[left]{\scriptsize 0};
\filldraw[fill=white] (.5,0) circle (.1) node[right]{\scriptsize 1};
\end{tikzpicture}
\end{minipage} 
& 
\begin{minipage}[c]{10mm}
\begin{tikzpicture}[scale=80/100]
\draw[double] (-.6,.3) .. controls (0,0) .. (-.4,-.3);
\filldraw[fill=white] (-.6,.3) circle (.1) node[left]{\scriptsize 0};
\filldraw[fill=white] (-.4,-.3) circle (.1) node[left]{\scriptsize 1};
\end{tikzpicture}
\end{minipage}
& 
\begin{minipage}[c]{14mm}
\begin{tikzpicture}[scale=80/100]
\draw[double,->] (-.6,.3) -- (0,0);
\draw[double,->] (-.4,-.3) -- (0,0);
\filldraw[fill=white] (-.6,.3) circle (.1) node[left]{\scriptsize 0};
\filldraw[fill=white] (-.4,-.3) circle (.1) node[left]{\scriptsize 1};
\filldraw[fill=white] (0,0) circle (.1) node[right]{\scriptsize 2};
\end{tikzpicture}
\end{minipage}
& 
\begin{minipage}[c]{14mm}
\begin{tikzpicture}[scale=80/100]
\draw[thick] (-.6,.3) -- (.4,.3) -- (-.4,-.3) -- (.6,-.3) -- (-.6,.3);
\filldraw[fill=white] (-.6,.3) circle (.1) node[above]{\scriptsize 0};
\filldraw[fill=white] (-.4,-.3) circle (.1) node[below]{\scriptsize 1};
\filldraw[fill=white] (.4,.3) circle (.1) node[above]{\scriptsize 2};
\filldraw[fill=white] (.6,-.3) circle (.1) node[below]{\scriptsize 3};
\end{tikzpicture}
\end{minipage}
\\
{\rm C}^{(1)}_1 & {\rm B}^{(1)}_1 & {\rm B}^{(1)}_2 & {\rm D}^{(1)}_3 
\end{array}
\]
\end{table}

Moreover, for Dynkin diagrams of type ${\rm A}^{(1)}_1$ certain diagram involutions coincide which are distinct when $n>1$ so that certain Satake diagrams can be identified. In Tables \ref{Tbl:isos:1}--\ref{Tbl:isos:3} we list the isomorphisms of generalized Satake diagrams of low rank.
Finally, if $A$ is of type ${\rm D}^{(1)}_4$ the group $\Aut(A)$ is strictly larger than $\Sigma_A$; some $\Aut(A)$-equivalence classes coincide with a $\Sigma_A$-equivalence class (see Table \ref{Tbl:isos:4a}) whereas others are the union of two $\Sigma_A$-equivalence classes (see Table \ref{Tbl:isos:4b}).

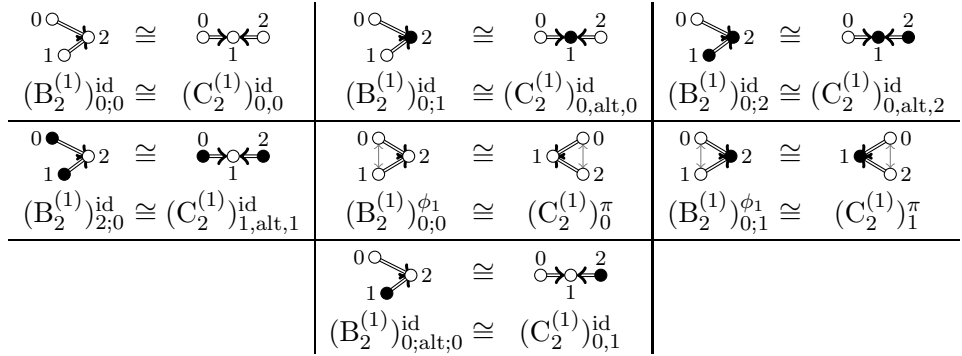
\captionsetup{width=\textwidth}

\begin{longtable}{p{\textwidth}}
\caption{Isomorphisms between the three representative Satake diagrams of type ${\rm A}^{(1)}_1$ (written in two different ways), ${\rm B}^{(1)}_1$ and ${\rm C}^{(1)}_1$, respectively. } \label{Tbl:isos:1} \\
\nopagebreak 
\vspace{-2em}
\[
\begin{array}{c@{\;\cong\;}c@{\;\cong\;}c@{\;\cong\;}c}
\begin{minipage}[c]{7mm}
\begin{tikzpicture}[scale=70/100]
\draw[thick,domain=0:360] plot ({.4*cos(\x)},{.4*sin(\x)});
\filldraw[fill=white] (0,.4) circle (.1) node[above]{\scriptsize 0};
\filldraw[fill=white] (0,-.4) circle (.1) node[below]{\scriptsize 1};
\end{tikzpicture}
\end{minipage}
&
\begin{minipage}[c]{16mm}
\begin{tikzpicture}[scale=70/100]
\draw[thick] (0,0) .. controls (.5,.5) .. (1,0) .. controls (.5,-.5) .. (0,0);
\filldraw[fill=white] (0,0) circle (.1) node[left]{\scriptsize 0};
\filldraw[fill=white] (1,0) circle (.1) node[right]{\scriptsize 1};
\end{tikzpicture}
\end{minipage}
&
\begin{minipage}[c]{10mm}
\begin{tikzpicture}[scale=70/100]
\draw[double] (-.6,.3) .. controls (0,0) .. (-.4,-.3);
\filldraw[fill=white] (-.6,.3) circle (.1) node[left]{\scriptsize 0};
\filldraw[fill=white] (-.4,-.3) circle (.1) node[left]{\scriptsize 1};
\end{tikzpicture}
\end{minipage}
&
\begin{minipage}[c]{12mm}
\begin{tikzpicture}[scale=70/100]
\draw[double] (0,0) -- (.5,0);
\filldraw[fill=white] (0,0) circle (.1) node[left]{\scriptsize 0};
\filldraw[fill=white] (.5,0) circle (.1) node[right]{\scriptsize 1};
\end{tikzpicture}
\end{minipage}
\\[0mm]
({\rm A}^{(1)}_1)^{\id}_\emptyset & ({\rm A}^{(1)}_1)^{\psi}_{0,0} & ({\rm B}^{(1)}_1)^{\id}_{0;0} & ({\rm C}^{(1)}_1)^{\id}_{0,0}, \\[1mm]
\hline
\begin{minipage}[c]{7mm}
\begin{tikzpicture}[scale=70/100]
\draw[thick,domain=0:360] plot ({.4*cos(\x)},{.4*sin(\x)});
\filldraw[fill=white] (0,.4) circle (.1) node[above]{\scriptsize 0};
\filldraw[fill=black] (0,-.4) circle (.1) node[below]{\scriptsize 1};
\end{tikzpicture}
\end{minipage}
&
\begin{minipage}[c]{16mm}
\begin{tikzpicture}[scale=70/100]
\draw[thick] (0,0) .. controls (.5,.5) .. (1,0) .. controls (.5,-.5) .. (0,0);
\filldraw[fill=white] (0,0) circle (.1) node[left]{\scriptsize 0};
\filldraw[fill=black] (1,0) circle (.1) node[right]{\scriptsize 1};
\end{tikzpicture}
\end{minipage}
&
\begin{minipage}[c]{10mm}
\begin{tikzpicture}[scale=70/100]
\draw[double] (-.6,.3) .. controls (0,0) .. (-.4,-.3);
\filldraw[fill=white] (-.6,.3) circle (.1) node[left]{\scriptsize 0};
\filldraw[fill=black] (-.4,-.3) circle (.1) node[left]{\scriptsize 1};
\end{tikzpicture}
\end{minipage}
&
\begin{minipage}[c]{12mm}
\begin{tikzpicture}[scale=70/100]
\draw[double] (0,0) -- (.5,0);
\filldraw[fill=white] (0,0) circle (.1) node[left]{\scriptsize 0};
\filldraw[fill=black] (.5,0) circle (.1) node[right]{\scriptsize 1};
\end{tikzpicture}
\end{minipage}
\\[0mm]
({\rm A}^{(1)}_1)^{\id}_{\rm alt} & ({\rm A}^{(1)}_1)^{\psi}_{0,1} & ({\rm B}^{(1)}_1)^{\id}_{0;1} & ({\rm C}^{(1)}_1)^{\id}_{0,{\rm alt},1} \\[1mm]
\hline
\begin{minipage}[c]{7mm}
\begin{tikzpicture}[scale=70/100]
\draw[thick,domain=0:360] plot ({.4*cos(\x)},{.4*sin(\x)});
\draw[<->,gray] (0,.3) -- (0,-.3);
\filldraw[fill=white] (0,.4) circle (.1) node[above]{\scriptsize 0};
\filldraw[fill=white] (0,-.4) circle (.1) node[below]{\scriptsize 1};
\end{tikzpicture}
\end{minipage}
&
\begin{minipage}[c]{7mm}
\begin{tikzpicture}[scale=70/100]
\draw[thick,domain=0:360] plot ({.4*cos(\x)},{.4*sin(\x)});
\draw[<->,gray] (0,.3) -- (0,-.3);
\filldraw[fill=white] (0,.4) circle (.1) node[above]{\scriptsize 0};
\filldraw[fill=white] (0,-.4) circle (.1) node[below]{\scriptsize 1};
\end{tikzpicture}
\end{minipage}
&
\begin{minipage}[c]{9mm}
\begin{tikzpicture}[scale=70/100]
\draw[double,domain=-90:90] plot ({.3*cos(\x)},{.3*sin(\x)});
\draw[<->,gray] (0,.2) -- (0,-.2);
\filldraw[fill=white] (0,.3) circle (.1) node[left]{\scriptsize 0};
\filldraw[fill=white] (0,-.3) circle (.1) node[left]{\scriptsize 1};
\end{tikzpicture}
\end{minipage}
&
\begin{minipage}[c]{9mm}
\begin{tikzpicture}[scale=70/100]
\draw[double,domain=90:270] plot ({.3*cos(\x)},{.3*sin(\x)});
\draw[<->,gray] (0,.2) -- (0,-.2);
\filldraw[fill=white] (0,.3) circle (.1) node[right]{\scriptsize 0};
\filldraw[fill=white] (0,-.3) circle (.1) node[right]{\scriptsize 1};
\end{tikzpicture}
\end{minipage}
\\[0mm]
({\rm A}^{(1)}_1)^{\pi}_\emptyset & ({\rm A}^{(1)}_1)^{\psi'}_{0,0} & ({\rm B}^{(1)}_1)^{\flL}_{0;0} & ({\rm C}^{(1)}_1)^\pi_{0}
\end{array}
\]
\end{longtable}

\vspace{-1em}

\begin{longtable}{p{\textwidth}}
\caption{Isomorphisms between the seven representative generalized Satake diagrams of types ${\rm B}^{(1)}_2$ and ${\rm C}^{(1)}_2$.
The diagram in the bottom row is the only weak Satake diagram. 
Because of the different labelling of the nodes the isomorphism is given by $1 \leftrightarrow 2$.} \label{Tbl:isos:2}\\
\nopagebreak 
\vspace{-1.5em}
\[
\begin{array}{c@{\;\cong\;}c|c@{\;\cong\;}c|c@{\;\cong\;}c}
\begin{minipage}[c]{14mm}
\begin{tikzpicture}[scale=80/100]
\draw[double,->] (-.6,.3) -- (0,0);
\draw[double,->] (-.4,-.3) -- (0,0);
\filldraw[fill=white] (-.6,.3) circle (.1) node[left]{\scriptsize 0};
\filldraw[fill=white] (-.4,-.3) circle (.1) node[left]{\scriptsize 1};
\filldraw[fill=white] (0,0) circle (.1) node[right]{\scriptsize 2};
\end{tikzpicture}
\end{minipage}
& \begin{minipage}[c]{12mm}
\begin{tikzpicture}[scale=80/100]
\draw[double,<-] (-.1,0) -- (-.5,0);
\draw[double,<-] (.1,0) -- (.5,0);
\filldraw[fill=white] (-.5,0) circle (.1) node[above]{\scriptsize 0};
\filldraw[fill=white] (0,0) circle (.1) node[below]{\scriptsize 1};
\filldraw[fill=white] (.5,0) circle (.1) node[above]{\scriptsize 2};
\end{tikzpicture} 
\end{minipage} 
& \begin{minipage}[c]{14mm}
\begin{tikzpicture}[scale=80/100]
\draw[double,->] (-.6,.3) -- (0,0);
\draw[double,->] (-.4,-.3) -- (0,0);
\filldraw[fill=white] (-.6,.3) circle (.1) node[left]{\scriptsize 0};
\filldraw[fill=white] (-.4,-.3) circle (.1) node[left]{\scriptsize 1};
\filldraw[fill=black] (0,0) circle (.1) node[right]{\scriptsize 2};
\end{tikzpicture}
\end{minipage}
& \begin{minipage}[c]{12mm}
\begin{tikzpicture}[scale=80/100]
\draw[double,<-] (-.1,0) -- (-.5,0);
\draw[double,<-] (.1,0) -- (.5,0);
\filldraw[fill=white] (-.5,0) circle (.1) node[above]{\scriptsize 0};
\filldraw[fill=black] (0,0) circle (.1) node[below]{\scriptsize 1};
\filldraw[fill=white] (.5,0) circle (.1) node[above]{\scriptsize 2};
\end{tikzpicture} 
\end{minipage}
& \begin{minipage}[c]{14mm}
\begin{tikzpicture}[scale=80/100]
\draw[double,->] (-.6,.3) -- (0,0);
\draw[double,->] (-.4,-.3) -- (0,0);
\filldraw[fill=white] (-.6,.3) circle (.1) node[left]{\scriptsize 0};
\filldraw[fill=black] (-.4,-.3) circle (.1) node[left]{\scriptsize 1};
\filldraw[fill=black] (0,0) circle (.1) node[right]{\scriptsize 2};
\end{tikzpicture}
\end{minipage}
& \begin{minipage}[c]{12mm}
\begin{tikzpicture}[scale=80/100]
\draw[double,<-] (-.1,0) -- (-.5,0);
\draw[double,<-] (.1,0) -- (.5,0);
\filldraw[fill=white] (-.5,0) circle (.1) node[above]{\scriptsize 0};
\filldraw[fill=black] (0,0) circle (.1) node[below]{\scriptsize 1};
\filldraw[fill=black] (.5,0) circle (.1) node[above]{\scriptsize 2};
\end{tikzpicture} 
\end{minipage}
\\[1mm]
({\rm B}^{(1)}_2)^{\id}_{0;0} & ({\rm C}^{(1)}_2)^{\id}_{0,0} 
& ({\rm B}^{(1)}_2)^{\id}_{0;1} & ({\rm C}^{(1)}_2)^{\id}_{0,{\rm alt},0} 
& ({\rm B}^{(1)}_2)^{\id}_{0;2} & ({\rm C}^{(1)}_2)^{\id}_{0,{\rm alt},2} \\[1mm]
\hline
\begin{minipage}[c]{14mm}
\begin{tikzpicture}[scale=80/100]
\draw[double,->] (-.6,.3) -- (0,0);
\draw[double,->] (-.4,-.3) -- (0,0);
\filldraw[fill=black] (-.6,.3) circle (.1) node[left]{\scriptsize 0};
\filldraw[fill=black] (-.4,-.3) circle (.1) node[left]{\scriptsize 1};
\filldraw[fill=white] (0,0) circle (.1) node[right]{\scriptsize 2};
\end{tikzpicture}
\end{minipage}
& \begin{minipage}[c]{12mm}
\begin{tikzpicture}[scale=80/100]
\draw[double,<-] (-.1,0) -- (-.5,0);
\draw[double,<-] (.1,0) -- (.5,0);
\filldraw[fill=black] (-.5,0) circle (.1) node[above]{\scriptsize 0};
\filldraw[fill=white] (0,0) circle (.1) node[below]{\scriptsize 1};
\filldraw[fill=black] (.5,0) circle (.1) node[above]{\scriptsize 2};
\end{tikzpicture} 
\end{minipage}
& \begin{minipage}[c]{13mm}
\begin{tikzpicture}[scale=80/100]
\draw[double,->] (-.5,.3) -- (0,0);
\draw[double,->] (-.5,-.3) -- (0,0);
\draw[<->,gray] (-.5,.2) -- (-.5,-.2);
\filldraw[fill=white] (-.5,.3) circle (.1) node[left]{\scriptsize 0};
\filldraw[fill=white] (-.5,-.3) circle (.1) node[left]{\scriptsize 1};
\filldraw[fill=white] (0,0) circle (.1) node[right]{\scriptsize 2};
\end{tikzpicture}
\end{minipage} 
& \begin{minipage}[c]{13mm}
\begin{tikzpicture}[scale=80/100]
\draw[double,->] (.5,.3) -- (0,0);
\draw[double,->] (.5,-.3) -- (0,0);
\draw[<->,gray] (.5,.2) -- (.5,-.2);
\filldraw[fill=white] (.5,.3) circle (.1) node[right]{\scriptsize 0};
\filldraw[fill=white] (.5,-.3) circle (.1) node[right]{\scriptsize 2};
\filldraw[fill=white] (0,0) circle (.1) node[left]{\scriptsize 1};
\end{tikzpicture}
\end{minipage} 
& \begin{minipage}[c]{13mm}
\begin{tikzpicture}[scale=80/100]
\draw[double,->] (-.5,.3) -- (0,0);
\draw[double,->] (-.5,-.3) -- (0,0);
\draw[<->,gray] (-.5,.2) -- (-.5,-.2);
\filldraw[fill=white] (-.5,.3) circle (.1) node[left]{\scriptsize 0};
\filldraw[fill=white] (-.5,-.3) circle (.1) node[left]{\scriptsize 1};
\filldraw[fill=black] (0,0) circle (.1) node[right]{\scriptsize 2};
\end{tikzpicture}
\end{minipage} 
& \begin{minipage}[c]{13mm}
\begin{tikzpicture}[scale=80/100]
\draw[double,->] (.5,.3) -- (0,0);
\draw[double,->] (.5,-.3) -- (0,0);
\draw[<->,gray] (.5,.2) -- (.5,-.2);
\filldraw[fill=white] (.5,.3) circle (.1) node[right]{\scriptsize 0};
\filldraw[fill=white] (.5,-.3) circle (.1) node[right]{\scriptsize 2};
\filldraw[fill=black] (0,0) circle (.1) node[left]{\scriptsize 1};
\end{tikzpicture}
\end{minipage} 
\\[2mm]
({\rm B}^{(1)}_2)^{\id}_{2;0} & ({\rm C}^{(1)}_2)^{\id}_{1,{\rm alt},1} 
& ({\rm B}^{(1)}_2)^{\flL}_{0;0} &  ({\rm C}^{(1)}_2)^{\pi}_0 
& ({\rm B}^{(1)}_2)^{\flL}_{0;1} & ({\rm C}^{(1)}_2)^{\pi}_1 
\\[1mm]
\hline
\multicolumn{2}{c|}{}
&
\begin{minipage}[c]{14mm}
\begin{tikzpicture}[scale=80/100]
\draw[double,->] (-.6,.3) -- (0,0);
\draw[double,->] (-.4,-.3) -- (0,0);
\filldraw[fill=white] (-.6,.3) circle (.1) node[left]{\scriptsize 0};
\filldraw[fill=black] (-.4,-.3) circle (.1) node[left]{\scriptsize 1};
\filldraw[fill=white] (0,0) circle (.1) node[right]{\scriptsize 2};
\end{tikzpicture}
\end{minipage}
& \begin{minipage}[c]{12mm}
\begin{tikzpicture}[scale=80/100]
\draw[double,<-] (-.1,0) -- (-.5,0);
\draw[double,<-] (.1,0) -- (.5,0);
\filldraw[fill=white] (-.5,0) circle (.1) node[above]{\scriptsize 0};
\filldraw[fill=white] (0,0) circle (.1) node[below]{\scriptsize 1};
\filldraw[fill=black] (.5,0) circle (.1) node[above]{\scriptsize 2};
\end{tikzpicture} 
\end{minipage}
&
\multicolumn{2}{c}{}
\\[1mm]
\multicolumn{2}{c|}{}
&
({\rm B}^{(1)}_2)^{\id}_{0;{\rm alt};0} & ({\rm C}^{(1)}_2)^{\id}_{0,1} 
&
\multicolumn{2}{c}{}
\end{array}
\]
\end{longtable}

\vspace{-1em}

\begin{longtable}{p{\textwidth}}
\caption{Isomorphisms between the nine representative Satake diagrams of type ${\rm A}^{(1)}_3$ and ${\rm D}^{(1)}_3$. Again, the nodes are labelled differently and accordingly the isomorphism is given by $1 \leftrightarrow 2$.} \label{Tbl:isos:3} \\
\nopagebreak 
\vspace{-1.5em}
\[
\begin{array}{c@{\,\cong\,}c|c@{\,\cong\,}c|c@{\,\cong\,}c}
\begin{minipage}[c]{16mm}
\begin{tikzpicture}[scale=80/100]
\draw[thick,domain=0:360] plot ({.5*cos(\x)},{.5*sin(\x)});
\filldraw[fill=white] (0,.5) circle (.1) node[above=-1pt]{\scriptsize 0};
\filldraw[fill=white] (.5,0) circle (.1) node[right]{\scriptsize 1};
\filldraw[fill=white] (0,-.5) circle (.1) node[below=-.5pt]{\scriptsize 2};
\filldraw[fill=white] (-.5,0) circle (.1) node[left]{\scriptsize 3};
\end{tikzpicture}
\end{minipage}
& \begin{minipage}[c]{14mm}
\begin{tikzpicture}[scale=80/100]
\draw[thick] (-.6,.3) -- (.4,.3) -- (-.4,-.3) -- (.6,-.3) -- (-.6,.3);
\filldraw[fill=white] (-.6,.3) circle (.1) node[above]{\scriptsize 0};
\filldraw[fill=white] (-.4,-.3) circle (.1) node[below]{\scriptsize 1};
\filldraw[fill=white] (.4,.3) circle (.1) node[above]{\scriptsize 2};
\filldraw[fill=white] (.6,-.3) circle (.1) node[below]{\scriptsize 3};
\end{tikzpicture}
\end{minipage}
& \begin{minipage}[c]{16mm}
\begin{tikzpicture}[scale=80/100]
\draw[thick,domain=0:360] plot ({.5*cos(\x)},{.5*sin(\x)});
\filldraw[fill=white] (0,.5) circle (.1) node[above=-1pt]{\scriptsize 0};
\filldraw[fill=black] (.5,0) circle (.1) node[right]{\scriptsize 1};
\filldraw[fill=white] (0,-.5) circle (.1) node[below=-.5pt]{\scriptsize 2};
\filldraw[fill=black] (-.5,0) circle (.1) node[left]{\scriptsize 3};
\end{tikzpicture}
\end{minipage}
& \begin{minipage}[c]{14mm}
\begin{tikzpicture}[scale=80/100]
\draw[thick] (-.6,.3) -- (.4,.3) -- (-.4,-.3) -- (.6,-.3) -- (-.6,.3);
\filldraw[fill=white] (-.6,.3) circle (.1) node[above]{\scriptsize 0};
\filldraw[fill=white] (-.4,-.3) circle (.1) node[below]{\scriptsize 1};
\filldraw[fill=black] (.4,.3) circle (.1) node[above]{\scriptsize 2};
\filldraw[fill=black] (.6,-.3) circle (.1) node[below]{\scriptsize 3};
\end{tikzpicture}
\end{minipage}
& \begin{minipage}[c]{16mm}
\begin{tikzpicture}[scale=80/100]
\draw[thick,domain=0:360] plot ({.5*cos(\x)},{.5*sin(\x)});
\draw[<->,gray] (0,.4) -- (0,-.4);
\draw[<->,gray] (-.4,0) -- (.4,0);
\filldraw[fill=white] (0,.5) circle (.1) node[above=-1pt]{\scriptsize 0};
\filldraw[fill=white] (.5,0) circle (.1) node[right]{\scriptsize 1};
\filldraw[fill=white] (0,-.5) circle (.1) node[below=-.5pt]{\scriptsize 2};
\filldraw[fill=white] (-.5,0) circle (.1) node[left]{\scriptsize 3};
\end{tikzpicture}
\end{minipage}
& \begin{minipage}[c]{12mm}
\begin{tikzpicture}[scale=80/100]
\draw[thick] (-.5,.3) -- (.5,.3) -- (-.5,-.3) -- (.5,-.3) -- (-.5,.3);
\draw[<->,gray] (-.5,.2) -- (-.5,-.2);
\draw[<->,gray] (.5,.2) -- (.5,-.2);
\filldraw[fill=white] (-.5,.3) circle (.1) node[above]{\scriptsize 0};
\filldraw[fill=white] (-.5,-.3) circle (.1) node[below]{\scriptsize 1};
\filldraw[fill=white] (.5,.3) circle (.1) node[above]{\scriptsize 2};
\filldraw[fill=white] (.5,-.3) circle (.1) node[below]{\scriptsize 3};
\end{tikzpicture}
\end{minipage}
\\[1mm]
({\rm A}^{(1)}_3)^{\id}_\emptyset & ({\rm D}^{(1)}_3)^{\id}_{0,0} 
& ({\rm A}^{(1)}_3)^{\id}_{\rm alt} & ({\rm D}^{(1)}_3)^{\id}_{0,2} 
& ({\rm A}^{(1)}_3)^{\pi}_\emptyset  & ({\rm D}^{(1)}_3)^{\flLR}_{0,0} \\[1mm]
\hline
\begin{minipage}[c]{16mm}
\begin{tikzpicture}[scale=80/100]
\draw[thick] (0,.4) -- (-.5,0) -- (0,-.4);
\draw[thick] (0,.4) -- (.5,0) -- (0,-.4);
\draw[<->,gray] (0,.3) -- (0,-.3);
\filldraw[fill=white] (-.5,0) circle (.1) node[left]{\scriptsize 0};
\filldraw[fill=white] (0,.4) circle (.1) node[above=-1pt]{\scriptsize 1};
\filldraw[fill=white] (0,-.4) circle (.1) node[below=-.5pt]{\scriptsize 3};
\filldraw[fill=white] (.5,0) circle (.1) node[right]{\scriptsize 2};
\end{tikzpicture} 
\end{minipage}
& \begin{minipage}[c]{13mm}
\begin{tikzpicture}[scale=80/100]
\draw[thick] (-.6,.3) -- (.5,.3) -- (-.4,-.3) -- (.5,-.3) -- (-.6,.3);
\draw[<->,gray] (.5,.2) -- (.5,-.2);
\filldraw[fill=white] (-.6,.3) circle (.1) node[above]{\scriptsize 0};
\filldraw[fill=white] (-.4,-.3) circle (.1) node[below]{\scriptsize 1};
\filldraw[fill=white] (.5,.3) circle (.1) node[above]{\scriptsize 2};
\filldraw[fill=white] (.5,-.3) circle (.1) node[below]{\scriptsize 3};
\end{tikzpicture}
\end{minipage}
& \begin{minipage}[c]{16mm}
\begin{tikzpicture}[scale=80/100]
\draw[thick] (0,.4) -- (-.5,0) -- (0,-.4);
\draw[thick] (0,.4) -- (.5,0) -- (0,-.4);
\draw[<->,gray] (0,.3) -- (0,-.3);
\filldraw[fill=white] (-.5,0) circle (.1) node[left]{\scriptsize 0};
\filldraw[fill=white] (0,.4) circle (.1) node[above=-1pt]{\scriptsize 1};
\filldraw[fill=white] (0,-.4) circle (.1) node[below=-.5pt]{\scriptsize 3};
\filldraw[fill=black] (.5,0) circle (.1) node[right]{\scriptsize 2};
\end{tikzpicture} 
\end{minipage}
& \begin{minipage}[c]{13mm}
\begin{tikzpicture}[scale=80/100]
\draw[thick] (-.6,.3) -- (.5,.3) -- (-.4,-.3) -- (.5,-.3) -- (-.6,.3);
\draw[<->,gray] (.5,.2) -- (.5,-.2);
\filldraw[fill=white] (-.6,.3) circle (.1) node[above]{\scriptsize 0};
\filldraw[fill=black] (-.4,-.3) circle (.1) node[below]{\scriptsize 1};
\filldraw[fill=white] (.5,.3) circle (.1) node[above]{\scriptsize 2};
\filldraw[fill=white] (.5,-.3) circle (.1) node[below]{\scriptsize 3};
\end{tikzpicture}
\end{minipage}
& \begin{minipage}[c]{16mm}
\begin{tikzpicture}[scale=80/100]
\draw[thick] (0,.4) -- (-.5,0) -- (0,-.4);
\draw[thick] (0,.4) -- (.5,0) -- (0,-.4);
\draw[<->,gray] (0,.3) -- (0,-.3);
\filldraw[fill=white] (-.5,0) circle (.1) node[left]{\scriptsize 0};
\filldraw[fill=black] (0,.4) circle (.1) node[above=-1pt]{\scriptsize 1};
\filldraw[fill=black] (0,-.4) circle (.1) node[below=-.5pt]{\scriptsize 3};
\filldraw[fill=black] (.5,0) circle (.1) node[right]{\scriptsize 2};
\end{tikzpicture} 
\end{minipage}
& \begin{minipage}[c]{13mm}
\begin{tikzpicture}[scale=80/100]
\draw[thick] (-.6,.3) -- (.5,.3) -- (-.4,-.3) -- (.5,-.3) -- (-.6,.3);
\draw[<->,gray] (.5,.2) -- (.5,-.2);
\filldraw[fill=white] (-.6,.3) circle (.1) node[above]{\scriptsize 0};
\filldraw[fill=black] (-.4,-.3) circle (.1) node[below]{\scriptsize 1};
\filldraw[fill=black] (.5,.3) circle (.1) node[above]{\scriptsize 2};
\filldraw[fill=black] (.5,-.3) circle (.1) node[below]{\scriptsize 3};
\end{tikzpicture}
\end{minipage}
\\[1mm]
({\rm A}^{(1)}_3)^{\psi}_{0,0} & ({\rm D}^{(1)}_3)^{\flR}_{0;0} 
& ({\rm A}^{(1)}_3)^{\psi}_{0,1} & ({\rm D}^{(1)}_3)^{\flR}_{\rm alt} 
& ({\rm A}^{(1)}_3)^{\psi}_{0,2} & ({\rm D}^{(1)}_3)^{\flR}_{0;2} \\[1mm]
\hline
\begin{minipage}[c]{16mm}
\begin{tikzpicture}[scale=80/100]
\draw[thick] (0,.4) -- (-.5,0) -- (0,-.4);
\draw[thick] (0,.4) -- (.5,0) -- (0,-.4);
\draw[<->,gray] (0,.3) -- (0,-.3);
\filldraw[fill=black] (-.5,0) circle (.1) node[left]{\scriptsize 0};
\filldraw[fill=white] (0,.4) circle (.1) node[above=-1pt]{\scriptsize 1};
\filldraw[fill=white] (0,-.4) circle (.1) node[below=-.5pt]{\scriptsize 3};
\filldraw[fill=black] (.5,0) circle (.1) node[right]{\scriptsize 2};
\end{tikzpicture} 
\end{minipage}
& \begin{minipage}[c]{13mm}
\begin{tikzpicture}[scale=80/100]
\draw[thick] (-.6,.3) -- (.5,.3) -- (-.4,-.3) -- (.5,-.3) -- (-.6,.3);
\draw[<->,gray] (.5,.2) -- (.5,-.2);
\filldraw[fill=black] (-.6,.3) circle (.1) node[above]{\scriptsize 0};
\filldraw[fill=black] (-.4,-.3) circle (.1) node[below]{\scriptsize 1};
\filldraw[fill=white] (.5,.3) circle (.1) node[above]{\scriptsize 2};
\filldraw[fill=white] (.5,-.3) circle (.1) node[below]{\scriptsize 3};
\end{tikzpicture}
\end{minipage}
& \begin{minipage}[c]{10mm}
\begin{tikzpicture}[scale=80/100]
\draw[thick,domain=90:270] plot({.4*cos(\x)},{.4*sin(\x)});
\draw[thick] (0,.4) -- (.5,.4);
\draw[thick] (0,-.4) -- (.5,-.4);
\draw[thick,domain=270:450] plot({.5+.4*cos(\x)},{.4*sin(\x)});
\draw[<->,gray] (0,.3) -- (0,-.3);
\draw[<->,gray] (.5,.3) -- (.5,-.3);
\filldraw[fill=white] (0,.4) circle (.1) node[above=-1pt]{\scriptsize 0};
\filldraw[fill=white] (0,-.4) circle (.1) node[below=-.5pt]{\scriptsize 3};
\filldraw[fill=white] (.5,.4) circle (.1) node[above=-1pt]{\scriptsize 1};
\filldraw[fill=white] (.5,-.4) circle (.1) node[below=-.5pt]{\scriptsize 2};
\end{tikzpicture} 
\end{minipage}
& \begin{minipage}[c]{10mm}
\begin{tikzpicture}[scale=80/100]
\draw[thick] (.1,.7) -- (-.1,.7);
\draw[thick] (.1,-.2) -- (-.1,-.2);
\draw[thick,domain=90:270] plot({-.1+.45*cos(\x)},{.25+.45*sin(\x)});
\draw[thick,domain=90:270] plot({-.1+.45*cos(\x)},{-.15+.45*sin(\x)});
\draw[thick,domain=90:270] plot({-.1+.65*cos(\x)},{.05+.65*sin(\x)});
\draw[thick,domain=90:270] plot({-.1+.25*cos(\x)},{.05+.25*sin(\x)});
\draw[<->,gray] (-.1,.2) -- (-.1,-.5);
\draw[<->,gray] (.1,.6) -- (.1,-.1);
\filldraw[fill=white] (-.1,.3) circle (.1) node[above=-1pt]{\scriptsize 1};
\filldraw[fill=white] (-.1,-.6) circle (.1) node[right]{\scriptsize 2};
\filldraw[fill=white] (.1,.7) circle (.1) node[right]{\scriptsize 0};
\filldraw[fill=white] (.1,-.2) circle (.1) node[right]{\scriptsize 3};
\end{tikzpicture}
\end{minipage}
& \begin{minipage}[c]{10mm}
\begin{tikzpicture}[scale=80/100]
\draw[thick,domain=90:270] plot({.4*cos(\x)},{.4*sin(\x)});
\draw[thick] (0,.4) -- (.5,.4);
\draw[thick] (0,-.4) -- (.5,-.4);
\draw[thick,domain=270:450] plot({.5+.4*cos(\x)},{.4*sin(\x)});
\draw[<->,gray] (0,.3) -- (0,-.3);
\draw[<->,gray] (.5,.3) -- (.5,-.3);
\filldraw[fill=white] (0,.4) circle (.1) node[above=-1pt]{\scriptsize 0};
\filldraw[fill=white] (0,-.4) circle (.1) node[below=-.5pt]{\scriptsize 3};
\filldraw[fill=black] (.5,.4) circle (.1) node[above=-1pt]{\scriptsize 1};
\filldraw[fill=black] (.5,-.4) circle (.1) node[below=-.5pt]{\scriptsize 2};
\end{tikzpicture} 
\end{minipage}
& \begin{minipage}[c]{10mm}
\begin{tikzpicture}[scale=80/100]
\draw[thick] (.1,.7) -- (-.1,.7);
\draw[thick] (.1,-.2) -- (-.1,-.2);
\draw[thick,domain=90:270] plot({-.1+.45*cos(\x)},{.25+.45*sin(\x)});
\draw[thick,domain=90:270] plot({-.1+.45*cos(\x)},{-.15+.45*sin(\x)});
\draw[thick,domain=90:270] plot({-.1+.65*cos(\x)},{.05+.65*sin(\x)});
\draw[thick,domain=90:270] plot({-.1+.25*cos(\x)},{.05+.25*sin(\x)});
\draw[<->,gray] (-.1,.2) -- (-.1,-.5);
\draw[<->,gray] (.1,.6) -- (.1,-.1);
\filldraw[fill=black] (-.1,.3) circle (.1) node[above=-1pt]{\scriptsize 1};
\filldraw[fill=black] (-.1,-.6) circle (.1) node[right]{\scriptsize 2};
\filldraw[fill=white] (.1,.7) circle (.1) node[right]{\scriptsize 0};
\filldraw[fill=white] (.1,-.2) circle (.1) node[right]{\scriptsize 3};
\end{tikzpicture}
\end{minipage} 
\\[1mm]
({\rm A}^{(1)}_3)^{\psi}_{1,1} & ({\rm D}^{(1)}_3)^{\flR}_{2;0}  
& ({\rm A}^{(1)}_3)^{\psi'}_{0,0} & ({\rm D}^{(1)}_3)^{\pi}_0 
& ({\rm A}^{(1)}_3)^{\psi'}_{0,1} & ({\rm D}^{(1)}_3)^{\pi}_{1} 
\end{array}
\]
\end{longtable}

\vspace{-1em}

\begin{table}[h]
\caption{Satake diagrams of type ${\rm D}^{(1)}_4$ whose $\Aut(A)$- and $\Sigma_A$-equivalence class coincide.} \label{Tbl:isos:4a}
\nopagebreak 
\vspace{-1.5em}
\[
\begin{array}{c @{\qquad} c @{\qquad} c}
\begin{minipage}[c]{17mm}
\begin{tikzpicture}[scale=80/100]
\draw[thick] (-.6,.3) -- (0,0) -- (-.4,-.3);
\draw[thick] (.4,.3) -- (0,0) -- (.6,-.3);
\filldraw[fill=white] (-.6,.3) circle (.1) node[left]{\scriptsize 0};
\filldraw[fill=white] (-.4,-.3) circle (.1) node[left]{\scriptsize 1};
\filldraw[fill=white] (0,0) circle (.1) node[above]{\scriptsize 2};
\filldraw[fill=white] (.6,-.3) circle (.1) node[right]{\scriptsize 4};
\filldraw[fill=white] (.4,.3) circle (.1) node[right]{\scriptsize 3};
\end{tikzpicture}
\end{minipage} 
& 
\begin{minipage}[c]{17mm}
\begin{tikzpicture}[scale=80/100]
\draw[thick] (-.6,.3) -- (0,0) -- (-.4,-.3);
\draw[thick] (.4,.3) -- (0,0) -- (.6,-.3);
\filldraw[fill=white] (-.6,.3) circle (.1) node[left]{\scriptsize 0};
\filldraw[fill=black] (-.4,-.3) circle (.1) node[left]{\scriptsize 1};
\filldraw[fill=black] (0,0) circle (.1) node[above]{\scriptsize 2};
\filldraw[fill=black] (.6,-.3) circle (.1) node[right]{\scriptsize 4};
\filldraw[fill=black] (.4,.3) circle (.1) node[right]{\scriptsize 3};
\end{tikzpicture}
\end{minipage} 
& 
\begin{minipage}[c]{17mm}
\begin{tikzpicture}[scale=80/100]
\draw[thick] (-.6,.3) -- (0,0) -- (-.4,-.3);
\draw[thick] (.4,.3) -- (0,0) -- (.6,-.3);
\filldraw[fill=black] (-.6,.3) circle (.1) node[left]{\scriptsize 0};
\filldraw[fill=black] (-.4,-.3) circle (.1) node[left]{\scriptsize 1};
\filldraw[fill=white] (0,0) circle (.1) node[above]{\scriptsize 2};
\filldraw[fill=black] (.6,-.3) circle (.1) node[right]{\scriptsize 4};
\filldraw[fill=black] (.4,.3) circle (.1) node[right]{\scriptsize 3};
\end{tikzpicture}
\end{minipage} 
\\[2mm]
({\rm D}^{(1)}_4)^{\id}_{0,0} &
({\rm D}^{(1)}_4)^{\id}_{0,4} &
({\rm D}^{(1)}_4)^{\id}_{2,2}
\end{array}
\]
\end{table}

\vspace{-1em}

\begin{longtable}{p{\textwidth}}
\caption{$\Aut(A)$-equivalent Satake diagrams of type ${\rm D}^{(1)}_4$ which are $\Sigma_A$-inequivalent.} \label{Tbl:isos:4b}\\
\nopagebreak 
\vspace{-1.5em}
\[
\begin{array}{c@{\;\not\cong_{\Sigma_A}\;}c|c@{\;\not\cong_{\Sigma_A}\;}c|c@{\;\not\cong_{\Sigma_A}\;}c}
\multicolumn{2}{c|}{} &
\begin{minipage}[c]{17mm}
\begin{tikzpicture}[scale=80/100]
\draw[thick] (-.6,.3) -- (0,0) -- (-.4,-.3);
\draw[thick] (.4,.3) -- (0,0) -- (.6,-.3);
\filldraw[fill=white] (-.6,.3) circle (.1) node[left]{\scriptsize 0};
\filldraw[fill=white] (-.4,-.3) circle (.1) node[left]{\scriptsize 1};
\filldraw[fill=white] (0,0) circle (.1) node[above]{\scriptsize 2};
\filldraw[fill=black] (.6,-.3) circle (.1) node[right]{\scriptsize 4};
\filldraw[fill=black] (.4,.3) circle (.1) node[right]{\scriptsize 3};
\end{tikzpicture}
\end{minipage}
& 
\begin{minipage}[c]{17mm}
\begin{tikzpicture}[scale=80/100]
\draw[thick] (-.6,.3) -- (0,0) -- (-.4,-.3);
\draw[thick] (.4,.3) -- (0,0) -- (.6,-.3);
\filldraw[fill=white] (-.6,.3) circle (.1) node[left]{\scriptsize 0};
\filldraw[fill=black] (-.4,-.3) circle (.1) node[left]{\scriptsize 1};
\filldraw[fill=white] (0,0) circle (.1) node[above]{\scriptsize 2};
\filldraw[fill=white] (.6,-.3) circle (.1) node[right]{\scriptsize 4};
\filldraw[fill=black] (.4,.3) circle (.1) node[right]{\scriptsize 3};
\end{tikzpicture}
\end{minipage}
\\[1mm]
\multicolumn{2}{c|}{} & ({\rm D}^{(1)}_4)^{\id}_{0,2} & ({\rm D}^{(1)}_4)^{\id}_{\rm alt} \\[1mm] \hline
\begin{minipage}[c]{16mm}
\begin{tikzpicture}[scale=80/100]
\draw[thick] (-.5,.3) -- (0,0) -- (-.5,-.3);
\draw[thick] (.5,.3) -- (0,0) -- (.5,-.3);
\draw[<->,gray] (-.5,.2) -- (-.5,-.2);
\draw[<->,gray] (.5,.2) -- (.5,-.2);
\filldraw[fill=white] (-.5,.3) circle (.1) node[left]{\scriptsize 0};
\filldraw[fill=white] (-.5,-.3) circle (.1) node[left]{\scriptsize 1};
\filldraw[fill=white] (0,0) circle (.1) node[above]{\scriptsize 2};
\filldraw[fill=white] (.5,.3) circle (.1) node[right]{\scriptsize 3};
\filldraw[fill=white] (.5,-.3) circle (.1) node[right]{\scriptsize 4};
\end{tikzpicture} 
\end{minipage}
& 
\begin{minipage}[c]{14mm}
\begin{tikzpicture}[scale=80/100]
\draw[thick] (.7,.6) -- (0,0) -- (.4,.2);
\draw[thick] (.7,-.1) -- (0,0) -- (.4,-.5);
\draw[<->,gray] (.4,.1) -- (.4,-.4);
\draw[<->,gray] (.7,.5) -- (.7,0);
\filldraw[fill=white] (0,0) circle (.1) node[left]{\scriptsize 2};
\filldraw[fill=white] (.4,.2) circle (.1) node[right=-2pt]{\scriptsize 1};
\filldraw[fill=white] (.4,-.5) circle (.1) node[right=-1.5pt]{\scriptsize 3};
\filldraw[fill=white] (.7,.6) circle (.1) node[right]{\scriptsize 0};
\filldraw[fill=white] (.7,-.1) circle (.1) node[right]{\scriptsize 4};
\end{tikzpicture}
\end{minipage}
& 
\begin{minipage}[c]{16mm}
\begin{tikzpicture}[scale=80/100]
\draw[thick] (-.5,.3) -- (0,0) -- (-.5,-.3);
\draw[thick] (.5,.3) -- (0,0) -- (.5,-.3);
\draw[<->,gray] (-.5,.2) -- (-.5,-.2);
\draw[<->,gray] (.5,.2) -- (.5,-.2);
\filldraw[fill=white] (-.5,.3) circle (.1) node[left]{\scriptsize 0};
\filldraw[fill=white] (-.5,-.3) circle (.1) node[left]{\scriptsize 1};
\filldraw[fill=black] (0,0) circle (.1) node[above]{\scriptsize 2};
\filldraw[fill=white] (.5,.3) circle (.1) node[right]{\scriptsize 3};
\filldraw[fill=white] (.5,-.3) circle (.1) node[right]{\scriptsize 4};
\end{tikzpicture} 
\end{minipage}
& 
\begin{minipage}[c]{14mm}
\begin{tikzpicture}[scale=80/100]
\draw[thick] (.7,.6) -- (0,0) -- (.4,.2);
\draw[thick] (.7,-.1) -- (0,0) -- (.4,-.5);
\draw[<->,gray] (.4,.1) -- (.4,-.4);
\draw[<->,gray] (.7,.5) -- (.7,0);
\filldraw[fill=black] (0,0) circle (.1) node[left]{\scriptsize 2};
\filldraw[fill=white] (.4,.2) circle (.1) node[right=-2pt]{\scriptsize 1};
\filldraw[fill=white] (.4,-.5) circle (.1) node[right=-1.5pt]{\scriptsize 3};
\filldraw[fill=white] (.7,.6) circle (.1) node[right]{\scriptsize 0};
\filldraw[fill=white] (.7,-.1) circle (.1) node[right]{\scriptsize 4};
\end{tikzpicture}
\end{minipage}
&
\begin{minipage}[c]{16mm}
\begin{tikzpicture}[scale=80/100]
\draw[thick] (-.5,.3) -- (0,0) -- (-.5,-.3);
\draw[thick] (.5,.3) -- (0,0) -- (.5,-.3);
\draw[<->,gray] (-.5,.2) -- (-.5,-.2);
\draw[<->,gray] (.5,.2) -- (.5,-.2);
\filldraw[fill=white] (-.5,.3) circle (.1) node[left]{\scriptsize 0};
\filldraw[fill=white] (-.5,-.3) circle (.1) node[left]{\scriptsize 1};
\filldraw[fill=black] (0,0) circle (.1) node[above]{\scriptsize 2};
\filldraw[fill=black] (.5,.3) circle (.1) node[right]{\scriptsize 3};
\filldraw[fill=black] (.5,-.3) circle (.1) node[right]{\scriptsize 4};
\end{tikzpicture} 
\end{minipage}
& 
\begin{minipage}[c]{14mm}
\begin{tikzpicture}[scale=80/100]
\draw[thick] (.7,.6) -- (0,0) -- (.4,.2);
\draw[thick] (.7,-.1) -- (0,0) -- (.4,-.5);
\draw[<->,gray] (.4,.1) -- (.4,-.4);
\draw[<->,gray] (.7,.5) -- (.7,0);
\filldraw[fill=black] (0,0) circle (.1) node[left]{\scriptsize 2};
\filldraw[fill=black] (.4,.2) circle (.1) node[right=-2pt]{\scriptsize 1};
\filldraw[fill=black] (.4,-.5) circle (.1) node[right=-1.5pt]{\scriptsize 3};
\filldraw[fill=white] (.7,.6) circle (.1) node[right]{\scriptsize 0};
\filldraw[fill=white] (.7,-.1) circle (.1) node[right]{\scriptsize 4};
\end{tikzpicture}
\end{minipage}
\\[1mm]
({\rm D}^{(1)}_4)^{\flLR}_{0,0} & ({\rm D}^{(1)}_4)^{\pi}_0 & ({\rm D}^{(1)}_4)^{\flLR}_{\rm alt} & ({\rm D}^{(1)}_4)^{\pi}_1 &  ({\rm D}^{(1)}_4)^{\flLR}_{0,2} & ({\rm D}^{(1)}_4)^{\pi}_2 \\[1mm]
\hline
\begin{minipage}[c]{16mm}
\begin{tikzpicture}[scale=80/100]
\draw[thick] (-.6,.3) -- (0,0) -- (-.4,-.3);
\draw[thick] (.5,.3) -- (0,0) -- (.5,-.3);
\draw[<->,gray] (.5,.2) -- (.5,-.2);
\filldraw[fill=white] (-.6,.3) circle (.1) node[left]{\scriptsize 0};
\filldraw[fill=white] (-.4,-.3) circle (.1) node[left]{\scriptsize 1};
\filldraw[fill=white] (0,0) circle (.1) node[above]{\scriptsize 2};
\filldraw[fill=white] (.5,.3) circle (.1) node[right]{\scriptsize 3};
\filldraw[fill=white] (.5,-.3) circle (.1) node[right]{\scriptsize 4};
\end{tikzpicture} 
\end{minipage}
& 
\begin{minipage}[c]{17mm}
\begin{tikzpicture}[scale=80/100]
\draw[thick] (-.6,.3) -- (0,0) -- (-.4,-.3);
\draw[thick] (.4,.3) -- (0,0) -- (.6,-.3);
\draw[<->,gray] (-.3,-.3) -- (.5,-.3);
\filldraw[fill=white] (-.6,.3) circle (.1) node[left]{\scriptsize 0};
\filldraw[fill=white] (-.4,-.3) circle (.1) node[left]{\scriptsize 1};
\filldraw[fill=white] (0,0) circle (.1) node[above]{\scriptsize 2};
\filldraw[fill=white] (.6,-.3) circle (.1) node[right]{\scriptsize 4};
\filldraw[fill=white] (.4,.3) circle (.1) node[right]{\scriptsize 3};
\end{tikzpicture}
\end{minipage}
&
\begin{minipage}[c]{16mm}
\begin{tikzpicture}[scale=80/100]
\draw[thick] (-.6,.3) -- (0,0) -- (-.4,-.3);
\draw[thick] (.5,.3) -- (0,0) -- (.5,-.3);
\draw[<->,gray] (.5,.2) -- (.5,-.2);
\filldraw[fill=white] (-.6,.3) circle (.1) node[left]{\scriptsize 0};
\filldraw[fill=white] (-.4,-.3) circle (.1) node[left]{\scriptsize 1};
\filldraw[fill=black] (0,0) circle (.1) node[above]{\scriptsize 2};
\filldraw[fill=black] (.5,.3) circle (.1) node[right]{\scriptsize 3};
\filldraw[fill=black] (.5,-.3) circle (.1) node[right]{\scriptsize 4};
\end{tikzpicture} 
\end{minipage}
&
\begin{minipage}[c]{17mm}
\begin{tikzpicture}[scale=80/100]
\draw[thick] (-.6,.3) -- (0,0) -- (-.4,-.3);
\draw[thick] (.4,.3) -- (0,0) -- (.6,-.3);
\draw[<->,gray] (-.3,-.3) -- (.5,-.3);
\filldraw[fill=white] (-.6,.3) circle (.1) node[left]{\scriptsize 0};
\filldraw[fill=black] (-.4,-.3) circle (.1) node[left]{\scriptsize 1};
\filldraw[fill=black] (0,0) circle (.1) node[above]{\scriptsize 2};
\filldraw[fill=black] (.6,-.3) circle (.1) node[right]{\scriptsize 4};
\filldraw[fill=white] (.4,.3) circle (.1) node[right]{\scriptsize 3};
\end{tikzpicture}
\end{minipage}
&
\begin{minipage}[c]{16mm}
\begin{tikzpicture}[scale=80/100]
\draw[thick] (-.6,.3) -- (0,0) -- (-.4,-.3);
\draw[thick] (.5,.3) -- (0,0) -- (.5,-.3);
\draw[<->,gray] (.5,.2) -- (.5,-.2);
\filldraw[fill=black] (-.6,.3) circle (.1) node[left]{\scriptsize 0};
\filldraw[fill=black] (-.4,-.3) circle (.1) node[left]{\scriptsize 1};
\filldraw[fill=white] (0,0) circle (.1) node[above]{\scriptsize 2};
\filldraw[fill=white] (.5,.3) circle (.1) node[right]{\scriptsize 3};
\filldraw[fill=white] (.5,-.3) circle (.1) node[right]{\scriptsize 4};
\end{tikzpicture} 
\end{minipage}
&
\begin{minipage}[c]{17mm}
\begin{tikzpicture}[scale=80/100]
\draw[thick] (-.6,.3) -- (0,0) -- (-.4,-.3);
\draw[thick] (.4,.3) -- (0,0) -- (.6,-.3);
\draw[<->,gray] (-.3,-.3) -- (.5,-.3);
\filldraw[fill=black] (-.6,.3) circle (.1) node[left]{\scriptsize 0};
\filldraw[fill=white] (-.4,-.3) circle (.1) node[left]{\scriptsize 1};
\filldraw[fill=white] (0,0) circle (.1) node[above]{\scriptsize 2};
\filldraw[fill=white] (.6,-.3) circle (.1) node[right]{\scriptsize 4};
\filldraw[fill=black] (.4,.3) circle (.1) node[right]{\scriptsize 3};
\end{tikzpicture}
\end{minipage}
\\[1mm]
({\rm D}^{(1)}_4)^{\flR}_{0;0} & ({\rm D}^{(1)}_4)^{(14)}_{\emptyset} & ({\rm D}^{(1)}_4)^{\flR}_{0;2} & ({\rm D}^{(1)}_4)^{(14)}_{\{1,2,4\}} & ({\rm D}^{(1)}_4)^{\flR}_{2;0} & ({\rm D}^{(1)}_4)^{(14)}_{\{0,3\}}
\end{array}
\]
\end{longtable}



\section{Summary of K-matrix properties} \label{App:summary}

In Table \ref{table:summary} we summarize some key properties of the K-matrices obtained in Section \ref{sec:Results}.
First we indicate the type of generalized Satake diagram. 
Then we give constraints on the classification parameters $n,\ell,r$ (or $N,t,\ell,r$ in case of A.3). 
For each K-matrix of the given type with the constraints as indicated we indicate (for generic values of any free parameters) the locations and the number of its nonzero entries, the parameters appearing in the bare K-matrix (including any powers of $q$) and $d_{\rm min}$, the degree of the minimal polynomial of $K(u)$ in the untwisted case and of $C^{-1} K(\wt q^{\, -1} u)$ in the twisted case. 

We list some further properties where appropriate, namely whether the K-matrix is twisted (tw), restrictable (rstr), non-quasistandard (nqs), half-period symmetric (hps), singly regular (sreg) and also the value of its effective degree $d_{\rm eff} = d_{\rm eff}(K)$ if this is unequal to $d_{\rm min}$. 
As an aid to the reader we also highlight if the corresponding Satake diagram is weak (W).
If the generalized Satake diagram is such that $I \backslash X = \{0\}$ then we will indicate this; in this case $K(u) = \Id$ and automatically we have the (trivial) properties that $d_{\rm eff}\!=\!0$, $K(u)$ is restrictable and half-period-symmetric; we will not specify these properties.
Similarly, if $K(u)$ is otherwise constant then we will indicate this (cst); in this case $K(u)$ is automatically restrictable,  half-period-symmetric and nonregular and satisfies $d_{\rm eff}\!=\!0$; again, we will not specify these properties.


\setlength\LTleft{0pt}
\setlength\LTright{0pt}

\setlength{\tabcolsep}{3.5pt}
\begin{longtable}{l l p{20pt} l l l l}
\caption{A summary of properties of K-matrices.} \label{table:summary}\\
\hline
\multirow{2}{*}{Type} & \multirow{2}{*}{Restrictions} & \multicolumn{2}{l}{Nonzero entries} & Parameters & \multirow{2}{*}{$d_{\rm min}$} & \multirow{2}{*}{Comments} \\
&& Loc. & Number & in bare K & & \\
\nobreakhline
\endfirsthead
\hline
\multirow{2}{*}{Type} & \multirow{2}{*}{Restrictions} & \multicolumn{2}{l}{Nonzero entries} & Parameters & \multirow{2}{*}{$d_{\rm min}$} & \multirow{2}{*}{Comments} \\
&& Loc. & Number & in bare K & & \\
\nobreakhline
\endhead
\hline
A.1 & $N > 2$ & \shape{\evendiag} & $N$ & - & 2 & tw, cst\\
\hline
A.2 & $N >2$ even & \shape{\evenAalt} & $N$ & $q^{\frac12}$ & 2 & tw, cst \\
\hline
A.4 & $N >2$ even & \shape{\evenArot} & $N$ & - & 4 & tw, $d_{\rm eff} = 1$ \\
\hline
\hline
& $0=\ell=r$ & \shape{\evendiag} & $N$ & - & $1$ & $I \backslash X = \{0\}$ \\*
\multirow{2}{*}{A.3a} 
& $0=\ell<r \le\frac{N}2$ & \shape{\evencrossin} & $N\!+\!2r$ & $\la,\mu$ & $2$ & rstr, hps \\*
& $0<\ell<r < \frac{N}2\!-\!\ell$ & \shape{\evencrossinout} & $N\!+\!2(r\!-\!\ell)$ & $\la,\mu$ & $4$ & hps \\*
& $0<\ell=r\le \frac{N}4$ & \shape{\evendiag} & $N$ & $\xi=\la \mu$ & $3$ & hps \\
\hline
& $0=\ell=r$ & \shape{\odddiag} & $N$ & - & $1$ & $I \backslash X = \{0\}$ \\*
\multirow{2}{*}{A.3b} 
& $0=\ell<r\le \frac{N+1}2$ & \shape{\oddcrossin} & $N\!+\!2r$ & $\la,\mu$ & $2$ & rstr, hps \\*
& $0<\ell<r\le \frac{N+1}2$ & \shape{\oddcrossinout} & $N\!+\!2(r\!-\!\ell)$ & $\la,\mu$ & $4$ & hps \\*
& $0<\ell=r\le \frac{N+1}2$ & \shape{\odddiag} & $N$ & $\xi=\la \mu$ & $3$ & hps \\
\hline
& $0=\ell=r$ & \shape{\evendiag} & $N$ & $\xi=\la \mu$ & $2$ & hps \\ 
A.3c & $0\le\ell<r< \frac{N}2\!-\!\ell$ & \shape{\evencrossinshift} & $N\!+\!2(r\!-\!\ell)$ & $\la,\mu$ & $4$ & hps \\*
& $0<\ell=r< \frac{N}4$ & \shape{\evendiag} & $N$ & $\xi=\la \mu$ & $3$ & hps \\*
\hline
\hline
& $0=\ell=r$ & \shape{\evendiag} & $N$ & - & 1 & $I\backslash X=\{0\}$ \\
& $0=\ell<r<n$ & \shape{\evencrossin} & $N\!+\!2r$ & $q^{\bar r},\mu$ & 2 & rstr, hps, W \\
& $0=\ell<r=n$ & \shape{\evencross} & $2N$ & $\la,\mu$ & 2 & rstr, hps \\
\multirow{2}{*}{C.1 \vspace{-7pt} }  & $0<\ell<r<n\!-\!\ell,\;r\ne \ell\!+\!2$ & \shape{\evencrossinout} & $N\!+\!2(r\!-\!\ell)$ & $q^{\bar r},q^{-\ell-1}$ & 4 & W \\
& $1\le\ell\le n\!-\!3,\; r=\ell\!+\!2$ & \shape{\evendiagmidblocks} & $N\!+\!12$ & $q^{\bar r},q^{-\ell-1},\nu$ & 5 & nqs, W \\
& $0<\ell<r=n\!-\!\ell$ & \shape{\evencrossinout} & $2(N\!-\!2\ell)$ & $q^{-\ell-1}$ & 4 & sreg, $d_{\rm eff}=3$, W \\
& $0<\ell=r<\tfrac n2$ & \shape{\evendiag} & $N$ & $q^{n-2\ell}$ & 3 & \\
& $0<\ell=r=\tfrac n2$ & \shape{\evendiag} & $N$ & - & 3 & sreg, $d_{\rm eff}=2$ \\
\hline
& $0=\ell=r$ & \shape{\odddiag} & $N$ & - & 1 & $I\backslash X=\{0\}$ \\*
& $0=\ell<r$ & \shape{\odddiagaltin} & $N\!+\!2r$ & $q^{\bar r-\frac32},\mu$ & 2 & rstr, hps, W \\*
\multirow{2}{*}{B.2 \vspace{-7pt} } & $1=\ell=r$ & \shape{\odddiag} & $N$ & $q^{n-\frac32},\mu$ & 3 & hps, $d_{\rm eff}=4$, W \\*
& $1=\ell<r$ & \shape{\odddiagaltinout} & $N\!+\!2(r\!-\!1)$ & $q^{\bar r-\frac32},\mu$ & 4 & hps, W \\*
& $1<\ell<r$ & \shape{\odddiagaltinout} & $N\!+\!2(r\!-\!\ell)$ & $q^{\bar r-\frac32},q^{-\ell+1}$ & 4 & W \\*
& $1<\ell=r$ & \shape{\odddiag} & $N$ & $q^{\bar \ell- \ell -\frac12}$ & 3 & - \\*
\hline
& $0=\ell=r$ & \shape{\evendiag} & $N$ & - & 1 & $I \backslash X=\{0\}$ \\*
& $0=\ell<r<n\!-\!1$ & \shape{\evendiagaltin}  & $N\!+\!2r$ & $q^{n-r-1},\mu$ & 2 & rstr, hps, W \\*
& $0=\ell<r=n\!-\!1,\;n$ odd & \shape{\evendiagaltin} & $N\!+\!2r$ & $\la,\mu$ & 2 & rstr, hps \\*
& $0=\ell<r=n,\;n$ even & \shape{\evendiagalt} & $N\!+\!2r$ & $\la,\mu$ & 2 & rstr, hps  \\*
& $1=\ell=r$ & \shape{\evendiag} & $N$ & $q^{n-2},\mu$ & 3 & hps  \\*
D.2 & $1=\ell<r<n\!-\!1$ & \shape{\evendiagaltinout} & $N\!+\!2(r\!-\!1)$ & $q^{n-r-1},\mu$ & 4 & hps, W  \\*
& $1=\ell<r=n\!-\!1$ & \shape{\evendiagaltinout} & $2N\!-\!4$ & $\la,\mu$ & 4 & hps  \\*
& $1<\ell<r<n\!-\!\ell\!\!$ & \shape{\evendiagaltinout} & $N\!+\!2(r\!-\!\ell)$ & $q^{n-r-1}, q^{-\ell+1}$ & 4 & W \\*
& $0<\ell<r=n\!-\!\ell\!\!$ & \shape{\evendiagaltinout} & $2(N\!-\!2\ell)$ & $q^{-\ell+1}$ & 4 & sreg, $d_{\rm eff}\!=\!3$, W \\*
& $1<\ell=r<\tfrac{n}{2}$ & \shape{\evendiag} & $N$ & $q^{n-2\ell}$ & 3 & - \\*
& $\ell=r=\tfrac{n}{2}$ & \shape{\evendiag}& $N$ & - & 3 & sreg, $d_{\rm eff}\!=\!2$ \\
\hline
\hline
& $0=\ell$ & \shape{\evendiag} & $N$ & $\xi=q^{-n}\la^2$ & 2 & hps \\*
C.4 & $0<\ell<\tfrac{n}{2}$ & \shape{\evendoublecrossin} & $N\!+\!4\ell$ & $\la, q^{n\!-\!2\ell}\la$ & 4 & \\*
& $\!\tfrac{n}{2}=\ell$ & \shape{\evendoublecross} & $2N$ & $\la$ & 4 & sreg, $d_{\rm eff}\!=\!2$ \\
\hline
& $0=\ell$ & \shape{\evendiag} & $N$ & $\xi=q^{-n}\la^2$ & 2 & hps \\*
\multirow{2}{*}{D.4 \vspace{-7pt} } 
& $1=\ell$ & \shape{\evendoublecrossinbig} & $N\!+\!4$ &  $\al, \la, q^{-n+2}\al \la$ & 4 & hps \\*
& $1 <\ell < \tfrac{n}{2}$ & \shape{\evendoublecrossin} & $N\!+\!4\ell$ & $\la,q^{-n\!+\!2\ell}\la$ & 4 & \\*
& $\!\tfrac{n}{2}=\ell$ & \shape{\evendoublecross} & $2N$ & $\la$ & 4 & sreg, $d_{\rm eff}\!=\!2$ \\
\hline
\hline
& $0=\ell=r$ & \shape{\evendiag} & $N$ & - & 1 & $I \backslash X=\{0\}$ \\*
& $0=\ell<r<n$ & \shape{\evendiagaltin}  & $N\!+\!2r$ & $q^{n-r}$ & 2 & rstr, hps \\*
& $0=\ell<r=n$ & \shape{\evenalt} & $N$ & - & 2 & cst  \\*
C.2 & $0<\ell<r<n\!-\!\ell\!\!$ & \shape{\evendiagaltinout} & $N\!+\!2(r\!-\!\ell)$ & $q^{n-r}, q^{-\ell}$ & 4 &  - \\*
& $0<\ell<r=n\!-\!\ell\!\!$ & \shape{\evendiagaltinout} & $2(N\!-\!2\ell)$ & $q^{-\ell}$ & 4 & sreg, $d_{\rm eff}\!=\!3$ \\*
& $0<\ell=r<\tfrac{n}{2}$ & \shape{\evendiag} & $N$ & $q^{n-2\ell}$ & 3 &  \\*
& $0<\ell=r=\tfrac{n}{2}$ & \shape{\evendiag}& $N$ & - & 3 & sreg, $d_{\rm eff}\!=\!2$ \\
\hline 
& $0=\ell=r$ & \shape{\odddiag} & $N$ & - & 1 & $I \backslash X=\{0\}$  \\*
& $0=\ell<r \ne 2$ & \shape{\oddcrossin} & $N\!+\!2r$ & $q^{\bar r-\frac12}$ & 2 & rstr, hps \\*
& $0=\ell<r = 2$ & \shape{\odddiagoutblocks} & $N\!+\!12$ & $q^{n-\frac32},\nu_0,\nu_1$ & 3 & rstr, nqs, hps, $d_{\rm eff}\!=\!4$ \\*
\multirow{2}{*}{B.1} & $0<\ell<r$, $\ell\!\ne\!n\!-\!1$, $r\!\ne\!\ell\!+\!2$ & \shape{\oddcrossinout} & $N\!+\!2(r\!-\!\ell)$ &  $q^{\bar r-\frac12}, q^{-\ell}$ & 4  & \\*
& $0<\ell<r = \ell\!+\!2$ & \shape{\odddiagmidblocks} & $N\!+\!12$ &  $q^{\bar \ell-\frac52}, q^{-\ell},\nu$ & 5  & nqs \\*
& $\ell=n\!-\!1$, $r = n$ & \shape{\odddiaginblock} & $N\!+\!6$ &  $q^{\frac12}, q^{1-n},\nu$ & 5 & nqs \\*
& $1=\ell=r$ & \shape{\odddiag} & $N$ &  $q^{n-\frac32},\al$ & 3 & hps, $d_{\rm eff}\!=\!4$ \\*
& $1<\ell=r$ & \shape{\odddiag} & $N$ & $q^{\bar \ell - \ell - \frac12}$ & 3  & \\
\hline
& $0=\ell=r$ & \shape{\evendiag} & $N$ & - & 1 & $I \backslash X=\{0\}$ \\*
& $0=\ell,\; r=2$ & \shape{\evendiagoutblocks} & $N\!+\!12$ & $q^{n-2}, \nu_0, \nu_1$ & 3 & rstr, nqs, hps, $d_{\rm eff}\!=\!4$ \\*
& $0=\ell<r<n,\; r \ne 2$ & \shape{\evencrossin} & $N\!+\!2r$ & $q^{n-r}$ & 2 & rstr, hps \\*
& $0=\ell<r=n$ & \shape{\evenantidiag} & $N$ & - & 2 & cst \\*
\multirow{2}{*}{D.1} & $0\!<\!\ell\!<\!r\!<\!n,\; r \!\notin\! \{\ell\!+\!2, n\!-\!\ell\}\,$ & \shape{\evencrossinout} & $N\!+\!2(r\!-\!\ell)$ & $q^{n-r}, q^{-\ell}$ & 4 &  \\*
& $0\!<\!\ell\!<\!r=n\!-\!\ell,\; r \ne \ell\!+\!2$ & \shape{\evencrossinout} & $2(N\!-\!2\ell)$ & $q^{-\ell}$ & 4 & sreg, $d_{\rm eff}\!=\!2$ \\*
& $0\!<\!\ell\!<\!\frac{n}2\!-\!1,\;r=\ell\!+\!2$ & \shape{\evendiagmidblocks} & $N\!+\!12$ & $q^{\bar \ell-3}, q^{-\ell}, \nu$ & 5 & nqs \\*
& $\ell=\frac{n}2\!-\!1,\;r=\frac{n}2\!+\!1$ & \shape{\evendiagmidblocks} & $N\!+\!12$ & $q^{\frac{n-3}2}, q^{\frac{1-n}2}, \nu$ & 5 & nqs, sreg, $d_{\rm eff}\!=\!4$ \\*
& $1<\ell=r<n\!-\!1\!$ & \shape{\evendiag} & $N$ & $q^{n\!-\!2\ell}$ & 3 & -\\
\hline 
\end{longtable}



\section{A K-matrix of generalized $q$-Onsager type} \label{App:qOns}

Here we present an example of a generalized $q$-Onsager algebra, namely the one of type A defined in \cite{BsBe1}. 
Let us return to the family A.1 considered in Section \ref{sec:A1}. 
Recall that the Satake diagram is given by $(X,\tau) = (\emptyset, \id)$. 
Assume $n\ge2$. 
Choose tuples $\bm c , \bm s \in (\K^\times)^I$ such that the generalized $q$-Onsager type constraints $c_j = (q^{-1}-1)^2\, s_j^2$ for all $j\in I$ are satisfied ({\it cf.}~\cite[Prop.~2.1]{BsBe1}). 
Then the associated algebra $B_{\bm c, \bm s}$ is the generalized $q$-Onsager algebra generated by
\eq{
b_j = y_j - c_j\, x_j k^{-1}_j - s_j\, k^{-1}_j \tx{for all} j \in I. \label{A1:b_j:qOns}
}
Introduce dressing parameters $\om_1,\om_2,\ldots,\om_N$ satisfying $\om_1 \om_2 \cdots \om_N =1$ and, for all $j\in I$, set $s_j = (q^{-1}-1)^{-1}\,\om_{j}^{-1}\om_{j+1}$, where $\om_0 := \eta \, \om_N$. 
Next, we solve the twisted boundary intertwining equation \eqref{intw-tw} for all $b_j$ in \eqref{A1:b_j:qOns}. This gives a solution
\eq{
K(u;\bm \om) = G(\bm \om)\, K(u)\, G(\bm \om) ,\qu\text{where}\qu
K(u) = \sum_{1 \leq i,j \leq N} \frac{q^{-1/2} u^{\delta_{i>j}} + q^{1/2} u^{\delta_{i \geq j}}}{1-u}\, E_{ij} ,
\label{KA1:$q$-Ons}
}
satisfying the twisted reflection equation \eqref{tRE}. 
This K-matrix was first obtained in \cite{Gb}.
It has the following additional properties (recall that $\wt q = (-1)^{N/2} q^{N/2}$):

\smallskip

\begin{description} [itemsep=1ex]

\item[Twisted unitarity] \hfill $(C^{-1} K(\wt q^{\,-1} u))^{-1} = (-1)^N C^{-1} K( \wt q^{\,-1} u^{-1})$.  \hfill \hphantom{Twisted unitarity}

\item[Eigenvalues of $C^{-1} K(\wt q^{-1} u)$] 
$(-1)^{\frac{N}{2}+1}$ (with multiplicity $\frac{N\!-\!1}{2}$), 
$(-1)^{\frac{N}{2}}$ (with multiplicity $\lceil \frac{N\!-\!1}{2} \rceil$), 
$(-1)^{-\frac{N}{2}} \frac{1-\wt q \, u}{\wt q-u}$ (with multiplicity 1).
This particular K-matrix also has a basis of eigenvectors, but we leave the explicit eigendecomposition of $C^{-1} K(\wt q^{-1} u)$ out of consideration.

\item[Affinization] 
\hfill $\displaystyle K(u) = \frac{K_0 + u\, C K_0^{-1} C^{-1}}{1-u}, \qu \text{where } K_0 := \lim_{u\to 0} K(u)$. \hfill \hphantom{}

\item[Rotations] \hfill $K^{\rho}(u) = K(u)$. \hfill \hphantom{}
\end{description}

Additionally, $K(u)$ satisfies $K(u)^{\t} + K(u^{-1}) = (q^{-1/2} - q^{1/2})\, \Id$ and $K(u)^{\t} = J K(u) J$.
Note that although the degree of the minimal polynomial of $C^{-1} K(\wt q^{-1} u)$ is 3, the effective degree is only 1.

\smallskip

Contrary to the solutions to the twisted RE in Section \ref{sec:K:tw}, $K(u)$ maps each basis vector of $\End(\K^N)$ to a linear combination of all of them.
This approach can be used to obtain similar ``dense'' reflection matrices associated with Satake diagrams of types B.1, C.1 and D.1 when the parameters $c_j$ and $s_j$ satisfy generalized $q$-Onsager type constraints. 
Examples of these most general non-diagonal \mbox{K-matrices} are already known: for type D such a K-matrix was obtained in \cite{DeGe} and a wider class of such dense K-matrices was found in \cite[Sec.~3.2]{MLS}.






\begin{thebibliography}{AACDFR1}


\bibitem[AACDFR1]{AACDFR1}
	D. Arnaudon, J. Avan, N. Cramp\'{e}, A. Doikou, L. Frappat, E. Ragoucy,
	{\it Classification of reflection matrices related to (super-) Yangians and application to open spin chain models}.
	Nucl. Phys. B {\bf 668} (2003), no. 3, 469--505.
	{\tt arXiv:math/0304150}.

\bibitem[AACDFR2]{AACDFR2}
	D. Arnaudon, J. Avan, N. Cramp\'{e}, A. Doikou, L. Frappat, E. Ragoucy,
	{\it General boundary conditions for the sl($N$) and sl($M\mid N$) open spin chains}.
	JSTAT 08 (2004) P005.
	{\tt arXiv:math-ph/0406021}.
	
\bibitem[AbRi]{AbRi}
	J. Abad, M. Rios,
	{\it Non-diagonal solutions to reflection equatons in $su(n)$ spin chains}.
	Phys. Lett. B {\bf 352} (1995), no. 1, 92--95.
	{\tt arXiv:hep-th/9502129}.

\bibitem[Ara]{Ara}
	Sh. Araki,
	{\it On root systems and an infinitesimal classification of irreducible symmetric spaces}.
	J. Math. Osaka City Univ. {\bf 13} (1962), no. 1, 1--34.

\bibitem[Ari]{Ari}
	S. Ariki,
	{\it Lectures on cyclotomic Hecke algebras}.
	LMS Lecture Note Series No.~290 (Ed. A. Pressley), Cambridge University Press, 2002.
	{\tt arXiv:math/9908005}.

\bibitem[BaWa]{BaWa}
	H. Bao, W. Wang,
	{\it A new approach to Kazhdan-Lusztig theory of type B via quantum symmetric pairs}.
	Preprint, {\tt arXiv:1310.0103}.
	
\bibitem[Ba1]{Ba1}
	R. Baxter,
	{\it Eight-vertex model in lattice statistics and one-dimensional anisotropic Heisenberg chain. I. Some fundamental eigenvectors.}
	Ann.~of Physics {\bf 76} (1973), no. 1, 1--24.

\bibitem[Ba2]{Ba2}
	\bysame,
	{\it Exactly Solved Models in Statistical Mechanics}.
	Academic Press Inc., 1982.

\bibitem[BBBR]{BBBR}
	V. Back-Valente, N. Bardy-Panse, H. Ben Massaoud, G. Rousseau,
	{\it Formes presque-d\'{e}ploy\'{e}es des alg\`{e}bres de Kac-Moody: Classication et racines relatives}.
	J. Algebra {\bf 171} (1995), 43--96.

\bibitem[BbRg]{BbRg}
	A. Babichenko, V. Regelskis,
	{\it On boundary fusion and functional relations in the Baxterized affine Hecke algebra}.
	J. Math. Phys. {\bf 55} (2014), 043503.
	{\tt arXiv:1305.1941}.
	
\bibitem[BCDR]{BCDR}
	P. Bowcock, E. Corrigan, P. Dorey, R. Rietdijk,
	{\it Classically integrable boundary conditions for affine Toda field theories}.
	Nucl. Phys. B {\bf 445} (1995), 469--500.
	{\tt arXiv:hep-th/9501098}.

\bibitem[BCR]{BCR}
	P. Bowcock, E. Corrigan, R. Rietdijk,
	{\it Background field boundary conditions for affine Toda field theories}.
	Nucl. Phys. B {\bf 465} (1996), 350--364.
	{\tt arXiv:hep-th/9510071}.

\bibitem[BeFo]{BeFo}
	S. Belliard, V. Fomin,
	{\it Generalized q-Onsager Algebras and Dynamical K-matrices}.
	J. Phys. A {\bf 45} (2012), 025201.
	{\tt arXiv:1106.1317}.
	
\bibitem[BGKNR]{BGKNR}
	H. Boos, A. G\"{o}hmann, A. Kl\"{u}mper, K.S. Nirov, A.V. Razumov,
	{\it Exercises with the universal R-matrix}.
	J. Phys. A {\bf 43} (2010), no. 41, 415208, 35pp.
	{\tt arxiv:1004.5342}.

\bibitem[BgKo1]{BgKo1}
	M. Balagovi\'{c}, S. Kolb,
	{\it The bar involution for quantum symmetric pairs}.
	Rep. Thy. of the Amer. Math. Soc. {\bf 19} (2015), no. 8, 186--210.
	{\tt arXiv:1409.5074}.
	
\bibitem[BgKo2]{BgKo2}
	\bysame,
	{\it Universal K-matrix for quantum symmetric pairs}.
	Journal f\"{u}r die reine und angewandte Mathematik (Crelles Journal).
	{\tt arXiv:1507.06276}.

\bibitem[BsBe1]{BsBe1}
	P. Baseilhac, S. Belliard,
	{\it Generalized q-Onsager algebras and boundary affine Toda field theories}.
	Lett. Math. Phys. {\bf 93} (2010), 213--228.
	{\tt arXiv:0906.1215}.
	
\bibitem[BsBe2]{BsBe2}
	\bysame,
	{\it The half-infinite XXZ chain in Onsager's approach}.
	Nucl. Phys. {\bf B873} (2013), 550--583.
	{\tt arXiv:1211.6304}.

\bibitem[BsKz1]{BsKz1}
	P. Baseilhac, K. Koizumi,
	{\it Sine-Gordon quantum field theory on the half-line with quantum boundary degrees of freedom}.
	Nucl. Phys. B {\bf 649} (2003), no. 3, 491--510.
	{\tt arXiv:hep-th/0703106}.

\bibitem[BsKz2]{BsKz2}
	\bysame,
	{\it Exact spectrum of the XXZ open spin chain from the q-Onsager algebra representation theory}.
	J. Stat. Mech.: Theory and Exp. {\bf 09} (2007), P09006.
	{\tt arXiv:hep-th/0703106}.

\bibitem[Bzh]{Bzh}
	V. V. Bazhanov,
	{\it Trigonometric solutions of triangle equations and classical Lie algebras}.
	Phys. Lett. B {\bf 159} (1985), 321--324.

\bibitem[Ca]{Ca}
	R. Carter,
	{\it Lie algebras of finite and affine type}.
	Cambridge studies in advanced mathematics, Cambridge University Press, 2005. 

\bibitem[CDRS]{CDRS}
	E. Corrigan, P. E. Dorey, R. H. Rietdijk and R. Sasaki,
	{\it Affine Toda field theory on a half line}.
	Phys. Lett. B {\bf 333} (1994), no. 1, 83--91.
	{\tt arXiv:hep-th/9404108}.

\bibitem[CGM]{CGM}
	H. Chen, N. Guay, X. Ma,
	{\it Twisted Yangians, twisted quantum loop algebras and affine Hecke algebras of type BC}.
	Trans. Amer. Math. Soc. {\bf 366} (2014), 2517--2574. 
	
\bibitem[Ch1]{Ch1}
	I. Cherednik,
	{\it Factorizing particles on a half-line and root systems}.
	Theor. and Math. Phys. {\bf 61} (1984), no. 1, 977--983.

\bibitem[Ch2]{Ch2}
	\bysame,
	{\it A unification of Knizhnik-Zamolodchikov and Dunkl operators via affine Hecke algebras}.
	Inv. Math. {\bf 106} (1991), 411--431.
	
\bibitem[Ch3]{Ch3}
	\bysame,
	{\it Quantum Knizhnik-Zamolodchikov equations and affine root systems}.
	Comm. Math. Phys. {\bf 150} (1992), 109--136.

\bibitem[ChPr1]{ChPr1} 
	V. Chari and A. Pressley, 
	{\it A guide to quantum groups}.
	Cambridge University Press, 1994.

\bibitem[ChPr2]{ChPr2} 
	\bysame,
	{\it Quantum affine algebras and affine Hecke algebras}.
	Pacific J. Math. {\bf 174} (1996), no. 2, 295--326.

\bibitem[CRAS]{CRAS}
	G. Cisar, M. Ruiz-Altaba, G. Sierra,
	{\it Quantum groups in two-dimensional physics}.
	Cambridge University Press, 2005.

\bibitem[DeGe]{DeGe}
	G. Delius, A. George,
	{\it Quantum affine reflection algebras of type $d_n^{(1)}$ and reflection matrices}.
	 Lett. Math. Phys. {\bf 62} (2002), 211--217.
	{\tt arXiv:math/0208043}. 

\bibitem[DeMk]{DeMk}
	G. Delius, N. Mackay,
	{\it Quantum group symmetry in sine-Gordon and affine Toda field theories on the half-line}.
	Comm. Math. Phys. {\bf 233} (2003), 173--190.
	{\tt arXiv:hep-th/0112023}.
	
\bibitem[DFZJ1]{DFZJ1}
	P. Di Francesco, P. Zinn-Justin,
	{\it The quantum Knizhnik-Zamolodchikov equation, generalized Razumov-Stroganov sum rules and extended Joseph polynomials}.
	J. Phys. A: Math. Gen. {\bf 38} (2005), L815.
	{\tt arXiv:math-ph/0508059}.

\bibitem[DFZJ2]{DFZJ2}
	\bysame,
	{\it Quantum Knizhnik-Zamolodchikov equation: reflecting boundary conditions and combinatorics}.
	J. Stat. Mech.: Theory and Exp. {\bf 12} (2007), P12009.
	{\tt arXiv:0709.3410}.
	
\bibitem[DKM]{DKM}
	J. Donin, P. Kulish, A. Mudrov,
	{\it On a universal solution to the reflection equation}.
	Lett. Math. Phys. {\bf 63} (2003), 179--194.
	{\tt arXiv:math/0210242}.

\bibitem[dLMR]{dLMR}
	M. de Leeuw, T. Matsumoto, V. Regelskis,
	{\it Coideal Quantum Affine Algebra and Boundary Scattering of the Deformed Hubbard Chain}.
	J. Phys. A {\bf 45} (2012), 065205.
	{\tt arXiv:1110.4596}.

\bibitem[dLR]{dLR}
	M. de Leeuw, V. Regelskis
	{\it Integrable boundaries in AdS/CFT: revisiting the Z=0 giant graviton and D7-brane}.
	JHEP {\bf 03} (2013) 030.
	{\tt arXiv:1206.4704}.
	
\bibitem[DMS]{DMS}
	G. W. Delius, N. J. MacKay, B. J. Short,
	{\it Boundary remnant of Yangian symmetry and the structure of rational reflection matrices}.
	Phys. Lett, B {\bf 522} (2001), no. 3, 335--344.
	{\tt arXiv:hep-th/0109115}.
	
\bibitem[Do1]{Do1}
	A. Doikou,
	{\it Quantum spin chain with `soliton nonpreserving' boundary conditions}.
	J. Phys. A: Math. Gen. {\bf 33} (2000), 8797--8807.
	{\tt arXiv:hep-th/0006197}.

\bibitem[Do2]{Do2}
	\bysame,
	{\it From affine Hecke algebras to boundary symmetries}.
	Nucl.Phys. B{\bf 725} (2005), 493--530.
	{\tt arXiv:math-ph/0409060}.

\bibitem[Dr1]{Dr1} 
	V. Drinfeld, 
	{\it Hopf algebras and the quantum Yang-Baxter equation}.
	 Soviet Math. Dokl. {\bf 32} (1985), 254--258.
	 
\bibitem[Dr2]{Dr2}
	\bysame,
	{\it Quantum groups}.
	Proceedings of the International Congress of Mathematicians, Berkeley, 1986, A. M. Gleason (ed), 798-820, Amer. Math. Soc., Providence, RI.

\bibitem[EFK]{EFK}
	P. Etingof, I. Frenkel, A. Kirillov,
	{\it Lectures on Representation Theory and Knizhnik-Zamolodchikov Equations}.
	Amer. Math. Soc., 1998.
	
\bibitem[EGHLSVY]{EGHLSVY}
	P. Etingof, O. Golberg, S. Hensel, T. Liu, A. Schwendner, D. Vaintrob, E. Yudovina,
	{\it Introduction to Representation Theory}, vol. 59 of Student Mathematical Library. AMS, 2011.	
	
\bibitem[Fa]{Fa}
	L. Faddeev,
	Les Houches XXXIX, edited by J. Zuber and R. Stora (1984).
	{\tt arXiv:hep-th/9605187}.
	
\bibitem[FgK\"{o}]{FgKo}
	A. Fring and R. K\"{o}berle,
	{\it Boundary bound states in affine Toda field theory}.
	Int. J. Mod. Phys. A, {\bf 10} (1995), no. 5, 739--751.
	{\tt arXiv:hep-th/9404188}.

\bibitem[FK]{FK}
	G. Filali and N. Kitanine,
	{\it Spin chains with non-diagonal boundaries and trigonometric SOS model with reflecting end}.
	SIGMA, {\bf 7} (2011), Paper 012, 22 pages.
	{\tt arXiv:1011.0660}.

\bibitem[FNR]{FNR}
	L. Frappat, R. I. Nepomechie, E. Ragoucy,
	{\it Complete Bethe Ansatz solution of the open spin-s XXZ chain with general integrable boundary terms}.
	J. Stat. Mech. (2007), P09009.
	{\tt arxiv:0707.0653}.

\bibitem[FrHz]{FrHz}
	E. Frenkel, D. Hernandez,
	{\it Baxter's Relations and Spectra of Quantum Integrable Models}.
	Duke Math. J. {\bf 164} (2015), no. 12, 2407--2460.
	{\tt arXiv:1308.3444}.

\bibitem[FrMn]{FrMn}
	E. Frenkel, E. Mukhin,
	{\it The Hopf algebra Rep $U_q(\wh\mfgl_\infty)$}.
	Sel. Math., New Ser {\bf 8} (2002), 537--635.

\bibitem[FrRt]{FrRt}
	I. Frenkel, N. Reshetikhin,
	{\it Quantum affine algebras and holonomic difference equations}.
	Commun. Math. Phys. {\bf 146} (1992), 1--60.

\bibitem[FRT]{FRT} 
	L. Faddeev, N. Reshetikhin and L. Takhtajan,
	{\it Quantization of Lie groups and Lie algebras}.
	Leningrad Math. J. {\bf 1} (1990), 193.

\bibitem[Gb]{Gb}
	G. Gandenberger,
	{\it New non-diagonal solutions to the $a^{(1)}_n$ boundary Yang-Baxter equation}.
	Preprint, {\tt arXiv:hep-th/9911178}.
	
\bibitem[GhZa]{GhZa}
	S. Ghoshal, A. Zamolodchikov,
	{\it Boundary S-matrices and boundary state in two-dimensional integrable quantum field theory}.
	Int. J. Mod. Phys. A \textbf{09} (1994), 3841--3886.
	{\tt arXiv:hep-th/9306002}.
		
\bibitem[GmTL]{GmTL}
	S. Gautam, V. Toledano-Laredo,
	{\it Meromorphic Kazhdan-Lusztig equivalence for Yangians and quantum loop algebras}.
	Preprint, {\tt arXiv:1403.5251}.
	
\bibitem[GoMo]{GoMo}
	L. Gow, A. Molev,
	{\it Representations of twisted q-Yangians}.
	Selecta Math. {\bf 16} (2010), 439--499.
	{\tt arXiv:0909.4905}.

\bibitem[GuRg1]{GuRg1}
	N. Guay, V. Regelskis,
	{\it Twisted Yangians for symmetric pairs of types B, C, D}.
	To appear in Math. Z., {\tt arXiv:1407.5247}.
	
\bibitem[GRW]{GRW}
	N. Guay, V. Regelskis, C. Wendlandt,
	{\it RTT presentation of quantum loop algebras of types B, C, D and their representations}.
	Unpublished.

\bibitem[He]{He}
	S. Helgason,
	{\it Differentian geometry, Lie groups, and symmetric spaces}.
	AMS (2012).
	
\bibitem[Is]{Is}
	A. Isaev,
	{\it Functional equations for transfer-matrix operators in open Hecke chain models}.
	Theor. Math. Phys. {\bf 150} (2007), no. 2, 187.
	{\tt arXiv:1003.3385}.

\bibitem[IsOg]{IsOg}
	A. Isaev, O. Ogievetsky,
	{\it Baxterized Solutions of Reflection Equation and Integrable Chain Models}.
	Nucl.Phys. {\bf B}760 (2007), 167--183.
	{\tt arXiv:math-ph/0510078}.
	
\bibitem[Ji1]{Ji1}
	M. Jimbo, 
	{\it A q-analogue of $U(\mfgl(N+1))$, Hecke algebra, and the Yang-Baxter equation}.
	Lett. Math. Phys. {\bf 11} (1986), no. 3, 247--252.
	
\bibitem[Ji2]{Ji2}
	\bysame,
	{\it Quantum R Matrix for the Generalized Toda System}.
	Comm. Math. Phys. {\bf 102} (1986), 537--547.

\bibitem[Ji3]{Ji3}
	\bysame,
	{\it Introduction to the Yang-Baxter equation}.
	Int. J. of Mod. Phys. A {\bf 4} (1989), no. 15, 3759--3777.

\bibitem[JiMi]{JiMi}
	M. Jimbo, T. Miwa,
	{\it Algebraic analysis of solvable lattice models}.
	CBMS Reg. Conf. Ser. in Math. {\bf 85}, AMS (1994).
	
\bibitem[JKKKM]{JKKKM}
	M. Jimbo, R. Kedem, T. Kojima, H. Konno, T. Miwa,
	{\it XXZ chain with a boundary}.
	Nucl. Phys. B {\bf 441} (1995), no. 3, 437-470.

\bibitem[JKKMW]{JKKMW}
	M. Jimbo, R. Kedem, H. Konno, T. Miwa, R. Weston,
	{\it Difference equations in a spin chain with a boundary}.
	Nucl. Phys. B {\bf 448} (1995), no. 3, 429--456.	

\bibitem[JoMa]{JoMa}
	D. Jordan, X. Ma,
	{\it Quantum symmetric pairs and representations of double affine Hecke algebras of type $C^\vee C_n$}.
	Selecta Math. {\bf 17} (2011), no. 1, 139--181.
	{\tt arXiv:0908.3013}.

\bibitem[Ka]{Ka}
	V. Kac,
	{\it Infinite dimensional Lie algebras}.
	3rd. ed., Cambridge University Press, 1994.
	
\bibitem[KaWa]{KaWa}
	V. Kac, S. Wang,
	{\it On Automorphisms of Kac-Moody Algebras and Groups}.
	Adv. Math. {\bf 92} (1992), 129--195.

\bibitem[KhTo]{KhTo}
	S. Khoroshkin, V. Tolstoy,
	{\it The universal R-matrix for quantum untwisted affine Lie algebras}.
	Funct. Anal. Appl. {\bf 26} (1992), 69--71.

\bibitem[KlSg]{KlSg}
	A. Klimyk, K. Schm\"{u}dgen,
	{\it Quantum groups and their representations}.
	Springer-Verlag, 1997.

\bibitem[KKMMNN]{KKMMNN}
	S.-J. Kang, M. Kashiwara, K. Misra, T. Miwa, T. Nakashima, A. Nakayashiki,
	{\it Affine crystals and vertex models}.
	Int. J. Mod. Phys. A {\bf 7}, Suppl. 1A (1992), 449--484.
	
\bibitem[Ko1]{Ko1}
	S. Kolb,
	{\it Quantum symmetric Kac-Moody pairs}.
	Adv. Math. {\bf 267} (2014), 395--469.
	{\tt arXiv:1207.6036}. 
	
\bibitem[Ko2]{Ko2}
	\bysame, unpublished.

\bibitem[KRS]{KRS}
	P. P. Kulish, N. Yu. Reshetikhin, E. Sklyanin,
	{\it Yang-Baxter equation and representation theory: I}.
	Lett. Math. Phys. {\bf 5} (1981), no. 5, 393--403.

\bibitem[KSS]{KSS}
	P. P. Kulish, R. Sasaki, C. Schwiebert,
	{\it Constant Solutions of Reflection Equations and Quantum Groups}.
	J. Math. Phys. {\bf 34} (1993), 286--304.
	{\tt arXiv:hep-th/9205039}.
	
\bibitem[KuMv]{KuMv}
	P. P. Kulish, A. Mudrov,
	{\it Baxterization of solutions to reflection equation with Hecke R-matrix}.
	Lett. Math. Phys. {\bf 75} (2006), no. 2, 151--170.
	{\tt arXiv:math/0508289}.

\bibitem[KuSk1]{KuSk1}
	P. P. Kulish, E. K. Sklyanin,
	{\it Quantum spectral transform method: recent developments}.
	Integrable Quantum Field Theories, Lecture Notes in Physics {\bf 151} (1982), 61--119.
	
\bibitem[KuSk2]{KuSk2}
	\bysame,
	{\it Algebraic structures related to reflection equations}.
	J. Phys. A {\bf 25} (1992), no. 22, 5963--5975.
	{\tt arxiv:hep-th/9209054}.

\bibitem[Le1]{Le1}
	G. Letzter,
	{\it Symmetric Pairs for Quantized Enveloping Algebras}.
	J. Algebra {\bf 220} (1999), 729--767. 
		
\bibitem[Le2]{Le2}
	\bysame,
	{\it Coideal Subalgebras and Quantum Symmetric Pairs}.
	New Directions in Hopf Algebras, MSRI publications 43, Cambridge University Press (2002), 117--166.
	{\tt arXiv:math/0103228}. 
	
\bibitem[Le3]{Le3}
	\bysame,
	{\it Quantum Symmetric Pairs and Their Zonal Spherical Functions}.
	Transformation Groups {\bf 8}, no. 3 (2003), 261--292
	{\tt arXiv:math/0204103}. 

\bibitem[Lu]{Lu}
	G. Lusztig,
	{\it Introduction to quantum groups}.
	Birkh\"auser, Boston, 1994.
	
\bibitem[LvMt]{LvMt}
	D. Levy, P. Martin, 
	{\it Hecke algebra solutions to the reflection equation}.
	J. Phys. A {\bf 27} (1994), no. 14.

\bibitem[MeNe1]{MeNe1}
	L. Mezincescu, R.I. Nepomechie,
	{\it Integrable open spin chains with nonsymmetric R-matrices}.
	J. Phys. A: Math. Gen. {\bf 24} (1991), no. 1, L17.

\bibitem[MeNe2]{MeNe2}
	L. Mezincescu, R.I. Nepomechie, 
	{\it Fusion procedure for open chains}.
	J. Phys. A: Math. Gen. {\bf 25} (1992), no. 9, 2533.
	
\bibitem[Mk]{Mk}
	N. Mackay,
	{\it Introduction to Yangian symmetry in integrable field theory}.
	Int. J. Mod. Phys. A{\bf 20} (2005), 7189--7218.
	{\tt arXiv:hep-th/0409183}.
	
\bibitem[MLS]{MLS}
	R. Malara, A. Lima-Santos,
	{\it On $\mc A^{(1)}_{n-1}$, $\mc B^{(1)}_n$, $\mc C^{(1)}_n$, $\mc D^{(1)}_n$, $\mc A^{(2)}_{2n}$, $\mc A^{(2)}_{2n-1}$ and $\mc D^{(2)}_{n+1}$ reflection $K$-matrices}.
	J. Stat. Mech: Theory and Exp. {\bf 9} (2006), P09013.
	{\tt arXiv:nlin/0412058}.

\bibitem[Mo]{Mo}
	A. Molev,
	{\it Representations of the twisted quantized enveloping algebra of type $C_n$}.
	Moscow Mathematical Journal, {\bf 6} (2006), 531--551.
	{\tt arXiv:math/0602206}.

\bibitem[MoRa]{MoRa} 
	A. Molev, E. Ragoucy, 
	{\it Representations of reflection algebras}.
	Rev. Math. Phys. \textbf{14} (2002), no. 3, 317--342.
	{\tt arXiv:math/0107213}.
	
\bibitem[MRS]{MRS}	
	A. Molev, E. Ragoucy, P. Sorba,
	{\it Coideal subalgebras in quantum affine algebras}.
	Rev. Math. Phys. {\bf 15} (2003), no. 8, 789--822.
	{\tt arXiv:math/0208140}.	
	
\bibitem[NDS]{NDS}
	M. Noumi, M.S. Dijkhuizen and T. Sugitani,
	{\it Multivariable Askey-Wilson polynomials and quantum complex Grassmannians}, 
	AMS Fields Inst. Commun. {\bf 14} (1997), 167-177.
	{\tt arXiv:q-alg/9603014}.

\bibitem[NoSu]{NoSu} 
	M. Noumi, T. Sugitani, 
	{\it Quantum symmetric spaces and related q-orthogonal polynomials}, in: Group Theoretical Methods in Physics (ICGTMP) (Toyonaka, Japan, 1994),
	World Scientific Publishing, River Edge, NJ, (1995), 28--40.
	{\tt arXiv:math/9503225}.
	
\bibitem[Ol]{Ol}
	G. Olshanskii,
	{\it Twisted Yangians and infinite-dimensional classical Lie algebras}.
	Quantum groups (Leningrad, 1990), 104--119, Lecture Notes in Math. 1510, Springer, Berlin, 1992.

\bibitem[RSV1]{RSV1}
	N. Reshetikhin, J. Stokman, B. Vlaar,
	{\it Boundary quantum Knizhnik-Zamolodchikov equations and Bethe vectors}.
	Commun. Math. Phys. {\bf 336} (2015), no. 2, 953--986.
	{\tt arXiv:1305.1113}.

\bibitem[RSV2]{RSV2}
	\bysame,
	{\it Boundary quantum Knizhnik-Zamolodchikov equations and fusion}.
	Annales Henri Poincar\'{e} (2015), 1--41.
	{\tt arXiv:1404.5492}.

\bibitem[RSV3]{RSV3}
	\bysame,
	{\it Integral solutions to boundary quantum Knizhnik-Zamolodchikov equations}.
	Preprint, {\tt arXiv:1602.08457}.

\bibitem[Sk]{Sk}
	E. Sklyanin,
	{\it Boundary conditions for integrable quantum systems}.
	J. Phys. A: Math. Gen. {\bf 21} (1988), 2375--2389.
	
\bibitem[Sm]{Sm}
	F. Smirnov,
	{\it A general formula for solution form factors in the quantum sine-Gordon model}.
	J. Phys. A {\bf 19} no. 10 (1986), L575--L578.

\bibitem[StVl]{StVl}
	J. Stokman, B. Vlaar,
	{\it Koornwinder polynomials and the XXZ spin chain}.
	J. of Appr. Theory {\bf 197} (2015), 69--100.
	{\tt arXiv:1310.5545}.

\bibitem[STS]{STS}
	M. A. Semenov-Tian-Shansky,
	{\it Quantum and Classical Integrable Systems}.
	Springer Berlin Heidelberg, 1997, 314--377.
	{\tt arXiv:q-alg/9703023}.

\bibitem[TaVa]{TaVa}
	V. Tarasov, A. Varchenko,
	{\it Jackson integral representations for solutions to the quantized Knizhnik-Zamolodchikov equation}.
	St. Petersburg Math. J. {\bf 6} (1995), no. 2, 275--313.
	{\tt arXiv:hep-th/9311040}.
	
\bibitem[Tw]{Tw}
	E. Twietmeyer,
	{\it Real forms of $U_q(\mfg)$}.
	Lett. Math. Phys. {\bf 24} (1992), 49--58.

\bibitem[dVGR]{dVGR}
	H. de Vega, A. Gonz\'{a}lez-Ruiz,
	{\it Boundary K-matrices for the XYZ, XXZ and XXX spin chains}.
	J. Phys. A: Math. Gen. {\bf 27} (1994), 6129.
	{\tt arXiv:hep-th/9306089}.

\bibitem[Vl]{Vl}
	B. Vlaar,
	{\it Boundary transfer matrices and boundary quantum KZ equations}.
	J. Math. Phys. {\bf 56} (2015), 071705.
	{\tt arXiv:1408.3364}.

\bibitem[WYCS]{WYCS}
	Y. Wang, W.-L. Yang, J. Cao, K. Shi,
	{\it Off-diagonal Bethe ansatz for exactly solvable models}.
	Springer (2015).

\bibitem[ZaZa]{ZaZa}
	A.B. Zamolodchikov, A.B. Zamolodchikov,
	{\it Factorized S-matrices in two dimensions as the exact solutions of certain relativistic quantum field models}.
	Ann. Phys. {\bf 120} (1979), 253--291.

\end{thebibliography}
\end{document}